\documentclass[11pt, a4paper, oneside]{Thesis} 
\graphicspath{{./Pictures/}} 

\usepackage{changepage}
\usepackage{nameref}
\usepackage[greek,english]{babel}
\usepackage[square, numbers, comma, sort&compress]{natbib}
\usepackage{amsmath}
\usepackage{amsfonts}
\usepackage{algorithm,algorithmic}
\usepackage{amssymb}
\usepackage{latexsym}
\usepackage{graphicx}
\usepackage{color}
\usepackage{url}
\usepackage{caption}
\usepackage{graphicx,subfigure}
\usepackage{balance,hyperref}
\hypersetup{urlcolor=blue, colorlinks=true} 

\DeclareMathOperator*{\gimv}{\times}

\title{\ttitle} 

\begin{document}

\frontmatter 

\setstretch{1.3} 

\fancyhead{} 
\rhead{\thepage} 
\lhead{} 

\pagestyle{fancy} 

\newcommand{\HRule}{\rule{\linewidth}{0.5mm}} 

\hypersetup{pdftitle={\ttitle}}
\hypersetup{pdfsubject=\subjectname}
\hypersetup{pdfauthor=\authornames}
\hypersetup{pdfkeywords=\keywordnames}


\begin{titlepage}
\begin{center}

\textsc{\LARGE \univname}\\[1.5cm] 
\textsc{\Large Doctoral Thesis}\\[0.5cm] 

\HRule \\[0.4cm] 
{\huge \bfseries \ttitle}\\[0.4cm] 
\HRule \\[1.5cm] 
 
\begin{minipage}{0.4\textwidth}
\begin{flushleft} \large
\emph{Author:}\\
\href{http://www.math.cmu.edu/~ctsourak/}{\authornames} 
\end{flushleft}
\end{minipage}
\begin{minipage}{0.4\textwidth}
\begin{flushright} \large
\emph{Supervisor:} \\
\href{http://www.math.cmu.edu/~af1p}{\supname} 
\end{flushright}
\end{minipage}\\[3cm]
 
\examname

\large \textit{A thesis submitted in fulfilment of the requirements\\ for the degree of \degreename}\\[0.3cm] 
\textit{in the}\\[0.4cm]
\groupname\\\deptname\\[2cm] 
 
 
\vfill
\end{center}

\end{titlepage}


\pagestyle{empty} 

\null\vfill 

\textit{``I don't keep checks and balances, I don't seek to adjust myself. I follow the deep throbbing of my heart."}

\begin{flushright}
Nikos Kazantzakis 
\end{flushright}

\vfill\vfill\vfill\vfill\vfill\vfill\null 

\clearpage 

\null \vfill
\keywordnames

\clearpage 


\addtotoc{Abstract} 

\abstract{\addtocontents{toc}{\vspace{-4em}} 

This dissertation contributes to mathematical and algorithmic problems that 
arise in the analysis of network and biological data. 
The driving force behind this dissertation is the importance of 
studying network and biological data. 
Studies of such data can provide us with insights to 
important emerging and long-standing problems, such as {\it why do some mutations cause cancer whereas others do not?
What is the structure of the Web graph? 
How do networks form and how does their structure affect the the spread of ideas and diseases?}
This thesis consists of two parts. The first part is devoted to {\it graphs and networks} and the second part 
to {\it computational cancer biology}. Our contributions to {\it graphs and networks} 
revolve around the following two axes.

$\bullet$ Empirical studies: In order to develop a good model, 
one has to study structural properties of real-world networks. 
Given the size of today's networks, such empirical studies and 
graph-structured computations become challenging tasks.  
Therefore, in the course of empirically studying  real-world networks, 
both novel and existing well-studied problems but in new computational models need to be solved. 
We provide both efficient algorithms with strong theoretical guarantees for various graph-structured
computations and graph processing systems that allow us to handle big graph data. 

$\bullet$ Models: The quality of a given model is judged by how well 
it matches reality. We analyze fundamental graph theoretic properties of random 
Apollonian networks, a random graph model which mimicks real-world networks. 
Furthermore, we use random graphs to perform an average case analysis for 
rainbow connectivity, an intriguing connectivity concept. We provide
simple randomized procedures which succeed to solve the problem of interest 
with high probability on random binomial and regular graphs. 

Our contributions to {\it computational cancer biology} include 
new models,  theoretical insights into existing models, novel algorithmic techniques 
and detailed experimental analysis of various datasets. Specifically, 
we contribute the following. 

$\bullet$ Novel algorithmic techniques for denoising array comparative 
genomic hybridization data. Our algorithmic results are of independent 
interest and provide  approximation techniques for speeding up
dynamic programming. 

$\bullet$ Based on empirical findings which strongly indicate an inherent
geometric structure in cancer genomic data, we introduce a geometric model 
for finding subtypes of cancer which overcomes difficulties of 
existing methods such as principal/independent component analysis
and separation methods for Gaussians. We provide a computational method 
which solves the optimization problem efficiently and is robust to outliers. 

$\bullet$ We find the necessary and sufficient conditions to reconstruct uniquely an 
oncogenetic tree, a popular tumor phylogenetic method. 
}

\clearpage 


\setstretch{1.3} 

\acknowledgements{\addtocontents{toc}{\vspace{1em}} 

First, I wish to express my deep gratitude to my advisor Professor Alan M. Frieze. 
I feel privileged to have worked under his supervision. His insights, advice and support 
shall remain invaluable to me. His dedication to noble research and open mind
are a source of inspiration. 

Also, I thank Professor  Russell Schwartz and Professor Gary L. Miller. 
Russell introduced me to computational biology and provided me with 
insights on how to perform research on this human-centric and interdisciplinary domain. 
The meetings I have had with Gary have been the best teaching sessions on algorithm design I have had. 
My interaction with both of them has been a great pleasure. 

A key figure during my Ph.D. years was Professor Mihail N. Kolountzakis. 
I want to thank him for his beneficial impact on my research, our collaboration and 
for hosting me during the summer of 2010 in the University of Crete. 
Furthermore, I would like to thank Professor Aristides Gionis for a fruitful collaboration
that began in Yahoo! Research and will continue in Aalto University for the next academic
semester. I also thank Professor Christos Faloutsos for advising me 
during the first two years of my Ph.D. in the School of Computer Science. 

I would like to thank the members of my thesis committee for their valuable 
feedback: Professors  Tom Bohman, Alan Frieze, Russell Schwartz and R. Ravi. 

I would like to thank my collaborators throughout my Ph.D. years in the School of Computer 
Science and the Department of Mathematical Sciences: Petros Drineas, Aristides Gionis, Christos Faloutsos,  
Alan Frieze, U Kang, Mihail N. Kolountzakis, Ioannis Koutis, Jure Leskovec, Eirinaios Michelakis, Gary L. Miller, 
Rasmus Pagh, Richard Peng, Aditya Prakash, Russell Schwartz, Stanley Shackney, Ayshwarya Subramanian, 
Jimeng Sun,  David Tolliver, Maria Tsiarli, Hanghang Tong.
 
I want to thank all my friends for their support and the good times we have had over
the last years.  I would like to especially thank  
Felipe Trevizan,  Jamie Flanick,  Ioannis Mallios, Evangelos Vlachos, Ioannis Koutis, 
Don Sheehy, Shay Cohen, Eirinaios Michelakis, 
Nikos Papadakis, Dimitris Tsaparas and Dimitris Tsagkarakis. 

I thank my parents Eftichios and Aliki and my sister Maria for their love and support.  
Finally, I want to thank my fianc\'ee Maria Tsiarli. 
}
\clearpage 


\pagestyle{fancy} 

\lhead{\emph{Contents}} 
\tableofcontents 

\lhead{\emph{List of Figures}} 
\listoffigures 

\lhead{\emph{List of Tables}} 
\listoftables 


\clearpage 

\setstretch{1.5} 


\setstretch{1.3} 

\pagestyle{empty} 

\dedicatory{Dedicated to: \\ 
my fianc\'ee Maria Tsiarli  \\
my parents Aliki and Eftichios and my sister Maria, \\ 
Mr. Andreas Varverakis 
} 

\addtocontents{toc}{\vspace{2em}} 


\mainmatter 

\pagestyle{fancy} 


\newcommand{\fennel}{{\sc Fennel}\xspace}
\newcommand{\metis}{\textsc{Metis}\xspace}

\newcommand{\gmvm}{Generalized Iterative Matrix-Vector multiplication}
\newcommand{\IGMV}{{\tt GIM-V}\xspace}
\newcommand{\pegasus}{\textsc{PeGaSus}\xspace}
\newcommand{\hadi}{\textsc{Hadi}\xspace}
\newcommand{\hcc}{\textsc{Hcc}\xspace}

\newcommand{\join}{\texttt{combine2}}
\newcommand{\aggrnp}{\texttt{combineAll}}
\newcommand{\aggr}{\texttt{combineAll$_i$}}
\newcommand{\aggrsid}{\texttt{combineAll$_{E.sid}$}}
\newcommand{\assign}{\texttt{assign}}



\newcommand{\mapreduce}{\textsc{MapReduce}\xspace}
\newcommand{\hadoop}{\textsc{Hadoop}\xspace}
\newcommand{\anf}{\emph{Centralized Method}}
\newcommand{\MFB}{{FM-bitstring}}
\newcommand{\mfb}{{b}}
\newcommand{\MFV}{{FM-vector}}
\newcommand{\mfv}{{\bf v}}
\newcommand{\mfbhi}{{ b(h,i) }}
\newcommand{\Nhi}{N(h,i)}   
\newcommand{\NNhi}{ {\cal N}(h,i)} 
\newcommand{\NNhhi}{ {\cal N}(h+1,i)} 
\newcommand{\NNhj}{ {\cal N}(h,j)} 
\newcommand{\hatmfbhi}{ $\hat{b}(h,i)$} 

\newcommand{\PassA}{{\tt Stage1}\xspace}
\newcommand{\PassAMap}{{\tt Stage1-Map}\xspace}
\newcommand{\PassARed}{{\tt Stage1-Reduce}\xspace}
\newcommand{\PassB}{{\tt Stage2}\xspace}
\newcommand{\PassBMap}{{\tt Stage2-Map}\xspace}
\newcommand{\PassBRed}{{\tt Stage2-Reduce}\xspace}
\newcommand{\PassC}{{\tt Stage3}\xspace}
\newcommand{\PassCMap}{{\tt Stage3-Map}\xspace}
\newcommand{\PassCRed}{{\tt Stage3-Reduce}\xspace}

\newcommand{\diameter}{{\tt diameter}}
\newcommand{\npairs}{{\tt npairs}}
\newcommand{\gcc}{{\tt gcc}}
\newcommand{\entropy}{{\tt entropy}}
\newcommand{\sv}{{\mbox{$\lambda_1$}}}
\newcommand{\avgd}{{\tt avgd}}

\newcommand{\mat}[1]{{\bf{#1}}}
\newcommand{\nnodes}{N}     
\newcommand{\nedges}{E}     
\newcommand{\deff}{{d_{eff}}}   
\newcommand{\Er}{{E_{r}}}   
\newcommand{\Nr}{{N_{r}}}   
\newcommand{\Nc}{{N_c}}   
\newcommand{\Ec}{{E_c}}   
\newcommand{\Ngcc}{{N_{gcc}}} 
\newcommand{\NPairc}{{N_{NPAIRS}}} 
\newcommand{\dc}{{d_{c}}} 
\newcommand{\Lambdac}{{\lambda_{c}}} 

\newcommand{\ben}{\begin{enumerate*}}
\newcommand{\een}{\end{enumerate*}}
\newcommand{\bit}{\begin{itemize*}}
\newcommand{\eit}{\end{itemize*}}

\newcommand{\citationsds}{{CITATIONS}}
\newcommand{\epinionsds}{{EPINIONS}}
\newcommand{\patentsds}{{PATENTS}}
\newcommand{\net}[1]{{\textsc{#1}}}


\newcommand{\OES}{{\textsf{optimal $(g,h,\alpha)$-edge-surplus}}}
\newtheorem{problem}{Problem}
\newtheorem{observation}{Observation}
\newcommand{\Polytime}{{\ensuremath{\mathbf{P}}}}

\newcommand{\trimmer}{\textsc{CGHtrimmer}}
\newcommand{\dnacopy}{\textsc{CBS}}
\newcommand{\cghseg}{\textsc{CGHseg}}

\newcommand{\beql}[1]{\begin{equation}\label{#1}}
\newcommand{\eeq}{\end{equation}}
\newcommand{\beq}[1]{\begin{equation}\label{#1}}
\def\hT{\widehat{T}}
\def\hD{\widehat{D}}
\def\hP{\widehat{P}}
\def\cC{\mathcal{C}}
\def\cA{{\cal A}}
\def\cR{\mathcal{R}}

\def\hC{\widehat{Q}}
\def\depth{\text{depth}}
\def\Bin{\text{Bin}}

\newcommand{\rdup}[1]{{\left\lceil #1 \right\rceil }}
\newcommand{\rdown}[1]{{\lfloor #1 \rfloor}}
\newcommand{\proofstart}{{\bf Proof\hspace{2em}}}
\newcommand{\proofend}{\hspace*{\fill}\mbox{$\Box$}}

\newcommand{\brac}[1]{\left(#1\right)}
\newcommand{\bfrac}[2]{\left(\frac{#1}{#2}\right)}
\newcommand{\whp}{\textit{whp}}

\def\red{}
\def\rred{}
\def\blue{\color{blue}}

\newtheorem{claim}{Claim}
\newcommand{\diam}{\frac{ \log{n}}{\log{\log{n}}}}
\newcommand{\set}[1]{\left\{#1\right\}}

\def\a{\alpha} \def\b{\beta} \def\d{\delta} \def\D{\Delta}
\def\e{\epsilon} \def\f{\phi} \def\F{{\Phi}} \def\vp{\varphi} \def\g{\gamma}
\def\G{\Gamma} \def\i{\iota} \def\k{\kappa}\def\K{\Kappa}
\def\z{\zeta} \def\th{\theta} \def\TH{\Theta} \def\Th{\Theta}  \def\l{\lambda}
\def\La{\Lambda} \def\m{\mu} \def\n{\nu} \def\p{\pi}
\def\r{\rho} \def\R{\Rho} \def\s{\sigma} \def\Si{\Sigma}
\def\t{\tau} \def\om{\omega} \def\OM{\Omega} \def\Om{\Omega}

\def\ve{\varepsilon}
\def\hY{\hat{Y}}
\def\hx{\hat{x}}
\def\hv{\hat{v}}
\def\bY{{\bf Y}}
\def\cY{{\cal Y}}
\def\cB{{\cal B}}
\def\cT{{\cal T}}
\def\leqO{\stackrel{O}{\leq}}
\def\geqO{\stackrel{\Omega}{\geq}}
\def\bhY{\hat{\bY}}
\def\hR{\widehat{R}}
\def\hr{\widehat{\r}}

\newcommand{\Ds}{\displaystyle}
\newcommand{\Fpart}[1]{{\left\{{#1}\right\}}}
\newcommand{\Abs}[1]{{\left|{#1}\right|}}
\newcommand{\Lone}[1]{{\left\|{#1}\right\|_{L^1}}}
\newcommand{\Linf}[1]{{\left\|{#1}\right\|_\infty}}
\newcommand{\Norm}[1]{{\left\|{#1}\right\|}}
\newcommand{\Mean}[1]{{\mathbb E}\left[{#1}\right]}
\newcommand{\Var}[1]{{\mathbb Var}\left[{#1}\right]}
\newcommand{\Floor}[1]{{\left\lfloor{#1}\right\rfloor}}
\newcommand{\Ceil}[1]{{\left\lceil{#1}\right\rceil}}

 \renewcommand{\algorithmicrequire}{\textbf{Input:}}
\renewcommand{\algorithmicensure}{\textbf{Output:}}

\newcommand{\reminder}[1]{{\textsf{\textcolor{red}{[#1]}}}}
\newcommand{\Prob}[1]{\ensuremath{{\bf{Pr}}\left[{#1}\right]}}
\newcommand{\NP}{\ensuremath{\mathbf{NP}}}
\newcommand{\NPhard}{{\ensuremath{\mathbf{NP}}-hard}\xspace}
\newcommand{\NPcomplete}{{\ensuremath{\mathbf{NP}}-complete}}
\newcommand{\sgn}{{\ensuremath{\mathrm{sgn}}}}

\newcommand{\IPOPT}{{\ensuremath{\mathrm{IP}^{*}}}}
\newcommand{\SDPOPT}{{\ensuremath{\mathrm{SDP}^{*}}}}
\newcommand{\TMAX}{{\ensuremath{T_{\max}}}}

\newcommand{\topk}{{top-\ensuremath{k}}}
\newcommand{\ER}{{Erd\H{o}s-R\'{e}nyi}}
\newcommand{\bigO}{{\ensuremath{\cal O}}}

\newcommand{\QC}{{quasi-clique}}
\newcommand{\aQC}{{\ensuremath{\alpha}-quasi-clique}\xspace}
\newcommand{\OQC}{{\textsf{optimal quasi-clique}}\xspace}
\newcommand{\OQCs}{{\textsf{optimal quasi-cliques}}\xspace}
\newcommand{\OQCP}{{\sc OQC-Problem}\xspace}
\newcommand{\COQCP}{{\sc Constrained-OQC-Problem}\xspace}

\newcommand{\DS}{{\textsf{densest subgraph}}}
\newcommand{\DSs}{{\textsf{densest subgraphs}}}
\newcommand{\DSP}{{\sc{DS-Problem}}}

\newcommand{\SDPalgo}{{{SDP-OQC}}}
\newcommand{\LSalgo}{{\sc{Local\-SearchOQC}}}
\newcommand{\Galgo}{{\sc{GreedyOQC}}}

\newcommand{\tr}{{\rm Tr\,}}
\newcommand{\Cov}[1]{ {\mathbb Cov}\left[{#1}\right]}
\newcommand{\Hastad}{H{$\aa$}stad}

\newcommand{\Set}[1]{{\left\{{#1}\right\}}}
\newcommand{\One}[1]{{\mathbf 1}\left(#1\right)}

\newcommand{\added}[1]{\textcolor{red}{#1}}
\newcommand{\hide}[1]{}
\newcommand{\vectornorm}[1]{\left|\left|#1\right|\right|}
\newcommand{\field}[1]{\mathbb{#1}} 
\renewcommand{\vec}[1]{{\mbox{\boldmath$#1$}}}

\newcommand{\spara}[1]{\smallskip\noindent{\bf #1}}
\newcommand{\mpara}[1]{\medskip\noindent{\bf #1}}
\newcommand{\para}[1]{\noindent{\bf #1}}

\newcommand{\squishlist}{
 \begin{list}{$\bullet$}
  {  \setlength{\itemsep}{0pt}
     \setlength{\parsep}{3pt}
     \setlength{\topsep}{3pt}
     \setlength{\partopsep}{0pt}
     \setlength{\leftmargin}{2em}
     \setlength{\labelwidth}{1.5em}
     \setlength{\labelsep}{0.5em}
} }
\newcommand{\squishlisttight}{
 \begin{list}{$\bullet$}
  { \setlength{\itemsep}{0pt}
    \setlength{\parsep}{0pt}
    \setlength{\topsep}{0pt}
    \setlength{\partopsep}{0pt}
    \setlength{\leftmargin}{2em}
    \setlength{\labelwidth}{1.5em}
    \setlength{\labelsep}{0.5em}
} }

\newcommand{\squishdesc}{
 \begin{list}{}
  {  \setlength{\itemsep}{0pt}
     \setlength{\parsep}{3pt}
     \setlength{\topsep}{3pt}
     \setlength{\partopsep}{0pt}
     \setlength{\leftmargin}{1em}
     \setlength{\labelwidth}{1.5em}
     \setlength{\labelsep}{0.5em}
} }

\newcommand{\squishend}{
  \end{list}
}


\newpage
\vspace*{\fill}
\begingroup

\endgroup
\vspace*{\fill}

\newpage

\clearpage

\chapter{Introduction} 
\label{introchapter}
\lhead{\emph{Introduction}}
The motivating force behind this dissertation is the importance 
of studying network and biological data. Such studies can provide 
us with  insights or even answers to various significant questions: 
{\it How do people establish connections among each other and how does 
the underlying social graph affect the spread of ideas and diseases? 
What will the structure of the Web be in some years from now? 
How can we design better marketing strategies? 
Why do some mutations cause cancer whereas others do not? 
Can we use cancer data to improve diagnostics and therapeutics? }

This dissertation contributes to mathematical and algorithmic
problems which arise in the course of studying {\it network} 
and {\it cancer} data. This Chapter is organized as follows:
in Sections~\ref{subsec:introrealworldnetworks} and~\ref{subsec:introcancerdata} 
we introduce the reader to some basic network and biology concepts respectively. 
In Section~\ref{subsec:connectingdotsoverview} we motivate our work and 
present the contributions of this dissertation.

\section{Graphs and Networks} 
\label{subsec:introrealworldnetworks} 

Networks appear throughout in nature and society \cite{albert2002statistical}. 
Furthermore, many applications which use other types of data,
such as text or image data, create graphs by data-processing. 
{\em Network Science}   has emerged over the last years as an interdisciplinary
area spanning traditional domains including mathematics, computer science, 
sociology, biology and economics. 
Since complexity in social, biological and economical systems, and more generally
in complex systems, arises through pairwise interactions there exists a surging
interest in understanding networks \cite{easley2010networks}.   
Graph theory \cite{bollobas} plays a key role in modeling and studying networks. 
Given the increasing importance of networks in our lives, Professor Daniel Spielman  
believes that Graph Theory is the new Calculus \cite{danspielmannotes}. 
Some important networks and the corresponding graph models follow.

\squishlist
\item  {\it Human brain}. Our brain consists of neurons which form a network. 
The number of neurons in the region of human cortex is estimated 
to be roughly $10^{10}$. Neurons are connected through synapses. 
The strength of synapses in general varies.  The average number of synapses of each neuron is 
in the range of 24\,000-80\,000 for humans \cite{valiant2005memorization}. 
The brain can be modeled as a graph, whose vertex set corresponds 
to neurons and edge set to synapses.

\item  {\it World Wide Web (WWW)}. The World Wide Web is a particularly
important network. Turing Award winner Jim Gray in his 1998 Turing award address  speech
mentioned ``The emergence of 'cyberspace' and the World Wide Web is like the discovery of a new continent''. 
The vertices of the WWW are 
Web pages and the edges are hyperlinks (URLs) 
that point from one page to another. 
In 2011, Google reported that WWW has more than a trillion of edges.

\item {\it Internet}. There exist two different types of Internet graphs.
In the first type, vertices  are routers and edges correspond to physical connections. 
The second level is the autonomous systems level, where 
a single vertex represents a domain, namely multiple routers and computers.
An edge is drawn if there is at least one route
that connects them. 
It is worth mentioning that 
the first type of topology can be studied with the 
{\it traceroute} tool, whereas the second with {\it BGP tables}.

\item  {\it Social and online social networks}. 
According to Aristotle man is by nature a social animal. Social networks precede 
the individual. These networks are modeled by using one vertex per human. 
Each edge corresponds to some sort of interaction, for instance friendship. 

Since the emergence of the World Wide Web, we live in the information age. 
Nowadays, online social networks and social media are 
a part of our daily lives which can affect society immensely. 
Again, such networks are modeled as graphs where 
each vertex corresponds to an account.
Each edge corresponds to a connection. 


\item {\it Collaboration networks}. There exist many types of collaboration networks. 
Vertices can represent for instance mathematicians or actors. 
An edge between two vertices is drawn when the corresponding mathematicians 
have co-authored a paper or when the corresponding actors have 
played in the same movie respectively.
Two famous collaboration networks are 
the Erd\"{o}s and the Kevin Bacon collaboration networks.

\item {\it Protein interaction networks}. 
Protein - protein interactions occur when two or more proteins bind together, often to carry out their biological function.
These interactions define protein interaction networks. 
The corresponding graph has a vertex for each protein. 
Proteins $i$ and $j$ are connected if they interact. 


\item {\it Wireless sensor networks}.  Vertices are autonomous sensors. These sensors  monitor physical 
or environmental conditions, such as temperature, brightness and humidity 
A directed edge $(i,j)$ suggests that information can be passed from sensor $i$ to sensor $j$.  

\item {\it Power grid}. 
An electrical grid is a network which delivers electricity 
from suppliers to consumers. The vertices of the power grid graph 
correspond  to generators, transformers and substations.
Edges correspond to high-voltage transmission lines. 

\item {\it Financial networks}. One of the globalization effects  
is the interdependence between financial instutitions and markets. 
Financial network models are represented by graphs 
whose vertex set corresponds to  financial institutions such as 
banks and countries. Edges correspond to different types of 
interactions, for instance borrower-lender relations.

\item {\it Blog networks}.  A blog is a web log. A blog may contain
a blogroll, namely a list of blogs that interest the blogger, or 
may reference other blogs through posts.
This network is represented by a graph whose vertex set corresponds
to blogs. An edge $(i,j)$ exists if the $i$-th blog
has blog $j$ in its blogroll or a URL pointing to $j$. 

\squishend

A remarkable fact which underpins network science 
is that many real-world networks, despite their different origins, 
share several common structural characteristics. 
In the next Section we present some established patterns, shared
by numerous real-world networks \cite{easley2010networks,frieze2013algorithmic}.

\subsection{Structural properties of real-world networks}

Real-world networks possess a variety of remarkable properties. 
For instance, it is known that 
they are robust to
random  but vulnerable to malicious attacks \cite{albert2002statistical,bollobas2004robustness}.
Here, we review some important empirical properties of real-world networks which are related to our work.

\subsubsection{Power laws}

Real-world networks typically have skewed degree distributions. 
These heavy tailed empirical distributions are frequently modeled as power laws. 

\begin{definition}
The degree sequence of a graph follows a power law distribution if the number 
of vertices $N_d$  with degree $d$ is given by $N_d \propto d^{-\alpha}$
where $\alpha>1$ is the power law degree exponent or slope. 
\end{definition}

One of the early papers that popularized power laws 
as the modeling choice for empirical degree distributions 
was the Faloutsos, Faloutsos and Faloutsos paper \cite{faloutsos1999power}.
The Faloutsos brothers found that the Internet at the autonomous systems level
follows a power law degree distribution with $\alpha \approx 2.4$. 
In general, there exist three main problems with the initial studies of power laws in networks. 
First, Internet graphs generated with traceroute sampling \cite{faloutsos1999power}
may produce power-law distributions due to the bias of the process, 
even if the true underlying graph is regular \cite{lakhina2003sampling}.
Secondly, there exist  methodological flaws in determining the exponent/slope of the 
power law distribution. Clauset et al. provide a proper methodology
for finding the slope of the distribution \cite{clauset2009power}. 
Finally, other distributions could potentially fit the data better but were 
not considered as candidates. Such distribution is the lognormal \cite{danspielmannotes}.
A nice review of power law and lognormal distributions appears in \cite{mitzenmacher2004brief}.

\subsubsection{Small-world}

Six degrees of separation is the theory that anyone in the world is no more than six relationships away from any other person. 
In the early 20th century  Nobel Peace Prize winner Guglielmo Marconi, the father of modern radio, 
suggested that it would take only six relay stations to cover and connect the earth by radio \cite{marconi}. 
It is likely that this idea was the seed for the six degrees of separation theory,
which was further supported by Frigyes Karinthy in a short story called Chains.
Since then many scientists, including Michael Gurevich, Ithiel De Sola Pool have worked on this theory. 
In a famous experiment, Stanley Milgram asked people to route a postcard to a fixed recipient
by passing them to direct acquintances  \cite{milgram}. 
Milgram observed that depending on the sample of people chosen 
the average number of intermediaries  was between 
4.4 and 5.7.  Milgram's experiment besides its existential aspect
has a strong algorithmic aspect as well, which was first studied
by Kleinberg \cite{kleinberg2000small}.

Nowadays, World Wide Web and online social networks provide us with data that reach the planetary scale. 
Recently, Backstrom, Boldi, Rosa, Ugander and Vigna  showed that the world 
is even smaller than what the six degrees of separation theory predicts \cite{DBLP:conf/websci/BackstromBRUV12}. 
Specifically, they perform the first  world-scale social-network graph-distance computation,
using the entire Facebook network of active users (at that time 721 million users, 69 billion friendship links)
and observe an average distance of 4.74.  
In Chapter~\ref{hadichapter} we shall see formal graph theoretical concepts which
quantify the small world phenomenon.

\subsubsection{Clustering coefficients}

Watts and Strogatz  in their 
influential paper \cite{j1998collective} proposed a simple graph model 
which reproduces two features of social networks: 
abundance of triangles and the existence of short paths among any pair of nodes. 
Their model combines the idea of homophily which leads to the wealth of triangles
in the network and the idea of weak ties which create short paths. In order to quantify the homophily, 
they introduce the definitions of
the clustering coefficient. 
The definition of the transitivity $T(G)$ of a graph $G$, 
introduced by Newman et al. \cite{citeulike:691419}, 
is closely related to the clustering coefficient and measures the 
probability that two neighbors of a vertex are connected.

\begin{definition}[Clustering Coefficient]
A vertex $v \in V(G)$ with degree $d(v)$ which participates into $t(v)$ triangles 
 has clustering coefficient $C(v)$  equal to the fraction of edges among 
its neighbors to the maximum number of triangles it could participate:

\begin{equation}
\label{eq:stogatz1}
C(v) = \frac{t(v)}{{d(v) \choose 2}}
\end{equation}

The clustering coefficient $C(G)$ of graph $G$ is the average of $C(v)$ over all $v \in V(G)$.
\end{definition}

\begin{definition}[Transitivity] 
The transitivity of a graph measures the probability that two neighbors of a vertex are connected:

\begin{equation}
\label{eq:transitivity} 
T(G) = \frac{ 3 \times t }{\text{number of connected triples}}
\end{equation}

\end{definition}

\subsubsection{Communities} 

Intuitively, communities are  sets of vertices which are densely
intra-connected and sparsely inter-connected \cite{girvan2002community,newman}. 
A large amount of research in network science has 
focused on finding communities.
The goal of community detection methods is to partition the graph vertices into communities so that there
are many edges among vertices in the same community and few edges among vertices in different communities.
A landmark study of communities by Leskovec et al. \cite{jure08ncp} placed various folklores revolving around 
community existence in question. 
Specifically, Leskovec et al. observed in the majority of the networks 
they studied that communities exist at small size scales. 
Specifically, as the size increases up to a value which empirically 
is close to 100 the quality of the community tends to increase. 
However, after the critical value of 100, the quality tends to decrease.
This indicates that communities blend in with the rest of the network
and lose their community-like profile. 

There exist various formalizations of the community notion \cite{jure08ncp}. 
However, what underpins all these formalizations is the attempt to understand 
the cut structure of the graph. A popular measure for the quality of a community 
is the conductance. 

\begin{definition}
Given a graph $G(V,E,w)$ and a set $S \subseteq V$ of vertices, the conductance 
of $S$ is defined as 

$$ \phi(S) = \frac{ \sum\limits_{(i,j)\in E, i \in S,j \in \bar{S}} w_{ij} }{\min (vol(S),vol(\bar{S}))},$$

\noindent where $\bar{S}=V\backslash S$ and the volume of a given set $A \subseteq V$ of vertices 
is defined as $vol(A) = \sum_{i \in A} d(i)$. 
The conductance of the graph $G$ is defined as 
$$ \phi = \min_{S \subset V } \phi(S).$$ 
\end{definition} 

It is not a coincidence that random walks are frequently 
used to find communities \cite{pons2005computing} as their use 
is common in the   general setting of graph partitioning \cite{Orecchia:EECS-2011-56}. 
A lot of interest exists into finding dense sets of vertices around a given seed. 
A popular method for finding such sets was first introduced by 
Lov\'{a}sz and Simonovits \cite{lovasz1990mixing} who show
that random walks of length $O(\frac{1}{\phi})$ can be used 
to compute a cut with sparsity at most $\tilde{O}(\sqrt \phi)$ 
if the sparsest cut has conductance $\phi$. 
Later, Spielman and Teng \cite{spielman2004nearly,spielman2008local} 
provided a local graph partitioning 
algorithm which implements efficiently the Lov\'{a}sz-Simonovits idea. 
Furthermore, their algorithm has a bounded work/volume ratio.
Another closely related approach 
which does not explicitly compute sequences of random walk distributions 
but computes a personalized Pagerank vector was introduced 
by  Andersen, Chung, and Lang \cite{andersen2006local}. 
A few other representative approaches for the problem of community detection include methods on
minimum cut~\cite{flake2000efficient},
modularity maximization~\cite{girvan2002community}, and
spectral methods~\cite{metis,DBLP:conf/nips/NgJW01}.
In fact the literature on the topic is so extensive that we do not attempt to make a proper review here;
a comprehensive survey has been conducted by Fortunato~\cite{fortunato2010community}.

In the typical setting of finding communities, 
the vertex set of the graph is partitioned. 
A relaxation of the latter requirement, allowing overlaps
between sets of vertices, yields the notion of overlapping communities 
\cite{arora2012finding,towardstsourakakis,zhu2013local}.

\subsubsection{Densification and Shrinking Diameters}

Leskovec, Kleinberg and Faloutsos \cite{1217301} studied how numerous real-world networks 
from a variety of domains evolve over time. In their work two important observations
were made. First, networks become denser over time 
and the densification follows a power law pattern. 
Secondly, effective diameters shrink over time. 
The second pattern is particularly surpising and creates a modeling challenge as well, 
since the vast majority of real-world networks have a diameter that grows as the network grows.

\subsubsection{Web graph}

We focus on the bow {\it bow-tie structure} of the Web graph. 
Other important properties of the Web graph include the 
abundance of bipartite cliques \cite{kleinberg1999web,kumar2000stochastic} and 
compressibility \cite{boldi2009permuting,chierichetti2009compressing,chierichetti2013models}.
 In 1999 Andrei Broder et al. \cite{broder2000graph} performed an influential study of the 
Web graph using strongly connected components (SCCs) as their building blocks. Specifically, 
they proposed the bow-tie model for the structure of the Web graph based on 
their findings on the index of pages and links of the AltaVista search engine. 
According to the bow-tie structure of the Web, there exists a single giant SCC. 
Broder et al. \cite{broder2000graph} positioned the remaining SCCs with respect to the giant SCC 
as follows: 

\begin{figure*}
\centering
\includegraphics[width=0.85\textwidth]{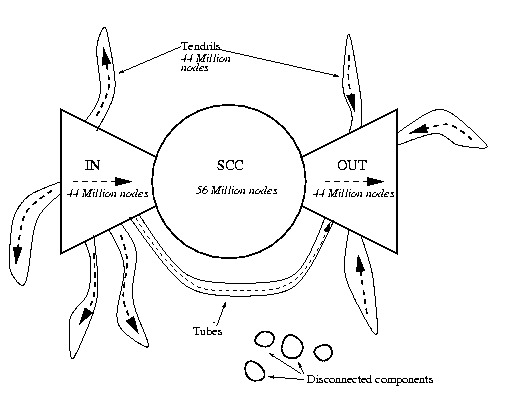}
\caption{\label{fig:bowtie} Bow-tie structure of the Web graph (Image  source: \cite{broder2000graph})}
\end{figure*}

$\bullet$ IN: vertices that can reach the giant SCC but cannot be reached from it. \\ 
$\bullet$ OUT: vertices that can be reached from the giant SCC but cannot be reach it. \\  
$\bullet$ Tendrils: These are vertices that either are reachable from IN but cannot reach
the giant SCC or the vertices that can reach OUT but cannot be reached from the giant SCC. \\  
$\bullet$ Disconnected: vertices that belong to none of the above categories. 
These are the vertices that even if we ignore the direction of the edges, 
have no path connecting them to the giant SCC. \\

A schematic picture of the bow-tie structure of the Web is shown in Figure~\ref{fig:bowtie}. 
This structure has been verified in other studies as 
well \cite{Donato:2007:WGF:1189740.1189744,Bharat:2001:LWM:645496.757722}.

\subsection{Network Models} 
\label{sec:chapter1networkmodels}
There exist two main research lines which model network formation. 
Random graph and strategic models. 
We focus on random graph models and specifically the models used
in this dissertation. The interested 
reader may consult the cited papers and references therein for more
random graph models and \cite{jackson2010social}  for a rich account of 
results on strategic models.

\subsubsection{Erd\"{o}s-R\'enyi random graphs} 

Let $\mathcal{G}$ be the family of all labeled 
graphs with vertex set $V=[n]=\{1,\ldots,n\}$. 
Notice $|\mathcal{G}|=2^{{n \choose 2}}$. 
Random graph models assign to each graph $G \in \mathcal{G}$ a probability. 
The random binomial graph model $G(n,p)$ 
has two parameters, $n$ the number of vertices 
and a probability parameter $0 \leq p \leq 1$. 
The $G(n,p)$ model assigns to a graph $G \in \mathcal{G}$ the following probability

$$ \Prob{G} = p^{|E(G)|} (1-p)^{{n \choose 2}-|E(G)|}.$$ 

We will refer to random binomial graphs as Erd\"{o}s-R\'{e}nyi graphs interchangeably. 
Historically, Gilbert \cite{gilbert1959random} introduced originally the $G(n,p)$ model but
Erd\"{o}s and R\'{e}nyi founded the field of random graphs \cite{erdos1960evolution,erdosrenyi}. 
They introduced a closely related model known as $G(n,m)$. 
This model has two parameters, the number of vertices 
$n$ and the number of edges $m$, where $0 \leq m \leq \binom{n}{2}$. 
This model assigns equal 
probability to all labelled graphs on the vertex set $[n]$ with exactly $m$ edges. In other words,

\[
 \Prob{G} =
  \begin{cases}
   \frac{1}{ { {n \choose 2} \choose m } } & \text{if } |E(G)|=m \\
   0       & \text{if } |E(G)| \neq m
  \end{cases}
\]

We shall be interested in understanding various graph theoretic properties.

\begin{definition}
Define a graph property $\mathcal{P}$ as a subset of all possible labelled graphs. 
Namely $\mathcal{P} \subseteq 2^{[n] \choose 2}$. 
\end{definition} 

For instance $\mathcal{P}$ can be the set of planar graphs or the set of graphs
that contain a Hamiltonian cycle. 
We will call a property $\mathcal{P}$ as monotone increasing if $G \in \mathcal{P}$ 
implies $G+e \in \mathcal{P}$. For instance the Hamiltonian property is monotone 
increasing. 
Similarly, we will call a property $\mathcal{P}$ as monotone decreasing if $G \in \mathcal{P}$ 
implies $G-e \in \mathcal{P}$.
For instance the planarity property is monotone decreasing. Since there is an underlying probability distribution, we shall be interested
in how likely a property is. 
Our claims relating to random graphs will be probabilistic. 
We will say that an event $A_n$ holds with high probability (\whp)
if $\lim_{n \rightarrow +\infty} \Prob{A_n}=1$.
The following are background definitions which can be found in any random graph theory textbook  
\cite{bollobas2001random,janson2000random}. 

\noindent Notice that in the $G(n,p)$ model we toss a coin independently for each possible edge 
and with probability $p$ we add it to the graph. 
In expectation there will be $p{n \choose 2}$ edges.  When $p=\frac{m}{{n \choose 2}}$,
then a random binomial graph has in expectation  $m$ edges and intuitively $G(n,p)$ and $G(n,m)$ 
should behave similarly. The following theorem 
quantifies this intuition.

\begin{theorem}
Let $0 \leq p_0 \leq 1, s(n) =n \sqrt{p(1-p)} \rightarrow +\infty$, and $\omega(n) \rightarrow +\infty$ 
as $n \rightarrow +\infty$.  Then, 

(a) if $\mathcal{P}$ is any graph property and for all $m \in \field{N}$ such that 
$|m- {n \choose 2}p| <\omega(n)s(n)$, the probability $\Prob{G(n,m) \in \mathcal{P} } \rightarrow p_0$,
then $\Prob{G(n,p) \in \mathcal{P} } \rightarrow p_0$ as $n \rightarrow +\infty$.

(b) if $\mathcal{P}$ is a monotone graph property and $p_{-} = p_0 - \frac{\omega{n}s(n)}{n^3}$, 
$p_+ = p_0 +\frac{\omega{n}s(n)}{n^3}$ then from the facts that 
$\Prob{ G(n,p_{-}) \in \mathcal{P} } \rightarrow p_0, \Prob{ G(n,p_{+}) \in \mathcal{P} } \rightarrow p_0$,
it follows that $\Prob{ G(n,p{n \choose 2}) \in \mathcal{P} } \rightarrow p_0$ as $n \rightarrow +\infty$.
\end{theorem}

\subsubsection{Configuration model and Random Regular Graphs} 

The configuration model can construct a multigraph in general with
a given degree sequence $d=(d_1,\ldots,d_n)$. 
We describe the configuration for random regular graphs \cite{wormald1999models},
as we use it in Chapter~\ref{rainbowchapter}. 
We follow the configuration model of Bollob\'as \cite{bollobas2001random} 
in our proofs, see \cite{janson2000random}  for further details.
Let $W=[2m=rn]$ be our set of {\em configuration points} and let $W_i=[(i-1)r+1,ir]$,
$i\in [n]$, partition $W$. The function $\f:W\to[n]$ is defined by
$w\in W_{\f(w)}$. Given a pairing $F$ (i.e. a partition of $W$ into $m$ pairs) we obtain a
(multi-)graph $G_F$ with vertex set $[n]$ and an edge $(\f(u),\f(v))$ for each
$\{u,v\}\in F$. Choosing a pairing $F$ uniformly at random from
among all possible pairings $\Omega_W$ of the points of $W$ produces a random
(multi-)graph $G_F$. 
Each $r$-regular simple graph $G$ on vertex set $[n]$ is equally likely to be generated as $G_F$.
Here, simple means without loops of multiple edges. 
Furthermore, if $r=O(1)$ then $G_F$ is simple with a probability bounded below by a positive value independent of $n$.
Therefore, any event that occurs \whp\ in $G_F$ will also occur \whp\ in $G(n,r)$.

\subsubsection{Preferential attachment}

The configuration model can be used to generate graphs with power law degree distributions.
Assuming that $d=(d_1,\ldots,d_n)$ is a graphical sequence following a power law distribution, 
each vertex $i$ obtains $d_i$ configuration points. Then a pairing $F$ results in a multigraph
with degree sequence $d$. However, this mechanism does not provide any insights on how a dynamic
network that evolves over time exhibits a power law degree distribution. 

Albert-L\'{a}szl\'{o} Barab\'{a}si and R\'eka Albert in a highly influential paper \cite{albert}
provide a dynamic mechanism that explains how power law degree sequences emerge in 
real-world networks. We present a generalized version of their model with 
two parameters $m$ and $\delta \geq -1$
as presented in  \S 8.1 in  \cite{van2009random}. The original version \cite{albert} is a subcase 
of this general model for $\delta=0$. 
The model generates a sequence of graphs which we denote by $\{ PA_t(m,\delta)\}_{t=1}^{+\infty}$
which for every $t$ yields a graph with $t$ vertices and $mt$ edges. We define the model for $m=1$. 
The model  $PA_t(m,\delta)$ where $m\geq 2$ is reduced to the case $m=1$ by running the model 
$ PA_{mt}(1,\delta)$ and collapsing sequences of $m$ consecutive vertices. 
Let $\{v_1,\ldots,v_t \}$ be the set of vertices of $PA_t(m,\delta)$
and let $d_i(t)$ be the degree of vertex $v_i$ at time $t$. 
Initially at time 1 the graph $PA_1(1,\delta)$ consists of a single vertex with a loop. 
The growth follows the following preferential rule. 
To obtain $PA_{t+1}(1,\delta)$ from $PA_t(1,\delta)$ a new vertex with 
a single edge is added to the graph. This edge chooses its second endpoint 
according to the following probability distribution. 
With probability $\frac{d_i(t)+\delta}{t(2+\delta)+(1+\delta)}$ 
vertex $v_i$ is chosen, where $i \in [t]$, and with the remaining 
probability $\frac{1+\delta}{t(2+\delta)+(1+\delta)}$ a self loop is created.

Define $p_k =\Big(2+\frac{\delta}{m} \Big)   \frac{ \Gamma(k+\delta) \Gamma(m+2+\delta+\tfrac{\delta}{m}) }{\Gamma(m+\delta) \Gamma(k+3+\delta+\tfrac{\delta}{m}) }$
for $k\geq m$ where $\Gamma(t)= \int_{0}^{+\infty} x^{t-1}e^{-x}dx$ is the $\Gamma$ function 
and $P_k(t) = \frac{1}{t} \sum_{i=1}^t  \mathbf{1}(d_i(t)=k)$. Also, let $N_k(t)=tP_k(t)$. 
It turns out that $\Mean{N_k(t)} \approx p_k t$ and that the degree sequence 
is strongly concentrated around its expectation. Specifically, for any $C > m\sqrt{8} $ as $t\rightarrow +\infty$ the following
concentration inequality holds. 

$$ \Prob{ \max_k | N_k(t) - \Mean{N_k(t)}| \geq C\sqrt{t\log t} }=o(1).$$    

\noindent For the special case $\delta=0$,

$$ p_k = \frac{2m(m+1)}{k(k+1)(k+2)},$$ 

namely the probability distribution follows a power law with slope 3, as $p_k\sim k^{-3}$.
It is worth mentioning that the model of preferential attachment was introduced conceptually
by  Barab\'{a}si and  Albert but it was Bollob{\'a}s and Riordan with their collaborators 
who formally defined and studied the model \cite{bollobasriordan,bollobasdegrees}. 
Power law degree sequences can also emerge with growth models based on optimization \cite{fabrikant}. 


\subsubsection{Kronecker graphs} 

Kronecker graphs \cite{Leskovec05Realistic,kronecker} are 
inspired by fractal theory \cite{mandelbrot1977fractals}. There exist two versions 
of Kronecker graphs, a deterministic and a randomized one. To define each one, we remind 
the definition of the Kronecker product. 

\begin{definition} 
Given two matrices $A^{m \times n}=(a_{ij}),B^{m' \times n'}=(b_{ij})$, the Kronecker product 
matrix $C^{mm'\times nn'}$ is given by 

\[ C = A \otimes B = \left( \begin{array}{ccc}
a_{11}B & \ldots & a_{1n}B \\
\vdots &  \ddots &  \vdots \\
a_{m1}B & \ldots& a_{mn}B\end{array} \right).\] 
\end{definition}

Deterministic Kronecker graphs are defined by a small initiator adjacency matrix 
$K_1$ and the order $k$. A deterministic Kronecker graph of order $k$ is defined by 

$$ K_1^{(k)} = \underbrace{ K_1 \otimes \ldots \otimes K_1}_{\text{k~times}}.$$

The stochastic Kronecker graphs use an initiator matrix with probabilities. The final 
adjacency matrix is the outcome of a randomized rounding of the $k$-th order 
Kronecker product of the initiator matrix. 
Kronecker graphs match several empirical properties 
such as   heavy-tailed degree distributions and triangles,  low diameters, and also 
obeys the densification power law. Most properties are analyzed
in the deterministic case \cite{leskovec2010kronecker,Tsourakakis:2008:FCT:1510528.1511415}. 
Mahdian and Xu in an elegant paper studied stochastic Kronecker graphs. 
They show a phase transition for the emergence of the giant component and 
for connectivity, and prove that such graphs have constant diameters beyond 
the connectivity threshold  \cite{mahdian2007stochastic}.
Two additional appealing features of Kronecker graphs is the existence
of methods to fit the parameters of the $2 \times 2$ initiator matrix 
to a given graph and their generation is embarassingly parallel.

Other popular models are the copying model \cite{kleinberg1999web,kumar2000stochastic}, 
the Cooper-Frieze model \cite{cooperfrieze}, 
the Aiello-Chung-Lu model \cite{aiello,aiello2000random}, protean graphs \cite{protean}
the Fabrikant-Koutsoupias-Papadimitriou model \cite{fabrikant},
and the forest fire model \cite{1217301}.


\subsection{Triangles} 
\label{sec:trianglestriangles}
Subgraphs play a central role in graph theory. 
It is not an exaggeration to claim that studying 
subgraphs has been an active thread of research 
since the early days of graph theory:
Euler paths and cycles, Hamilton paths and cycles, 
matchings, cliques, neighborhoods of vertices typically 
sketched via the degree sequence are special types of subgraphs.

Among various subgraphs, triangles play a 
major role in network analysis. A triangle is a clique of order 3. 
The number of triangles in a graph is a computationally expensive, crucial graph statistic in complex network 
analysis, in random graph models and in various important applications. Despite the fact 
that real-world networks tend to be sparse in edges, they are dense in triangles. 
This observation implies that when two vertices share a common neighbor, then it is 
more likely that they are/become connected. For instance, it has been observed in the 
MSN Messenger social network that if two people have a common contact it is 18\,000 times 
more likely that they are connected \cite{achyou}. The transitivity of adjacency is striking in social networks
and in other types of networks too.
There exist two processes that 
generate triangles in a social network: {\it homophily} and {\it transitivity}.
According to the former, people tend to choose friends with similar characteristics to themselves
(e.g., race, education) \cite{citeulike:8417412,wasserman_faust94} and according to the latter friends 
of friends tend to become friends themselves \cite{wasserman_faust94}. 
We survey a wide range of applications which rely on the number of triangles in a given graph.

\spara{Clustering Coefficients and Transitivity of a Graph:} 
Despite the fact that Erd\"{o}s-R\'enyi graphs have a short diameter 
they do not model social networks well. Social networks 
have many triangles. This was the main motivation of 
Watts and Strogatz \cite{j1998collective} in their 
influential paper to introduce clustering coefficients and 
the notion of transitivity which we defined in a previous section.

\spara{Uncovering Hidden Thematic Structures:} Eckmann and Moses \cite{eckmann2002curvature} propose the use of the clustering coefficient 
for detecting subsets of web pages with a common topic. 
The key idea is that reciprocal links between pages indicate a mutual recognition/respect 
and then triangles due to their transitivity properties can be used
to extend  ``seeds'' to larger subsets of vertices with similar thematic structure in the Web graph. 
In other words, regions of the World Wide Web with high curvature 
indicate a common topic, allowing the authors to extract useful meta-information. 
This idea has found more applications, in natural language processing
\cite{dorow} and in bioinformatics \cite{DBLP:journals/bmcbi/RougemontH03,Kalna:2007:CCW:1365534.1365536}.

\spara{Exponential Random Graph Model:} Frank and Strauss \cite{citeulike:1665611} proved under the assumption that two edges are dependent
only if they share a common vertex that the sufficient statistics for Markov graphs 
are the counts of triangles and stars. 
Wasserman and Pattison \cite{citeulike:1692025} proposed the exponential random graph (ERG) model  
which generalized the Markov graphs \cite{robins07introduction}. 
Triangles are frequently used as one of the sufficient statistics of the ERG model 
and counting them is necessary for parameter estimation, e.g., using Markov chain Monte Carlo (MCMC) procedures \cite{bhamidi}.

\spara{Spam Detection:} Becchetti et al. \cite{1401898} show that the distribution of triangles among
spam hosts and non-spam hosts can be used as a feature for classifying 
a given host as spam or non-spam. The same result holds also 
for web pages, i.e., the spam and non-spam triangle distributions differ at a detectable level
using standard statistical tests from each other.

\spara{Content Quality and Role Behavior Identification:} Nowadays, there exist many online 
forums where acknowledged scientists participate, e.g., 
MathOverflow, CStheory Stackexchange and discuss problems of their fields.
This yields significant information for researchers.  
Several interesting questions arise such as which participants comment on each other. 
This question including several others were studied in \cite{wesler}.
The number of triangles that a user participates was shown to play a critical role in answering 
these questions. For further applications in assesing the role behavior of users see \cite{1401898}.

\spara{Structural Balance and Status Theory:} Balance theory appeared first in Heider's seminal work \cite{heider} 
and is based on the concept ``the friend of my friend is my friend'', 
``the enemy of my friend is my enemy'' etc. \cite{wasserman_faust94}. 
To quantify this concept edges become signed, i.e., there is a function $c: E(G) \rightarrow \{+,-\}$. 
If all triangles are positive, i.e., the product of the signs of the edges is $+$, then the graph is balanced. 
Status theory is based on interpreting a positive edge $(u,v)$ 
as $u$ having lower status than $v$, while the negative edge $(u,v)$
means that $u$ regards $v$ as having a lower status than himself/herself.
Recently, Leskovec et al.\cite{Leskovec:2010:PPN:1772690.1772756} have performed experiments
to quantify which of the two theories better apply 
to online social networks and predict signs of incoming links. 
Their algorithms require counts of signed triangles in the graph.

\spara{Microscopic Evolution of networks:} Leskovec et al. \cite{Leskovec:2008:MES:1401890.1401948}
present an extensive experimental study of network evolution
using detailed temportal information. One of their findings is that as edges arrive in the network, 
they tend to close triangles, i.e., connect people with common friends.

\spara{Community Detection: }  Counting triangles is important as subroutine
in community detection algorithms. Berry et al. use triangle counting to deduce the edge
support measure in their community detection algorithm \cite{berry2011tolerating}. 
Gleich and Seshadhri \cite{Gleich-2012-neighborhoods} 
show that heavy-tailed degree distributions and abundance in triangles 
imply that there exist vertices which together with their neighbors
form a low-conductance set, i.e., community. 

\spara{Motif Detection:} Triangles are abudant not only in social networks but in biological networks  \cite{citeulike:307440,Yook04functional}.
This fact can be used e.g., to correlate the topological and functional properties of protein interaction networks \cite{Yook04functional}.

\spara{Triangular Connectivity \cite{batagelj2007short}:} Two vertices $u,v$ are triangularly connected
if there is a sequence of triangles $(\Delta_1,\ldots,\Delta_s)$ 
such that $u$ is a vertex in the $\Delta_1$,
$v$ in $\Delta_s$ and $\Delta_i$ shares at least one vertex with $\Delta_{i-1}$.

\spara{$k$-truss:} The $k$-truss of a graph $G$ \cite{cohen2009graph} is the maximum subgraph
of $G$ where every edge appears in at least $k-2$ triangles. 

\spara{Link recommendation:} Triangle listing is used in link recommendation
\cite{tsourakakis2009spectral,tsourakakis2011spectral}.

\spara{CAD applications:} Fudos and Hoffman \cite{fudos} introduced a graph-constructive approach to solving systems of geometric constraints,
a problem which arises frequently in Computer-Aided design (CAD) applications. 
One of the steps of their algorithm computes the number of triangles in an appropriately defined graph.

Given the large number of applications, there exists a lot of interest in developing 
efficient triangle listing and counting algorithms. 

\subsubsection{Triangle counting methods}
\label{sec:subgraphrelated}

There exist {\it exact} and {\it approximate} triangle counting algorithms.
It is worth noting that for most of the applications described in Section~\ref{sec:subgraphintro} 
the exact number of triangles is not crucial.
Hence, approximate counting algorithms which are fast and output a high quality estimate are desirable for the 
applications in which we are interested.

\spara{Exact Counting:} Naive triangle counting by checking all triples of vertices takes $O(n^3)$ units of time. 
The state of the art algorithm is due to Alon, Yuster and Zwick \cite{739463} 
and runs in $O(m^{\frac{2\omega}{\omega+1}})$,
where currently the fast matrix multiplication exponent $\omega$ is 2.3727 
\cite{williams2011breaking}. Thus, the Alon, Yuster, Zwick (AYZ) algorithm currently runs in $O(m^{1.407})$ time.
It is worth mentioning that from a practical point of view algorithms based
on matrix multiplication are not used due to the prohibitive memory requirements.
Even for medium sized networks, i.e., networks with hundreds of thousands of edges, 
matrix-multiplication based algorithms are not applicable. 
Itai and Rodeh in 1978 showed an algorithm which finds a triangle in any graph in $O(m^\frac{3}{2})$ \cite{itairoder}. 
This algorithm can be extended to list the set of triangles in the graph with the same time complexity.
Chiba and Nishizeki showed that triangles can be found in time $O(m\alpha(G))$ where $\alpha(G)$ 
is the {\it arboricity} of the graph. Since $\alpha(G)$ is at most $O(\sqrt{m})$ their algorithm
runs in $O(m^{3/2})$ in the worst case \cite{chiba}. 
For special types of graphs more efficient triangle counting algorithms exist. 
For instance in planar graphs, triangles can be found in $O(n)$ time \cite{chiba,itairoder,papadimitriou1981clique}.

Even if listing algorithms solve a more general problem than the counting one, they
are preferred in practice for large graphs, due to the smaller memory requirements
compared to the matrix multiplication based algorithms. 
Simple representative algorithms are the node- and the edge-iterator algorithms.
The former counts for each vertex $v$ the number of triangles $t_v$ it is involved in,
i.e., the number of edges among its neighbors, whereas the latter
algorithm counts for each edge $(i,j)$ the common neighbors  of vertices $i,j$.
Both of these algorithms have the same asymptotic complexity $O(mn)$, 
which in dense graphs results in $O(n^3)$ time, the complexity of the naive counting algorithm. 
Practical improvements over this family of algorithms have been achieved using various techniques, such as 
hashing and sorting by the degree \cite{latapy,schank2005finding}.

\spara{Approximate Counting:} On the approximate counting side, most of the triangle counting algorithms have been developed in the 
streaming setting. In this scenario, the graph is represented as a stream. 
Two main representations of a graph as a stream are the edge stream and the incidence stream. In the former, edges arrive
one at a time. In the latter scenario all edges incident to the same vertex appear successively in the stream. The ordering of the vertices 
is assumed to be arbitrary. A streaming algorithm produces a $(1+\epsilon)$ approximation 
of the number of triangles \whp\ by making only a constant number of passes over the stream. 
However, sampling algorithms developed in the streaming literature can be applied in the setting where the graph fits in the memory as well. 
Monte Carlo sampling techniques have been proposed to give a fast estimate of the number of triangles.
According to such an approach, a.k.a. naive sampling \cite{paper:schank:2004}, 
we choose three nodes at random repetitively and check if they form a triangle or not. 
If one makes $$ r = \log({\frac{1}{\delta}})\frac{1}{\epsilon^2}(1+\frac{T_0+T_1+T_2}{T_3})$$
independent trials where $T_i$ is the number of triples with $i$ edges 
and outputs as the estimate of triangles the random variable $T_3'$ equaling to the 
fractions of triples picked that form triangles times the total number of
triples ${n \choose 3}$, then
$$(1-\epsilon)T_3 < T_3' < (1+\epsilon)T_3 $$
with probability at least $1- \delta$.
This is suitable only when  $T_3=o(n^2)$.

In \cite{yosseff} the authors reduce the problem of triangle counting efficiently to estimating
moments for a stream of node triples. Then, they use the Alon-Matias-Szegedy (AMS) algorithms \cite{amsalgos}   to proceed. 
The key is that the triangle computation reduces to estimating the zero-th, first and second frequency moments, which can be done efficiently. 
Furthermore, as the authors suggest their algorithm is efficient only on graphs with $\Omega(n^2/\log{\log{n}})$
triangles, i.e., triangle dense graphs as in the naive sampling. 
The AMS algorithms are also used by \cite{jowhary}, where simple sampling techniques are used, such
as choosing an edge from the stream at random and checking how many common neighbors its two endpoints share
considering the subsequent edges in the stream. 
Along the same lines, Buriol et al.  \cite{buriol} proposed two space-bounded sampling algorithms to estimate the number of triangles. 
Again, the underlying sampling procedures are simple. For instance, in the case of the edge stream representation, they sample randomly
an edge and a node in the stream and check if they form a triangle. The three-pass algorithm presented therein, counts  in the 
first pass the number of edges, in the second pass it samples uniformly at random an edge $(i,j)$ and a node $k \in V-\{i,j\}$ and in the third pass it tests
whether the edges $(i,k),(k,j)$ are present in the stream. The number of samples $r$ needed to obtain an $(1\pm \epsilon)$-approximation
with probability $1-\delta$ is 

$$ r = O\Bigg( \log{ \big(\frac{1}{\delta}\big) } \frac{T_1+2T_2+3T_3}{T_3\epsilon^2} \Bigg)  = O\Bigg( \log{ \big(\frac{1}{\delta}\big) } \frac{mn}{t\epsilon^2} \Bigg) .$$

Even if the term $T_0$ in the nominator is missing\footnote{Notice that $m(n-2)=T_1+2T_2+3T_3$ and $t=T_3$.} compared to the naive sampling, the graph has still to be fairly dense with respect
to the number of triangles in order to get  a $(1+\pm \epsilon)$-approximation \whp. 
Buriol et al. \cite{buriol} show how to turn the three-pass algorithm into a single pass algorithm 
for the edge stream representation and similarly they provide a three- and one-pass algorithm for the incidence stream representation. 
Kane et al. show how to count other subgraphs in the streaming model \cite{kane2012counting}. 
In \cite{1401898} the semi-streaming model for counting triangles is introduced, which allows $\log{n}$ passes over the edges. 
The key observation is that since counting triangles reduces to computing the intersection of two sets, namely the induced neighborhoods
of two adjacent nodes, ideas from  locality sensitivity hashing \cite{broder2008minwise} are applicable to the problem.

Another line of work is based on linear algebraic arguments. Specifically, in the case of ``power-law'' networks it was 
shown in \cite{tsourakakis1} that the spectral counting of triangles 
can be efficient due to their special spectral properties \cite{chung}.
This idea was further extended in \cite{TsourakakisKAIS} using the randomized Singular Value Decomposition (SVD) 
approximation algorithm by \cite{drineas:frieze}.
More recently, Avron proposed a new approximate triangle counting method based on a randomized algorithm
for trace estimation \cite{Haim10}.

\spara{Graph Sparsifiers:} A sparsifier of a graph $G(V,E,w)$ is a sparse graph $H$ that is similar
to $G$ in some useful notion.  We discuss in the following the Bencz{\'u}r-Karger cut 
sparsifier \cite{benczurstoc,DBLP:journals/corr/cs-DS-0207078} and 
the Spielman-Srivastava spectral sparsifier \cite{DBLP:conf/stoc/SpielmanS08}.

\textit{Bencz{\'u}r-Karger Sparsifier:}  Bencz{\'u}r and Karger introduced in \cite{benczurstoc} the notion of cut sparsification to accelerate cut algorithms whose running time
depends on the number of edges. Using a non-uniform sampling scheme they show that given a graph $G(V,E,w)$ with $|V|=n, |E|=m$ 
and a parameter $\epsilon$ there exists a graph $H(V,E',w')$ with $O(n\log{(n)}/\epsilon^2)$ edges such that the weight
of every cut in $H$ is within a factor of $(1\pm \epsilon)$ of its weight in $G$. Furthermore, they provide a nearly-linear 
time algorithm which constructs such a sparsifier. 
The key quantity used in the sampling scheme of Bencz{\'u}r and Karger is the strong connectivity $c_{(u,v)}$ of an edge $(u,v) \in E$
\cite{benczurstoc,DBLP:journals/corr/cs-DS-0207078}. The latter quantity is defined to be the maximum value $k$
such that there is an induced subgraph $G_0$ of $G$ containing both $u$ and $v$, and every cut in $G_0$ has weight at least $k$.

\textit{Spielman-Srivastava Sparsifier:} In \cite{DBLP:conf/stoc/SpielmanS08} Spielman and Teng introduced the notion 
of a spectral sparsifier in order to strengthen the notion of a cut sparsifier.  
A quantity that plays a key role in spectral sparsifiers is the {\em effective resistance}. 
The term effective resistance comes from electrical network analysis, see Chapter IX in \cite{bollobas}. 
In a nutshell, let $G(V,E,w)$ be a weighted graph with vertex set $V$, edge set $E$ and weight function $w$. 
We call the weight $w(e)$ {\it resistance} of the edge $e$. We define the {\it conductance} $r(e)$ of $e$ 
to be the inverse of the resistance $w(e)$. Let $\mathcal{G}$ be the resistor network constructed from 
$G(V,E,w)$  by replacing each edge $e$ with an electrical resistor whose electrical resistance is $w(e)$. 
Typically, in $\mathcal{G}$  vertices are called {\it terminals}, a convention that emphasizes the 
electrical network perspective of a graph $G$.
The {\em effective resistance} $R(i,j)$ between two vertices $i,j$ is the electrical resistance measured 
across vertices $i$ and $j$ in $\mathcal{G}$. Equivalently, the effective resistance is the potential difference 
that appears across terminals $i$ and $j$ when we apply a unit current source between them. 
Finally, the {\em effective conductance}  $C(i,j)$ between two vertices $i,j$ is defined  as $C(i,j) = R^{-1}(i,j)$. 

Spielman and Srivastava in their seminal work \cite{DBLP:conf/stoc/SpielmanS08} proposed to include 
each edge of $G$ in the sparsifier $H$ with probability proportional to its effective resistance. 
They provide  a nearly-linear time algorithm that produces spectral sparsifiers with $O(n\log{n})$ edges.

\subsection{Dense Subgraphs}
\label{sec:denserelatedch1}

Finding dense subgraphs is a key problem for many applications 
and the key primitive for community detection. Here we review some
important concepts of dense subgraphs. 

\spara{Cliques:}
A clique is a subset of vertices all connected to each other.
The problem of finding whether there exists a clique of a given size in a graph is \NP-complete.
A \emph{maximum} clique of a graph is a clique of maximum possible size
and its size is called the graph's clique number.
H$\mathring{a}$stad~\cite{hastad} shows that, unless $\mathbf{P} = \mathbf{NP}$,
there cannot be any polynomial time algorithm that approximates the
maximum clique within a factor better than $\bigO(n^{1-\epsilon})$, for any $\epsilon > 0$.
Feige~\cite{feige} proposes a polynomial-time algorithm
that finds a clique of size $\bigO( (\frac{\log{n}}{\log\log{n}})^2 \big)$
whenever the graph has a clique of size $\bigO(\frac{n}{\log^b{n}})$ for any constant $b$.
Based on this, an algorithm that approximates the maximum clique problem within a factor of $\bigO( n \frac{(\log\log{n})^2}{\log{n}^3}\big)$ is also defined.
A \emph{maximal} clique is a clique that is not a subset of any other clique.
The Bron-Kerbosch algorithm~\cite{BronKerbosch} finds all maximal cliques in a graph in exponential time.
A near optimal time algorithm for sparse graphs was introduced in~\cite{Eppstein}.

\spara{Densest Subgraph:}
Let $G(V,E)$ be a graph, $|V|=n$, $|E|=m$. The average degree of a vertex set $S \subseteq V$ is defined as
$\frac{2e[S]}{|S|}$, where $e[S]$ is the number of edges in the induced graph $G[S]$.
The densest subgraph problem is to find a set $S$ that maximizes the average degree.
The densest subgraph can be identified in polynomial time by solving a maximum-flow problem
\cite{GGT89,Goldberg84}.
Charikar~\cite{Char00} shows that the greedy algorithm proposed by Asashiro et al.
\cite{AITT00} produces a $\frac{1}{2}$-approximation of the densest subgraph in linear time.
Both algorithms are efficient in terms of running times and scale to large networks.
In the case of directed graphs, the densest subgraph problem is solved in polynomial
time as well. Charikar~\cite{Char00} provided a linear programming approach
which requires the computation of $n^2$ linear programs
and a $\frac{1}{2}$-approximation algorithm which runs in $\bigO(n^3+n^2m)$ time.
Khuller and Saha~\cite{Khuller} improved significantly the state-of-the art by providing
an exact combinatorial algorithm and a fast $\frac{1}{2}$-approximation algorithm
which runs in $\bigO(n+m)$ time.
Kannan and Vinay~\cite{Kannan} gave a spectral $\bigO(\log{n})$ approximation algorithm for a related
notion of density.

In the classic definition of densest subgraph there is no size restriction of the output.
When restrictions on the size $|S|$ are imposed, the problem becomes {\NPhard}.
Specifically, the {\it DkS} problem of finding the densest subgraph of $k$ vertices is known to be {\NPhard}~\cite{AHI02}.
For general $k$, Feige et al.~\cite{FPK01} provide an approximation guarantee
of $\bigO(n^{\alpha})$, where $\alpha < \frac{1}{3}$.
The greedy algorithm by Asahiro et al.~\cite{AITT00} gives instead an approximation factor of $\bigO(\frac{n}{k})$.
Better approximation factors for specific values of $k$ are provided by algorithms based on semidefinite programming~\cite{FL01}.
From the perspective of (in)approximability, Khot~\cite{Khot06} shows that there cannot exist any PTAS for the {\it DkS} problem under a reasonable complexity assumption.
Arora et al.~\cite{AKK95} propose a PTAS for the special case $k=\Omega(n)$ and $m=\Omega(n^{2})$.
Finally, two variants of the {\it DkS} problem are introduced by Andersen and Chellapilla~\cite{AndersenChellapilla}.
The two problems ask for the set $S$ that maximizes the average degree subject to $|S| \leq k$  (DamkS) and  $|S| \geq k$ (DalkS), respectively.
The authors provide constant factor approximation algorithms for both DamkS and DalkS.

\spara{Quasi-cliques:}
A set of vertices $S$ is an $\alpha$-quasi-clique (or pseudo-clique) if  $e[S] \geq \alpha {|S| \choose 2}$, i.e.,
if the edge density of the induced subgraph $G[S]$ exceeds a threshold parameter $\alpha \in (0,1)$.
Similarly to cliques, \emph{maximum} quasi-cliques~\cite{MaximumQuasiCliques}
and \emph{maximal} quasi-cliques~\cite{MaximalQuasiCliques} are quasi-cliques
of maximum size and quasi-cliques not contained into any other quasi-clique, respectively.
Abello et al.~\cite{abello} propose an algorithm for finding a single maximal
$\alpha$-quasi-clique, while Uno~\cite{uno} introduces an algorithm to enumerate
all $\alpha$-quasi-cliques.

\spara{$k$-core, $k$-clubs, $kd$-cliques:}  
A $k$-core  is a maximal connected subgraph in which all vertices have degree at least $k$.
There exists a linear time algorithm for finding $k$-cores by repeatedly removing the vertex having the smallest degree~\cite{batagelj}.
A $k$-club is a subgraph whose diameter is at most $k$~\cite{mokken}. $kd$-cliques differ from $k$-clubs as
 the shortest paths used to compute the diameter of a $kd$-clique are allowed to use vertices not belonging to that $kd$-clique.
All these clique variants are clearly conceptually different from the \OQCs\ we study in this paper.

\subsection{Graph Partitioning}

Graph partitioning is a fundamental computer science problem. 
As we discussed above, the problem of finding communities 
is reduced to understanding the cut structure of the graph. 
In distributed computing applications, 
the following version of the graph partitioning problem 
plays a key role. 
The interested reader may consult the cited work and the references therein
for more information on the balanced graph partitioning problem.

{\bf Balanced graph partitioning:}  The {\em balanced graph partitioning} 
problem is a classic \NPhard problem 
of fundamental importance to parallel and distributed computing  \cite{Garey:1974:SNP:800119.803884}.
The input of this problem is an undirected graph $G(V,E)$ and an integer $k \in \field{Z}^+$, the 
output is a partition of the vertex set in $k$ balanced parts such that the number of edges across the clusters
is minimized.
Formally, the balance constraint is defined by the imbalance parameter $\nu$. 
Specifically, the $(k,\nu)$-balanced graph partitioning 
asks to divide the vertices of a graph in $k$ clusters each of size 
at most $\nu \frac{n}{k}$, where $n$ is the number of vertices in $G$. 
The case $k=2,\nu=1$ is equivalent to the \NPhard   minimum bisection problem.
Several approximation algorithms, e.g., \cite{Feige:2002:PAM:586842.586910},
and heuristics, e.g., \cite{Fiduccia:1982:LHI:800263.809204,kl} exist for this problem. 
When $\nu=1+\epsilon$ for any desired but fixed $\epsilon$ there exists a
$O(\epsilon^{-2} \log^{1.5}{n})$ approximation algorithm \cite{Andreev:2004:BGP:1007912.1007931}. 
When $\nu = 2$ there exists an $O(\sqrt{\log{k}\log{n}})$ approximation algorithm
based on semidefinite programming (SDP) \cite{krauthgamer}.  
Due to the practical importance of $k$-partitioning there exist several heuristics, 
among which \metis \cite{schloegel} and its parallel version~\cite{schloegel_parmetis}
stand out for their good performance. 
\metis is widely used in many existing systems~\cite{karagiannis}.
There are also heuristics that improve efficiency and partition quality of \metis 
in a distributed system~\cite{satuluri}. 

{\bf Streaming balanced graph partitioning:}  

Despite the large amount of work on the balanced graph partitioning problem, 
neither state-of-the-art approximation algorithms nor heuristics 
such as \metis are well tailored to the computational restrictions that the size of today's graphs impose. 
Motivated by this fact, Stanton and Kliot introduced the streaming balanced graph partitioning problem, 
where the graph arrives in the stream and decisions about the partition need to be taken with  on the fly quickly~\cite{stanton}. 
Specifically, the vertices of the graph arrive in a stream with the set of edges incident to them. 
When a vertex arrives, a partitioner decides  where to place the vertex.
A vertex is never moved after it has been assigned to one of the $k$ machines. 
A realistic assumption that can be used in real-world streaming graph partitioners is the existence
of a small-sized buffer. Stanton and Kliot evaluate partitioners with or without buffers. 
The work of \cite{stanton} can be   adapted to edge streams.
Stanton showed that streaming graph partitioning algorithms with a single pass 
under an adversarial stream order cannot approximate the optimal cut size within
$o(n)$. The same bound holds also for random stream orders \cite{stantonstreaming}.
Finally,  Stanton \cite{stantonstreaming} analyzes two variants of well performing algorithms 
from~\cite{stanton} on random graphs. Specifically, she proves that if the graph 
$G$ is sampled according to  the planted partition model, then the two algorithms
despite their similarity can perform differently and that 
one of the two can recover the true partition \whp, assuming  that 
inter-, intra-cluster edge probabilities are constant, and their gap is a large constant.

We conclude our brief exposition by outlining the differences between community
detection methods and the balanced partitioning problem. 
One main difference is the lack of restriction on the number of vertices per subset
in the community detection problem. 
A second difference is that in realistic applications the number
of clusters in the balanced partitioning problem is part of the input,
as it represents the number of machines/clusters available to distribute the graph. 
In community detection the number of clusters is not known a priori, 
or even worse, their existence is not clear.
It is worth mentioning at this point that in Chapters~\ref{densestchapter} and \ref{fennelchapter} we introduce measures
conceptually close to the modularity measure  \cite{girvan2002community,newman2004finding,newman}. 
Despite the popularity of modularity, few rigorous results exist. 
Specifically,  Brandes et al. proved that maximizing modularity is \NPhard \cite{brandes2007finding}. 
Approximation algorithms without theoretical guarantees whose performance 
is evaluated in practice  also exist \cite{kempe}.

\subsection{Big Graph Data Analytics}

Except for the algorithmic `dasein' of computer science, there is an engineering one too.
An important engineering law is Moore's law. 
Gordon Moore based on observations from 1958 until 1965 extrapolated that the number 
of components in integrated circuits would keep doubling for at least until 1975 \cite{moore1998cramming}.
It is remarkable that Moore's prediction remains (more or less) valid since then. 
However, as we are approaching the end of its validity, it is becoming clear that in order 
to perform demanding computational tasks, we need more than one machine. 
At the same time, input size increases. Currently, the growth rate
is unprecedented. Eron Kelly, the general manager of product marketing for Microsoft SQL Server, 
predicts that as humankind we will generate    more data as humankind 
than we generated in the previous 5,000 years \cite{techmsr}. 
The term {\em big data}  describes collections of large and complex datasets 
which are difficult to manipulate and process using traditional tools. 
Mainly, for these two reasons, i.e., hardware reaching its limits and big data, 
parallel and distributed computing are the {\it de facto} solutions 
for processing large scale data.  For this reason, there exists a lot of interest
in developing efficient graph processing systems. Popular graph processing platforms are 
Pregel \cite{malewicz} and its open-source version Apache Giraph that build on MapReduce , 
and GraphLab \cite{DBLP:conf/uai/LowGKBGH10}. 
It is worth mentioning that for dynamic graphs 
there exist other platforms which are suitable for stream/micro-batch processing, 
such as Twitter's Storm \cite{stormtwitter}. 

In the following we discuss the details of \mapreduce \cite{dean}, 
which we use in this dissertation as the underlying distributed system to develop
efficient large-scale graph processing algorithms and systems.

\subsubsection{\mapreduce Basics}
\label{sec:mapreducebasics} 

While the PRAM model \cite{jaja} and the bulk-synchronous parallel model (BSP) \cite{valiant1990bridging}
are powerful models, \mapreduce has largely ``taken over'' both industry and academia \cite{hadoopusers}. In few words, 
this success is due to two reasons: first, \mapreduce is a simple and powerful programming model which makes the programmer's 
life easy.  Secondly, \mapreduce is publicly available via its open source version \hadoop. \mapreduce was introduced in 
\cite{dean} by Google, one of the largest users of multiple processor computing in the world, for facilitating the 
development of scalable and fault tolerant applications. In the \mapreduce paradigm, a parallel computation is 
defined on a set of values and consists of a series of map, shuffle and reduce steps. Let $(x_1,\ldots, x_n)$ be the set 
of values, $m$ denote the mapping function which takes a value $x$ and returns a pair of a key $k$ and a value $u$ and $r$ 
the reduce function.

\begin{enumerate} 
 \item In the map step a mapping function $m$ is applied to a value $x_i$ and a pair $(k_i,u_i)$ of a 
 key $k_i$ and a value $u_i$ is generated.
 \item The shuffle step starts upon having mapped all values $x_i$ for $i=1$ to $n$ to pairs. 
  In this step, a set of lists is produced using the key-value pairs generated from the map 
  step with an important feature. Each list is characterized by the key $k$ and has the form $L_k = \{ k : u_1,..,u_{j(k)} \}$ 
  if and only if there exists a pair $(k,u_i)$ for $i=1$ to $j$. 
 \item Finally in the reduce step, the reduce function $r$ is applied to the lists 
   generated from the shuffle step to produce the set of values $(w_1,w_2,\ldots)$. 
\end{enumerate}

To illustrate the aforementioned abstract concepts consider the problem of counting how 
many times each word in a given document appears. 
The set of values is the ``bag-of-words'' appearing in the document. 
For example, if the document is the sentence ``The dog runs in the forest'', 
then $\{x_1,x_2,x_3,x_4,x_5,x_6\}= \{$ the, dog, runs, in, the, forest$\}$. 
One convenient choice for the \mapreduce functions is the following and results in the following steps:
The map function m will map a value x to a pair of a key and a value. A convenient 
choice for $m$ is something close to the identity map. Specifically, we choose 
$m(x) =(x,\$)$, where we assume that the dollar sign $\$$ an especially reserved symbol.
The shuffle step for our small example will produce the following set of lists: 
(the:$,$), (dog:$\$$), (runs:$\$$), (in:$\$$), (runs:$\$$), (forest:$\$$) 
The reduce function $r$ will process each list defined by each different word appearing 
in the document by counting the number of dollar signs $\$$.  
This number will also be the count of times that specific word appears in the text. 

\hadoop implements \mapreduce and was originally created by Doug Cutting. 
Even if \hadoop is well known for \mapreduce it is actually a collection of subprojects that are 
closely related to distributed computing. For example HDFS (\hadoop filesystem) is a distributed 
filesystem that provides high  throughput access to application data and HBase is a scalable, 
distributed database that supports structured data storage for large tables (column-oriented database). 
Another subproject is Pig, which is a high-level data-flow language and execution framework for parallel 
computation \cite{gates2009building}.  Pig runs on HDFS and \mapreduce. For more details and other subprojects, 
the interested reader can visit the website that hosts the \hadoop project \cite{hadoopusers}.

\section{Computational Cancer Biology} 
\label{subsec:introcancerdata}

Human cancer is caused by the accumulation of genetic alternations in cells 
\cite{michor,weinberg}.
It is a complex phenomenon often characterized by the 
successive acquisition of combinations of genetic aberrations that
result in malfunction or disregulation of genes.  
Finding driver genetic mutations, i.e., mutations which confer growth 
advantage on the cells carrying them 
and have been positively selected during the evolution of the cancer
and uncovering their temporal sequence have been central goals of 
cancer research the last decades \cite{nature}.
In this Section we review three problems arising in computational 
cancer biology. In Section~\ref{subsec:preprocessacgh} 
we present background on a data denoising problem.
In Section~\ref{subsec:oncotreeshere} we review 
oncogenetic trees, a popular model for oncogenesis.
Finally, in Section~\ref{subsec:unmixinghere} we 
discuss the problem of discovering cancer subtypes.

\subsection{Denoising array-based Comparative Genomic Hybridization (aCGH) data}
\label{subsec:preprocessacgh}

There are many forms of chromosome aberration that can contribute to cancer development,
including polyploidy, aneuploidy, interstitial deletion, 
reciprocal translocation, non-reciprocal translocation, as well as amplification,
again with several different types of the latter 
(e.g., double minutes, HSR and distributed insertions \cite{albertson}).  Identifying the specific
recurring aberrations, or sequences of aberrations, 
that characterize particular cancers provides important clues about the genetic basis of
tumor development and possible targets for diagnostics or therapeutics.  
Many other genetic diseases are also characterized by
gain or loss of genetic regions, such as Down Syndrome 
(trisomy 21)~\cite{downsyndrome}, Cri du Chat (5p deletion)~\cite{criduchat}, and
Prader-Willi syndrome (deletion of 15q11-13)~\cite{praderwilli} 
and recent evidence has begun to suggest that inherited copy number
variations are far more common and more important to human 
health than had been suspected just a few years ago~\cite{CNVs}.  These facts have
created a need for methods for assessing DNA copy 
number variations in individual organisms or tissues.

In Chapter~\ref{acghchapter}, we focus specifically on array-based comparative genomic 
hybridization (aCGH) \cite{bignell,pollack,Kallioniemi:1992,Genomic04high}, a method
for copy number assessment using DNA microarrays that remains, 
for the moment, the leading approach for high-throughput typing of copy number
abnormalities.  The technique of aCGH is schematically 
represented in Figure~\ref{fig:acghfig1}.  A test and a 
reference DNA sample are
differentially labeled and hybridized to a microarray 
and the ratios of their fluorescence intensities is measured for each spot. A typical
output of this process is shown in Figure~\ref{fig:acghfig1} (3), where the genomic profile of the cell line GM05296 \cite{coriell} is
shown for each chromosome. The $x$-axis corresponds to 
the genomic position and the $y$-axis corresponds to a noisy measurement of the ratio
$\log_2{ \frac{T}{R}}$ for each genomic position, typically referred to as ``probe'' by biologists.  For healthy diploid organisms, $R$=2 and $T$ is the DNA copy number we want to infer from
the noisy measurements. For more details on the 
use of aCGH to detect different types of  chromosomal aberrations, see \cite{albertson}.

\begin{figure} 
\centering
\includegraphics[width=.7\textwidth]{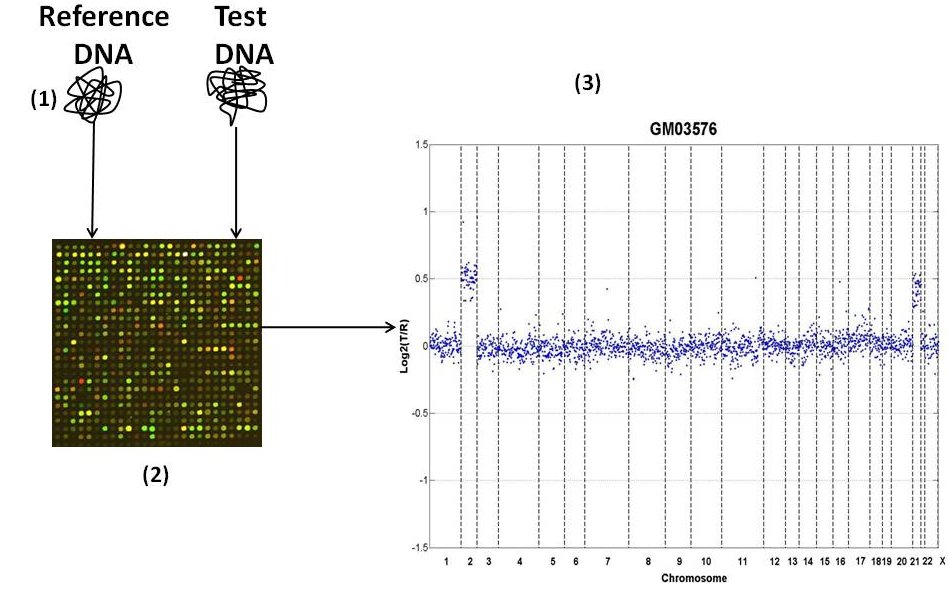}
\caption{Schematic representation of array CGH.  Genomic DNA from two cell populations (1)
is differentially labeled and hybridized in a microarray (2). Typically the reference 
DNA comes from a normal subject. For humans  this means that the reference DNA comes
from a normal diploid genome. The ratios on each spot 
are measured and normalised so that the median $\log_2$ ratio is zero. 
The final result is an ordered tuple  containing values 
of the fluorescent ratios  in each genomic position per each chromosome. This is shown in (3)
where we see the genomic profile of the cell line GM05296 \cite{coriell}. 
The problem of denoising array CGH  data is to infer the true DNA copy number $T$ 
per genomic position from a set of noisy measurements of the quantity $\log_2{ \frac{T}{R}}$,
where $R$=2 for normal diploid humans. }
\label{fig:acghfig1}
\end{figure}

Converting raw aCGH log fluorescence ratios into discrete 
DNA copy numbers is an important but non-trivial problem.  
Finding DNA regions that consistently exhibit chromosomal losses or gains in cancers provides a crucial
means for locating the specific genes involved in development of different cancer types.  It is therefore important to distinguish,
when a probe shows unusually high or low fluorescence, whether that aberrant signal reflects experimental noise or a probe that is truly
found in a segment of DNA that is gained or lost.  
Furthermore, successful discretization of array CGH data is crucial for
understanding the process of cancer evolution, since discrete inputs are required for a large family of successful evolution algorithms,
e.g., \cite{DBLP:journals/jcb/DesperJKMPS99,DBLP:journals/jcb/DesperJKMPS00}.
It is worth noting that manual annotation of such regions, 
even if possible \cite{coriell}, is tedious and prone to mistakes due to
several sources of noise (impurity of test sample, noise from array CGH method, etc.).
A well-established observation that we use in Chapter~\ref{acghchapter} is 
that near-by probes tend to have the similar DNA copy number. 

Many algorithms and objective functions have  been proposed for
the problem of discretizing and segmenting aCGH data.  Many methods,
starting with~\cite{1016476}, treat aCGH segmentation
as a hidden Markov model (HMM) inference problem.  The HMM approach
has since been extended in various ways, e.g., through the use of
Bayesian HMMs \cite{guha}, incorporation of prior knowledge of
locations of DNA copy number polymorphisms~\cite{citeulike:789789}, and
the use of Kalman filters~\cite{DBLP:conf/recomb/ShiGWX07}.  Other
approaches include wavelet decompositions~\cite{hsu}, quantile
regression~\cite{citeulike:774308}, expectation-maximization in
combination with edge-filtering~\cite{citeulike:774287}, genetic
algorithms~\cite{DBLP:journals/bioinformatics/JongMMVY04},
clustering-based methods~\cite{xing,citeulike:773210}, variants on
Lasso regression~\cite{citeulike:2744846,huang}, and various
problem-specific Bayesian~\cite{barry}, likelihood~\cite{1093217}, and
other statisical models~\cite{DBLP:conf/recomb/LipsonABLY05}.  A
dynamic programming approach, in combination with expectation
maximimization, has been previously used by Picard et
al.~\cite{picard2}. In~\cite{1181383} and~\cite{citeulike:387317} an 
extensive experimental analysis of available methods has been conducted.
Two methods stand out as the leading approaches in practice.
One of these top methods is \cghseg\ \cite{picard}, which assumes that a given CGH
profile is a Gaussian process whose distribution parameters are
affected by abrupt changes at unknown coordinates/breakpoints.  The
other method which stands out for its performance is Circular Binary Segmentation~\cite{olshen} (CBS), a
modification of binary segmentation, originally proposed by Sen and
Srivastava~\cite{sensrivastava}, which uses a statistical comparison
of mean expressions of adjacent windows of nearby probes to identify
possible breakpoints between segments combined with a greedy algorithm
to locally optimize breakpoint positions.

\subsection{Cancer subtypes}
\label{subsec:unmixinghere}

Genomic studies have dramatically improved our understanding of the biology of tumor formation and
treatment. In part this has been accomplished by harnessing tools that profile the genes and proteins in tumor cells, 
revealing previously indistinguishable tumor subtypes that are likely to exhibit distinct sensitivities to treatment 
methods \cite{Perou00, Golub99}. As these tumor subtypes are uncovered, 
it becomes possible to develop novel therapeutics more specifically targeted to the particular 
genetic defects that cause each cancer \cite{Pegram00,Atkins02,Bild09}. 
While recent advances have had a profound impact on our understanding of tumor biology, 
the limits of our understanding of the molecular nature of cancer obstruct the burgeoning efforts in 
``targeted therapeutics'' development. These limitations are apparent in the high failure rate of the discovery 
pipeline for novel cancer therapeutics \cite{Kamb07} as well as in the continuing difficulty of predicting which 
patients will respond to a given therapeutic.  A striking example is the fact that trastuzumab, the targeted therapeutic 
developed to treat HER2-amplified breast cancers, is ineffective in many patients who have HER2-overexpressing tumors and 
yet effective in some who do not \cite{Paik08}. Furthermore, subtypes typically remain poorly defined 
--- {\em e.g.}, the ``basal-like'' breast cancer subtype, for which different studies have 
inferred very distinct genetic signatures \cite{Perou00,Sotiriou03} --- and yet many 
patients do not fall into any known subtype. Our belief, then, is that clinical treatment of cancer 
will reap considerable benefit from the identification of new cancer subtypes and genetic signatures.

One promising approach for better elucidating the common mutational patterns by which tumors develop 
is to recognize that tumor development is an evolutionary process and apply phylogenetic methods to tumor 
data to reveal these evolutionary relationships.
Much of the work on tumor evolution models flows from the seminal efforts of \cite{DesperJKMPS99} on 
inferring {\it oncogenetic trees} from comparative genomic hybridization (aCGH) profiles of tumor cells. 
A strength in this model stems from the extraction of ancestral structure from many probe sites per tumor, 
potentially utilizing measurements of the expression or copy number changes across the entire genome. 
However, this comes at the cost of overlooking the diversity of cell populations within tumors, which can 
provide important clues to tumor progression but are conflated with one another in tissue-wide assays like aCGH. 
The cell-by-cell approaches, such as \cite{PenningtonSSR07,Shackney04}, use this heterogeneity information 
but at the cost of allowing only a small number of probes per cell. Schwartz and 
Shackney~\cite{SchwartzS10} proposed bridging the gap between these two methodologies by computationally 
inferring cell populations from tissue-wide gene expression samples.  This inference was accomplished through 
``geometric unmixing,'' a mathematical formalism of the problem of separating components of mixed samples in which
each observation is presumed to be an unknown convex combination\footnote{A point $p$ is a convex combination 
combination of basis points $v_0,...,v_k$ if and only if the constraints 
$p=\sum_{i=0}^k \alpha_i v_i$, $\sum_i \alpha_i=1$ and $\forall i~:~\alpha_i\ge 0$ obtain. 
The fractions $\alpha_i$ determine a mixture over the basis points $\{v_i\}$ 
that produce the location $p$.} of several hidden fundamental components.
Other approaches to inferring common pathways include mixture models of oncogenetic 
trees \cite{BeerenwinkelRKHSL05}, PCA-based methods \cite{hoglund1}, conjunctive Bayesian
networks \cite{GerstungBMB09} and clustering \cite{DBLP:journals/bioinformatics/LiuMCRKB06}.

Unmixing falls into the class of methods that seek to recover a set of pure sources 
from a set of mixed observations. Analogous problems have been coined ``the cocktail problem,'' 
``blind source separation,'' and ``component analysis'' and various communities have formalized 
a collection of models with distinct statistical assumptions. In a broad sense, the classical approach 
of principal component analysis (PCA) \cite{Pearson01} seeks to factor the data under the constraint that, 
collectively, the fundamental components form an orthonormal system. Independent component analysis (ICA) 
\cite{comon94} seeks a set of statistically independent fundamental components. 
These methods, and their ilk, have been extended to represent non-linear data distributions through 
the use of kernel methods (see \cite{ScholkopfSM98,ScholkopfS02} for details), which often confound 
modeling with black-box data transformations.  Both PCA and ICA break down as pure source separators 
when the sources exhibit a modest degree of correlation.  Collectively, these methods place  strong 
independence constraints on the fundamental components that are unlikely to hold for tumor samples, 
where we expect components to correspond to closely related cell states.

Extracting multiple correlated fundamental components, has motivated the development of new methods for unmixing genetic data. 
Similar unmixing methods were first developed for tumor samples by Billheimer and colleagues 
\cite{Etzioni05} to improve the power of statistical tests on tumor samples in the presence of 
contaminating stromal cells. Similarly, a hidden Markov model approach to unmixing was developed by 
Lamy {\em et al.} \cite{Lamy07} to correct for stromal contamination in DNA copy number data. 
These recent advances demonstrate the feasibility of unmixing-based approaches for separating 
cell sub-populations in tumor data. Outside the bioinformatics community, geometric unmixing has 
been successfully applied in the geo-sciences \cite{EhrlichF87} and in hyper-spectral image analysis \cite{ChanCHM09}. 

The recent work by \cite{SchwartzS10} applied the hard geometric unmixing model 
to gene expression data with the goal of recovering 
expression signatures of tumor cell subtypes, with the specific goal of facilitating 
phylogenetic analysis of tumors. The results showed promise in identifying meaningful 
sub-populations and improving phylogenetic inferences. 

\subsection{Oncogenetic trees}
\label{subsec:oncotreeshere}

Among the triumphs of cancer research stands the breakthrough work of Vogelstein and his collaborators \cite{fearon,vogelstein} 
which provides significant insight  into the evolution of colorectal cancer. 
Specifically, the so-called ``Vogelgram'' models colorectal tumorigenesis as a linear accumulation of certain genetic events. 
Few years later, Desper et al. \cite{desper} considered more general evolutionary models compared to 
the ``Vogelgram'' and presented one of the first theoretical approaches to the problem \cite{michor},
the so-called {\it oncogenetic trees}. 
Before we provide a description of oncogenetic trees which are the focus of our work,
we would like to emphasize that since then a lot of research work has followed from several groups of researchers,
influenced by the seminal work of Desper et al. \cite{desper}. 
Currently there exists a wealth of methods that infer evolutionary models from 
microarray-based data such as gene expression and array Comparative Genome Hybridization (aCGH) data: 
distance based oncogenetic trees \cite{desper2}, maximum likelihood oncogenetic trees \cite{heydebreck},
hidden variable oncogenetic trees \cite{tofigh}, conjunctive Bayesian networks \cite{beer4}
and their extensions \cite{beer6,beer3}, mixture of trees \cite{beer5}.
The interested reader is urged to read the surveys of Attolini et al. \cite{michor} and Hainke et al. \cite{hainke}
and the references therein on established progression modeling methods. 
Furthermore, oncogenetic trees have successfully shed light into many types of cancer such as renal cancer \cite{desper}, 
hepatic cancer \cite{hepatic} and head and neck squamous cell carcinomas \cite{huang}.

An oncogenetic tree is a rooted directed tree\footnote{Typically, the term 
{\it tree} is reserved for the undirected case and the term {\it branching} for the directed
case. In the context of oncogenetic trees,  we consistently 
use the term {\it tree} for a directed tree as in \cite{desper}.}. The root  represents the 
healthy state of tissue with no mutations. 
Any other vertex $v \in V$ represents a mutation. Each edge represents a ``cause-and-effect'' relationships.  
Specifically, for a mutation represented by vertex $v$ to occur, all the mutations corresponding to vertices that lie on the directed path 
from the root to $v$ must be present in the tumor.
In other words, if two mutations $u,v$ are connected by an edge $(u,v)$ then $v$ cannot occur if $u$ has not occured. 
The edges are labeled with probabilities.
Each tumor corresponds to a rooted subtree of the oncogenetic tree and the probability of occurence is determined 
as described by \cite{desper}. Desper et al. provide an algorithm that finds a likely oncogenetic tree that fits 
the observed data.

\section{Thesis Overview}
\label{subsec:connectingdotsoverview}

In this Section we motivate our work and provide an overview of this dissertation. 



{\bf Rainbow Connectivity (Chapter~\ref{rainbowchapter}) } 
\label{subsec:rainbowoverview}

Connectivity is a fundamental graph theoretic property \cite{bondy1976graph}.
The most well-studied connectivity concept asks for the minimum 
number of vertices or edges which need to be removed in order 
to disconnect the graph. 
However, there exist other graph theoretic concepts that strengthen 
the connectivity concept: imposing bounds on the diameter, 
existence of edge disjoint spanning trees etc. 
In 2006 Chartrand et al. \cite{chartrand2008rainbow} defined the  concept of 
{\it rainbow connectivity}, also referred as {\it rainbow connection}. 
We prefer to provide two motivating examples rather than the exact 
definition which is found in Chapter~\nameref{rainbowchapter}. 

Suppose we wish to route messages in a cellular network $G$, between any two
vertices in a pipeline, and require that each link on the route between the vertices
(namely, each edge on the path) is assigned a distinct channel (e.g., a distinct
frequency). The minimum number of distinct channels we need to use 
is  the rainbow connectivity of $G$. 

Another motivating example is related to securing communication between government agencies \cite{li2012rainbow}. 
The Department of Homeland Security of USA was created in 2003 in response to the weaknesses 
discovered in the secure transfer of classified information after the September 11, 2001 
terrorist attacks. Ericksen \cite{ericksen2007matter}  
observed that because of the unexpected aftermath 
law enforcement and intelligence agencies could not communicate. 
Given that this situation could not have been easily predicted, 
the technologies utilized were separate entities and prohibited shared access, meaning
that there was no way for officers and agents to cross check information between various organizations.
While the information needs to be protected since it relates to national security,
there must also be procedures that permit access between appropriate parties. This
twofold issue can be addressed by assigning information transfer paths between
agencies which may have other agencies as intermediaries while requiring a large
enough number of passwords and firewalls that is prohibitive to intruders, yet small
enough to manage. Equivalently, this number has to be large enough so that one or 
more paths between every pair of agencies have no password repeated. 
Rainbow connectivity arises as the natural answer to the following question: 
what is the minimum number of passwords or firewalls needed that allows 
at least one path between every two agencies so that the passwords along each path are distinct?

{\bf Contributions:} 

In \cite{FriezeTsourakakisRainbow,DBLP:conf/approx/FriezeT12} we prove the following results
on the rainbow connectivity of sparse random graphs. 

\squishlist
 \item For an Erd\"{o}s-R\'{e}nyi random graph  $G = G(n,p)$ at the connectivity threshold, i.e., 
 $p=\frac{\log{n}+{\om}}{n}$, $\om\to\infty,{\om}=o(\log{n})$, we prove Theorem~\ref{thrm:Rainbowmainthrm} 
 which characterizes optimally the rainbow connectivity \whp.
 Our proof is constructive in the following sense: a random coloring is \whp\ a valid rainbow coloring.
 \item For random regular graphs \cite{wormald1999models} we prove Theorem~\ref{thrm:Rainbowregular}. 
 The proof of Theorem~\ref{thrm:Rainbowregular} is still constructive, 
 but requires an unexpected use of a Markov Chain Monte Carlo  algorithm.  
\squishend

{\bf Random Apollonian Graphs (Chapter~\ref{ranchapter})}
\label{subsec:ranoverview} 

 In Chapter~\ref{ranchapter} we analyze Random Apollonian Networks (RANs) \cite{maximal}, 
 a popular random graph model for real-world networks. 
 Compared to other models, RANs generate planar graphs. 
 This makes RANs special for at least two reasons. 
 Firstly, planar graphs form an important family of graphs for various reasons. 
 They model several significant types of spatial real-world networks 
such as power grids, water distribution networks and  road networks. 
For instance, a street network has edges corresponding
to roads and vertices represent roads' intersections
and endpoints. Since edges intersect only at vertices,
street networks are planar. 
It is worth mentioning that planarity of street networks
is almost always violated in practice because of bridges.
However planarity is a good approximation \cite{lammer2006scaling}.
It has been observed through various experimental studies 
that real-world planar graphs have distinct features  
\cite{christensen1999cities,cardillo2006structural,clark1951urban,jiang2007topological,lammer2006scaling} from 
random planar graphs \cite{mcdiarmid2005random}. One such feature is that the degrees are skewed,
obeying a power law degree distribution \cite{lammer2006scaling}.
A recent paper which surveys properties and models of real-world planar graphs is \cite{barthelemy2011spatial}. 
Despite the outstanding amount of work on modeling real-world networks with random graph models, e.g.,
\cite{aiello2000random,albert,borgs2010hitchhiker,borgs2007first,lattanzi2009affiliation,
leskovec2005realistic,mahdian2007stochastic,flaxman2006geometric,flaxman2007geometric,durrett2007random,fabrikant}, 
real-world planar graphs have been overlooked.
Secondly, real-world networks tend to have small vertex separators.
By the planar separator theorem \cite{lipton1979separator} and the planarity of RANs, 
this property is satisfied. This should be seen 
in constrast to the popular preferential attachment model \cite{albert} 
where the generated graph is an expander.

{\bf Contributions:} 

In \cite{DBLP:conf/waw/FriezeT12,rans} we perform the first rigorous analysis of RANs. 

\squishlist 
 \item We prove in Theorem~\ref{thrm:RANdegreesequence} tight   results on the degree sequence of RANs. 
  Previous results were weaker or even erroneous, see \cite{zhang,comment}.
 \item We prove in Theorem~\ref{thrm:RANthrm1} tight asymptotic expressions for the top-$k$ largest degrees, where $k$ is constant. 
 \item We provide in Theorem~\ref{thrm:RANthrm2} tight asymptotic expressions for the top-$k$ largest eigenvalues, where $k$ is constant. 
 \item We provide a simple first moment  argument that upper-bounds the asymptotic diameter growth. 
  By observing a bijection between RANs and random ternary trees, we are able to prove Theorem~\ref{thrm:RANternary},  a refined upper bound 
  on the diameter. 
\squishend

{\bf Triangle Counting (Chapter~\ref{trianglecountingchapter}) } 

We motivated the importance of triangles in real-world networks 
in Section~\ref{sec:trianglestriangles}.

{\bf Contributions}

In Chapter~\ref{trianglecountingchapter} we present results from our work 
\cite{tsourakakis2,DBLP:journals/im/KolountzakisMPT12,tsourakakis4,tsourakakis3}. 

\squishlist 
\item

In Section~\ref{sec:trianglesparsifiers} we present a randomized 
algorithm for approximately counting the number 
of triangles in a graph $G$. The algorithm proceeds as follows: 
keep each edge independently with probability $p$, enumerate
the triangles in the sparsified graph $G'$ and return the number of triangles found in $G'$ 
multiplied by  $p^{-3}$. We prove that under mild assumptions 
on $G$ and $p$ our algorithm returns a good 
approximation for the number of triangles with high probability. 
We illustrate the efficiency of our algorithm on various large real-world datasets
where we achieve significant speedups. 
Furthermore, we investigate the performance of existing 
sparsification procedures
namely the Spielman-Srivastava spectral sparsifier \cite{DBLP:conf/stoc/SpielmanS08} and 
the the Bencz{\'u}r-Karger cut sparsifier \cite{benczurstoc,DBLP:journals/corr/cs-DS-0207078}
and show that they are not optimal/suitable with respect to triangle counting. 

\item 

In Section~\ref{sec:degreepartitioning} we extend the results 
from Section~\ref{sec:trianglesparsifiers} by introducing a powerful idea 
of Alon, Yuster and Zwick \cite{739463}. 
As a result, we propose a Monte Carlo algorithm which 
approximates the true number of triangles 
within accuracy (1+$\epsilon$) and runs in $O \left( m + \frac{m^{3/2} \log{n} }{t \epsilon^2} \right)$ time, 
where $n,m,t,\epsilon>0$ are the number of vertices, edges, triangles and a small constant respectively. 
We extend our method  to the semi-streaming model \cite{feigenbaum2005graph} using three passes and a
memory overhead of $O\left(m^{1/2}\log{n} + \frac{m^{3/2} \log{n} }{t \epsilon^2} \right)$.
We propose a random projection based method for triangle counting and provide a sufficient 
condition to obtain an estimate with low variance. 

\item 
In Section~\ref{sec:IPL} we present a new sampling approach to approximating 
the number of triangles in a graph  $G(V,E)$, 
that significantly improves existing sampling approaches. Furthermore, 
it is easily implemented in parallel. 
The key idea of our algorithm is to correlate the sampling of edges 
such that if two edges of a triangle are sampled, 
the third edge is always sampled. Compared to Section~\ref{sec:trianglesparsifiers}, 
this sampling decreases
the degree of the multivariate polynomial that expresses the number 
of sampled triangles. 
As a result, we are able to obtain more ``aggressive'' sampling techniques 
compared to Section~\ref{sec:trianglesparsifiers},
while strong concentration results remain valid.

\squishend

{\bf Densest Subgraphs (Chapter~\ref{densestchapter}) } 

Extracting dense subgraphs from large graphs is a key primitive in a variety of application domains~\cite{aggarwal}.
In the Web graph, dense subgraphs may correspond to thematic groups or even spam link farms, as observed by Gibson et al.~\cite{gibson}.
In biology, finding dense subgraphs can be used for discovering regulatory motifs in genomic DNA~\cite{fratkin}, and
finding correlated genes~\cite{langston}, and detecting transcriptional modules~\cite{Everett}.
In the financial domain, extracting dense subgraphs has been applied to
finding price value motifs~\cite{Du}. Other applications include graph compression~\cite{Buehrer}, graph visualization~\cite{hamelin},
reachability and distance query indexing~\cite{JinXRF09}, and finding stories and events in
micro-blogging streams~\cite{Angel}.

{\bf Contributions}
In Chapter~\ref{densestchapter} we present results
most of which are included in \cite{barcelona}. Our contributions are summarized as follows.

\squishlist

 \item  We introduce a general framework for finding dense subgraphs,
 which subsumes popular density functions. We provide theoretical insights into our framework by
 showing that a large family of objectives are efficiently solvable
 but there also exist subcases which are \NPhard.

  \item Our  framework provides a principled way to derive novel algorithms/heuristics
  geared towards the requirements of the application of interest. 
 As a special instance of our general framework we introduce the problem of extracting \OQC,
 which in general is \NPhard.

\item
We design two efficient algorithms for extracting optimal quasi-cliques. 
The first one is a greedy algorithm where the smallest-degree vertex is repeatedly removed from the graph,
and achieves an additive approximation error. 
The second algorithm is a heuristic based on the local-search paradigm.

\item  For a shifted-version of our objective, we show that the problem can be
approximated within a constant factor of 0.796 using a semidefinite-programming algorithm.

\item We evaluate our efficient algorithms on numerous datasets, 
both synthetic and real, showing that it produces high quality dense subgraphs. 
In particular, in the synthetic data experiments, we plant a clique in  Erd\"{o}s-R\'{e}nyi and 
in random power-law graphs, and measure precision and recall of the methods in 
``recovering'' the planted clique: our method clearly outperforms the \DS\ in this task. 
We also develop applications of our method in data mining and bioinformatic tasks, 
such as forming a successful team of domain experts and finding highly correlated genes from a microarray dataset.

\item Finally, motivated by real-world scenarios, we define and evaluate interesting variants of our original problem definition, such as
($i$) finding the \topk\ \OQCs, and ($ii$) finding {\OQCs} that contain a given set of vertices. 
\squishend

{\bf Structure of the Web Graph  (Chapter~\ref{hadichapter}) } 

The Web graph describes the directed links between pages of the World Wide Web (WWW) \cite{broder2000graph,bonato2005survey}.  
It is a graph which occupies a special position among real-world networks. 
The World Wide Web grew in a decentralized way, under the influence and decision of multiple participants. 
Understanding the structure of WWW and developing realistic models of evolution are two major research problems
that have attracted a lot of interest \cite{abonato}.

{\bf Contributions}

In Chapter~\ref{hadichapter} we present results from our work \cite{Kang:2011:HMR:1921632.1921634,DBLP:conf/sdm/KangTAFL10}. 
Our contributions can be summarized as follows:

The key contributions of this Chapter are the following:
\squishlist
\item We propose \hadi, a  scalable algorithm to compute the radii and diameter of network.
\item We analyze one of the largest public Web graphs, 
   with several {\em billions} of nodes and edges. We validate the small-world phenomenon and find several new structural patterns.
\squishend

{\bf FENNEL: Streaming Graph Partitioning for Massive Scale Graphs (Chapter~\ref{fennelchapter}) }

A key computational problem underlying all existing 
large-scale graph-processing platforms is {\it balanced graph partitioning}. 
When the graph is partitioned, the sizes of the partitions have to be balanced to exploit the 
speedup of parallel computing over different partitions. 
Furthermore, it is critical that the number of edges between distinct partitions is small in order to minimize 
the communication cost incurred due to messages exchanged between different partitions. 
Pregel \cite{malewicz}, Apache Giraph \cite{giraph}, PEGASUS \cite{DBLP:conf/icdm/KangTF09}
and  GraphLab \cite{DBLP:conf/uai/LowGKBGH10}  use as a default partitioner hash partitioning on vertices, 
which essentially corresponds to assigning each vertex to one of the $k$ partitions uniformly at random. 
This heuristic would balance the number of vertices per partition, but as it is entirely oblivious to 
the graph structure, may well result in grossly suboptimal fraction of edges cut. 
Balanced graph partitioning becomes even harder in the case of dynamic graphs: 
whenever a new edge or a new vertex with its neighbors arrives, it has to be assigned 
to one of partition parts.

{\bf Summary of our Contributions}

In Chapter~\ref{fennelchapter} we present results from our work \cite{fennel}. 
Our contributions can be summarized in the following points:

\squishlist 
\item We introduce a general framework for graph partitioning 
that relaxes the hard cardinality constraints on the number of vertices 
in a cluster~\cite{DBLP:conf/stoc/AroraRV04,krauthgamer}.
Our formulation provides a unifying framework that subsumes two of the most
popular heuristics used for streaming balanced graph partitioning:
the folklore heuristic of \cite{Prabhakaran:2012:MLG:2342821.2342825} 
which places a vertex to the cluster with the fewest non-neighbors, 
and the degree-based heuristic of \cite{stanton}, which serves as the current 
state-of-the-art method with respect to performance. 

\item Our framework allows us to define formally the notion of interpolation
between between the non-neighbors heuristic  \cite{Prabhakaran:2012:MLG:2342821.2342825} and the neighbors heuristic 
\cite{stanton}. This provides improved performance for the balanced partitioning problem
in the streaming setting.

\item We evaluate our proposed streaming graph partitioning method, 
\textsf{Fennel}, on a wide range of graph datasets, both real-world and 
synthetic graphs, showing that it produces high quality graph partitions. 
Table~\ref{tab:Fenneltab1} shows the performance of \textsf{Fennel} versus the 
best previously-known heuristic, which is the linear weighted degrees 
\cite{stanton}, and the baseline \textsf{Hash Partition} of vertices. 
We observe that \textsf{Fennel} achieves simultaneously significantly 
smaller fraction of edges cut and balanced cluster sizes. 

\item We also demonstrate the performance gains with respect to communication cost and run time while running iterative computations over partitioned input graph 
data in a distributed cluster of machines. Specifically, we evaluated  \textsf{Fennel} 
and other partitioning methods by computing PageRank in the graph processing platform Apache Giraph. 
We observe significant gains with respect to the byte count among different clusters and run time in 
comparison with the baseline \textsf{Hash Partition} of vertices. 

\item Furthermore, modularity--a popular measure for community detection \cite{girvan2002community,newman2004finding,newman}--
is also a special instance of our framework. We establish an approximation algorithm for a shifted objective, 
achieving a guarantee of  $O(\log(k)/k)$ for partitioning into $k$ clusters. 

\squishend

{\bf PEGASUS: A System for Large-Scale Graph Processing (Chapter~\ref{pegasuschapter})}

As we discussed previously, 
an appealing solution to improve upon scalability is to partition massive graphs into smaller partitions 
and then use a large distributed system to process them.
Designing and implementing efficient graph processing platforms is a major  problem.

{\bf Contributions}

In Chapter~\ref{pegasuschapter} we present results from our work \cite{DBLP:conf/icdm/KangTF09,DBLP:journals/kais/KangTF11}. 
Our main contributions are the following:

\squishlist
\item We introduce a generic framework which allows us to perform 
various important graph mining tasks efficiently by generalizing 
the standard matrix-vector multiplication    (\IGMV).
\item We implement \pegasus, an optimized graph mining library. 
The source code is available online at \url{http://www.cs.cmu.edu/~pegasus/}.
\item We analyze the performance of our system showing that it 
scales well to large-scale graphs. 
\item We apply \pegasus on several large, real-world networks 
and we obtain insights into their structure. 
\squishend

{\bf Approximation Algorithms for Speeding up Dynamic Programming and Denoising aCGH data (Chapter~\ref{acghchapter})}

As we previously discussed in Section~\ref{subsec:preprocessacgh}, 
a major computational problem in cancer biology is assessing DNA copy 
number variations in individual organisms or tissues. 

{\bf Contributions}

In Chapter~\ref{acghchapter} we present results from our work 
\cite{DBLP:conf/soda/MillerPST11,DBLP:journals/jea/TsourakakisPTMS11}. 
Our contributions can be summarized as follows:

\squishlist
\item We propose a new formulation of the array Comparative Genomic Hybridization (aCGH) denoising problem.
Specifically, based on the well-established observation that near-by probes tend to have the same DNA copy
number, we formulate the problem of denoising aCGH data as the problem of approximating 
a signal $P$ with another signal $F$ consisting of a few piecewise constant segments.
Specifically, let  $P=(P_1, P_2, \ldots, P_n) \in \field{R}^n$ be the input
signal -in our setting the sequence of the noisy aCGH measurements- 
and let  $C$ be a constant. Our goal is to find a function $F:[n]\rightarrow \field{R}$ which optimizes the following objective function:

\begin{equation}
\min_{F} \sum_{i=1}^n (P_i-F_i)^2 + C\times (|\{i:F_i \neq F_{i+1} \} |+1).
\label{eq:optimizationobjective}
\end{equation}

\item We solve the problem using a dynamic programming algorithm in $O(n^2)$ time. 
\item We provide a technique which approximates the optimal value of our 
objective function within additive $\epsilon$ error and runs in 
$\tilde{O}(n^{\tfrac{4}{3}+\delta} \log{ (\frac{U}{\epsilon}) )}$ time, where $\delta$ 
is an arbitrarily small positive constant and $U = \max \{ \sqrt{C},(|P_i|)_{i=1,\ldots,n} \}$.
\item We provide a technique for approximate dynamic programming
which solves the corresponding recurrence within a multiplicative factor 
of (1+$\epsilon$) and runs in $O(n \log{n} / \epsilon )$.
\item We validate our proposed model on both synthetic and real data. Specifically, 
our segmentations result in superior precision and recall compared to leading competitors on benchmarks of synthetic
data and real data from the Coriell cell lines.  In addition, we are able to find
several novel markers not recorded in the benchmarks but supported in the oncology literature.
\squishend

{\bf Robust Unmixing of Tumor States in Array Comparative Genomic Hybridization Data (Chapter~\ref{unmixingchapter})}

We discussed in Section~\ref{subsec:introcancerdata} the phenomenon 
of {\it inter-tumor heterogeneity}. We propose a geometric approach 
robust to noise to the problem of detecting cancer subtypes. 

{\bf Contributions}

In Chapter~\ref{unmixingchapter} we present results from our 
work \cite{DBLP:journals/bioinformatics/TolliverTSSS10}. 
Our contributions can be summarized as follows: 

\squishlist
 \item We introduce a novel method for finding cancer subtypes 
 using tissue-wide DNA copy number data as assessed by array comparative 
 genomic hybridization (aCGH) data.
 \item We develop efficient computational tools to solve our optimization 
 problem which is robust to noise. 
 \item We apply our method to an aCGH data set taken from \cite{Navin09} and show that 
 the method identifies state sets 
corresponding to known subtypes consistent with much of the analysis performed by the authors.
\squishend 

{ \bf Perfect Reconstruction of Oncogenetic Trees (Chapter~\ref{oncotreeschapter}) }

Human cancer is caused by the accumulation of genetic alternations in cells \cite{michor,weinberg}.
Finding driver genetic mutations, i.e., mutations which confer growth advantage on the cells carrying them 
and have been positively selected during the evolution of the cancer
and uncovering their temporal sequence have been central goals of cancer research the last decades \cite{nature}.
Among the triumphs of cancer research stands the breakthrough work of Vogelstein and his collaborators \cite{fearon,vogelstein} 
which provides significant insight  into the evolution of colorectal cancer. 
Specifically, the so-called ``Vogelgram'' models colorectal tumorigenesis as a linear accumulation of certain genetic events. 
Few years later, Desper et al. \cite{desper} considered more general evolutionary models compared to 
the ``Vogelgram'' and presented one of the first theoretical approaches to the problem \cite{michor},
the so-called {\it oncogenetic trees}. 

Oncogenetic trees have been a successful tumorigenesis model for various cancer types. 
For this reason understanding its properties is an important problem.

{\bf Contributions}

In Chapter~\ref{oncotreeschapter} we present results from our 
work \cite{DBLP:journals/corr/abs-1201-4618}. Our main contribution is the following:

\squishlist
 \item We provide necessary and sufficient conditions for the unique 
 reconstruction of an oncogenetic tree \cite{desper}. 
\squishend

It is worth outlining that these conditions may be used 
to understand better the phenomenon of {\it intra-tumor heterogeneity} 
\cite{tsourakakis2013modeling}.

{\bf  Conclusion and Future Directions (Chapter~\ref{conclchapter}) } 

Our work leaves numerous interesting problems open. 
In Chapter~\ref{conclchapter} we conclude and provide several new research directions.

\chapter{Theoretical Preliminaries}
\label{theoryprelim}
\lhead{\emph{Theoretical Preliminaries}} 

In this Chapter we review theoretical preliminaries that are used in later Chapters. 

\section{Concentration of Measure Inequalities} 

\noindent The use of Chebyshev's inequality is known as {\it the second moment method}. 

\begin{lemma}[Chebyshev's Inequality \cite{alon}]
Let $X$ be a random variable, $\mu = \Mean{X}, \sigma = \sqrt{\Var{X}}$. For any positive $\lambda >0$

$$ \Prob{ |X-\mu| \geq \lambda \sigma } \leq \frac{1}{\lambda^2}. $$
\label{lem:chebyshevinequality}
\end{lemma} 

Chernoff bounds allow us to obtain strong concentration results, when applicable. 
We use the following version in later Chapters. 

\begin{lemma}[Chernoff Inequality \cite{alon}]
Let $X_1, X_2, \ldots, X_k$ be independently distributed $\{0,1\}$ variables with $E[X_i]=p$. Then for any $\epsilon > 0$, we have
$$\Prob{ |\frac{1}{k} \sum_{i=1}^k X_i - p| > \epsilon p } \leq 2 e^{-\epsilon^2pk/2}$$
\label{lem:chernoff}
\end{lemma}

The theory of discrete time martingales \cite{alon} will be 
the key to establish concentration inequalities in our proofs for degree sequences.

\begin{lemma}[Azuma-Hoeffding inequality]
Let $\lambda >0$. Also, let $(X_t)_{t=0}^n$ be a martingale sequence with $|X_{t+1}-X_t| \leq c$ for 
$t=0,\ldots,n-1$.Then:
$$ \Prob{ |X_n - X_0| \geq \lambda } \leq \exp{ \Big(-\frac{\lambda^2}{2c^2n}\Big) }.$$
\label{lem:RANazuma}
\end{lemma}

The Kim-Vu theorem is an important concentration result since it allows us to obtain strong concentration
when the polynomial of interest is not {\it smooth}. Specifically, for the purposes of our work, 
let $Y=Y(t_1,\ldots,t_m)$ be a positive polynomial of $m$ Boolean variables $[t_i]_{i=1..m}$ which are independent.
A common task in combinatorics is to show that $Y$ is concentrated around its expected value. 
In the following we state the necessary definitions and the main concentration result which we will use in our method. $Y$ is totally positive if all of its coefficients are non-negative variables.
$Y$ is homogeneous if all of its monomials have the same degree and we call this value the degree of the polynomial.
Given any multi-index $\alpha=(\alpha_1,\ldots,\alpha_m) \in \field{Z}^m_{+}$, define the partial derivative
${\partial^{\alpha}Y}= ( \frac{\partial}{\partial t_1} )^{\alpha_1} \ldots ( \frac{\partial}{\partial t_m} )^{\alpha_m} Y(t_1,\ldots,t_m)$
and denote by $\Abs{\alpha} = {\alpha_1}+\cdots{\alpha_m}$ the order of $\alpha$.
For any order $d \geq 0$, define ${\mathbb E}_d(Y)=\max_{\alpha:|\alpha|=d}{\mathbb E}(\partial^{\alpha}Y)$
and ${\mathbb E}_{\ge d}(Y) =\max_{d' \ge d}{\mathbb E}_{d'}(Y)$.

Now, we refer to the main theorem of Kim and Vu of \cite[$\S$1.2]{kim-vu} as phrased in 
Theorem 1.1 of \cite{vu} or as Theorem 1.36 of \cite{tao-vu}.

\begin{theorem}
\label{thrm:kim-vu}
There is a constant $c_k$ depending on $k$ such that the following holds.
Let $Y(t_1, \ldots, t_m)$ be a totally positive polynomial of degree $k$, where $t_i$ can have arbitrary
distribution on the interval $[0, 1]$. Assume that:
\beql{cond1}
\Mean{Y} \ge {\mathbb E}_{\ge 1}(Y) 
\eeq
Then for any $\lambda \ge $ 1:
\beql{res1}
\Prob{\Abs{Y-\Mean{Y}} \ge c_k \lambda^k (\Mean{Y} {\mathbb E}_{\ge 1}(Y))^{1/2}} \le e^{-\lambda + (k-1)\log m}.
\eeq
\end{theorem}

Typically, when a polynomial $Y$ is smooth, it is strongly concentrated. 
By smoothness one usually means a small Lipschitz coefficient or in other words, when one changes the value of one 
variable $t_j$, the value $Y$ changes no more than a constant. However, as stated in \cite{vu} this is 
restrictive in many cases. Thus one can demand ``average smoothness'' as defined in \cite{vu}
which is quantified via the expectation of partial derivatives of any order.

\section{Random Projections} 
\label{subsec:randomproj}

A random projection $x \to Rx$ from $\field{R}^d \to \field{R}^k$ approximately preserves all Euclidean distances.
One version of the Johnson-Lindenstrauss lemma \cite{lindenstrauss} is the following:

\begin{lemma}[Johnson Lindenstrauss]
Suppose $x_1,\ldots,x_n \in \field{R}^d$ and $\epsilon>0$ and take $k=C \epsilon^{-2} \log n$. Define the random matrix $R \in \field{R}^{k\times d}$ by taking all $R_{i,j} \sim N(0,1)$ (standard gaussian) and independent. Then, with probability bounded below by a constant the points $y_j = R x_j \in \field{R}^k$ satisfy
$$
(1-\epsilon) \Abs{x_i-x_j} \le \Abs{y_i - y_j} \le (1+\epsilon) \Abs{x_i - x_j}
$$
for $i, j=1,2,\ldots,n$ where $|\cdot|$ represents the Euclidean norm.
\label{thrm:jllemma}
\end{lemma}

\section{Extremal Graph Theory}
\label{subsec:extremal} 

Hajnal and Szemer\'{e}di \cite{HajnalSzemeredi} proved in 1970 the following conjecture of Paul Erd\"{o}s:

\begin{theorem}[Hajnal-Szemer\'{e}di Theorem]
\label{lem:hajnal}
Every graph with $n$ vertices and maximum vertex degree at most $k$ is $k+1$ colorable with all 
color classes of size  $\lfloor \tfrac{n}{k+1} \rfloor$ or  $\lceil \tfrac{n}{k+1} \rceil$.
\end{theorem}

Ahlswede and Katona consider the following problem: which graph with a given number of vertices $n$ 
and a given number of edges $m$ maximizes the number of edges in its line graph $L(G)$? 
The problem is equivalent to maximizing the sum of squares of the degrees of the vertices 
under the constraint that their sum equals twice the number of the edges. 
The following theorem was given in \cite{ahlswede} and answers this question. 

\begin{lemma}[Ahlswede-Katona theorem]
The maximum value of the sum of the squares of all vertex degrees $\sum_{v \in V(G)} d(v)^2$ over 
the set of all graphs with $n$ vertices and $m$ edges 
occurs at one or both of two special types of graphs, the quasi-star graph or the quasi-complete graph.
\end{lemma}

For further progress on other questions related to the above optimization problem 
such as when does the optimum occur at both graphs, see the work of Abrego, Fern\'{a}ndez-Merchant, Neubauer and Watkins 
\cite{abrego}.

\section{Two useful lemmas}
\label{subsec:usefullemmas} 

Two more useful lemmas we will use in later chapters follow.

\begin{lemma}[Lemma 3.1, \cite{chunglu}]
Suppose that a sequence $\{a_t\}$ satisfies the recurrence 

$$ a_{t+1} = (1 - \frac{b_t}{t+t_1}) a_t + c_t$$

\noindent for $t \geq t_0$.  Furthermore suppose $\displaystyle\lim_{t \rightarrow +\infty} b_t = b >0$
and $\displaystyle\lim_{t \rightarrow +\infty} c_t = c$. Then $\displaystyle\lim_{t \rightarrow +\infty} \frac{a_t}{t}$ exists
and 

$$ \displaystyle\lim_{t \rightarrow +\infty} \frac{a_t}{t} = \frac{c}{1+b}. $$ 

\label{lem:RANchunglubook} 
\end{lemma}


Graphs can be viewed as electrical networks. 
Given two vertices $s,t \in V(G)$, we can ensure an electrical current from $s$ to $t$ 
of value 1. The potential/voltage difference between $s$ and $t$ is defined to be the 
{\it effective resistance} $R(s,t)$. 
For further details the interested reader should read \cite{bollobas}.
The following theorem is due to Foster \cite{foster49}. 

\begin{theorem}[Foster's theorem \cite{foster49}] 
\label{thrm:fosterstheorem}
Let $G$ be a connected graph of order $n$. Then
$$ \sum_{(u,v) \in E(G)} R(u,v) = n-1.$$
\end{theorem}

\section{Semidefinite Bounds}
\label{sec:semidefinitebounds}

Semidefinite programs are generalizations of linear 
programs which can be solved in polynomial time using 
interior point methods. Semidefinite programming 
uses symmetric, positive semidefinite matrices. 

\begin{definition}
A matrix $A \in \field{R}^{n \times n}$ is positive semidefinite 
if and only if for all $x \in \field{R}^n$, $x^T A x \geq 0$. 
\end{definition} 

We define the scalar product $\langle A,B \rangle$ of matrices $A,B\field{R}^{n \times n}$
as $\langle A,B \rangle = Trace(B^TA)$ which is an inner product 
on the vector space of $n\times n$ matrices.
A semidefinite program is defined by the symmetric $n\times n$ matrices $C,A_1,\ldots,A_m$ 
and vector $b \in \field{R}^m$ as follows:

\begin{equation}
\label{equation:SDPgw}
\framebox{
\begin{minipage}{0.85\linewidth}
\begin{align}
\nonumber	
\mbox{\bf max}  \qquad & \langle C,X \rangle  \\
\nonumber
\nonumber
\mbox{{subject to }}\, & \langle A_i,X \rangle, \,\mbox{for all}\,  i \in \{1,..,m\} \\
\nonumber
\mbox{{and }}\, & X \succeq 0,\, X \,\mbox{symmetric.} \\
\nonumber
\end{align}
\end{minipage}}
\end{equation}

Goemans and Williamson  significantly advanced the field of approximation
algorithms by introducing a randomized rounding  for the 
MAX-CUT problem\cite{goemans}.
The MAX-CUT problem takes as input a graph $G(V,E,w), w:E \rightarrow \field{R}^+$ 
and asks for a non-empty set $S$ such that weight
$\sum_{(i,j) \in E(G)} w_{ij}$   of the cut $(S,V\backslash S)$ 
is maximized. This problem is \NPhard. 
The first step of the Goemans-Williamson algorithm is a semidefinite
relaxation of the following quadratic integer program which is equivalent to MAX-CUT.

\begin{align*} 
p^* &= \max \frac{1}{2} \sum_{(i,j) \in E} w_{ij} (1-x_ix_j) 
     &\text{subject to~} x\in \{-1,+1\}^n.
\end{align*} 

Notice that $X=xx^T$ is a symmetric, positive semidefinite matrix with rank 1
and $X_{ii}=1$ for $i\in [n]$. 
The semidefinite relaxation relaxes the rank 1 condition  as follows

\begin{align*} 
s^* &= \max \frac{1}{2} \sum_{(i,j) \in E} w_{ij} (1-X_{ij}) 
     &\text{subject to~}  X \succeq 0,\, X \,\mbox{symmetric.}, X_{ii}=1 \,\mbox{for all i}.
\end{align*} 

The second step of the Goemans-Williamson algorithm consists of a 
randomized rounding procedure. Specifically, let $X$ be the optimal solution
of the semidefinite relaxation. As $X$ is positive semidefinite $X=U^TU$ 
where $U \in \field{R}^{d \times n}$,  $d \leq n$. Let $u_j$ be the $j$-th
column of $U$. As $X_{ii}=1$ for all $i$, $|| u_i ||_2=1$. 
Goemans and Williamson proposed generating a random unit vector $r \in \field{R}^d$ 
and letting $S$ be the set of vertices which correspond to columns
of $U$ such that $u_j^T r > 0$. They proved that this algorithm
provides a $\beta$-approximation where $\beta > 0.87856$. 
Their technique has been extended to the MAX-k-CUT problem \cite{friezejerrum}.

\section{Speeding up Dynamic Programming} 
\label{subsec:speed}

Dynamic programming is a powerful problem solving technique introduced 
by Bellman \cite{862270} with numerous applications in biology, e.g., \cite{picard,360861,9303}, 
in control theory, e.g., \cite{517430}, in operations research and many other fields. 
Due to its importance, a lot of research has focused on speeding up 
basic dynamic programming implementations. A successful example of 
speeding up a naive dynamic programming implementation is the computation 
of optimal binary search trees. Gilbert and Moore solved the problem 
efficiently using dynamic programming \cite{bb12353}. Their algorithm 
runs in $O(n^3)$  time and for several years this running time was considered to be tight. 
In 1971 Knuth \cite{DBLP:journals/acta/Knuth71} showed that the same computation can be carried out in $O(n^2)$ time. 
This remarkable result was generalized by Frances Yao in \cite{fyao,804691}.
Specifically, Yao showed that this dynamic programming speedup technique works for a large class
of recurrences. She considered the recurrence $c(i,i) =0$, $c(i,j) = \min_{i<k\leq j}{ (  c(i,k-1) + c(k,j)) }  + w(i,j)$
for $i <j$ where the weight function $w$ satisfies the quadrangle inequality (see Section~\ref{sec:Trimmermonge}) and proved
that the solution of this recurrence can be found in $O(n^2)$ time. 
Eppstein, Galil and Giancarlo have considered similar recurrences where they showed that naive $O(n^2)$ implementations
of dynamic programming can run in $O(n\log{n})$ time \cite{Eppstein88speedingup}. 
Larmore and Schieber \cite{larmore} further improved the running time, giving a linear time algorithm when the weight function 
is concave. Klawe and Kleitman give in \cite{kleitman} an algorithm which runs in $O(n\alpha(n))$ time when the weight
function is convex, where $\alpha(\cdot)$ is the inverse Ackermann function. 
Furthermore, Eppstein, Galil, Giancarlo and Italiano have also explored the effect of sparsity \cite{146650,146656}, 
another key concept in speeding up dynamic programming.  
Aggarwal, Klawe, Moran, Shor, Wilber developed an algorithm, widely known as the SMAWK algorithm,
\cite{10546} which can compute in $O(n)$ time the row maxima of a totally monotone $n \times n$ matrix. 
The connection between the Knuth-Yao technique and the SMAWK algorithm was made 
clear in \cite{1109562}, by showing that the Knuth-Yao technique is a special case of the use of totally monotone matrices. 
The basic properties which allow these speedups are the convexity or concavity of the weight function.
Such properties date back to Monge \cite{monge} and are well studied in the literature, see for example \cite{240860}. 

Close to our work lies the work on histogram construction, an important problem for database applications. 
Jagadish et al. \cite{671191} originally provided a simple dynamic programming 
algorithm which runs in $O(kn^2)$ time, where $k$ is the number of buckets and $n$ the input size
and outputs the best V-optimal histogram. 
Guha, Koudas and Shim \cite{1132873} propose a $(1+\epsilon)$  approximation algorithm which runs in linear time. 
Their algorithms exploits monotonicity properties of the key quantities involved in the problem.
Our $(1+\epsilon)$ approximation algorithm in Section~\ref{sec:Trimmermonge} uses a
decomposition technique similar to theirs.

\section{Reporting Points in a Halfspace} 
\label{subsec:compgeom} 

Let $S$ be a set of points in $\field{R}^d$ and let $k$ denote the size of the output, i.e., the number of points to be reported. 
Consider the problem of preprocessing $S$ such that for any 
halfspace query $\gamma$ we can report efficiently whether the set $S \cap \gamma$ is empty or not.  
This problem is a well studied special case of the more general range searching problem. For
an extensive survey see the work by Agarwal and Erickson \cite{Agarwal99geometricrange}. 
For $d=2$, the problem has been solved optimally by Chazelle, Guibas and Lee \cite{chazelle_power_1985}.  
For $d=3$, Chazelle and Preparata in \cite{323248} gave a solution with nearly linear space and $O(\log{n}+k)$
query time, while Aggarwal, Hansen and Leighton \cite{100260} gave a solution with a more expensive preprocessing 
but $O(n\log{n})$ space. When the number of dimensions is greater than 4, i.e., $d \geq 4$, 
Clarkson and Shor \cite{82363} gave an algorithm that requires $O(n^{\Floor{d/2} +\epsilon})$ preprocessing time and space,
where $\epsilon$ is an arbitrarily small positive constant, but can subsequently answer queries in $O(\log{n}+k)$ time.  
Matou\v{s}ek in \cite{DBLP:conf/focs/Matousek91} provides improved results on the problem, which are used by Agarwal, Eppstein, Matou\v{s}ek  \cite{10.1109/SFCS.1992.267816}
in order to create dynamic data structures that trade off insertion and query times. 
We refer to Theorem 2.1(iii) of their paper \cite{10.1109/SFCS.1992.267816}:

\begin{theorem}[Agarwal, Eppstein, Matou\v{s}ek \cite{10.1109/SFCS.1992.267816}]
Given a set $S$ of $n$ points in $\field{R}^d$ where $d \geq 3$ 
and a parameter $m$ between $n$ and $n^{\Floor{\frac{d}{2}}}$
the halfspace range reporting  problem can be solved with the following performance:
$O(\frac{n}{m^{1/\Floor{d/2}}} \log{n})$ query time,
$O(m^{1+\epsilon})$ space and preprocessing time, 
$O(m^{1+\epsilon}/n)$ amortized update time.
\label{thrm:matousek}
\end{theorem}

Substituting for $d=4$, $m=n^{\tfrac{4}{3}}$ we obtain the following corollary, which will be used as a subroutine in our proposed method:

\begin{corollary}
Given a set $S$ of $n$ points in $\field{R}^4$ the halfspace range reporting problem can be solved with 
$O(n^{\tfrac{1}{3}}\log{n})$ query time, $O(n^{\tfrac{4}{3}+\delta})$ space and preprocessing time, 
and $O(n^{\tfrac{1}{3}+\delta})$  update time, where $\delta$ is an arbitrarily small positive constant. 
\label{cor:matousek}
\end{corollary}

\section{Monge Functions and Dynamic Programming}
\label{sec:Trimmermonge}

Here, we refer to one of the results in \cite{larmore} which we use in Section~\ref{sec:Trimmermonge} as a subroutine for our
proposed method. A function $w$ defined on pairs of integer indices is Monge (concave) if for any 4-tuple of indices $i_1 < i_2 < i_3 < i_4$,
$w(i_1,i_4)+w(i_2,i_3) \geq w(i_1, i_3) + w(i_2, i_4)$. Furthermore,  we assume that $f$ is a function such that the values 
$f(a_j)$ for all $j$ are easily evaluated. The following results holds:

\begin{theorem}[\cite{larmore}]
Consider the one dimensional recurrence $a_i =\min_{j < i} \{ f(a_j) + w(j,i)\}$ for $i=1,\ldots,n$, where the basis
$a_0$ is given. There exists an algorithm which solves the recurrence {\it online} in $O(n)$ time\footnote{Thus, obtaining $O(n)$ speedup compared to the straight-forward dynamic programming algorithm which runs in $O(n^2)$ units of time.}.
\label{thm:larmore}
\end{theorem}

\newpage 
\vspace*{\fill}
\begingroup
\centering

{\Huge  I {\em Graphs and Networks} }

\endgroup
\vspace*{\fill}

\clearpage
\chapter{Rainbow Connectivity of Sparse Random Graphs}
\label{rainbowchapter}
\lhead{\emph{Rainbow Connectivity of Sparse Random Graphs}} 
\section{Introduction}
\label{sec:RainbowIntro}

Connectivity is a fundamental graph theoretic property. Recently, the  
concept of {\em rainbow connectivity} was introduced by Chartrand et al. in \cite{chartrand2008rainbow}. 
An edge colored graph $G$ is rainbow edge connected if any two vertices are connected by a 
path whose edges have distinct colors. The rainbow connectivity $rc(G)$ of a connected graph 
$G$ is the smallest number of colors that are needed in order to make $G$ rainbow edge connected. 
Notice, that by definition a rainbow edge connected graph is also connected and furthermore
any connected graph has a trivial edge coloring that makes it rainbow edge connected, since 
one may color the edges of a given spanning tree with distinct colors. 
Other basic facts established in \cite{chartrand2008rainbow} are that $rc(G)=1$ if and only if $G$ is a 
clique and $rc(G)=|V(G)|-1$ if and only if $G$ is a tree. 
Besides its theoretical interest, rainbow connectivity is also of interest in applied settings, such 
as securing sensitive information \cite{li2012rainbow}, transfer and networking \cite{chakraborty}.

The concept of  rainbow connectivity has attracted the interest of various researchers.
Chartrand et al. \cite{chartrand2008rainbow} determine the rainbow connectivity of several special classes
of graphs, including multipartite graphs. Caro et al. \cite{caro} prove that for a connected 
graph $G$ with $n$ vertices and minimum degree $\delta$, the rainbow connectivity satisfies
$rc(G)\leq \frac{\log{\delta}}{\delta}n(1+f(\delta))$, where $f(\delta)$ tends to zero 
as $\delta$ increases. The following simpler bound was also proved in \cite{caro},
$rc(G) \leq n \frac{4\log{n}+3}{\delta}$. 
Krivelevich and Yuster \cite{krivelevichyuster} removed the logarithmic factor
from the Caro et al. \cite{caro} upper bound. Specifically they proved that
$rc(G) \leq \frac{20n}{\delta}$. 
Due to a construction of a graph
with minimum degree $\delta$ and diameter $\frac{3n}{\delta+1}-\frac{\delta+7}{\delta+1}$ by 
Caro et al. \cite{caro}, the best upper bound one can hope for is $rc(G) \leq \frac{3n}{\delta}$.
Chandran, Das, Rajendraprasad and Varma \cite{Chandran} have subsequently
proved an upper bound of $\frac{3n}{\delta+1}+3$, which is therefore essentially optimal.

As Caro et al. point out, the random graph setting poses several intriguing questions. 
Specifically, let $G=G(n,p)$ denote the binomial random graph on $n$ 
vertices with edge probability $p$ \cite{erdosrenyi}. 
Caro et al. \cite{caro} proved that $p=\sqrt{\log{n}/n}$ is the sharp threshold for 
the property $rc(G(n,p))\leq 2$. 
He and Liang \cite{heliang} studied further the rainbow connectivity of random graphs. 
Specifically, they obtain the sharp threshold for the property $rc(G) \leq d$
where $d$ is constant. For further results and references we refer 
the interested reader to the recent monograph of Li and Sun \cite{li2012rainbow}. 
In this work we look at the rainbow connectivity of the binomial graph
at the connectivity threshold  $p=\frac{\log{n}+{\om}}{n}$ where ${\om}=o(\log{n})$. 
This range of values for $p$ poses problems that cannot be tackled 
with the techniques developed in the aforementioned work. 
Rainbow connectivity has not been studied in random regular graphs to the best of our knowledge. 

Let
\begin{equation}
\label{Ldef} 
L = \diam
\end{equation}

and let $A\sim B$ denote $A=(1+o(1))B$ as $n\to\infty$.

\noindent We establish the following theorems:
\begin{theorem}
\label{thrm:Rainbowmainthrm}
Let $G = G(n,p),p=\frac{\log{n}+{\om}}{n}$, $\om\to\infty,{\om}=o(\log{n})$. Also, let $Z_1$ be the number of 
vertices of 
degree 1 in $G$.
Then, with high probability({\it whp})
$$  rc(G) \sim \max\set{Z_1, L},$$
\end{theorem}

It is known that \whp\ the diameter of $G(n,p)$ is asymptotic to $L$ for $p$ as in the above range,
see for example Theorem 10.17 of Bollob\'as \cite{bollobas}. 
Theorem \ref{thrm:Rainbowmainthrm} gives asymptotically optimal results. Our next theorem is not quite as precise.
\begin{theorem}
 \label{thrm:Rainbowregular} 
Let $G=G(n,r)$ be a random $r$-regular graph where $r \geq 3$ is a fixed integer. 
Then, {\it whp} 
$$ rc(G) =\begin{cases}
          O(\log^4n)&r=3\\
          O(\log^{2\th_r}n)&r\geq 4.
          \end{cases}
$$ 
where $\th_r=\frac{\log (r-1)}{\log (r-2)}$.
\end{theorem}

\noindent All logarithms whose base is omitted are natural.
It will be clear from our proofs that the colorings in the above two theorems can be
constructed in a low order polynomial time. The second theorem, while weaker, contains
an unexpected use of a Markov Chain Monte-Carlo (MCMC) algorithm for randomly coloring a graph.

The Chapter is organized as follows: After giving a sketch of our approach in Section \ref{sec:Rainbowsketch},
in Sections~\ref{sec:Rainbowproofs},~\ref{sec:Rainbowregular} we prove 
Theorems~\ref{thrm:Rainbowmainthrm},~\ref{thrm:Rainbowregular}
respectively.

\section{Sketch of approach}
\label{sec:Rainbowsketch}

The general idea in the proofs of both theorems is as follows:
\begin{enumerate}
\item Randomly color the edges of the graph in question. For Theorem \ref{thrm:Rainbowmainthrm} we can (in the main)
use a uniformly
random coloring. The distribution for Theorem \ref{thrm:Rainbowregular} is a little more complicated. 
\item To prove that this works, we have to find, for each pair of vertices $x,y$, a large collection of edge disjoint
paths joining them. It will then be easy to argue that at least one of these paths is rainbow colored.
\item To find these paths we pick a typical vertex $x$. We grow a regular tree $T_x$ with root $x$. The depth is chosen
carefully. We argue that for a typical pair of vertices $x,y$, many of the leaves of $T_x$ and $T_y$ can be put
into 1-1 correspondence $f$ so that (i) the path $P_x$ from $x$ to leaf $v$ of $T_x$ is rainbow colored, (ii) the path 
$P_y$ from $y$ to the leaf $f(v)$ of $T_y$ is rainbow colored and (iii) $P_x,P_y$ do not share color.
\item We argue that from most of the leaves of $T_x,T_y$ we can grow a tree of depth approximately equal to half the diameter.
These latter trees themselves contain a bit more than $n^{1/2}$ leaves. These can be constructed so that they are vertex disjoint. Now
we argue that each pair of trees, one associated with $x$ and one associated with $y$, are joined by an edge. 
\item We now have, by construction, a large set of edge disjoint paths joining leaves $v$ of $T_x$ to leaves
$f(v)$ of $T_y$. A simple estimation shows that \whp\ for at least one leaf $v$ of $T_x$, the path from $v$ to $f(v)$ is 
rainbow colored and does not
use a color already used in the path from $x$ to $v$ in $T_x$ or the path from $y$ to $f(v)$ in $T_y$.
\end{enumerate}
We now fill in the details of both cases.

\section{Proof of Theorem~\ref{thrm:Rainbowmainthrm}} 
\label{sec:Rainbowproofs}

Observe first that $rc(G)\geq \max\set{Z_1,diameter(G)}$. First of all, each edge incident to a vertex of 
degree one must have a distinct color.
Just consider a path joining two such vertices. Secondly, if the shortest distance between two vertices is $\ell$ 
then we need at least $\ell$ colors.
Next observe that \whp\ the diameter $D$ is asymptotically equal to $L$, 
see for example \cite{bollobas}. 
We break the proof of Theorem~\ref{thrm:Rainbowmainthrm} into several lemmas. 

Let a vertex be {\em large} if $d(x)\geq \log{n}/100$ and {\em small} otherwise.
\begin{lemma} 
\label{Rainbowlem2}
{\em Whp}, there do not exist two small vertices within distance 
at most $3L/4$. 
\end{lemma}

\begin{proof}

\begin{align*}
&\Prob{\exists x,y \in [n]:\;d(x),d(y) \leq \log{n}/100\text{ and }dist(x,y) \leq \frac{3L}{4}}\\
&\leq \binom{n}{2} \sum_{k=1}^{3L/4} n^{k-1}p^k \brac{\sum_{i=0}^{\log n/100}\binom{n-1-k}{i}p^i(1-p)^{n-1-k}}^2 \\
&\leq \sum_{k=1}^{3L/4} n(2\log n)^k\brac{2\binom{n}{\log n/100}p^{\log n/100}(1-p)^{n-1-\log n/100}}^2\\
&\leq \sum_{k=1}^{3L/4} n(2\log n)^k\brac{2(100e^{1+o(1)})^{\log n/100}n^{-1+o(1)}}^2\\
&\leq \sum_{k=1}^{3L/4} n(2\log n)^k n^{-1.9}\\
&\leq 2n(2\log n)^{3L/4}n^{-1.9}\\
&\leq n^{-.1}.
\end{align*}
 
\end{proof}

\noindent 
We use the notation $e[S]$ for the number of edges induced by a given set of vertices $S$. Notice that 
if a set $S$ satisfies $e[S] \geq s+t$ where $t\geq 1$, the induced subgraph $G[S]$ has at least $t+1$ cycles.

\begin{lemma}
\label{Rainbowlem3} 
Fix $t\in \field{Z}^+$ and $0<\alpha<1$. Then, {\em whp} there does not exist a subset $S \subseteq [n]$, 
such that $|S| \leq \alpha t L$ and $e[S] \geq |S|+t$. 
\end{lemma}

\begin{proof} 

For convenience, let $s=|S|$ be the cardinality of the set $S$.Then,
\begin{align*}
\Prob{ \exists S: s \leq \alpha t L \text{~~and~~} e[S] \geq s+t } 
&\leq \sum_{s \leq \alpha t L} \binom{n}{s} \binom{\binom{s}{2}}{s+t} p^{s+t}\\
&\leq \sum_{s \leq \alpha t L } \bfrac{ne}{s}^s \bfrac{es^2p}{2(s+t)}^{s+t} \\ 
&\leq \sum_{s \leq \alpha t L } (e^{2+o(1)}\log{n})^s \bfrac{es\log{n}}{n}^t \\
&\leq  \alpha tL \brac{ (e^{2+o(1)}\log{n})^{ \alpha L} \bfrac{e \alpha t\log^2{n} }{n\log{\log{n}}}}^t \\ 
&< \frac{1}{n^{(1-\a-o(1))t}}.
\end{align*}

\end{proof}

\begin{remark}\label{rem1}
Let $T$ be a rooted tree of depth at most $4L/7$ and let $v$ be a vertex not in $T$, but with $b$ neighbors in $T$.
Let $S$ consist of $v$, the neighbors of $v$ in $T$ plus the ancestors of these neighbors. Then 
$|S|\leq 4bL/7+1\le 3bL/5$ and $e[S]=|S|+b-2$. It follows from {\red the proof of} 
Lemma \ref{Rainbowlem3} with $\a=3/5$ and $t=8$, that 
we must have $b\leq 10$
with probability $1-o(n^{-3})$.
\end{remark}

\noindent Our next lemma shows the existence of the subgraph $G'_{x,y}$ described next and shown in 
Figure~\ref{fig:Rainbowfig1}
for a given pair of vertices $x,y$. 
We first deal with paths between large vertices.

\begin{figure*}
\centering
\includegraphics[width=0.7\textwidth]{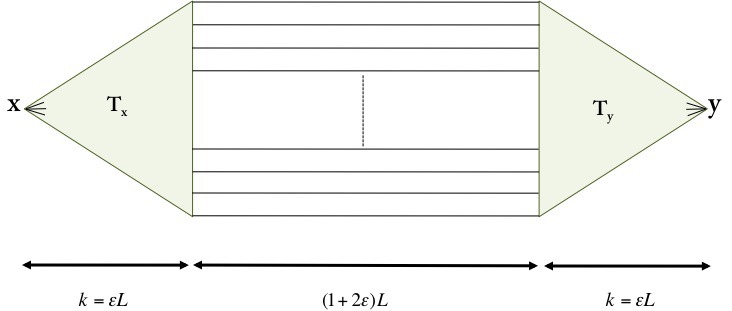}
\caption{\label{fig:Rainbowfig1}Structure of Lemma~\ref{Rainbowlem4}.}
\end{figure*}

Now let 
\beq{eps}
\text{$\e=\e(n)=o(1)$ be such that $\frac{\e\log\log n}{\log 1/\e}\to \infty$ and let $k=\e L$.}
\eeq
Here $L$ is defined in \eqref{Ldef} and we could take $\e=1/(\log\log n)^{1/2}$.
\begin{lemma}
\label{Rainbowlem4} 
{\red Whp,} for all pairs of large vertices $x,y \in [n]$
there exists a subgraph $G_{x,y}(V_{x,y},E_{x,y})$ of $G$ as shown in figure~\ref{fig:Rainbowfig1}.
The subgraph consists of two isomorphic vertex disjoint trees $T_x,T_y$ rooted at $x,y$ each of depth $k$.
$T_x$ and $T_y$ both have a branching factor of $\log n/101$. I.e. each vertex of $T_x,T_y$ has
at least  $\log n/101$ neighbors, excluding its parent in the tree.
Let the leaves of $T_x$ be $x_1,x_2,\ldots,x_\t$ where $\t\geq n^{4\e/5}$ and those of $T_y$ be
$y_1,y_2,\ldots,y_\t$. Then $y_i=f(x_i)$ where $f$ is a natural 
isomporphism that preserves the parent-child relation.
Between each pair of leaves $(x_i,y_i),i=1,2,\ldots,\t$ 
there is a path $P_i$ of length  $(1+2\epsilon) L$. The paths $P_i,i=1,2,\ldots,\t$ are edge disjoint.
\end{lemma}

\begin{proof}

Because we have to do this for all pairs $x,y$, we note without further comment that likely (resp. unlikely) 
events will
be shown to occur with probability $1-o(n^{-2})$ (resp. $o(n^{-2}$)).

To find the subgraph shown in Figure~\ref{fig:Rainbowfig1} we grow tree structures as shown 
in Figure~\ref{fig:Rainbowfig2}. 
Specifically, we first grow a tree from $x$ using BFS until it reaches  depth $k$. 
Then, we grow a tree starting from $y$ again using BFS until it reaches depth $k$. 
Finally, we grow trees from the leaves of $T_x$ and $T_y$ using BFS for depth $\gamma=(\frac{1}{2}+\epsilon)L$. 
Now we analyze these processes. Since the argument is the same we explain
it in detail for $T_x$ and we outline the differences for the other trees. 
We use the notation $D_i^{(\r)}$ for the number of vertices at depth $i$
of the BFS tree rooted at $\r$.

\begin{figure*}
\centering
\includegraphics[width=0.7\textwidth]{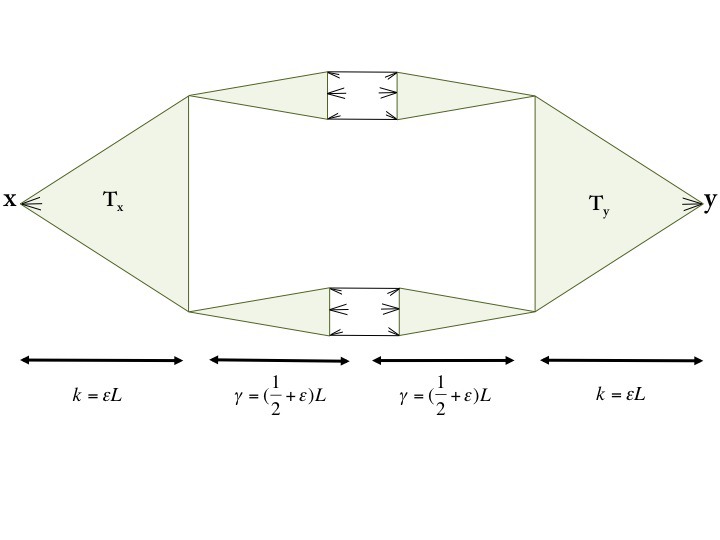}
\caption{Subgraph found in the proof of Lemma~\ref{Rainbowlem4}.}
\label{fig:Rainbowfig2}
\end{figure*}

First we grow $T_x$. As we grow the tree via BFS from a vertex $v$ at depth $i$ to vertices 
at depth $i+1$ certain {\em bad} edges from $v$ may point 
to vertices already in $T_x$. Remark \ref{rem1} shows with probability $1-o(n^{-3})$ there can be at most
10 bad edges emanating from $v$.

 Furthermore, 
Lemma \ref{Rainbowlem2} implies that there exists at most one vertex of degree less than $\frac{\log{n}}{100}$ 
at each level {\it whp}. 
Hence, we obtain the recursion  

\beq{rec1}
D_{i+1}^{(x)} \geq \brac{\frac{\log{n}}{100}-10} (D_i^{(x)}-1) \geq \frac{\log{n}}{101}D_i^{(x)}.
\eeq

\noindent Therefore the number of leaves satisfies 

\beq{rec2}
D_{k}^{(x)} \geq \bfrac{\log n}{101}^{\e L}\geq n^{4\e/5}.
\eeq
We can make the branching factor exactly $\frac{\log n}{101}$ by pruning.
We do this so that the trees $T_x$ are isomorphic to each other.

With a similar argument 
\beq{rec3}
D_{k}^{(y)} \geq n^{\frac{4}{5}\epsilon}.
\eeq
The only difference is that now we also say an edge is bad if the other endpoint is in $T_x$.
This immediately gives
$$D_{i+1}^{(y)} \geq \brac{\frac{\log{n}}{100}-20} (D_i^{(y)}-1) \geq \frac{\log{n}}{101}D_i^{(y)}$$
and the required conclusion \eqref{rec3}.

Similarly, from each leaf $x_i \in T_x$ and $y_i \in T_y$ we grow trees $\hT_{x_i},\hT_{y_i}$ of depth 
$\gamma = \big(\frac{1}{2}+\epsilon\big) L$ using the same procedure and arguments 
as above. Remark \ref{rem1} implies that there are at most 20 edges from the vertex $v$ being explored to 
vertices in any of the trees already constructed. At most 10 to $T_x$ plus any trees rooted at an $x_i$ 
and another 10 for $y$.
The numbers of leaves of each $\hT_{x_i}$ now satisfies
$$\hD_\g^{(x_i)}\geq \frac{\log n}{100}\bfrac{\log n}{101}^{\g}\geq n^{\frac12+\frac{4}{5}\epsilon}.$$
Similarly for $\hD_\g^{(y_i)}$.

\noindent Observe next that BFS does not condition the edges between the leaves $X_i,Y_i$ of the trees $\hT_{x_i}$ 
and $\hT_{y_i}$.
I.e., we do not need to look at these edges in order to carry out our construction. 
On the other hand we have conditioned on the occurence of certain events to imply a certain growth rate.
We handle this technicality as follows. We go through the above construction and halt if ever we find that we cannot
expand by the required amount. Let ${\bf A}$ be the event that we do not halt the construction
i.e. we fail the conditions of Lemmas \ref{Rainbowlem2} or \ref{Rainbowlem3}. We have
$\Prob{{\bf A}}=1-o(1)$ and so,
$$\Prob{\exists i:e(X_i,Y_i)=0\mid {\bf A}}\leq \frac{\Prob{\exists i:e(X_i,Y_i)=0}}{\Pr({\bf A})}\leq
2n^{\frac{4\e}{5}}(1-p)^{n^{1+\frac{8\e}{5}}}\leq n^{-n^\e}.$$
We conclude that {\em whp} there is always an edge between each $X_i,Y_i$ and thus a path of length at most
$(1+2\e)L$ between each $x_i,y_i$.
\end{proof}

\noindent Let $q=(1+5\e)L$ be the number of available colors. 
We color the edges of $G$ randomly. 
We show that the probability of having a rainbow path between $x,y$ in the 
subgraph $G_{x,y}$ of Figure~\ref{fig:Rainbowfig1}  is at least $1-\frac{1}{n^3}$. 

\begin{lemma}
\label{Rainbowlem5}
Color each edge of $G$ using one color at random from $q$ available. Then, the probability of having at least one 
rainbow path 
between two fixed large vertices $x,y \in [n]$ is at least $1-\frac{1}{n^3}$. 
\end{lemma}

\begin{proof} 
We show that the subgraph $G_{x,y}$ contains such a path.
We break our proof into two steps:
 
\begin{figure*}
\centering 
\includegraphics[width=0.7\textwidth]{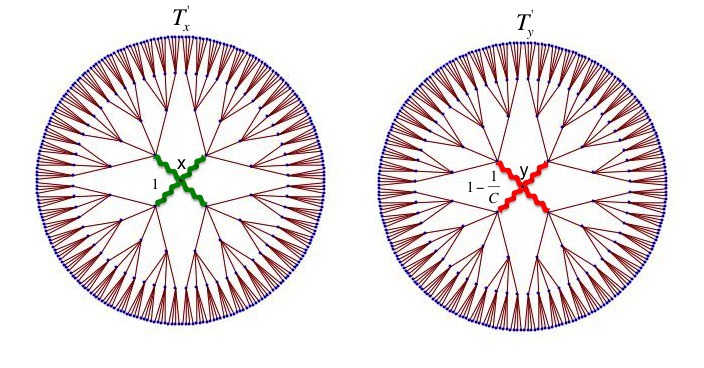}
\caption{\label{fig:Rainbowfig3}Figure shows $\frac{\log{n}}{101}$-ary trees $T_x,T_y$. The two roots are shown respectively at the 
center of the trees. 
In our thinking of the random coloring as an evolutionary process, the green edges incident to $x$ survive
with probability 1, the red edges incident to $y$ with probability $1-\frac{1}{q}$ and all the other
edges with probability $p_0= \Big( 1- \frac{2k}{q} \Big)^2$ where $k$ is the depth of both trees
and $q$ the number of available colors. 
Our analysis in Lemma~\ref{Rainbowlem4} using these probabilities gives a lower bound on the number of 
alive pairs of leaves after coloring $T_x,T_y$ from the root to the leaves respectively.}
\end{figure*}

Before we proceed, we provide certain necessary definitions. 
Think of the process of coloring $T_x,T_y$ as an evolutionary process that colors edges 
by starting from the two roots $x, f(x)=y$ until it reaches the leaves. 
In the following, we call a vertex $u$ of $T_x$ ($T_y$)
{\red {\em alive/living}} if the path $P(x,u)$ ($P(y,u)$) from $x$ ($y$) to $u$ is rainbow, i.e., the edges have received 
distinct colors. 
We  call a pair of vertices $\{u,f(u)\}$ alive, $u \in T_x, f(u) \in T_y$ if $u,f(u)$ are both {\it alive}
and the paths $P(x,u),P(y,f(u))$ share no color. 
Define $A_j=|\{(u,f(u)): (u,f(u)) \text{~is alive and depth}(u)=j  \}|$ for $j=1,..,k$.

\noindent
\underline{$\bullet$ {\sc Step 1:} Existence of at least $n^{\tfrac{4}{5}\epsilon}$ living pairs of leaves}\\

Assume the pair of vertices $\{u,f(u)\}$ is alive where $u \in T_x, f(u) \in T_y$. 
It is worth noticing that $u,f(u)$
have the same depth in their trees. We are interested in the number of pairs of children 
$\{u_i, f(u_i)\}_{i=1,..,\log{n}/101}$ 
that will be alive after coloring the edges from $\depth(u)$ to $\depth(u)+1$. 
A living pair $\{u_i,f(u_i)\}$ by definition has the following properties:
edges $(u,u_i) \in E(T_x)$ and $(f(u),f(u_i))\in E(T_y)$ receive two distinct colors, 
which are different from the set of colors used in paths $P(x,u)$ and $P(y,f(u))$.  
Notice the latter set of colors has cardinality $2 \times \text{depth}(u) \leq 2 k$. 

Let $A_j$ be the number of living pairs at depth $j$. We first bound the size of $A_1$.
\beq{A1}
{\red \Prob{A_1\leq \frac{\log n}{200}}\leq 2^{\log n/101}\bfrac{1}{q}^{\log n/300}=O(n^{-\Omega(\log\log n)}).}
\eeq
{\red Here $2^{\log n/101}$ bounds the number of choices for $A_1$. For a fixed set $A_1$ there will
be at least $\frac{\log n}{101}-\frac{\log n}{200}\geq \frac{\log n}{300}$ edges incident with 
$x$ that have the same color as their corresponding edges incident with $y$, under $f$. The factor $q^{-\log n/300}$
bounds the probability of this event}. 

For $j>1$ we see that the random variable equal to the number of living 
pairs of children of $(u,f(u))$ stochastically dominates
the random variable $X \sim \Bin\brac{\frac{\log{n}}{101}, p_0}$, where $p_0 =  
\brac{1-\frac{2k}{q}}^2 = \big(\frac{1+3\epsilon}{1+5\epsilon}\big)^2$. 
The colorings of the descendants of each live pair are independent and so
we have using the Chernoff bounds for $2\leq j\leq k$,
\begin{multline}\label{indu}
\Prob{A_j < \bfrac{\log{n}}{200}^j p_0^{j-1} \bigg| A_{j-1} \geq \bfrac{\log{n}}{200}^{j-1} p_0^{j-2}} \\
\leq \exp\set{-\frac12\cdot\bfrac{99}{200}^2\cdot\frac{\log n}{101}\cdot\bfrac{\log n}{200}^{j-1}p_0^j}=
O(n^{-\Omega(\log\log n)}).
\end{multline}

\eqref{A1} and \eqref{indu} justify assuming that $A_k \geq \bfrac{\log n}{200}^k
{\red p_0^{k-1}}\ge n^{\frac{4}{5}\epsilon}$.

\noindent \\
\underline{$\bullet$ {\sc Step 2:} Existence of rainbow paths between $x,y$ in $G_{x,y}$}\\
Assuming that there are $\geq n^{4\e/5}$ living pairs of leaves $(x_i,y_i)$ for vertices $x,y$, 
$$\Pr(x,y\text{ are not rainbow connected})\leq \brac{1-\prod_{i=0}^{2\g-1}\brac{1-\frac{2k+i}{q}}}^{n^{4\e/5}}.$$
But
$$\prod_{i=0}^{2\g-1}\brac{1-\frac{2k+i}{q}}\geq \brac{1-\frac{2k+2\g}{q}}^{2\g}=\bfrac{\e}{1+5\e}^{2\g}.$$
So
\begin{multline}\label{lem10}
\Pr(x,y\text{ are not rainbow connected})\leq \exp\set{-n^{4\e/5}\bfrac{\e}{1+5\e}^{2\g}}\\
=\exp\set{-n^{4\e/5-O(\log(1/\e)/\log\log n)}}.
\end{multline}
Using \eqref{eps} and the union bound taking \eqref{lem10} over all large $x,y$ completes the proof of 
Lemma \ref{Rainbowlem5}.
\proofend

\begin{figure}
  \centering
  \includegraphics[width=0.5\textwidth]{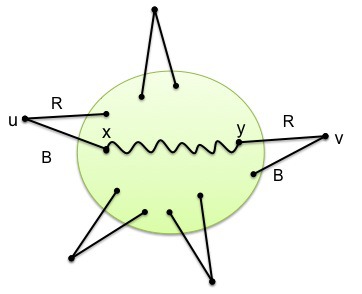}                
  \caption{\label{fig4:4a}Taking care of small vertices.}
  \label{fig:Rainbowfig4}
\end{figure}

\noindent We now finish the proof of Theorem~\ref{thrm:Rainbowmainthrm} i.e. take care of small vertices. 

We showed in Lemma~\ref{Rainbowlem5} that \whp\ for any two large vertices, a random coloring
results in a rainbow path joining them. We divide
the small vertices into two sets: vertices of degree 1, $V_1$ and the
vertices of degree at least 2, $V_2$. Suppose that our colors are $1,2,\ldots,q$
and $V_1=\set{v_1,v_2,\ldots,v_s}$. We begin by giving the edge incident with $v_i$ the color $i$.
Then we slightly modify the argument in Lemma \ref{Rainbowlem5}. If $x$ is the neighbor
of $v_i\in V_1$ then color $i$ cannot be used in Steps 1 and 2 of that procedure. In terms of {\red analysis} this 
replaces $q$ by $(q-1)$ ($(q-2)$ if $y$ is also a neighbor of $V_1$) and the argument is essentially unchanged
i.e. \whp\ there will be a rainbow path between each pair of large vertices.  
Furthermore, any path starting at $v_i$ can only use color $i$ once and so there will be rainbow paths between
$V_1$ and $V_1$ and between $V_1$ and the set of large vertices.

The set $V_2$ is treated by using only two extra colors. 
Assume that Red and Blue have not been used in our coloring. Then we use
Red and Blue to color two of the edges incident to a vertex $u\in V_2$ (the remaining edges are colored arbitrarily).
This is shown in Figure~\ref{fig4:4a}. Suppose that $V_2=\set{w_1,w_2,\ldots,w_t}$. Then if we want a rainbow path 
joining $w_i,w_j$ where $i<j$ then we use the red edge to go to its neighbor $w_i'$. Then we take the already  
constructed rainbow path to $w_j''$, the neighbor of $w_j$ via a blue edge. Then we can continue to $w_j$.
\end{proof}

\section{Proof of Theorem~\ref{thrm:Rainbowregular}} 
\label{sec:Rainbowregular} 

We first observe that simply randomly coloring the edges of $G=G(n,r)$ with $q=n^{o(1)}$ colors will 
not do. This is because there will \whp\ be {\red $\Omega(nq^{1-r^2})=\Omega(n^{1-o(1)})$} vertices $v$ where all 
edges at distance
at most two from $v$ have the same color.

We follow a similar 
strategy to the proof in Theorem \ref{thrm:Rainbowmainthrm}. We grow small trees $T_x$ from each vertex $x$. 
Then for a pair
of vertices $x,y$ we build disjoint trees on the 
leaves of $T_x,T_y$ so that \whp\ we can find edge disjoint paths between any set of leaves $S_x$ of 
$T_x$ and any set of leaves of $S_y$ of the same size. A bounded number of leaves of $T_x,T_y$
will be excluded from this statement. The main difference will
come from our procedure for coloring the edges. Because of the similarities, we will give a little less detail in
the common parts of our proofs. We are in effect
talking about building a structure like that shown in Figure \ref{fig:Rainbowfig2}. There is one difference, we will have
to take care of which leaves of $T_x$ we pair with which leaves of $T_y$, for a pair of vertices $x,y$.

Having grown the trees, we have the problem of coloring the edges.
Instead of independently and randomly coloring the edges, we use a greedy algorithm that produces
a coloring that is guaranteed to color edges differently, if they are close. This will guarantee that 
the edges of $T_x$ are rainbow, for all vertices $x$. We then argue that we can find, for each vertex pair
$x,y,$ a partial mapping $g$ from the leaves of $T_x$ to the leaves of $T_y$ such that the path from $x$ to leaf
$v$ in $T_x$ and the path from $y$ to leaf $g(v)$ in $T_y$ do not share a color. This assumes that $v$ has an image under
the partial mapping $g$. We will have to argue that $g$ is defined on enough vertices in $T_x$. Given this,
we then consider the colors on a set of edge disjoint paths that we can construct from the leaves of $T_x$ to their 
$g$-counterpart in the leaves of $T_y$.

We use the the
 configuration model of Bollob\'as \cite{bollobas2001random} in our proofs, see Chapter~\ref{sec:chapter1networkmodels}

\subsection{Tree building}
We will grow a Breadth First Search tree $T_x$ from each vertex. We will grow each tree to depth 
$$k=k_r=\begin{cases}\rdup{\log_{r-2}\log n}&r\geq 4.\\\rdup{2\log_2\log n-2\log_2\log_2\log n}&r=3.\end{cases}$$
Observe that 

\beq{neq3}
T_x\text{ has at most }r(1+(r-1)+(r-1)^2+\cdots+(r-1)^{k-1})=r\frac{(r-1)^k-1}{r-1}\text{ edges.}
\eeq
It is useful to observe that
\begin{lemma}\label{Rainbowdensity}
Whp, no set of $s\leq \ell_1=\frac{1}{10}\log_{r-1}n$ vertices contains more than $s$ edges.
\end{lemma}
\proofstart 
Indeed,
\begin{align}
\Pr(\exists S\subseteq [n],|S|\leq \ell_1,e[S]\geq |S|+1)&\leq 
\sum_{s=3}^{\ell_1}\binom{n}{s}\binom{\binom{s}{2}}{s+1}\bfrac{r^2}{rn-rs}^{s+1}\label{neq1}\\
&\leq \frac{r\ell_1}{n}
\sum_{s=3}^{\ell_1}\binom{n}{s}\binom{\binom{s}{2}}{s}\bfrac{r^2}{rn-rs}^{s}\nonumber\\
&\leq \frac{r\ell_1}{n}\sum_{s=3}^{\ell_1}\brac{\frac{ne}{s}\cdot \frac{se}{2}\cdot\frac{2r}{n}}^s\nonumber\\
&\leq \frac{r\ell_1}{n}\cdot\ell_1\cdot (e^2r)^{\ell_1}=o(1).\label{neq2}
\end{align}
{\bf Explanation of \eqref{neq1}:} The factor $\bfrac{r^2}{rn-rs}^{s+1}$ can be justified as follows. We can estimate
$$\Pr(e_1,e_2,\ldots,e_{s+1}\in E(G_F))=\\\prod_{i=0}^{s}\Pr(e_{i+1}\in E(G_F)\mid e_1,e_2,\ldots,e_i\in E(G_F))
\leq \bfrac{r^2}{rn-rs}^{s+1}$$
if we pair up the lowest index endpoint of each $e_i$ in some arbitrary order. The fraction $\frac{r^2}{rn-rs}$ 
is an upper bound
on the probability that this endpoint is paired with the other endpoint, regardless of previous pairings.
\proofend

Denote the leaves of $T_x$ by $L_x$.
\begin{corollary}\label{Rainbownlem1}
Whp, {\red $(r-1)^k\leq |L_x|\leq r(r-1)^{k-1}$ for all $x\in [n]$.}
\end{corollary}
\proofstart
This follows from the fact that \whp\ the vertices spanned by each $T_x$ span at most one cycle. 
This in turn follows from Lemma \ref{Rainbowdensity}.
\proofend

\noindent Consider two vertices $x,y \in V(G)$ where $T_x\cap T_y=\emptyset$. We will show that \whp\ we can find 
a subgraph $G'(V',E'), V' \subseteq V, E'\subseteq E$ with 
similar structure to that shown in Figure~\ref{fig:Rainbowfig2}. Here $k=k_r$ and $\gamma=\brac{\frac12+\e}\log_{r-1}n$
for some small positive constant $\e$. 
\begin{remark}\label{Rainbowrem3}
In our analysis we expose the pairing $F$, only as necessary. For example
the construction of $T_x$ involves exposing all pairings involving non-leaves of $T_x$ and one pairing for each leaf.
There can be at most one exception to this statement, for the rare case where $T_x$ contains a unique cycle.
In particular, if we expose the point $q$ paired with a currently unpaired point $p$ of a leaf of $T_x$ then $q$ is chosen
randomly from the remaining unpaired points.
\end{remark}
Suppose that we have 
constructed $i=O(\log n)$ {\red vertex disjoint trees of depth $\g$ rooted at some of the leaves 
of $T_x$}. We grow the $(i+1)$st tree $\hT_z$ via BFS, without using edges that go into $y$ or previously 
constructed trees. 
Let a leaf $z\in L_x$ be {\em bad} if we have to omit a single edge as we construct the first $\ell_1/2$ levels 
of $\hT_z$.
The previously constructed
trees plus $y$ account for $O(n^{1/2+\e})$ vertices and pairings, so the probability that $z$ is bad, given all
the pairings we have exposed so far, is at most 
$O((r-1)^{\ell_1/2}n^{-1/2+\e})=O(n^{-1/3})$. Here bad edges can only join two leaves. This probability bound holds
regardless of whichever other vertices are bad. This follows from the way we build the pairing $F$, see the final
statement of Remark \ref{Rainbowrem3}. So \whp\ there will be at most 3 bad leaves on any 
$T_x$. Indeed, $\Pr(\exists x:x\text{ has }\geq 4
\text{ bad leaves})\leq n\binom{O(\log n)}{4}n^{-4/3}=o(1)$. 

If a leaf is not bad then the first $\ell_1/2$ 
levels produce $\Theta(n^{1/20})$ leaves. 
From this, we see that \whp\ the next $\g-\ell_1$ levels grow at a rate $r-1-o(n^{-1/25})$. Indeed, given that a 
level has $L$ vertices where 
$n^{1/20}\leq L\leq n^{3/4}$, the number of vertices in the next level dominates 
$Bin\brac{(r-1)L,1-O\bfrac{n^{3/4}}{n}}$,
after accounting for the configuration points used in building previous trees. 
Indeed, $(r-1)L$ configuration points associated with good leaves will be unpaired and for each of them,
the probability it is paired with a point associated with a vertex in any of the trees constructed so far is
$O(n^{1/2+2\e}/n)$. This probability bound holds regardless of the pairings of the other leaf configuration points.  
We can thus assert that \whp\ we will have that all but at most three of the leaves $L_x$ of
$T_x$ are roots of vertex disjoint trees $\hT_1,\hT_2,\ldots,$ 
each with $\Theta(n^{1/2+\e/2})$ leaves. Let $L_x^*$ denote these {\em good} leaves.
The same analysis applies when we build trees $\hT_1',\hT_2',\ldots,$ with roots at $L_y$. 

Now the probability that there is no edge joining the leaves of $\hT_i$ to the leaves of $\hT_j'$ is at most
$$\brac{1-\frac{(r-1)\Theta(n^{1/2+\e/2})}{rn}}^{(r-1)
n^{1/2+\e/2}}\leq e^{-\Omega(n^{\e})}.$$ 
To summarise,
\begin{remark}\label{Rainbowrem2}
{\em Whp} we will succeed in finding in $G_F$ and hence in $G=G(n,r)$, for all $x,y\in V(G_F)$,
for all $u\in L_x^*,v\in L_y^*$, a path $P_{u,v}$ from $u$ to $v$ of length
$O(\log n)$ such that if $u\neq u'$ and $v\neq v'$ then
$P_{u,v}$ and $P_{u',v'}$ are edge disjoint. These paths avoid $T_x,T_y$ except at their start and endpoints.
\end{remark}
\subsection{Coloring the edges}
We now consider the problem of coloring the edges of $G$. Let $H$ denote the line graph of 
$G$ and let $\G=H^{2k}$ denote the 
graph with the same vertex set as $H$ and an edge between
vertices $e,f$ of $\G$ if there there is a path of length at most $k$ between $e$ and $f$ in $H$. We will
construct a proper coloring
of $\G$ using 
$$q=10(r-1)^{2k}\sim100\log^{2\th_r}n\text{ where }\th_r=\frac{\log (r-1)}{\log (r-2)}$$ 
colors. We do this as follows: Let $e_1,e_2,\ldots,e_m$ be an 
arbitrary ordering of the vertices of $\G$. For $i=1,2,\ldots,m$, color $e_i$ with a random color, chosen 
uniformly from the set of colors not currrently
appearing on any neighbor in $\G$. At this point only $e_1,e_2,\ldots,e_{i-1}$ will have been colored.

Suppose then that we color the edges of $G$ using the above method. Fix a pair of vertices $x,y$ of $G$.
We see immediately, that no color appears twice in $T_x$ and no color appears twice in $T_y$. 
This is because the distance between edges in $T_x$ is at most $2k$. 
This also deals with the case where $V(T_x)\cap V(T_y)\neq\emptyset$, for the same reason. So assume now that 
$T_x,T_y$ are vertex disjoint.
We can find lots of paths joining $x$ and $y$. We know that the first and last $k$ edges of each
path will be individually rainbow colored. We will first show that we have many choices of path where these $2k$ edges 
are rainbow colored
when taken together. 

\subsection{Case 1: $r\geq 4$:}
We argue now that we can  
find $\s_0=(r-2)^{k-1}$ leaves 
$u_1,u_2,\ldots,u_\t\in T_x$ and $\s_0$ leaves $v_1,v_2,\ldots,v_\t\in T_y$
such for each $i$ the $T_x$ path from $x$ to $u_i$ and the $T_y$ path from $y$ to $v_i$ do not share any colors. 
\begin{lemma}\label{Rainbowlemcol}
Let $T_1,T_2$ be two vertex disjoint copies of an edge colored complete $d$-ary tree with $\ell$ levels, where 
$d\geq 3$. Let $T_1,T_2$ be rooted at $x,y$ respectively.
Suppose that the colorings
of $T_1,T_2$ are both rainbow. Let $\k=(d-1)^{\ell}$. Then there exist leaves $u_1,u_2,\ldots,u_\k$ of $T_1$ and 
leaves 
$v_1,v_2,\ldots v_\k$ of 
$T_2$ such that the following is true: If $P_i,P_i'$ are the paths from $x$ to $u_i$ in $T_1$ and from $y$ to $v_i$ in 
$T_2$ respectively, then 
$P_i\cup P_i'$ is rainbow colored for $i=1,2,\ldots,\k$.
\end{lemma}
\proofstart
Let $A_\ell$ be the minimum number of rainbow path pairs that we can find in any such pair of edge colored trees.
We prove that $A_\ell\geq (d-1)^\ell$ by induction on $\ell$. This is true trivially for $\ell=0$. 
Suppose that $x$ is incident with $x_1,x_2,\ldots,x_d$ and that the sub-tree rooted at $x_i$ is $T_{1,i}$ for 
$i=1,2,\ldots,d$.
Define $y_i$ and $T_{2,i},\,i=1,2,\ldots,d$ similarly with respect to $y$. Suppose that the color of the 
edge $(x,x_i)$ is 
$c_i$ for $i=1,2,\ldots,d$
and let $Q_x=\set{c_1,c_2,\ldots,c_d}$. 
Similarly, suppose that the color of the edge $(y,y_i)$ is 
$c_i'$ for $i=1,2,\ldots,d$
and let $Q_y=\set{c_1',c_2',\ldots,c_d'}$. 
Next suppose that $Q_j$ is the set of colors in $Q_x$ that appear on the 
edges $E(T_{2,j})\cup \{(y,y_j)\}$ .
The sets $Q_1,Q_2,\ldots,Q_d$ are pair-wise disjoint. 
Similarly, suppose that $Q_i'$ is the set of colors in $Q_y$ that appear on the 
edges $E(T_{1,i})\cup \{(x,x_i)\}$.
The sets $Q_1',Q_2',\ldots,Q_d'$ are pair-wise disjoint. 

Now define a bipartite graph $H$ with vertex set $A+B=[d]+[d]$ and an edge $(i,j)$ iff $c_i\notin Q_j$
and $c_j'\notin Q_i'$. We claim that if $S\subseteq A$ then its neighbor set $N_H(S)$ satisfies the inequality
\beq{HM}
d|S|-|N_H(S)|-|S|\leq |S|\cdot |N_H(S)|.
\eeq
Here the LHS of \eqref{HM} bounds from below, the size of the set $S:N_H(S)$ of edges between $S$ and $N_H(S)$. 
This is because there are at most 
$|S|$ edges missing from $S:N_H(S)$ due to $i\in S$ and $j\in N_H(S)$ and $c_i\in Q_j$. 
At most $|N_H(S)|$ edges are missing
for similar reasons. On the other hand, $d|S|$ is the number there would be without these missing edges. The RHS
of \eqref{HM} is a trivial upper bound.

Re-arranging we get that 
$$|N_H(S)|-|S|\geq \rdup{\frac{(d-2-|S|)|S|}{|S|+1}}\geq -1.$$
(We get -1 when $|S|=d$).

Thus $H$ contains a matching $M$ of size $d-1$. Suppose without loss of generality that this matching is 
$(i,i),i=1,2,\ldots,d-1$.
We know by induction that for each $i$ we can find paths $(P_{i,j},\hP_{i,j}),\,j=1,2,\ldots,(d-1)^{\ell-1}$ 
where $P_{i,j}$ is a root to leaf path 
in $T_{1,i}$ and $\hP_{i,j}$ is a root to leaf path 
in $T_{2,i}$  and that $P_{i,j}\cup \hP_{i,j}$ is rainbow for all $i,j$. Furthermore, $(i,i)$ being an edge of $H$, 
means that the edge sets $\set{(x,x_i)}\cup E(P_{i,j})\cup E(\hP_{i,j})\cup \{(y,y_i\}$ are all rainbow. 
\proofend

Let 
$$V_1=\set{x:V(T_x)\text{ contains a cycle}}.$$ 
When $x,y\notin V_1$ we apply this Lemma to $T_x,T_y$ by deleting one of the $r$ sub-trees attached to each of $x,y$
and applying the lemma directly to the $(r-1)$-ary trees that remain. This will yield $(r-2)^k$ pairs
of paths. 
If $x\in V_1$, we delete $r-2$ sub-trees attached to $x$ leaving at least two $(r-1)$-ary trees of depth $k-1$
with roots adjacent to $x$. We can do the same at $y$. Let $c_1,c_2$ be the colors of the two edges 
from $x$ to the roots of these two trees $T_1,T_2$. Similarly, let $c_1',c_2'$ be the colors of the two analogous 
edges from $y$ to the trees $T_1',T_2'$. If color $c_1$ does not appear in $T_1'$ then we apply
the lemma to $T_1$ and $T_1'$. Otherwise, we can apply the lemma to $T_1$ and $T_2'$. In both cases we 
obtain  $(r-2)^{k-1}$ pairs
of paths.  

Accounting for bad vertices we put 
$$\s=\s_0-6=(r-2)^{k-1}-6\geq \frac{\log n}{r-2}-6$$ 
and we see {\red from Remark \ref{Rainbowrem2} 
that we can \whp\ find $\s$ paths $P_1,P_2,\ldots,P_\s$ 
of length $O(\log n)$ from $x$ to $y$. Path $P_i$ goes from $x$ to a leaf $u_i\in L_x^*$ via $T_x$ and then
traverses $Q_i=P(u_i,v_i)$ where {\rred $v_i=\f(u_i)\in L_y^*$} and then goes from $v_i$ to a $y$ via $T_y$. 
Here $\f$ is some partial map from $L_x^*$ to $L_y^*$. It is a random variable that depends on the coloring
$\cC$ of the edges of $T_x$ and $T_y$.
The paths $P_1,P_2,\ldots,P_\s$ depend 
on the choice of $\f$ and hence $\cC$ and so we should write $P_i=P_i(\cC)$.

We fix the coloring $\cC$ and hence $P_1,P_2,\ldots,P_\s$. Let $\cR$ be the event that at least one of the paths
$P_1,P_2,\ldots, P_\s$ is rainbow colored. 
We show that $\Pr(\neg\cR\mid\cC)$ is small.

We let $c(e)$ denote the color of edge $e$ in a given coloring. 
We remark next that for a particular coloring $c_1,c_2,\ldots,c_m$ of the edges $e_1,e_2,\ldots,e_m$ we have
$$\Pr(c(e_i)=c_i,\,i=1,2,\ldots,m)=\prod_{i=1}^m\frac{1}{a_i}$$
where $q-\D\leq a_i\leq q$ is the number of colors available for the color of the edge $e_i$ 
given the coloring so far i.e. the number of colors
unused by the neighbors of $e_i$ in $\G$ when it is about to be colored.

Now fix an edge $e=e_i$ and the colors $c_j,\,j\neq i$. Let $C$ be the set of colors not used by the neighbors of $e_i$ in $\G$.
The choice by $e_i$ of its color under this conditioning is not quite random, but close. Indeed, we claim that for $c,c'\in C$
$$\frac{\Pr(c(e)=c\mid c(e_j)=c_j,\,j\neq i)}{\Pr(c(e)=c'\mid c(e_j)=c_j,\,j\neq i)}\leq \bfrac{q-\D}{q-\D-1}^\D.$$
This is because, changing the color of $e_i$ only affects the number of colors available to neighbors of $e_i$, and only by at most one.

Thus, for $c\in C$, we have
$$\Pr(c(e)=c\mid c(e_j)=c_j,\,j\neq i)\leq \frac{1}{q-\D}\bfrac{q-\D}{q-\D-1}^\D.$$
Now $\D\leq (r-1)^{2k}=q/10$ and we deduce that 
$$\Pr(c(e)=c\mid c(e_j)=c_j,\,j\neq i)\leq \frac{2}{q}.$$
It follows that for $i\in[\s]$,
$$\Pr(P_i\text{ is rainbow colored}\mid \cC,\text{ coloring of }\bigcup_{j\neq i}Q_j)\geq 
\brac{1-\frac{4(k+\g)}{q}}^{2\g}.$$
This is because when we consider the coloring of $Q_i$ there will always be at most $2k+2\g$ colors forbidden by 
non-neighboring edges, if it is to be rainbow colored.

It then follows that
\begin{align*}
\Pr(\neg\cR\mid \cC)&\leq \brac{1-\brac{1-\frac{4(k+\g)}{q}}^{2\g}}^{\s}\\
&\leq \bfrac{8\g(k+\g)}{q}^{\s}\\
&\leq \bfrac{(2+10\e)\log_{r-1}^2n}{10\log^{\th_r}n}^{\s}=o(n^{-2}).
\end{align*}
This completes the proof of Theorem \ref{thrm:Rainbowregular} when $r\geq 4$.

\noindent
{\bf Case 2: $r=3$:}\\
When $r=3$ we can't use $(r-2)^k$ to any effect. Also, we need to increase $q$ to $\log^4n$.
This necessary for a variety of reasons. One reason is that we will reduce $\s$ to $2^{k/2}$.
We want this to be $\Omega(\log n)$ and 
this will force $k$ to (roughly) double what it would have been if we had followed the recipe for $r\geq 4$.
This makes $\D$ close to $\log^4n$ and we need $q\gg\D$. 

And we need to modify the 
argument based on Lemma \ref{Rainbowlemcol}. Instead of inducting on the trees at depth one from the roots $x,y$, 
we now 
induct on the trees at depth two. Assume first that $x,y\notin V_1$.
After ignoring one branch for $T_x$ and $T_y$ we now consider the sub-trees $T_{x,i},T_{y,i},\,i=1,2,3,4$ of 
$T_x,T_y$ whose
roots $x_1,\ldots,x_4$ and $y_1,\ldots,y_4$ are at depth two. We cannot necessarily make this construction when 
$x\in V_1$.
Let $P_i$ be the path from $x$ to $x_i$ in $T_x$ and let $\hP_j$ be the path from $y$ to $y_j$ in $T_y$.
Next suppose that $\widehat{Q}_j$ is the set of colors in $Q$ that appear on the 
edges $E(T_{y,j})\cup E(\hP_j)$.
Similarly, suppose that $Q_i'$ is the set of colors in $Q'$ that appear on the 
edges $\{E(T_{x,i})\cup E(P_i)\}$. 

Re-define $H$ to be the bipartite graph with vertex set $A+B=[4]+[4]$. The edges of $H$ are as before:
$(i,j)$ exists iff $c_i\notin Q_j$ and $c_j'\notin \widehat{Q}_i$. This time we can only say that a color is in at 
most 
two $\widehat{Q}_i$'s 
and
similarly for the $Q_j'$'s.
The effect of this is to replace \eqref{HM} by
$$4|S|-2(|N_H(S)|+|S|)\leq |S|\cdot |N_H(S)|$$
from which we can deduce that
$$|S|-|N_H(S)|\leq \frac{|S|\cdot |N_H(S)|}{2}\leq 2|N_H(S)|.$$
It follows that $|N_H(S)|\geq \rdup{|S|/3}\geq |S|-2$ and so $H$ contains a matching of size two. 
An inductive argument then shows that we are able to find $2^{\rdown{k/2}}$ 
rainbow pairs of paths. The proof now continues as in the case $r\geq 4$, arguing about the coloring of
paths $P_1,P_2,\ldots,P_\s$ where now $\s=2^{\rdown{k/2}}$. 

We finally deal with the vertices in $V_1$. 
We classify them according to the size of the cycle $C_x$ that is contained in $V(T_x)$.
If $T_x$ contains a cycle $C_x$ then necessarily 
$|C_x|\leq 2k$ and so there are at most $2k$ types in our classification. 
It follows from Lemma \ref{Rainbowdensity} that if $x,y\in V_1$ and $T_x\cap T_y\neq\emptyset$
then $C_x=C_y$ \whp. Note next that the distance from $x$ to $C_x$ is at most $k-|C_x|/2$. If $C$ is a cycle
of length at most $2k$, let $V_C=\set{x:C=C_x}$ and let $E_C$ be the set of edges contained in $V_C$. We have
\beq{VK}
|V_C|=O( |C|2^{k-|C|/2})=O(2^k)=O(\log^2n/\log\log n).
\eeq
We introduce $2k$ new sets $\hC_i,i=3,4,\ldots,2k$ of
$O(\log^2n/\log\log n)$ colors, distinct from $Q$. Thus we introduce $O(\log^2n)$ new colors overall.
We re-color each $E_C$ with the colors from $\hC_{|C|}$. It is important to observe that if $|C|=|C'|$ then the graphs 
induced by
$V_C$ and $V_{C'}$ are isomorphic and so we can color them isomorphically. By the latter we mean that we choose some 
isomorphism $f$ from $V_C$ to $V_{C'}$ and then if $e$ is an edge
of $V_C$ then we color $e$ and $f(e)$ with the same color. After this re-coloring, we see that if $T_x$ and $T_y$ are 
not vertex disjoint,
then they are contained in the same $V_C$. The edges of $V_C$ are rainbow colored and so now we only need to concern 
ourselves with
$x,y\in V_1$ such that $T_x$ and $T_y$ are vertex disjoint. Assume now that $x,y\in V_1$.

Assume first that $x,y$ are of the same type and that they are at the same distance from $C_x,C_y$ respectively.
Our aim now is to define binary trees $T_x',T_y'$ ``contained`` 
in $T_x,T_y$ that can be used as in Lemma \ref{Rainbowlemcol}. 
If we delete an edge $e=(u,v)$ of $C_x$ then the graph that remains on $V(T_x)$
is a tree with at most two vertices $u,v$ of degree two. Now delete one of the three sub-trees of $T_x$.
If there are vertices of degree two, make sure one of them is in this sub-tree. 
If necessary, shrink the path of length two with the remaining vertex of degree two in the middle to an edge $e_x$. 
It has leaves at depth $k-1$ and leaves at depth $k-2$. 
The resulting binary tree will be our $T_x'$. The leaves at depth $k-1$ come in pairs. Delete one vertex from each 
pair
and shrink the paths of length two through the vertex at depth $k-2$ to an edge.

The edges that are obtained by shrinking paths of length two will have two colors. 
Because $x,y$ are at the same distance from their cycles, we can delete $f(e)$ from $C_y$ and do the construction so 
that
$T_x'$ and $T_y'$ will be isomorphically colored.

It is now easy to find $2^{k-2}$ pairs of paths whose unions are rainbow colored. Each leaf of $T_x,T_y$
can be labelled by a $\{0,1\}$ string of length $k-2$. We pair string $\xi_1\xi_2\cdots\xi_{k-1}\xi_{k-2}$ in $T_x$
with $(1-\xi_1)\xi_2\cdots\xi_{k-1}\xi_{k-2}$ in $T_y$. The associated paths will have a rainbow union.
The proof now continues as in the case $r\geq 4$, arguing about the coloring of
paths $P_1,P_2,\ldots,P_\s$ where now $\s=2^{k-2}$.

If $x$ is further from $C_x$ than $y$ is from $C_y$ then let $z$ be the vertex on the path from $x$ to $C_x$ at the
same distance from $C_x$ as $y$ is from $C_y$. We have a rainbow path from $z$ to $y$ and adding the $T_x$ path 
from $x$ to $z$
gives us a rainbow path from $x$ to $y$. This relies on the fact that $V_{C_x}$ and $V_{C_y}$ are isomorphically
colored.

If $x,y$ are of a different type, then $T_x$ and $T_y$ are re-colored with distinct colors and we can proceed as 
as in the case $r\geq 4$, arguing about the coloring of
paths $P_1,P_2,\ldots,P_\s$ where now $\s=2^{k}$, using Corollary \ref{Rainbownlem1}.

If $x\in V_1$ and $y\notin V_1$ then we can proceed as if 
both are not in $V_1$.
This is because of the re-coloring of the edges of $T_x$. We can proceed as 
as in the case $r\geq 4$, arguing about the coloring of
paths $P_1,P_2,\ldots,P_\s$ where now $\s=2^{k}$, using Corollary \ref{Rainbownlem1}.

This completes our proof of Theorem \ref{thrm:Rainbowregular}.

\noindent We conclude this Chapter with mentioning that if the degree 
$r$ in Theorem \ref{thrm:Rainbowregular} is allowed to grow as fast 
as $\log n$ then one can prove a result closer to that of Theorem \ref{thrm:Rainbowmainthrm}.

\newpage 
\clearpage
\chapter{Random Apollonian networks}
\label{ranchapter}
\lhead{\emph{Random Apollonian networks}} 
\section{Model \& Main Results}
\label{sec:RANmainresults}

As we outlined in Chapter~\ref{introchapter} planar graphs model several significant types of 
spatial real-world networks such as power grids and  road networks. 
Despite the outstanding amount of work  on modeling real-world networks with random graph models 
\cite{aiello2000random, albert, borgs2010hitchhiker,borgs2007first,lattanzi2009affiliation,
leskovec2005realistic,mahdian2007stochastic,flaxman2006geometric,flaxman2007geometric,durrett2007random,fabrikant}, 
real-world planar graph generators have received considerably less attention. 
In this Chapter we focus on Random Apollonian Networks (RANs),
a popular random graph model for generating planar graphs with power law properties \cite{maximal}. 
Before we state our main results we briefly describe the model.

{\bf Model:} An example of a RAN is shown in Figure~\ref{fig:RANfig1}. At time $t=1$ 
the RAN is shown in Figure~\ref{fig:RANfig1}(a). At each step $t \geq 2$ 
a face $F$ is chosen uniformly at random among the faces of $G_t$. Let $i,j,k$
be the vertices of $F$. We add a new vertex inside $F$ and we connect it to $i,j,k$. 
Higher dimensional RANs also exist where instead of triangles we have
$k$-simplexes $k\geq 3$, see \cite{zhang}. 
It is easy to see that the number of vertices $n_t$, edges $m_t$ and faces $F_t$ at time $t \geq 1$ in a RAN $G_t$
satisfy:
$$n_t=t+3,~~ m_t=3t+3,~~ F_t=2t+1.$$ 
\noindent Note that a RAN is a maximal planar graph since for any planar graph $m_t \leq 3n_t - 6 \leq 3t+3$.

Surprisingly, despite the popularity of the model 
various important properties have been analyzed experimentally and heuristically with lack of rigor. 
In this Chapter, we prove the following theorems. 

\begin{theorem}[Degree Sequence] 
\label{thrm:RANdegreesequence} 

\noindent Let $Z_k(t)$ denote the number of vertices of degree $k$ at time $t$, $k \geq 3$.
 For any $t \geq 1$  and any $k \geq 3$ there exists a constant $b_k$ depending on $k$ such that 

$$ |\Mean{ Z_k(t) } - b_k t| \leq K, \text{~~where~~} K=3.6.$$

\noindent Furthermore, for $t$ sufficiently large and any $\lambda > 0$ 

\begin{equation}
\label{eq:degree}
\Prob{|Z_k(t) - \Mean{Z_k(t)}| \geq \lambda } \leq e^{-\frac{\lambda^2}{72t}}.
\end{equation} 

\end{theorem} 

\noindent For previous weaker results on the degree sequence see \cite{comment,maximal}.
An immediate corollary which proves strong concentration of $Z_k(t)$ around its expectation
is obtained from Theorem~\ref{thrm:RANdegreesequence} and a union bound by setting $\lambda=10\sqrt{t\log{t}}$. 
Specifically: 

\begin{corollary}
For all possible degrees $k$ $$\Prob{|Z_k(t) -  \Mean{Z_k(t)} | \geq 10\sqrt{t\log{t}}} = o(1).$$
\end{corollary}

\noindent The next theorem provides insight into the asymptotic growth of the highest degrees of RANs
and is crucial in proving Theorem~\ref{thrm:RANthrm2}. 

\begin{figure*}
  \centering
   \begin{tabular}{cccc}
  \includegraphics[width=0.2\textwidth]{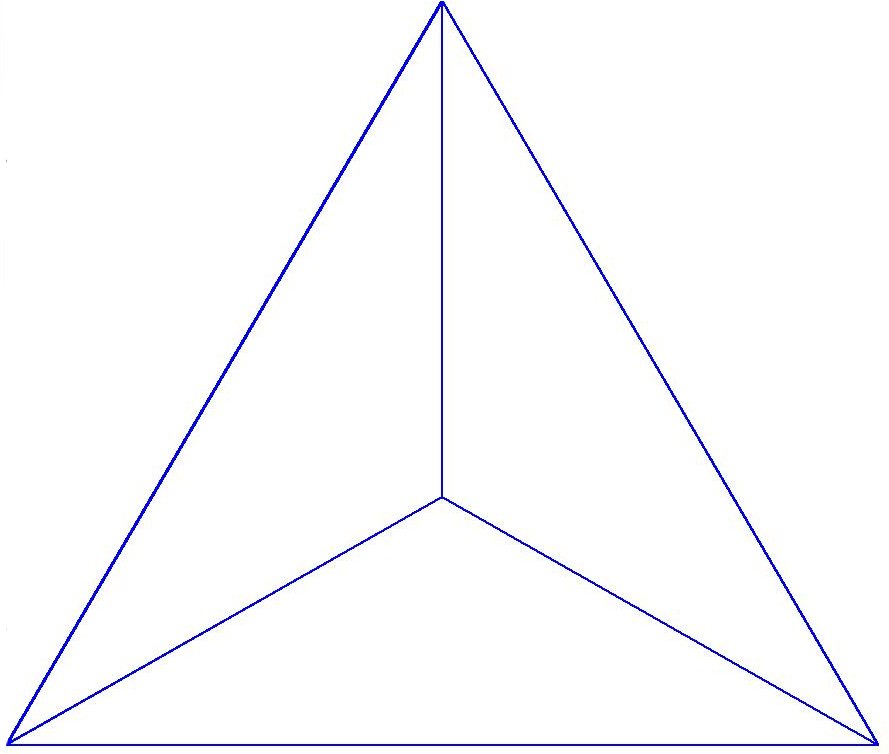} &
  \includegraphics[width=0.2\textwidth]{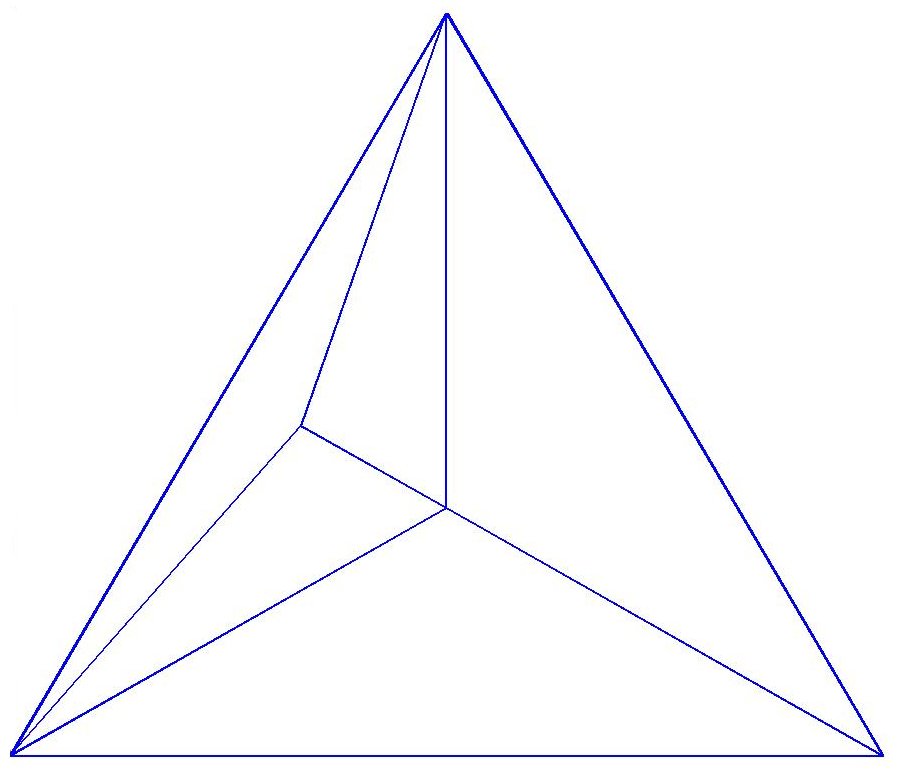} &
  \includegraphics[width=0.2\textwidth]{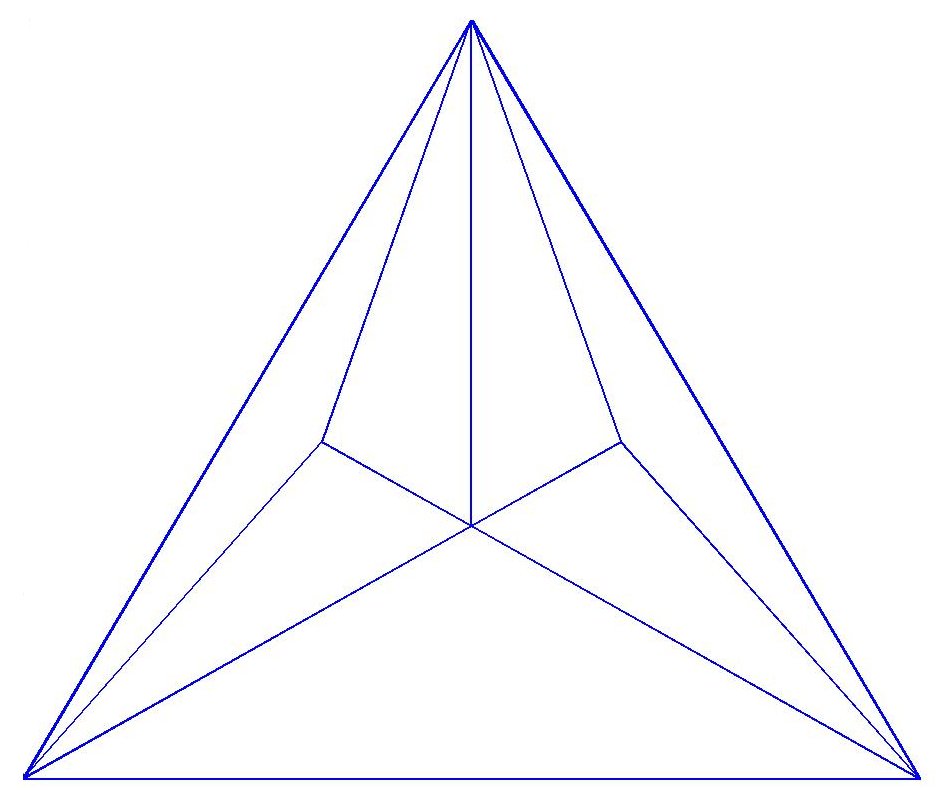} &
  \includegraphics[width=0.2\textwidth]{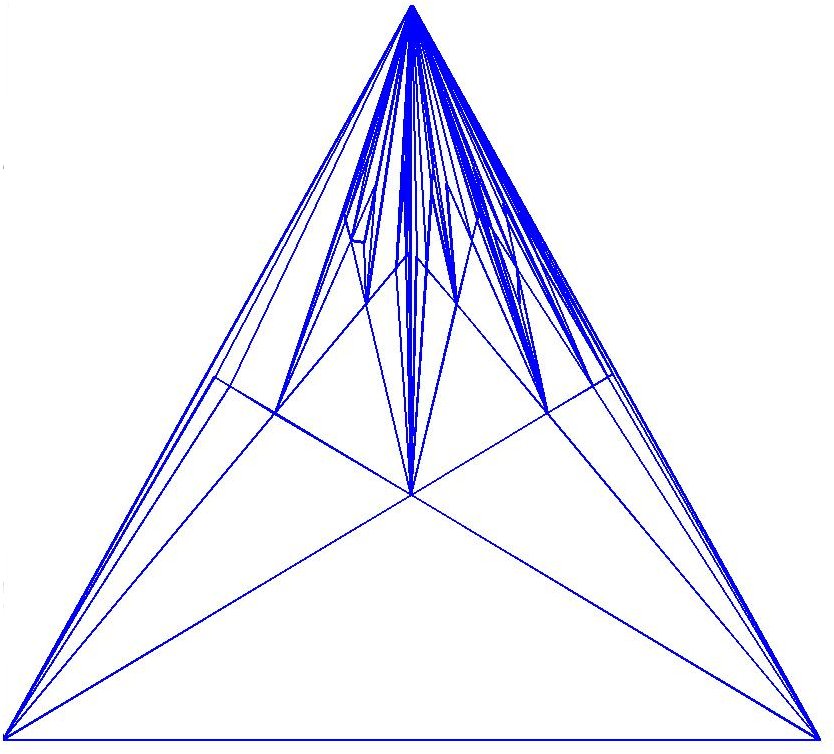} \\ 
   (a) & (b) & (c)  & (d) \\
  \end{tabular} 
  \caption{\label{fig:RANfig1}Snapshots of a Random Apollonian Network (RAN) at: (a) $t=1$ (b) $t=2$ (c) $t=3$ (d) $t=100$. }
\end{figure*}

\begin{theorem}[Highest Degrees]
Let $\Delta_1 \geq \Delta_2 \geq \ldots \geq \Delta_k$ be the $k$ highest degrees of the RAN $G_t$ 
at time $t$ where $k$ is a fixed positive integer. Also, let $f(t)$ be a function such that $f(t) \rightarrow +\infty$ as $t \rightarrow +\infty$. 
Then \whp\footnote{An event $A_t$ holds with high probability (\whp) if $\displaystyle\lim_{t \to +\infty} \Prob{A_t}=1$.}

$$ \frac{t^{1/2}}{f(t)} \leq \Delta_1 \leq t^{1/2} f(t) $$ 

\noindent and for $i=2,\ldots,k$

$$  \Delta_{i-1} - \Delta_i \geq \frac{t^{1/2}}{f(t)}. $$ 
\label{thrm:RANthrm1} 
\end{theorem}

\noindent The growing function $f(t)$ cannot be removed, see \cite{flaxman}.
Using Theorem~\ref{thrm:RANthrm1} and the technique of Mihail and Papadimitriou \cite{mihail} 
we show how the top eigenvalues of the adjacency matrix representation of a RAN grow
asymptotically as $t \rightarrow +\infty$ \whp. 

\begin{theorem}[Largest Eigenvalues] 
Let $k$ be a fixed positive integer. Also, let $\lambda_1 \geq \lambda_2 \geq \ldots \geq \lambda_k$ be the largest $k$ eigenvalues
of the adjacency matrix of $G_t$. Then \whp ~$ \lambda_i =(1\pm o(1))\sqrt{\Delta_i}.$
\label{thrm:RANthrm2} 
\end{theorem}

\noindent Also, we show the following refined upper bound for the asymptotic growth of the diameter. 

\begin{theorem}[Diameter] 
The diameter $d(G_t)$ of $G_t$ satisfies in probability $ d(G_t) \leq \rho \log{t}$
where $\frac{1}{\rho}=\eta$ is the unique solution less than 1 of the equation $\eta - 1 - \log{\eta} = \log{3}$.
\label{thrm:RANternary} 
\end{theorem}

\noindent A straight-forward calculation of $\eta$, gives us the following corollary. 

\begin{corollary}[Diameter] 
The diameter $d(G_t)$ of $G_t$ satisfies asymptotically 

 $$\Prob{ d(G_t) > 7.1 \log{t} } \rightarrow 0.$$

\label{thrm:RANcorrternary} 
\end{corollary}

The outline of this Chapter is as follows:  in Section~\ref{sec:RANrelated} we present briefly related work  and technical preliminaries 
needed for our analysis. We prove Theorems~\ref{thrm:RANdegreesequence},~\ref{thrm:RANthrm1},~\ref{thrm:RANthrm2}
and \ref{thrm:RANternary} in Sections~\ref{sec:RANdegreeseq},~\ref{sec:RANdegrees},~\ref{sec:RANeigen} and ~\ref{sec:RANdiam} respectively. 

\section{Related Work}
\label{sec:RANrelated}

Apollonius of Perga was a Greek geometer and astronomer noted for his writings on conic sections.
He introduced the problem of space filling packing of spheres whose classical solution,
the so-called Apollonian packing \cite{graham}, exhibits a power law behavior. Specifically, 
the circle size distribution follows a power law with exponent around 1.3 \cite{boyd}.
Apollonian Networks (ANs)  were introduced in \cite{adrade} and independently in \cite{doye}. 
Zhou et al. \cite{maximal} introduced Random Apollonian Networks (RANs). 
Their degree sequence was analyzed inaccurately in \cite{maximal} (see comment in \cite{comment})
and subsequently using physicist's methodology in \cite{comment}. 
Eigenvalues of RANs have been studied only experimentally \cite{adrade2}. 
Concerning the diameter of RNAs it has been shown to grow logarithmically \cite{maximal}
using heuristic arguments (see for instance equation B6, Appendix B in \cite{maximal}).
RANs are planar 3-trees, a special case of random $k$-trees \cite{kloks}. 
Cooper and Uehara \cite{cooper} and Gao \cite{gao} analyzed the degree distribution of random $k$-trees, a closely 
related model to RANs. In RANs --in contrast to random $k$-trees-- the random $k$ clique 
chosen at each step has never previously been selected. For example, in the two dimensional 
RAN any chosen face is being subdivided into three new faces by connecting
the incoming vertex to the vertices of the boundary. 
Random $k$-trees due to their power law properties have been proposed as a 
model for complex networks, see, e.g., \cite{cooper,gao2} and references therein.
Recently, a variant of $k$-trees, namely ordered increasing $k$-trees has 
been proposed and analyzed in \cite{panholzer}.
Closely related to RANs but not the same are random Apollonian network structures
which have been analyzed by Darrasse, Soria et al. \cite{darase2,darase,darase3}. 

Bollob\'{a}s, Riordan,  Spencer and Tusn\'{a}dy \cite{bollobasdegrees} proved rigorously
the power law distribution of the Barab\'{a}si-Albert model \cite{albert}. 
Chung, Lu, Vu \cite{chung} Flaxman, Frieze, Fenner \cite{flaxman} and Mihail, Papadimitriou
\cite{mihail} have proved rigorous results for eigenvalue related properties
of real-world graphs using various random graph models. 

\section{Proof of Theorem~\ref{thrm:RANdegreesequence}}
\label{sec:RANdegreeseq}

We decompose our proof in a sequence of Lemmas. 
For brevity let $N_k(t)=\Mean{Z_k(t)}$, $k \geq 3$.
Also, let $d_v(t)$ be the degree of vertex $v$ at time $t$
and  $\mathbf{1}(d_v(t)=k)$ be an indicator variable which equals
1 if $d_v(t)=k$, otherwise 0. Then, for any $k\geq 3$ we
can express the expected number $N_k(t)$ of vertices of degree $k$ 
as a sum of expectations of indicator variables:

\begin{equation}
N_k(t) = \sum_{v} \Mean{\mathbf{1}(d_v(t)=k)}.
\label{eq:RANeqindicators}
\end{equation}
 
\noindent We distinguish two cases in the following. 
\newline

\noindent \underline{$\bullet$ {\sc Case 1:} $k=3$:}\\

\noindent Observe that a vertex of degree 3 is created only by an insertion of a new vertex. 
The expectation $N_3(t)$ satisfies the following recurrence\footnote{The 
three initial vertices participate in one less face than their degree. However,
this leaves our results unchanged.}

\begin{equation}
N_3(t+1) = N_3(t)+1-\frac{3N_3(t)}{2t+1}.
\label{eq:RANbasis}
\end{equation}

\noindent The basis for Recurrence~\eqref{eq:RANbasis} is $N_3(1)=4$. We prove the following lemma
which shows that $\displaystyle\lim_{t \to +\infty}{\frac{N_3(t)}{t}} = \frac{2}{5}$.

\begin{lemma} 
$N_3(t)$ satisfies the following inequality:
\begin{equation} 
|N_3(t) - \frac{2}{5} t|\leq K, \text{~~where~~} K=3.6
\label{eq:RANdk3}
\end{equation} 
\label{lem:RANdk3}
\end{lemma}

\begin{proof} 
We use induction. Assume that $N_3(t)=\frac{2}{5}t+e_3(t)$, where $e_3(t)$ stands for the error term.
We wish to prove that for all $t$, $|e_3(t)|\leq K$. 
The result trivially holds for $t=1$. We also see that for $t=1$ inequality \eqref{eq:RANdk3} is tight. 
Assume the result holds for some $t$. We show it holds for $t+1$.

\begin{align*}
N_3(t+1)  &= N_3(t)+1-\frac{3N_3(t)}{2t+1} \Rightarrow \\
e_3(t+1)    &= e_3(t) + \frac{3}{5} - \frac{6t+15e_3(t)}{10t+5} = e_3(t)\Big(1- \frac{3}{2t+1}\Big) + \frac{3}{5(2t+1)}  \Rightarrow \\ 
|e_3(t+1)|  &\leq K(1-\frac{3}{2t+1}) + \frac{3}{5(2t+1)} \leq K \\ 
\end{align*}

\noindent Therefore inductively Inequality~\eqref{eq:RANdk3} holds for all $t \geq 1$. 
\end{proof}

\noindent \underline{$\bullet$ {\sc Case 2:} $k \geq 4$:}\\

\noindent  For $k \geq 4$ the following holds: 

\begin{equation} 
\Mean{\mathbf{1}(d_v(t+1)=k)} = \Mean{ \mathbf{1}(d_v(t)=k)} (1-\frac{k}{2t+1}) + \Mean{\mathbf{1}(d_v(t)=k-1)} \frac{k-1}{2t+1}
\label{eq:RANdegree}
\end{equation}

\noindent  Therefore, we can rewrite Equation~\eqref{eq:RANeqindicators} for $k \geq 4$ as follows:

\begin{equation}
N_k(t+1) = N_k(t)(1-\frac{k}{2t+1})+ N_{k-1}(t)\frac{k-1}{2t+1}
\label{eq:RANeqindicators2}
\end{equation}


\begin{lemma} 
For any $k \geq 3$, the limit $\displaystyle\lim_{t \to +\infty} \frac{N_k(t)}{t}$ exists.
Specifically, let $b_k=\displaystyle\lim_{t \to +\infty} \frac{N_k(t)}{t}$. Then,
$b_3=\frac{2}{5}, b_4=\frac{1}{5}, b_5=\frac{4}{35}$ and for $k\geq 6$
$b_k = \frac{24}{k(k+1)(k+2)}$.
Furthermore, for all $k \geq 3$ 

\begin{equation} 
|N_k(t) - b_k t| \leq K, \text{~~where~~} K=3.6.
\label{eq:RANdk4}
\end{equation}
\label{lem:RANdegreelemma}

\end{lemma}

\begin{proof} 
For $k=3$ the result holds by Lemma~\ref{lem:RANdk3} and specifically $b_3=\frac{2}{5}$. Assume the result holds for some $k$.
We show that it holds for $k+1$ too. Rewrite Recursion~\eqref{eq:RANeqindicators2} as: 
$ N_k(t+1) = (1- \frac{b_t}{t+t_1})N_k(t) + c_t$ where $b_t=k/2$, $t_1=1/2$,
$c_t = N_{k-1}(t)\frac{k-1}{2t+1}$. Clearly $\displaystyle\lim_{t \rightarrow +\infty}{b_t}=k/2>0$
and $\displaystyle\lim_{t \rightarrow +\infty}{c_t}= \displaystyle\lim_{t \rightarrow +\infty}{ b_{k-1}t \frac{k-1}{2t+1}} = b_{k-1}(k-1)/2$.
Hence by Lemma~\ref{lem:RANchunglubook}:

$$ \lim_{t \rightarrow +\infty} \frac{N_k(t)}{t} = \frac{ (k-1)b_{k-1}/2 }{1+k/2}= b_{k-1} \frac{k-1}{k+2}. $$ 

\noindent Since $b_3=\frac{2}{5}$ we obtain that $b_4=\frac{1}{5}$, $b_5=\frac{4}{35}$
for any $k \geq 6$, $b_k = \frac{24}{k(k+1)(k+2)}$. This shows that the degree sequence of RANs follows a power law
distribution with exponent 3. 

Now we prove Inequality~\eqref{eq:RANdk4}.
The case $k=3$ was proved in Lemma~\ref{lem:RANdk3}. Let $e_k(t)=N_k(t)-b_kt$. 
Assume the result holds for some $k \geq 3$, i.e., $|e_k(t)| \leq K$ 
where $K=3.6$. We show it holds for $k+1$ too.  Substituting in Recurrence~\eqref{eq:RANeqindicators}
and using the fact that $b_{k-1}(k-1)=b_k(k+2)$ we obtain the following: 

\begin{align*}
e_k(t+1) &= e_k(t) + \frac{k-1}{2t+1} e_{k-1}(t) - \frac{k}{2t+1} e_k(t) \Rightarrow \\
|e_k(t+1)| &\leq | (1-\frac{k}{2t+1}) e_k(t) | + | \frac{k-1}{2t+1} e_{k-1}(t) | \leq K(1-\frac{1}{2t+1}) \leq K
\end{align*}

\noindent Hence by induction, Inequality~\eqref{eq:RANdk4} holds for all $k \geq 3$. 
\end{proof}

\noindent Using integration and a first moment argument, 
it can be seen that Lemma~\ref{lem:RANdegreelemma} agrees 
with Theorem~\ref{thrm:RANthrm1} where it is shown that the maximum degree is $\approx t^{1/2}$. 
(While $b_k=O(k^{-3})$ suggests a maximum degree of order $t^{1/3}$, summing $b_k$ over $k\geq K$
suggests a maximum degree of order $t^{1/2}$).

Finally, the next Lemma proves the concentration of $Z_k(t)$ around its expected
value for $k \geq 3$. This lemma applies Lemma~\ref{lem:RANazuma} and 
completes the proof of Theorem~\ref{thrm:RANdegreesequence}.

\begin{lemma} 
Let $\lambda >0$. For $k\geq 3$

\begin{equation}
\Prob{|Z_k(t) - \Mean{Z_k(t)}| \geq \lambda } \leq e^{-\frac{\lambda^2}{72t}}.
\end{equation}

\label{lem:RANdeg} 
\end{lemma}

\begin{proof}

Let $(\Omega, \mathcal{F}, \field{P})$ be the probability space induced
by the construction of a RAN after $t$ insertions. Fix $k$, where $k \geq 3$, and let $(X_i)_{i \in \{0,1,\ldots,t\}}$
be the martingale sequence defined by $X_i = \Mean{Z_k(t)|\mathcal{F}_i}$,
where $\mathcal{F}_0=\{\emptyset,\Omega\}$ and 
$\mathcal{F}_i$ is the $\sigma$-algebra generated by the RAN process 
after $i$ steps. Notice $X_0=\Mean{Z_k(t)|\{\emptyset,\Omega\}}=N_k(t)$, $X_t = Z_k(t)$.
We show that $|X_{i+1}-X_{i}| \leq 6$ for $i=0,\ldots,t-1$. 
Let $P_j=(Y_1,\ldots,Y_{j-1},Y_j)$, $P_j'=(Y_1,\ldots,Y_{j-1},Y_j')$  be two sequences of 
face choices differing only at time $j$. Also, let $\bar{P},\bar{P'}$ continue from $P_j,P_j'$ 
until $t$. We call the faces $Y_j,Y_j'$ special with respect to $\bar{P},\bar{P'}$. 
We define a measure preserving map $\bar{P} \mapsto \bar{P'}$ in the following way:
for every choice of a non-special face in process $\bar{P}$ at time $l$ we make the same face choice in $\bar{P'}$ 
at time $l$. For every choice of a face inside the special face $Y_j$ in process $\bar{P}$ 
we make an isomorphic (w.r.t., e.g., clockwise order and depth) 
choice of a face inside the special face $Y_j'$ in process $\bar{P}'$.
Since the number of vertices of degree $k$ can change by at most 6, i.e., the (at most) 6 vertices involved
in the two faces $Y_j,Y_j'$ the following holds: 

$$ |\Mean{Z_k(t)|P} - \Mean{Z_k(t)|P'}| \leq 6 .$$

Furthermore,  this holds for any $P_j,P_j'$. 
We deduce that $X_{i-1}$ is a weighted mean of values, whose pairwise differences are all at most 6. 
Thus, the distance of the mean $X_{i-1}$ is at most 6 from each of these values.
Hence, for any one step refinement  $|X_{i+1}-X_{i}| \leq 6$ $\forall i\in \{0,\ldots,t-1\}$. 
By applying the  Azuma-Hoeffding inequality as stated in Lemma~\ref{lem:RANazuma} we obtain

\begin{equation}
\Prob{|Z_k(t) - \Mean{Z_k(t)}| \geq \lambda } \leq 2e^{-\frac{\lambda^2}{72t}}.
\end{equation}
 
\end{proof}

\section{Proof of Theorem~\ref{thrm:RANthrm1}}
\label{sec:RANdegrees}

We decompose the proof of Theorem~\ref{thrm:RANthrm1} into several lemmas
which we prove in the following. Specifically, the proof follows directly 
from Lemmas~\ref{lem:RANlemma7},~\ref{lem:RANlemma8},~\ref{lem:RANlemma9},~\ref{lem:RANlemma10},~\ref{lem:RANlemma11}.
We partition the vertices into three sets: those added before $t_0$, between $t_0$ and
$t_1$ and after $t_1$ where $t_0=\log{\log{\log{(f(t))}}}$ and $t_1=\log{\log{(f(t))}}$.
Recall that $f(t)$ is a function such that $\displaystyle\lim_{t \to +\infty}f(t)=+\infty$. 
We define a supernode to be a collection of vertices and the degree of the supernode the sum of the degrees 
of its vertices. 

\begin{lemma} 
Let $d_t(s)$ denote the degree of vertex $s$ at time $t$. and let $a^{(k)}=a(a+1)\ldots (a+k-1)$\
denote the rising factorial function. Then, for any positive integer $k$

\begin{equation}
\Mean{ d_t(s)^{(k)} } \leq \frac{ (k+2)!}{2} \big( \frac{2t}{s} \big)^{\frac{k}{2}}.
\label{eq:rising} 
\end{equation}

\label{lem:RANlemma6}
\end{lemma}

\begin{proof}

As we mentioned in the proof of Theorem~\ref{thrm:RANdegreesequence} 
the three initial vertices $1,2,3$  have  one less face than their 
degree whereas all other vertices have degree equal to the number of faces
surrounding them.
In this proof we treat both cases but we omit it in all other proofs.

\noindent \underline{$\bullet$ {\sc Case 1:} $s \geq 4$}\\

\noindent  Note that $d_s(s)=3$. By conditioning successively we obtain 

\begin{align*} 
\Mean{d_t(s)^{(k)}}                &= \Mean{ \Mean{ d_t(s)^{(k)} | d_{t-1}(s)} }\\
                                   &= \Mean{ (d_{t-1}(s))^{(k)} \big( 1-\frac{d_{t-1}(s)}{2t-1} \big) +  (d_{t-1}(s)+1)^{(k)} \frac{d_{t-1}(s)}{2t-1}} \\ 
                                   &=  \Mean{  (d_{t-1}(s))^{(k)} \big( 1-\frac{d_{t-1}(s)}{2t-1} \big) + (d_{t-1}(s))^{(k)} \frac{d_{t-1}(s)+k}{d_{t-1}(s)} \frac{d_{t-1}(s)}{2t-1}}\\
                                   &=  \Mean{ (d_{t-1}(s))^{(k)}}\big( 1+\frac{k}{2t-1} \big) = ... = 3^{(k)} \prod_{t'=s+1}^t ( 1+\frac{k}{2t'-1} ) \\
                                   &\leq  3^{(k)}  \exp{\Big( \sum_{t'=s+1}^t \frac{k}{2t'-1} \Big)} \leq 3^{(k)} \exp{ \Big( k \int_{s}^{t} \! \frac{\mathrm{d}x}{2x-1}   \,   \Big) }  \\ 
                                   &\leq \frac{(k+2)!}{2} \exp{ \Big( \tfrac{k}{2} \log{ \frac{t-1/2}{s-1/2} } \Big)} \leq \frac{ (k+2)!}{2} \Big( \frac{2t}{s} \Big)^{\frac{k}{2}}. 
\end{align*}

\noindent \underline{$\bullet$ {\sc Case 2:} $s \in \{1,2,3\}$}\\

\noindent  Note that initially the degree of any such vertex is 2. For any $k\geq 0$

\begin{align*} 
\Mean{d_t(s)^{(k)}}                &= \Mean{ \Mean{ d_t(s)^{(k)} | d_{t-1}(s)} }\\
                                   &= \Mean{ (d_{t-1}(s))^{(k)} \big( 1-\frac{d_{t-1}(s)-1}{2t-1} \big) +  (d_{t-1}(s)+1)^{(k)} \frac{d_{t-1}(s)-1}{2t-1} }\\ 
                                   &=  \Mean{ (d_{t-1}(s))^{(k)} \big( 1+\frac{k}{2t-1} \big) - (d_{t-1}(s))^{(k)} \frac{k}{(2t-1)d_{t-1}(s)}  }\\
                                   &\leq \Mean{ (d_{t-1}(s))^{(k)}} \big( 1+\frac{k}{2t-1} \big) \leq \ldots \leq \frac{ (k+2)!}{2} \big( \frac{2t}{s} \big)^{\frac{k}{2}}.
\end{align*}

\end{proof}

\begin{lemma}
The degree $X_t$ of the supernode $V_{t_0}$ of vertices added before time $t_0$ is at least $t_0^{1/4}\sqrt{t}$ \whp.
\label{lem:RANlemma7}
\end{lemma}

\begin{proof} 

We consider a modified process $\mathcal{Y}$ coupled with the RAN process, see also Figure~\ref{fig:fig3}.
Specifically, let $Y_t$ be the modified degree of the supernode in the modified process $\mathcal{Y}$ which 
is defined as follows: for any type of insertion in the original RAN process --note there exist three types of insertions with respect 
to how the degree $X_t$ of the supernode (black circle) gets affected, see also Figure~\ref{fig:fig3}-- 
$Y_t$ increases by 1. We also define $X_{t_0}=Y_{t_0}$.  Note that $X_t \geq Y_t$ for all $t \geq t_0$.
Let $d_0=X_{t_0}=Y_{t_0}=6t_0+6$ and  $p^*=\Prob{Y_t=d_0+r| Y_{t_0}=d_0}$.

\begin{figure*}
\centering
\includegraphics[width=0.35\textwidth]{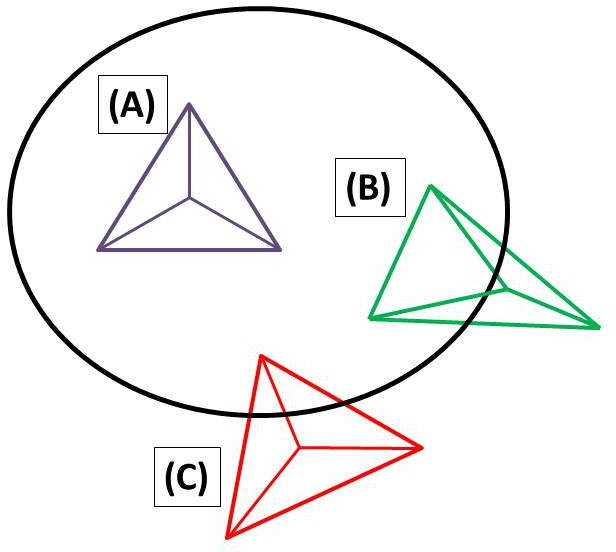}
\caption{Coupling used in Lemma~\ref{lem:RANlemma7}.}
\label{fig:fig3}
\end{figure*}

\noindent 
\begin{claim}[1]
$$  p^* \leq  {d_0+r-1 \choose d_0-1} \Big(\frac{2t_0+3}{2t+1}\Big)^{d_0/2} e^{\frac{3}{2}+t_0-\frac{d_0}{2}+\frac{2r}{3\sqrt{t}}} $$
\label{claim:lemma7}
\end{claim}

\begin{proof}

Let $\tau=(t_0 \equiv \tau_0,\underbrace{\tau_1,\ldots,\tau_r}_{\text{insertion times}},\tau_{r+1} \equiv t)$ be a vector denoting that $Y_t$ 
increases by 1 at $\tau_i$ for $i=1,\ldots,r$. 
We upper bound the probability $p_{\tau}$ of this event in the following.Note that 
we consider the case where the vertices have same degree as the number of faces around them. 
As we mentioned earlier, the other case is analyzed in exactly the same way, modulo a negligible error term.

\begin{align*}
p_{\tau} &= \Bigg[ \prod_{k=1}^r \frac{d_0+k-1}{2\tau_k+1} \Bigg] \Bigg[ \prod_{k=0}^r \prod_{j=\tau_k+1}^{\tau_{k+1}-1} \Big(1-\frac{d_0+k}{2j+1}\Big) \Bigg] \\ 
         &= d_0 (d_0+1) \ldots (d_0+r-1)\Bigg[ \prod_{k=1}^r \frac{1}{2\tau_k+1} \Bigg]
         \exp{\Bigg( \sum_{k=0}^r \sum_{j=\tau_k+1}^{\tau_{k+1}-1} \log{\Big(1-\frac{d_0+k}{2j+1}\Big)}  \Bigg)} \\
         &= \frac{(d_0+r-1)!}{(d_0-1)!}\Bigg[ \prod_{k=1}^r \frac{1}{2\tau_k+1} \Bigg]
         \exp{\Bigg( \sum_{k=0}^r \sum_{j=\tau_k+1}^{\tau_{k+1}-1} \log{\Big(1-\frac{d_0+k}{2j+1}\Big)}  \Bigg)} \\
\end{align*}

\noindent Consider now the inner sum which we upper bound using an integral: 

\begin{align*} 
\sum_{j=\tau_k+1}^{\tau_{k+1}-1} \log{\Big(1-\frac{d_0+k}{2j+1}\Big)}  &\leq  \int_{\tau_k+1}^{\tau_{k+1}} \! \log{\Big(1-\frac{d_0+k}{2x+1}\Big)}  \, \mathrm{d}x  \\
&\leq -\big( \tau_{k+1}+\tfrac{1}{2} \big)  \log{(2\tau_{k+1}+1)} + \\
& \frac{2\tau_{k+1}+1-(d_0+k)}{2} \log{(2\tau_{k+1}+1-(d_0+k))} + \\
& \big( \tau_{k}+\tfrac{3}{2} \big)  \log{(2\tau_{k}+3)} - \frac{2\tau_{k}+3-(d_0+k)}{2} \log{(2\tau_{k}+3-(d_0+k))} 
\end{align*}

\noindent since 

$$ \int \log{\Big(1-\frac{d_0+k}{2x+1}\Big)} = -\big( x+\tfrac{1}{2} \big)  \log{(2x+1)} + \frac{2x+1-(d_0+k)}{2} \log{(2x+1-(d_0+k))} $$

\noindent Hence we obtain $\sum_{k=0}^r \sum_{j=\tau_k+1}^{\tau_{k+1}-1} \log{\Big(1-\frac{d_0+k}{2j+1}\Big)}  \leq A+\sum_{k=1}^r B_k$
where 

\begin{align*} 
A &= \big( \tau_{0}+\tfrac{3}{2} \big)  \log{(2\tau_{0}+3)} - \frac{2\tau_{0}+3-d_0}{2} \log{(2\tau_{0}+3-d_0)} \\
  & -\big( \tau_{r+1}+\tfrac{1}{2} \big)  \log{(2\tau_{r+1}+1)} + \frac{2\tau_{r+1}+1-(d_0+r)}{2} \log{(2\tau_{r+1}+1-(d_0+r))}
\end{align*}

\noindent and 

\begin{align*} 
B_k &= \big( \tau_{k}+\tfrac{3}{2} \big)  \log{(2\tau_{k}+3)} - \frac{2\tau_{k}+3-(d_0+k)}{2} \log{(2\tau_{k}+3-(d_0+k))} \\
    & -\big( \tau_{k}+\tfrac{1}{2} \big)  \log{(2\tau_{k}+1)} + \frac{2\tau_{k}+1-(d_0+k-1)}{2} \log{(2\tau_{k}+1-(d_0+k-1))}.
\end{align*} 

We first upper bound the quantities $B_k$ for $k=1,\ldots,r$. By rearranging terms and using the identity $\log{(1+x)} \leq x$ 
we obtain 

\begin{align*} 
B_k &= \big( \tau_{k}+\tfrac{1}{2} \big)  \log{\big(1+\frac{1}{\tau_{k}+\tfrac{1}{2}}\big)} + \log{(2\tau_k+3)} \\
    &  -\frac{1}{2} \log{\big(2\tau_k+3-(d_0+k)\big)} - \frac{2\tau_{k}+2-(d_0+k)}{2} \log{\Big( 1+\frac{1}{2\tau_{k}+2-(d_0+k)} \Big)}.\\
    &\leq \frac{1}{2} + \frac{1}{2} \log{\big(2\tau_k+3\big)} -\frac{1}{2} \log{\big( 1-\frac{d_0+k}{2\tau_{k}+3} \big)}
\end{align*} 

First we rearrange terms and then we bound the term $e^A$ by using the inequality $e^{-x-x^2/2} \geq 1-x$ which is valid for $0<x<1$:

\begin{align*} 
A   &= -\big( \tau_{0}+\tfrac{3}{2} \big)  \log{\big(1-\frac{d_0}{2\tau_{0}+3}\big)} + \big( \tau_{r+1}+\tfrac{1}{2} \big)  \log{\big(1-\frac{d_0+r}{2\tau_{r+1}+1}\big)} +\frac{d_0}{2} \log{\big( 2\tau_0 +3 -d_0\big)} \\
    &    - \frac{d_0+r}{2} \log{\Big( 2\tau_{r+1}+1-(d_0+r) \Big)}.\Rightarrow \\
e^A &= \big(1-\frac{d_0}{2\tau_{0}+3}\big)^{-(\tau_{0}+\tfrac{3}{2}) } \big(1-\frac{d_0+r}{2\tau_{r+1}+1}\big)^{\tau_{r+1}+\tfrac{1}{2}}( 2\tau_0 +3 -d_0)^{\frac{d_0}{2}}( 2\tau_{r+1}+1-(d_0+r))^{- \frac{d_0+r}{2}} \\
    &= \bigg( \frac{2\tau_0+3}{2\tau_{r+1}+1} \bigg)^{d_0/2} (2\tau_{r+1}+1)^{-r/2} \bigg( 1-\frac{d_0}{2\tau_0+3}\bigg)^{-(\tau_{0}+\tfrac{3}{2})+\tfrac{d_0}{2}} \bigg( 1-\frac{d_0+r}{2\tau_{r+1}+1} \bigg)^{\tau_{r+1}+\tfrac{1}{2} -\frac{d_0+r}{2}} \\
    &\leq \bigg( \frac{2t_0+3}{2t+1} \bigg)^{d_0/2} (2t+1)^{-r/2} \bigg( 1-\frac{d_0}{2\tau_0+3}\bigg)^{-(\tau_{0}+\tfrac{3}{2})+\tfrac{d_0}{2}} e^{ \Big( -\frac{d_0+r}{2t+1}- \big(-\frac{d_0+r}{2t+1}\big)^2/2 \Big) (t+1/2-\frac{d_0+r}{2})  } \\
    &= \bigg( \frac{2t_0+3}{2t+1} \bigg)^{d_0/2} (2t+1)^{-r/2} \bigg( 1-\frac{d_0}{2\tau_0+3}\bigg)^{-(\tau_{0}+\tfrac{3}{2})+\tfrac{d_0}{2}} e^{ -\frac{d_0+r}{2}+\frac{(d_0+r)^2}{8t+4}+\frac{(d_0+r)^3}{4(2t+1)^2} }
\end{align*} 

\noindent Now we upper bound the term $\exp{\big( A+\sum_{k=1}^r B_k \big)}$ using the above upper bounds:

\begin{align*} 
e^{A+\sum_{k=1}^r B_k} &\leq e^A e^{r/2} \prod_{i=1}^r \sqrt{\frac{2\tau_k+3}{1-\frac{d_0+k}{2\tau_k+3}}}\\
&\leq  \bigg( 1-\frac{d_0}{2\tau_0+3}\bigg)^{-(\tau_{0}+\tfrac{3}{2})+\tfrac{d_0}{2}}  e^{ -\frac{d_0}{2}+\frac{(d_0+r)^2}{8t+4}+\frac{(d_0+r)^3}{4(2t+1)^2} } \bigg( \frac{2t_0+3}{2t+1} \bigg)^{d_0/2} \times  \\
& (2t+1)^{-r/2} \prod_{i=1}^r \sqrt{\frac{2\tau_k+3}{1-\frac{d_0+k}{2\tau_k+3}}} 
\end{align*}

\noindent  Using the above upper bound we get that

\begin{align*}
p_{\tau} &\leq  C(r,d_0,t_0,t) \prod_{k=1}^r  \Big[ (2\tau_k+3-(d_0+k))^{-1/2} \big( 1+\frac{1}{\tau_k+1/2}\big)\Big]
\end{align*}

\noindent where 

$$ C(r,d_0,t_0,t)=\frac{(d_0+r-1)!}{(d_0-1)!}\bigg( 1-\frac{d_0}{2\tau_0+3}\bigg)^{-(\tau_{0}+\tfrac{3}{2})+\tfrac{d_0}{2}}  e^{ -\frac{d_0}{2}+\frac{(d_0+r)^2}{8t+4}+\frac{(d_0+r)^3}{4(2t+1)^2} } \bigg( \frac{2t_0+3}{2t+1} \bigg)^{d_0/2}(2t+1)^{-r/2}$$

We need to sum over all possible insertion times to bound the probability of interest $p^*$. We set $\tau_k' \leftarrow \tau_k -  \lceil \frac{d_0+k}{2} \rceil$ for $k=1,\ldots,r$. For $d=o(\sqrt{t})$ and $r=o(t^{2/3})$ we obtain: 

\begin{align*} 
p^* &\leq C(r,d_0,t_0,t) \sum_{t_0+1 \leq \tau_1 < .. <\tau_r \leq t}  \prod_{k=1}^r  \Big[ (2\tau_k+3-(d_0+k))^{-1/2} \big( 1+\frac{1}{\tau_k+1/2}\big)\Big] \\
  &\leq C(r,d_0,t_0,t) \sum_{t_0-\lceil \tfrac{d_0}{2} \rceil +1 \leq \tau_1' \leq .. \leq \tau_r' \leq t-\lceil \tfrac{d_0+r}{2} \rceil}  \prod_{k=1}^r  \Big[ (2\tau_k'+3)^{-1/2} \big( 1+\frac{1}{\tau_k'+\frac{d_0+k}{2} +1/2}\big)\Big] \\
  &\leq \frac{C(r,d_0,t_0,t)}{r!} \Bigg( \sum_{t_0-\lceil \tfrac{d_0}{2} \rceil}^{t-\lceil \tfrac{d_0+r}{2} \rceil} (2\tau_k'+3)^{-1/2}+\tfrac{1}{\sqrt{2}} (\tau_k'+3/2)^{-3/2} \Bigg)^r \\
  &\leq \frac{C(r,d_0,t_0,t)}{r!}  \Bigg( \int_{0}^{t-\frac{d+r}{2}} \! \Big[  (2x+3)^{-1/2} + \tfrac{1}{\sqrt{2}} (x+3/2)^{-3/2} \Big] \, \mathrm{d}x\Bigg)^r \\ 
  &\leq \frac{C(r,d_0,t_0,t)}{r!} \Big(\sqrt{2t+3-(d_0+r)}+2/3\Big)^{r}  \\ 
  &\leq \frac{C(r,d_0,t_0,t)}{r!} (2t)^{r/2} e^{-\tfrac{r}{2} \tfrac{d_0+r-3}{2t} } e^{\frac{2r}{3\sqrt{2t-(d_0+r)+3}}} \\
  &\leq {d_0+r-1 \choose d_0-1} \big(\frac{2t_0+3}{2t+1}\big)^{d_0/2} \Big[ \big( 1-\frac{d_0}{2t_0+3}\big)^{-(1-\frac{d_0}{2t_0+3})} \Big]^{t_0+3/2} \times \\ 
  & \big(\frac{2t}{2t+1}\big)^{r/2} \exp{\Big( -\frac{d_0}{2} + \frac{(d_0+r)^2}{8t+4}+ \frac{(d_0+r)^3}{4(2t+1)^2} - \frac{r(d_0+r-3)}{4t}+\frac{2r}{3\sqrt{2t+3-(d_0+r)}}  \Big)} \\ 
\end{align*}

\noindent By removing the $o(1)$ terms in the exponential and using the fact that $x^{-x}\leq e$ 
we obtain the following bound on the probability $p^*$.

$$  p^* \leq  {d_0+r-1 \choose d_0-1} \Big(\frac{2t_0+3}{2t+1}\Big)^{d_0/2} e^{\frac{3}{2}+t_0-\frac{d_0}{2}+\frac{2r}{3\sqrt{t}}}. $$ 
\end{proof}

Let $\mathcal{A}_1$ denote the event that the supernode consisting of the first $t_0$ vertices has degree $Y_t$
in the modified process $\mathcal{Y}$ less than $t_0^{1/4}\sqrt{t}$. 
Note that since $\{ X_t \leq t_0^{1/4} \sqrt{t}\} \subseteq \{ Y_t \leq t_0^{1/4}\sqrt{t}\} $ it 
suffices to prove that $\Prob{Y_t \leq t_0^{1/4}\sqrt{t}} =  o(1)$. Using Claim (1) we obtain

\begin{align*} 
\Prob{\mathcal{A}_1} &\leq \sum_{r=0}^{ t_0^{1/4}\sqrt{t} -(6t_0+6)} {r+6t_0+5 \choose 6t_0+5}  \Big(\frac{2t_0+3}{2t+1}\Big)^{3t_0+3} e^{-\frac{3}{2}-2t_0+\frac{2t_0^{1/4}}{3}} \\
&\leq t_0^{1/4}t^{1/2} \frac{\big( t_0^{1/4}t^{1/2} \big)^{6t_0+5}}{(6t_0+5)!} \Big(\frac{2t_0+3}{2t+1}\Big)^{3t_0+3}
e^{-\frac{3}{2}-2t_0+\frac{2t_0^{1/4}}{3}} \\ 
&\leq \Big( \frac{t}{2t+1} \Big)^{3t_0+3}  \frac{t_0^{3t_0/2+3/2} (2t_0+3)^{3t_0+3}}{(6t_0+5)^{6t_0+5}} e^{4t_0+7/2+2/3 t_0^{1/4}}\\
&\leq 2^{-(3t_0+3)} \frac{e^{4t_0+7/2+2/3 t_0^{1/4}}}{(6t_0+5)^{\tfrac{3}{2}t_0+\tfrac{1}{2}}}  =o(1).
\end{align*}

\end{proof}

\begin{lemma}
No vertex added after $t_1$ has degree exceeding $t_0^{-2}t^{1/2}$ \whp.
\label{lem:RANlemma8}
\end{lemma}

\begin{proof} 

Let $\mathcal{A}_2$ denote the event that some vertex added after $t_1$ has
degree exceeding $t_0^{-2}t^{1/2}$. We use a union bound, a third moment argument
and Lemma~\ref{lem:RANlemma6} to prove that $\Prob{\mathcal{A}_2}=o(1)$. 
Specifically 

\begin{align*}
\Prob{\mathcal{A}_2} &\leq  \sum_{s=t_1}^t \Prob{d_t(s) \geq t_0^{-2}t^{1/2}} = \sum_{s=t_1}^t \Prob{d_t(s)^{(3)} \geq (t_0^{-2}t^{1/2})^{(3)}} \\ 
                     &\leq  t_0^{6}t^{-3/2} \sum_{s=t_1}^t \Mean{d_t(s)^{(3)}} \leq 5!\sqrt{2} t_0^6 \sum_{s=t_1}^t s^{-3/2}            \leq 5!2\sqrt{2} t_0^6 t_1^{-1/2} =o(1). 
\end{align*} 
 
\end{proof}

\begin{lemma}
No vertex added before $t_1$ has degree exceeding $t_0^{1/6}t^{1/2}$ \whp. 
\label{lem:RANlemma9}
\end{lemma}

\begin{proof} 
Let $\mathcal{A}_3$ denote the event that some vertex added before $t_1$ 
has degree exceeding $t_0^{1/6}t^{1/2}$. We use again a third moment argument
and Lemma~\ref{lem:RANlemma6} to prove that  $\Prob{\mathcal{A}_3}=o(1)$.

\begin{align*} 
\Prob{\mathcal{A}_3} &\leq \sum_{s=1}^{t_1} \Prob{ d_t(s) \geq t_0^{1/6}t^{1/2}} = \sum_{s=1}^{t_1} \Prob{ d_t(s)^{(3)} \geq (t_0^{1/6}t^{1/2})^{(3)} } \\ 
           &\leq t_0^{-1/2}t^{-3/2} \sum_{s=1}^{t_1} \Mean{ d_t(s)^{(3)} }  \leq t_0^{-1/2}t^{-3/2} \sum_{s=1}^{t_1} 5!\sqrt{2} \frac{t^{3/2}}{s^{3/2}} \\ 
           &\leq 5!\sqrt{2}\zeta(3/2)t_0^{-1/2} =o(1)
\end{align*} 

\noindent where $\zeta(3/2)=\sum_{s=1}^{+\infty} s^{-3/2} \approx 2.612$.
 
\end{proof}

\begin{lemma}
The $k$ highest degrees are added before $t_1$ and have degree $\Delta_i$ bounded by 
$ t_0^{-1} t^{1/2} \leq \Delta_i \leq t_0^{1/6}t^{1/2}$  \whp.
\label{lem:RANlemma10}
\end{lemma}

\begin{proof} 
For the upper bound it suffices to show that $\Delta_1 \leq t_0^{1/6}t^{1/2}$. 
This follows immediately by Lemmas~\ref{lem:RANlemma8} and~\ref{lem:RANlemma9}.
The lower bound follows directly from Lemmas~\ref{lem:RANlemma7},~\ref{lem:RANlemma8} and~\ref{lem:RANlemma9}.
Assume that at most $k-1$ vertices added before $t_1$ have degree exceeding the lower bound 
$t_0^{-1} t^{1/2}$. Then the total degree of the supernode formed by the first $t_0$ vertices
is $O(t_0^{1/6}\sqrt{t})$. This contradicts Lemma~\ref{lem:RANlemma7}. 
Finally, since each vertex $s\geq t_1$ has degree at most $t_0^{-2}\sqrt{t} \ll t_0^{-1} t^{1/2}$
the $k$ highest degree vertices are added before $t_1$ \whp.  
\end{proof}

\noindent The proof of Theorem~\ref{thrm:RANthrm1} is completed with the following lemma. 

\begin{lemma}
The $k$ highest degrees satisfy $ \Delta_{i} \leq \Delta_{i-1} - \frac{\sqrt{t}}{f(t)}$ \whp.
\label{lem:RANlemma11}
\end{lemma}

\begin{proof} 
Let $\mathcal{A}_4$ denote the event that there are two vertices among the first $t_1$
with degree $t_0^{-1} t^{1/2}$ and within $\frac{\sqrt{t}}{f(t)}$ of each other. 
By the definition of conditional probability and Lemma~\ref{lem:RANlemma8}

\begin{align*}
\Prob{\mathcal{A}_4}  &= \Prob{\mathcal{A}_4|\bar{\mathcal{A}_3}} \Prob{\bar{\mathcal{A}_3}} + \Prob{\mathcal{A}_4|\mathcal{A}_3}\Prob{\mathcal{A}_3} \leq \Prob{\mathcal{A}_4|\bar{\mathcal{A}_3}} + o(1) 
\end{align*} 

\noindent it suffices to show that $\Prob{\mathcal{A}_4 | \bar{\mathcal{A}_3}}=o(1)$. Note that by a simple union bound

\begin{align*}
\Prob{\mathcal{A}_4} &\leq \sum_{1 \leq s_1 < s_2 \leq t_1} \sum_{l=-\frac{\sqrt{t}}{f(t)}}^{\frac{\sqrt{t}}{f(t)}} p_{l,s_1,s_2} =O\Big(t_1^2 \frac{\sqrt{t}}{f(t)} \max{ p_{l,s_1,s_2} }\Big)
\end{align*}

\noindent where $p_{l,s_1,s_2} = \Prob{d_t(s_1)-d_t(s_2)=l|\bar{\mathcal{A}_3}}$. 

We consider two cases and we show that in both cases $\max{ p_{l,s_1,s_2}}=o(\frac{f(t)}{t_1^2\sqrt{t}})$.

\noindent
\underline{$\bullet$ {\sc Case 1} $(s_1,s_2) \notin E(G_t)$:}\\
 
Note that at time $t_1$ there exist $m_{t_1}=3t_1+3 < 4t_1$ edges in $G_{t_1}$. 

\begin{align*} 
p_{l,s_1,s_2} &\leq \sum_{r=t_0^{-1}t^{1/2}}^{t_0^{1/6}t^{1/2}} \sum_{d_1,d_2=3}^{4t_1} \Prob{d_t(s_1)=r \wedge d_t(s_2)=r-l | d_{t_1}(s_1)=d_1, d_{t_1}(s_2)=d_2} \tag{2}\\ 
              &\leq t_0^{1/6}t^{1/2}  \sum_{d_1,d_2=3}^{4t_1} {2t_0^{1/6}t^{1/2} \choose d_1-1} {2t_0^{1/6}t^{1/2} \choose d_2-1} \Big(\frac{2t_0+3}{2t+1}\Big)^{(d_1+d_2)/2} e^{\frac{3}{2}+t_1+\frac{2t_0^{1/6}}{3}} \tag{3}\\ 
              &\leq t_0^{1/6}t^{1/2} \sum_{d_1,d_2=3}^{4t_1}  (2t_0^{1/6}t^{1/2})^{d_1+d_2-2}  \Big(\frac{2t_0+3}{2t+1}\Big)^{(d_1+d_2)/2} e^{2t_1} \\ 
              &\leq  t_0^{1/6}t^{1/2} e^{2t_1} t_1^2 (2t_0^{1/6}t^{1/2})^{8t_1-2} \Big(\frac{2t_0+3}{2t+1}\Big)^{4t_1} \\
              &= t_0^{4t_1/3+1/6}t^{-1/2} e^{2t_1} t_1^2 2^{8t_1} (2t_0+3)^{4t_1} \Big( \frac{t}{2t+1} \Big)^{4t_1} \\ 
              &= o\Big( \frac{f(t)}{t_1^2\sqrt{t}} \Big)
\end{align*} 

\noindent Note that we omitted the tedious calculation justifying the transition from (2) to (3) 
since calculating the upper bound of the joint probability distribution is very similar to the calculation of Lemma~\ref{lem:RANlemma7}.

\noindent \\
\underline{$\bullet$ {\sc Case 2} $(s_1,s_2) \in E(G_t)$ :}\\

Notice that in any case $(s_1,s_2)$ share at most two faces (which may change over time). 
Note that the two connected vertices $s_1,s_2$ share a common face only if $s_1,s_2 \in \{1,2,3\}$\footnote{We analyze the case where $s_1,s_2 \geq 4$. The other case is treated in the same manner.}.
Consider the following modified process $\mathcal{Y}'$: whenever an incoming vertex ``picks'' one of the two common faces 
we don't insert it. We choose two other faces which are not common to $s_1,s_2$ and add one vertex in each of those. 
Notice that the number of faces increases by 1 for both $s_1,s_2$ as in the original process and the difference of the degrees
remains the same. An algebraic manipulation similar to Case 1 gives the desired result. 
\end{proof}

\section{Proof of Theorem~\ref{thrm:RANthrm2}}
\label{sec:RANeigen}

Having computed the highest degrees of a RAN in Section~\ref{sec:RANdegrees}, eigenvalues are computed 
by adapting existing techniques \cite{chung,flaxman,mihail}.
We decompose the proof of Theorem~\ref{thrm:RANthrm2} in Lemmas~\ref{lem:RANlemma12},~\ref{lem:RANlemma13},~\ref{lem:RANlemma14},~\ref{lem:RANlemma15}.
Specifically, in Lemmas~\ref{lem:RANlemma12},~\ref{lem:RANlemma13} we bound the degrees and co-degrees respectively. 
Having these bounds, we decompose the graph into a star forest and show in Lemmas~\ref{lem:RANlemma14} and~\ref{lem:RANlemma15} that 
its largest eigenvalues, which are $(1\pm o(1))\sqrt{\Delta_i}$, dominate the eigenvalues of the remaining graph.  
This technique was pioneered by Mihail and Papadimitriou \cite{mihail}.

We partition the vertices into three set $S_1,S_2,S_3$. Specifically, let $S_i$ be the set of vertices 
added after time $t_{i-1}$ and at or before time $t_i$ where 

$$ t_0=0, t_1=t^{1/8}, t_2=t^{9/16}, t_3=t.$$

In the following we use the recursive variational characterization of eigenvalues \cite{chungbook}. Specifically,
let $A_G$ denote the adjacency matrix of a simple, undirected graph $G$ and let $\lambda_i(G)$ denote
the $i$-th largest eigenvalue of $A_G$. Then

$$ \lambda_i(G) = \min_{S} \max_{x \in S,x \neq 0} \frac{x^TA_Gx}{x^Tx}$$ 

\noindent where $S$ ranges over all $(n-i+1)$ dimensional subspaces of $\field{R}^n$. 

\begin{lemma}
For any $\epsilon>0$ and any $f(t)$ with $f(t) \rightarrow +\infty$ as $t \rightarrow +\infty$
the following holds \whp: for all $s$ with $f(t) \leq s \leq t$, for all vertices $r \leq s$,
then $d_s(r) \leq s^{\tfrac{1}{2}+\epsilon}r^{-\tfrac{1}{2}}$. 
\label{lem:RANlemma12}
\end{lemma}

\begin{proof} 

Set $q=\Ceil{\frac{4}{\epsilon}}$. We use Lemma~\ref{lem:RANlemma6}, a union bound and Markov's inequality to obtain:

\begin{align*}
\Prob{ \bigcup_{s=f(t)}^t \bigcup_{r=1}^s \{ d_s(r) \geq s^{1/2+\epsilon} r^{-1/2} \} } &\leq \sum_{s=f(t)}^t \sum_{r=1}^s  \Prob{ d_s(r)^{(q)} \geq (s^{1/2+\epsilon} r^{-1/2})^{(q)}  }\\
&\leq \sum_{s=f(t)}^t \sum_{r=1}^s  \Prob{ d_s(r)^{(q)} \geq (s^{-(q/2+q\epsilon)} r^{q/2}) } \\ 
&\leq \sum_{s=f(t)}^t \sum_{r=1}^s  \frac{(q+2)!}{2} \Big( \frac{2s}{r} \Big)^{q/2} s^{-q/2} s^{-q\epsilon} r^{q/2}  \\
&= \frac{(q+2)!}{2} 2^{q/2} \sum_{s=f(t)}^t s^{1-q\epsilon} \\
&\leq \frac{(q+2)!}{2} 2^{q/2}   \int_{f(t)-1}^t \! x^{1-q\epsilon}  \, \mathrm{d}x    \\
&\leq \frac{(q+2)!}{2(q\epsilon-2)} 2^{q/2} (f(t)-1)^{2-q\epsilon} = o(1).
\end{align*}

\end{proof}

\begin{lemma}
Let $S_3'$ be the set of vertices in $S_3$ which are adjacent to more than one vertex of $S_1$. Then $ |S_3'| \leq t^{1/6}$ \whp.
\label{lem:RANlemma13} 
\end{lemma}

\begin{proof} 
First, observe that when vertex $s$ is inserted it becomes adjacent to more than one vertex of $S_1$ 
if the face chosen by $s$ has at least two vertices in $S_1$. We call the latter property $\mathcal{A}$
and we write $s \in \mathcal{A}$ when $s$ satisfies it. 
At time $t_1$ there exist $2t_1+1$ faces total, which consist of faces whose three vertices 
are all from $S_1$. At time $s \geq t_2$ there can be at most $6t_1+3$ faces with at least two vertices
in $S_1$ since each of the original $2t_1+1$ faces can give rise to at most 3 new faces with 
at least two vertices in $s_1$.
Consider a vertex $s \in S_3$, i.e., $s \geq t_2$. By the above argument, 
$ \Prob{|N(s) \cap S_1| \geq 2 } \leq \frac{6t_1+3}{2t+1}$. Writing $|S_3'|$ as a sum
of indicator variables, i.e., $|S_3'|= \sum_{s=t_2}^t I(s \in \mathcal{A})$
and taking the expectation we obtain

\begin{align*} 
\Mean{|S_3'|} &\leq \sum_{s=t_2}^t \frac{6t_1+3}{2t+1} \leq (6t_1+3) \int_{t_2}^t \! (2x+1)^{-1} \, \mathrm{d}x  \\ 
         &\leq (3t^{\tfrac{1}{8}}+\tfrac{3}{2})   \ln{\frac{2t+1}{2t_2+1}}= o(t^{1/7})
\end{align*}

\noindent By Markov's inequality:

$$ \Prob{ |S_3'| \geq t^{1/6} } \leq \frac{ \Mean{|S_3'|} }{t^{1/6}}  = o(1).$$ 

\noindent Therefore, we conclude that $|S_3'| \leq t^{1/6}$ \whp.  
\end{proof}

\begin{lemma}
Let $F \subseteq G$ be the star forest consisting of edges between $S_1$ and $S_3 - S_3'$. 
Let $\Delta_1 \geq \Delta_2 \geq \ldots \geq \Delta_k$ denote the $k$ highest degrees of $G$. 
Then $\lambda_i(F) = (1-o(1))\sqrt{\Delta_i}$ \whp. 
\label{lem:RANlemma14} 
\end{lemma}

\begin{proof}

It suffices to show that $\Delta_i(F) = (1-o(1))\Delta_i(G)$ for $i=1,\ldots,k$. 
Note that since the $k$ highest vertices are inserted before $t_1$ \whp, the edges they lose
are the edges between $S_1$ and the ones incident to $S_3'$ and $S_2$ 
and we know how to bound the cardinalities of all these sets. 
Specifically by Lemma~\ref{lem:RANlemma13} $|S_3'| \leq t^{1/6}$ \whp~ and by 
Theorem~\ref{thrm:RANthrm1} the maximum degree in $G_{t_1},G_{t_2}$ is less 
than $t_1^{1/2+\epsilon_1} = t^{1/8}$, $t_2^{1/2+\epsilon_2} = t^{5/16}$ for $\epsilon_1=1/16,\epsilon_2=1/32$
respectively \whp. Also by Theorem~\ref{thrm:RANthrm1}, $\Delta_i(G) \geq \frac{\sqrt{t}}{\log{t}}$.
Hence, we obtain

$$ \Delta_i(F) \geq \Delta_i(G) - t^{1/8} - t^{5/16} - t^{1/6} = (1-o(1))\Delta_i(G).$$
 
\end{proof}

\noindent To complete the proof of Theorem~\ref{thrm:RANthrm2} it suffices to prove that $\lambda_1(H)$ is $o(\lambda_k(F))$
where $H=G-F$. We prove this in the following lemma. The proof is based on bounding maximum
degree of appropriately defined subgraphs using Lemma~\ref{lem:RANlemma12} and standard inequalities 
from spectral graph theory \cite{chungbook}. 

\begin{lemma} 
$\lambda_1(H) = o(t^{1/4})$ \whp.
\label{lem:RANlemma15}
\end{lemma}

\begin{proof} 
From Gershgorin's theorem \cite{strang} the maximum eigenvalue of any graph is bounded by the maximum degree. 
We bound the eigenvalues of $H$ by bounding the maximum eigenvalues of six different induced
subgraphs. Specifically, let $H_i=H[S_i]$, $H_{ij}=H(S_i,S_j)$ where $H[S]$ is the subgraph
induced by the vertex set $S$ and $H(S,T)$ is the subgraph containing only edges with one vertex
is $S$ and other in $T$. We use Lemma~\ref{lem:RANlemma14} to bound 
$\lambda_1(H(S_1,S_3))$ and Lemma~\ref{lem:RANlemma13} for the other eigenvalues. 
We set $\epsilon=1/64$.

$$ \lambda_1(H_1) \leq \Delta_1(H_1)      \leq t_1^{1/2+\epsilon} = t^{33/512}. $$
$$ \lambda_1(H_2) \leq \Delta_1(H_2)    \leq t_2^{1/2+\epsilon}t_1^{-1/2} = t^{233/1024}. $$
$$ \lambda_1(H_3) \leq \Delta_1(H_3)    \leq t_3^{1/2+\epsilon}t_2^{-1/2} = t^{15/64}. $$
$$ \lambda_1(H_{12}) \leq \Delta_1(H_{12}) \leq t_2^{1/2+\epsilon} = t^{297/1024}. $$
$$ \lambda_1(H_{23}) \leq \Delta_1(H_{23}) \leq t_3^{1/2+\epsilon} t_1^{-1/2} = t^{29/64}. $$
$$ \lambda_1(H_{13}) \leq \Delta_1(H_{13}) \leq t^{1/6}. $$

Therefore \whp ~we obtain

$$ \lambda_1(H) \leq \sum_{i=1}^3 \lambda_1(H_i) + \sum_{i<j} \lambda_1(H_{i,j}) = o(t^{1/4}).$$
 
\end{proof}

\section{Proof of Theorem~\ref{thrm:RANternary}}
\label{sec:RANdiam}

Before we give the proof of Theorem~\ref{thrm:RANternary}, we give a simple proof that 
the diameter of a RAN is $O(\log{t})$ \whp. 

\begin{figure}[h]
\centering
\includegraphics[width=0.25\textwidth]{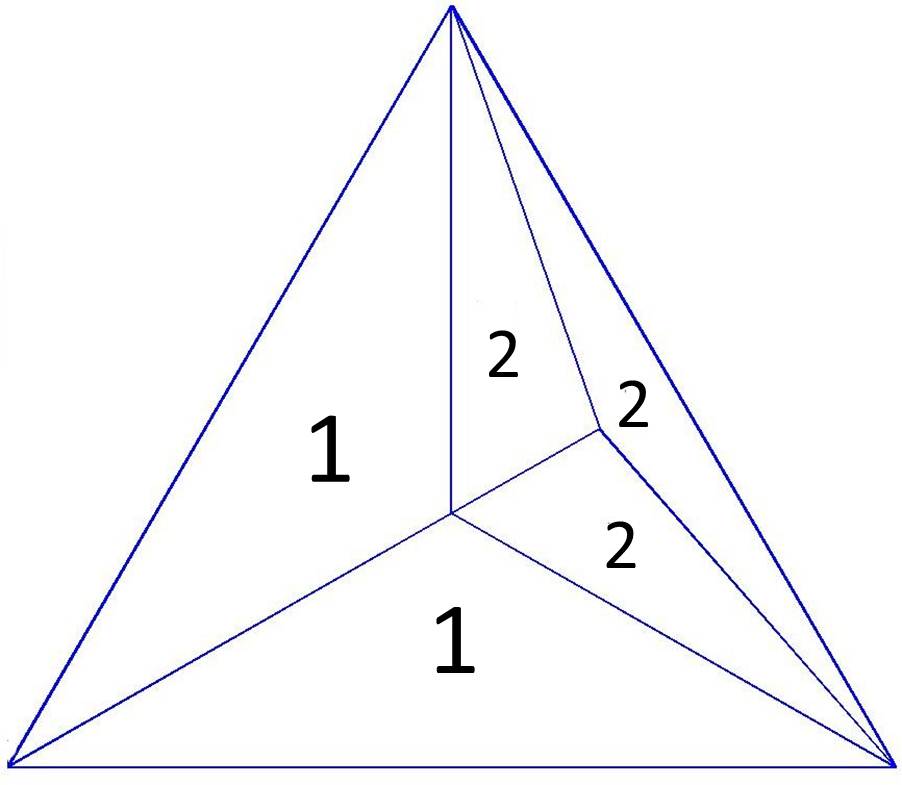}
\caption{An instance of the process for $t=2$. Each face is labeled with its depth.}
\label{fig:RANdepth}
\end{figure}

We begin with a necessary definition for the proof of Claim (2). 
We define the {\it depth of a face} recursively. Initially, we have three faces,
see Figure~\ref{fig:RANfig1}(a), whose depth equals 1. 
For each new face $\beta$ created by subdividing a face $\gamma$, we have $depth(\beta)=depth(\gamma)+1$. 
An example is shown in Figure~\ref{fig:RANdepth}, where each face is labeled with its
corresponding depth. 

\begin{claim}[2]
The diameter $d(G_t)$ satisfies  $d(G_t)=O(\log{t})$ \whp. 
\label{claim:diameter} 
\end{claim}

\begin{proof} 

A simple but key observation is that if $k^*$ is the maximum depth of a face then  $d(G_t)=O(k^*)$.
Hence, we need to upper bound the depth of a given face after $t$ rounds. 
Let $F_t(k)$ be the number of faces of depth $k$ at time $t$, then: 

\begin{align*}
\Mean{F_t(k)} &= \sum_{1 \leq t_1<t_2<\ldots<t_k\leq t} \prod_{j=1}^k \frac{1}{2t_j+1} \leq \frac{1}{k!} (\sum_{j=1}^t \frac{1}{2j+1} )^k \leq \frac{1}{k!} (\frac{1}{2} \log{t})^k \leq   (\frac{e\log{t}}{2k})^{k+1}
\end{align*}

\noindent By the first moment method we obtain $k^*=O(\log{t})$ \whp~ and by our observation $d(G_t)=O(\log{t})$ \whp.  
\end{proof}

\begin{figure}[h]
\centering
\includegraphics[width=0.99\textwidth]{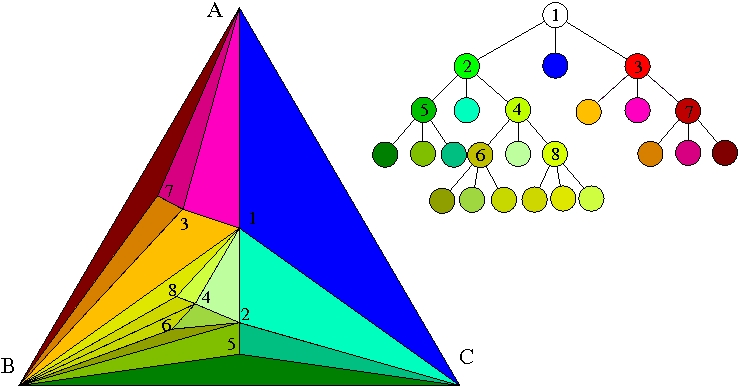}
\caption{RANs as random ternary trees.}
\label{fig:RANternary}
\end{figure}

\begin{figure}[h]
\centering
\includegraphics[width=0.5\textwidth]{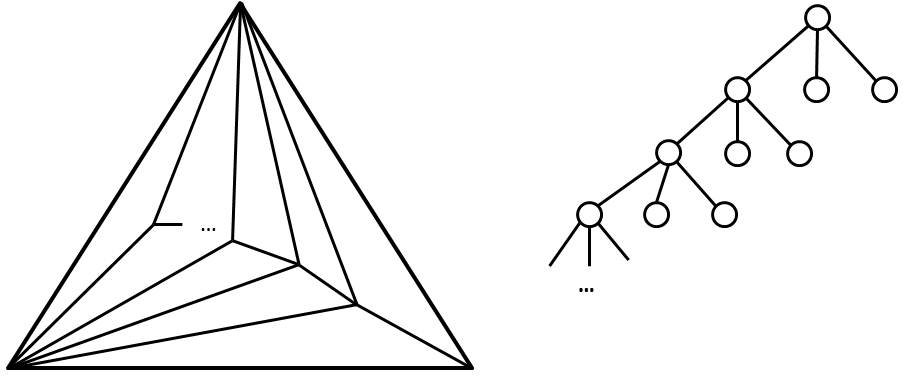}
\caption{The height of the random ternary tree cannot be used to lower bound the diameter. The height of the 
random ternary tree can be arbitrarily large but the diameter is 2.}
\label{fig:RANcounter}
\end{figure}

The depth of a face can be formalized via a bijection between random ternary trees and RANs.
Using this bijection we prove Theorem~\ref{thrm:RANternary} which gives a refined upper bound on the 
asymptotic growth of the diameter.

\begin{proof} 
Consider the random process which starts with a single vertex tree and at every
step picks a random leaf and adds three children to it. Let $T$
be the resulting tree after $t$ steps. 
There exists a natural bijection between the RAN process and this process, see \cite{darase}
and also Figure~\ref{fig:RANternary}.  
The depth of $T$ in probability 
is $\frac{\rho}{2} \log{t}$ where $\frac{1}{\rho}=\eta$ is the unique solution less than 1 
of the equation $\eta - 1 - \log{\eta} = \log{3}$, see Broutin and Devroye \cite{devroye}, pp. 284-285\footnote{ There is a typo in \cite{devroye}. Specifically it says ``$\rho$ is the unique solution greater than 1 of ...,''.
However, based on their Theorem 1, they should replace ``greater than'' with ``less than''. 
Thanks to Abbas Mehrabian for pointing this out.}.
Note that the diameter $d(G_t)$ is at most twice the height of the tree and hence the result follows. 
 
\end{proof} 	

The above observation, i.e., the bijection between RANs and random ternary trees
cannot be used to lower bound the diameter. A counterexample is shown in Figure~\ref{fig:RANcounter}
where the height of the random ternary tree can be made arbitrarily large but the diameter is 2. 
Albenque and Marckert proved in \cite{albenque} that if $v,u$ are two i.i.d. uniformly random
internal vertices, i.e., $v,u \geq 4$, then the distance $d(u,v)$ tends
to $\frac{6}{11} \log{n}$ with probability 1 as the number of vertices $n$ of the RAN grows to infinity.
Finally, it is worth mentioning that the diameter of the RAN grows faster
asymptotically than the diameter of the classic preferential attachment model \cite{albert} which \whp~
grows as $\frac{\log{t}}{\log{\log{t}}}$, see Bollob\'{a}s and Riordan \cite{bollobasriordan}.

\clearpage
\chapter{Triangle Counting in Large Graphs}
\label{trianglecountingchapter}
\lhead{\emph{Triangle Counting in Large Graphs}}
\section{Introduction} 
\label{sec:subgraphintro}

In this Chapter we focus on algorithmic techniques for approximate triangle counting in large graphs. 
Despite the fact that the subgraph of interest is a triangle, our techniques are extendable
to counting other types of fixed-size subgraphs. The outline of this Chapter is as follows: 
in Section~\ref{sec:trianglesparsifiers} we present {\it triangle sparsifiers}, 
a notion inspired by the seminal work of Bencz{\'u}r and Karger \cite{benczurstoc}
and Spielman and Srivastava \cite{DBLP:conf/stoc/SpielmanS08}. 
Specifically, we present a new randomized algorithm for approximately counting the number 
of triangles in a graph $G$. The algorithm proceeds as follows: keep each edge independently with probability $p$, enumerate
the triangles in the sparsified graph $G'$ and return the number of triangles found in $G'$ 
multiplied by  $p^{-3}$. We prove that under mild assumptions on $G$ and $p$ our algorithm returns a good 
approximation for the number of triangles with high probability.

We illustrate the efficiency of our algorithm on various large real-world datasets
where we achieve significant speedups. 
Furthermore, we investigate the performance of existing sparsification procedures
namely the Spielman-Srivastava spectral sparsifier \cite{DBLP:conf/stoc/SpielmanS08} and 
the the Bencz{\'u}r-Karger cut sparsifier \cite{benczurstoc,DBLP:journals/corr/cs-DS-0207078}
and show that they are not optimal/suitable with respect to triangle counting.

In Section~\ref{sec:degreepartitioning} we propose a new triangle counting method which provides a $(1 \pm \epsilon)$ multiplicative
approximation to the number of triangles in the graph and runs in $O \left( m + \frac{m^{3/2} \log{n} }{t \epsilon^2} \right)$ time.
The key idea of the method is to combine the sampling scheme introduced by Tsourakakis et al. in \cite{Tsourakakiskdd09,tsourakakis2}
with the partitioning idea of Alon, Yuster and Zwick \cite{739463} in order to obtain a more efficient sampling scheme. 
Furthermore, we show that this method can be adapted to the semistreaming model \cite{feigenbaum2005graph}
with a constant number of passes and $O\left(m^{1/2}\log{n} + \frac{m^{3/2}  \log{n} }{t \epsilon^2} \right)$ space. 
We apply our methods in various networks with several millions of edges and we obtain excellent results both with respect to the accuracy
and the running time. 
Finally, we propose a random projection based method for triangle counting and provide a sufficient condition to obtain an 
estimate with low variance. Even if such a method is unlikely to be practical it raises some interesting theoretical issues. 
 
In Section~\ref{sec:IPL} we present an (almost) optimal algorithm for triangle counting.
The proposed randomized algorithm is analyzed via the second moment method and provides tight theoretical guarantees. 
We discuss various aspects of the proposed algorithm, including 
an implementation in \mapreduce.

\section{Triangle Sparsifiers}
\label{sec:trianglesparsifiers}

\subsection{Proposed Algorithm}
\label{subsec:Sparsifieralgorithm}

\begin{algorithm}[!ht]
\caption{\label{alg:trianglesparsifier} Triangle Sparsifier} 
\begin{algorithmic}
\REQUIRE Set of edges $E \subseteq {[n]\choose 2}$ \\
\COMMENT{Unweighted graph $G([n],E)$} 
\REQUIRE Sparsification parameter $p$
\STATE Pick a random subset $E'$ of edges such that the events $\Set{e \in E'}, \mbox{for all $e \in E$}$ are independent and the probability of each is equal to $p$.
\STATE $t' \leftarrow $ count triangles on the graph $G'([n],E')$
\STATE Return $T \leftarrow \frac{t'}{p^3}$ 
\end{algorithmic}
\end{algorithm}

Our proposed algorithm {\em Triangle Sparsifier} is shown in Algorithm~\ref{alg:trianglesparsifier}.
The algorithm takes an unweighted, simple graph $G(V,E)$, where without loss of 
generality we assume $V=[n]$, and a sparsification parameter $p \in (0,1)$ as input.
The algorithm first chooses a random subset $E'$ of the set $E$ of edges.
The random subset is such that the events
$$
\Set{e \in E'}, \mbox{for all $e \in E$},
$$
are independent and the probability of each is equal to $p$.
Then, any triangle counting algorithm can be used to count triangles 
on the sparsified graph with edge set $E'$. Clearly, the expected size of $E'$ 
is $pm$. 
The output of our algorithm is the number of triangles in the sparsified
graph multiplied by $\frac{1}{p^3}$, or equivalently we are counting the 
number of weighted triangles in $G'$ where each edge has weight $\frac{1}{p}$. 
It follows immediately that the expected value $\Mean{T}$ of our estimate is 
the number of triangles  in $G$, i.e., $t$. Our main theoretical result is the 
following theorem:

\begin{theorem}
\label{thrm:ctsourak}
Suppose $G$ is an undirected graph with $n$ vertices, $m$ edges and $t$ triangles. Let also $\Delta$ denote the size of the largest collection of triangles with a common edge. Let $G'$ be the random graph that arises from $G$ if we keep every edge with probability $p$ and write $T$ for the number of triangles of $G'$. Suppose that $\gamma>0$ is a constant and 
\beql{cond}
\frac{pt}{\Delta} \ge \log^{6+\gamma} n, \ \ \ \mbox{if $p^2\Delta \ge 1$},
\eeq
and
\beql{cond-small}
p^3 t \ge \log^{6+\gamma} n, \ \ \ \mbox{if $p^2\Delta < 1$}.
\eeq
for $n \ge n_0$ sufficiently large.
Then
$$
\Prob{\Abs{T-\Mean{T}} \ge \epsilon \Mean{T}} \le n^{-K}
$$
for any constants $K, \epsilon > 0$ and all large enough $n$ (depending on $K$, $\epsilon$ and $n_0$).
\end{theorem}

\begin{proof}
Write $X_e=1$ or $0$ depending on whether the edge $e$ of graph $G$ survives in $G'$. Then
$T = \sum_{\Delta(e,f,g)} X_e X_f X_g$ where
$\Delta(e,f,g) = \One{\mbox{edges $e,f,g$ form a triangle}}$.
Clearly $\Mean{T} = p^3 t$.

Refer to Theorem~\ref{thrm:kim-vu}. We use $T$ in place of $Y$, $k=3$.

We have
$$
\Mean{\frac{\partial T}{\partial X_e}} = \sum_{\Delta(e,f,g)} \Mean{X_f X_g} = p^2 \Abs{\Delta(e)}.
$$

\noindent We first estimate the quantities ${\mathbb E}_j (T), j=0,1,2,3,$ defined before Theorem \ref{thrm:kim-vu}.
We get \beql{e1}
{\mathbb E}_1 (T) = p^2 \Delta.
\eeq

We also have
$$
\Mean{\frac{\partial^2 T}{\partial X_e \partial X_f}} = p\One{\exists g: \Delta(e,f,g)},
$$
hence
\beql{e2}
{\mathbb E}_2 (T) \le p.
\eeq
Obviously, ${\mathbb E}_3(T) \le 1$.

Hence
$$
{\mathbb E}_{\ge 3} (T) \le 1,\ 
{\mathbb E}_{\ge 2} (T) \le 1,
$$
and
$$
{\mathbb E}_{\ge 1} (T) \le \max\Set{1, p^2\Delta},\ 
{\mathbb E}_{\ge 0} (T) \le \max\Set{1, p^2\Delta, p^3t}.
$$

\noindent
\underline{$\bullet$ {\sc Case 1} ($p^2\Delta < 1$):}\\
We get ${\mathbb E}_{\ge 1} (T) \le 1$, and from \eqref{cond-small}, $\Mean{T} \ge {\mathbb E}_{\ge 1}(T)$.

\noindent \\
\underline{$\bullet$ {\sc Case 2} ($p^2\Delta \ge 1$):}\\
We get ${\mathbb E}_{\ge 1} (T) \le p^2\Delta$ and, from \eqref{cond}, $\Mean{T} \ge {\mathbb E}_{\ge 1}(T)$.

We get, for some constant $c_3>0$, from Theorem \ref{thrm:kim-vu}:
\beql{dev}
\Prob{\Abs{T-\Mean{T}} \ge c_3 \lambda^3 (\Mean{T} {\mathbb E}_{\ge 1}(T))^{1/2}} \le e^{-\lambda + 2\log n}.
\eeq
Notice that since in both cases we have $\Mean{T} \ge {\mathbb E}_{\ge 1}(T)$.

We now select $\lambda$ so that the lower bound inside the probability on the left-hand side of \eqref{dev} becomes $\epsilon\Mean{T}$.
In Case 1 we pick
$$
\lambda = \frac{\epsilon^{1/3}}{c_3^{1/3}} (p^3 t)^{1/6}
$$
while in Case 2
$$
\lambda = \frac{\epsilon^{1/3}}{c_3^{1/3}} \left(\frac{pt}{\Delta}\right)^{1/6}
$$
to get
\beql{dev1}
\Prob{\Abs{T-\Mean{T}} \ge \epsilon \Mean{T}} \le \exp(-\lambda+2\log n) 
\eeq
Since $\lambda \ge (K+2) \log n$ follows from our assumptions \eqref{cond} and \eqref{cond-small} if $n$ is sufficiently large, we get
$\Prob{\Abs{T-\Mean{T}} \ge \epsilon \Mean{T}} \le n^{-K}$, in both cases. 
\end{proof}

\spara{Complexity Analysis:} The expected running time of edge sampling i
s sublinear, i.e., $O(pm)$, see Claim~\ref{claim:claim1}.
The complexity of the counting step depends on which  algorithm we use to count triangles\footnote{We assume for fairness that we use the same algorithm
in both the original graph $G$ and the sparsified graph $G'$ to count triangles.}. 
For instance, if we use \cite{739463} as our triangle counting algorithm,
the expected running time of Triangle Sparsifier is $O(pm+ (pm)^{\frac{2\omega}{\omega+1}})$,
where $\omega$ currently is   2.3727  \cite{williams2011breaking}.
If we use the node-iterator (or any other standard listing triangle algorithm) 
the expected running time is $O(pm+ p^2 \sum_i d_i^2)$.

\begin{claim}[Sparsification in sublinear expected time]
\label{claim:claim1}
The edge sampling can run in $O(pm)$ expected time. 
\end{claim}

\begin{proof}
We do not ``toss a $p$-coin'' $m$ times in order to construct $E'$.
This would be very wasteful if $p$ is small.
Instead we construct the random set $E'$ with the following procedure which
produces the right distribution.
Observe that the number $X$ of unsuccessful events, i.e., edges which are not selected in our sample,
until a successful one follows a geometric distribution. Specifically, $\Prob{X=x}=(1-p)^{x-1}p$. 
To sample from this distribution it suffices to generate a uniformly distributed variable $U$ in $[0,1]$
and set $X \leftarrow \Ceil{ \frac{\text{ln}U}{1-p} } $. Clearly the probability that $X=x$ is 
equal to $\Prob{ (1-p)^{x-1} > U \geq (1-p)^x   } = (1-p)^{x-1} - (1-p)^{x} = (1-p)^{x-1}p $ as required. 
This provides a practical and efficient way to pick the subset $E'$ of edges in subliner expected time $O(pm)$. 
For more details see \cite{knuth}.
\end{proof}

\spara{Expected Speedup:} The expected speedup with respect to the triangle counting task depends on the triangle 
counting subroutine that we use. 
If we use \cite{739463} as our subroutine which is the fastest known 
algorithm, the expected speedup is $p^{-\frac{2\omega}{\omega+1}}$, i.e., currently $p^{-1.407}$
where $\omega$ currently is  2.3727  \cite{williams2011breaking}.
As already outlined, in practice 
$p^{-\frac{2\omega}{\omega+1}}$, i.e., currently $p^{-1.41}$, and $p^{-2}$ respectively.

\spara{Discussion:} This theorem states the important result that the estimator of the number of triangles is concentrated around its expected value, which 
is equal to the actual number of triangles $t$ in the graph under mild conditions on the triangle density of the
graph. The mildness comes from condition \eqref{cond}: picking $p=1$, given that our graph is not triangle-free, i.e., $\Delta \geq 1$,
gives that the number of triangles $t$ in the graph has to satisfy $t \geq \Delta \log^{6+\gamma} n$. 
This is a mild condition on $t$ since  $\Delta \leq n$ and thus it suffices that $t \geq n \log^{6+\gamma} n $ 
(after all, we can always add two dummy connected nodes that connect to every other node, as in Figure 1(a),
even if in empirically $\Delta$ is smaller than $n$). 
The critical quantity besides the number of triangles $t$, is $\Delta$. Intuitively, if the sparsification
procedure throws away the common edge of many triangles, the triangles in the resulting graph may differ significantly
from the original. 
A significant problem is the choice of $p$ for the sparsification. Conditions \eqref{cond} and \eqref{cond-small} tell us how small we can afford to choose $p$, but the quantities involved, namely $t$ and $\Delta$, are unknown.
We discuss a practical algorithm using a doubling procedure in Section~\ref{subsec:doubling}. 
Furthermore, our method justifies significant speedups. For a graph $G$ with 
 $t \ge n^{3/2+\epsilon}$ and $\Delta\sim n$ , we get $p = n^{-1/2}$ implying a linear  
expected speedup if we use a practical exact counting method as the node iterator. 
Finally, it is worth pointing out that {\it Triangle Sparsifier} essentially outputs a sparse graph  
$H(V,E',w)$ with $w=1/p$ for all edges $e \in E'$ which approximates $G(V,E)$ with respect to the count of triangles
(a triangle formed by the edges $(e_1,e_2,e_3)$ in a weighted graph counts for $w(e_1)w(e_2)w(e_3)$ unweighted triangles). 
As we shall see in Chapter~\ref{conclchapter}, see Figure~\ref{fig:sparsifierfig2}, 
{\it Triangle Sparsifier} is not recommended for weighted graphs. Finally, it is worth mentioning
that the sparsification scheme which has been used for speeding up the 
computation of linear algebraic decompositions
\cite{achlioptas2001fast,mach}   has also been used to count triangles 
based on spectral properties of real-world networks \cite{tsourakakis1,tsourakakis2009spectral,tsourakakis2011spectral}.

\subsection{Experimental Results}
\label{subsec:Sparsifierexperiments}

In this Section we present our experimental findings. Specifically, in Section~\ref{subsec:sparsifierdatasets} we describe the datasets we used, 
in Section~\ref{subsec:sparsifiersetup} we give details with respect to the experimental setup and in Section~\ref{subsec:sparsifierresults} 
the experimental results.

\subsubsection{Datasets}
\label{subsec:sparsifierdatasets}

The graphs we used with the exceptions of Livejournal-links and Flickr are available on the Web. 
Table~\ref{tab:sparsifierresources} summarizes the data resources. 
We preprocessed the graphs  by first making them undirected and removing all self-loops.
Furthermore, a common phenomenon was to have multiple edges in the edge file, i.e., a file whose each line
corresponds to an edge, despite the fact that the graphs were claimed to be simple.
Those multiple edges were removed. 
Table~\ref{tab:sparsifierdatasets} summarizes the datasets we used after the preprocessing.

\begin{table}
\centering
\begin{tabular}{|c|c|} \hline
Description                                            & Availability \\ \hline 
SNAP                                                   & \url{http://snap.stanford.edu/} \\ \hline 
UF Sparse Matrix Collection                            & \url{http://www.cise.ufl.edu/research/sparse} \\\hline 
Max Planck                                             & \url{http://socialnetworks.mpi-sws.org/} \\ \hline 
\end{tabular} 
\caption{\label{tab:sparsifierresources} Dataset sources.}
\end{table}

\begin{table}[ht]
\begin{tabular}{|l|r|r|r|l|} \hline
 Name (Abbr.)  & Nodes & Edges &  Triangle Count \\ \hline 
\textcolor{red}{$\odot$}   AS-Skitter (AS) & 1,696,415 & 11,095,298 & 28,769,868 \\ \hline
\textcolor{cyan}{$\star$}Flickr (FL) & 1,861,232  & 15,555,040 &  548,658,705   \\ \hline
\textcolor{cyan}{$\star$}Livejournal-links (LJ) & 5,284,457 & 48,709,772 & 310,876,909   \\ \hline
\textcolor{cyan}{$\star$}Orkut-links (OR)  & 3,072,626 & 116,586,585 & 621,963,073   \\ \hline
\textcolor{cyan}{$\star$}Soc-LiveJournal (SL) & 4,847,571 & 42,851,237 & 285,730,264   \\ \hline
\textcolor{cyan}{$\star$}Youtube (YOU) &   1,157,822  & 2,990,442&   4,945,382    \\ \hline
\textcolor{green}{$\diamond$}Web-EDU (WE) &   9,845,725  &    46,236,104   & 254,718,147  \\ \hline
\textcolor{green}{$\diamond$}Web-Google (WG) & 875,713 & 3,852,985& 11,385,529  \\ \hline
\textcolor{green}{$\diamond$}Wikipedia 2005/11 (W0511)    & 1,634,989  & 18,540,589   & 44,667,095  \\\hline
\textcolor{green}{$\diamond$}Wikipedia 2006/9  (W0609) & 2,983,494  & 35,048,115   &84,018,183  \\ \hline
\textcolor{green}{$\diamond$}Wikipedia 2006/11 (W0611)& 3,148,440  & 37,043,456 &88,823,817    \\ \hline
\textcolor{green}{$\diamond$}Wikipedia 2007/2 (W0702) & 3,566,907 & 42,375,911  & 102,434,918   \\\hline
\end{tabular}
\caption{\label{tab:sparsifierdatasets}Datasets used in our experiments. Abbreviations are included. Symbol \textcolor{red}{$\odot$} stands for Autonomous Systems graphs, 
\textcolor{cyan}{$\star$} for online social networks and \textcolor{green}{$\diamond$} for Web graphs.
Notice that the networks with the highest triangle counts are online social networks (Flickr, Livejournal, Orkut), verifying the 
folklore that online social networks are abundant in triangles.}
\end{table}

\subsubsection{Experimental Setup}
\label{subsec:sparsifiersetup}

The experiments were performed on a single machine, with Intel Xeon CPU
at 2.83 GHz, 6144KB cache size and and 50GB of main memory. 
The algorithm was implemented in C++, and compiled using gcc version 4.1.2 and the -O3 optimization flag.
Time was measured by taking the user time given by the linux time command.
IO times are included in that time since the amount of memory operations performed in setting up the graph is non-trivial.
However, we use a modified IO routine that's much faster than the standard C/C++ scanf.
Furthermore, as we mentioned in Section~\ref{subsec:Sparsifieralgorithm} picking a random subset of expected
size $p|S|$ from a set $S$ can be done in expected sublinear time \cite{knuth}. 
A simple way to do this in practice is to generate the
differences between indices of entries retained.
This allows us to sample in a sequential way and also results in better cache performance.
As a competitor we use the single pass algorithm of \cite[$\S$ 2.2]{buriol}.  

\subsubsection{Experimental Results}
\label{subsec:sparsifierresults}

Table~\ref{tab:sparsifierdatasets} shows the count of triangles for each graph used in our experiments. 
Notice that Orkut, Flickr and Livejournal graphs have  $\sim$622M, 550M and 311M triangles respectively. 
This confirms the folklore that online social networks are abundant in triangles. 
Table~\ref{tab:sparsifierresults} shows the results we obtain for $p=0.1$ over 5 trials.  
All running times are reported in seconds. The first column shows the running time for the exact counting algorithm over 5 runs. 
Standard deviations are neglibible for the exact algorithm and therefore are not reported. 
The second and third column show the error and running time averaged over 5 runs for each dataset (two decimal digits of accuracy).
Standard deviations are also included (three decimal digits of accuracy). 
The last column shows the running time averaged over 5 runs for the 1-pass algorithm as stated in 
\cite[$\S$2.2]{buriol} and the standard deviations. 
For each dataset the number of samples needed by the 1-pass algorithm was set to a value that achieves {\em at most} as good accuracy as the ones  
achieved by our counting method. Specifically, for any dataset,  if $\alpha,\beta (\%)$ are the errors obtained by our algorithm 
and the Buriol et al. algorithm, we ``tune'' the number of samples in the latter algorithm in such way that 
$ \alpha \leq \beta \leq \alpha+1\%$. 
Even by favoring in this way the 1-pass algorithm of Buriol et al. \cite{buriol},
one can see that the running times achieved by our method are consistently better. 
However, it is important to outline once again that our method and other triangle counting methods can be combined. 
For example, in Section~\ref{sec:degreepartitioning} we show that Triangle Sparsifiers and other sampling methods 
can be combined to obtain a superior performance both in practice and theory by improving 
the sampling scheme of Buriol et al. \cite{buriol}. As we will see in detail, 
this is achieved by distinguishing vertices into two subsets according to their 
degree and using two sampling schemes, one for each subset \cite{tsourakakis4,DBLP:journals/im/KolountzakisMPT12}.
We also tried other competitors, but our running times  outperform  them significantly. 
For example, even the exact counting method outperforms other approximate counting methods.  
As we show in Section~\ref{subsec:doubling} smaller values of $p$ values work as well and these can be found by a simple procedure. 

\begin{table}[ht]
\begin{center}
\begin{tabular}{p{1cm}|c|c|c|c|}\cline{2-5}
                  &  \multicolumn{4}{|c|}{Results}                                        \\ \cline{2-5}
                  &  \multicolumn{1}{|c|}{Exact}          & \multicolumn{2}{|c|}{Triangle Sparsifier} & \multicolumn{1}{|c|}{Buriol et al. \cite{buriol}} \\ \cline{2-5}
                  &   Avg. time                           &  Avg. err.\% (std) & Avg. time (std) &  Avg. time (std)   \\ \cline{1-5} \hline
AS                & 4.45                                  & 2.60 (0.022)       & 0.79 (0.023)    &  2.72 (0.128)       \\ \cline{1-5}
FL                & 41.98                                 & 0.11 (0.003)       & 0.96 (0.014)    &  3.40 (0.175)        \\  \cline{1-5}
LJ                & 50.83                                 & 0.34 (0.001)       & 2.85 (0.054)    & 12.40 (0.250)         \\  \cline{1-5}
OR                & 202.01                                & 0.60 (0.004)       & 5.60 (0.159)    & 11.71 (0.300)          \\  \cline{1-5}
SL                & 38.27                                 & 8.27 (0.006)       & 2.50 (0.032)    & 8.92  (0.115)           \\  \cline{1-5}
YOU               & 1.35                                  & 1.50 (0.050)       & 0.30 (0.002)    & 10.91 (0.130)            \\ \cline{1-5}
WE                & 8.50                                  & 0.70 (0.005)       & 2.79 (0.090)    & 6.56  (0.025)             \\  \cline{1-5}
WG                & 1.60                                  & 1.58 (0.011)       & 0.40 (0.004)    & 1.85  (0.047)              \\ \cline{1-5}
W0511             & 32.47                                 & 1.53 (0.010)       & 1.19 (0.020)    & 3.71  (0.038)               \\ \cline{1-5}
W0609             & 86.62                                 & 0.40 (0.055)       & 2.07 (0.014)    & 8.10  (0.040)                \\ \cline{1-5}
W0611             & 96.11                                 & 0.62 (0.008)       & 2.16 (0.042)    & 7.90  (0.090)                 \\ \cline{1-5}
W0702             & 122.34                                & 0.80 (0.015)       & 2.48 (0.012)    & 11.00 (0.205)                  \\ \cline{1-5}
\end{tabular}
\end{center}
\caption{\label{tab:sparsifierresults}Results of experiments averaged over 5 trials using $p=0.1$. All running times are reported in seconds. 
The first column shows the running time for the exact counting algorithm averaged over 5 runs.
The second and third column show the error and running time averaged over 5 runs for each dataset (two decimal digits of accuracy).
Standard deviations are also included (three decimal digits of accuracy). The last column shows the running time averaged over 
5 runs for the 1-pass algorithm as stated in \cite[$\S$2.2]{buriol} and the corresponding standard deviations. 
The number of samples for each dataset was set to a value that achieves {\em at most} as good accuracy as the ones 
achieved by our counting method. See Section~\ref{subsec:sparsifierresults} for all the details.}
\end{table}

\subsubsection{The ``Doubling'' Algorithm}
\label{subsec:doubling}

As we saw in Section~\ref{subsec:Sparsifieralgorithm}, setting optimally the parameter $p$ requires knowledge 
of the quantity we want to estimate, i.e., the number of triangles. 
To overcome this problem we observe that when we have concentration, the squared coefficient
of variation $\frac{\text{Var}[T]}{\Mean{T}^2}$ is ``small''. Furthermore, by the Chebyshev inequality
and by the median boosting trick \cite{jerrum} it suffices to sample $\{T_1,\ldots,T_s\}$ where $ s= O( \frac{\text{Var}[T]}{\Mean{T}^2} \frac{1}{\epsilon^2}\ln{\frac{1}{\delta}} )$
in order to obtain a $(1\pm \epsilon)$ approximation $\Mean{T}=t$ with probability at least 1-$\delta$. 
Hence, one can set a desired value for the number of samples $s$ and of the failure probability $\delta$ 
and calculate the expected error  $\epsilon = O( \sqrt{ \frac{\text{Var}[T]}{\Mean{T}^2} \frac{1}{s}\ln{\frac{1}{\delta}}} )$.
If this value is significantly larger than the desired error threshold then one increases $p$ and repeats the same procedure
until the stopping criterion is satisfied. One way one can change $p$ is to use the multiplicative 
rule $p\leftarrow c p$,  where $c>1$ is a constant. For example, if $c=2$ then we have a doubling procedure. 
Notice that we've placed the word doubling in the title of this section in quotes in order to emphasize that 
one may use any $c>1$ to change $p$ from one round to the next.

For how many rounds can this procedure run? Let's consider the realistic scenario where one wishes
to be optimistic and picks as an initial guess for $p$ a value $p_0=n^{-\alpha}$ where $\alpha$ is a positive
constant, e.g., $\alpha=1/2$. Let $p*$ be the minimum value over all possible $p$ with the property 
that for $p*$ we obtain a concentrated estimate of the true number of triangles. Clearly, $p*\leq 1$ 
and hence the number of rounds performed by our procedure is less that $r$ where $ p_0 c^r = 1$. 
Hence, for any constant $c>1$ we obtain that the number of rounds performed by our algorithm is $O(\log{n})$. 
Furtermore,  note that the running time of the doubling procedure is dominated by the last iteration.
To see why, consider for simplicity the scenario where $r+1$ rounds are needed to deduce concentration,
$c=\sqrt{2}$ and the use of the node-iterator algorithm to count triangles in the triangle sparsifier.  
Then, the total running time shall be $p_0^2 \sum_{v\in V(G)} d(v)^2 \Bigg(  1+2+\ldots+2^{r-1}+2^r   \Bigg)$. 
Finally, observe that $1+2+...+2^{r-1}=O(2^{r})$.
In practice, this procedure works even for small values of $s$. An instance of this procedure 
with $s=2$, $\delta=1/100$ and error threshold equal to $3\%$ is shown in Table~\ref{tab:doublingexample}.

\begin{table}[ht]
\begin{center} 
\begin{tabular}{|c|c|c|c|} \hline
p    & $\{T_1 , T_2 \}$                 & $ \sqrt{ \frac{\text{Var}[T]}{\Mean{T}^2} \frac{1}{s}\ln{\frac{1}{\delta}}} $ &  err(\%) \\ \hline 
0.01 & $\{ 42,398,007 ~\& ~50,920,488\}$  & 0.1960    & 4.46     \\ \hline 
0.02 & $\{ 42,540,941 ~\& ~43,773,753 \}$ & 0.0307    & 3.38     \\ \hline 
0.04 & $\{ 44,573,294 ~\& ~43,398,549 \}$ & 0.0287    & 1.52    \\ \hline 
\end{tabular}
\end{center}
\caption{Doubling procedure for the Wikipedia 2005 graph with 44,667,095 triangles.} 
\label{tab:doublingexample}
\end{table}

\subsection{Theoretical Ramifications} 
\label{subsec:Sparsifierdiscussion}

In this Section  we investigate the performance of the Bencz{\'u}r-Karger  cut sparsifier 
and the Spielman-Srivastava spectral sparsifier with respect to triangle counting. 

\begin{figure}[htbp]
\begin{center}
\includegraphics[width=8cm]{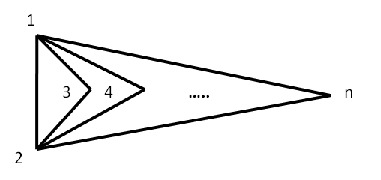}
\end{center}
\caption{Graph with linear number $O(n)$ of triangles.}
\label{fig:fig1}
\end{figure}

Consider the graph $G$ shown in Figure~\ref{fig:fig1}. The strong edge connectivity of any edge in the graph 
is 2 and therefore the Bencz{\'u}r-Karger algorithm does not distinguish the importance
of the edge $e=(1,2)$ with respect to triangle counting.  
The Spielman-Srivastava sparsifier with probability $1-o(1)$ throws away the critical edge $e=(1,2)$
as the number of vertices $n$ tends to infinity. To prove this claim, we need 
use Foster's theorem~\ref{thrm:fosterstheorem}.

\begin{claim}
The effective resistance $R(1,2)$ of the edge $(1,2)$ tends to 0 as $n$ grows to infinity, i.e., $R(1,2)=o(1)$. 
Furthermore, all other edges have constant effective resistance. 
\end{claim}

\begin{proof} 

Using the in-series and in-parallel network simplification rules \cite{bollobas}, 
the effective conductance of the edge $(1,2)$ is $ 1 + \sum_{i=1}^{n-2}\frac{1}{2}=\frac{n}{2}$.
Hence, the effective resistance of the edge $e=(1,2)$ is $2/n$, which also 
proves the first part of our claim. By Foster's theorem~\ref{thrm:fosterstheorem}, the sum of the effective resistances
of the edges of $G$ is $n-1$. 
Due to symmetry, we obtain that $R_{(1,3)}=R_{(2,3)}=R_{(1,4)}=R_{(2,4)}= \ldots =R_{(1,n)}=R_{(2,n)}=R_n$. 
Therefore we obtain $ \frac{2}{n} + 2(n-2) R_n = n-1 \rightarrow R_n = \frac{n^2 - n- 2}{2n^2-4n}$. 
Asymptotically as $n \rightarrow +\infty$, $R_n \rightarrow \frac{1}{2}$. 
\end{proof} 

Clearly, the Spielman-Srivastava sparsifier fails to capture the importance of the edge $(1,2)$ with respect to triangle counting.
Finding an easy-to-compute quantity which allows a sparsification that preserves triangles more efficiently is an interesting problem. 
It is worth outlining that our analysis does not exclude effective resistances which can be computed very efficiently \cite{yannis}, 
but the use of them as is typically done in the context of spectral sparsifiers.

\section{Efficient Triangle Counting in Large Graphs via Degree-based Vertex Partitioning}
\label{sec:degreepartitioning}

\subsection{Proposed Method}
\label{subsec:DegreeBasedmethod}

Our algorithm combines two approaches that have been taken on triangle counting:
sparsify the graph by keeping a random subset of the edges, see Section~\ref{sec:trianglesparsifiers},
followed by a triple sampling using the idea of vertex partitioning due to Alon, Yuster and Zwick \cite{DBLP:journals/algorithmica/AlonYZ97}.
In the following, we shall assume that the input is in the form of an edge file, i.e., a file whose 
each line contains an edge. Notice that given this representation, computing the degrees takes linear time.

\subsubsection{Edge Sparsification}

As we saw in Section~\ref{sec:trianglesparsifiers}  the following method performs very well 
in practice: keep each edge with probability $p$ independently.
Then for each triangle, the probability of being kept is $p^3$. So the
expected number of triangles left is $p^3t$.
This is an inexpensive way to reduce the size of the graph as it can be done in
one pass over the edge list using $O(mp)$ random variables.

We also proved that from the number of triangles in the sampled graph we can o
btain a concentrated estimate  around the actual triangle count
as long as $p^3 \geq \tilde{\Omega}(\frac{\Delta}{t})$\footnote{We use the tilde notation to hide polylogarithmic factors $polylog(n)$.}.
Here, we show a similar bound using more elementary techniques.
Suppose we have a set of $k$ triangles such that no two share an edge.
For each such triangle we define a random variable $X_i$ which is $1$ if
the triangle is kept by the sampling and $0$ otherwise.
Then as the triangles do not have any edges in common,
the $X_i$s are independent and take value $0$ with probability $1-p^3$ and
$1$ with probability $p^3$.
So by Chernoff inequality~\ref{lem:chernoff} 

$$Pr \left[ |\frac{1}{k} \sum_{i=1}^k X_i - p^3| > \epsilon p^3 \right] \leq 2e^{-\epsilon^2p^3k/2}.$$

So when $p^3 k \epsilon^2 \geq 4d\log{n}$ where $d$ is a positive constant, the probability of sparsification 
returning an $\epsilon$-approximation is at least $1-n^{-d}$.
This is equivalent to $p^3 k \geq (4d\log{n}) /( \epsilon^2)$ which suggests 
that in order to sample with small $p$ and hence discard many edges we need like $k$ to be large.
To show that such a large set of independent triangles exist, we invoke the Hajnal-Szemer\'{e}di 
Theorem~\ref{lem:hajnal} on an auxiliary graph $H$ which we construct as follows. 
For each triangle $i$ ($i=1,\ldots,t$) in $G$ we create a vertex $v_i$ in $H$.
We connect two vertices $v_{i},v_{j}$ in $H$ if and only if 
they represent triangles  $i,j$ respectively which share an edge in $G$.
Notice that the maximum degree in the auxiliary graph $H$ is $O(\Delta)$.
Hence, we obtain the following Corollary. 

\begin{corollary} 
\label{cor:partition}
Given $t$ triangles such that  no edge belongs to more than $\Delta$ triangles,
we can partition the triangles into sets $S_1 \dots S_l$ such that
$|S_i| > \Omega(t / \Delta)$ and $l$ is bounded by $O(\Delta)$.
\end{corollary}

\noindent Combining Corollary~\ref{cor:partition} and the Chernoff bound allows 
us to prove the next theorem.

\begin{theorem}
If $p^3 \in \Omega(\frac{ \Delta\log{n}}{\epsilon^2 t})$, then with probability
$1-n^{-2}$, the sampled graph has a triangle count that $\epsilon$-approximates
$t$.
\end{theorem}

\begin{proof}
Consider the partition of triangles given by corollary \ref{cor:partition}
and let $d=5$. 
By choice of $p$ we get that the probability that the triangle count in
each set is preserved within a factor of $\epsilon/2$ is at least $1-n^{-d}$.
Since there are at most $n^3$ such sets, an application of the union
bounds gives that their total is approximated within a factor
of $\epsilon/2$ with probability at least $1 - n^{3-d}$.
This gives that the triangle count is approximated
within a factor of $\epsilon$ with probability at least $1 - n^{3-d}$.
Substituting  $d=5$ completes the proof. 
\end{proof} 

\subsubsection{Triple Sampling}

Since each triangle corresponds to a triple of vertices, we can construct a set of
triples $U$ that include all triangles.  From this list, 
we can then sample triples uniformly at random. Let these samples be numbered from $1$ to $s$.
Also, for the $i^{th}$ triple sampled, let $X_i$ be $1$ if it is a triangle and $0$ otherwise.
Since we pick triples randomly from $U$ and $t$ of them are triangles,
we have $E(X_i) = \frac{t}{|U|}$ and $X_i$s are independent.
So by Chernoff bound we obtain:

$$Pr \left[ |\frac{1}{s} \sum_{i=1}^s X_i - \frac{t}{|U|}| > \epsilon \frac{t}{|U|}
 \right] \leq 2e^{-\epsilon^2ts/(2|U|)}$$

If $s = \Omega( \tfrac{ |U|\log{n} }{ t\epsilon^2 })$, then we have that $|U| \sum_{i=1}^s \tfrac{X_i}{s}$
approximates $t$ within a factor of $\epsilon$ with probability at least $1-n^{-d}$ for any $d$ of our choice.
As $|U| \leq n^3$, this immediately gives an algorithm with runtime $O(n^3\log{n}/(t\epsilon^2))$ that approximates $t$ within a factor
of $\epsilon$. Slightly more careful bookkeeping can also give tighter bounds on $|U|$ in sparse graphs.

A simple but crucial observation which allows us to decide whether 
we will sample a triple of vertices or an edge and a vertex is the following. 
Consider any triple containing vertex $u$, $(u,v,w)$.
Since $uv, uw \in E$, we have the number of such triples involving $u$ is at
most $d(u)^2$. From an edge-vertex sampling point of view, as $vw \in E$, another bound on the number of such triples is $m$.
When $d(u) > m^{1/2}$ , the second  bound is tighter, and the first is in the other case.

These two cases naturally suggest that low degree vertices with degree at most
$m^{1/2}$ be treated separately from high degree vertices with degree greater than $m^{1/2}$.
For the number of triangles around low degree vertices, 
the value of $\sum_u d(u)^2$ is maximized
when all edges are concentrated in as few vertices as possible \cite{ahlswede}.
Since the maximum degree of such a vertex is $m^{1/2}$, the number of such
triangles is upper bounded by $m^{1/2} \cdot (m^{1/2})^2 = m^{3/2}$.
Also, as the sum of all degrees is $2m$, there can be at most $2m^{1/2}$ high
degree vertices, which means the total number of triangles incident to these
high degree vertices is at most $2m^{1/2} \cdot m = 2m^{3/2}$.
Combining these bounds give that $|U|$ can be upper bounded by $3m^{3/2}$.
Note that this bound is asymptotically tight when $G$ is a complete graph ($n = m^{1/2}$).
However, in practice the second bound can be further reduced by summing over the degree of all $v$ adjacent to $u$, becoming 
$\sum_{uv \in E} d(v)$.
As a result, an algorithm that implicitly constructs $U$ by picking the better one among these two cases by examining the degrees of all neighbors will achieve $|U| \leq O(m^{3/2})$.

\noindent This improved bound on $U$ gives an algorithm that $\epsilon$ approximates the number of triangles in time:

$$O \left( m + \frac{m^{3/2} \log{n}}{t \epsilon^2} \right)$$

As our experimental data in Section 4.1 indicate, 
the value of $t$ is usually $\Omega(m)$ in practice.
In such cases, the second term in the above calculation becomes negligible
compared to the first one.
In fact, in most of our data, just sampling the first type of triples (aka. pretending
all vertices are of low degree) brings the second term below the first.

\subsubsection{Hybrid algorithm}

Edge sparsification with a probability of $p$ allows us to only work on $O(mp)$
edges, therefore the total runtime of the triple sampling algorithm after
sparsification with probability $p$ becomes:
$$O\left(mp+\frac{\log{n}(mp)^{3/2} }{\epsilon^2 tp^3} \right) =
O\left(mp+\frac{\log{n} m^{3/2}}{\epsilon^2 t p^{3/2}} \right). $$

As stated above, since the first term in most practical cases are much larger,
we can set the value of $p$ to balance these two terms out:

\begin{align*}
pm &= \frac{m^{3/2}\log{n}}{p^{3/2}t\epsilon^2} \Rightarrow p^{5/2} t \epsilon^2 = m^{1/2} \log{n} \Rightarrow p = \left( \frac{m^{1/2}\log{n}}{t \epsilon^2} \right) ^{2/5}
\end{align*}

The actual value of $p$ picked would also depend heavily on constants in front
of both terms, as sampling is likely much less expensive due to factors such as
cache effect and memory efficiency.
Nevertheless, our experimental results in section 4 does seem to indicate that
this type of hybrid algorithms can perform better in certain situations.

\subsubsection{Sampling in the Semi-Streaming Model}

The previous analysis of triangle counting by Alon, Yuster and Zwick was done
in the streaming model \cite{739463}, where the assumption was constant available space.
We show that our sampling algorithm can be done in a slightly weaker model
with space usage equaling:

$$O\left(m^{1/2}\log{n} + \frac{m^{3/2} \log{n}}{t \epsilon^2} \right)$$

We assume the edges adjacent to each vertex are given in order \cite{feigenbaum2005graph}.
We first need to identify high degree vertices, specifically the ones with degree
higher than $m^{1/2}$. This can be done by sampling $O(m^{1/2}\log{n})$ edges
and recording the vertices that are endpoints of one of those edges.

\begin{lemma}
Suppose $d m^{1/2} \log{n}$ samples were taken, then the probability of all
vertices with degree at least $m^{1/2}$ being chosen is at least $1-n^{-d+1}$.
\end{lemma}
\proof
Consider some vertex $v$ with degree at least $m^{1/2}$.
The probability of it being picked in each iteration is at least
$m^{1/2} / m = m^{-1/2}$.
As a result, the probability of it not picked in $d m^{1/2} \log{n}$ iterations is:
$$(1 - m^{-1/2})^{dm^{1/2}\log{n}}
= \left [(1-m^{1/2})^{m^{1/2}} \right ] ^ {d \log{n}}
\leq \left ( \frac{1}{e} \right ) ^ {d \log{n}} = n^{-d}$$
As there are at most $n$ vertices, applying union bound gives that all vertices
with degree at least $m^{1/2}$ are sampled with probability at least $1-n^{-d + 1}$.
\qed

Our proposed method is comprised of the following three steps/passes over the stream.

\begin{enumerate} 
\item Identifying high degree vertices requires one pass of the graph. Also, note that the number of potential candidates 
can be reduced to $m^{1/2}$ using another pass over the edge list.
\item For all the low degree vertices, we can read their $O(m^{1/2})$ neighbors
and sample from them.
For the high degree vertices, we do the following: for each edge, obtain a random
variable $y$ from a binomial distribution equal to the number of
edge/vertices pairs that this edge is involved in.
Then pick $y$ vertices from the list of high degree vertices randomly.
These two sampling procedures can be done together in another pass over the data.
\item 
Finally, we need to check whether each edge in the sampled triples belong to the
edge list.
We can store all such queries into a hash table as there are at most
$O(\frac{m^{3/2} \log{n}}{t \epsilon^2})$ edges sampled w.h.p.
Then going through the graph edges in a single pass and looking them up in
table yields the desired answer.
\end{enumerate} 

\subsection{Experiments}
\label{subsec:DegreeBasedexperiments}

\subsubsection{Data}

The graphs used in our experiments are shown in Table~\ref{tab:degreedatasets}. Multiple edges and self loops were removed (if any).  
All graphs with the exceptions of Livejournal-links and Flickr are available on the Web. 
Table~\ref{tab:sparsifierresources} summarizes the resources.

\begin{table}[ht]
\begin{center}
\begin{tabular}{|l|r|r|r|l|} \hline
 Name  & Nodes & Edges &  Triangle Count & Description\\ \hline 

AS-Skitter & 1,696,415 & 11,095,298 & 28,769,868& Autonomous Systems \\ \hline
Flickr & 1,861,232  & 15,555,040 &  548,658,705 & Person to Person \\ \hline
Livejournal-links & 5,284,457 & 48,709,772 & 310,876,909 & Person to Person \\ \hline
Orkut-links  & 3,072,626 & 116,586,585 & 621,963,073 & Person to Person \\ \hline
Soc-LiveJournal & 4,847,571 & 42,851,237 & 285,730,264 & Person to Person \\ \hline
Web-EDU&   9,845,725  &    46,236,104   & 254,718,147 &   Web Graph (page to page) \\ \hline
Web-Google & 875,713 & 3,852,985& 11,385,529 & Web Graph \\ \hline
Wikipedia 2005/11    & 1,634,989  & 18,540,589   & 44,667,095 & Web Graph    (page to page) \\\hline
Wikipedia 2006/9  & 2,983,494  & 35,048,115   &84,018,183& Web Graph  (page to page) \\ \hline
Wikipedia 2006/11 & 3,148,440  & 37,043,456 &88,823,817  & Web Graph (page to page) \\ \hline
Wikipedia 2007/2 & 3,566,907 & 42,375,911  & 102,434,918 & Web Graph (page to page) \\\hline
Youtube  &   1,157,822  & 2,990,442&   4,945,382  & Person to Person \\ \hline
\end{tabular}
\end{center}
\caption{\label{tab:degreedatasets}Datasets used in our experiments.}
\end{table}

\subsubsection{Experimental Setup and Implementation Details} 

The experiments were performed on a single machine, with Intel Xeon CPU
at 2.83 GHz, 6144KB cache size and and 50GB of main memory. 
The graphs are from real world web-graphs, some details regarding them are in Table~\ref{tab:sparsifierresources} and 
in Table~\ref{tab:degreedatasets}.
The algorithm was implemented in C++, and compiled using gcc version 4.1.2 and
the -O3 optimization flag. Time was measured by taking the user time given by the linux time command.
IO times are included in that time since the amount of memory operations
performed in setting up the graph is non-negligible. 
However, we use a modified IO routine that's much faster than the standard C/C++ scanf.

A major optimization that we used was to sort the edges in the graph and store
the input file in the format as a sequence of neighbor lists per vertex.
Each neighbor list begins with the size of the list, followed by the neighbors.
This is similar to how software like Matlab stores sparse matrices.
The preprocessing time to change the data into this format is not included.
It can significantly improve the cache property of the graph stored, and hence the overall performance.

Some implementation details are based on this graph storage format.
Specifically, since each triple that we check  by definition 
has 2 edges already in the graph, it suffices to check/query whether the 3rd edge is present in the graph.  
In order to do this efficiently, rather than querying the existence of an
edge upon sampling each triple, we store the entire set of the queries 
and answer them in one pass through the graph..  
Finally, in the next section we discuss certain details behind
efficient binomial sampling. Specifically picking a random subset of expected
size $p|S|$ from a set $S$ can be done in expected sublinear time,
as we already saw in Claim~\ref{claim:claim1}.

\subsubsection{Binomial Sampling in Expected Sublinear time}
\label{sec:sublin}

Most of our algorithms have the following routine in their core: given a list of values,
keep each of them with probability $p$ and discard with probability $1-p$.
If the list has length $n$, this can clearly be done using $n$ random variables.
As generating random variables can be expensive, it's preferrable to use
$O(np)$ random variables in expectation if possible.
One possibility is to pick $O(np)$ random elements, but this would likely
involve random accesses in the list, or maintaining a list of the indices picked
in sorted order. A simple way that we use in our code to perform this sampling is to generate the
differences between indices of entries retained \cite{knuth}.
This variable clearly belongs to an exponential distribution, and if $x$ is a uniform
random number in $(0, 1)$, taking $\lceil \log_{(1-p)}x \rceil$ as the value of the random variable, see \cite{knuth}. 
The primary advantage of doing so is that sampling can be done while accessing
the data in a sequential fashion, which results in much better cache performances.

\subsubsection{Results}

The six variants of the code involved in the experiment are first separated by
whether the graph was first sparsified by keeping each edge with probability $p = 0.1$.
In either case, an exact algorithm based on hybrid sampling with performance
bounded by $O(m^{3/2})$ was run.
Then two triple based sampling algorithms are also considered.
They differ in whether an attempt to distinguish between low and high degree vertices,
so the simple version is essentially sampling all 'V' shaped triples off each
vertex. Note that no sparsification and exact also generates the exact number of
triangles. Errors are measured by the absolute value of the difference between the value
produced and the exact number of triangles divided by the exact number.
The results on error and running time are averaged over five runs. 
The results are shown in Tables~\ref{tab:degreeresults1},~\ref{tab:degreeresults2}.

\begin{table}[ht]
\begin{center}
\begin{tabular}{|p{2.5cm}|c|c|c|c|c|c|c}\hline 
           &  \multicolumn{6}{|c|}{No Sparsification} \\ \hline
Graph      &  \multicolumn{2}{|c|}{Exact} & \multicolumn{2}{|c|}{Simple} &\multicolumn{2}{|c|}{Hybrid} \\ \hline
                  & err(\%) & time & err(\%) & time & err(\%) & time \\ \hline 
AS-Skitter        & 0.000   & 4.452 & 1.308 & 0.746 & 0.128 & 1.204  \\ \hline
Flickr            & 0.000   & 41.981 & 0.166 & 1.049 & 0.128 & 2.016  \\ \hline
Livejournal-links & 0.000   & 50.828 & 0.309 & 2.998 & 0.116 & 9.375  \\ \hline
Orkut-links       & 0.000   & 202.012 & 0.564 & 6.208 & 0.286 & 21.328  \\ \hline
Soc-LiveJournal   & 0.000   & 38.271 & 0.285 & 2.619 & 0.108 & 7.451  \\ \hline
Web-EDU           & 0.000   & 8.502 & 0.157 & 2.631 & 0.047 & 3.300   \\ \hline
Web-Google        & 0.000   & 1.599 & 0.286 & 0.379 & 0.045 & 0.740  \\ \hline
Wiki-2005         & 0.000   & 32.472 & 0.976 & 1.197 & 0.318 & 3.613  \\ \hline
Wiki-2006/9       & 0.000   & 86.623  & 0.886 & 2.250 & 0.361 & 7.483   \\ \hline
Wiki-2006/11      & 0.000   & 96.114 & 1.915 & 2.362 & 0.530 & 7.972  \\ \hline
Wiki-2007         & 0.000   & 122.395 & 0.943 & 2.728 & 0.178 & 9.268  \\ \hline
Youtube           & 0.000   & 1.347   & 1.114 & 0.333 & 0.127 & 0.500  \\ \hline 
\end{tabular}
\end{center}
\caption{\label{tab:degreeresults1}Results of experiments averaged over 5 Trials using only triple sampling.}
\end{table}

\subsubsection{Remarks}

From Table~\ref{tab:degreedatasets} it is evident that social networks are abundant 
in triangles. For example, the Flickr graph with only $\sim$1.9M vertices 
has $\sim$550M triangles and the Orkut graph with $\sim$3M vertices has $\sim$620M 
triangles. 
Furthermore, from Table~\ref{tab:degreeresults1} and Table~\ref{tab:degreeresults2} it is clear that none 
of the variants clearly outperforms the others on all the data. 
The gain/loss from sparsification is likely due to the fixed sampling rate. 
Adapting a doubling procedure for the sampling rate as in Section~\ref{subsec:doubling} 
is likely to mitigate this discrepancy. The difference between simple and hybrid sampling are due to the fact
that handling the second case of triples has a much worse cache access
pattern as it examines vertices that are two hops away.
There are alternative implementations of how to handle this situation, which
would be interesting for future implementations.
A fixed sparsification rate of $p = 10\%$ was used mostly to simplify the setups
of the experiments. In practice varying $p$ to look for a rate where the results stabilize is the preferred
option.

When compared with previous results on this problem, the error rates and running
times of our results are all significantly lower.
In fact, on the wiki graphs our exact counting algorithms have about the same order 
of speed with other appoximate triangle counting implementations. 
This is also why we did not include any competitors in the exposition of the results 
since our implementation is a highly optimized C/C++ implementation with an emphasis on performance for huge graphs.

\begin{table}[ht]
\begin{center}
\begin{tabular}{|p{2.5cm}|c|c|c|c|c|c|c}\hline 
           &  \multicolumn{6}{|c|}{Sparsified ($p = 0.1$)} \\ \hline
Graph      &  \multicolumn{2}{|c|}{Exact} & \multicolumn{2}{|c|}{Simple} &\multicolumn{2}{|c|}{Hybrid} \\ \hline
                  & err(\%) & time & err(\%) & time & err(\%) & time \\ \hline 
AS-Skitter        & 2.188 & 0.641 & 3.208 & 0.651 & 1.388 & 0.877 \\ \hline
Flickr            & 0.530 & 1.389 & 0.746 & 0.860 & 0.818 & 1.033 \\ \hline
Livejournal-links & 0.242 & 3.900 & 0.628 & 2.518 & 1.011 & 3.475 \\ \hline
Orkut-links       & 0.172 & 9.881 & 1.980 & 5.322 & 0.761 & 7.227 \\ \hline
Soc-LiveJournal   & 0.681 & 3.493 & 0.830 & 2.222 & 0.462 & 2.962 \\ \hline
Web-EDU           & 0.571 & 2.864 & 0.771 & 2.354 & 0.383 & 2.732 \\ \hline
Web-Google        & 1.112 & 0.251 & 1.262 & 0.371 & 0.264 & 0.265 \\ \hline
Wiki-2005         & 1.249 & 1.529 & 7.498 & 1.025 & 0.695 & 1.313 \\ \hline
Wiki-2006/9       & 0.402 & 3.431 & 6.209 & 1.843  & 2.091 & 2.598 \\ \hline
Wiki-2006/11      & 0.634 & 3.578 & 4.050 & 1.947 & 0.950 & 2.778 \\ \hline
Wiki-2007         & 0.819 & 4.407 & 3.099 & 2.224 & 1.448 & 3.196 \\ \hline
Youtube           & 1.358 & 0.210 & 5.511 & 0.302 & 1.836 & 0.268 \\ \hline 
\end{tabular}
\end{center}
\caption{Results of experiments averaged over 5 trials using sparsification and triple sampling.}
\label{tab:degreeresults2}
\end{table}

As we mentioned earlier in Section~\ref{sec:subgraphintro} there exists a lot of interest into signed networks. 
It is clear that our method applies to this setting as well, by considering individually each possible configuration of a signed triangle. 
However, we do not include any of our experimental findings here due to the small size of the 
signed networks available to us via the Stanford Network Analysis library (SNAP).

\subsection{Theoretical Ramifications} 
\label{subsec:DegreeBasedramifications}

In Section~\ref{subsec:rptriangles} we discuss random projections and triangles, motivated by the simple observation
that the inner product of two rows of the adjacency matrix corresponding to two connected
vertices forming edge $e$ gives the count of triangles $\Delta(e)$. 

\subsubsection{Random Projections and Triangles}
\label{subsec:rptriangles} 

Consider any two vertices $i,j \in V$ which are connected, i.e., $(i,j) \in E$. Observe that the inner product
of the $i$-th and $j$-th column of the adjacency matrix of graph $G$ gives the number of triangles that 
edge $(i,j)$ participates in. Viewing the adjacency matrix as a collection of $n$ points in $\field{R}^n$,
a natural question to ask is whether we can use results from the theory of random projections \cite{lindenstrauss}
to reduce the dimensionality of the points while preserving the inner products which contribute to the count
of triangles. Magen and Zouzias \cite{magen} have considered a similar problem, namely random projections which 
preserve approximately the volume for all subsets of at most $k$ points.

According to Lemma~\ref{thrm:jllemma} projecton $x \to Rx$ from $\field{R}^d \to \field{R}^k$ approximately preserves all Euclidean distances.
However it does not preserve all pairwise inner products. This can easily be seen by considering the set of points
$e_1, \ldots, e_n \in \field{R}^n = \field{R}^d$ where $e_1=(1,0,\ldots,0)$ etc. 
Indeed, all inner products of the above set are zero, which cannot happen for the points $R e_j$ as they belong to a lower dimensional space 
and they cannot all be orthogonal. For the triangle counting problem we do not need 
to approximate {\em all} inner products. 
Suppose $A \in \Set{0,1}^n$ is the adjacency matrix of a simple undirected graph $G$ 
with vertex set $V(G) = \Set{1,2,\ldots,n}$ and write $A_i$ for the $i$-the column of $A$. 
The quantity we are interested in is the number of triangles in $G$ (actually six times the number of triangles)
$$
t = \sum_{u,v,w \in V(G)} A_{uv} A_{vw} A_{wu}.
$$

\noindent If we apply a random projection of the above kind to the columns of $A$
$$
A_i \to R A_i
$$
and write
$$
X = \sum_{u,v,w \in V(G)} (RA)_{uv} (RA)_{vw} (RA)_{wu}
$$

\noindent  it is easy to see that $\Mean{X}=0$ since $X$ is a linear combination of triple products $R_{ij} R_{kl} R_{rs}$ of entries of the 
random matrix $R$ and that all such products have expected value $0$, no matter what the indices. So we cannot expect this kind of random projection to work.

Therefore we consider the following approach which still has limitations as we will show in the following. Let
$$
t = \sum_{u \sim v} A_u^\top A_v,\ \ \mbox{where $u \sim v$ means $A_{uv}=1$},
$$
and look at the quantity
\begin{eqnarray*}
Y &=& \sum_{u \sim v} (R A_u)^\top (RA_v)\\
 &=& \sum_{l=1}^k \sum_{i,j=1}^n \left(\sum_{u \sim v} A_{iu} A_{jv} \right) R_{li} R_{lj}\\
 &=& \sum_{l=1}^k \sum_{i,j=1}^n \#\Set{i-*-*-j} R_{li} R_{lj}.
\end{eqnarray*}
This is a quadratic form in the gaussian $N(0,1)$ variables $R_{ij}$.
By simple calculation for the mean value and diagonalization for the variance we see that if the $X_j$ are independent $N(0,1)$ variables and
$$
Z = X^\top B X,
$$
where $X = (X_1,\ldots,X_n)^\top$ and $B \in \field{R}^{n \times n}$ is \textit{symmetric}, that

\begin{eqnarray*}
\Mean{Z} &=& \tr{B}\\
\Var{Z} &=& \tr{B^2} = \sum_{i,j=1}^n (B_{ij})^2.
\end{eqnarray*}

Hence $\Mean{Y} = \sum_{l=1}^k \sum_{i=1}^n \#\Set{i-*-*-i} = k\cdot t$ so the mean value is the quantity we want (multiplied by $k$).
For this to be useful we should have some concentration for $Y$ near $\Mean{Y}$. 
We do not need exponential tails because we have only one quantity to control. 
In particular, a statement of the following type
$$
\Prob{\Abs{Y-\Mean{Y}} > \epsilon \Mean{Y}} < 1-c_{\epsilon},
$$
where $c_{\epsilon}>0$ would be enough.
The simplest way to check this is by computing the standard deviation of $Y$. 
By Chebyshev's inequality it suffices that the standard deviation be much smaller than $\Mean{Y}$.
According to the formula above for the variance of a quadratic form we get
\begin{eqnarray*}
\Var{Y} &=& \sum_{l=1}^k \sum_{i,j=1}^n \#\Set{i-*-*-i}^2\\
 &=& C\cdot k \cdot \#\Set{x-*-*-*-*-*-x} \\
 &=& C\cdot k \cdot \mbox{(number of circuits of length 6 in $G$)}.
\end{eqnarray*}
Therefore, to have concentration it is sufficient that
\beql{condition}
\Var{Y} = o(k \cdot (\Mean{Y})^2).
\eeq

Observe that \eqref{condition} is a sufficient -and not necessary-  condition.
Furthermore,\eqref{condition} is certainly not always true as there are graphs with many 6-circuits and no triangles at all (the circuits {\em may} repeat vertices or edges).

\section{Colorful Triangle Counting}
\label{sec:IPL}

In this Section we present a new sampling approach to approximating the number of triangles in a graph  $G(V,E)$, 
that significantly improves existing sampling approaches. Furthermore, it is easily implemented in parallel. 
The key idea of our algorithm is to correlate the sampling of edges such that if two edges of a triangle are sampled, 
the third edge is always sampled. This decreases 
the degree of the multivariate polynomial that expresses the number of sampled triangles.
This Section is organized as follows: in Section~\ref{subsec:colorfulalgo} we present our randomized algorithm.
In Section~\ref{subsec:colorfulanalysis} we present our main theoretical results, we analyze our algorithm and we
discuss some of its important properties. In Section~\ref{subsec:colorfulmapreduce} we present an implementation
of our algorithm in the popular \textsc{MapReduce} framework.

\subsection{Algorithm} 
\label{subsec:colorfulalgo} 

\begin{algorithm}[t]
\caption{\label{alg:colorfulcounting}Colorful Triangle Sampling} 
\begin{algorithmic}
\REQUIRE Unweighted graph $G([n],E)$
\REQUIRE Number of colors  $N=1/p$
\STATE Let $f: V \rightarrow [N]$ have uniformly random values
\STATE $E' \leftarrow \{ \{u,v\}\in E \; | \; f(u)=f(v) \}$
\STATE $T \leftarrow $ number of triangles in the graph $(V,E')$
\STATE {\bf return} $T/p^2$ 
\end{algorithmic}
\end{algorithm}

Our algorithm, summarized as Algorithm~\ref{alg:colorfulcounting}, 
samples each edge with probability~$p$, where $N=1/p$ is integer, as follows.
Let  $f: [n] \rightarrow [N]$ be a random coloring of the vertices of $G([n],E)$, 
such that for all $v \in [n]$ and $i \in [N]$, $\Prob{ f(v) = i } = p$.
We call an edge {\em monochromatic} if both its endpoints have the same color. 
Our algorithm samples exactly the set $E'$ of monochromatic edges, counts the number 
$T$ of triangles in $([n],E')$ (using any exact or approximate triangle counting algorithm), 
and multiplies this count by $p^{-2}$.

Work presented in Sections~\ref{sec:trianglesparsifiers},~\ref{sec:degreepartitioning} has used a related sampling idea, 
the difference being that edges were sampled {\em independently} with probability $p$.
Some intuition why this sampling procedure is less efficient than what we propose can 
be obtained by considering the case where a graph 
has $t$ edge-disjoint triangles. With independent edge sampling there will be no triangles 
left (with probability $1-o(1)$) if $p^3 t=o(1)$. 
Using our colorful sampling idea there will be $\omega(1)$ triangles in the sample with 
probability $1-o(1)$ as long as $p^2 t = \omega(1)$. 
This means that we can choose a smaller sample, and still get accurate estimates from it.

\subsection{Analysis} 
\label{subsec:colorfulanalysis}

We wish to pick $p$ as small as possible but at the same time have a strong concentration of the estimate around its expected value. How
small can $p$ be? In Section~\ref{subsec:colorfulsecondmoment} we present a second moment argument which gives a 
sufficient condition for picking $p$. 
Our main theoretical result, stated as Theorem~\ref{thrm:colorfulsecondmoment} in Section~\ref{subsec:colorfulconc}, provides a sufficient condition to this question. 
In Section~\ref{subsec:colorfulrunning} we analyze the complexity of our method. Finally, in Section~\ref{subsec:colorfuldiscussion} we discuss several aspects
of our work.

\subsubsection{Second Moment Method} 
\label{subsec:colorfulsecondmoment} 

Using the second moment method we are able to obtain the following strong theoretical guarantee: 

\begin{theorem} 
Let $n$, $t$, $\Delta$, $T$ denote the number of vertices in $G$, the number of triangles in $G$,
the maximum number of triangles an edge of $G$ is contained in and the number of monochromatic triangles
in the randomly colored graph respectively. Also let $N=\frac{1}{p}$ the number of colors used.
If $p \geq \max{ ( \frac{\Delta \log{n} }{t}, \sqrt{ \frac{\log{n}}{t}} )}$, then $T \sim \Mean{T}$ 
with probability $1-\frac{1}{\log{n}}$. 
\label{thrm:colorfulsecondmoment} 
\end{theorem}

\begin{proof} 

By Chebyshev's inequality~\ref{lem:chebyshevinequality}, if $\Var{T}=o(\Mean{T}^2)$ then $T \sim \Mean{T}$ with probability $1-o(1)$ \cite{alon}.
Let $X_i$ be a random variable for the $i$-th triangle, $i=1,\ldots,t$, such that $X_i=1$
if the $i$-th triangle is monochromatic. 
The number of monochromatic triangles $T$ is equal to the sum of these indicator variables, i.e., $T = \sum_{i=1}^t X_i$. 
By the linearity of expectation and by the fact that $\Prob{X_i=1} = p^2$ we obtain that $\Mean{T}=p^2t$. 
It is easy to check that the only case where two indicator variables are dependent
is when they share an edge. In this case the covariance is non-zero and for any $p>0$,
$ \Cov{X_i,X_j} = p^3-p^4 < p^3$. We write $i \sim j$ if and only if $X_i,X_j$ are dependent. 

We obtain the following upper bound on the variance of $T$, where $\delta_e$ 
is the number of triangles edge $e$ is contained in and $\Delta =\max_{e\in E(G)} \delta_e$: 

\begin{align*}
\Var{T} \leq \Mean{T} + \sum_{i \sim j} \Cov{X_i \wedge X_j} < p^2t + p^3 \sum_e \delta_e^2 \leq p^2t + 3p^3 t \Delta 
\end{align*}

\noindent We pick $p$ large enough to obtain $\Var{T}=o(\Mean{T}^2)$. It suffices:

\beql{eqsec}
p^4t^2 \gg p^2t + 3 p^3 t \Delta  \Rightarrow  p^2 t  \gg 1 + 3 p\Delta 
\eeq

\noindent We consider two cases, determined by which of the two terms of the right hand side  is larger:

\noindent
\underline{$\bullet$ {\sc Case 1} ($p\Delta < 1/3$):}\\
Since the right hand side of Inequality~\eqref{eqsec} is constant, 
it suffices that $p^2t = \omega(n)$ where $\omega(n)$ is some slowly growing function.
We pick $\omega(n)=\log{n}$ and hence $p \geq \sqrt{\frac{\log{n}}{t}}$. 

\noindent \\
\underline{$\bullet$ {\sc Case 2} ($p\Delta \ge 1/3$):}\\
In this case the right hand side of Inequality~\eqref{eqsec} is $\Theta(p\Delta)$
and therefore it suffices to pick $\frac{pt}{\Delta} = \log{n}$. 

\noindent Combining the above two cases we get that if 

$$p \geq \max{ ( \frac{\Delta\log{n} }{t}, \sqrt{ \frac{\log{n}}{t}})}$$

inequality~\eqref{eqsec} is satisfied and hence by Chebyshev's inequality $T \sim \Mean{T}$ with probability $1-\frac{1}{\log{n}}$. 

\end{proof}

\spara{Extremal Cases and Tightness of Theorem~\ref{thrm:colorfulsecondmoment}}

Given the assumptions of Theorem~\ref{thrm:colorfulsecondmoment}, is the condition on $p$ tight?  
The answer is affirmative as shown in Figure~\ref{fig:colorfultight}. Specifically, in Figure~\ref{fig:colorfultight}(a)
G consists of $t/\Delta$ ``books'' of triangles, each of size $\Delta$. This shows
that $p$ has to be at least $\omega(n)\frac{\Delta}{t}$ to hope for concentration,
where $\omega(n)$ is some growing function of $n$. Similarly, 
when $G$ consists of $t$ disjoint triangles as shown in Figure~\ref{fig:colorfultight}(b)
$p$ has to be at least $\omega(n)t^{-1/2}$.
Therefore, unless we know more about $G$, we cannot hope for milder conditions on $p$,
i.e., Theorem~\ref{thrm:colorfulsecondmoment} provides an optimal condition on $p$.

\begin{figure}
  \centering
  \begin{tabular}{cc}
   \includegraphics[width=0.35\textwidth]{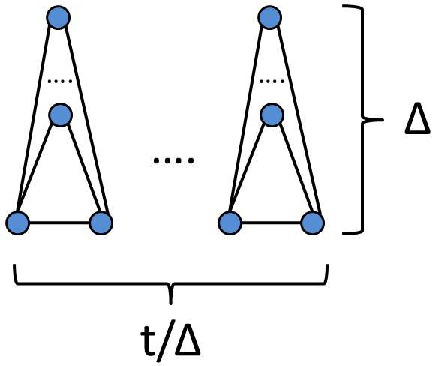}  &   \includegraphics[width=0.35\textwidth]{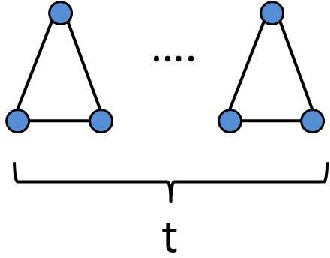} \\ 
    (a) & (b) 
  \end{tabular}
  \caption{\label{fig:colorfultight}Conditions of Theorem~\ref{thrm:colorfulsecondmoment} are tight. In order to hope for concentration $p$ has to be greater than (a)  $\frac{\Delta}{t}$ 
    and (b) $t^{-1/2}$.}
\end{figure}

\subsubsection{Concentration via the Hajnal-Szemer\'{e}di Theorem} 
\label{subsec:colorfulconc}

Here, we present a different approach to obtaining concentration, based on partitioning 
the set of triangles/indicator variables in sets containing many independent random indicator variables and then taking 
a union bound. Our theoretical result is the following theorem:

\begin{theorem} 

Let $t_{\max}$ be the  maximum number of triangles a vertex $v$ is contained in. Also, let $n,t,p,T$ be defined as above and $\epsilon$ a small positive constant.
If $p^2 \geq \frac{16 t_{\max}  \log{n}}{\epsilon^2 t} $, then $\Prob{ |T - \Mean{T}| > \epsilon \Mean{T} } \leq n^{-1}$.

\label{thrm:colorfulconcentration} 
\end{theorem}

\begin{proof} 

Let $X_i$ be defined as above, $i=1,\ldots,t$. Construct an auxiliary graph $H$ as follows: add a vertex in $H$ for every triangle in $G$ and connect two vertices representing triangles $t_1$ and $t_2$ if and only if they have a common vertex. The maximum degree of $H$ is  $3t_{\max}=O(\delta^2)$, where $\delta=O(n)$ is the maximum degree in the graph. 
Invoke the Hajnal-Szemer\'{e}di Theorem on $H$: we can partition the vertices of $H$ (triangles of $G$) into sets $S_1, \dots, S_q$ such that $|S_i| > \Omega(\tfrac{t}{t_{\max}})$ and $q=\Theta(t_{\max})$. Let $k=\tfrac{t}{t_{\max}}$. Note that the set of indicator variables $X_i$
corresponding to any set $S_j$ is independent. Applying the Chernoff bound for each set $S_i, i=1,\ldots,q$  we obtain 

$$Pr \left[ |\frac{1}{k} \sum_{i=1}^k X_i - p^2| > \epsilon p^2 \right] \leq 2e^{-\epsilon^2p^2k/2}$$.

If $p^2 k \epsilon^2 \geq 4d'\log{n}$, then $2e^{-\epsilon^2p^2k/2}$ is upper bounded by $n^{-d'}$,
where $d'>0$ is a constant. Since $q=O(n^3)$ by taking a union bound over all sets $S_i$ we see that the triangle count is approximated
within a factor of $\epsilon$ with probability at least $1 - n^{3-d'}$ Setting $d'=4$ completes the proof. 
\end{proof}

It's worth noting that for any constant $K>0$ the above proof gives that if $p^2 \geq \frac{4(K+3) t_{\max}  \log{n}}{\epsilon^2 t} $
then $\Prob{ |T - \Mean{T}| > \epsilon \Mean{T} } \leq n^{-K}$.

\subsubsection{Complexity}
\label{subsec:colorfulrunning}

The running time of our procedure of course depends on the subroutine we use on 
the second step, i.e., to count triangles in the edge set $E'$.
Let $\text{d}(i)$ denote the degree of vertex $i$.
Assuming we use node iterator, i.e., the exact method that examines each vertex independently and counts the number of edges among its neighbors,
our algorithm runs in $O(n+m+p^2 \sum_{i \in [n]} \text{d}^2(i))$ expected time
\footnote{We assume that uniform sampling of a color takes constant time. If not, then we obtain the term $O(n\log{(\frac{1}{p})}$
for the vertex coloring procedure.} by efficiently storing the graph and retrieving the neighbors
of $v$ colored with the same color as $v$ in $O(1+p\,\text{d}(v))$ expected time. 
Note that this implies that the speedup with respect to the counting task is $1/p^2$.

\subsubsection{Discussion}
\label{subsec:colorfuldiscussion}

Despite the fact that the second moment argument gave us strong conditions on $p$, 
the use of the Hajnal-Szemer\'{e}di theorem, see Theorem~\ref{lem:hajnal} and \cite{HajnalSzemeredi}, 
has the potential of improving the $\Delta$ factor. 
The condition we provide on $p$ is {\em sufficient} to obtain concentration. 
Note --see Figure~\ref{fig:colorfulfig1}-- that it was necessary to partition the triangles into 
vertex disjoint  rather than edge disjoint triangles since we need mutually 
independent variables per chromatic class in order to apply the Chernoff bound. 
If we were able to remove the dependencies in the chromatic classes defined by 
edge disjoint triangles, then the overall result could  probably be improved. 
It's worth noting that for $p=1$ we obtain that $t \geq n \omega(n)$, where $\omega(n)$ is 
any slowly growing function of $n$. 
This is --to the best of our knowledge-- the mildest condition on the 
triangle density needed for a randomized algorithm to obtain concentration.
Finally, notice that when $t \leq \frac{ \Delta^2 \log{n}}{t_{\max}}$ and $t_{\max} \geq 1$ 
Theorem \ref{thrm:colorfulconcentration} yields a better
bound than Theorem~\ref{thrm:colorfulsecondmoment}. The same holds when $t > \Delta^2 \log{n}$ 
and $t_{\max} \leq 1$. The latter scenario is far more restrictive 
and both Theorem~\ref{thrm:colorfulsecondmoment} and Theorem~\ref{thrm:colorfulconcentration} give for 
instance the same bound $p \geq \sqrt{\frac{\log{n}}{t}}$ for the graph of Figure~\ref{fig:colorfultight}(b).

Furthermore, the powerful theorem of Kim and Vu~\ref{thrm:kim-vu} that was used in Section~\ref{sec:trianglesparsifiers}
 is not immediately applicable here: let  $Y_e$ be an indicator variable 
for each edge $e$ such that $Y_e=1$ {\it if and only if} $e$ is monochromatic, i.e., 
both its endpoints receive the same color. Note that the number of triangles is a Boolean polynomial 
$T = \frac{1}{3} \sum_{\Delta(e,f,g)} \big( Y_e Y_f+ Y_f Y_g+Y_e Y_g \big)$ but the Boolean variables are not independent
as the Kim-Vu \cite{kim-vu} theorem requires. It is worth noting that the degree of the polynomial
is two. Essentially, this is the reason for which our method obtains better results than work 
in Section~\ref{sec:trianglesparsifiers}
where the degree of the multivariate polynomial is three.
Finally, it is worth noting that using a simple doubling procedure as the one outlined
in Section~\ref{sec:trianglesparsifiers},  we can pick $p$ effectively in practice despite the fact 
that it depends on the quantity $t$ which we want to estimate by introducing an extra logarithm in the running time.

\begin{figure*} 
\centering
\includegraphics[width=0.4\textwidth]{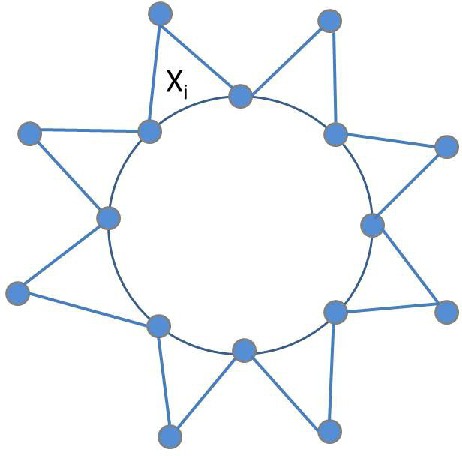}
\caption{Consider the indicator variable $X_i$ corresponding to the $i$-th triangle. 
Note that $\Prob{X_i|\text{rest are monochromatic}}=p \neq \Prob{X_i}=p^2$.
The indicator variables are {\em pairwise} but not mutually independent. }
\label{fig:colorfulfig1} 
\end{figure*}

Finally, from an experimentation point of view, it is interesting to see how well the upper bound $3\Delta t$ matches the sum $\sum_{e \in E(G)} \delta_e^2$,
where $\delta_e$  is the number of triangles edge $e$ is contained in, and the typical values for $\Delta$ and $t_{\max}$ in 
real-world graphs. The following table shows these numbers for five graphs
\footnote{AS:Autonomous Systems, Oregon: Oregon route views, 
Enron: Email communication network, ca-HepPh and AstroPh:Collaboration networks. Self-edges were removed.} taken from the SNAP library \cite{snap}.
We see that $\Delta$ and $t_{\max}$ are significantly less than their upperbounds and that typically $3\Delta t$ is significantly larger than 
$\sum_{e \in E(G)} \delta_e^2$ except for the collaboration network of Arxiv Astro Physics. The results are shown in Table~\ref{tab:colorfuldatasets}.

\begin{table}[ht]
\begin{center}
\begin{tabular}{|l|r|r|r|r|r|r|r|} \hline
Name     & Vertices($n$)     & Edges($m$)   &  Triangle Count($t$) & $\Delta$     & $t_{\max}$& $\sum_{e \in E(G)} \delta_e^2$ & 3$\Delta t$        \\ \hline 
AS       & 7,716        & 12,572     &  6,584             &   344        & 2,047     &        595,632                 & 6,794,688           \\ \hline
Oregon   & 11,492       & 23,409     &  19,894            &   537        & 3,638     &       2,347,560                & 32,049,234         \\ \hline
Enron    & 36,692       & 183,831    &  727,044           &   420        & 17,744    &      75,237,684                & 916,075,440          \\ \hline
ca-HepPh & 12,008       & 118,489    & 3,358,499          &   450        & 39,633    &    1.8839 $\times 10^9$        & 4.534$\times 10^9$  \\ \hline
AstroPh  & 18,772       & 198,050    & 1,351,441          &   350        & 11,269    &      148,765,753               & 1.419$\times 10^9$  \\ \hline
\end{tabular}
\end{center}
\caption{Values for the variables involved in our formulae for five real-world graphs. Typically, 
$\Delta$ and $t_{\max}$ are significantly less than the obvious upper bounds $n-2$
and ${n-1 \choose 2}$ respectively. Furthermore,  $3\Delta t$ 
is significantly larger than $\sum_{e \in E(G)} \delta_e^2$. }
\label{tab:colorfuldatasets}
\end{table}

\subsection{A \textsc{MapReduce} Implementation} 
\label{subsec:colorfulmapreduce} 

\textsc{MapReduce} \cite{dean} has become the {\it de facto} standard in academia and industry
for analyzing large scale networks. For a brief overview of \mapreduce see Section~\ref{sec:mapreducebasics}.
Recent work by Suri and Vassilvitskii \cite{suri} proposes two algorithms for counting triangles. 
The first is an efficient \textsc{MapReduce} implementation of the node iterator algorithm, see also \cite{schank2005finding}
and the second is based on partitioning the graph into overlapping subsets so that each triangle is present in at
least one of the subsets. 

Our method is amenable to being implemented in \textsc{MapReduce}
and the skeleton of such an implementation is shown in Algorithm 2\footnote{It's worth pointing out for completeness reasons 
that in practice one would not scale the triangles
after the first reduce. It would emit the count of monochromatic triangles
which would be summed up in a second round and scaled by $1/p^2$.}.
We implicitly assume that in a first round vertices have received 
a color uniformly at random from the $N$ available colors and that we have the coloring
information for the endpoints of each edge. 
Each mapper receives an edge together with the colors of its edgepoints. 
If the edge is monochromatic, then it's emitted with the color as the key and the edge as the value. 
Edges with the same color are shipped to the same reducer where locally a triangle counting
algorithm is applied. The total count is scaled appropriately. Trivially, the following lemma  
holds by the linearity of expectation and the fact that the endpoints of any edge 
receive a {\em given} color $c$ with probability $p^2$.

\begin{algorithm}[t]
\caption{\textsc{MapReduce} Colorful Triangle Counting $G(V,E),p=1/N$ } 
\begin{algorithmic}
\STATE {\bf Map:} Input $\langle e=(u,f(u),v,f(v));1 \rangle$
\COMMENT{Let $f$ be a uniformly at random coloring of the vertices with $N$ colors}
\STATE  {\it if} $f(u)=f(v)$ {\it then} emit $\langle f(u); (u,v)\rangle$
\STATE {\bf Reduce:} Input $\langle c; E_c=\{ (u,v) \} \subseteq E \rangle$ 
\COMMENT{ Every edge $(u,v) \in E_c$ has color $c$, i.e., $f(u)=f(v)$}
\STATE Scale each triangle by $\frac{1}{p^2}$. 
\end{algorithmic}
\end{algorithm}

\begin{lemma} 
The expected size to any reduce instance is $O(p^2m)$ and the expected total space used at the end of the map
phase is $O(pm)$. 
\label{subsec:colorfulmrc} 
\end{lemma}

\clearpage
\chapter{Dense Subgraphs}
\label{densestchapter}
\lhead{\emph{Densest Subgraphs}} 
\section{Introduction} 
\label{sec:densestintro} 

Given a graph $G=(V,E)$ and a subset of vertices $S\subseteq V$,  
let $G[S]=(S,E[S])$ be the subgraph induced by $S$, and
let $e[S]$ be the size of $E[S]$. The \emph{edge density} of 
the set $S$ is defined as $\delta(S) = e[S] / {|S| \choose 2}$.
Finding a dense subgraph of $G$ would in principle require to 
find a set of vertices $S \subseteq V$ that maximizes $\delta(S)$.
However, the direct maximization of $\delta$ is not a 
meaningful problem, as even a single edge achieves maximum density.
Therefore,  effort has been devoted to define 
alternative density functions whose maximization allows for extracting
subgraphs having large $\delta$ and, at the same 
time, non-trivial size.
Different choices of the density function lead to 
different variants of the dense-subgraph problem.
Some variants can be solved in polynomial time, 
while others are \NPhard, or even inapproximable.

\begin{table}[t]
\centering \small
\begin{tabular}{r|llll|llll|}
\multicolumn{1}{c}{} &  \multicolumn{4}{c}{\textsf{densest subgraph}} &  \multicolumn{4}{c}{\textsf{optimal \QC}} \\
\cline{2-9}
&
\multicolumn{1}{c}{$\frac{|S|}{|V|}$} &
\multicolumn{1}{c}{$\delta$} &
\multicolumn{1}{c}{$D$} &
\multicolumn{1}{c|}{$\tau$} &
\multicolumn{1}{c}{$\frac{|S|}{|V|}$} &
\multicolumn{1}{c}{$\delta$} &
\multicolumn{1}{c}{$D$} &
\multicolumn{1}{c|}{$\tau$} \\ \cline{2-9}
\textsf{Dolphins}   &  0.32 & 0.33 & 3  & 0.04 &  0.12 & 0.68 & 2 &  0.32 \\
\textsf{Football}   &  1 & 0.09 & 4  &  0.03 & 0.10  & 0.73 & 2 &  0.34 \\
\textsf{Jazz}             &  0.50 & 0.34 & 3 &  0.08 &  0.15 & 1 &  1 & 1 \\
\textsf{Celeg.\ N.} & 0.46   & 0.13 & 3 &  0.05 &  0.07 & 0.61 & 2 &   0.26 \\
\cline{2-9}
\end{tabular}
\caption{\label{tab:densestintro} Difference between densest subgraph and optimal \QC\  on some popular graphs. $\delta= e[S] / {|S| \choose 2}$ is the edge density of the extracted subgraph, $D$ is the diameter, and $\tau = t[S] / {|S| \choose 3}$ is the triangle density.}
\end{table}

In this Chapter we introduce a new framework for finding dense subgraphs. 
We focus on a special case of the framework which we refer to as \OQCs.
Table~\ref{tab:densestintro} offers a preview of the results that will follow.
Specifically, it compares our \OQCs\ with \DSs\ 
on some popular graphs.\footnote{\scriptsize \textsf{Densest subgraphs} 
are extracted here with the exact Goldberg's
algorithm~\cite{Goldberg84}. As far as  \OQCs, 
we optimize $f_\alpha$ with $\alpha=\frac 13$ and
use our local-search method.} The results in the table 
clearly show that {\OQCs} have much larger edge density 
than \DSs, smaller diameters and larger triangle densities. 
Moreover, \DSs\ are usually quite large-sized: in the 
graphs we report in Table~\ref{tab:densestintro}, 
the \DSs\ contain always more than the 30\% of the vertices in the input graph.
For instance, in the \textsf{Football} graph, the \DS\ 
corresponds to the whole graph, with edge density $< 0.1$ and diameter~4,
while the extracted \OQC\ is a 12-vertex subgraph 
with edge density $0.73$ and diameter~2.
The \textsf{Jazz} graph contains a perfect clique 
of 30 vertices: our method finds this clique 
achieving perfect edge density, diameter, and triangle density scores.
By contrast, the \DS\ contains 100 vertices, 
and has edge density $0.34$ and triangle density~$0.08$.

\section{A General Framework}
\label{sec:densestproblem}

\enlargethispage*{\baselineskip}
Let $G=(V,E)$ be a graph, with $|V| = n$ and $|E| = m$.  For a set of vertices $S \subseteq V$, let $e[S]$ be the
number of edges in the subgraph induced by~$S$. We define the following function.
\begin{definition} [Edge-surplus]
Let $S \subseteq V$ be a subset of the vertices of a graph $G=(V,E)$, and let $\alpha > 0$ be a constant.
We define  {\em edge-surplus}
as:
\[
f_{\alpha}(S) =  g( e[S] ) - \alpha h(|S|),
\]

where functions $g,h$ are both strictly increasing.
We also define $f_{\alpha}(\emptyset)=0$. 
\end{definition}

We note that the first term $g(e[S])$ encourages subgraphs abundant in edges whereas
the second term $-\alpha h(|S|)$ penalizes large subgraphs.
Our framework for finding dense subgraphs is based on the following optimization problem.

\begin{problem}[\OES]
Given a graph $G=(V,E)$, a positive real $\alpha$ and a pair of functions $g,h$,
find a subset of vertices $S^* \subseteq V$  such that
$f_{\alpha}(S^*) \ge f_{\alpha}(S)$, for all sets ${S \subseteq V}$.
We refer to the set $S^*$ as the {\OES} of the graph $G$.
\end{problem}

The edge-surplus definition subsumes numerous popular existing density measures
by choosing appropriately $g,h,\alpha$. Three important cases follow.

\squishlist

 \item By setting $g(x)=h(x)=\log{x}, \alpha=1$, the {\OES} problem
  becomes equivalent to maximizing $ \log{e[S]}-\log{|S|} = \log{ \frac{e[S]}{|S|} }$.
  This is equivalent to the popular {\DS\ }.
  \item By setting $g(x)=\log{x}, h(x)=\log{ \Big( \tfrac{x(x-1)}{2} \Big) }, \alpha =1$ the {\OES} problem becomes equivalent
   to maximizing  $\frac{e[S]}{ {|S| \choose 2}}$.
 \item By setting $g(x)=x, h(x)=\tfrac{x(x-1)}{2}$  and restricting $\alpha \in (0,1)$
  we obtain the following  {\OES} problem: $\max_{\emptyset \neq S \subseteq V} e[S]-\alpha {|S| \choose 2}$.
  We call this problem {\OQCP}. We notice that it turns the quasi-clique condition into
  an objective. To the best of our knowledge, this optimization problem does not appear in the existing literature.
\squishend

 
The \DS\ problem on the one hand is polynomially time solvable but results typically in very large subgraphs.
Also, maximizing the edge density $\delta(S)$ results in trivial subgraphs
such as edges or triangles which achieve the maximum possible density value 1.
We wish to better understand the {\OQCP} problem. We start by discussing properties of the objective. 
{\bf Understanding the objective:} 
Consider the case  $\alpha=0$. Clearly, the optimal solution is the whole graph. 
When $0<\alpha<1$ the problem in general is \NPhard. We discuss the case where $\alpha>1/2$. 
It is straightforward to check that by setting $\alpha=1-\frac{1}{\Omega(n^2)}$, e.g., $\alpha=1-n^{-3}$,
one solves the maximum clique problem. 
Furthermore, 
assuming that finding a hidden clique of order $O(n^{1/2-\delta})$  where $\delta>0$ 
in a random binomial graph $G\sim G(n,1/2)$ is hard, 
then one can see that our problem is hard. 
To see why, notice that for any set of $S$ vertices 
the expected score is $\big(\frac{1}{2}-\alpha\big) {n^{1/2-\delta} \choose 2}$ 
and for the hidden clique $(1-\alpha){n^{1/2-\delta} \choose 2}$. 
Therefore, if we could optimally solve the {\OQCP} problem, 
then by setting $\alpha>1/2$,  we could solve in expectation the hidden clique problem.

\begin{theorem}
The \OQCP\ is {\NPhard}.
\end{theorem}

Proving that the problem is \NPhard for any $\alpha \in (0,1)$ remains open. 
When $\alpha=1$ all cliques receive score equal to 0, independent of their size. 
When $\alpha>1$ then the problem stops being interesting as 
the optimal solution will be any single edge. 
An important property of the edge-surplus abstraction is that it
allows us to model scenaria in numerous practical situations
where we wish to find a dense subgraph with bounds on its size. For instance,
by relaxing the monotonicity property of function $h()$ 
the $k$-\DS\ problem can be modeled as an  {\OES} problem
by setting $g(x)=x$ and
    \begin{equation*}
          h(x) = \begin{cases}
               0               & x = k\\
               +\infty & \text{otherwise.}
           \end{cases}
    \end{equation*}
{\it By choosing $h(x)$ to penalize severely undesired size, a good algorithm
will avoid outputing a subgraph of undesired size. }

Finally, we outline that we cannot make any general statement on the
{\OES} problem, since certain cases are polynomially time solvable
whereas others \NPhard.
The following theorem provides a family of {\OES} problems that
are efficiently solvable

\begin{theorem}
Let $g(x)=x$ and $h(x)$ be a concave function. Then the {\OES} problem is in \Polytime.
\end{theorem}

\begin{proof}

The {\OES} problem becomes $\max_{\emptyset \neq S \subseteq V} e[S]-\alpha h(|S|)$
where $h(x)$ is a concave function. The claim follows directly from combining the following
facts.

\noindent
\underline{Fact 1} The function defined by the map $S \mapsto e[S]$ is a supermodular function. \\
\underline{Fact 2} The function $h(|S|)$ is submodular given that $h$ is concave.
Since $\alpha > 0$, the function $-\alpha h(|S|)$ is supermodular. \\
\underline{Fact 3} Combining the above facts with the fact that the sum of two supermodular
functions is supermodular, we obtain $f_\alpha(S)$ is a supermodular function.  \\
\underline{Fact 4} Maximizing supermodular functions is strongly polynomially time solvable \cite{schrijver}.
\end{proof}

However, we outline that the scenario where $g(x)=x$ and $h(x)$ is concave 
cannot be useful in real applications as the  output subgraph will be large. 

\section{Optimal Quasi-cliques}
\label{sec:densestoqc}

\spara{Scalable Algorithms.} The first efficient algorithm we propose is an adaptation of the greedy algorithm by Asashiro et al.~\cite{AHI02}, 
which has been shown to provide a $\frac 1 2$-approximation for the \DS\ problem~\cite{Char00}.
The outline of our algorithm, called \Galgo, is shown as Algorithm \ref{alg:additive}.
The algorithm iteratively removes the vertex with the smallest degree.
The output is the subgraph produced over all iterations that maximizes the objective function $f_\alpha$.
The algorithm can be implemented in $\bigO(n + m)$ time: the trick consists in keeping a list of vertices 
for each possible degree and updating the degree of any vertex $v$ during the various iterations of the 
algorithm simply by moving $v$  to the appropriate degree list.

\begin{algorithm}[t]
\caption{\label{alg:additive}\Galgo} 
 \begin{algorithmic} 
\REQUIRE Graph $G(V,E)$
\ENSURE Subset of vertices $\bar{S} \subseteq V$
\STATE $S_n \leftarrow V$ 
\FOR{$i \leftarrow n$ downto $1$} 
\STATE Let $v$ be the vertex with the smallest degree in $G[S_i]$
\STATE $S_{i-1} \leftarrow S_i  \setminus \{v\}$
\ENDFOR
\STATE $\bar{S} \leftarrow \arg\max_{i= 1, \ldots, n} f_\alpha(S_i)$
\end{algorithmic}
\end{algorithm}

The \Galgo\ algorithm provides an additive approximation guarantee for the {\OQCP}, as shown next.

\begin{theorem}\label{theorem:additive}
Let $\bar{S}$ be the set of vertices outputted by the \Galgo\ algorithm and let $S^*$ be the optimal vertex set.
Consider also the specific iteration of the algorithm where a vertex within $S^*$ is removed for the first time  and let $S_I$ denote the vertex set currently kept in that iteration.
It holds that:
$$
f_\alpha(\bar{S}) \geq f_\alpha(S^*) - \frac{\alpha}{2} |S_I|(|S_I|-|S^*|).
$$
\end{theorem}

\begin{proof}
Given a subset of vertices $S \subseteq V$ and a vertex $u \in S$, let $d_S(u)$ denote the degree of $u$ in $G[S]$.

We start the analysis by considering the first vertex belonging to $S^*$ removed by the algorithm from the current vertex set.
Let $v$ denote such a vertex, and let also $S_I$ denote the set of vertices still present just before the removal of $v$. 
By the optimality of $S^*$, we obtain:
\begin{eqnarray*}
\lefteqn{f_\alpha(S^*) \geq f_\alpha(S^* \setminus \{u\}), \ \forall u \in S^*}\\
\!\!& \Leftrightarrow &\!\!\! e[S^*] \!-\! \alpha {|S^*| \choose 2} \!\geq\! ( e[S]\! -\! d_{S^*}(u) ) \!- \! \alpha {|S^*| \! - \! 1 \choose 2}\!, \forall u \!\in\! S^* \\
\!\!& \Leftrightarrow &\!\!\! d_{S^*}(u) \geq \alpha(|S^*|-1), \ \forall u \in S^*.
\end{eqnarray*}

As the algorithm greedily removes vertices with the smallest degree in each iteration, it is easy to see that $d_V(u) \geq d_{S_I}(u) \geq  d_{S^*}(u) \geq \alpha(|S^*|-1)$, $\forall u$. 
Therefore, noticing also that $S^* \subseteq S_I$, it holds that:
\begin{eqnarray*}
\lefteqn{f_{\alpha}(S_I)= e[S_I] - \alpha {|S_I| \choose 2}        + \alpha{n \choose 2}} \\
  &  = & \frac{1}{2} \left ( \sum_{u \in S^*} d_{S^*}(u)+  \sum_{u \in S^*} \left (  d_{S_I}(u) - d_{S^*}(u) \right ) + \right.\\ 
  & & + \left. \sum_{u \in S_I\setminus S^*} d_{S_I}(u) \right ) - \alpha{|S_I| \choose 2} + \alpha{n \choose 2}\\ 
& \geq &  \frac{1}{2} \left ( \sum_{u \in S^*}\! d_{S^*}(u) \ + \!\!\! \sum_{u \in S_I\setminus S^*} \!\!\! d_{S_I}(u) \right ) - \alpha{|S_I| \choose 2}  + \alpha{n \choose 2}\\
& =  &  e[S^*] +  \frac{1}{2} \! \sum_{u \in S_I\setminus S^*} \!\!\!d_{S_I}(u) - \alpha{|S_I| \choose 2}   + \alpha{n \choose 2} \\   
& \geq & e[S^*] + \frac{1}{2} (|S_I|-|S^*|)\alpha(|S^*|-1) - \alpha{|S_I| \choose 2}  + \alpha{n \choose 2}  \\ 
 & = & f_\alpha(S^*) - \frac{\alpha}{2} |S_I|(|S_I|-|S^*|).
\end{eqnarray*}

\noindent As the final output of the algorithm is the best over all iterations, we finally obtain:
$$ 
f_{\alpha}(\bar{S}) \  \geq \ f_{\alpha}(S_I) \ \geq \  f_\alpha(S^*) - \frac{\alpha}{2} |S_I|(|S_I|-|S^*|).
$$ 
%
%
\end{proof}

The above result can be interpreted as follows.
Assuming that $|S_I|$ is $\bigO(|\bar{S}|)$, the additive approximation factor proved in Theorem \ref{theorem:additive} becomes $f_{\alpha}(\bar{S}) \geq   f_\alpha(S^*) - \frac{\alpha}{2} |\bar{S}|(|\bar{S}|-|S^*|)$.
Thus, the error achieved by the \Galgo\ algorithm is guaranteed to be bounded by an additive factor proportional to the size of the {\OQC} outputted.
As {\OQCs} are typically small graphs, this results in an approximation guarantee that is very tight in practice. 

Finally, we present a  local search heuristic for solving the {\OQCP}.
The algorithm, called \LSalgo, performs local operations and outputs a vertex set 
$S$ that is guaranteed to be locally optimal, i.e., if any single vertex is added to or removed 
from $S$, then the objective function decreases. 

The outline of \LSalgo\ is shown as Algorithm \ref{algorithm:localsearch}. 
The algorithm initially selects a random vertex and then it keeps adding vertices to the current set $S$ while the objective improves. 
When no vertex can be added, the algorithm tries to find a vertex in $S$ whose removal may improve the objective.
As soon as such a vertex is encountered, it is removed from $S$ and the algorithm re-starts from the adding phase.
The process continues until a local optimum is reached or the number of iterations exceeds \TMAX.
The time complexity of \LSalgo\ is $\bigO(\TMAX\,m)$.

\begin{algorithm}[t]
\caption{\label{algorithm:localsearch}\LSalgo}
\begin{algorithmic}
\REQUIRE Graph $G=(V,E)$;  maximum number of iterations $T_{MAX}$
\ENSURE Subset of vertices $\bar{S} \subseteq V$
\STATE $S \leftarrow \{v\}$, where $v$ is chosen uniformly at random
\STATE $b_1,b_2 \leftarrow$ TRUE, $t \leftarrow 1$.
\WHILE{$b_1$ and $t \leq T_{MAX}$}

\WHILE{$b_2$}
\STATE If there exists $u \in V \backslash S$ such that $f_{\alpha}(S \cup \{u\}) \geq f_{\alpha}(S)$ 
\STATE then let $S \leftarrow S \cup \{u\}$
\STATE otherwise set $b_2 \leftarrow$ FALSE
\ENDWHILE
\STATE If there exists $u \in S$ such that $f_{\alpha}(S \backslash \{u\}) \geq  f_{\alpha}(S)$ 
\STATE then let $S \leftarrow S \backslash \{u\}$ 
\STATE otherwise, set $b_1 \leftarrow$ FALSE
\STATE $t \leftarrow t+1$
\ENDWHILE
\STATE $\bar{S} \leftarrow \arg\max_{\hat{S} \in \{S, V \setminus S\}} f_\alpha(\hat{S})$
\end{algorithmic}
\end{algorithm}

In order to enhance the performance of the \LSalgo\ algorithm,
one may use the following heuristic \cite{Gleich-2012-neighborhoods}.
Let $v^*$ be the vertex that maximizes the ratio $\frac{t(v^*)}{d(v^*)}$, where
$t(v^*)$ is the number of triangles of $v^*$ and $d(v^*)$ its degree (we approximate the number of triangles in which 
each vertex participates with the technique described in~\cite{tsourakakis4}).
Given vertex $v^*$, we use as a seed the set $\{v^* \cup N(v^*)\}$, where $N(v^*)=\{u: (u,v^*) \in E \}$ is the neighborhood of $v^*$.

\spara{Parameter Selection.} Finally, a natural question that arises whenever a parameter exists, is how to choose an appropriate value.
No doubt, there exist different possible, principled ways which lead to different choices of $\alpha$.
For instance, setting $\alpha$ to be the graph edge density $\delta(G)$ results into a normalized version of our criterion.
However, since real-world networks are sparse, we do not encourage this for practical purposes.
Instead, we provide a simple criterion to pick $\alpha$.

Let us consider two disjoint sets of vertices $S_1,S_2$ in the graph $G$.
Assume that $G[S_1 \cup S_2]$ is disconnected, i.e., $G[S_1]$ and $G[S_2]$ form two separate connected components.
Also, without any loss of generality, assume that $f_\alpha(S_1) \leq f_\alpha(S_2)$.
As our goal is to favor small dense subgraphs, a natural condition to satisfy is $f_\alpha(S_1 \cup S_2) \leq f_\alpha(S_1) \leq f_\alpha(S_2)$, i.e., we require for our objective to prefer the set $S_1$ (or $S_2$) rather than the larger set $S_1 \cup S_2$.
Therefore, we obtain:
\[
\!e[S_1]+ e[S_2]- \alpha \!{|S_1|\!+\!|S_2| \choose 2} \!+  \alpha \!{n \choose 2} \!\leq\!
e[S_1] - \alpha \!{|S_1| \choose 2} + \alpha \!{n \choose 2}\!,
\]
which, considering that $e[S_2] \leq {|S_2| \choose 2}$, leads to:
$$
\alpha \geq  \frac{{|S_2| \choose 2}}{  {|S_1|+|S_2| \choose 2}-{|S_1| \choose 2}} = \frac{|S_2| -1}{2|S_1| + |S_2| -1}.
$$
Let us now assume for simplicity that $|S_1| = |S_2| = k$; then the above condition becomes:
$
\alpha \geq \frac{k-1}{3k -1}.
$
As $\frac{k-1}{3k -1} < \frac{1}{3}$, it suffices choosing $\alpha \geq \frac{1}{3}$ to have the condition satisfied.

On the other hand, it is easy to see that, when $\alpha$ is close to 1, $f_\alpha$ tends to be maximized even by subgraphs of trivial structure (e.g., single edges), which is clearly something that we want to avoid.
By combining the two arguments above, we conclude that a good choice for $\alpha$ is a value around $\frac{1}{3}$, which is the value we adopt in our experiments.

\spara{Multiplicative Approximation, a 0.796-approximation algorithm for a shifted objective.} 
We also design a multiplicative approximation algorithm for a shifted 
objective which works for any $\alpha>0$.
Notice that this shifting is not necessary since the optimal objective value is positive
in the interesting range of $0<\alpha<1$ as a single edge results in a positive 
score $1-\alpha$. 
Our algorithm is based on semidefinite programming and in particular on the techniques developed
by Goemans-Williamson~\cite{goemans}.
Our algorithm is a  $\beta$-approximation algorithm, where $\beta >0.796$, 
for a {\it shifted} objective. 
Specifically, we {\em shift} our objective by a constant $c=c(n)$ to make it non-negative.

\begin{observation}
Consider $f'_{\alpha}(S)= f_{\alpha}(S)+\alpha {n \choose 2}$. Then $f'_{\alpha}(S) \geq 0$ for any $S \subseteq V$
since ${n \choose 2}-{s \choose 2} \geq 0$ for any $S \subseteq V$ and $e[S] \geq 0$.
Furthermore we have $f'_{\alpha}(S_1) \geq f'_{\alpha}(S_2)$ if and only if $f_{\alpha}(S_1) \geq f_{\alpha}(S_2)$.
\end{observation}

All the guarantees we obtain in this Section refer to the shifted objective $f'_{\alpha}$. 
Therefore, from now on we will abuse slightly the notation and use $f_{\alpha}$ to denote $f'_{\alpha}$.
We formulate our maximization problem as an integer program.
We introduce a variable $x_i \in \{-1,+1\}$ for each vertex $i \in V=\{1,\ldots,n\}$
and an extra variable $x_0$ which expresses whether a vertex belongs to $S$ or not:
\[
\text{It is }
i \in S 
\text{ if and only if }
x_0x_1=1.
\]
Notice that the term $\frac{1+x_0x_i+x_0x_j+x_ix_j}{4}$ equals 1 if and only if
both $i,j$ belong in $S$, otherwise it equals 0. Furthemore, the term ${n \choose 2}$
enters the objective as $\frac{1}{2} \sum_{i \neq j} 1$. Therefore, we get the following
integer program:

\begin{equation}
\label{equation:IP1}
\framebox{
\begin{minipage}{0.85\linewidth}
\begin{align}
\nonumber	
\mbox{\bf max} \quad   &  \sum_{e=(i,j)}\frac{1+x_0x_i+x_0x_j+x_ix_j}{4} + \\
\nonumber
& \frac{\alpha}{2} \sum_{i \neq j} \Big( 1- \frac{1+x_0x_i+x_0x_j+x_ix_j}{4} \Big) \\
\nonumber
\mbox{{subject to }}\, & x_i \in \{-1,+1\}, \,\mbox{for all }\, i \in \{0,1,..,n\}. \\
\nonumber
\end{align}
\end{minipage}}
\end{equation}

\noindent We relax the integrality constraint and we allow the variables to be vectors in the unit
sphere in $\field{R}^{n+1}$. By using the variable transformation $y_{ij} = x_{i} x_j$, 
we obtain the following semidefinite programming relaxation:

\begin{equation}
\label{equation:IP2}
\framebox{
\begin{minipage}{0.85\linewidth}
\begin{align}
\nonumber	
\mbox{\bf max}  \qquad &
\alpha \sum_{e=(i,j)}\frac{1+y_{0i}+y_{0j}+y_{ij}}{4} + \\
\nonumber
& \frac{1}{2} \sum_{i \neq j} \Big( 1- \frac{1+y_{0i}+y_{0j}+y_{ij}}{4} \Big) \\
\nonumber
\mbox{{subject to }}\, & y_{ii}=1, \,\mbox{for all}\,  i \in \{0,1,..,n\} \\
\nonumber
\mbox{{and }}\, & Y \succeq 0,\, Y \,\mbox{symmetric.} \\
\nonumber
\end{align}
\end{minipage}}
\end{equation}

\noindent The above SDP can be solved within an additive error of $\delta$ of the optimum
in polynomial time 
by interior point algorithms or the ellipsoid method~\cite{alizadeh}.
In what follows, we refer to the optimal value of the integer program as {\IPOPT}
and of the semidefinite program as {\SDPOPT}. 
Our algorithm, \SDPalgo, is shown as Algorithm~\ref{algorithm:sdp}.

\begin{algorithm}[t]
\caption{\label{algorithm:sdp}\SDPalgo}
\begin{algorithmic}
\REQUIRE $G=(V,E)$
\STATE
{\bf 1. Relaxation} 
\STATE\quad 
Solve the semidefinite program~(\ref{equation:IP2}) 
\STATE\quad Compute a Cholesky decomposition of the resulting $Y$
\STATE\quad Let $v_0,v_1,\ldots,v_n$ be the resulting vectors
\STATE
{\bf 2. Randomized Rounding} 
\STATE\quad Randomly choose a unit length vector $r \in \field{R}^{n+1}$ 
\STATE\quad Set $S =\{ i \in [n]: \sgn(v_ir)=\sgn(v_0r) \}$
\STATE 
{\bf 3. Boosting the success probability} 
\STATE\quad Repeat steps 1--2 for $t=1,..,T$ 
\STATE\quad Output the best solution over
	$T=c_{\epsilon,\alpha,\beta} \log{n}$ runs
\STATE\quad \%	
	Here $\epsilon>0$ and 
$c_{\epsilon,\alpha,\beta} \ge \frac{2(\alpha+1)}{3\epsilon\alpha\beta}+1.$
\end{algorithmic}
\end{algorithm}

\begin{theorem}
\label{thrm:babissdp}
Algorithm {\SDPalgo} is a $\beta$-approximation algorithm for $f_{\alpha}$ where $\beta>0.796$ with probability at least $1-\bigO(n^{-1})$.
\end{theorem}

\begin{proof}
First, notice that we can rewrite the objective as
\begin{align*}
& \sum_{e=(i,j)}\frac{1+y_{0i}+y_{0j}+y_{ij}}{4} + \\
& \frac{\alpha }{4}\Big( \sum_{i \neq j} \frac{1-y_{0i}}{2} +
\sum_{i \neq j} \frac{1-y_{0j}}{2} +\sum_{i \neq j}\frac{1-y_{ij}}{2} \Big).
\end{align*}

Let $Y = [v_0 v_1 \ldots v_n]^T [v_0 v_1 \ldots v_n]$ be the Cholesky decomposition
of matrix $Y$. We analyze the randomized rounding step which is equivalent to 
considering a random hyperplane $\mathcal{H}$ that goes through the origin
and placing in set $S$ all the vertices whose corresponding vector
for the Cholesky decomposition $v_i$ is on the same side of $\mathcal{H}$
with $v_0$. 
Our goal now is to lower bound the expectation of our objective upon this randomized rounding. 
The expectation of the terms of the form $ \frac{1-y_{ij}}{2}$ is equal to the probability 
that the two vectors $v_i,v_j$ are on different sides of the random hyperplane.
As in Goemans-Williamson \cite{goemans} this probability  is 
\[
\Prob{ \sgn(v_ir) \neq \sgn(v_jr)} = \frac{\arccos(v_iv_j)}{\pi}.
\]

Furthermore,  again as in Goemans-Williamson \cite{goemans}, for any $0 \leq \theta \leq \pi$ we have
\[
\frac{\theta}{\pi} \geq 0.87856 \frac{1-\cos \theta}{2}
> \beta \frac{1-\cos \theta}{2}.
\]

Now, we lower bound the expectation of the first term in our objective. To do so, 
we need to compute the probability that $\sgn(v_ir) = \sgn(v_jr) = \sgn(v_0r)$. 
Consider the following events:
\[
A: \qquad \sgn(v_ir) = \sgn(v_jr) = \sgn(v_0r)
\]
\[
B_i: \qquad \sgn(v_ir) \neq \sgn(v_jr) = \sgn(v_0r)
\]
\[
B_j: \qquad \sgn(v_jr) \neq \sgn(v_ir) = \sgn(v_0r)
\]
\[
B_0: \qquad \sgn(v_0r) \neq \sgn(v_jr) = \sgn(v_ir)
\]
Notice that $\Prob{B_i} = \Prob{ \sgn(v_jr)=\sgn(v_0r)} - \Prob{A}$, and  that similar equations hold for indices $(j,0)$.
Furthermore, by the pidgeonhole principle 
$\Prob{A}+\Prob{B_i}+\Prob{B_j}+\Prob{B_0}=1$.
Hence, by solving for $\Prob{A}$ and  by using elementary calculus we obtain the following lower bound:

\begin{eqnarray*}
\Prob{A} & = &
1 -\frac{1}{2\pi} \big( \arccos(v_0v_i)+\arccos(v_0v_j)+\arccos(v_iv_j) \big) \\
& \geq & \frac{\beta}{4} \big( 1+v_0v_i+v_0v_j+v_iv_j \big).
\end{eqnarray*}

Combining the above lower bounds, we obtain
\begin{eqnarray*}
\Mean{f_{\alpha}(S)} & \geq &
\beta \Big(  \sum_{e=(i,j)}\frac{1+y_{0i}+y_{0j}+y_{ij}}{4} \Big) + \\
& & \frac{\alpha \beta}{4}\Big( \sum_{i \neq j} \frac{1-y_{0i}}{2} +
\sum_{i \neq j} \frac{1-y_{0j}}{2} +\sum_{i \neq j} \frac{1-y_{ij}}{2} \Big) \\
&=& \beta\, \SDPOPT \geq \beta\, \IPOPT.
\end{eqnarray*}

Now, we boost the probability of success of our randomized algorithm. Specifically,
we analyze part 3 of our algorithm. First, we lower bound {\IPOPT}. Consider adding each
vertex with probability $\frac{1}{2}$ to $S$. The expected value of the objective
is $\frac{m}{4}+ \alpha \Big( {n \choose 2} - {n/2 \choose 2} \Big) \approx  \frac{m}{4}+\alpha \frac{3n^2}{8} \geq \frac{3\alpha}{2} {n \choose 2}.$
Hence,
\[
\IPOPT  \geq \frac{3\alpha}{2} {n \choose 2}.
\]
Also notice that the objective is upper bounded always by $(\alpha+1){n \choose 2}$.
Define a constant $\gamma$ as
\[
\gamma = \frac{ \Mean{f_{\alpha}(S)}}{(\alpha+1){n \choose 2}}.
\]
Note that $\gamma \leq 1$.
We obtain the following lower bound on~$\gamma$:
\[
(\alpha+1) {n \choose 2} \geq \Mean{f_{\alpha}(S)} = \gamma (\alpha+1) {n \choose 2} \geq \beta\, \IPOPT \geq \beta\, \frac{3\alpha}{2} {n \choose 2},
\]
and hence $1 \geq \gamma \geq \frac{3\alpha \beta}{2(\alpha+1)}$.

Let $p=\Prob{W < (1-\epsilon)\Mean{f_{\alpha}(S)}}$, where $W$ is the actual objective value 
achieved by our randomized algorithm. We obtain the following (generous) upper bound on $p$ as follows:
\[
\Mean{f_{\alpha}(S)} \leq p(1-\epsilon)\Mean{f_{\alpha}(S)}+(1-p)(\alpha+1) {n \choose 2},
\]
which by solving for $p$ gives
\begin{eqnarray*}
p & \leq & 1-\frac{\epsilon\gamma}{1-\gamma+\epsilon \gamma}  \leq   1 -\frac{\epsilon\frac{3\alpha\beta}{2(\alpha+1)}}{1-\frac{(1-\epsilon)3\alpha\beta}{2(\alpha+1)}}=1-q.
\end{eqnarray*}
Let $c_{\epsilon,\alpha,\beta}  \ge \frac{2(\alpha+1)}{3\epsilon\alpha\beta}+1.$
Running the algorithm $c_{\epsilon,\alpha,\beta} \log{n} \geq \frac{1}{q}\log{n}$ times 
times gives that the success probability is $1-o(1)$, i.e.,         

\[ \Prob{f_{\alpha}(S)\geq (1-\epsilon)\Mean{f_{\alpha}(S)}} \geq 1-(1-q)^{\frac{\log{n}}{q}} \approx 1-\bigO(n^{-1}).
\]
\end{proof}

\section{Problem variants}
\label{sec:densestvariants}

We present two variants of our basic problem, that have many practical applications: finding {\topk} {\OQCs} (Section~\ref{subsec:densesttopk}) and  finding an {\OQC} that contains a given set of {\em query vertices} (Section~\ref{subsec:densestparty}).
\subsection{Top-{\large $k$} optimal quasi-cliques}\enlargethispage*{2\baselineskip}
\label{subsec:densesttopk}
The {\topk} version of our problem is as follows:  given a graph $G=(V,E)$ and a constant $k$, find {\topk} \emph{disjoint} {\OQCs}.
This variant is particularly useful in scenarios where finding a single dense subgraph is not sufficient, rather a set of $k >1$ dense components is required.

From a formal viewpoint, the problem would require to find $k$ subgraphs for which the sum of the various objective function values computed on each subgraph is maximized.
Due to its intrinsic hardness, however, here we heuristically tackle the problem in a greedy fashion:
we find one dense subgraph at a time, we remove all the vertices of the subgraph from the graph, and we continue until we find $k$ subgraphs or until we are left with an empty graph.
Note that this iterative approach allows us to automatically fulfil a very common requirement of finding {\topk} subgraphs that are pairwise disjoint.


\subsection{Constrained optimal quasi-cliques}
\label{subsec:densestparty}

The constrained \OQCs\ variant consists in finding an  {\OQC} that contains a set of pre-specified {\em query vertices}.
This variant is inspired by the {\em com\-munity-search problem}~\cite{Sozio}, which has many applications, such as finding thematic groups, organizing social events, tag suggestion.
Next, we formalize the problem, prove that it is {\NPhard}, and adapt our scalable algorithms (i.e., {\Galgo} and {\LSalgo}) for this  variant.

Let $G=(V,E)$ be a graph, and $Q \subseteq V$ be a set of query vertices.
We want to find a set of vertices $S\subseteq V$, so that the induced subgraph contains the query vertices~$Q$ and maximizes our objective function $f_\alpha$.
Formally, we define the following problem.
\begin{problem}[\COQCP]
\label{problem:party}
Given a graph $G=(V,E)$ and set $Q \subseteq V$,
find $S^* \subseteq V$ such that $f_\alpha(S^*)=\max_{Q \subseteq S \subseteq V} f_\alpha(S)$.
\end{problem}
It is easy to see that, when $Q=\emptyset$,  the \COQCP\ reduces to the \OQCP.
The following hardness result is immediate from Theorem 1 in \cite{uno}.  

\begin{theorem}
The \COQCP\ is {\NPhard}.
\end{theorem}
The {\Galgo} algorithm can be adapted to solve the \COQCP\ simply by ignoring the nodes $u \in Q$ during
the execution of the algorithm, so as to never remove vertices of $Q$.

Similarly, our {\LSalgo} algorithm can solve the {\COQCP} with a couple of simple modifications:
the set $S$ is initialized to the set of query vertices~$Q$,
while, during the iterative phase of the algorithm, we never allow a vertex $u \in Q$ to leave $S$.

\section{Experimental evaluation}
\label{sec:densestexperiments}

In this section we present our empirical evaluation, first on publicly available real-world graphs 
(Section~\ref{subsec:densestreal}), whose main characteristics are shown in Table~\ref{tab:densestdatasets}, and 
then on synthetic graphs where the ground truth is known (Section~\ref{subsec:densestsynthetic}).

We compare our {\OQCs}  with {\DSs}. The latter is chosen as our baseline 
given its popularity. 
For extracting {\OQCs} we use our scalable algorithms, {\Galgo} and {\LSalgo}.
As far as the  semi\-de\-fi\-nite-pro\-gram\-ming algorithm presented in Section~\ref{sec:densestoqc}, 
we recall that it has been introduced mainly to show theoretical properties of the problem tackled in this paper.
From a practical viewpoint, we have been able to run it only on the smallest datasets.
We do not present the results we obtained with the semi\-de\-fi\-nite-pro\-gram\-ming algorithm presented in Section~\ref{sec:densestoqc},
which we implemented using {\sc sdtp3} \cite{sdpt3} in {\sc Matlab} for the following reasons:
($i$) it does not scale to large networks; ($ii$) the results are inferior to the results of the scalable algorithms
for the small graphs we tested it on. For instance, for the {\it polbooks} network, the corresponding
edge density is 0.19 compared to 0.67 and 0.61 that we obtain using  {\Galgo} and {\LSalgo} respectively;
($iii$) it serves mainly as a theoretical contribution.

Following our discussion in Section~\ref{sec:densestoqc}, we run our algorithms with $\alpha=\frac{1}{3}$.
For \LSalgo, we set $T_{\max}=50$.
For finding {\DSs}, we use the Goldberg's exact algorithm~\cite{Goldberg84} for small graphs, while for graphs whose size does not allow the Goldberg's algorithm to terminate in reasonable time we use Charikar's approximation algorithm~\cite{Char00}.

All algorithms are implemented in {\sc java}, and all experiments are performed on a single machine with Intel Xeon {\sc cpu}
at 2.83GHz and 50GB {\sc ram}.

\begin{table}[t]
\begin{center}
\small	
\hspace{-6mm}\begin{tabular}{r|rrl|}
\multicolumn{1}{c}{}  &
\multicolumn{1}{c}{\sf Vertices} &
\multicolumn{1}{c}{\sf Edges} &
\multicolumn{1}{c}{\sf Description}\\ \cline{2-4}
\textsf{Dolphins} & 62 & 159  & Biological Network\\
\textsf{Polbooks} & 105 & 441 & Books Network \\
\textsf{Adjnoun} & 112 & 425 & Adj.\ and Nouns in  \\
  &   &  & `David Copperfield' \\
\textsf{Football} & 115  & 613 & Games Network\\
\textsf{Jazz}     & 198   & 2\,742 & Musicians Network \\
\textsf{Celegans N.} & 297   & 2\,148 & Biological Network \\
\textsf{Celegans M.} & 453 & 2\,025 & Biological Network\\
\textsf{Email}    & 1\,133 & 5\,451 & Email Network \\
\textsf{AS-22july06} & 22\,963  & 48\,436 & Auton.\ Systems\\
\textsf{Web-Google} & 875\,713 & 3\,852\,985  & Web Graph \\
\textsf{Youtube} &   1\,157\,822  & 2\,990\,442&  Social Network \\
\textsf{AS-Skitter} & 1\,696\,415 & 11\,095\,298 &  Auton.\ Systems \\
\textsf{Wikipedia 2005}    & 1\,634\,989  & 18\,540\,589   &     Web Graph  \\
\textsf{Wikipedia 2006/9}  & 2\,983\,494  & 35\,048\,115    & Web Graph  \\
\textsf{Wikipedia 2006/11} & 3\,148\,440  & 37\,043\,456    & Web Graph \\
 \cline{2-4}
\end{tabular}
\end{center}
\caption{\label{tab:densestdatasets}Graphs used in our experiments.}	
\end{table}

\subsection{Real-world graphs}
\label{subsec:densestreal}

We experiment with the real graphs in Table~\ref{tab:densestdatasets}. The results are shown in Table~\ref{tab:densestrealresults}.
We compare \OQCs\ out\-put\-ted by the {\Galgo} and {\LSalgo} algorithms with \DSs\ extracted with Charikar's algorithm.
Particularly, we use Charikar's method to be able to handle the largest graphs.
For consistence, Table~\ref{tab:densestrealresults} reports on results achieved by Charikar's method also for the smallest graphs.
We recall that the results in Table \ref{tab:densestintro}  refer instead to the exact Goldberg's method.
However, a comparison of the two tables on their common rows shows that Charikar's algorithm, even though it is approximate, produces almost identical results with the results produced by Goldberg's algorithm.

Table~\ref{tab:densestrealresults} clearly confirms the preliminary results reported in the Introduction: \OQCs\ have larger edge and triangle densities, and smaller diameter than \DSs.
Particularly, the edge density of \OQCs\ is evidently larger on all graphs.
For instance, on \textsf{Football} and  \textsf{Youtube}, the edge density of \OQCs\ (for both the {\Galgo} and {\LSalgo} algorithms) is about 9 times larger than the edge density of \DSs, while on \textsf{Email} the difference increases up to 20 times ({\Galgo}) and 14 times ({\LSalgo}).
Still, the triangle density of the \OQCs\ outputted by both {\Galgo} and {\LSalgo} is one order of magnitude larger than the triangle density of {\DSs} on 11 out of 15 graphs.

Comparing our two algorithms, we can see that {\LSalgo} performs generally better than {\Galgo}.
Indeed, the edge density achieved by {\LSalgo} is higher than that of {\Galgo} on 10 out of 15 graphs, while the diameter of the {\LSalgo} \OQCs\ is never larger than the diameter of the {\Galgo} \OQCs.

Concerning efficiency, all algorithms are linear in the number of edges of the graph.
Charikar's and {\Galgo} algorithm are somewhat slower than {\LSalgo}, but mainly due to bookkeeping.
\LSalgo\ algorithm's running times vary from milliseconds for the
small graphs (e.g., 0.004s for \textsf{Dolphins}, 0.002s for \textsf{Celegans \ N.}), 
few seconds for the larger graphs (e.g., 7.94s for \textsf{Web-Google} and 3.52s 
for \textsf{Youtube}) and less than one minute for the largest graphs (e.g., 59.27s for \textsf{Wikipedia 2006/11}).

\begin{table*}[t]
\small
 \begin{adjustwidth}{-3.8cm}{}
\begin{tabular}{r|rrr|lll|rrr|lll|}
\multicolumn{1}{c}{}  &  \multicolumn{3}{c}{$|S|$} &  \multicolumn{3}{c}{$\delta$} &  \multicolumn{3}{c}{$D$} &  \multicolumn{3}{c}{$\tau$}\\
\cline{2-13}

& \multicolumn{1}{|c}{\textsf{densest}}  & \multicolumn{2}{c|}{\textsf{opt. quasi-clique}}  & \multicolumn{1}{|c}{\textsf{densest}}  & \multicolumn{2}{c|}{\textsf{opt. quasi-clique}} & \multicolumn{1}{|c}{\textsf{densest}}  & \multicolumn{2}{c|}{\textsf{opt. quasi-clique}} & \multicolumn{1}{|c}{\textsf{densest}}  & \multicolumn{2}{c|}{\textsf{opt. quasi-clique}}\\
\cline{3--3} \cline{4--4} \cline{6--6} \cline{7--7} \cline{9--9} \cline{10--10} \cline{12--12} \cline{13--13}
&
\multicolumn{1}{c}{\textsf{subgraph}} &
\multicolumn{1}{c}{\sc{greedy}} &
\multicolumn{1}{c|}{\sc{ls}} &
\multicolumn{1}{c}{\textsf{subgraph}} &
\multicolumn{1}{c}{\sc{greedy}} &
\multicolumn{1}{c|}{\sc{ls}} &
\multicolumn{1}{c}{\textsf{subgraph}} &
\multicolumn{1}{c}{\sc{greedy}} &
\multicolumn{1}{c|}{\sc{ls}} &
\multicolumn{1}{c}{\textsf{subgraph}} &
\multicolumn{1}{c}{\sc{greedy}} &
\multicolumn{1}{c|}{\sc{ls}} \\
\cline{2-13}

\textsf{Dolphins}     &  19  &   13  &   8  &  0.27  &  0.47  &  0.68  &  3  &  3  &  2  &  0.05  &  0.12  &  0.32\\
\textsf{Polbooks}     &  53  &   13  &  16  &  0.18  &  0.67  &  0.61  &  6  &  2  &  2  &  0.02  &  0.28  &  0.24\\
\textsf{Adjnoun}      &  45  &   16  &  15  &  0.20  &  0.48  &  0.60  &  3  &  3  &  2  &  0.01  &  0.10  &  0.12\\
\textsf{Football}     &  115 &   10  &  12  &  0.09  &  0.89  &  0.73  &  4  &  2  &  2  &  0.03  &  0.67  &  0.34\\
\textsf{Jazz}         &  99  &   59  &  30  &  0.35  &  0.54  &  1     &  3  &  2  &  1  &  0.08  &  0.23  &  1   \\
\textsf{Celeg.\ N.}   &  126 &   27  &  21  &  0.14  &  0.55  &  0.61  &  3  &  2  &  2  &  0.07  &  0.20  &  0.26\\
\textsf{Celeg.\ M.}   &  44  &   22  &  17  &  0.35  &  0.61  &  0.67  &  3  &  2  &  2  &  0.07  &  0.26  &  0.33\\
\textsf{Email}        &  289 &   12  &   8  &  0.05  &  1     &  0.71  &  4  &  1  &  2  &  0.01  &  1     &  0.30\\
\textsf{AS-22july06}  &  204 &   73  &  12  &  0.40  &  0.53  &  0.58  &  3  &  2  &  2  &  0.09  &  0.19  &  0.20\\
\textsf{Web-Google}   &  230 &   46  &  20  &  0.22  &  1     &  0.98  &  3  &  2  &  2  &  0.03  &  0.99  &  0.95\\
\textsf{Youtube}      &  1874&  124  & 119  &  0.05  &  0.46  &  0.49  &  4  &  2  &  2  &  0.02  &  0.12  &  0.14\\
\textsf{AS-Skitter}   &  433 &  319  &  96  &  0.41  &  0.53  &  0.49  &  2  &  2  &  2  &  0.10  &  0.19  &  0.13\\
\textsf{Wiki '05}     & 24555&  451  & 321  &  0.26  &  0.43  &  0.48  &  3  &  3  &  2  &  0.02  &  0.06  &  0.10\\
\textsf{Wiki '06/9}   &  1594&  526  & 376  &  0.17  &  0.43  &  0.49  &  3  &  3  &  2  &  0.10  &  0.06  &  0.11\\
\textsf{Wiki '06/11}  &  1638&  527  &  46  &  0.17  &  0.43  &  0.56  &  3  &  3  &  2  &  0.31  &  0.06  &  0.35\\
\cline{2-13}
\end{tabular}
\begin{center} 
\caption{\label{tab:densestrealresults}
\textsf{Densest subgraphs} extracted with Charikar's method vs.\ \OQCs\ extracted with the proposed \Galgo\ algorithm ({\sc greedy}) and \LSalgo\ algorithm ({\sc ls}).
$\delta = e[S] / {|S| \choose 2}$ is the edge density of the extracted subgraph $S$, $D$ is the diameter, and
$\tau = t[S] / {|S| \choose 3}$ is the triangle density.}
\end{center}
 \end{adjustwidth}
\end{table*}

\begin{figure}[t]
\centering
\begin{tabular}{@{}c@{}@{\ }c@{}}
\includegraphics[width=0.45\textwidth]{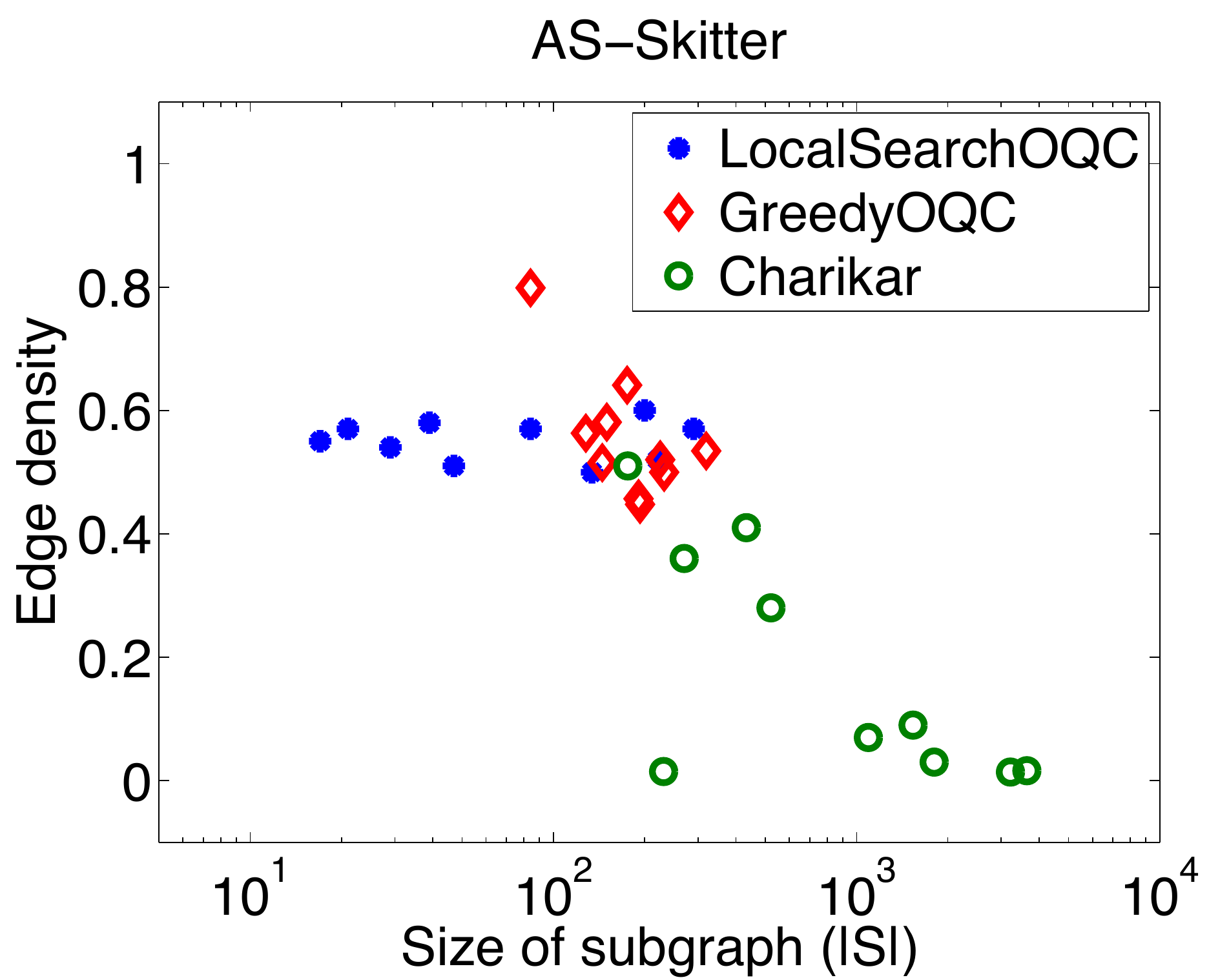} & \includegraphics[width=0.45\textwidth]{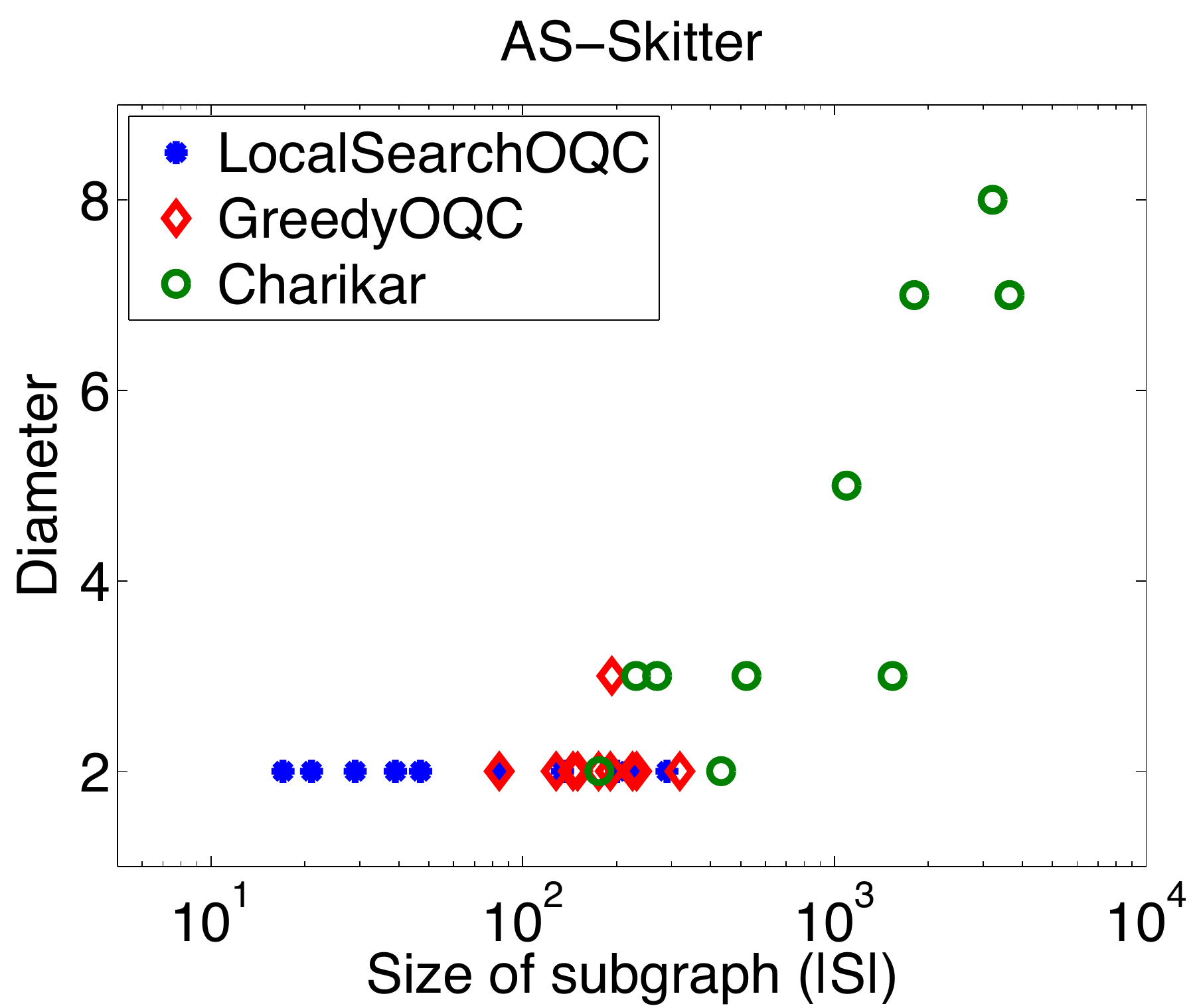}\\
\includegraphics[width=0.45\textwidth]{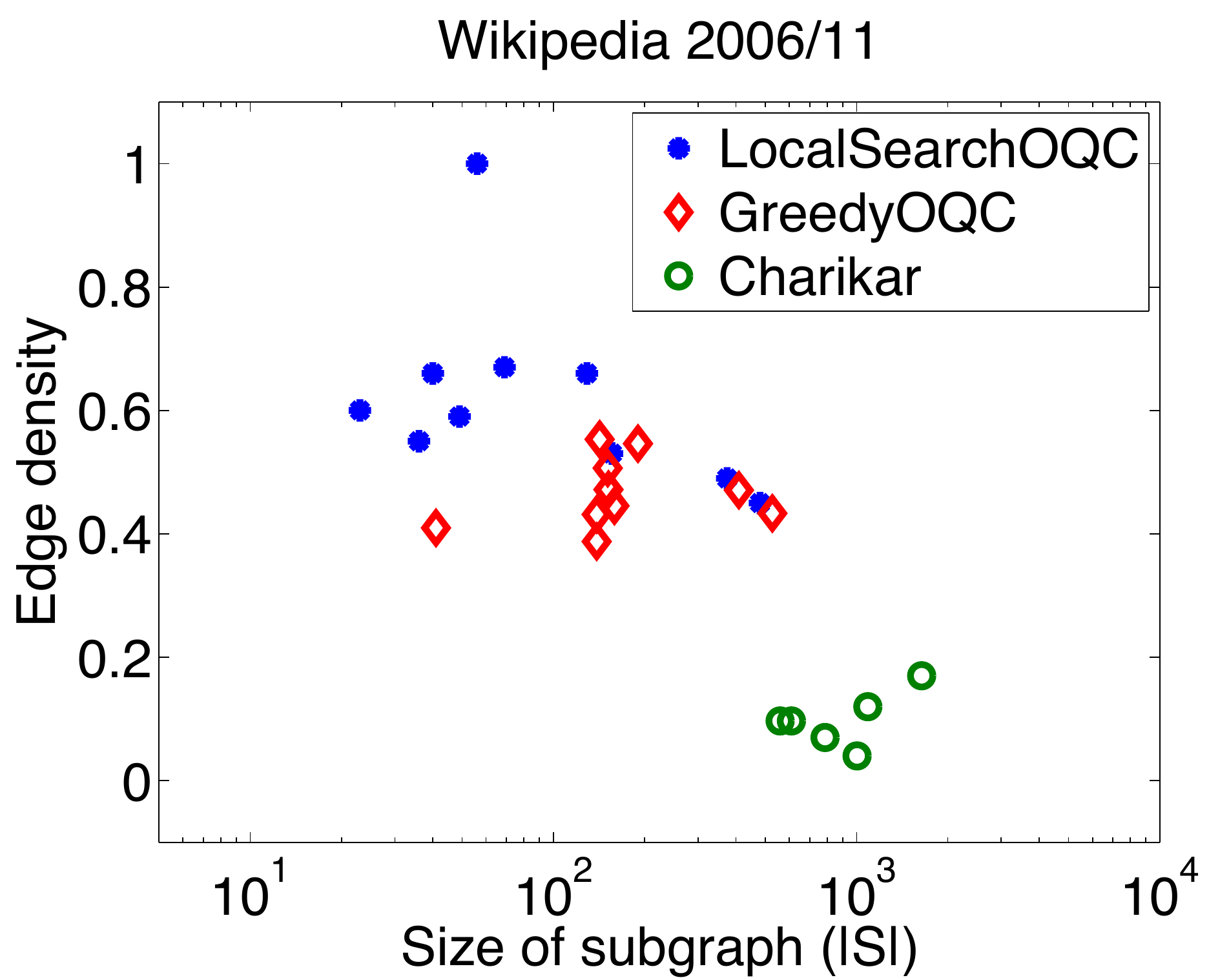} & \includegraphics[width=0.45\textwidth]{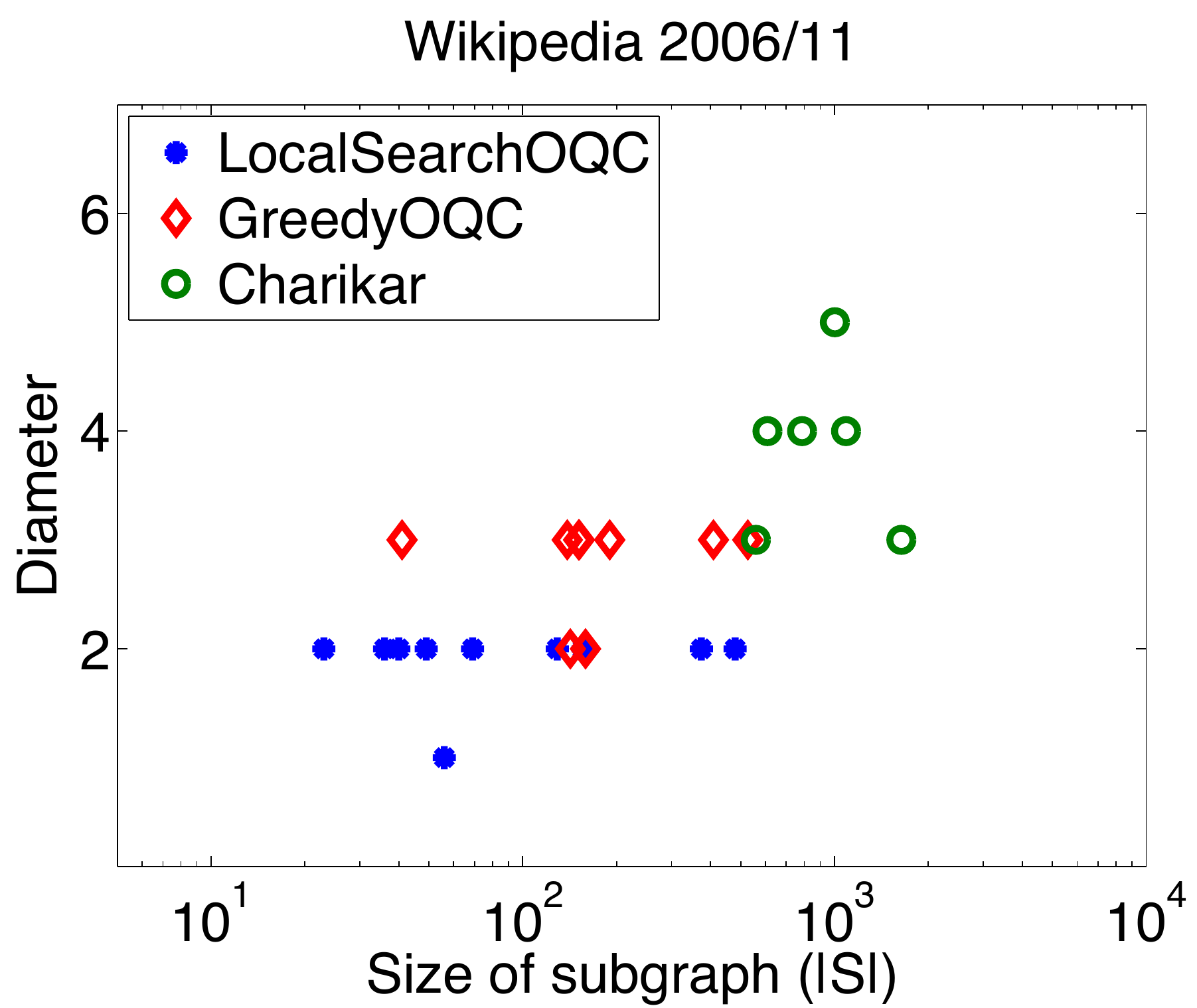}
\end{tabular}
\caption{\label{fig:topk1}
Edge density and diameter of the top-10 subgraphs found by our {\Galgo} and {\LSalgo} methods, and Charikar's algorithm, on the \textsf{AS-skitter} graph (top) and the \textsf{Wikipedia 2006/11} graph (bottom).}
\end{figure}

\spara{Top-$k$ \OQCs.}
Figure~\ref{fig:topk1} evaluates \topk\ \OQCs\ and \topk\ \DSs\ on the \textsf{AS-Skitter} and \textsf{Wikipedia 2006/11} graphs 
using the iterative method  described in Section~\ref{subsec:densesttopk}.
Similar results hold for the other graphs but are omitted due to space constraints.

For each graph we show two scatterplots.
The $x$ axis in logarithmic scale reports the size of each of the top-$k$ dense components, while the $y$ axes show the edge density and the diameter, respectively.
In all figures, \OQCs\ correspond to  blue filled circles ({\LSalgo}) or red diamonds ({\Galgo}), while \DSs\ correspond to green circles.
It is evident that \OQCs\ are significantly better in terms of both edge density and diameter also in this top-$k$ variant.
The edge density is in the range $0.4-0.7$  and the diameter is always 2 or 3, except for a 56-vertex clique in \textsf{Wikipedia 2006/11} with diameter 1.
On the contrary, the \DSs\ are large graphs, with diameter ranging typically from 3 to 5, with significantly smaller edge densities: besides few exceptions, the edge density of \DSs\ is always around $0.1$ or even less.

\subsection{Synthetic graphs}
\label{subsec:densestsynthetic}

Experiments on synthetic graphs deal with the following task:
a (small) clique is planted in two different types of random graphs, and the goal is to check if the dense subgraph algorithms are able to recover those cliques.
Two different random-graph models are used as host graphs for the cliques:
($i$) {\ER} and ($ii$) random power-law graphs.
In the former model, each edge exists with probability $p$ independently of the other edges.
To generate a random power-law graph, we follow the Chung-Lu model~\cite{chunglu}:
we first generate a degree sequence $(d_1,\ldots,d_n)$ that follows a power law with a pre-specified
slope and we connect each pair of vertices $i,j$ with probability proportional to $d_id_j$.

We evaluate our algorithms by measuring how ``close'' are the returned subgraphs to the planted clique.
In particular, we use the measures of {\em precision} $P$ and {\em recall} $R$, defined as
\begin{eqnarray*}
P & = &\frac{\#\{\text{returned vertices from hidden clique}\}}{size\{\text{subgraph returned}\}}, \mbox{ and} \\
R & = & \frac{\#\{\text{returned vertices from hidden clique}\}}{size\{\text{hidden clique}\}}.
\end{eqnarray*}
Next we discuss the results obtained.
For the {\ER} model we also provide a theoretical justification of the outcome of the two tested algorithms.

\spara{{\ER} graphs.}
We plant a clique of 30 vertices on {\ER} graphs with $n=3\,000$ and edge probabilities $p\in\{0.5,0.1,0.008\}$.
Those values of $p$ are selected to represent very dense, medium-dense, and sparse graphs.

We report in Table~\ref{tab:densesterdos} the results of running our \LSalgo\ and \Galgo\ algorithms for extracting \OQCs, as well as the Goldberg's algorithm for extracting \DSs.
We observe that our two algorithms, {\LSalgo} and {\Galgo}, produce \emph{identical} results, thus we refer to both of them as {\OQCs} algorithms.
We see that the algorithms produce two kinds of results:
they either find the hidden clique, or they miss it and return the whole graph.
In the very dense setting $(p=0.5)$ all algorithms miss the clique,
while in the sparse setting $(p=0.008)$ all algorithms recover it.
However, at the middle-density setting $(p=0.1)$ only the \OQCs\ algorithms find the clique,
while the Goldberg's algorithm misses it.

To better understand the results shown on Table~\ref{tab:densesterdos}, we provide a theoretical explanation of the behavior of the algorithms depending on their objective.
Assume that $h$ is the size of the hidden clique.
If $np \geq h-1$ the \DS\ criterion always returns the whole graph.
In our experiments, this happens with $p=0.5$ and $p=0.1$.
On the other hand, if $np < h-1$ the \DS\ corresponds to the hidden clique, and therefore the Goldberg's algorithm cannot miss it.

Now consider our objective function , i.e., the edge-surplus function $f_\alpha$, and let us discard for simplicity the constant ${n \choose 2}$, because this does not affect the validity of the following reasoning.
The expected score for the hidden clique is $\Mean{f_\alpha(H)}= f_\alpha(H)=(1-\alpha){h \choose 2}$.
The expected score for the whole network is $\Mean{f_\alpha(V)} = \big( p{n \choose 2} + (1-p) {h \choose 2}\big) - \alpha {n \choose 2}$.

We obtain the following two cases:
(A) when $p> \alpha$, we have $\Mean{f_\alpha(V)} \geq f_\alpha(H)$.
(B) when $p < \alpha$, we have $f_\alpha(H) \geq \Mean{f_\alpha(V)}$.
This rough analysis explains our findings.\footnote{\scriptsize The analysis can be tightened via Chernoff bounds, but we avoid this here due to space constraints.}

\spara{Power-law graphs.}
We plant a clique of 15 vertices in random power-law graphs of again 3\,000 vertices,
with power-law exponent varying from 2.2 to 3.1.
We select these values since most real-world networks have power-law exponent in this range \cite{newman2003structure}.
For each exponent tested, we generate five random graphs, and all the figures we report are averages over these five trials.

Again, we compare our {\Galgo} and {\LSalgo} algorithms with the Goldberg's algorithm.
The {\LSalgo} algorithm is run seeded with one of the vertices of the clique.
The justification of this choice is that we can always re-run the algorithm until it finds such a vertex with high probability.\footnote{\scriptsize
If the hidden clique is of size $\bigO(n^{\epsilon})$, for some $0 \leq \epsilon < 1$,
a rough calculation
shows that it suffices to run the algorithm a sub-linear number of times
(i.e., $\bigO((1-\gamma)n^{1-\epsilon})$ times) in order to obtain one of the vertices
of the clique as a seed with probability at least $1-\gamma$.
}

\begin{table}[t]
\begin{center}
\small
\begin{tabular}{rl|rrr|rrr}

\multicolumn{2}{c}{{\ER}} & \multicolumn{3}{c}{}  & \multicolumn{3}{c}{} \\
\multicolumn{2}{c}{parameters} & \multicolumn{3}{c}{\textsf{densest subgraph}}  & \multicolumn{3}{c}{\textsf{optimal \QC}} \\
\hline
\multicolumn{1}{c}{$n$} &
\multicolumn{1}{c|}{$p$} &
\multicolumn{1}{c}{$|S|$} &
\multicolumn{1}{c}{$P$} &
\multicolumn{1}{c|}{$R$} &
\multicolumn{1}{c}{$|S|$} &
\multicolumn{1}{c}{$P$} &
\multicolumn{1}{c}{$R$} \\
\hline
3\,000 & 0.5   & 3\,000 & 0.01 & 1.00 & 3\,000 & 0.01 & 1.00 \\
3\,000 & 0.1   & 3\,000 & 0.01 & 1.00 &     30 & 1.00 & 1.00 \\
3\,000 & 0.008 &     30 & 1.00 & 1.00 &     30 & 1.00 & 1.00 \\
\hline
\end{tabular}
\end{center}
\caption{\label{tab:densesterdos}
Subgraphs returned by Goldberg's max-flow algorithm and by our
two algorithms (\Galgo, \LSalgo) on {\ER} graphs with 3\,000 vertices and three values of $p$,
and with a planted clique of 30 vertices.}	
\end{table}

\begin{figure}[t]
\centering
\includegraphics[width=0.49\columnwidth]{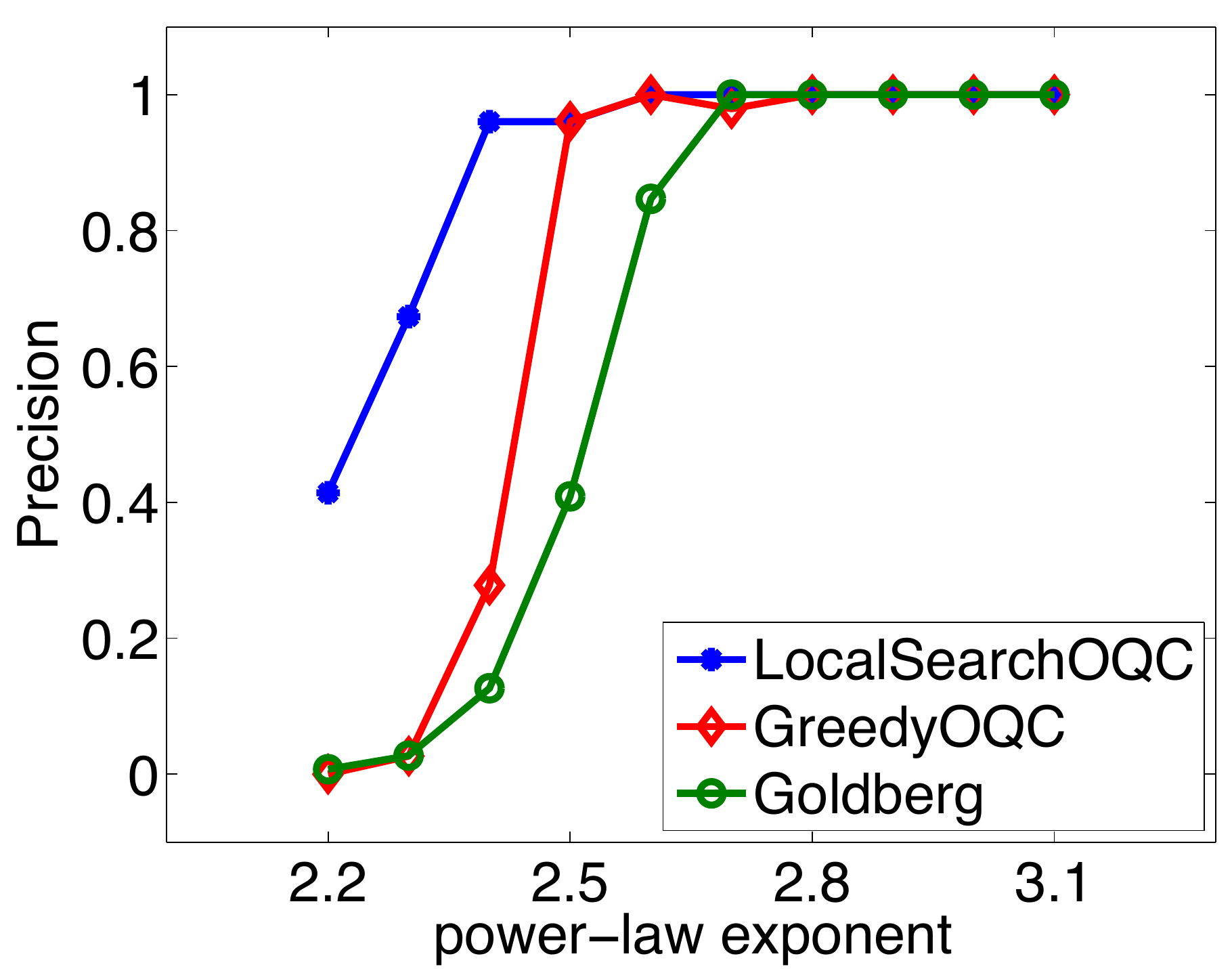}
\includegraphics[width=0.49\columnwidth]{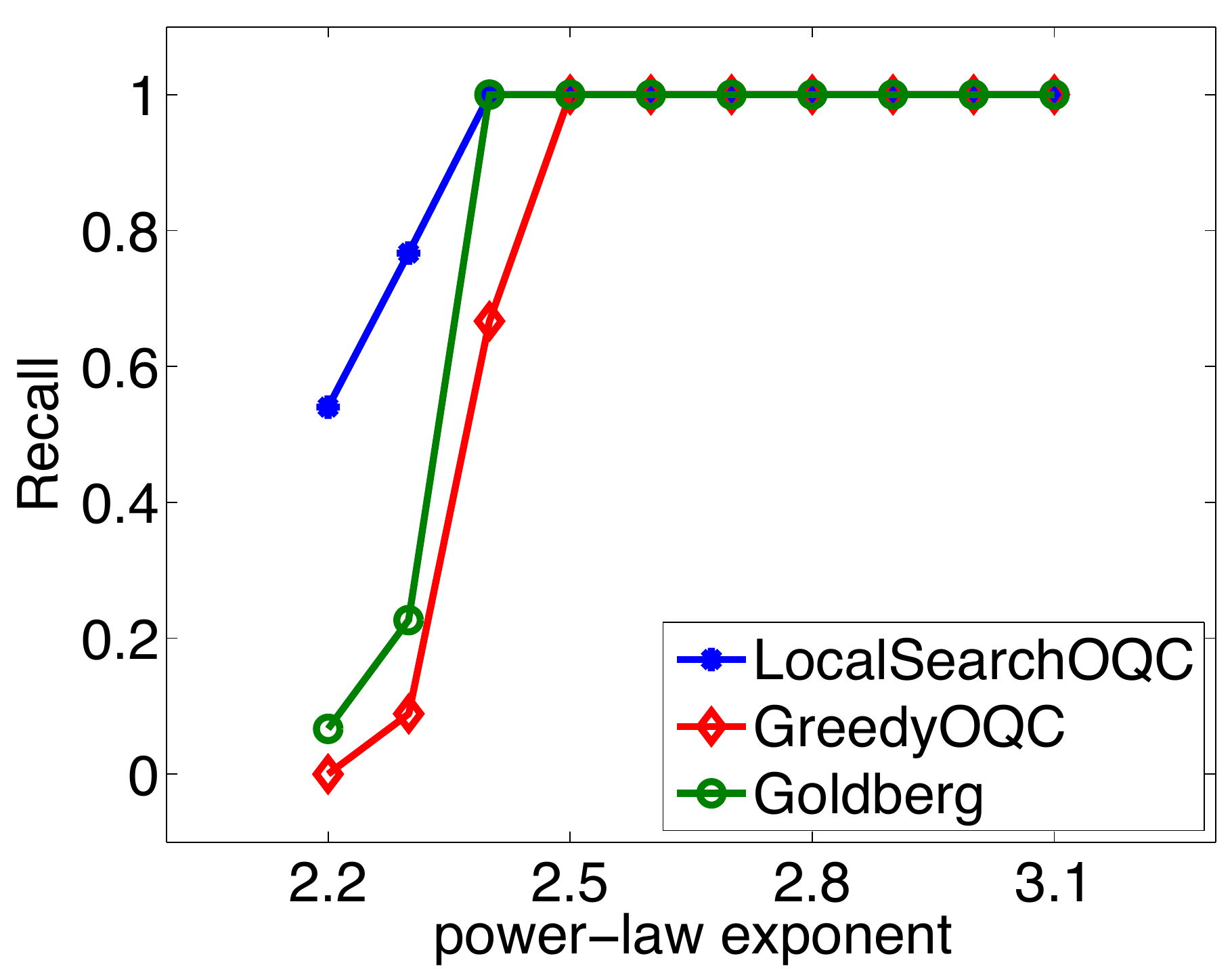}
\caption{\label{fig:PR}
Precision and recall for our method and Goldberg's algorithm vs.\ the power-law exponent of the host graph.}
\end{figure}

The precision and recall scores of the three competing algorithms as a function of the power-law exponent are shown in Figure~\ref{fig:PR}.
As the exponent increases the host graph becomes sparser and both algorithms
have no difficulty in finding the hidden clique.
However, for exponent values ranging between 2.2 and 2.6 the \OQCs\ are significantly better than
the \DSs.
Indeed, in terms of precision, the Goldberg's algorithm is outperformed by both our algorithms.
In terms of recall, our {\LSalgo} is better than Goldberg's, while our {\Galgo} performs slightly worse.
An explanation for this is that the  {\Galgo} algorithm detects other high-density subgraphs, but not 
exactly the planted clique. As an example, with power-law exponent 2.3, 
{\Galgo} finds a subgraph with  23 vertices and edge density 0.87.

\spara{Stability with respect to $\mathbf{\alpha}$.}
We also test the sensitivity of our density measure with respect to the parameter $\alpha$.
We use again the planted-clique setting, and we test the ability of our algorithms to recover the clique as we vary the parameter $\alpha$.
We omit detailed plots, due to space constraints, but we report that the behavior of both algorithms is extremely stable with respect to $\alpha$.
Essentially, the algorithms again either find the clique or miss it, depending on the graph-generation parameters, as we saw in the previous section, namely, the probability $p$ of the {\ER} graphs, or the exponent of the power-law graphs.
Moreover, in all cases, the performance of our algorithms, measured by precision and recall as in the last experiment, does not depend on $\alpha$.

\section{Applications}
\label{sec:densestapps}

In this section we show experiments concerning our constrained \OQCs\ variant introduced in Section \ref{subsec:densestparty}.
To this end, we focus on two applications that can be commonly encountered in real-world scenarios: finding thematic groups and finding highly correlated genes from a microarray dataset.
For the sake of brevity of presentation, we show next results for only one of our scalable algorithms, particularly the {\LSalgo} algorithm.

\subsection{Thematic groups}
\label{subsec:densestparties}

\spara{Motivation.}
Suppose that a set of scientists $Q$ wants to organize a workshop.
How do they invite other scientists to participate in the workshop so that the set of all the participants, including $Q$, have similar interests?

\spara{Setup.}
We use a co-authorship graph extracted from the {\sc dblp} dataset.
The dataset contains publications in all major computer-science journals.
There is an undirected edge between two authors if they have coauthored a journal article.
Taking the largest connected component gives a graph of 226K vertices and 1.4M edges.

We evaluate the results of our algorithm  qualitatively,
in a sanity check form rather than a  strict and quantitative way,
which is not even well-defined.
We perform the following two queries:
$ Q_1 = \{ \text{Papadimitriou}, \text{Abiteboul} \}$
and $ Q_2 = \{ \text{Papadimitriou},\text{Blum} \}$.

\spara{Results.}
Papadimitriou is one of the most prolific computer scientists and has worked on a wide range of areas.
With query $Q_1$ we invoke his interests in database theory given that Abiteboul is an expert in this field.
As we can observe from Figure~\ref{fig:dblpDB}, the \OQC\ outputted contains database scientists.

On the other hand, with query $Q_2$ we invoke Papadimitriou's interests in theory,
given that Blum is a Turing-award theoretical computer scientist.
As we can see in Figure~\ref{fig:dblpTCS}, the returned \OQC\ contains well-known theoretical computer scientists.

\begin{figure}[t]
\begin{minipage}[b]{\linewidth}
\centering
\begin{tabular}{c}
\hline
Abiteboul,
Bernstein,
Brodie,
Carey,
Ceri,
Crof, \\
DeWitt,
Ehrenfeucht,
Franklin,
Gawlick,
Gray,
Haas, \\
Halevy,
Hellerstein,
Ioannidis,
Jagadish,
Kanellakis, \\
Kersten,
Lesk,
Maier,
Molina,
Naughton,
Papadimitriou, \\
Pazzani,
Pirahesh,
Schek,
Sellis,
Silberschatz,
Snodgrass, \\
Stonebraker,
Ullman,
Weikum,
Widom,
Zdonik \\
\hline
\end{tabular}
\caption{\label{fig:dblpDB}
Authors returned by our \LSalgo\ algorithm when queried with \textsf{Papadimitriou} and \textsf{Abiteboul}.
The set includes well-known database scientists.
The induced subgraph has 34 vertices and 457 edges.
The edge density is 0.81, the diameter is 3, the triangle density is 0.66.}
\end{minipage}
\begin{minipage}[b]{\linewidth}
\centering

\bigskip

\begin{tabular}{c}
\hline
Alt,
Blum,
Garey,
Guibas,
Johnson, \\
Karp,
Mehlhorn,
Papadimitriou,
Preparata, \\
Tarjan,
Welzl,
Widgerson,
Yannakakis,
\\
\hline
\end{tabular}
\caption{\label{fig:dblpTCS}
Authors returned by our \LSalgo\ algorithm when queried with \textsf{Papadimitriou} and \textsf{Blum}.
The set includes well-known theoretical computer scientists.
The induced subgraph has 13 vertices and 38 edges.
The edge density is 0.49, the diameter is 3, the triangle density is 0.14.}
\end{minipage}
\end{figure}

\subsection{Correlated genes}
\label{subsec:densestcorrgenes}

\spara{Motivation.}
Detecting correlated genes has several applications.
For instance, clusters of genes with similar expression levels are typically under similar transcriptional control.
Furthermore, genes  with similar expression patterns may imply co-regulation or relationship in functional pathways.
Detecting gene correlations has played a key role in discovering unknown types of breast cancer~\cite{tib}.
Here, we wish to illustrate that {\OQCs} provide a useful graph theoretic framework for
gene co-expression network analysis~\cite{langston}, without delving deeply into biological
aspects of the results.

\spara{Setup.}
We use the publicly available breast-cancer dataset of van de Vijner et al.~\cite{veer},
which consists of measurements across 295 patients of 24\,479 probes.
Upon running a standard probe selection algorithm based on Singular Value Decomposition (SVD),
we obtain a 295$\times$1000 matrix.
The graph $G$ in input to our \LSalgo\ algorithm is derived using the well-established approach defined in~\cite{langston}:
each gene  corresponds to a vertex in $G$, while an edge between any pair of genes
$i,j$ is drawn if and only if the the modulus of the Pearson's correlation coefficient $|\rho(i,j)|$ exceeds a given threshold
$\theta$ ($\theta = 0.99$ in our setting).
We perform the following query, along the lines of the previous section: ``find highly
correlated genes with the tumor protein 53 (p53)''.
We select p53 since it is known to play a key role in  tumorigenesis.

\begin{figure}[t]
\begin{minipage}[b]{\linewidth}
\centering
\begin{tabular}{c}
\hline
p53,
BRCA1,
ARID1A,
ARID1B,
ZNF217,
FGFR1,
KRAS, \\
NCOR1,
PIK3CA,
APC,
MAP3K13,
STK11,
AKT1,
RB1  \\
\hline
\end{tabular}
\caption{\label{fig:p53sub}
Genes returned by our \LSalgo\ algorithm when queried with \textsf{p53}.
The induced subgraph is a clique with 14 vertices.}
\end{minipage}
\begin{minipage}[b]{\linewidth}
\bigskip	
\centering
\includegraphics[width=0.72\textwidth]{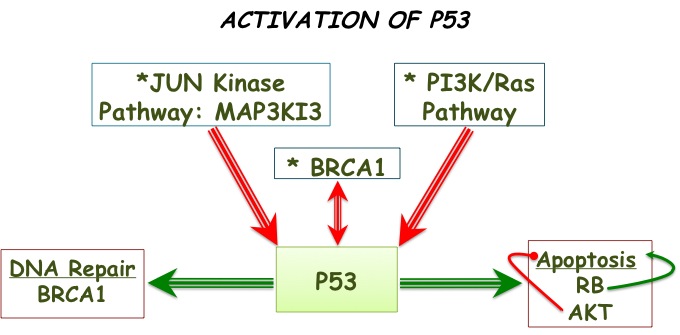}
\caption{A tumorigenesis pathway consistent with our findings: 
p53 lies in the nexus of many signaling pathways. For instance, signals from the 
JNK or PIK3/AKT pathways can activate p53.
In turn, p53 can promote apoptosis or DNA repair. RB and AKT are important for apoptosis,
 whereas BRCA1 is critical for DNA repair. Mutations in these pathways can lead to cancer.}
\label{fig:p53}
\end{minipage}
\vspace{-4mm}
\end{figure}

\spara{Results.}
The output of our algorithm is a clique consisting of 14 genes shown in Figure~\ref{fig:p53sub}.
A potential explanation of our finding is the pathway depicted in Figure~\ref{fig:p53}, which shows
that the activation of the p53 signaling can be initiated by signals coming from the PI3K/AKT pathway.
Both PI3KCA and AKT1 that are detected by our method are key players of this pathway.
Furthermore, signals from the JUN kinase pathway can also trigger the p53-cascade;
MAP3K13 is a member of this pathway.

One of the results of p53 signaling is apoptosis, a process promoted by RB.
The latter can also regulate the stability and the apoptotic function of p53.
Finally, our output includes BRCA1, which is known to physically associate with
p53 and affect its actions \cite{weinberg}.

\clearpage
\chapter{ Structure of the Web Graph  }
\label{hadichapter}
\lhead{\emph{Stucture of the Web Graph}}
\section{Introduction} 
\label{sec:hadiintro}

In this Chapter we design scalable algorithms which allow us to 
understand better the structure of information, social and economic networks,
including the Web graph, a prominent information network. 
Our goal is to shed light on important properties of the Web graph 
and other large real-world networks such as the effective diameter, number and 
sizes of connected components and temporal patterns of evolution. 
Specifically, we perform experiments on the {\it YahooWeb} graph with 60 billion edges.
We reveal facts about the structure of the Web, 
including the well known small-world phenomenon, the multimodal shape of the radius
distribution and time-varying patters. 

Analyzing networks of this scale is a challenge. For this purpose
we design \hadi, a \mapreduce algorithm which scales to large-scale networks.
Our algorithm relies on approximate counting of distinct elements of a multiset. 
We use the algorithm of Flajolet and Martin \cite{flajolet85probabilistic} as our ``black-box'',
which is historically the first algorithm proposed for this problem. Since then, many other algorithms have been proposed,
including the Hyperloglog counters \cite{Flajolet07hyperloglog}.
It is worth mentioning that the current state-of-the-art method \cite{hyperanf} which 
outperforms \hadi is based on the latter counters \cite{Flajolet07hyperloglog},
which provide an exponential space improvement over the original 
Flajolet-Martin counters \cite{flajolet85probabilistic}.

The rest of the Chapter is organized as follows: 
Section~\ref{sec:hadirelated}  defines the notions of effective radius and diameter
and reviews briefly related work.  
Section~\ref{sec:hadimoments} provides an explanation of how 
the number of distinct elements of multisets is related to finding the diameter
of a graph. 
Section~\ref{sec:hadimeth} presents algorithmic tools 
for finding the distribution of radii and the diameter of large-scale networks. 
Section~\ref{sec:haditimings} provides wall-clock times.
Section~\ref{sec:hadifindings} presents findings on the structure of real-world 
networks. 

\section{Related Work}
\label{sec:hadirelated}

\spara{Definitions:} First, we review basic graph theoretic definitions \cite{bondy1976graph} and then introduce the 
notions of effective radius and diameter. Let $G(V,E)$ be a directed graph.
The radius/eccentricity of a vertex $v$ is the greatest shortest-path distance between $v$ and any other vertex. 
The radius $r(G)$ is the minimum radius of any vertex.
The diameter $d(G)$ is the maximum radius of any vertex.
Since the radius and the diameter are susceptible to outliers
(e.g., long chains), we follow the literature~\cite{Leskovec05Realistic}
and define the {\em effective} radius and diameter as follows.

\begin{definition} [Effective Radius] 
For a node $v$ in a graph $G$, the effective radius $r_{eff}(v)$ of $v$ 
is the 90th-percentile of all the shortest distances from $v$.
\end{definition}

\begin{definition} [Effective Diameter] 
The effective diameter $d_{eff}(G)$ of a graph $G$ is the minimum number of hops in 
which 90\% of all connected pairs of nodes can reach each other.
\end{definition}

In Section~\ref{sec:hadifindings} we use the following three radius-based Plots:

\squishlist
    \item \textbf{Static Radius Plot} (or just ``Radius plot'')
    of graph $G$ shows the distribution (count)
    of the effective radius of nodes at a specific time.

    \item \textbf{Temporal Radius Plot}
    shows the distributions of effective radius of nodes at several times.

    \item \textbf{Radius-Degree Plot} shows the scatter-plot
    of the effective radius $r_{eff}(v)$ versus the  degree $d_v$ for each node $v$.
\squishend

\spara{Computing Radius and Diameter:}
Typical algorithms to compute the radius and the diameter of a graph
include Breadth First Search (BFS) and Floyd's algorithm~\cite{cormen:algorithms}
when no negative cycles are present. 
Both approaches are prohibitively slow for large-scale graphs, requiring $O(n^2 + nm)$
and $O(n^3)$ time respectively. For the same reason, related BFS or all-pair shortest-path based algorithms 
like ~\cite{Ferrez1998HPCSA,Bader2008,MaJCST1993,Sinha1986} can not handle large-scale graphs.

A sampling approach starts BFS from a subset of nodes, typically chosen at random as in \cite{broder2000graph}.
Despite its practicality, this approach has no obvious solution for choosing the representative sample for BFS.
An interesting approach has been proposed by Cohen \cite{Cohen:1997:SFA:269994.270010}, 
but according to practitioner's experience \cite{Boldi:2011:HAN:1963405.1963493} it appears not to 
be as scalable as the ANF algorithm  \cite{DBLP:conf/kdd/PalmerGF02}.
The latter is closely related to our work since it is a sequential algorithm 
based on Flajolet-Martin sketches \cite{flajolet85probabilistic}.
We review its key idea in the next Section. 

\spara{Distinct Elements in Multisets:} Let $A = \{ a_1, \ldots, a_m \}$ 
be a multiset where $a_i \in [n]$ for all $i = 1,\ldots,m$. 
Let $m_i = | \{ j : a_j = i \} |$. For each $k \geq 0$ define 
$F_k = \sum_{i=1}^n m_i^k$. The numbers $F_k$ are called 
frequency moments of the multiset and provide useful statistics. 
We notice that $F_0$ is the number of distinct elements in $A$, 
$F_1=m$. 
Historically, Morris was the first to show that $F_1$ can be approximated 
with $O(\log{\log{m}} =O(\log{\log{n}}$ \cite{morris1978counting}. 
Flajolet and Martin designed an algorithm that needs $O(\log{n})$ bits of memory
to approximate $F_0$ \cite{flajolet85probabilistic}. 
Since then, many excellent researches have looked into this problem, 
see \cite{Alon:1996:SCA:237814.237823, Bar-Yossef:2002:CDE:646978.711822, Bar-Yossef:2002:RSA:545381.545464,
Beyer:2007:SDE:1247480.1247504,Cohen:1997:SFA:269994.270010,DBLP:conf/esa/DurandF03,Flajolet07hyperloglog,
DBLP:conf/focs/IndykW03,Woodruff:2004:OSL:982792.982817}.
Recently, Kane, Nelson and Woodruff provided an optimal algorithm for estimating $F_0$ \cite{Kane:2010:OAD:1807085.1807094}.

\section{Distinct Elements and the Diameter}
\label{sec:hadimoments}

In this Section, we sketch the key idea of \cite{DBLP:conf/kdd/PalmerGF02}, 
which shows how one can use a space efficient algorithm for estimating distinct elements
in a multiset to estimate the diameter and radius distribution of a graph. 
Assume that for each vertex $v$ in the graph, we maintain the number of neighbors reachable from $v$ within $h$ hops.
As $h$ increases, the number of neighbors increases until it stabilzes.
The diameter is $h$ where the number of neighbors within $h+1$ does not increase for every node. 

To generate the Radius plot, we need to calculate the effective radius of every node.
In addition, the effective diameter is useful for tracking the evolution of networks.
Assume we have a `set' data structure that supports two functions: \emph{add()} for adding an 
item, and \emph{size()} for returning the count of distinct items. With the set, radii of vertices can be computed as follows:

\begin{enumerate}
  \item For each vertex $i$, create a set $S_i$ and initialize it by adding $i$ to it.
  \item For each vertex $i$, continue updating $S_i$ by adding 1,2,3,...-step 
neighbors of $i$ to $S_i$. When the size of $S_i$ stabilizes for first time, 
then the vertex $i$ reached its radius. Iterate until all vertices reach their radii.
\end{enumerate}

Although simple and clear, the above algorithm requires  $\Omega(n^2)$ space, 
since there are $n$ vertices and each vertex requires $\Omega(n)$ space. 
This is prohibitive in practice. We describe the ANF algorithm
which serves as the basis of our investigation into the structure of real-world networks
\cite{DBLP:conf/kdd/PalmerGF02}. 
We use the  Flajolet-Martin algorithm \cite{flajolet85probabilistic,DBLP:conf/kdd/PalmerGF02}
for counting the number of distinct elements in a multiset.
The main idea of the Flajolet-Martin algorithm is as follows.
We maintain a bitstring $\textit{BITMAP}[0 \dots L-1]$ of length $L$ which encodes the set.
For each item we add, we do the following:

\begin{enumerate}
  \item Pick an $index \in [0 \dots L-1]$ with probability $1/2^{index+1}$.
  \item Set $\textit{BITMAP}[index]$ to 1.
\end{enumerate}

Let $R$ denote the index of the leftmost `0' bit in $\textit{BITMAP}$.
It is clear that $2^R$ should be a good estimate of the number of distinct elements. 
The main result of Flajolet-Martin is that the unbiased estimate of the size of the set is given by

\begin{equation}
\frac{1}{\varphi} 2^R
\end{equation}

\noindent where $\varphi=0.77351\cdots$. 
A more concentrated estimate can be obtained by using multiple bitstrings and averaging the $R$. 
If we use $K$ bitstrings $R_1$ to $R_K$, the size of the set can be estimated by

\begin{equation}
\frac{1}{\varphi} 2^{
\frac{1}{K}\sum_{l=1}^{K}
R_l
}
\end{equation}

\noindent 
The application of the Flajolet-Martin algorithm to radius and diameter estimation is straight-forward.
We maintain $K$ Flajolet-Martin (FM) bitstrings
$b(h,i)$ for each vertex $i$ and the current hop number $h$.
$b(h,i)$ encodes the number of vertices reachable from vertex $i$ within $h$ hops, and
can be used to estimate radii and diameter as shown below.
The bitstrings $b(h,i)$ are iteratively updated
until the bitstrings of all vertices stabilize.
At the $h$-th iteration,
each vertex receives the bitstrings of its neighboring vertices,
and updates its own bitstrings $b(h-1,i)$
handed over from the previous iteration:

\begin{equation}
b(h,i) = b(h-1,i) \text{ BIT-OR } \{b(h-1,j)|(i,j) \in E\}
\label{Nhupdate}
\end{equation}

\noindent where ``BIT-OR'' denotes bitwise-OR function.
After $h$ iterations, a vertex $i$ has $K$ bitstrings that encode
the {\em neighborhood function} $N(h, i)$, that is,
the number of vertices within $h$ hops from the vertex $i$.
$N(h,i)$ is estimated from the $K$ bitstrings by

\begin{equation}
N(h,i) = \frac{1}{0.77351} 2^{\frac{1}{K}\sum_{l=1}^{K}
b_l(i)} 
\label{Nhi}
\end{equation}

where $b_l(i)$ is the position of leftmost `0' bit
of the $l^{th}$ bitstring of vertex $i$.
The iterations continue until the bitstrings of all vertices stabilize,
which is a necessary condition that the current iteration number $h$ exceeds
the diameter $d(G)$.
After the iterations finish at $h_{max}$,
we calculate the effective radius for every node
and the diameter of the graph, as follows:

\squishlist
  \item $r_{eff}(i)$ is the smallest $h$ such that $N(h,i) \geq 0.9 \cdot N(h_{max}, i)$.
  \item $d_{eff}(G)$ is the smallest $h$ such that $N(h)=\sum_{i}N(h,i) = 0.9 \cdot N(h_{max})$. If $N(h) > 0.9 \cdot N(h_{max}) > N(h-1)$, then $d_{eff}(G)$ is linearly interpolated from $N(h)$ and $N(h-1)$. That is, $d_{eff}(G) = (h-1) + \frac{0.9\cdot N(h_{max}) - N(h-1)}{N(h)-N(h-1)}$.
\squishend

The parameter $K$ is typically set to 32 \cite{flajolet85probabilistic}, 
and $MaxIter$ is set to $256$ since real graphs 
have relatively small effective diameter.

\section{\hadoop Algorithms} 
\label{sec:hadimeth}

In this Section we present a way to implement ANF \cite{DBLP:conf/kdd/PalmerGF02} 
on the top of both a \mapreduce system and a parallel SQL DBMS.
Our algorithm is  \hadi, a parallel radius and diameter estimation algorithm. 
It is important to notice that \hadi is a disk-based algorithm.
\hadi saves two pieces of information to a distributed file system (such as HDFS (Hadoop Distributed File System) in the case of \hadoop):

\begin{itemize}
  \item \textbf{Edge} has a format of ($src id$, $dst id$).
  \item \textbf{Bitstrings} has a format of ($node id$, $bitstring_1$, ..., $bitstring_K$).
\end{itemize}

Section~\ref{subsec:hadinaive}  presents \hadi-naive
which gives the big picture and explains why this kind 
of an implementation should not be used in practice. 
Section~\ref{subsec:hadiplain}  presents \hadi-plain,
a significantly improved implementation. 
Section~\ref{subsec:hadioptimized} presents the optimized version of \hadi-optimized,
which scales almost linearly as a function of the number of machines allocated. 
Section~\ref{subsec:hadianalysis} presents the space and time complexity of \hadi. 
Finally, Section~\ref{subsec:hadisql} shows how \hadi can run on the top of a relational
database management system (RDBMS). 

\subsection{ HADI-naive in \mapreduce}
\label{subsec:hadinaive}

\spara{Data:} The edge file is saved as a sparse adjacency matrix in HDFS.
Each line of the file contains a nonzero element of the adjacency matrix
of the graph, in the format of ($src id$, $dst id$).
Also, the bitstrings of each node are saved in a file in the format of ($node id$, $flag$, $bitstring_1$, ..., $bitstring_K$).
The $flag$ variable records  whether a bitstring changed or not. 

\textbf{Main Program Flow}
The main idea of \hadi-naive is to use the bitstrings file as a logical ``cache'' to machines which contain edge files.
The bitstring update operation in Equation~\eqref{Nhupdate} of Section~\ref{sec:hadimoments} requires 
that the machine which updates the bitstrings of node $i$ should have access to (a) all edges adjacent 
from $i$, and (b) all bitstrings of the adjacent nodes.
To meet the requirement (a), it is needed to reorganize the edge file so that edges with a same source 
id are grouped together. That can be done by using an identity mapper which outputs the given input 
edges in ($src id$, $dst id$) format.
The most simple yet naive way to meet the requirement (b) is sending the bitstrings to every reducer 
which receives the reorganized edge file.

Thus, \hadi-naive iterates over two-stages of \mapreduce.
The first stage updates the bitstrings of each node
and sets the `Changed' flag
if at least one of the bitstrings of the node
is different from the previous bitstring.
The second stage counts the number of changed vertex and
stops iterations when the bitstrings stabilized, as illustrated in the swim-lane diagram of Figure~\ref{fig:naiveANF}.

\begin{figure}
	\centering
        \includegraphics[width=1.0\textwidth]{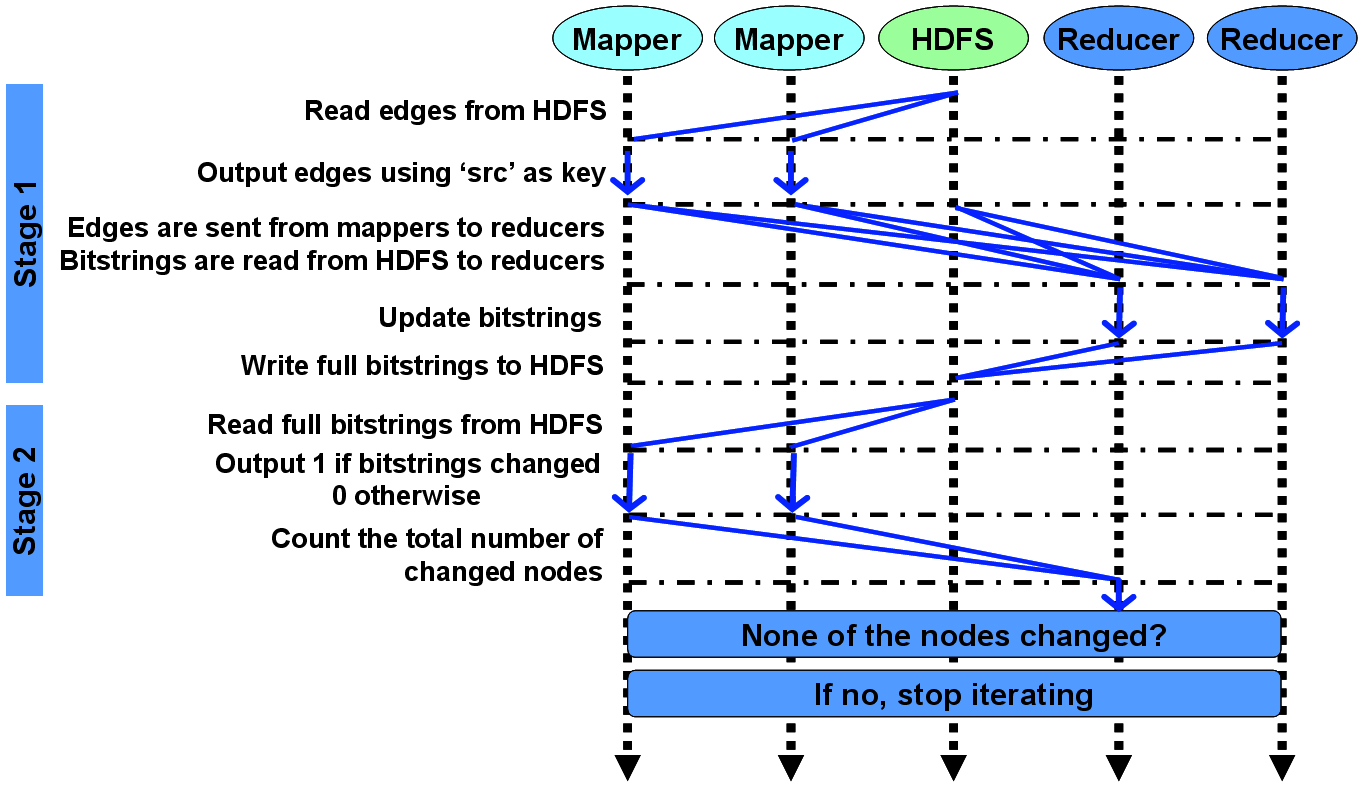}
	\caption{\label{fig:naiveANF} One iteration of \hadi-naive.
	\emph{Stage 1}. Bitstrings of all vertex are sent to every reducer.
    \emph{Stage 2}. Sums up the count of changed nodes.
}
\end{figure}

Although conceptually simple and clear, \hadi-naive is
unnecessarily expensive, because it ships
all the bitstrings to all reducers.

\subsection{ HADI-plain in \mapreduce}
\label{subsec:hadiplain} 

\hadi-plain improves \hadi-naive by
\textit{copying only the necessary bitstrings to each reducer}.
The details follow. 

\spara{Data:} As in \hadi-naive, the edges are saved in the format of
($src id$, $dst id$), and bitstrings are saved in the format of
($node id$, $flag$, $bitstring_1$, ..., $bitstring_K$) in files over HDFS.
The initial bitstrings generation  can be performed in completely parallel way.
The $flag$ of each node records the following information:
\begin{itemize}
  \item {\it Effective Radii} and {\it Hop Numbers} to calculate the effective radius.
  \item {\it Changed} flag to indicate whether at least a bitstring has been changed or not.
\end{itemize}

\spara{Main Program Flow: }
As mentioned in the beginning, \hadi-plain copies only the necessary bitstrings to each reducer.
The main idea is to replicate bitstrings of node $j$ exactly $x$ times where $x$ is the in-degree of node $j$.
The replicated bitstrings of node $j$ is called the {\em partial bitstring}
and represented by $\hat{b}(h,j)$.
The replicated $\hat{b}(h,j)$'s are used to update $b(h,i)$, the bitstring of node $i$ where $(i,j)$ is an edge in the graph.
\hadi-plain iteratively runs three-stage \mapreduce jobs
until all bitstrings of all vertex stop changing.
Algorithm~\ref{alg:HADIp1}, \ref{alg:HADIp2}, and \ref{alg:HADIp3}
shows \hadi-plain, and Figure~\ref{fig:hadi-plain} shows the swim-lane.
We use  $h$ for denoting the current iteration number which starts from $h$=1.
Output($a$,$b$) means to output a pair of data with the key $a$ and the value $b$.

\begin{figure}
	\centering
        \includegraphics[width=1.0\textwidth]{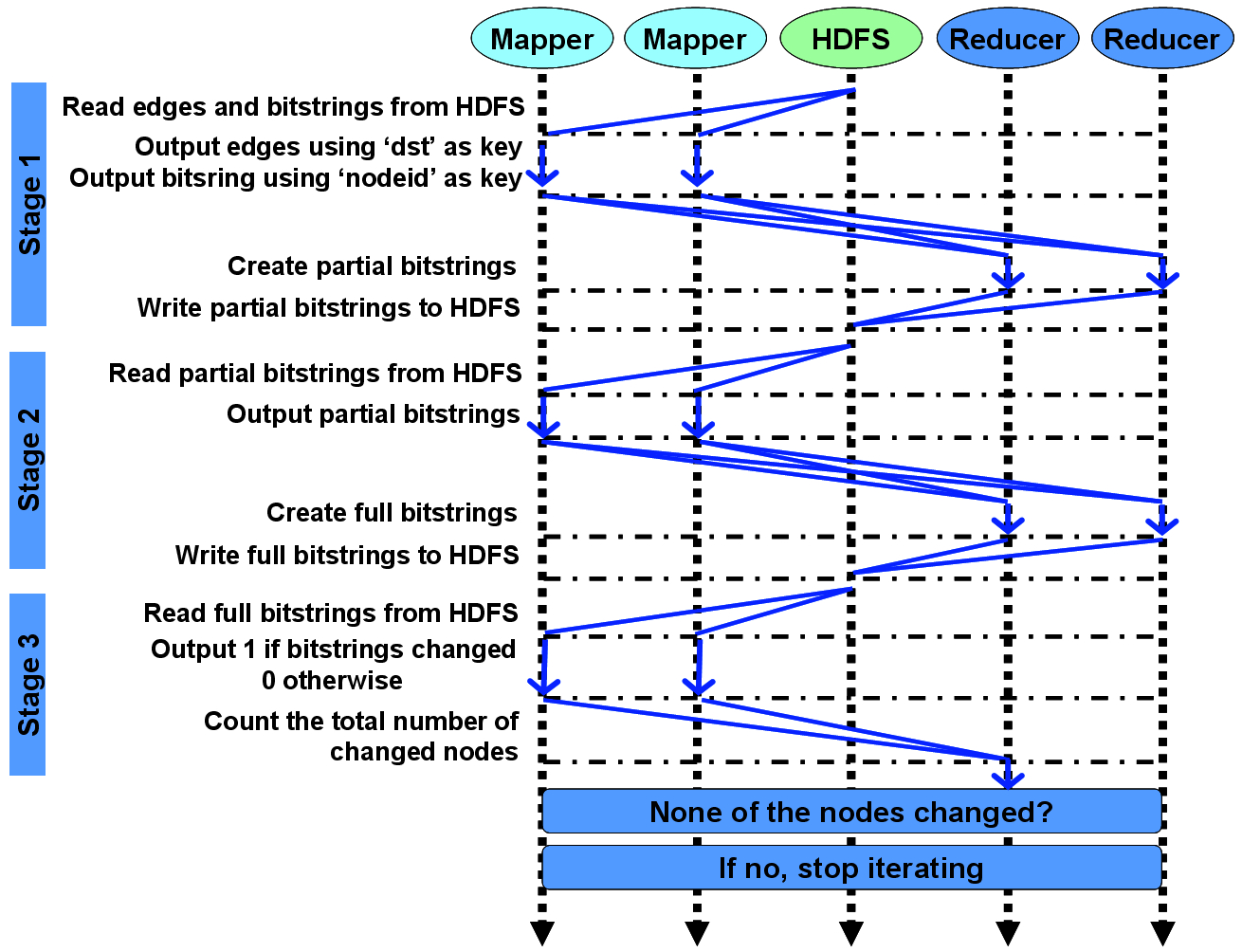}
	\caption{One iteration of \hadi-plain.
	\emph{Stage 1}. Edges and bitstrings are matched to create partial bitstrings.
    \emph{Stage 2}. Partial bitstrings are merged to create full bitstrings.
    \emph{Stage 3}. Sums up the count of changed nodes, and compute N(h), the neighborhood function. Computing N(h) is not drawn in the figure for clarity.
}
	\label{fig:hadi-plain}	
\end{figure}

\begin{algorithm}[!t]
\begin{algorithmic}[1]
\REQUIRE Edge data $E=\{(i, j)\}$ 
\REQUIRE Current bitstring $B=\{(i, b(h-1, i))\}$ or Bitstring Creation Command $BC=\{(i, cmd)\}$
\STATE Stage1-Map(key $k$, value $v$)
\IF{($k,v$) is of type B or BC}
    \STATE Output($k,v$)
\ELSIF{($k,v$) is of type E }
    \STATE Output($v,k$)
\ENDIF
\STATE Stage1-Reduce(key $k$, values $V$[])
\STATE SRC $\leftarrow$ [] 
\FOR{$v \in V$} 
\IF{($k,v$) is of type   BC}
    \STATE  $\hat{b}(h-1,k) \leftarrow $NewFMBitstring()
\ELSIF{ ($k,v$) is of type B }
    \STATE $\hat{b}(h-1,k) \leftarrow v$
\ELSIF{      ($k,v$) is of type E }
    \STATE Add $v$ to $SRC$
\ENDIF
\ENDFOR 
\FOR{$src \in SRC$} 
\STATE Output$(src, \hat{b}(h-1,k))$
\ENDFOR 
\STATE Output$(k, \hat{b}(h-1,k))$
\end{algorithmic}
\caption{\label{alg:HADIp1} \hadi Stage 1. Output is a set of partial bitstrings $B'=\{(i, b(h-1,j))\}$} 
\end{algorithm}

\begin{algorithm}[!t]
\begin{algorithmic}[1]
\REQUIRE Partial bitstring $B=\{(i, \hat{b}(h-1,j)\}$
\STATE Stage2-Map(key $k$, value $v$)   \\
\COMMENT{ Identity Mapper } \\  
\STATE Output($k,v$) 

\STATE Stage2-Reduce(key $k$, values $V$[])
\STATE $b(h,k) \leftarrow $ 0 
\FOR{$v \in V$} 
\STATE $b(h,k) \leftarrow b(h,k)$ BIT-OR $v$
\ENDFOR 
\STATE Update $flag$ of $b(h,k)$
\STATE Output$(k, b(h,k))$
\IF{($k,v$) is of type   BC}
    \STATE  $\hat{b}(h-1,k) \leftarrow $NewFMBitstring()
\ELSIF{ ($k,v$) is of type B }
    \STATE $\hat{b}(h-1,k) \leftarrow v$
\ELSIF{      ($k,v$) is of type E }
    \STATE Add $v$ to $SRC$
\ENDIF
\FOR{$src \in SRC$} 
\STATE Output$(src, \hat{b}(h-1,k))$
\ENDFOR 
\STATE Output$(k, \hat{b}(h-1,k))$
\end{algorithmic}
\caption{\label{alg:HADIp2} \hadi Stage 2. Output is a set of full bitstring $B=\{(i, b(h,i)\}$} 
\end{algorithm}

\begin{algorithm}[!t]
\begin{algorithmic}[1]
\REQUIRE Full bitstring $B=\{(i, b(h,i))\}$
\STATE Stage3-Map(key $k$, value $v$) 
\STATE Analyze $v$ to get ($changed$, $N(h,i)$)
\STATE Output($key\_for\_changed$,$changed$)
\STATE Output($key\_for\_neighborhood$, $N(h,i)$) 

\STATE Stage3-Reduce(key $k$, values $V$[])\

\STATE $Changed \leftarrow $ 0 
\STATE $N(h) \leftarrow $ 0
\FOR{$v \in V$} 
  \IF{k is key\_for\_changed }
      \STATE   $Changed \leftarrow Changed + v$ 
  \ELSIF{ k is key\_for\_neighborhood }
      \STATE $N(h) \leftarrow N(h) + v$
  \ENDIF
\ENDFOR
\STATE Output($key\_for\_changed$,$Changed$)
\STATE Output($key\_for\_neighborhood$, $N(h)$) 
\end{algorithmic}
\caption{\label{alg:HADIp3} \hadi Stage 3. Output is the number of changed nodes, Neighborhood $N(h)$} 
\end{algorithm}

\textbf{Stage 1}
We generate (key, value) pairs, where the key is the node id $i$ and
the value is the partial bitstrings $\hat{b}(h,j)$'s
where $j$ ranges over all the neighbors adjacent from node $i$.
To generate such pairs,
the bitstrings of node $j$ are grouped together
with edges whose $dstid$ is $j$.
Notice that at the very first iteration,
bitstrings of vertex do not exist; they have to be generated on the fly,
and we use the {\em Bitstring Creation Command } for that.
The NewFMBitstring() function generates $K$ FM bitstrings \cite{flajolet85probabilistic}.
Notice also that line 19 of Algorithm~\ref{alg:HADIp1} is used to propagate the bitstrings of one's own node.
These bitstrings are compared to the newly updated bitstrings
at Stage 2 to check convergence.

\textbf{Stage 2}
Bitstrings of node $i$ are updated by combining partial bitstrings of itself and vertex adjacent from $i$.
For the purpose, the mapper is the Identity mapper
(output the input without any modification).
The reducer combines them, generates new bitstrings, and sets $flag$ by recording (a) whether 
at least a bitstring changed or not, and (b) the current iteration number $h$ and the neighborhood value $N(h,i)$.
This $h$ and $N(h,i)$ are used to calculate the effective radius of vertex after all bitstrings converge.
Notice that only the last neighborhood $N(h_{last}, i)$ and other neighborhoods $N(h',i)$ that 
satisfy $N(h',i) \geq 0.9 \cdot N(h_{last}, i)$ need to 
be saved to calculate the effective radius.
The output of Stage 2 is fed into the input of Stage 1 at the next iteration.

\textbf{Stage 3}
We calculate the number of changed vertex and
sum up the neighborhood value of all vertex to calculate $N(h)$.
We use only
two unique keys(key\_for\_changed and key\_for\_neighborhood),
which correspond to the two calculated values.
The analysis of line 2 can be done by checking the  $flag$ field
and using Equation~\eqref{Nhi} in Section~\ref{sec:hadimoments}.
The variable $changed$ is set to $1$ or $0$, based on whether the bitmask of node $k$ changed or not.

When all bitstrings of all vertex converge, 
a \mapreduce job to finalize the effective radius and 
diameter is performed and the program finishes. 
Compared to \hadi-naive, the advantage of \hadi-plain is clear: 
bitstrings and edges are evenly distributed over machines so that the 
algorithm can handle as much data as possible, given sufficiently many machines.

\subsection{HADI-optimized in \mapreduce}
\label{subsec:hadioptimized}

\hadi-optimized further improves \hadi-plain.
It uses two orthogonal ideas:
``block operation'' and ``bit shuffle encoding''.
Both try to address some subtle performance issues.
Specifically,
\hadoop\ has
the following two
major bottlenecks:
\begin{itemize}
  \item Materialization: at the end of each map/reduce stage,
  the output is written to the disk,
  and it is also read at the beginning of next reduce/map stage.
  \item Sorting: at the \textit{Shuffle} stage,
  data is sent to each reducer and sorted
  before they are handed over to the \textit{Reduce} stage.
\end{itemize}
\hadi-optimized addresses these two issues.

\spara{Block Operation:} Our first optimization is the block encoding of the edges and the bitstrings.
The main idea is to group $w$ by $w$ sub-matrix into a super-element in the adjacency matrix E, and group $w$ bitstrings into a super-bitstring.
Now, \hadi-plain is performed on these super-elements and super-bitstrings, instead of the original edges and bitstrings.
Of course, appropriate decoding and encoding is necessary at each stage.
Figure~\ref{fig:blockoperation} shows an example of converting data into block-format.

\begin{figure}[h]
\begin{center}
  \includegraphics[width=0.6\textwidth]{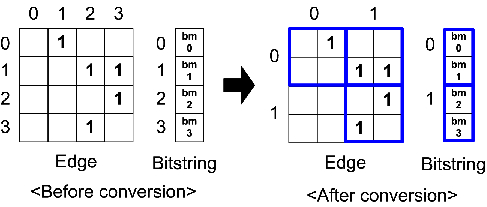}
  \caption{ \label{fig:blockoperation}
  Converting the original edge and bitstring to blocks.
  The 4-by-4 edge and length-4 bitstring are converted to 2-by-2 super-elements and length-2 super-bitstrings. Notice the lower-left super-element of the edge is not produced since there is no nonzero element inside it.
  }
\end{center}
\end{figure}

By this block operation, the performance of \hadi-plain changes as follows:
\begin{itemize}
  \item \textit{Input size} decreases in general, since we can use fewer bits to index elements inside a block.
  \item \textit{Sorting time} decreases, since the number of elements to sort decreases.
  \item \textit{Network traffic} decreases since the result of matching a super-element and a super-bitstring is a bitstring which can be at maximum $block\_width$ times smaller than that of \hadi-plain.
  \item \textit{Map and Reduce functions} take more time, since the block must be decoded to be processed, and be encoded back to block format.
\end{itemize}
For reasonable-size blocks, the performance improvement are significant.

\spara{Bit Shuffle Encoding: } In our effort to decrease the input size,
we propose an encoding scheme that can compress the bitstrings.
Recall that in \hadi-plain, we use $K$ (e.g., 32, 64) bitstrings
for each node, to increase the accuracy of our estimator.
Since \hadi requires $O(K(m+n) \log n)$ space,
the amount of data increases
when $K$ is large.
For example, the YahooWeb graph spans 120 GBytes
(with 1.4 billion nodes, 6.6 billion edges).
However the required disk space for just the bitstrings is
$32 \cdot (1.4B + 6.6B) \cdot 8$ byte = $2$ Tera bytes (assuming 8 byte for each bitstring),
which is more than 16 times larger than the input graph.

The main idea of Bit Shuffle Encoding is to carefully reorder
the bits of the bitstrings of each node, and then use Run Length Encoding.
By construction, the leftmost part of each bitstring is  almost full of
one's, and the rest is almost full of zeros.
Specifically, we make the reordered bit strings to contain
long sequences of 1's and 0's:
we get all the first bits from all $K$ bitstrings, then
get the second bits, and so on.
As a result we get a single bit-sequence of length $K \cdot |bitstring|$,
where most of the first bits are `1's, and most of the last bits are `0's.
Then we encode only the length of each bit sequence, achieving
good space savings (and, eventually, time savings, through fewer I/Os).

\subsection{Analysis} 
\label{subsec:hadianalysis} 

\hadi depends on four parameters: the number $M$ of machines, 
the number of vertices $n$ and edges $m$, and the diameter $d$. 
The time complexity is dominated by the shorting time needed for the shuffling
during Stage1. Specifically, \hadi takes $O(\frac{m+n}{M}\log\frac{m+n}{M})$ time 
to run. It requires  $O((m+n) \log n)$ space units and is required 
$O((m+n) \log n)$.

\subsection{\hadi in SQL}
\label{subsec:hadisql}

Using relational database management systems (RDBMS) for graph mining is a promising research direction,
especially given the findings of \cite{pavlo09}.
We mention that
\hadi\ can be implemented on the top of an Object-Relational DBMS
(parallel or serial):
it needs repeated joins of the edge table with the appropriate table of
bit-strings, and a user-defined function for bit-OR-ing.
We sketch a potential implementation of \hadi in a RDBMS.

\spara{Data:} In parallel RDBMS implementations, data is saved in tables.
The edges are saved in the table $E$ with attributes $src$ (source node id) and $dst$ (destination node id). Similarly, the bitstrings are saved in the table $B$ with 

\spara{Main Program Flow:} The main flow comprises iterative execution of SQL statements with appropriate user defined functions.
The most important and expensive operation is updating the bitstrings of nodes. 
Observe that the operation can be concisely expressed as a SQL statement:

\begin{table}[h]
        \centering
        \begin{tabular}{|l|}\hline
SELECT INTO B\_NEW E.src, BIT-OR(B.b)  \\
~ FROM E, B  \\
~ WHERE E.dst=B.id \\
~ GROUP BY E.src \\
\hline
        \end{tabular}
        \label{tab:hadi_rdbms}
\end{table}

The SQL statement requires BIT-OR(), a UDF function that implements the bit-OR-ing of the Flajolet-Martin bitstrings.
The RDBMS implementation iteratively runs the SQL until $B\_NEW$ is same as $B$. $B\_NEW$ created at an iteration is used as 
$B$ at the next iteration.

\section{Wall-clock Times} 
\label{sec:haditimings}

In this section, we perform experiments to answer the following questions:
\squishlist
  \item Q1: How fast is \hadi?
  \item Q2: How does it scale up with the graph size and the number of machines?
  \item Q3: How do the optimizations help performance?
\squishend

\subsection{Experimental Setup}

We use both real and synthetic graphs in our experiments. 
These datasets are shown in Table~\ref{tab:datasummary}.

\squishlist
  \item YahooWeb: web pages and their hypertext links
  indexed by Yahoo! Altavista search engine in 2002.
  \item Patents: U.S. patents, citing each other (from 1975 to 1999).
  \item LinkedIn: people connected to other people (from 2003 to 2006).
  \item Kronecker: Synthetic Kronecker graphs~\cite{Leskovec05Realistic} using a chain of length two as the seed graph.
\squishend

\begin{table}[h]
  \begin{center}
  {\small
  \begin{tabular}{|r||r|r|r|l| } \hline 
    \textbf{Graph} & \textbf{Nodes} & \textbf{Edges} & \textbf{File} & \textbf{Description} \\ \hline \hline
    YahooWeb   & 1.4 B & 6.6 B & 116G & page-page \\ \hline
    LinkedIn & 7.5 M & 58 M & 1G & person-person\\ \hline
    Patents & 6 M & 16 M & 264M & patent-patent\\
    \hline
    Kronecker
      & 177 K & 1,977 M & 25G & synthetic\\
      & 120 K & 1,145M & 13.9G & \\
      & 59 K & 282 M & 3.3G  &\\
    \hline
    Erd{\H o}s-R\'enyi
       & 177 K & 1,977 M & 25G  & random $G_{n,p}$\\
       & 120 K &  1,145 M & 13.9G &\\
       & 59 K & 282 M & 3.3G &\\
       \hline 
  \end{tabular}
  } 
  \end{center}
  \caption{Datasets (B: Billion, M: Million, K: Thousand, G: Gigabytes)}
  \label{tab:datasummary}
\end{table}

In order to test the scalability of \hadi, we use synthetic graphs,
namely Kronecker and Erd{\H o}s-R\'enyi  graphs. 
\hadi runs on {\em M45}, one of the fifty most powerful supercomputers in the world.
M45 has 480 hosts (each with  2 quad-core Intel Xeon 1.86 GHz, running RHEL5),
with 3Tb aggregate RAM, and over 1.5 Peta-byte disk size.

Finally, we use the following notations to indicate different optimizations of \hadi:
\squishlist
  \item \hadi-BSE: \hadi-plain with bit shuffle encoding.
  \item \hadi-BL: \hadi-plain with block operation.
  \item \hadi-OPT: \hadi-plain with both bit shuffle encoding and block operation.
\squishend

\subsection {Running Time and Scale-up}

\begin{figure}[h]
\begin{center}
  \includegraphics[width=0.6\textwidth]{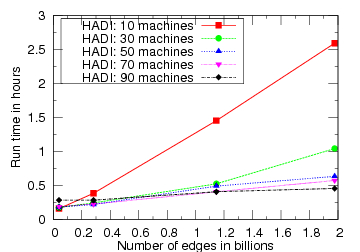}
  \caption{ \label{fig:runtime_edges}
  Running time versus number of edges with \hadi-OPT on Kronecker graphs for three iterations.
  Notice the excellent scalability: linear on the graph size (number of edges).
}
\end{center}
\end{figure}

\begin{figure}[h]
\begin{center}
  \includegraphics[width=0.6\textwidth]{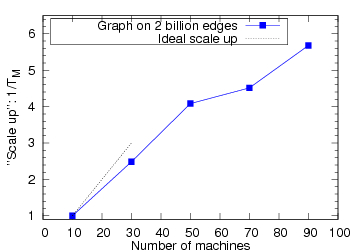}
  \caption{ \label{fig:runtime_machines}
  ``Scale-up'' (throughput $1/T_M$)
  versus number of machines $M$, for the Kronecker graph
  (2B edges).
  Notice the near-linear growth in the beginning,
  close to the ideal(dotted line).
  }
\end{center}
\end{figure}

Figure~\ref{fig:runtime_edges} gives the wall-clock time of \hadi-OPT versus
the number of edges in the graph. Each curve corresponds
to a different number of machines used (from 10 to 90).
\hadi has excellent scalability, with  its running time being
linear on the number of edges.
The rest of the  \hadi versions (\hadi-plain, \hadi-BL, and \hadi-BSE),
were slower, but had a similar linear trend.

Figure~\ref{fig:runtime_machines}
gives the throughput $1/T_M$ of \hadi-OPT.
We also tried \hadi with one machine;
however it didn't complete,
since the machine would take so long that it would often fail in the meanwhile.
For this reason, we do not report the typical scale-up score $s = T_1 / T_M$ (ratio of time with 1 machine, over time with $M$ machine),
and instead we report just the inverse of $T_M$.
\hadi scales up near-linearly with the number
of machines $M$, close to the ideal scale-up.

\begin{figure}[h]
\begin{center}
  \includegraphics[width=0.6\textwidth]{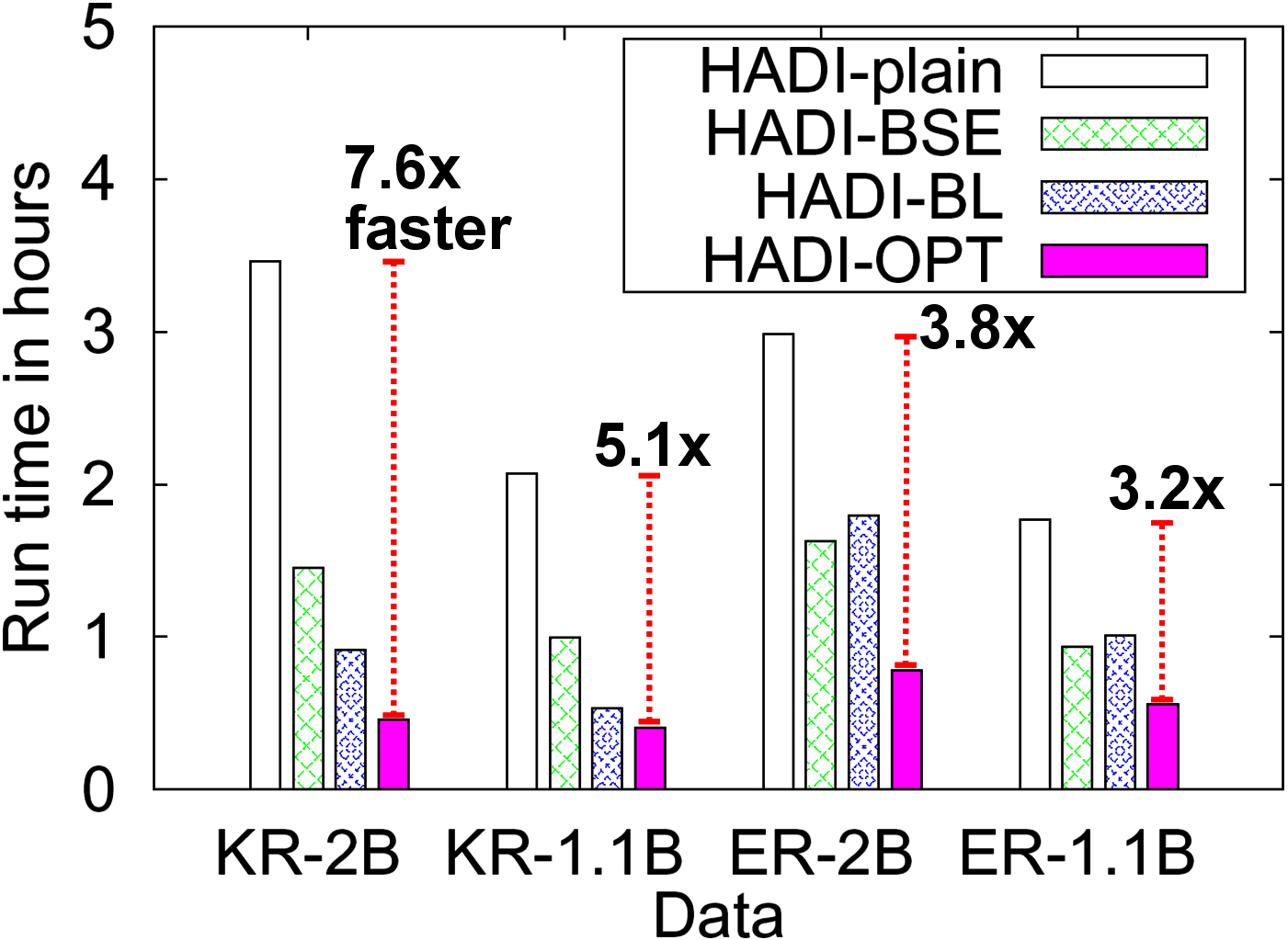}
  \caption{ \label{fig:hadiopteff}
  Run  time of \hadi with/without optimizations for Kronecker and Erd{\H o}s-R\'enyi graphs with 
  several billions of edges, on the M45 \hadoop\ cluster using 90 machines for 3 iterations.
   \hadi-OPT is up to 7.6$\times$ faster than \hadi-plain.
  }
\end{center}
\end{figure}

\subsection{Effect of Optimizations}
Among the optimizations that we mentioned earlier,
which one helps the most, and by how much?
Figure~\ref{fig:hadiopteff}
plots the running time of different graphs versus different \hadi optimizations.
For the Kronecker graphs,
we see that block operation is more efficient than bit shuffle encoding. 
Here, \hadi-OPT achieves 7.6$\times$ better performance than \hadi-plain.
For the Erd{\H o}s-R\'enyi graphs, however, we see that block operations
do not help more than bit shuffle encoding, because the adjacency matrix has no block structure,
while Kronecker graphs do.
Also notice that \hadi-BLK and \hadi-OPT run faster on Kronecker graphs
than on Erd{\H o}s-R\'enyi graphs of the same size. Again, the reason is that
Kronecker graphs have fewer nonzero blocks  (i.e., ``communities'') by their construction,
and the ``block'' operation yields more savings.

\section{Structure of Real-World Networks}
\label{sec:hadifindings} 

\begin{figure}[ht]
    \begin{tabular}{cc}
    \includegraphics[width=0.45\textwidth]{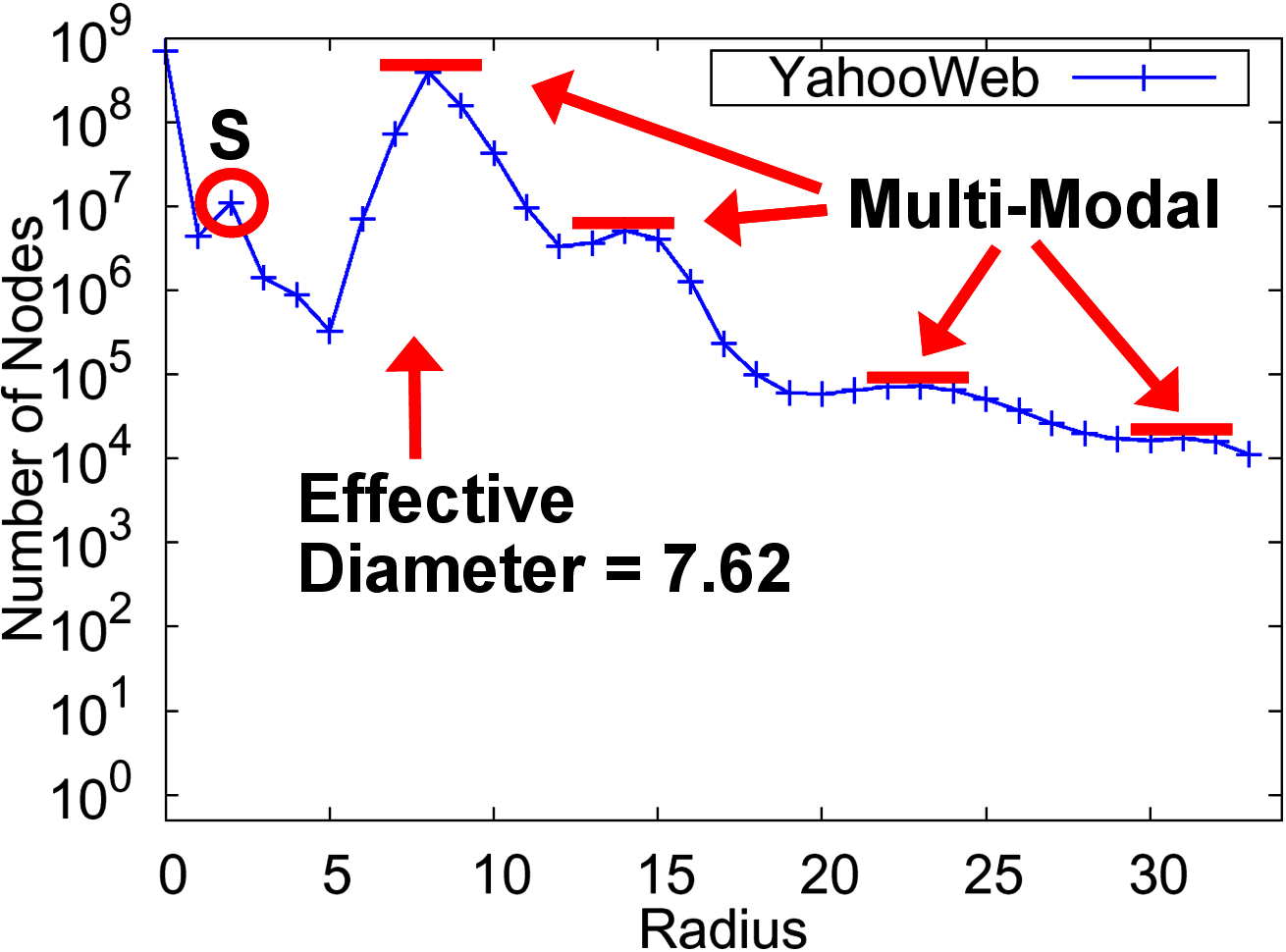}&
    \includegraphics[width=0.45\textwidth]{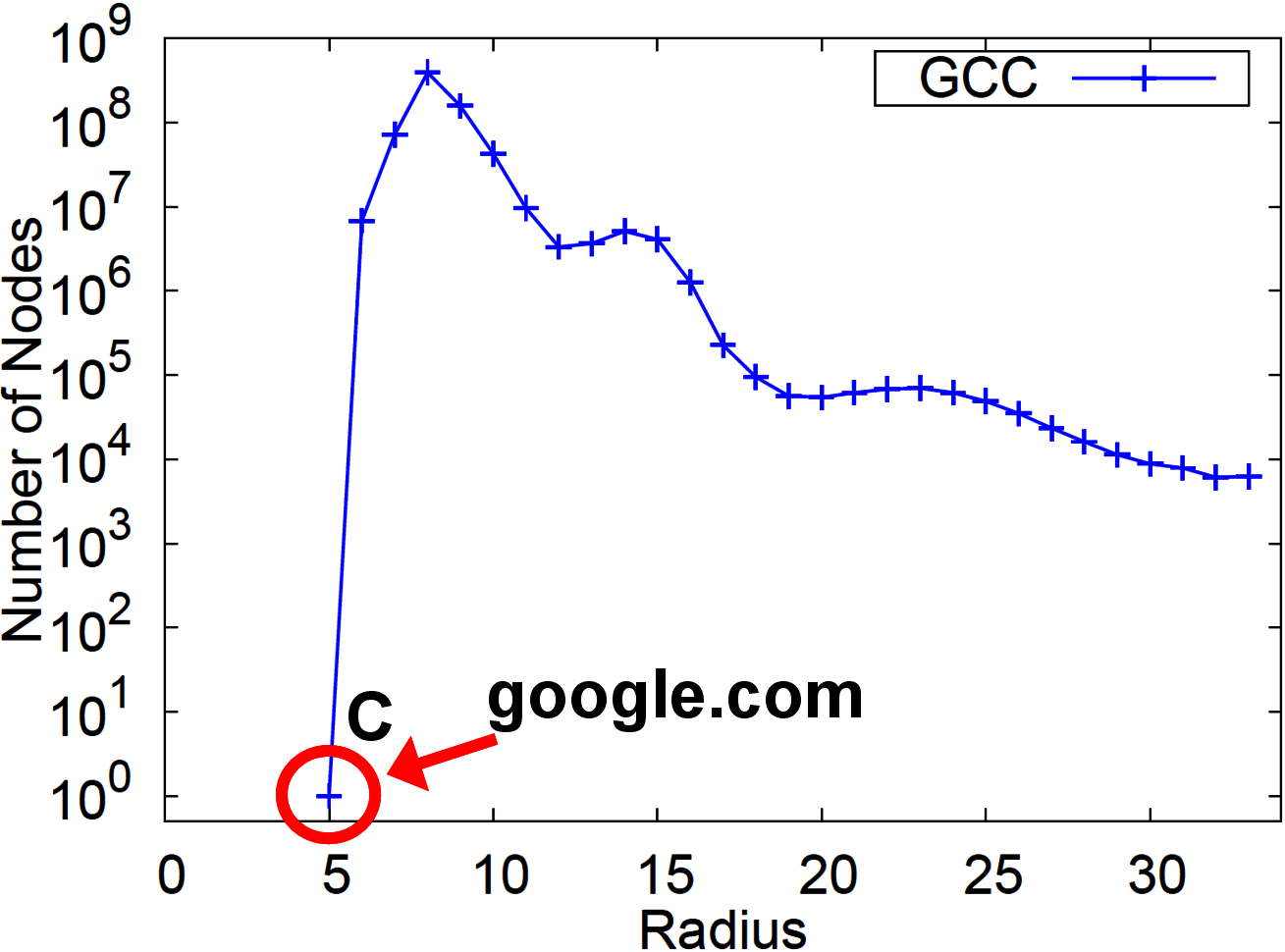}\\
         (a) Radius plot of YahooWeb &  (b) Radius plot of GCC of YahooWeb 
    \end{tabular}
    \caption{\label{fig:hadiyahoo} \textbf{ (a)} Radius plot (Count versus Radius) of the YahooWeb graph. Notice the effective 
    diameter is surprisingly small. Also notice the peak (marked `S') at radius 2, due to star-structured disconnected components. \newline
    \textbf{ (b)} Radius plot of GCC (Giant Connected Component)
    of YahooWeb graph. The {\em only} vertex with radius 5  (marked `C')
    is {\tt google.com}.}
\end{figure}

\hadi reveals new patterns in massive graphs which we present in this section.
We distinguish these new patterns into {\it static}  (Section~\ref{subsec:statichadi}) 
and {\it temporal}  (Section~\ref{subsec:temporalhadi}).

\subsection{Static Patterns}
\label{subsec:statichadi}

\subsubsection{Diameter}
What is the diameter of the Web?
Albert et al.~\cite{Albert99Diameter}
computed the diameter on a directed Web graph with approximately 0.3 million vertices and
conjectured that it should be around 19 for a 1.4 billion-vertex Web graph as shown in the upper line of Figure~\ref{fig:dia_comp}.
Broder et al.~\cite{broder2000graph} used their sampling approach from approximately 200 million-vertices and reported 
16.15 and 6.83 as the diameter for the directed and the undirected  cases, respectively.
What should the effective diameter be, for a significantly larger crawl of the Web, with billions of vertices ?
Figure~\ref{fig:hadiyahoo} gives the surprising answer:

\begin{observation}[Small Web]
\label{observation:small_web}
The effective diameter of the YahooWeb graph
 (year: 2002)
is surprisingly small, between 7 and 8. 
\end{observation}

The previous results from Albert et al. \cite{Albert99Diameter} and Broder et al. \cite{broder2000graph} also consider
the undirected version of the Web graph. We  compute the average diameter and show the comparison of 
diameters of different graphs in Figure~\ref{fig:dia_comp}.
We first observe that the average diameters of all graphs are relatively small  ($<$ 20) for both the directed and the undirected cases.
We also observe that the Albert et al.'s conjecture for the diameter of the directed graph is over-pessimistic: 
both the sampling approach and \hadi output smaller values for the diameter of the directed graph.
For the diameter of the undirected graph, we observe the constant/shrinking diameter pattern~\cite{1217301}.

\begin{figure}[h]
\centering
\includegraphics[width=0.6\textwidth]{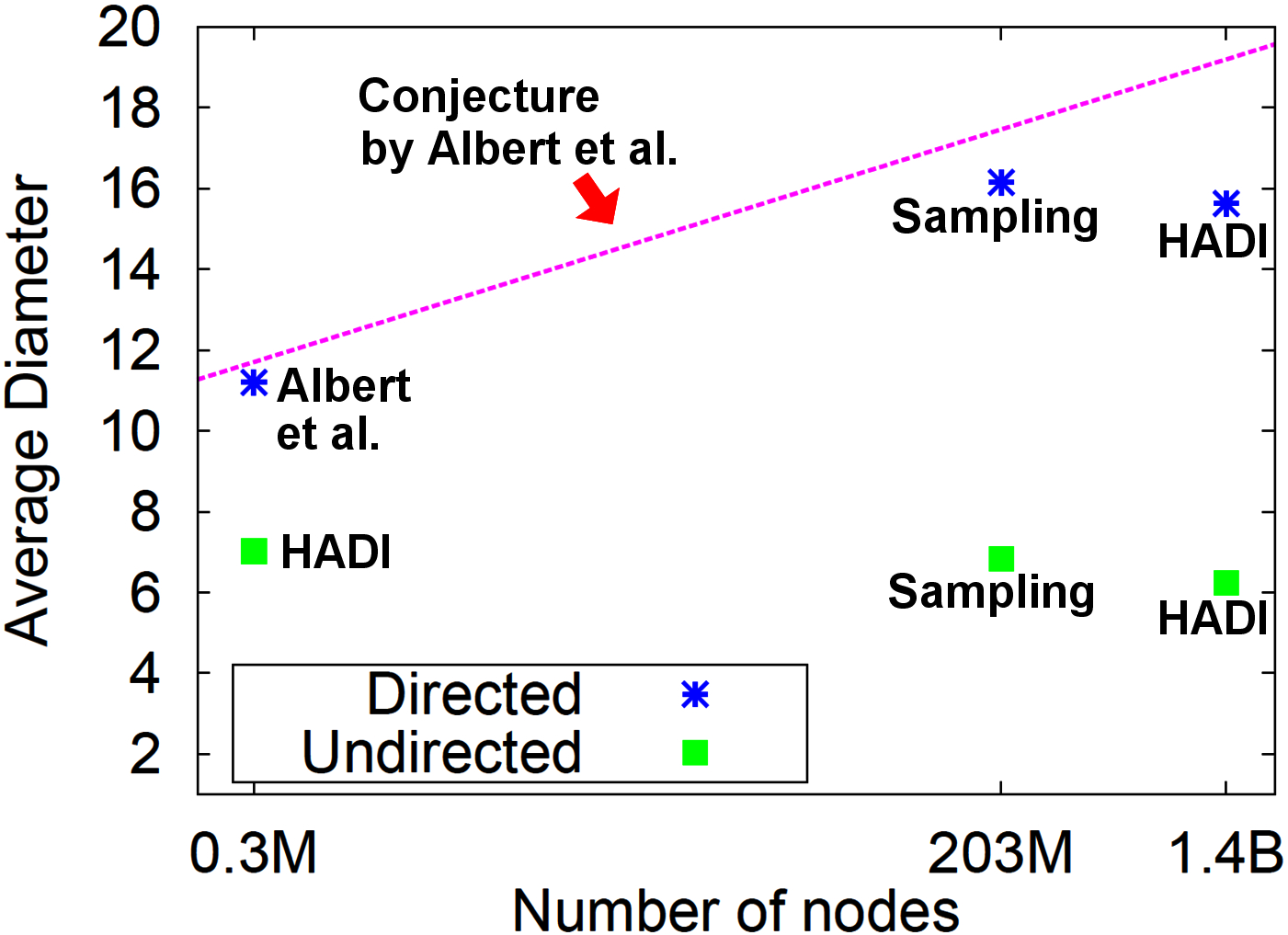}
\caption{\label{fig:dia_comp}
Average diameter vs. number of vertices in lin-log scale for the three different Web graphs, where M and B stand for 
millions and billions respectively.
 (0.3M): web pages inside nd.edu at 1999, from Albert et al.'s work.
 (203M): web pages crawled by Altavista at 1999, from Broder et al.'s work
 (1.4B): web pages crawled by Yahoo at 2002  (YahooWeb in Table~\ref{tab:datasummary}).
Notice the relatively small diameters for both the directed and the undirected cases. 
}
\end{figure}

\subsubsection{Shape of Distribution}

Figure~\ref{fig:hadiyahoo} shows that the radii distribution in the Web Graph is multimodal.
In other relatively smaller networks, we observe a bimodal structure.
As shown in the Radius plot of U.S. Patent and LinkedIn network in
Figure~\ref{fig:static_radius_others}, they have a peak at zero,
a dip at a small radius value  (9, and 4, respectively) and another peak
very close to the dip.

\begin{figure*}[htbp]
    \begin{tabular}{c c}
    \includegraphics[width=0.45\textwidth]{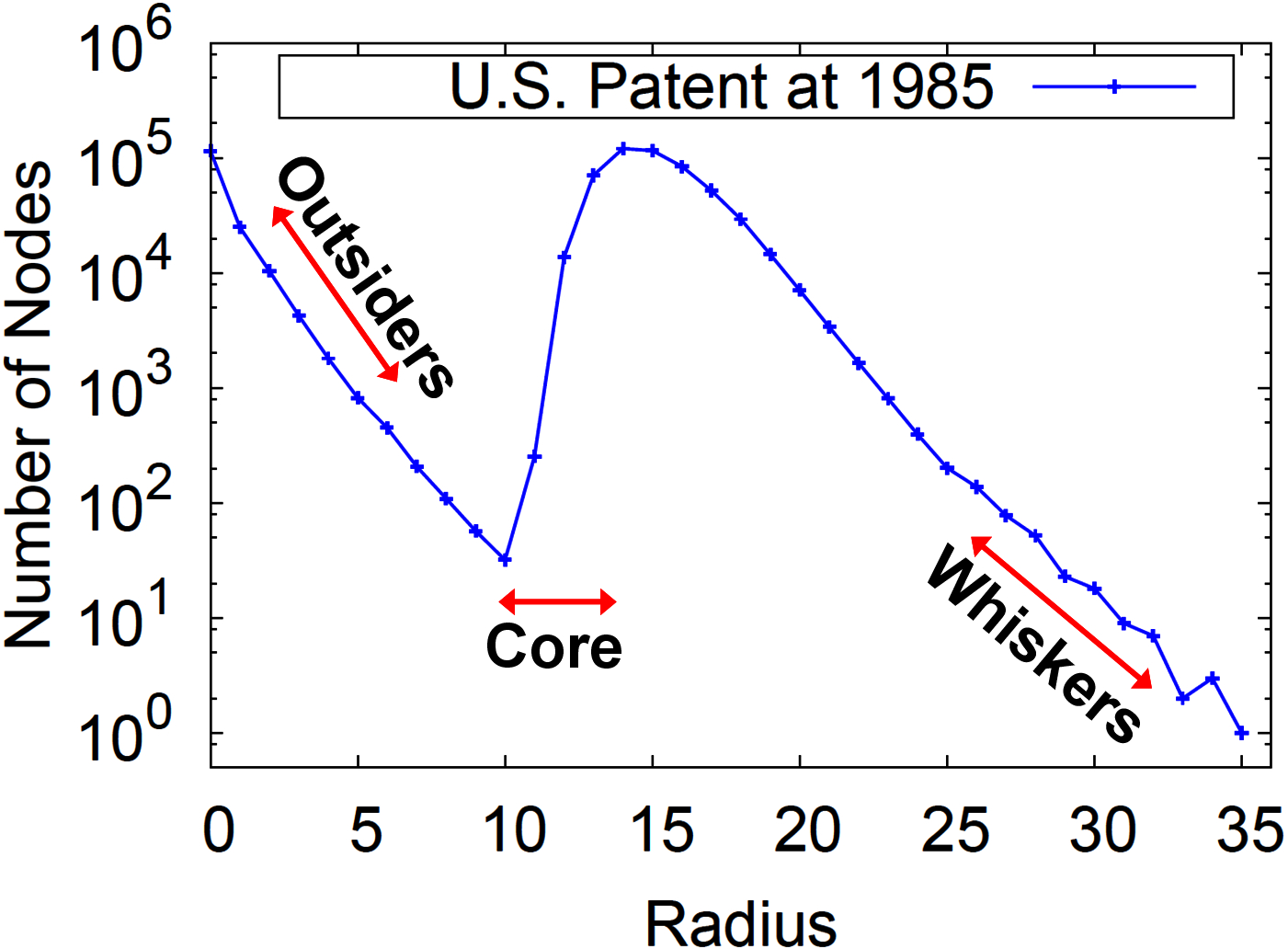} &
    \includegraphics[width=0.45\textwidth]{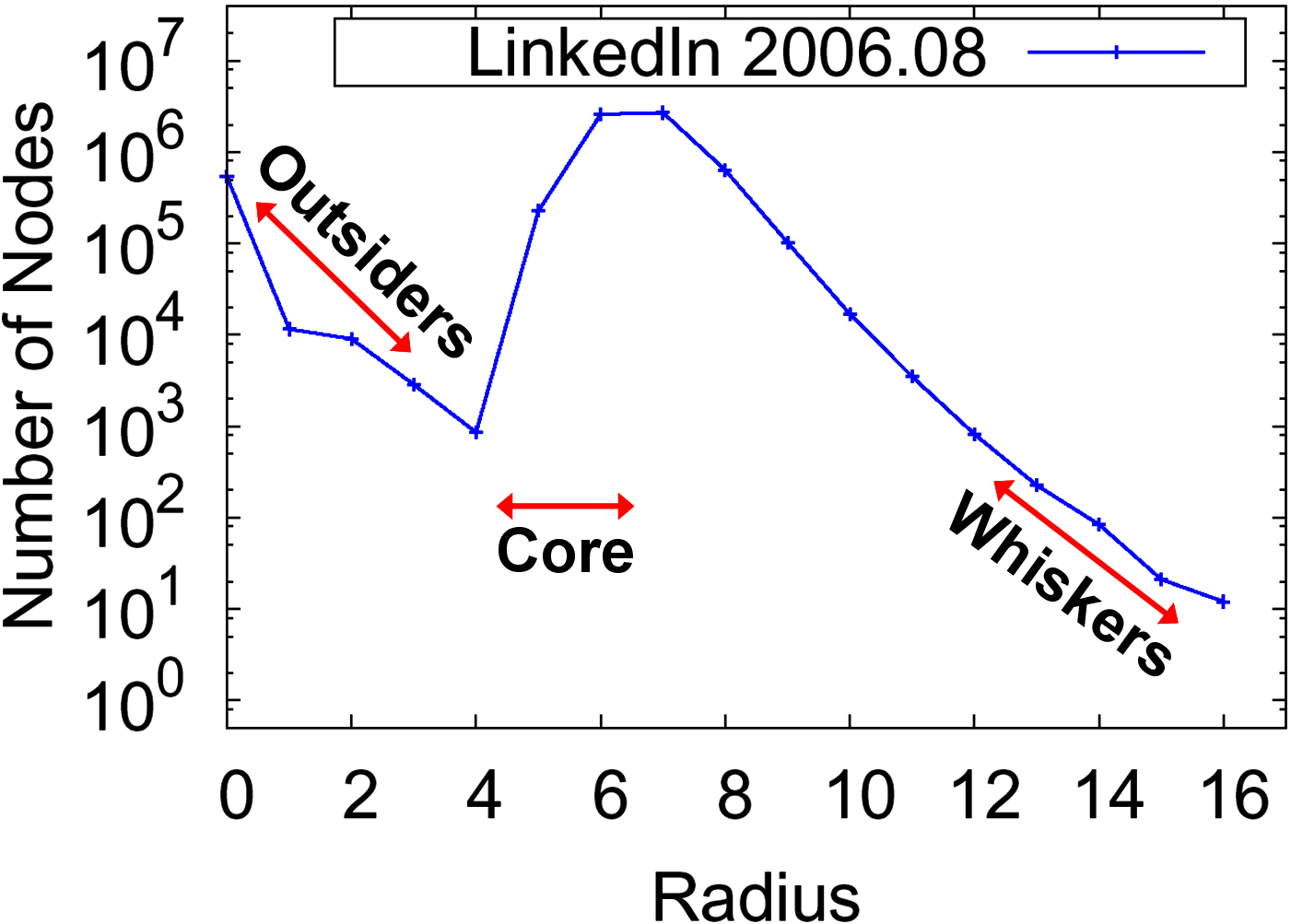}\\
       (a) U.S. Patent &  (b) LinkedIn \\
    \end{tabular}
    \caption{\label{fig:static_radius_others}
    Static Radius Plot (Count versus Radius) of U.S. Patent and LinkedIn graphs. Notice the bimodal structure 
    with `outsiders' (vertices in the DCs), `core' (central vertices in the GCC), and `whiskers'  (vertices connected to the GCC with long paths).
    }
\end{figure*}

\begin{observation}[Multi-modal and Bi-modal]
\label{observation:radiusthreearea}
The radius distribution of the Web graph has a multimodal structure.
Smaller networks have a bimodal structure.
\end{observation}

A natural question to ask with respect to the bimodal structure 
is what are the common properties of the vertices that belong to the first
peak; similarly, for the vertices in the first dip,
and the same for the vertices of the second peak.
After investigation, the former are vertices that belong to disconnected
components  (DCs); vertices in the dip are usually core vertices in the
giant connected component  (GCC), and the vertices at the second
peak are the vast majority of well connected vertices in the GCC.
Figure~\ref{fig:patent_uradius_t1_decom} exactly shows the radii distribution
for the vertices of the GCC  (in blue), and the vertices of the few largest
remaining components.

\begin{figure}[htbp]
\begin{center}
  \includegraphics[width=0.6\textwidth]{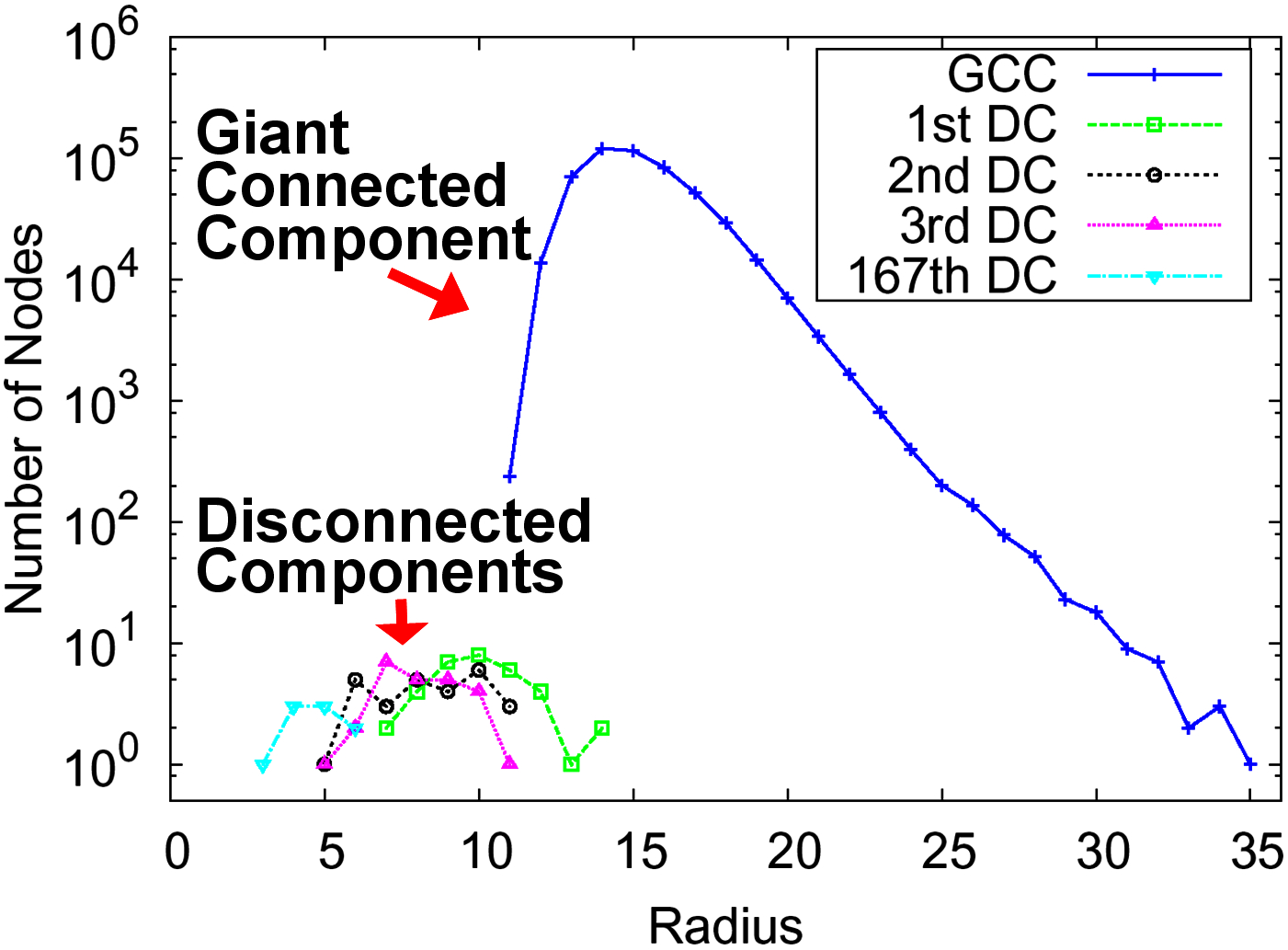}
  \caption{ \label{fig:patent_uradius_t1_decom}
  Radius plot  (Count versus radius)
  for several connected components of the U.S. Patent data in 1985.
  In blue: the distribution for the giant connected component;
  rest colors: several disconnected components.
  }
\end{center}
\end{figure}

    In Figure~\ref{fig:patent_uradius_t1_decom}, 
	we clearly see that the second peak of the bimodal structure
    came from the giant connected component.
    But, where does the first peak around radius 0 come from?
    We can get the answer from the distribution of connected component
    of the same graph in Figure~\ref{fig:static_cc}.
    Since the ranges of radius are limited by the size of connected components,
    we see the first peak of Radius plot came
    from the disconnected components whose size follows a power law.

    \begin{figure*}[htbp]
    \begin{center}
        \begin{tabular}{c c}
        \includegraphics[width=0.45\textwidth]{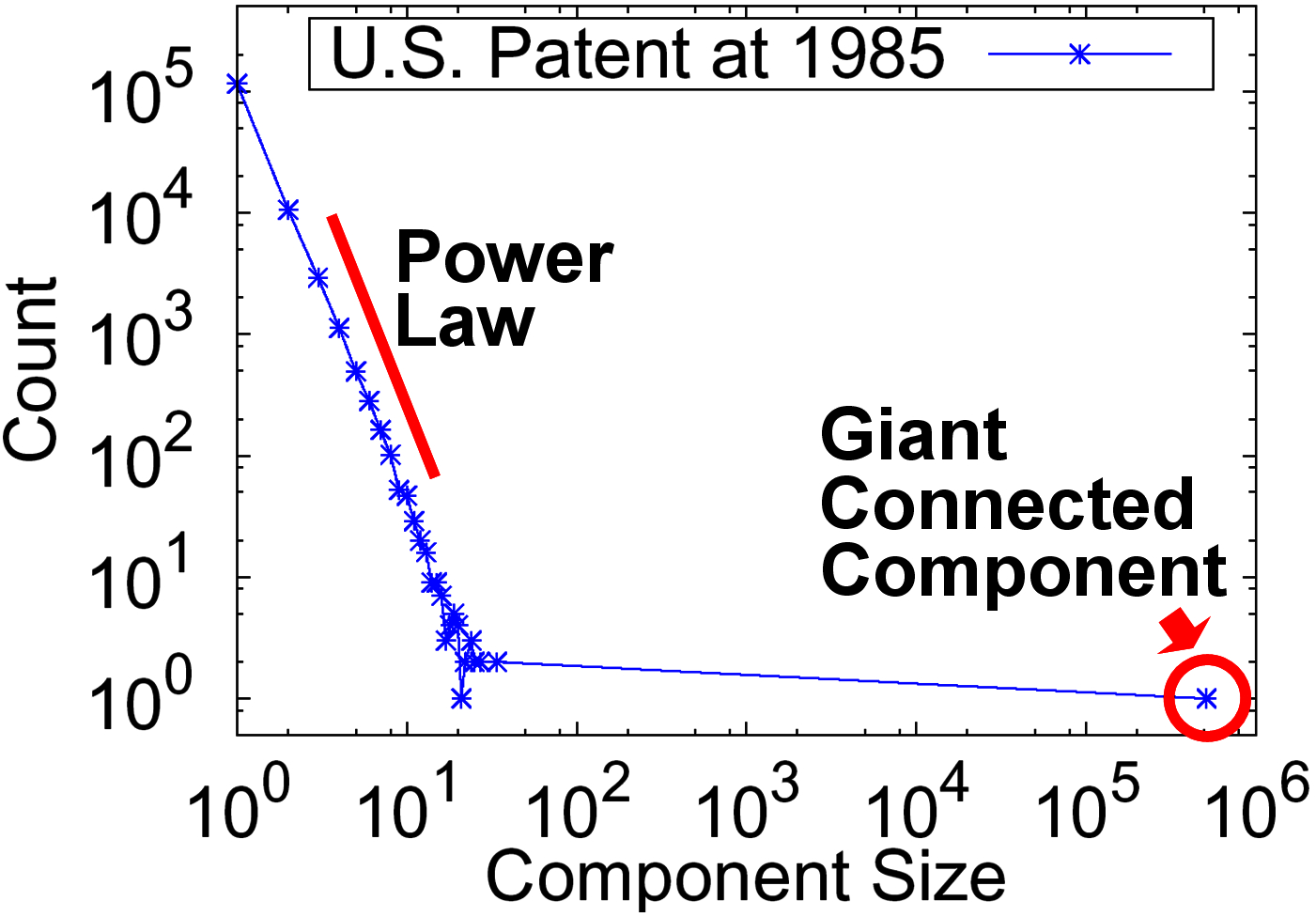}&
        \includegraphics[width=0.45\textwidth]{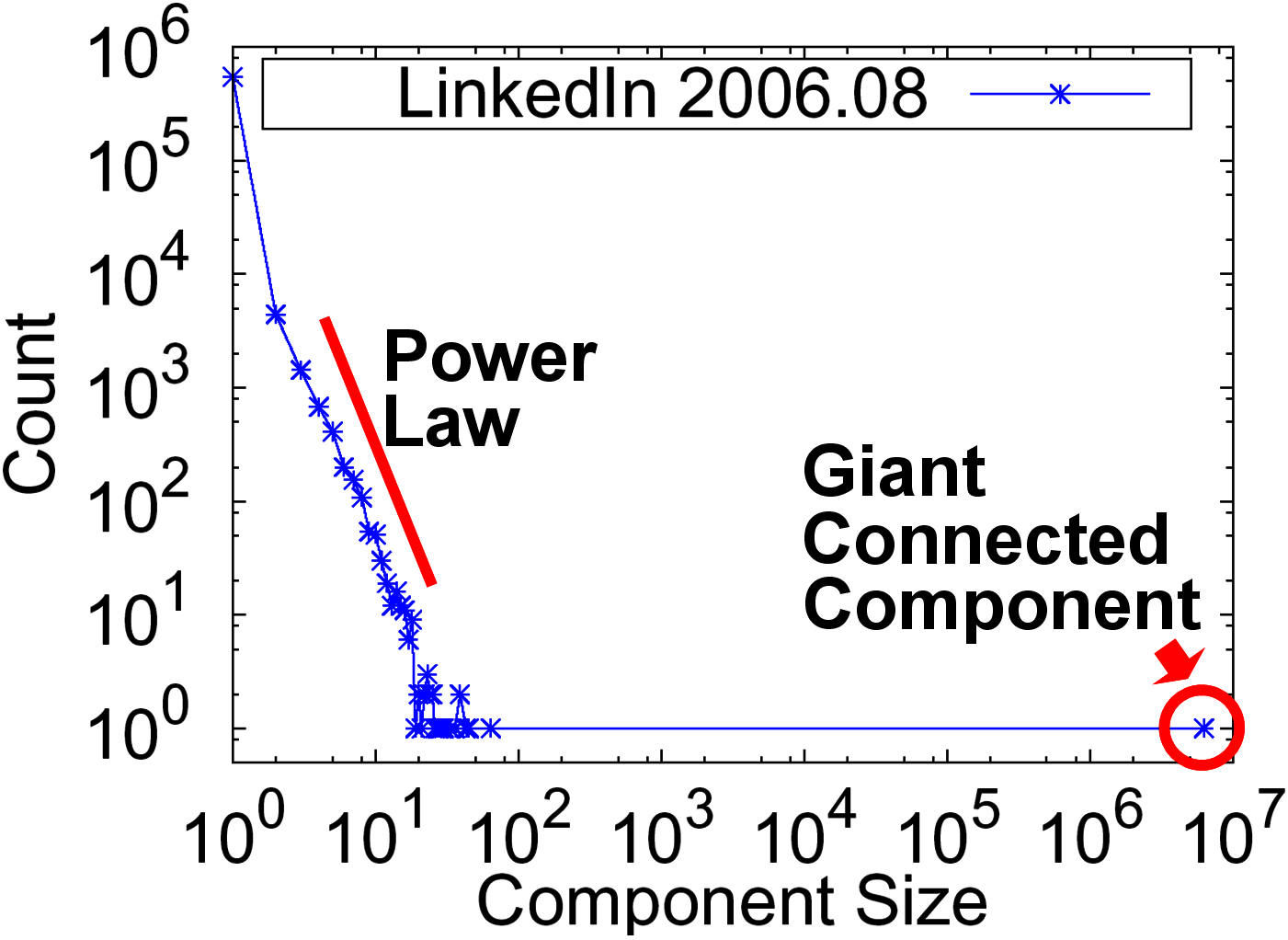}\\
           (a) Patent &  (b) LinkedIn
        \end{tabular}
        \caption{
        Size distribution of connected components. Notice the size of the disconnected components (DCs) follows a power-law which explains the first peak around radius 0 of the radius plots in Figure~\ref{fig:static_radius_others}.
        }
        \label{fig:static_cc}
        \end{center}
    \end{figure*}

\begin{figure*}[htbp]
\begin{center}
\setlength{\tabcolsep}{0cm}
    \begin{tabular}{c c}
    \includegraphics[width=0.45\textwidth]{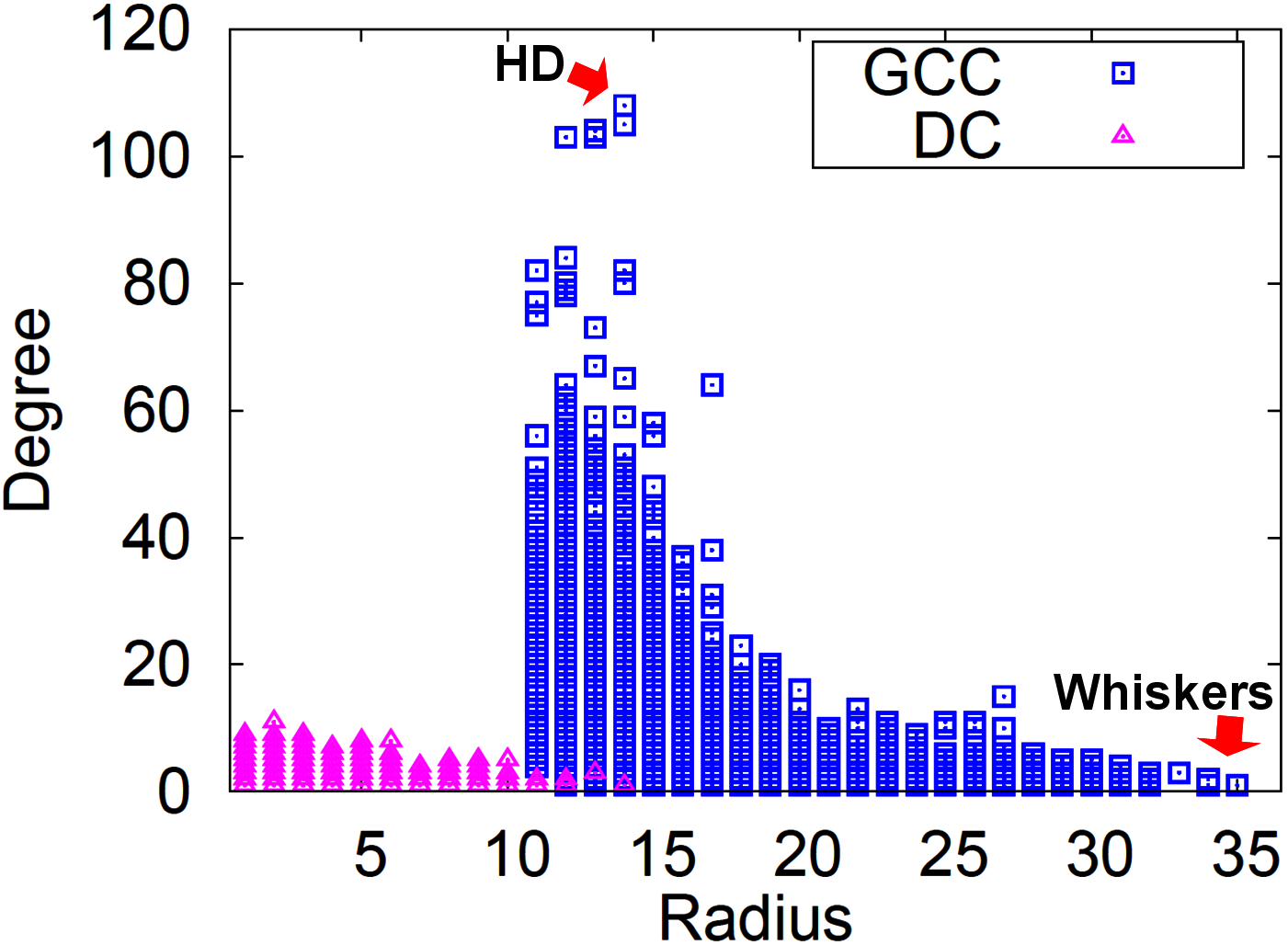}&
    \includegraphics[width=0.45\textwidth]{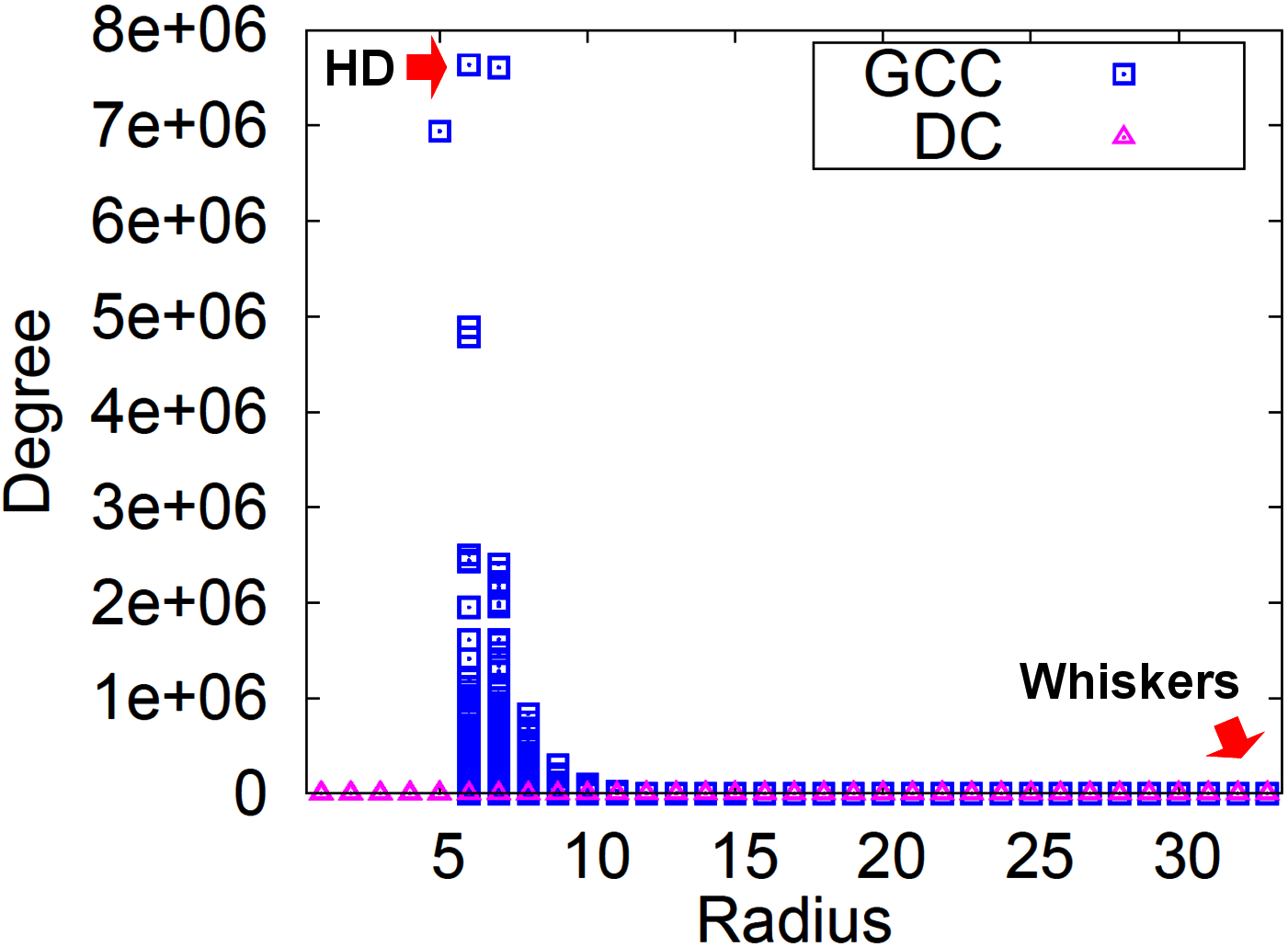} \\
       (a) Patent &  (b) YahooWeb \\
    \includegraphics[width=0.45\textwidth]{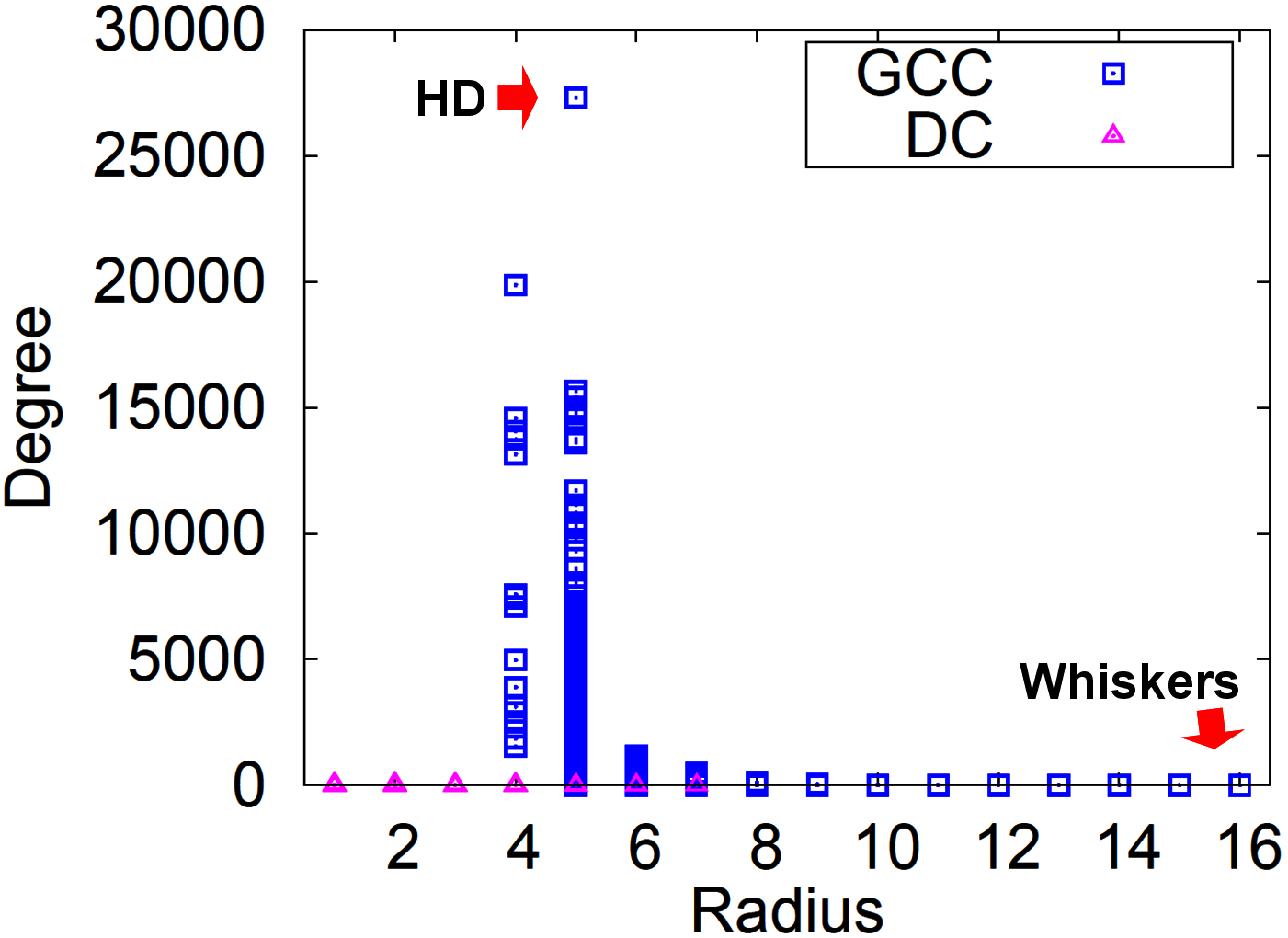} &\\
       (c) LinkedIn &\\
    \end{tabular}
    \caption{
    Radius-Degree plots of real-world graphs. HD represents the vertex with the highest degree.
    Notice that HD belongs to core vertices inside the GCC, and whiskers have small degree.
    }
    \label{fig:patent_t1_dr}
    \end{center}
\end{figure*}

Now we can explain the three important areas of
Figure~\ref{fig:static_radius_others}:
`{\em outsiders}' are the vertices in the disconnected components, and responsible for the first peak and the negative slope to the dip.
`{\em Core}' are the central vertices with the smallest radii from the giant connected component.
`{\em Whiskers}'~\cite{jure08ncp} are the vertices  connected to the GCC with long paths.

\subsubsection{Radius plot of GCC}
Figure~\ref{fig:hadiyahoo}(b) shows that all vertices of the GCC of 
the YahooWeb graph have radius 6 or more except for {\tt google.com}
that has radius one.

\subsubsection{``Core'' and ``Whisker'' vertices }

Figure~\ref{fig:patent_t1_dr} shows the Radius-Degree plot
of Patent, YahooWeb and LinkedIn graphs. 
The Radius-Degree plot is a scatterplot  with one dot for every vertex plotting the degree of the vertex versus its radius. 
The points corresponding to vertices in the GCC are colored with blue, while the rest is in magenta. 
We observe that the highest degree vertices    belong to the set of core vertices inside the GCC but are not
necessarily the ones with the smallest radius.
Finally, the whisker vertices have small degree and belong to chain subgraphs.

\subsection{Temporal Patterns}
\label{subsec:temporalhadi}

Here we study the radius distribution as a function of time. 
We know that the diameter of a graph typically grows with time, spikes at the `gelling point',
and then shrinks~\cite{Mcglohon2008,1217301}. 
Indeed, this holds for our datasets as shown in Figure~\ref{fig:evolution_diam}.

\begin{figure*}[htbp]
\begin{center}
    \begin{tabular}{c c}
    \includegraphics[width=0.45\textwidth]{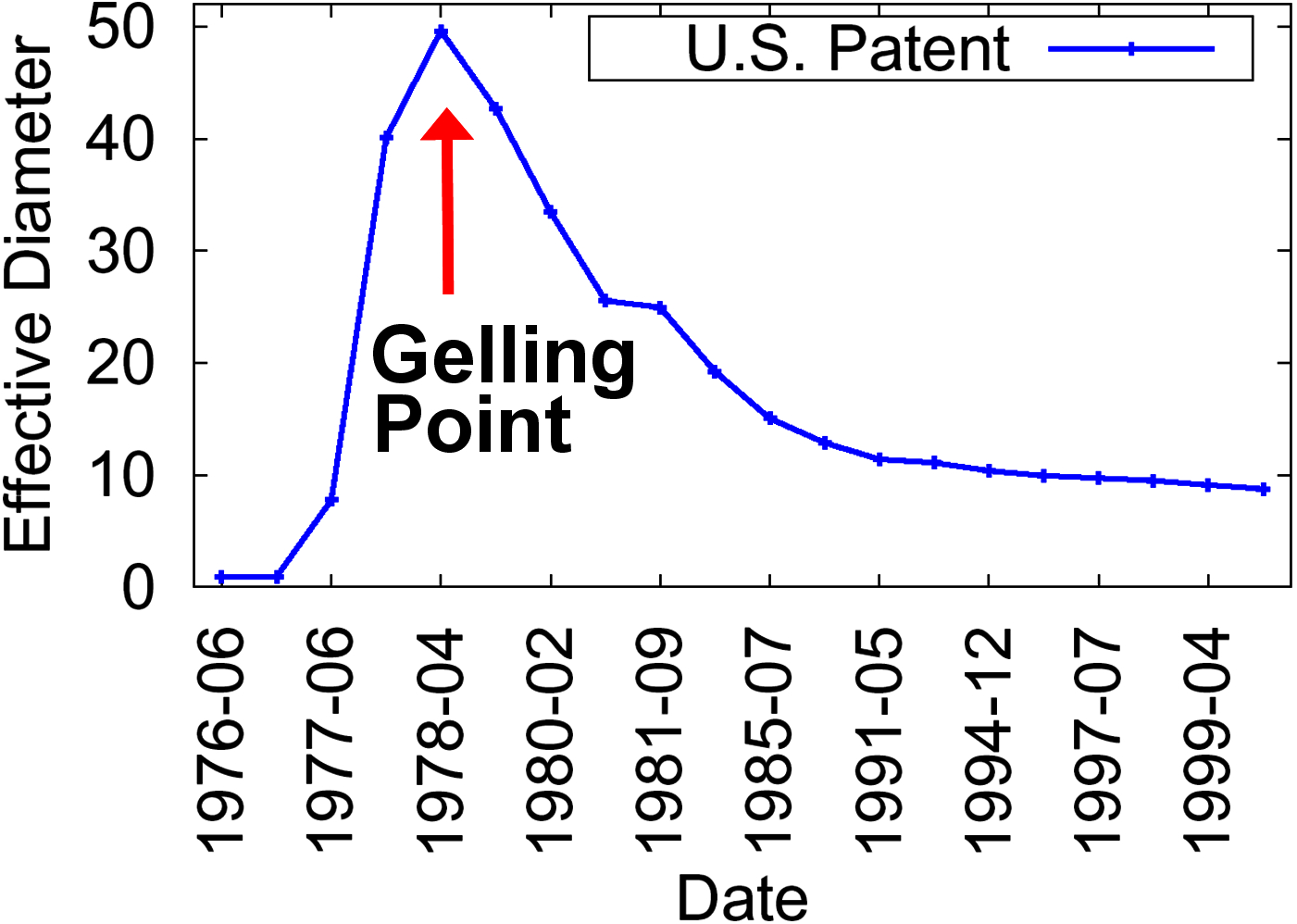}&
    \includegraphics[width=0.45\textwidth]{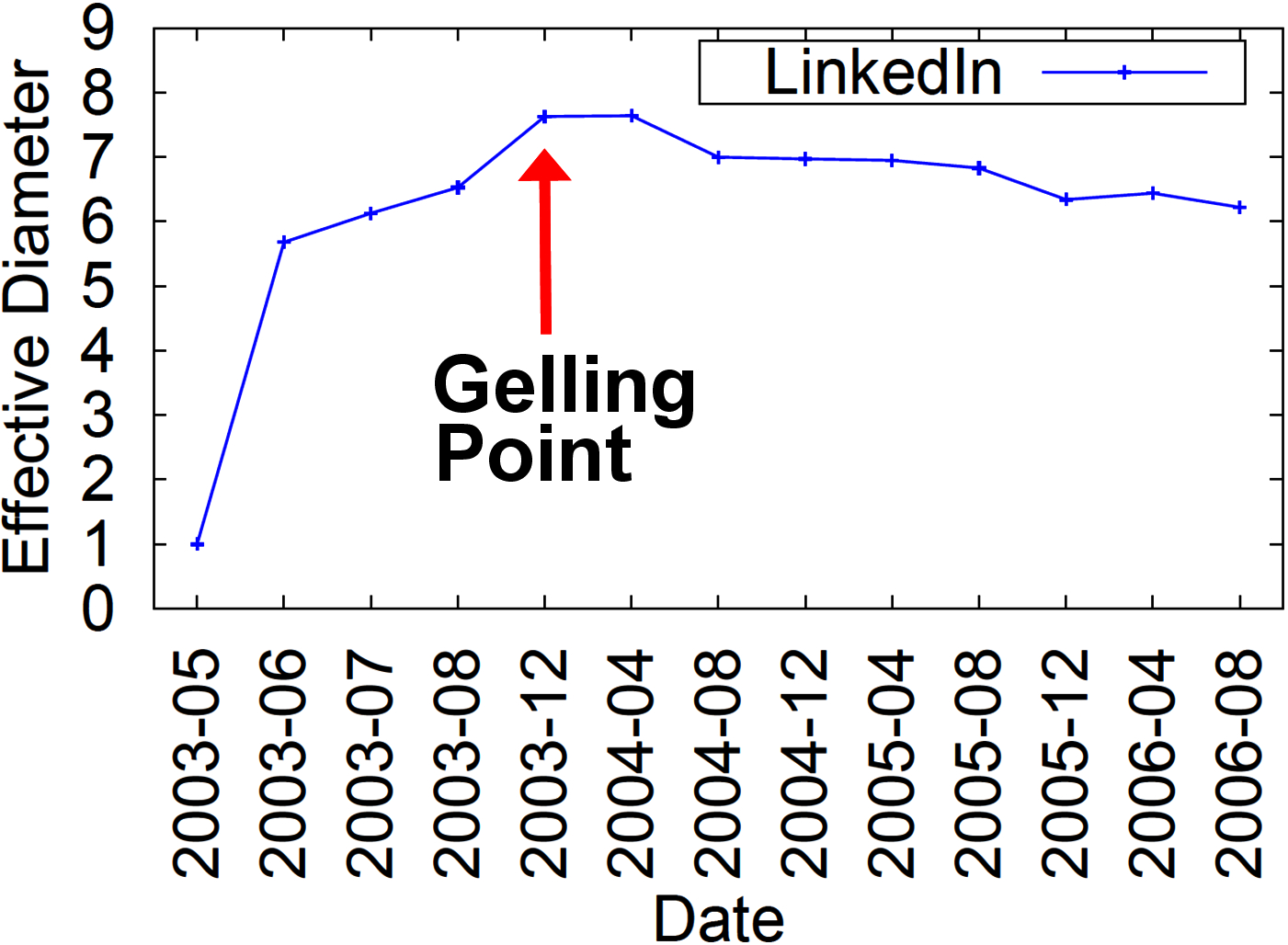}\\
       (a) Patent &  (b) LinkedIn
    \end{tabular}
    \caption{
    Evolution of the effective diameter of real graphs.
    The diameter increases until a `gelling' point, and starts to decrease after the point.
    }
    \label{fig:evolution_diam}
    \end{center}
\end{figure*}

Figure~\ref{fig:patent_uradius_overtime} shows our findings. 
The radius distribution expands to the right until it reaches the gelling point.
Then, it contracts to the left.
Finally, the decreasing segments of several, real radius plots seem to decay exponentially, that is
\begin{equation}
count (r) \propto \exp{  (- c r ) }
\end{equation}
for every time tick {\em after} the gelling point. $count (r)$ is the number of vertices 
with radius $r$, and $c$ is a constant. For the Patent and LinkedIn graphs,
the absolute correlation coefficient was  at least 0.98. 

\begin{figure*}[htbp]
\begin{center}
    \begin{tabular}{c c}
    \includegraphics[width=0.45\textwidth]{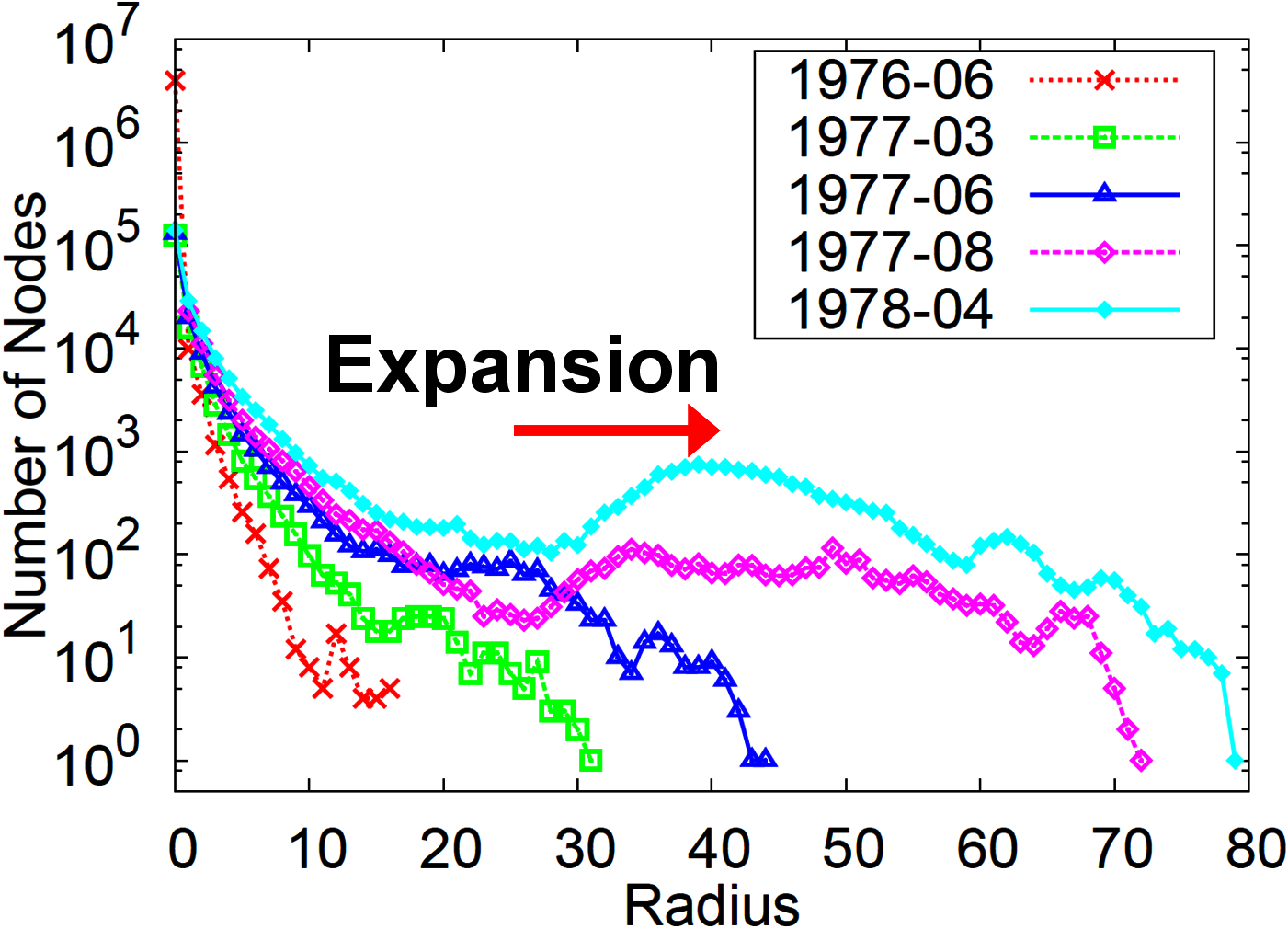}&
    \includegraphics[width=0.45\textwidth]{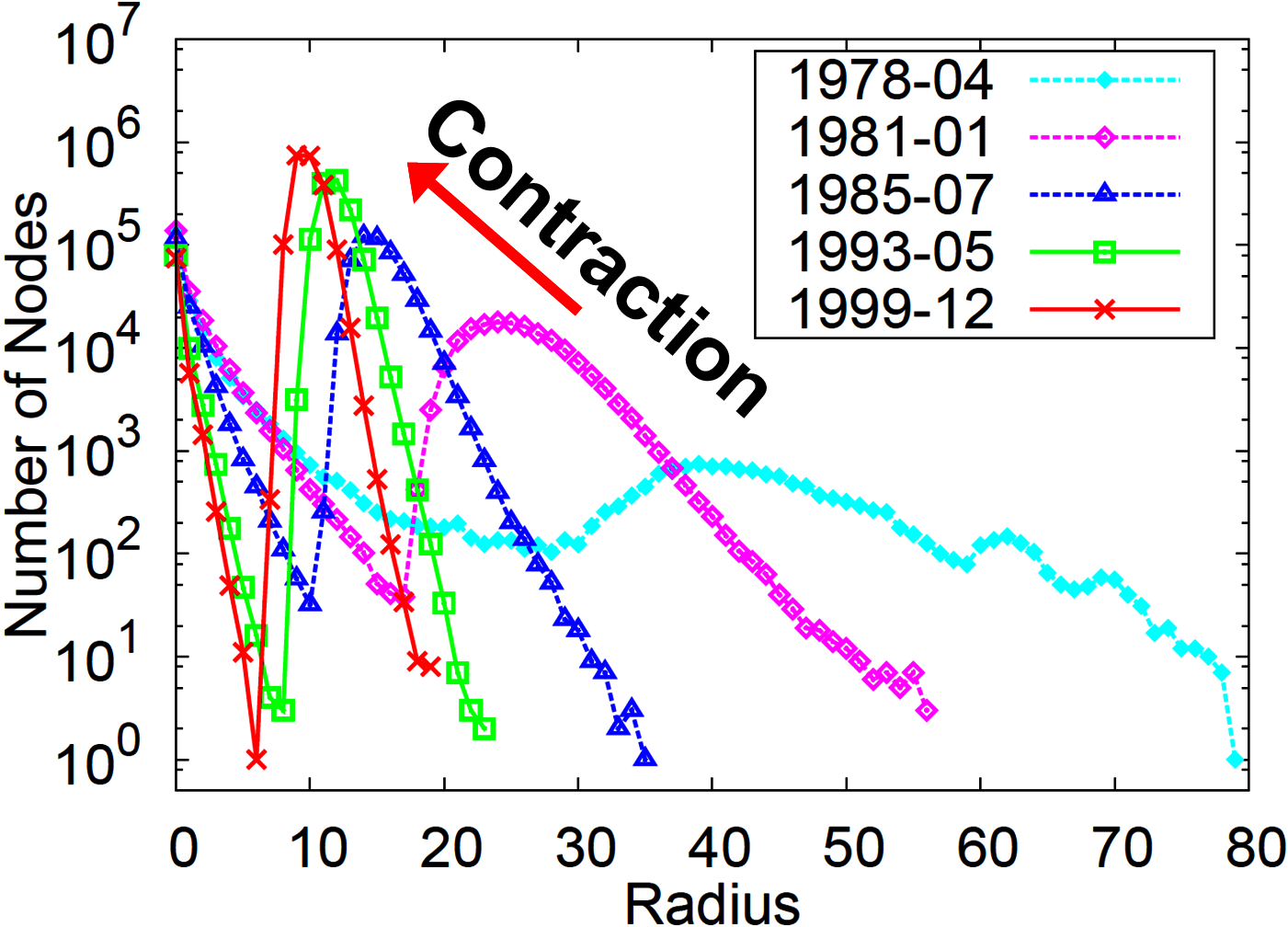}\\
       (a) Patent-Expansion &  (b) Patent-Contraction\\
    \includegraphics[width=0.45\textwidth]{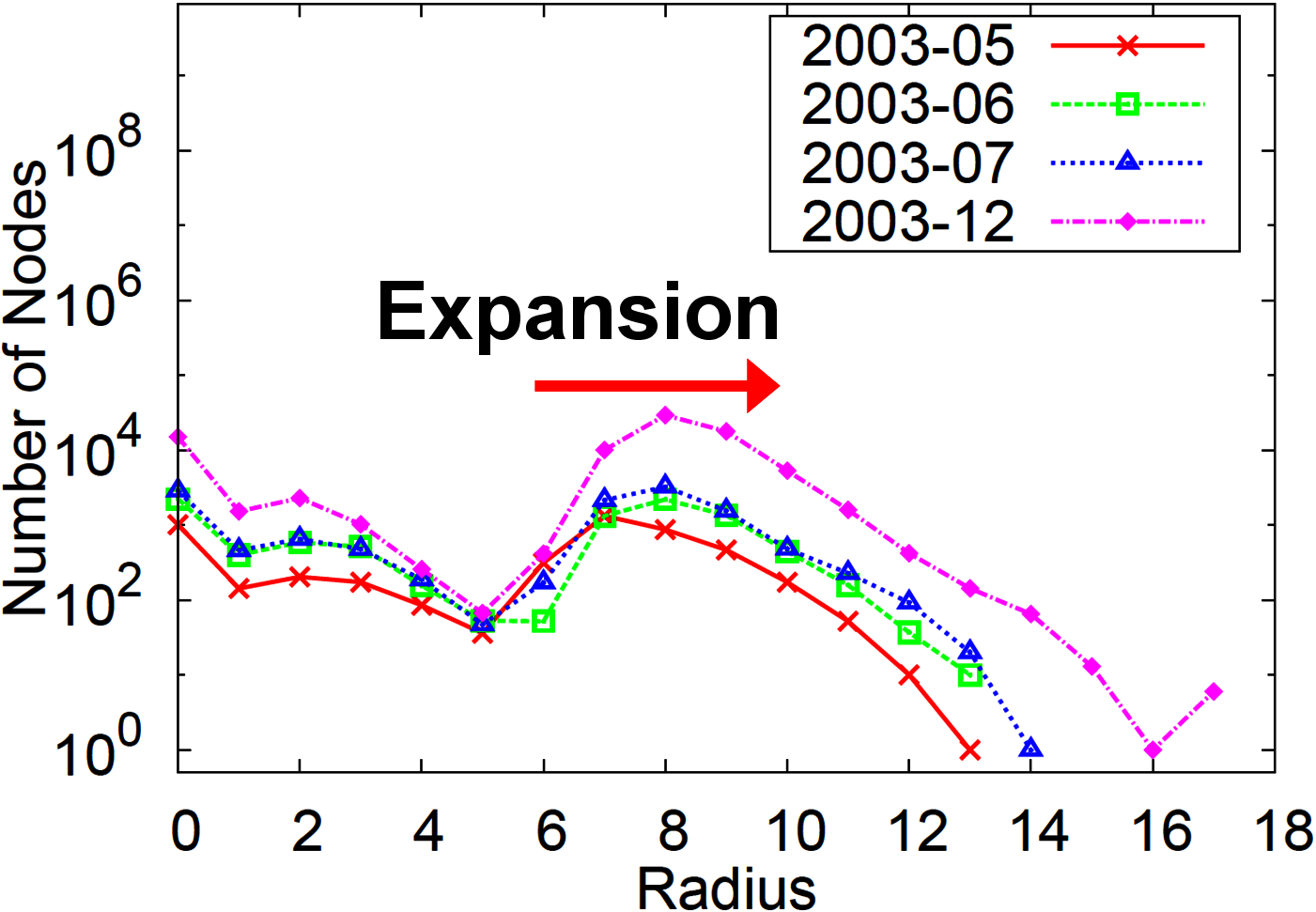}&
    \includegraphics[width=0.45\textwidth]{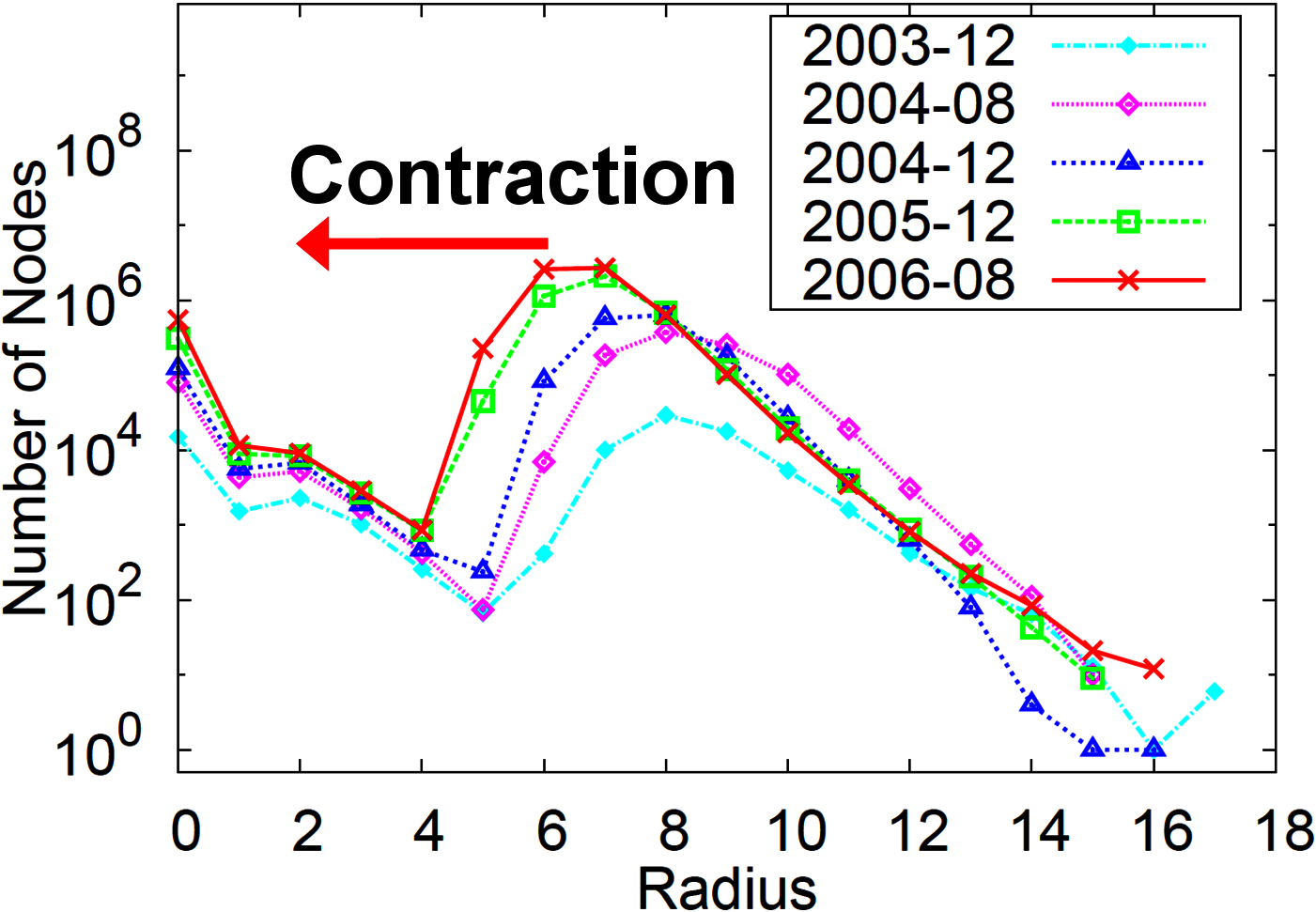}\\
       (c) LinkedIn-Expansion &  (d) LinkedIn-Contraction\\
    \end{tabular}
    \caption{
    Radius distribution over time.
    ``Expansion'': the radius distribution moves
    to the right until the gelling point.
    ``Contraction'': the radius distribution moves to the left
    after the gelling point.
    }
    \label{fig:patent_uradius_overtime}
    \end{center}
\end{figure*}

\newpage

\clearpage
\chapter{ FENNEL: Streaming Graph Partitioning for Massive Scale Graphs}
\label{fennelchapter}
\lhead{\emph{FENNEL: Streaming Graph Partitioning  for Massive Scale Graphs}} 
\section{Introduction}
\label{sec:Fennelintro}

\begin{table*}
\centering 
\begin{tabular}{|c|c|c|c|c|c|c|c|c|}  \hline
\multicolumn{1}{|c}{} & \multicolumn{2}{|c|}{\textsf{Fennel}} &  \multicolumn{2}{|c|}{Best competitor} &  \multicolumn{2}{|c|}{\textsf{Hash Partition}}  &  \multicolumn{2}{|c|}{\textsf{METIS}}\\ 
\cline{1-9}
\multicolumn{1}{|c|}{\# Clusters ($k$)} & \multicolumn{1}{c}{$\lambda$ } & \multicolumn{1}{c|}{$\rho$ } & 
                            \multicolumn{1}{c}{$\lambda$ } & \multicolumn{1}{c|}{$\rho$ } &
						    \multicolumn{1}{c}{$\lambda$ } & \multicolumn{1}{c|}{$\rho$ } & 
							\multicolumn{1}{c}{$\lambda$ } & \multicolumn{1}{c|}{$\rho$ }  \\ 
\cline{1-9}
2 & 6.8\% & 1.1 & 34.3\% & 1.04  & 50\%   & 1 & 11.98\%             & 1.02 \\
4 & 29\%  & 1.1 & 55.0\% & 1.07  & 75\%   & 1 & 24.39\%             & 1.03\\
8 & 48\%  & 1.1 & 66.4\% & 1.10  & 87.5\% & 1 & 35.96\%      & 1.03 \\ \hline
\end{tabular} 
\caption{\label{tab:Fenneltab1}Fraction of edges cut $\lambda$ and the normalized maximum load $\rho$ for \textsf{Fennel}, the previously best-known heuristic 
(linear weighted degrees \cite{stanton}) and hash partitioning of vertices for the Twitter graph with approximately 1.5 billion edges. 
\textsf{Fennel} and best competitor require around 40 minutes, METIS more than 8$\tfrac{1}{2}$ hours. }
\end{table*} 

Big volumes of data are typically managed and analyzed on large distributed systems \cite{dean}. 
Specifically, the data is partitioned across a large number of cheap, commodity machines 
which are  typically connected by gigabit Ethernet. 
Minimizing the communication between the machines is critical, since the communication
cost often dominates computation cost. 
A key problem towards minimizing communication cost for big graph 
data is the {\em  balanced graph partitioning} (BGPA) problem:
partition the vertex set of the graph in a given number of machines
in such a way that each partition is balanced and the number  edges cut is minimized.
The {\em  balanced graph partitioning} problem is a classic \NPhard   
problem \cite{Andreev:2004:BGP:1007912.1007931} for which
several approximation algorithms have been designed.
In practice, systems aim at providing good partitions in order to enhance their performance, e.g., \cite{malewicz,pujol}. 
It is worth emphasizing that the balanced graph partitioning problem appears in various guises in numerous domains
\cite{Karypis:1998:FHQ:305219.305248}.

Another major challenge in the area of big graph data is efficient processing of dynamic graphs.
 For example, new accounts are created and deleted every day in online services such as Facebook, Skype and Twitter. 
Furthermore, graphs created upon post-processing datasets such as Twitter posts 
are also dynamic, see for instance \cite{Angel}. It is crucial to have efficient graph partitioners of 
dynamic graphs. For example, in the Skype service, each time a 
user logs in, his/her online contacts get notified. It is expensive when messages have 
to be sent across different graph partitions since this would typically involve using network infrastructure. 
The balanced graph partitioning problem in the dynamic setting is known as 
{\em streaming graph partitioning} \cite{stanton}. Vertices (or edges) arrive and the decision of 
the placement of each vertex (edge) has to be done ``on-the-fly'' in order to incur as little computational overhead as possible. 

It is worth noting that the state-of-the-art work on graph partitioning seems to roughly 
divide in two main lines of research. Rigorous mathematically work and algorithms that do not 
scale to massive graphs, e.g., \cite{krauthgamer}, and heuristics that are used 
in practice~\cite{metis,Karypis:1998:FHQ:305219.305248,Prabhakaran:2012:MLG:2342821.2342825,stanton}. 
Our work contributes towards bridging the gap between theory and practice. 

The remainder of the Chapter is organized as follows. 
In Section~\ref{sec:Fenneldefinition}, we introduce our graph partitioning framework and 
present our main theoretical result. In Section~\ref{sec:Fennelstreaming}, 
we present our scalable, streaming algorithm. In Section~\ref{sec:Fennelexperiments}, 
we evaluate our method versus the state-of-the-art work on a broad set of real-world 
and synthetic graphs, while in Section~\ref{sec:Fennelsystem} 
we provide our experimental results in the Apache Giraph. 

\section{Proposed Framework}
\label{sec:Fenneldefinition}

{\bf Notation}. Throughout this Chapter we use the following notation. Let $G(V,E)$ be a simple, undirected graph. 
Let the number of vertices and edges be denoted as $|V|=n$ and $|E|=m$. For a subset of vertices 
$S \subseteq V$, let $e(S,S)$ be the set of edges with both end vertices in the set $S$, and let 
$e(S,V\setminus S)$ be the set of edges whose one end-vertex is in the set $S$ and the other is not. 
For a given vertex $v$ let $t_S{v}$ be the number of triangles $(v,w,z)$ such that $w,z \in S$.
We define a \emph{partition} of vertices ${\mathcal P}=(S_1,\ldots,S_k)$ 
to be a family of pairwise disjoint sets vertices, i.e., $S_i \subseteq V$, $S_i \cap S_j = \emptyset$ 
for every $i \neq j$. We call $S_i$ to be a cluster of vertices. Finally, for a graph $G = (V,E)$ and 
a partition $\mathcal{P} = (S_1,S_2,\ldots,S_k)$ of the vertex set $V$, let ${\partial e(\mathcal{P})}$ 
be the set of edges that cross partition boundaries, i.e. ${\partial e(\mathcal{P})} = \cup_{i=1}^k e(S_i,V\setminus S_i)$. 
 
{\bf Graph Partitioning Framework.} We formulate a graph partitioning framework that is based on accounting for the cost of internal edges and the cost of edges cut by 
a partition of vertices in a single global objective function. 

{\it The size of individual partitions}. We denote with $\sigma(S_i)$ the size of the cluster of vertices $S_i$, where $\sigma$ is a mapping to the set of real numbers. Special instances of interest are (1) \emph{edge cardinality} where the size of the cluster $i$ is proportional to the total number of edges with at least one end-vertex in the set $S_i$, i.e. $|e(S_i,S_i)| + |e(S_i,V\setminus S_i)|$, (2) \emph{interior-edge cardinality} where the size of cluster $i$ is proportional to the number of internal edges $|e(S_i,S_i)|$, and (3) \emph{vertex cardinality} where the size of partition $i$ is proportional to the total number of vertices $|S_i|$. The edge cardinality of a cluster is an intuitive measure of cluster size. This is of interest for computational tasks over input graph data where the computational complexity within a cluster of vertices is linear in the number of edges with at least one vertex in the given cluster. For example, this is the case for iterative computations such as solving the power iteration method. The vertex cardinality is a standard measure of the size of a cluster and for some graphs may serve as a proxy for the edge cardinality, e.g. for the graphs with bounded degrees.

{\it The global objective function}. We define a global objective function that consists of two elements: (1) the inter-partition cost $c_{\mathrm{OUT}}:N^k \rightarrow \field{R}_+$ and (2) the intra-partition cost $c_{\mathrm{IN}}:N^k \rightarrow \field{R}_+$. These functions are assumed to be increasing and super-modular (or convex, if extended to the set of real numbers). For every given partition of vertices $\mathcal{P} = (S_1,S_2,\ldots,S_k)$, we define the global cost function as
\begin{eqnarray*}
f({\mathcal P}) &=& c_{\mathrm{OUT}}(|e(S_1,V\setminus S_1)|, \ldots, |e(S_k,V\setminus S_k)|)\\
&& + c_{\mathrm{INT}}(\sigma(S_1),\ldots,\sigma(S_k)).
\end{eqnarray*}

It is worth mentioning some particular cases of interest. Special instance of interest for the inter-partition cost is the linear function in the total number of cut edges $|{\partial e({\mathcal P})}|$. This case is of interest in cases where an identical cost is incurred per each edge cut, e.g. in cases where messages are exchanged along cut edges and these messages are transmitted through some common network bottleneck. For the intra-partition cost, a typical goal is to balance the cost across different partitions and this case is accomodated by defining $c_{\mathrm{INT}}(\sigma(S_1),\ldots,\sigma(S_k)) = \sum_{i=1}^k c(\sigma(S_i))$, where $c(x)$ is a convex increasing function such that $c(0) = 0$. In this case, the intra-partition cost function, being defined as a sum of convex functions of individual cluster sizes, would tend to balance the cluster sizes, since the minimum is attained when sizes are equal.  

We formulate the graph partitioning problem as follows. 

\begin{center}\fbox{
\begin{minipage}{0.8\linewidth}
\textbf{Optimal $k$-Graph Partitioning}
\medskip\par
Given a graph $G=(V,E)$, find a partition $\mathcal{P}^*=\{S_1^*,\ldots,S_k^*\}$ 
of the vertex set $V$, such that $f(\mathcal{P}^*) \ge f(\mathcal{P})$, 
for all partitions $\mathcal{P}$ such that $|\mathcal{P}|=k$. \\

We refer to the partition ${\mathcal P}^*$ as the optimal $k$ graph partition of the graph $G$. 
\end{minipage}
}
\end{center}

{\it Streaming setting}. The streaming graph partitioning 
problem can be defined as follows. Let $G=(V,E)$ be an 
input graph and let us assume that we want to partition the 
graph into $k$ disjoint subsets of vertices. The vertices 
arrive in some order, each one with the set of its 
neighbors. 
We consider three different stream orders, as in \cite{stanton}. 

\squishlist
\item Random: Vertices arrive according to a random permutation. 
\item BFS: This ordering is generated by selecting 
a vertex uniformly at random and performing 
breadth first search starting from that vertex. 
\item DFS: This ordering is identical to the BFS ordering, except 
that we perform depth first search.
\squishend

A $k$-partitioning streaming algorithm has to decide whenever a new vertex arrives to which cluster it is going to be placed. A vertex is never moved after it has been assigned to a cluster. The formal statement of the problem follows.

{\bf Classic Balanced Graph Partitioning.} We consider the traditional instance 
of a graph partitioning problem that is a special case 
of our framework by defining the inter-partition cost 
to be equal to the total number of edges cut and the 
intra-partition cost defined in terms of the vertex cardinalities.  
 
The starting point in the existing literature, e.g., \cite{Andreev:2004:BGP:1007912.1007931, krauthgamer}, is to admit hard cardinality constraints, so that $|S_i^*| \leq \nu \frac{n}{k}$ for $i=1,\ldots,k$ , where $\nu \geq 1$ is a fixed constant. This set of constaints makes the problem significantly hard. Currently, state-of-the-art work depends on the impressive ARV barrier \cite{DBLP:conf/stoc/AroraRV04} which results in a $O(\sqrt{\log{n}})$ approximation factor. 
The typical formulation is the following:

\begin{equation*} \label{typicalIP2}
\framebox{
\begin{minipage}[b]{0.85\linewidth}
\begin{equation*}
\begin{array}{rl}
\hbox{minimize}_{\mathcal{P}=(S_1,\ldots,S_k)}  & |{\partial e(\mathcal{P})}| \\
\hbox{subject to} & |S_i| \leq \nu \frac{n}{k},\ \forall i \in \{1,\ldots,k\}
\end{array}
\end{equation*}
\end{minipage}\medskip\par}
\end{equation*}

{\it Our approach: Just Relax !} The idea behind our approach is to {\em relax} the hard cardinality constraints by introducing a term in the objective $c_{\mathrm{IN}}(\mathcal{P})$ whose minimum is achieved when $|S_i| =  \frac{n}{k}$ for all $i \in \{1,\ldots,k\}$. Therefore, our framework is based on a well-defined global graph partitioning objective function, which allows for a principled design of approximation algorithms and heuristics as shall be demonstrated in 
Section~\ref{sec:Fennelstreaming}. Our graph partitioning method is based on solving the following optimization problem:

\begin{equation} \label{fennel}
\framebox{
\begin{minipage}[b]{0.85\linewidth}
\begin{equation*}
\begin{array}{lr}
\hbox{minimize}_{\mathcal{P}=(S_1,\ldots,S_k)}  & |{\partial e(\mathcal{P})}| + c_{\mathrm{IN}}(\mathcal{P}) \\
\end{array}
\end{equation*}
\end{minipage}\medskip\par}
\end{equation}

{\it Intra-partition cost:} With the goal in mind to favor balanced partitions, we may define the intra-partition cost function by $c_{\mathrm{IN}}({\mathcal P}) = \sum_{i=1}^k c(|S_i|)$ where $c(x)$ is an increasing function choosen to be \emph{super-modular}, so that the following increasing returns property 
holds $c(x + 1) - c(x) \geq c(y+1) - c(y)$, for every $0\leq y\leq x$. 

We shall focus our attention to the following family of functions $c(x) = \alpha x^\gamma$, for $\alpha > 0$ and $\gamma \geq 1$. 
By the choice of the parameter $\gamma$, this family of cost functions allows us to 
control how much the imbalance of cluster sizes is accounted for
in the objective function. In one extreme case where $\gamma = 1$, 
we observe that the objective corresponds to minimizing the number of cut-edges, 
thus entirely ignoring any possible imbalance of the cluster sizes. 
On the other hand, by taking larger values for the parameter $\gamma$, 
the more weight is put on the cost of partition imbalance, 
and this cost may be seen to approximate hard constraints on the imbalance in the limit of large $\gamma$. 
Parameter $\alpha$ is also important. We advocate a principled choice of $\alpha$ independently of whether
it is suboptimal compared to other choices. Specifically, we choose $\alpha = m\frac{k^{\gamma-1}}{n^{\gamma}}$. 
This provides us a proper scaling, since for this specific choice of $\alpha$, our optimization problem 
is equivalent to minimizing a natural normalization of the objective function $ \frac{\sum_{i=1}^k e(S_i,V\setminus S_i)}{m} + \frac{1}{k}\sum_{i=1}^k \left(\frac{|S_i|}{\tfrac{n}{k}}\right)^{\gamma}$.

{\it An equivalent maximization problem}. We note that the optimal $k$ graph partitioning problem admits an equivalent formulation as a maximization problem. It is of interest to consider this alternative formulation as it allows us to make a connection with the concept of graph modularity, which we do later in this section. For a graph $G=(V,E)$ and $S\subseteq V$, we define the function $h: 2^{V} \rightarrow \field{R}$ as: 
$$
h(S) =  |e(S,V\setminus S)| -  c(|S|)
$$
where $h(\emptyset)=h(\{v\})=0$ for every $v \in V$. Given $k\geq 1$ and a partition $\mathcal{P}=\{S_1,\ldots,S_k\}$ of the vertex set $V$, we define the function $g$ as 
$$ 
g(\mathcal{P}) =  \sum_{i=1}^{k} h(S_i).
$$
Now, we observe that maximizing the function $g({\mathcal P})$ over all possible partitions ${\mathcal P}$ of the vertex set $V$ such that $|{\mathcal P}| = k$ corresponds to the $k$ graph partitioning problem. Indeed, this follows by noting that
\begin{eqnarray*}
g({\mathcal P}) &=& \sum_{i=1}^k |e(S_i,S_i)| - c(|S_i|)\\
&=& (m - \sum_{i=1}^k |e(S_i,V\setminus S_i)|) - c(|S_i|)\\
&=& m - f({\mathcal P}). 
\end{eqnarray*}
Thus, maximizing function $g({\mathcal P})$ corresponds to minimizing function $f({\mathcal P})$, which is precisely the objective of our $k$ graph partitioning problem.

{\it Modularity:} We note that when the function $c(x)$ is taken from the family $c(x) = \alpha x^{\gamma}$, for $\alpha > 0$ and $\gamma=2$, our objective has a combinatorial interpretation. Specifically, our problem is equivalent to maximizing the function
$$
\sum_{i=1}^k  [|e(S_i,S_i)| - p {|S_i| \choose 2}]
$$
where $p = \alpha / 2$. In this case, each summation element admits the following intepretation: it corresponds to the difference between 
the realized number of edges within a cluster and the expected number of edges within the cluster under the null-hypothesis that 
the graph is an Erd\"{o}s-R\'{e}nyi random graph with parameter $p$. 
This is intimately related to the concepts of graph 
modularity \cite{girvan2002community,newman2004finding,newman} and quasi-cliques \cite{uno}.
Recall that an approximation algorithm for a maximization problem is meaningful as a notion if the optimum solution is positive. 
However, in our setting our function $g$ does not results as it can easily be seen in a non-negative optimum. 
For instance, if $G$ is the empty graph on $n$ vertices, any partition of $G$ in $k$ parts results in a negative 
objective (except when $k=n$ when the objective becomes 0). Therefore, we need to shift our objective in 
order to come up with a multiplicative approximation algorithm.
We define the following shifted objective, along the lines of Chapter~\ref{densestchapter}.

\begin{definition}
$k \geq 1$. Also let $\mathcal{P}^*=\{S_1^*,\ldots,S_k^*\}$ be a partition of the vertex set $V$.
We define the  function $g$  as 
$$
\tilde g(\mathcal{P}) =  \alpha {n \choose 2} + \sum_{i=1}^{k} f(S_i).
$$
\end{definition}

\begin{claim} 
For any partition $\mathcal{P}$,  $\tilde g(\mathcal{P})  \geq 0$. 
\end{claim}

\begin{proof} 
The proof follows directly from the fact that for any positive-valued $s_1,s_2,\ldots,s_k$ such that $\sum_{i=1}^k s_i=n$, the following holds $n^2 \geq s_1^2 +\cdots + s_k^2$.
\end{proof} 

Due to the combinatorial interpretation of the objective function, we design 
a semidefinite programming approximation algorithm. Before that, we
see how random partitioning performs. 

\spara{Random Partitioning:}
Suppose each vertex is assigned to one of $k$ partitions uniformly at random. This simple 
graph partition is a faithful approximation of hash partitioning of vertices that is commonly used in practice. 
In expectation, each of the $k$ clusters will have $\frac{n}{k}$ vertices. Let $S_1,\ldots,S_k$ be the 
resulting $k$ clusters. How well does this simple algorithm perform with respect to our objective? 
Let $\mathcal{P}^*$ be an optimal partition for the optimal quasi-clique problem, i.e. 
$g({\mathcal P}^*) \geq g({\mathcal P})$, for every partition ${\mathcal P}$ of the vertex set 
$V$ into $k$ partitions. Notice that ${\mathcal P}^*$ is also an optimal partition for the optimal 
quasi-clique problem with shifted objective function. Now, note that an edge $e=(u,v)$ has 
probability $\frac{1}{k}$ that both its endpoints belong to the same cluster. 
By the linearity of expectation, we obtain by simple calculations:

\begin{eqnarray*} 
\Mean{\tilde g(S_1,\ldots,S_k)} &=& \frac{|E|}{k}+ \alpha \frac{k-1}{k} {n \choose 2}\\
&\geq & \frac{1}{k} \Bigg(|E|+ \alpha {n \choose 2} \Bigg)\\
&\geq & \frac{1}{k} \tilde g({\mathcal P}^*)
\end{eqnarray*} 
where last inequality comes from the simple upper bound $\tilde g(\mathcal{P}) \leq |E| + \alpha {n \choose 2}$ for any partition ${\mathcal P}$.

\spara{An SDP Rounding Algorithm}

We define a vector variable $x_i$ for each vertex $i \in V$ and we allow $x_i$ to be one of the unit vectors $e_1,\ldots,e_k$,
where $e_j$ has only the $j$-th coordinate 1. 

\medskip
\begin{equation} \label{IP}
\framebox{
\begin{minipage}[b]{0.8\linewidth}
\begin{equation*}
\begin{array}{rl}
\hbox{maximize}  &  \sum_{e=(i,j)} x_ix_j+ \alpha \sum_{i < j} \Big( 1- x_ix_j \Big)\\
\hbox{subject to} & x_i \in \{e_1,\ldots,e_k\},\ \forall i \in \{1,\ldots,n\}
\end{array}
\end{equation*}
\end{minipage}\medskip\par}
\end{equation}

\noindent We obtain the following semidefinite programming relaxation:

\begin{equation} \label{IP}
\framebox{
\begin{minipage}[b]{0.8\linewidth}
\begin{equation*}
\begin{array}{rl}
\hbox{maximize}  &  \sum_{e=(i,j)} y_{ij} +\alpha \sum_{i < j} \Big( 1 - y_{ij} \Big)\\
\hbox{subject to} & y_{ii}=1,\ \forall i \in \{1,\ldots,n\}\\
&  y_{ij} \geq 0,\ \forall i \neq j\\
& Y \succeq 0,\ Y \text{~symmetric}
\end{array}
\end{equation*}
\end{minipage}\medskip\par}
\end{equation}

\noindent The above SDP can be solved within an additive error of $\delta$ of the optimum in time polynomial in the size of the input and $\log{(\frac{1}{\delta})}$ by interior point algorthms or the ellipsoid method \cite{alizadeh}. In what follows, we refer to the optimal value of the integer program as $\mathrm{OPT}_{\mathrm{IP}}$ 
and of the semidefinite program as $\mathrm{OPT}_{\mathrm{SDP}}$. Our algorithm is the following: 

\begin{center}\fbox{
\begin{minipage}{0.8\linewidth}
\textbf{SDP-Relax}
\begin{itemize}
    \item {\em Relaxation:} Solve the semidefinite program \cite{alizadeh} and compute a Cholesky decomposition of $Y$. Let $v_0,v_1,\ldots,v_n$ 
    be the resulting vectors. 
    \item {\em Randomized Rounding:} Randomly choose $t=\lceil \log{k} \rceil$ unit length vectors $r_i \in \field{R}^{k}, i=1,\ldots, t$. 
    These $t$ random vectors define $2^t=k$ possible regions in which the vectors $v_i$ can fall: one region for each distinct possibility 
    of whether $r_jv_i \geq 0$ or $r_j v_i < 0$.  Define a cluster by adding all vertices whose vector $v_i$ fall in a given region. 
\end{itemize}
\end{minipage}
}
\end{center}

\begin{theorem} 
\label{thrm:cleansdp}
Algorithm SDP-Relax is a $\Omega(\frac{\log{k}}{k})$ approximation algorithm for the Optimal Quasi-Clique Partitioning. 
\label{thm:sdpclean}
\end{theorem}

\begin{proof}
 
Let $C_k$ be the score of the partition produced by our randomized rounding. Define $A_{i,j}$ to be the event that vertices $i$ and $j$ are assigned to the same partition. Then,

\begin{align*}
\Mean{C_k} &= \sum_{e=(i,j)} \Prob{A_{i,j}} + \alpha \sum_{i < j} \Big( 1 - \Prob{A_{i,j}}  \Big) \\ 
\end{align*}

As in Goemans-Williamson \cite{goemans}, given a random hyperplane with normal vector $r$ that goes through the origin, the probability
of $\hbox{sgn}(v_i^Tr) =\hbox{sgn}(v_j^Tr)$, i.e., $i$ and $j$ fall on the same side of the hyperplane, is $1-\frac{\arccos(v_i^Tv_j)}{\pi}$. 
Since we have $t$ independent hyperplanes 
$$
\Prob{A_{i,j}}= \left(1-\frac{\arccos(v_i^Tv_j)}{\pi}\right)^t.
$$

Let us define, for  $t \geq 1$,
$$
f_1(\theta) = \frac{\left(1-\frac{\theta}{\pi}\right)^t}{\cos(\theta)},\ \theta \in [0,\frac{\pi}{2}) 
$$
and
$$
\rho_1 = \min_{0\leq \theta < \frac{\pi}{2}} f_1(\theta).
$$

Similarly, define for  $t \geq 1$,
$$
f_2(\theta) = \frac{1-\left(1-\frac{\theta}{\pi}\right)^t}{1-\cos(\theta)},\ \theta \in [0,\frac{\pi}{2}) 
$$
and
$$
\rho_2 = \min_{0\leq \theta < \frac{\pi}{2}} f_2(\theta).
$$

We wish to find a $\rho$ such that $\Mean{C_k} \geq \rho \mathrm{OPT}_{\mathrm{SDP}}$. Since, $\mathrm{OPT}_{\mathrm{SDP}} \geq \mathrm{OPT}_{\mathrm{IP}}$, this would then imply $\Mean{C_k} \geq \rho \mathrm{OPT}_{\mathrm{IP}}$. 

We note that
$$
f_1'(\theta) = \frac{-\frac{t}{\pi}\left(1-\frac{\theta}{\pi}\right)^{t-1} \cos(\theta) + \left(1-\frac{\theta}{\pi}\right)^t \sin(\theta)}{\cos^2(\theta)}.
$$
It follows that $f_1'(\theta(t)) = 0$ is equivalent to 
$$
t =(\pi-\theta(t))\tan(\theta(t)).
$$
Notice that the following two hold
$$
\lim_{t\rightarrow \infty} \theta(t) = \frac{\pi}{2}
$$
and
\begin{equation}
\frac{\pi}{2}\tan(\theta(t)) \leq t \leq \pi \tan(\theta(t)).
\label{equ:tbounds}
\end{equation}

We next show that $f_1(\theta(t)) \geq \frac{1}{\pi} t 2^{-t}$, by the the following series of relations
\begin{eqnarray*}
f_1(\theta(t)) & = & \frac{\left(1-\frac{\theta(t)}{\pi}\right)^t}{\cos(\theta(t))} = \sqrt{1+\tan^2(\theta(t))}\left(1-\frac{\theta(t)}{\pi}\right)^t\\
& \geq & \sqrt{1+\tan^2(\theta(t))} 2^{-t}  \geq \sqrt{1 + \left(\frac{t}{\pi}\right)^2}2^{-t}\\
& \geq & \frac{1}{\pi} t 2^{-t}
\end{eqnarray*}
where the second equality is by the fact $\cos(\theta) = \frac{1}{\sqrt{1+\tan^2(\theta)}}$, the first inequality is by the fact $\theta(t) \leq \frac{\pi}{2}$, and the second inequality is by (\ref{equ:tbounds}).

Thus, for $t = \log_2(k)$, we conclude
$$
\rho_1 \geq \frac{1}{\pi\log(2)} \frac{\log(k)}{k}.
$$

Now, we show that $\rho_2=\frac{1}{2}$.  First we show that for any $t \geq 1$, $f_2(\theta) \geq 1/2$. To this end, we note
\begin{align*}
\frac{1-\left(1-\frac{\theta}{\pi}\right)^t}{1-\cos(\theta)} &\geq \frac{1}{2}  \Leftrightarrow \frac{1}{2} \Big( 1+ \cos(\theta)\Big) \geq \left(1-\frac{\theta}{\pi}\right)^t.
\end{align*}
Notice that for all $\theta \in [0,\pi/2)$, if $t_1 \geq t_2 \geq 1$, then $ \Big( 1- \frac{\theta}{\pi} \Big)^{t_1}  \leq \Big( 1- \frac{\theta}{\pi} \Big)^{t_2}$.
Hence, it suffices $\frac{1}{2} \Big( 1+ \cos(\theta) \Big) \geq \left(1-\frac{\theta}{\pi}\right)$.
With the use of simple calculus, the latter is true and is also tight for $\theta=0$ and $\theta = \pi/2$. 
It is worth observing the opposite trend of the values of $\rho_1$ and $\rho_2$. The reason is that $\Prob{A_{i,j}}$
drops as we use more hyperplanes and, of course, $1-\Prob{A_{i,j}}$  grows. 
Now, we can establish the following lower bound on the expected score of our randomized rounding procedure. Let $\theta_{i,j} = \arccos(v_i^Tv_j)$. 

\begin{align*}
\Mean{C_k} &= \sum_{e=(i,j)} \Prob{A_{i,j}} + \alpha \sum_{i < j} \Big( 1 - \Prob{A_{i,j}}  \Big) \\ 
                  &= \sum_{e=(i,j)}  \left(1-\frac{\theta_{i,j}}{\pi}\right)^t + \alpha \sum_{i < j} \Big( 1 -\left(1-\frac{\theta_{i,j}}{\pi}\right)^t   \Big) \\ 
                  &\geq  \sum_{e=(i,j)} \rho_1 \cos(\theta_{i,j}) + \alpha \sum_{i < j} \rho_2 \Big( 1 -  \cos(\theta_{i,j})  \Big) \\  
                  &\geq \min\{\rho_1,\rho_2\} \mathrm{OPT}_{\mathrm{SDP}}  \\ 
                  &\geq \min\{\rho_1,\rho_2\} \mathrm{OPT}_{\mathrm{IP}}
\end{align*}

It suffices to set $\rho = \min\{\rho_1,\rho_2\}$. The above analysis shows that our algorithm is a $\rho$-approximation
algorithm, where $\rho = \Omega(\frac{\log(k)}{k})$. 

\end{proof}

The approximation guarantees that hold for the shifted objective function have the following meaning for the 
original objective function. Suppose that $\mathcal{P}$ is a $\rho$-approximation with respect 
to the shifted-objective quasi-clique partitioning problem, i.e. $\tilde g(\mathcal{P}) \geq \rho \tilde g(\mathcal{P}^*)$. Then, it holds
$$
g(\mathcal{P}) \geq \rho g(\mathcal{P}^*) - (1-\rho)\alpha {n \choose 2}.
$$
Notice that this condition is equivalent to 
$$
g(\mathcal{P}) - g(\mathcal{P}^*) \geq -(1-\rho)[g(\mathcal{P}^*) +\alpha {n \choose 2}].
$$
We wish to outline that 
this approximation guarantee is near to a $\rho$-approximation if the parameter $\alpha$ is small 
enough so that the term $g(\mathcal{P}^*)$ dominates the term $\alpha {n\choose 2}$. For example, for an input 
graph that consists of $k$ cliques, 
we have that $g(\mathcal{P}^*) \geq \alpha {n\choose 2}$ corresponds to $\frac{\alpha}{1-\alpha} \leq \frac{1}{k}\frac{n+k}{n+1}$, 
from which we conclude that it suffices that $\alpha \leq \frac{1}{k+1}$.

\spara{Alternative Roundings} 

One may ask whether the relaxation of Frieze and Jerrum \cite{friezejerrum}, Karger, Motwani and Sudan \cite{kargercoloring} can improve significalty 
the approximation factor. We provide negative evidence. 
Before we go into further details, we notice that the main ``bottleneck'' in our approximation is the 
probability of two vertices being in the same cluster, as $k$ grows.
The probability  that $i,j$ are in the same cluster in our rounding is $p(\theta)$ where 

$$
p(\theta) = \Big( 1- \tfrac{\theta}{\pi} \Big)^{\log{k}}.
$$

\noindent Suppose $\theta = \frac{\pi}{2} (1-\epsilon)$. Then, 
\begin{align*} 
p(\theta) &= \Big( 1- \tfrac{\theta}{\pi} \Big)^{\log{k}} = \Big( \frac{1+\epsilon}{2} \Big)^{\log{k}} \\
          & = \frac{1}{k} + \epsilon \frac{\log{k}}{k} + O(\epsilon^2). 
\end{align*}

As we see from Lemma~5 in \cite{friezejerrum}, for this $\theta$, the asymptotic expression matches ours:
\begin{eqnarray*}
N_k(\cos(\theta))& \approx & \frac{1}{k} + \frac{2\log{k}}{k} \cos{\Big( \frac{\pi}{2}(1-\epsilon)\Big)}\\
& = & \frac{1}{k} + \pi \epsilon \frac{\log{k}}{k} + O(\epsilon^2). 
\end{eqnarray*}

\section{One-Pass Streaming Algorithm}
\label{sec:Fennelstreaming}

We derive a streaming algorithm by using a greedy assignment of vertices to partitions as follows: assign each arriving vertex 
to a partition such that 
the objective function of the $k$ graph partitioning problem, defined as a maximization problem, is increased the most. 
Formally, given that current vertex partition is $\mathcal{P} = (S_1,S_2,\ldots,S_k)$, a vertex $v$ is assigned to partition $i$ such that
\begin{eqnarray*}
&& g(S_1,\ldots,S_i\cup \{v\},\ldots,S_j,\ldots,S_k)\\
&\geq & g(S_1,\ldots,S_i,\ldots,S_j\cup \{v\},\ldots,S_k), \hbox{ for all } j\in [k].
\end{eqnarray*}  
Defining $\delta g(v,S_i) = g(S_1,\ldots,S_i\cup \{v\},\ldots,S_j,\ldots,S_k) - g(S_1,\ldots,S_i,\ldots,S_j,\ldots,S_k)$, the above greedy assignment of vertices corresponds to that in the following algorithm.

\begin{center}\fbox{
\begin{minipage}{0.8\linewidth}
\textbf{Greedy vertex assignment}
\begin{itemize}
    \item Assign vertex $v$ to partition $i$ such that $\delta g(v,S_i) \geq \delta g(v,S_j)$, for all $j\in [k]$\\
\end{itemize}
\end{minipage}
}
\end{center}

{\it Special case: edge-cut and balanced vertex cardinality}. This is a special case of introduced 
that we discussed in Section~\ref{sec:Fenneldefinition}. In this case, $\delta g(v,S_l) = |N(v) \cap S_l| - \delta c(|S_l|)$, where $\delta c(x) = c(x+1)-c(x)$, for $x\in \field{R}_+$, and $N(v)$ denotes the set of neighbors of vertex $v$. 
The two summation elements in the greedy index $\delta g(v,S_l)$ account for the two underlying objectives of minimizing the number of cut edges and balancing of the partition sizes. 
Notice that the component $|N(v)\cap S_i|$ corresponds to the number of neighbours 
of vertex $v$ that are assigned to partition $S_i$. In other words, this corresponds 
to the degree of vertex $v$ in the subgraph induced by $S_i$. On the other hand, 
the component $\delta c(|S_i|)$ can be interpreted as the marginal cost of increasing 
the partition $i$ by one additional vertex. 

For our special family of cost functions $c(x) = \alpha x^\gamma$, 
we have $\delta c(x) = \alpha \gamma x^{\gamma - 1}$. 
For $\gamma = 1$, the greedy index rule corresponds to 
assigning a new vertex $v$ to partition $i$ with the largest 
number of neighbours in $S_i$, i.e $|N(v) \cap S_i|$. This is 
one of the greedy rules considered by Stanton and Kliot~\cite{stanton}, 
and is a greedy rule that may result in highly imbalanced partition sizes. 

On the other hand, in case of quadratic cost $c(x) = \frac{1}{2} x^2$, 
the greedy index is $|N(v)\cap S_i| - |S_i|$, and the greedy assignment 
corresponds to assigning a new vertex $v$ to partition $i$ that minimizes 
the number of \emph{non-neighbors} of $v$ inside $S_i$, i.e. $|S_i \setminus N(v)|$. 
Hence, this yields the following heuristic: 
{\em place a vertex to the partition with the least number of non-neighbors} \cite{Prabhakaran:2012:MLG:2342821.2342825}. 
This assignment accounts for both the cost of cut edges and the balance of partition sizes.

Finally, we outline that in many applications there exist very strict constraints 
on the load balance. Despite the fact that we investigate the effect of the parameter 
$\gamma$ on the load balance, one may apply the following algorithm, which enforces 
to consider only machines whose load is at most $\nu \times  \frac{n}{k}$.
This algorithm for $1 \leq \gamma \leq 2$ amounts to interpolating between 
the basic heuristics of \cite{stanton} and \cite{Prabhakaran:2012:MLG:2342821.2342825}.
The overall complexity of our algorithm is $O(n+m)$.

\begin{center}\fbox{
\begin{minipage}{0.85\linewidth}
\textbf{Greedy vertex assignment with threshold $\nu$}
\begin{itemize}
    \item Let $I_{\nu}=\{i: \mu_i \leq \nu \frac{n}{k} \}$. 
    Assign vertex $v$ to partition $i\in I_\nu$ such that $\delta g(v,S_i) \geq \delta g(v,S_j)$, for all $j\in I_\nu$
\end{itemize}
\end{minipage}
}
\end{center}

\section{Experimental Evaluation}
\label{sec:Fennelexperiments}

In this section we present results of our experimental evaluations of the quality of graph partitions 
created by our method and compare with alternative methods. We first describe our experimental setup 
in Sections~\ref{subsec:setup}, and then present our findings using synthetic and real-world graphs, 
in Section~\ref{subsec:synthetic} and \ref{subsec:realworld}, respectively. 

\begin{table}[t]
\begin{center}
\small	
\begin{tabular}{r|rrl|}
\multicolumn{1}{c}{}  &
\multicolumn{1}{c}{\sf Nodes} &
\multicolumn{1}{c}{\sf Edges} &
\multicolumn{1}{c}{\sf Description}\\ \cline{2-4}
\textsf{amazon0312} & 400\,727 & 2\,349\,869  & Co-purchasing  \\
\textsf{amazon0505} & 410\,236 & 2\,439\,437 & Co-purchasing  \\
\textsf{amazon0601} & 403\,364 & 2\,443\,311 & Co-purchasing   \\
\textsf{as-735} & 6\,474  & 12\,572 & Auton.\ Sys.\  \\
\textsf{as-Skitter}     & 1\,694\,616   & 11\,094\,209 & Auton.\ Sys.\  \\
\textsf{as-caida} & 26\,475   & 53\,381 & Auton.\ Sys.\ \\
\textsf{ca-AstroPh} & 17\,903 & 196\,972  & Collab.\  \\
\textsf{ca-CondMat}    & 21\,363 & 91\,286 & Collab.\  \\
\textsf{ca-GrQc} & 4\,158  & 13\,422 & Collab.\  \\
\textsf{ca-HepPh} & 11\,204 & 117\,619  & Collab.\  \\
\textsf{ca-HepTh} &   8\,638  & 24\,806 & Collab.\   \\
\textsf{cit-HepPh} & 34\,401 & 420\,784 &  Citation  \\
\textsf{cit-HepTh}    & 27\,400 & 352\,021  &     Citation    \\
\textsf{cit-Patents}  & 3\,764\,117  & 16\,511\,740    & Citation    \\
\textsf{email-Enron} & 33\,696  & 180\,811    & Email  \\
\textsf{email-EuAll} & 224\,832 & 339\,925  & Email   \\
\textsf{epinions} & 119\,070  & 701\,569   & Trust \\
\textsf{Epinions1} & 75\,877  & 405\,739    & Trust \\
\textsf{LiveJournal1}    & 4\,843\,953  & 42\,845\,684   &    Social   \\
\textsf{p2p-Gnutella04} &  10\,876  &  39\,994 &  P2P   \\
\textsf{p2p-Gnutella05} & 8\,842 & 31\,837 &  P2P  \\
\textsf{p2p-Gnutella06}    & 8\,717  & 31\,525  &    P2P   \\
\textsf{p2p-Gnutella08} & 6\,299 & 20\,776 &  P2P  \\
\textsf{p2p-Gnutella09}    & 8\,104  & 26\,008  & P2P   \\
\textsf{p2p-Gnutella25} & 22\,663 & 54\,693 &  P2P  \\
\textsf{p2p-Gnutella31}    & 62\,561  & 147\,878  &   P2P   \\  
\textsf{roadNet-CA}  & 1\,957\,027  & 2\,760\,388   & Road    \\
\textsf{roadNet-PA} & 1\,087\,562  & 1\,541\,514   &Road   \\
\textsf{roadNet-TX} & 1\,351\,137 & 1\,879\,201 &  Road    \\
\textsf{Slashdot0811}  & 77\,360  & 469\,180    & Social   \\
\textsf{Slashdot0902} & 82\,168  & 504\,230    & Social  \\
\textsf{Slashdot081106} & 77\,258 & 466\,661 &  Social  \\
\textsf{Slashdot090216}    & 81\,776  & 495\,661   & Social   \\
\textsf{Slashdot090221}  & 82\,052  & 498\,527   & Social   \\
\textsf{usroads} & 126\,146  & 161\,950   & Road  \\
\textsf{wb-cs-stanford}    & 8\,929  & 2\,6320   &     Web   \\
\textsf{web-BerkStan}  & 654\,782  & 6\,581\,871    & Web   \\
\textsf{web-Google} & 855\,802  & 4\,291\,352   & Web  \\
\textsf{web-NotreDame} & 325\,729 & 1\,090\,108 &  Web   \\
\textsf{web-Stanford}    & 255\,265  & 1\,941\,926   &     Web   \\
\textsf{wiki-Talk}  & 2\,388\,953  & 4\,656\,682  & Web   \\
\textsf{Wikipedia-20051105} & 1\,596\,970  & 18\,539\,720   & Web  \\
\textsf{Wikipedia-20060925} & 2\,935\,762  & 35\,046\,792   & Web  \\ 
\textsf{Twitter} & 41\,652\,230  & 1\,468\,365\,182   & Social \\
 \cline{2-4}
\end{tabular}
\end{center}
\caption{\label{tab:datasets}Datasets used in our experiments.}	
\end{table}

\subsection{Experimental Setup}
\label{subsec:setup}

\begin{table}
\centering 
\begin{tabular}{|c|c|c|c|c|}  \hline
\multicolumn{1}{|c}{} & \multicolumn{2}{|c|}{\textsf{BFS}} &  \multicolumn{2}{|c|}{Random} \\
\cline{1-5}
\multicolumn{1}{|c|}{Method} & \multicolumn{1}{c}{$\lambda$ } & \multicolumn{1}{c|}{$\rho$ } &                             
							\multicolumn{1}{c}{$\lambda$ } & \multicolumn{1}{c|}{$\rho$ }  \\ 
\cline{1-5}
\textsf{H}   & 96.9\% & 1.01 & 96.9\% & 1.01 \\ 
\textsf{B}  \cite{stanton} & 97.3\% & 1.00 & 96.8\% & 1.00 \\ 
\textsf{DG}  \cite{stanton} & 0\% & 32 & 43\% & 1.48 \\
\textsf{LDG} \cite{stanton} & 34\% & 1.01 & 40\% & 1.00 \\
\textsf{EDG} \cite{stanton} & 39\% & 1.04 & 48\% & 1.01 \\
\textsf{T}  \cite{stanton} & 61\% & 2.11 & 78\% & 1.01 \\
\textsf{LT} \cite{stanton}& 63\% & 1.23 & 78\% & 1.10 \\
\textsf{ET} \cite{stanton}  & 64\% & 1.05 & 79\% & 1.01 \\
\textsf{NN} \cite{Prabhakaran:2012:MLG:2342821.2342825}  & 69\% & 1.00 & 55\% & 1.03 \\ \hline
\textsf{Fennel} & 14\% & 1.10 & 14\% & 1.02 \\
\textsf{METIS} \cite{metis} & 8\% & 1.00 & 8\% & 1.02 \\ \hline
\end{tabular} 
\caption{\label{tab:amazon}
 Performance of various existing methods on \textsf{amazon0312}, $k$ is set to 32.}
\end{table}

\begin{figure*}[t]
\centering
\includegraphics[width=0.45\textwidth]{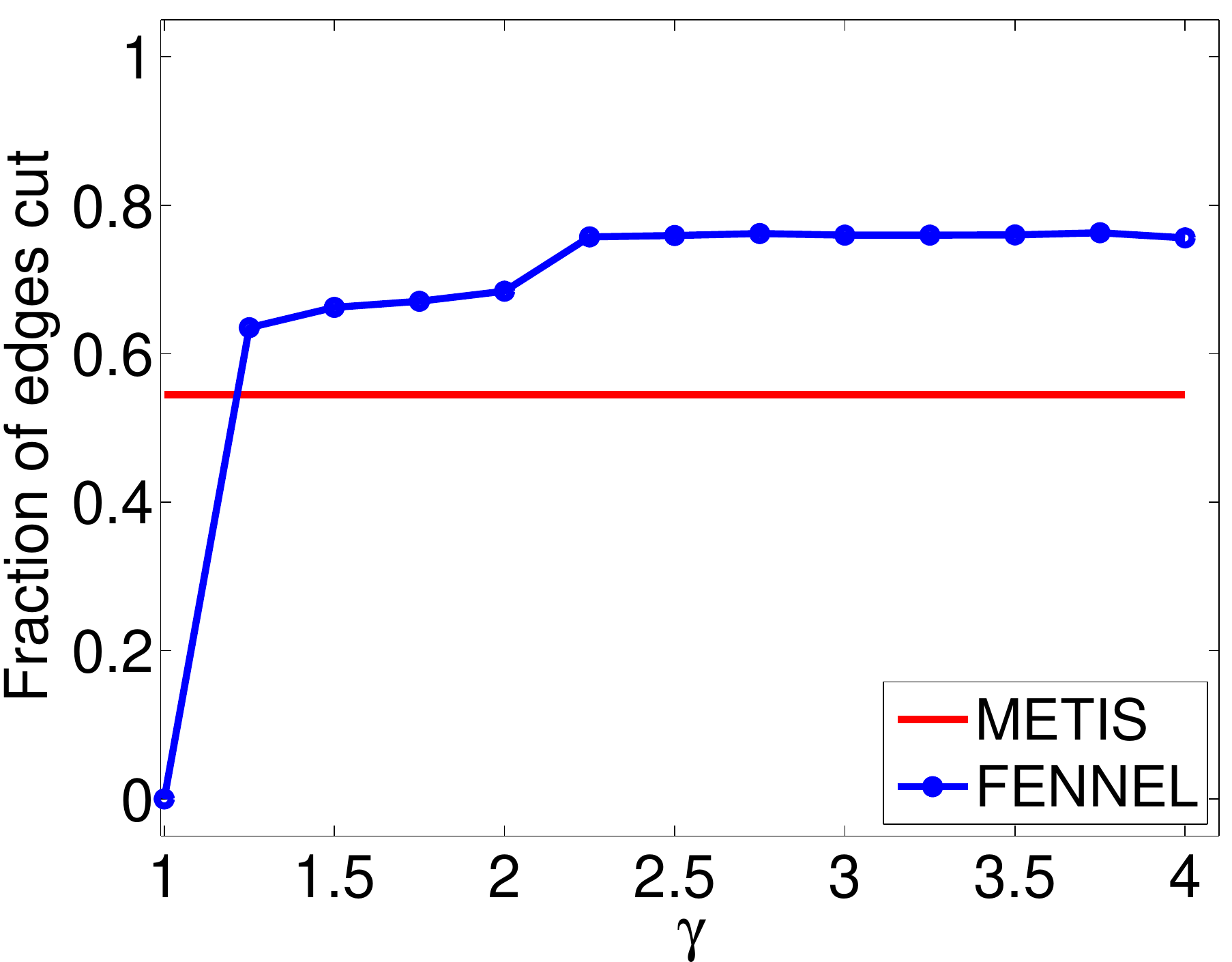}
\hspace{0.03\textwidth}
\includegraphics[width=0.45\textwidth]{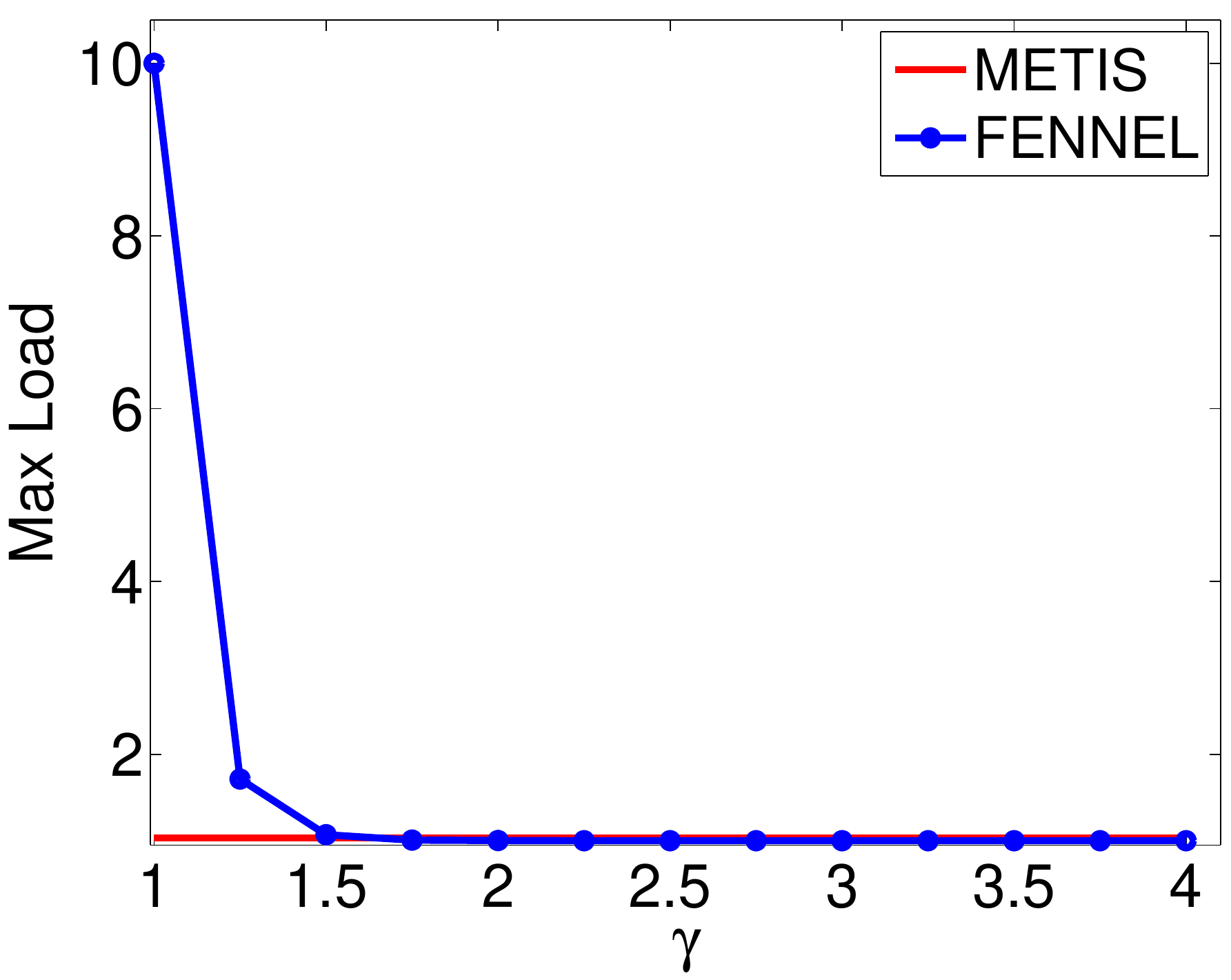}
\hspace{0.03\textwidth}
\caption{\label{fig:powerlaw}
Fraction of edges cut $\lambda$ and maximum load normalized $\rho$ as a function of $\gamma$,
ranging from 1 to 4 with a step of 0.25, over five randomly generated power law graphs with slope 2.5.
The straight lines show the performance of \textsf{METIS}. }
\end{figure*}

The real-world graphs used in our experiments are shown in Table~\ref{tab:datasets}.  
Multiple edges, self loops, signs and weights were removed, if any.
Furthermore, we considered the largest connected component from each graph
in order to ensure that there is a non-zero number of edges cut. 
All graphs are publicly available on the Web. All algorithms have been implemented in {\sc java}, and all experiments were performed on a 
single machine, with Intel Xeon {\sc cpu} at 3.6GHz, and 16GB of main memory. Wall-clock times include only the algorithm execution time, excluding the required time to load the graph into memory. 

In our synthetic experiments, we use two random graph models. 
The first model is the hidden partition model \cite{DBLP:conf/random/CondonK99}. 
It is specified by four parameters parameters: the number of vertices $n$, 
the number of clusters $k$, the intercluster and intracluster edge probabilities
$p$ and $q$, respectively. First, each vertex is assigned to one of $k$ clusters uniformly 
at random. We add an edge between two vertices of the same (different) cluster(s)
with probability $p$ ($q$) independently of the other edges. 
We denote this model as \textsf{HP}$(n,k,p,q)$. 
The second model we use is a standard model for generating random power law graphs.
Specifically, we first generate a power-law degree sequence with a given slope $\delta$ 
and use the Chung-Lu random graph model
to create an instance of a power law graph \cite{chunglu}. 
The model \textsf{CL}$(n,\delta)$ has two parameters: 
the number of vertices $n$ and the slope $\delta$ of the expected 
power law degree sequence. 

We evaluate our algorithms by measuring two quantities from the resulting 
partitions. In particular, for a fixed partition $\mathcal{P}$ 
we use the measures of the {\em fraction of edges cut} $\lambda$ and the {\em normalized maximum load} $\rho$, defined as
\begin{eqnarray*}
\lambda & = & \frac{\#\text{ edges cut by~} \mathcal{P}}{\#\text{ total edges}} = \frac{|{\partial e(\mathcal{P})}|}{m}, \mbox{ and} \\
\rho & = & \frac{\text{maximum load}}{\frac{n}{k}}.
\end{eqnarray*}

Throughout this section, we also use the notation $\lambda_{{\mathcal M}}$ and $\mu_{{\mathcal M}}$ to indicate the partitioning method ${\mathcal M}$ used in a particular context. In general, we omit indices whenever it is clear to which partition method we refer to. 
Notice that $k \geq \rho \geq 1$ since the maximum load of a cluster is at most $n$ and there always exists at least one cluster with at least $\frac{n}{k}$ vertices. 

In Section~\ref{subsec:synthetic}, we use the greedy vertex assignment without any threshold. 
Given that we are able to control ground truth, we are mainly interested in understanding 
the effect of the parameter $\gamma$ on the tradeoff between the fraction of edges cut and the normalized maximum load. 
In Section~\ref{subsec:realworld}, the setting of the parameters we use throughout our experiments is 
$\gamma=\frac{3}{2}$, $\alpha= \sqrt{k} \frac{m}{n^{3/2}}$, and $\nu=1.1$. The choice of $\gamma$ is based on our findings from Section~\ref{subsec:synthetic}
and of $\alpha$ based on Section~\ref{sec:Fenneldefinition}. Finally, $\nu=1.1$ is a reasonable load balancing factor for real-world settings.

As our competitors we use state-of-the-art heuristics. 
Specifically, in our evaluation we consider the following 
heuristics from \cite{stanton}, which we briefly describe here for completeness. 
Let $v$ be the newly arrived vertex. 

\squishlist 
\item Balanced (\textsf{B}): place $v$ to the cluster $S_i$ with minimal size. 
\item Hash partitioning (\textsf{H}): place $v$ to a cluster chosen uniformly at random. 
\item Deterministic Greedy (\textsf{DG}): place $v$ to $S_i$ that maximizes $|N(v)\cap S_i|$. 
\item Linear Weighted Deterministic Greedy (\textsf{LDG}): place $v$ to $S_i$ that maximizes $|N(v)\cap S_i| \times (1-\frac{|S_i|}{\tfrac{n}{k}})$. 
\item Exponentially Weighted Deterministic Greedy (\textsf{EDG}): place $v$ to $S_i$ that maximizes $|N(v)\cap S_i| \times  \Big(1-  \exp{ \big(|S_i|- \tfrac{n}{k}\big) } \Big)$. 
\item Triangles (\textsf{T}): place $v$ to $S_i$ that maximizes $t_{S_i}(v)$. 
\item Linear Weighted Triangles (\textsf{LT}):  place $v$ to $S_i$ that maximizes $t_{S_i}(v) \times \Big(1-\frac{|S_i|}{\tfrac{n}{k}} \Big)$. 
\item Exponentially Weighted Triangles (\textsf{ET}): place $v$ to $S_i$ that maximizes $t_{S_i}(v) \times  \Big(1- \exp{ \big(|S_i|- \tfrac{n}{k}\big) } \Big)$. 
\item Non-Neighbors (\textsf{NN}): place $v$ to $S_i$ that minimizes $|S_i\setminus N(v)|$.
\squishend

In accordance with \cite{stanton}, we observed that \textsf{LDG} is the best performing heuristic. 
Even if Stanton and Kliot do not compare with \textsf{NN}, \textsf{LDG} outperforms it also. 
Non-neighbors typically have very good load balancing properties, as \textsf{LDG} as well, 
but cut significantly more edges. 
Table~\ref{tab:amazon} shows the typical performance we observe across all datasets.
Specifically, it shows $\lambda$ and $\rho$ for both BFS and random order
for \textsf{amazon0312}. DFS order is omitted since qualitatively 
it does not differ from BFS. We observe that \textsf{LDG} is the best competitor, 
\textsf{Fennel} outperforms all existing competitors and is inferior to \textsf{METIS}, but of comparable
performance. In whatever follows, whenever we refer to the best competitor, unless 
otherwise mentioned we refer to \textsf{LDG}.

\begin{table}
\centering 
\begin{tabular}{|c|c|c|c|c|c|c|c|}  \hline
\multicolumn{1}{|c}{} & \multicolumn{1}{|c|}{} & \multicolumn{2}{|c|}{\textsf{Fennel}}  &  \multicolumn{2}{|c|}{\textsf{METIS}} \\ 
\cline{1-6}
\multicolumn{1}{|c}{$m$} & \multicolumn{1}{|c|}{$k$} &  \multicolumn{1}{c}{$\lambda$ } & \multicolumn{1}{c|}{$\rho$ } & 
																			    \multicolumn{1}{c}{$\lambda$ } & \multicolumn{1}{c|}{$\rho$ } \\ 
\cline{1-6}
7\,185\,314 & 4   & 62.5 \%   & 1.04 &  65.2\%   & 1.02 \\
6\,714\,510 & 8   & 82.2 \%   & 1.04 &  81.5\%   & 1.02 \\ 
6\,483\,201 & 16  & 92.9 \%   & 1.01 &  92.2\%   & 1.02 \\
6\,364\,819 & 32  & 96.3\%    & 1.00 &  96.2\%   & 1.02 \\
6\,308\,013 & 64  & 98.2\%    & 1.01 &  97.9\%   & 1.02 \\
6\,279\,566 & 128 & 98.4 \%   & 1.02 &  98.8\%   & 1.02 \\ \hline
\end{tabular} 
\caption{\label{tab:hiddenpartition} Fraction of edges cut $\lambda$ and normalized maximum load $\rho$
for \textsf{Fennel} and \textsf{METIS}~\cite{metis} averaged over 5 random graphs
generated according to the \textsf{HP}(5000,0.8,0.5) model. }
\end{table}

\subsection{Synthetic Datasets}
\label{subsec:synthetic}

Before we delve into our findings, it is worth summarizing the main 
findings of this section. (a) For all synthetic graphs we generated,
the value $\gamma=\frac{3}{2}$ achieves the best performance {\it pointwise},
not in average.  
(b) The effect of the stream order is minimal on the results. 
Specifically, when $\gamma \geq \tfrac{3}{2}$  
all orders result in the same qualitative results.
When $\gamma< \tfrac{3}{2}$ 
BFS and DFS orders result in the same 
results which are worse with respect to load balancing
--and hence better for the edge cuts--  compared
to the random order. 
(c) \textsf{Fennel}'s performance is comparable to \textsf{METIS}.

{\it Hidden Partition:} We report averages over five randomly generated graphs
according to the model \textsf{HP}$(5000,k,0.8,0.5)$ for each value of $k$ we use. 
We study (a) the effect of the parameter $\gamma$, which parameterizes 
the function $c(x) = \alpha x^{\gamma}$, and (b) the effect of the 
number of clusters $k$. 

We range $\gamma$ from 1 to 4 with a step of 1/4, for six different
values of $k$ shown in the second column of Table~\ref{tab:hiddenpartition}.
For all $k$, we observe, consistently, the following behavior:
for $\gamma=1$ we observe that $\lambda=0$ and $\rho=k$. This means that 
one cluster receives all vertices. 
For any $\gamma$ greater than 1, we obtain excellent load 
balancing with $\rho$ ranging from 1 to 1.05, and 
the same fraction of edges cut with \textsf{METIS} up the the first decimal digit. 
This behavior was not expected a priori, since in general we expect
$\lambda$ shifting from small to large values and see $\rho$ 
shifting from large to small values as $\gamma$ grows. 
Given the insensitivity of \textsf{Fennel} to $\gamma$ in this setting, 
we fix $\gamma=\frac{3}{2}$ and present in Table~\ref{tab:hiddenpartition} 
our findings. For each $k$ shown in the second column we generate five random 
graphs. The first column shows the average number of edges. Notice 
that despite the fact that we have only 5,000 vertices, we obtain graphs with
several millions of edges. 
The four last columns show the performance of \textsf{Fennel} and \textsf{METIS}. 
As we see, their performance is comparable and in one case (k=128)
\textsf{Fennel} clearly outperforms \textsf{METIS}.

{\it Power Law:}  It is well known that power law graphs have no good cuts \cite{gonzales}, 
but they are commonly observed in practice. We examine the effect of parameter $\gamma$ for $k$ fixed to 10. 
In contrast to the hidden partition experiment, we observe the expected
tradeoff between $\lambda$ and $\rho$ as $\gamma$ changes. 
We generate five random power law graphs CL(20\,000,2.5), 
since this value matches the slope of numerous real-world networks \cite{newman2003structure}.
Figure~\ref{fig:powerlaw} shows the tradeoff when $\gamma$ ranges from 
1 to 4 with a step of 0.25 for the random stream order.  
The straight line shows the performance of \textsf{METIS}. 
As we see, when $\gamma<1.5$, $\rho$ is unacceptably large 
for demanding real-world applications. 
When $\gamma=1.5$ we obtain essentially the same 
load balancing performance with \textsf{METIS}.
Specifically, $\rho_{\text{Fennel}} = 1.02, ~\rho_{\text{METIS}} = 1.03$. 
The corresponding cut behavior for $\gamma=1.5$  is 
$\lambda_{\text{Fennel}} =62.58\%, \lambda_{\text{METIS}} =54.46\%$. 
Furthermore, we experimented with the random, BFS and DFS 
stream orders. We observe that the only major difference between the stream orders
is obtained for $\gamma =1.25$. For all other $\gamma$ values the behavior is 
identical. For $\gamma=1.25$ we observe that BFS and DFS stream orders result
in significantly worse load balancing properties. Specifically, 
$\rho_{\text{BFS}}=3.81,~\rho_{\text{DFS}}=3.73,~\rho_{\text{Random}}=1.7130$. 
The corresponding fractions of edges cut are
$\lambda_{\text{BFS}}=37.83\%$,  $\lambda_{\text{DFS}}=38.85\%$, and $\lambda_{\text{Random}}=63.51\%$. 

\subsection{Real-World Datasets}
\label{subsec:realworld}

\begin{figure*}[bt]
\centering
\includegraphics[width=0.45\textwidth]{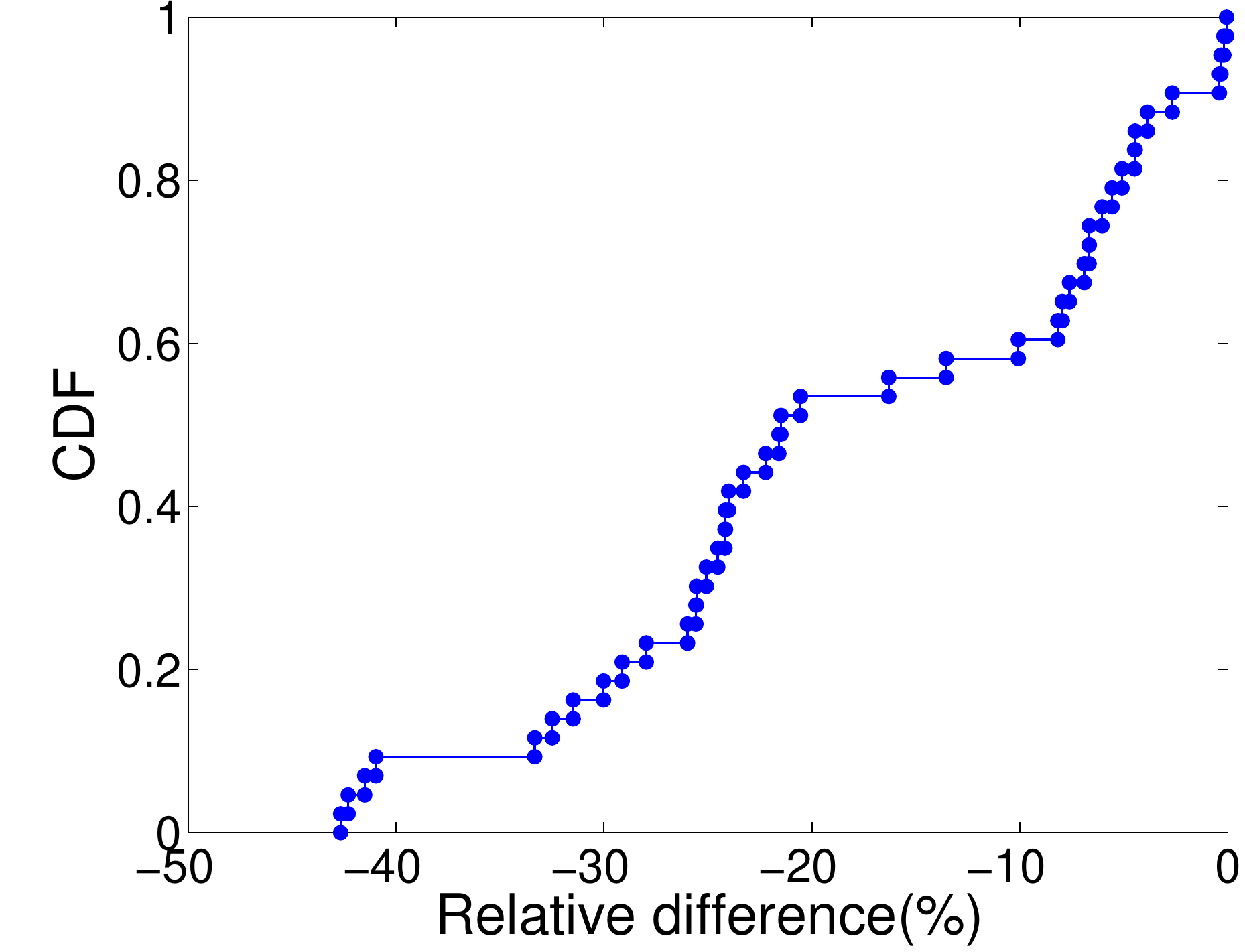}
\hspace{0.03\textwidth}
\includegraphics[width=0.45\textwidth]{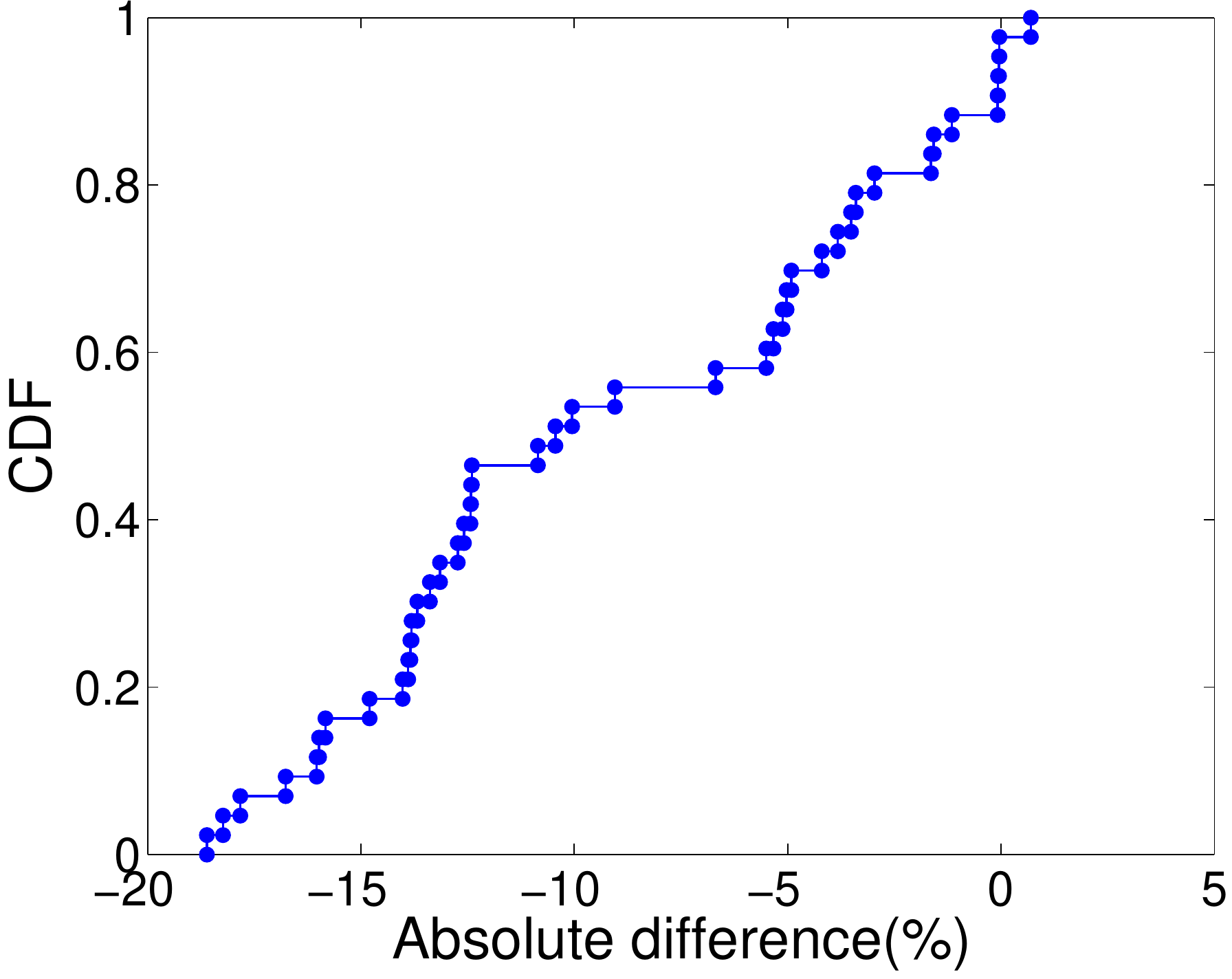}
\hspace{0.03\textwidth}
\caption{\label{fig:relative} (Left) CDF of the relative difference 
$( \frac{\lambda_{\text{Fennel}} - \lambda_{\text{c}} }{ \lambda_{\text{c}} }) \times 100\%$ 
of percentages of edges cut of our method and the best competitor. (Right) Same but for the absolute difference
$( \lambda_{\text{Fennel}} - \lambda_{\text{c}}) \times 100\%$.}
\end{figure*}

\begin{table}[t]
\centering 
\begin{tabular}{|c|c|c|}  \hline
$k$ & Relative Gain & $\rho_{\text{Fennel}}-\rho_{\text{c}}$ \\ \hline
2	 & 25.37\% & 0.47\% \\ \hline
4	 & 25.07\% & 0.36\%\\ \hline
8	 & 26.21\% & 0.18\% \\\hline
16	 & 22.07\% & -0.43\% \\\hline
32	 & 16.59\% & -0.34\%    \\\hline
64	 & 14.33\% & -0.67\%  \\\hline
128	 & 13.18\% & -0.17\%      \\\hline
256	 & 13.76\% & -0.20\%	\\\hline
512	 & 12.88\% &-0.17\%		\\\hline
1024 & 11.24\% &-0.44\%		\\\hline
\end{tabular} 
\caption{\label{tab:realavg}
The relative gain $( 1-\frac{\lambda_{\text{Fennel}}}{ \lambda_{\text{c}} }) \times 100\%$ 
and  load imbalance, where subindex $c$ stands for the best competitor,
averaged over all datasets in Table~\ref{tab:datasets} as a function of $k$. 
}
\end{table} 

Again, before we delve into the details of the experimental results, 
we summarize the main points of this Section: (1) 
{\it \textsf{Fennel} is superior to existing streaming partitioning algorithms.}
Specifically, it consistently, over a wide range of $k$ values 
and over all datasets, performs better than the 
current state-of-the-art. \textsf{Fennel} achieves 
excellent load balancing with significantly 
smaller edge cuts. (2) 
For smaller values of $k$ (less or equal than 64) the observed gain is 
more pronounced. 
(c) {\it \textsf{Fennel} is fast.} Our implementation scales well with 
the size of the graph. It takes about 40 minutes to partition
the Twitter graph which has more than 1 billion of edges.

{\it Twitter Graph.} Twitter graph is the largest graph in our collection 
of graphs, with more than 1.4 billion edges. This feature
makes it the most interesting graph from our collection, even if, 
admittedly, is a graph that can be loaded into the main memory. 
The results of \textsf{Fennel} on this graph are excellent. 
Specifically, Table~\ref{tab:Fenneltab1} shows the performance of \textsf{Fennel}, 
the best competitor \textsf{LDG}, the baseline \textsf{Hash Partition} and \textsf{METIS}
for $k=2,4$ and $8$. 
All methods achieve balanced partitions, with $\rho \leq 1.1$. 
\textsf{Fennel}, is the only method that always attains this upper bound. 
However, this reasonable performance comes with a high gain 
for $\lambda$. Specifically, we see that \textsf{Fennel} achieves better 
performance of $k=2$ than \textsf{METIS}. Furthermore, \textsf{Fennel} requires 42 minutes
whereas \textsf{METIS} requires 8$\tfrac{1}{2}$ hours. 
Most importantly, \textsf{Fennel} outperforms \textsf{LDG} consistently. 
Specifically, for $k=16,32$ and $64$, \textsf{Fennel} achieves the following results $(\lambda,\rho) = (59\%,1.1)$, $(67\%,1.1)$, and $(73\%,1.1)$, respectively. 
Linear weighted degrees (\textsf{LDG}) achieves $(76\%,1.13)$, $(80\%,1.15)$, and $(84\%,1.14)$, respectively.  
Now we turn our attention to smaller bur reasonably-sized datasets.

In Figure~\ref{fig:relative}, we show the distribution of the difference 
of the fraction of edges cut of our method and that of the best competitor, 
conditional on that the maximum observed load is at most $1.1$. 
This distribution  is derived from the values observed by partitioning each input graph from our set 
averaged over a range of values of parameter $k$ that consists of values $2,4, \ldots, 1024$. 
These results demonstrate that the fraction of edges cut by our method is smaller 
than that of the best competitor in all cases. Moreover, we observe
that the median difference (relative difference) is in the excess of 20\% (15\%), thus providing appreciable performance gains.

Furthermore, in Table~\ref{tab:realavg}, we present the average performance gains conditional 
on the number of partitions $k$. These numerical results amount to an average 
relative reduction of the fraction of edges cut in the excess of 18\%. Moreover, 
the performance gains observed are consistent across different values of parameter 
$k$, and are more pronounced for smaller values of $k$.

{\it Bicriteria.} 
In our presentation of experimental results so far, we focused on the fraction of edges cut by 
conditioning on the cases where the normalized maximum load was smaller than a fixed threshold. 
We now provide a closer look at both criteria and their relation. In Figure~\ref{fig:scat}, 
we consider the difference of the fraction of edges cut vs. the difference of normalized maximum loads 
of the best competitor and our method. We observe that in all the cases, the differences of 
normalized maximum loads are well within $10\%$ while the fraction of edges cut by our method is significantly 
smaller. These results confirm that the observed reduction of the fraction of edges cut by our method 
is not at the expense of an increased maximum load.

\begin{figure*}[t]
\centering
\includegraphics[width=0.45\textwidth]{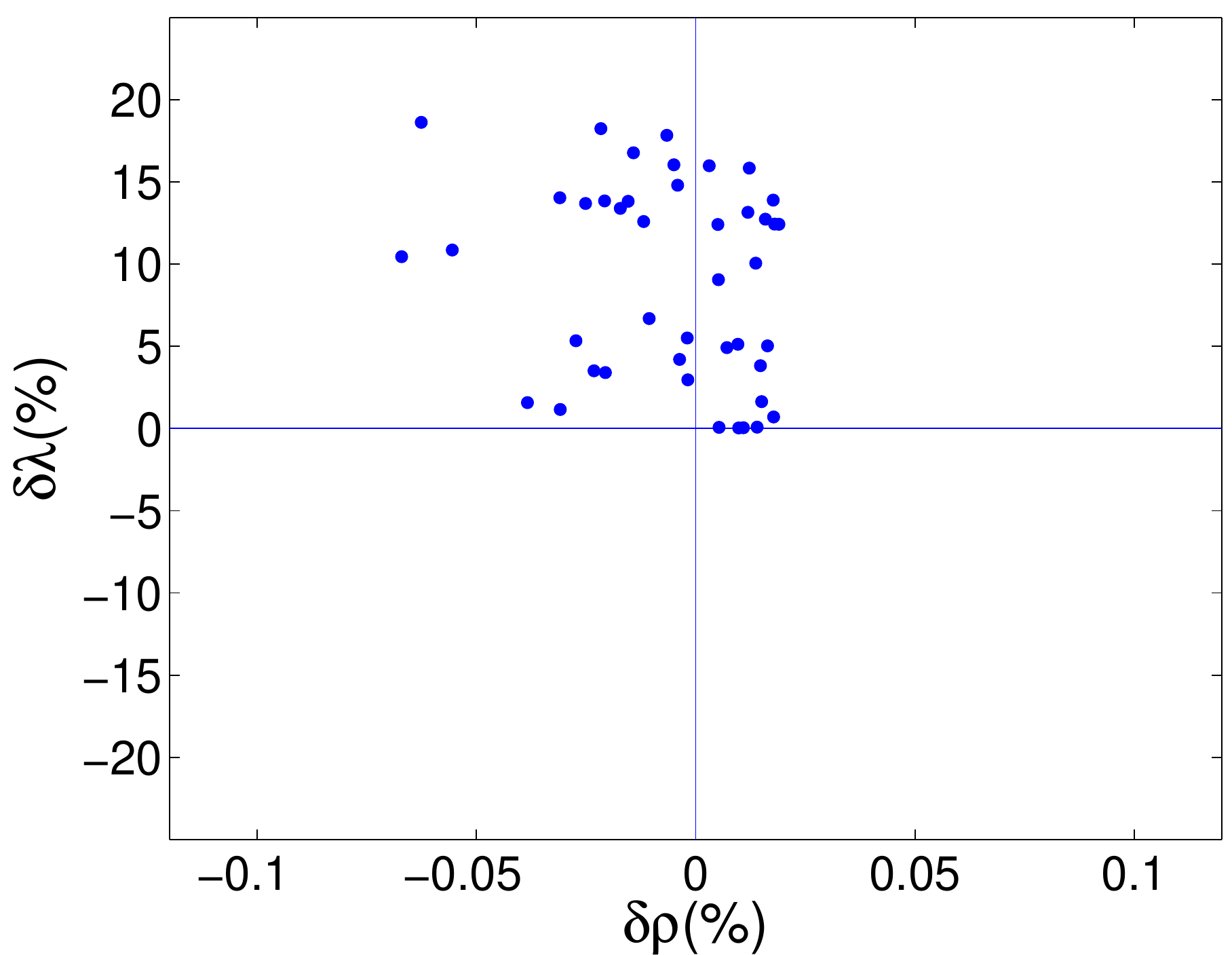}
\hspace{0.03\textwidth}
\includegraphics[width=0.45\textwidth]{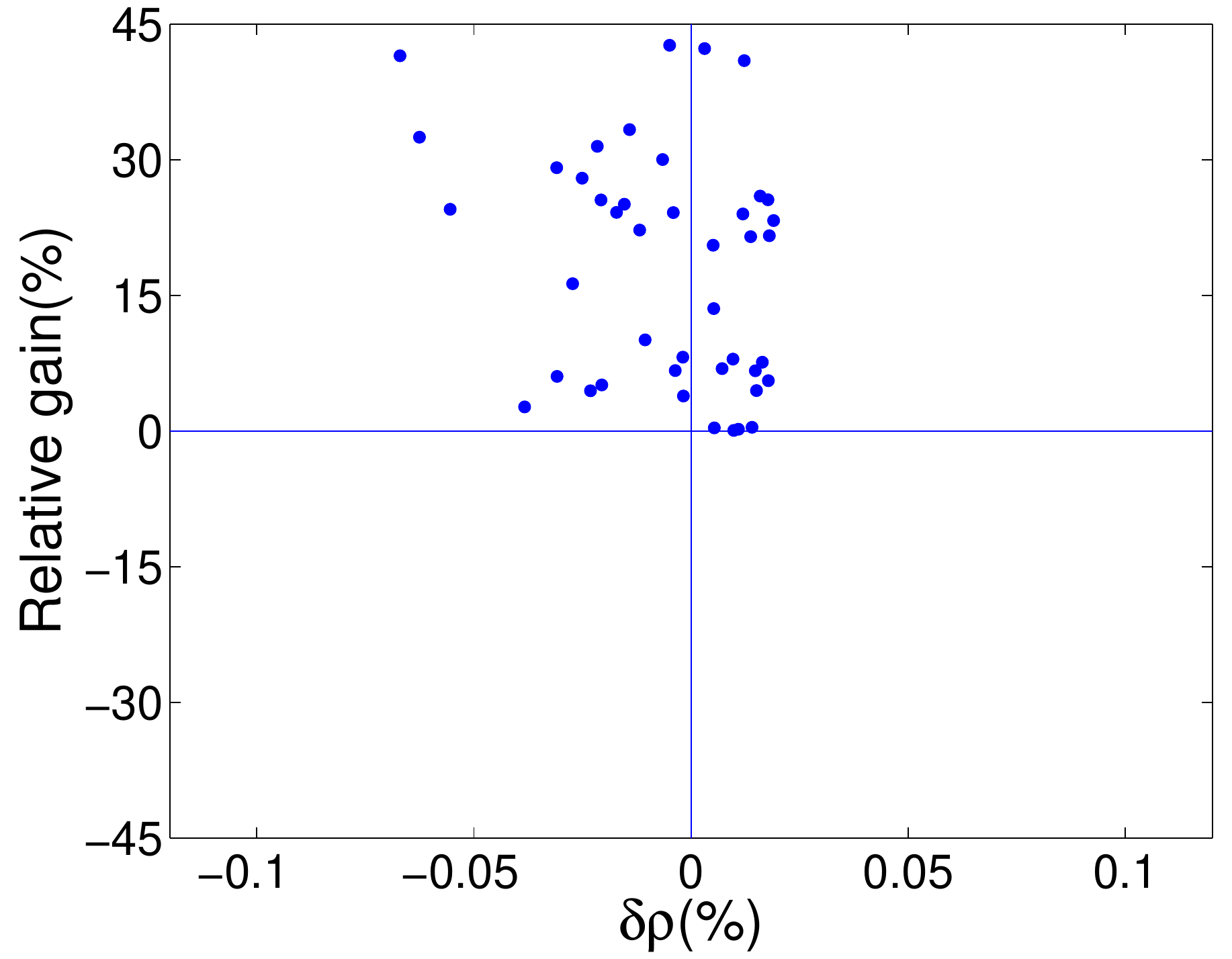}
\hspace{0.03\textwidth}
\caption{\label{fig:scat} Absolute difference $\delta\lambda$ and Relative gain versus the maximum load imbalance $\delta\rho$.}
\end{figure*}

{\it Speed of partitioning}. We now turn our attention to the efficiency 
of our method with respect to the running time to partition a graph. Our graph 
partitioning algorithm is a one-pass streaming algorithm, which allows for fast graph partitioning. 
In order to support this claim, in Figure~\ref{fig:time}, we show the run time it took to partition each graph 
from our dataset vs. the graph size in terms of the number of edges. We observe 
that it takes in the order of minutes to partition large graphs of tens of millions of edges. 
As we also mentioned before, partitioning the largest graph from our dataset collection took about 40 minutes. 

\begin{figure}[t]
\centering
\includegraphics[width=0.5\textwidth]{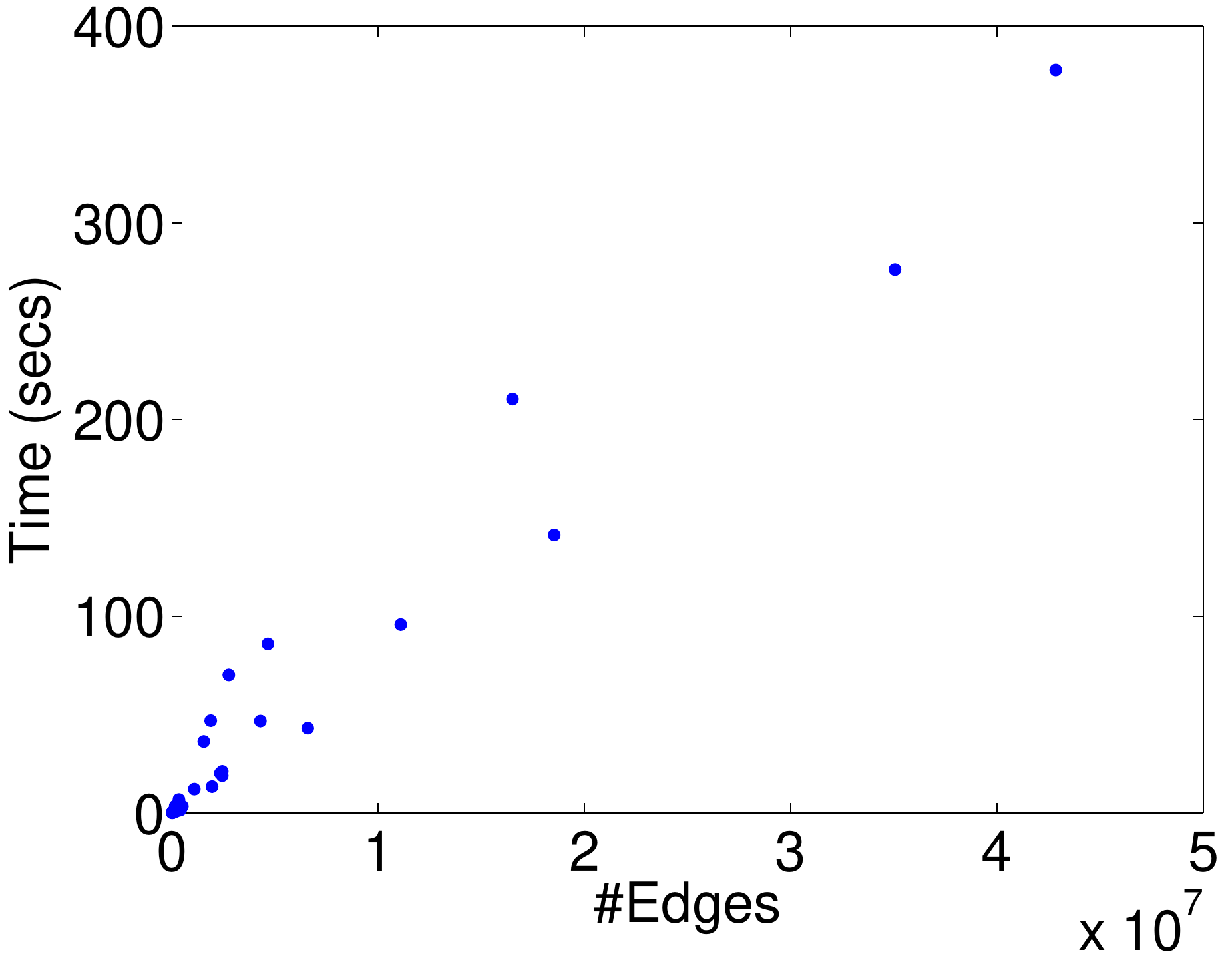}
\caption{\label{fig:time} \textsf{Fennel}: time vs. number of edges.}
\end{figure}

\section{System Evaluation}
\label{sec:Fennelsystem}

\begin{table*}
\centering 
\begin{tabular}{|c|c|c|c|c|}  \hline
\multicolumn{1}{|c|}{} & \multicolumn{2}{|c|}{\textsf{Run time [s]}}  &  
\multicolumn{2}{|c|}{\textsf{Communication [MB]}} \\[1pt] 
\cline{1-5}
\multicolumn{1}{|c|}{\# Clusters ($k$)} &  
\multicolumn{1}{c}{Hash} &   \multicolumn{1}{c|}{{\bf Fennel} } & 
\multicolumn{1}{c}{Hash} &   \multicolumn{1}{c|}{{\bf Fennel} } \\
\cline{1-5}
 4 & 32.27 &   {\bf 25.49} & 321.41 &  {\bf 196.9} \\ \hline
 8 & 17.26  & {\bf 15.14} & 285.35  & {\bf 180.02} \\ \hline
 16 & 10.64   & {\bf 9.05} & 222.28 &  {\bf 148.67} \\ \hline
\end{tabular} 
\caption{\label{tab:Fennelsystem} The average duration of a step and the average amount of MB exchanged per node and per step during
 the execution of PageRank on LiveJournal data set. }
\end{table*} 
Evaluating a partitioning algorithm is not an easy task from a systems perspective,
since it depends on system characteristics. 
For instance, in a large-scale production data-center it is more 
important to balance the traffic across clusters than the traffic or the amount of 
computation executed within a cluster. 
However, in a small-scale cluster consisting of few tens of nodes, 
rented by a customer from a large cloud provider such as  Amazon's Elastic Map-Reduce,
it is important to minimize the network traffic across the nodes but also to balance well 
the computational load on each node. 
Given this diversity of scenaria, a detailed evaluation is out of the scope of this work. 
Here, we perform a basic experiment to verify the superiority of our proposed method
versus the de facto standard of hash partitioning with respect to speeding up 
a large-scale computation. 

We select Pagerank as the computation of interest. 
Notice, that an advantage of Fennel is that it gives a flexibility in choosing a 
suitable objective that accomodates the needs of the specific application. 
We demonstrate the efficiency and flexibility of Fennel with the typical Elastic Map-Reduce 
scenario in mind.  We set up a cluster and we vary the number of nodes to 4, 8 and 16 nodes. 
Each node is equipped with Intel Xeon CPU at 2.27 GHz and 12 GB of main memory. 
On the cluster we run Giraph, a graph processing platform running on the top of Hadoop. 
We implemented a PageRank algorithm on Giraph and we run it on Live Journal data set\footnote{Twitter data set was too large to fit on a 16-nodes Giraph cluster.}. 

Since the complexity of PageRank depends on the number of edges and not vertices, we use a version of the Fennel 
objective (eq.~\ref{fennel}) that balances the number of edges per cluster. In particular, we choose the 
$c_{\mathrm{IN}}(\mathcal{P}) = \sum_{i=1}^k e(S_i, S_i)^\gamma$ with $\gamma = 1.5$. 

We compare with hash partitioning, the default partitioning scheme used by Giraph.
We look at two metrics. The first is the average duration of an iteration of the PageRank algorithm. 
This metric is directly proportional to the actual running time and incorporates both the 
processing and the communication time. 
The second metric is the average number of Megabytes transmitted by a cluster node 
in each iteration. This metric directly reflects the quality of the cut and is proportional to the incurred network load. 

The results are shown in Table~\ref{tab:Fennelsystem}. We see that Fennel has the best run time in all cases. 
This is because it achieves the best balance between the computation and communication load. 
Hash partitioning takes 25\% more time than Fennel and it also has a much higher traffic load.

\section{Discussion}
\label{sec:Fenneldiscussion}

In this section we discuss some of the extensions that can be accomodated by our framework and discuss some details about distributed implementation.

{\it Assymetric edge costs.} As discussed in Section~\ref{sec:Fennelsystem}, 
in some application scenarios some edges that cross partition boundaries 
may be more costly than other. For example, this is the case if individual 
graph partitions are assigned to machines in a data center and these machines 
are connected with an asymmetric network topology, so that the available network 
bandwidth varies across different pairs of machines, e.g. intra-rack vs. inter-rack
 machines in standard data center architectures~\cite{Greenberg09}. Another example
 are data center network topologies where the number of hops between different pairs
 of machines vary substantially, e.g. torus topologies~\cite{costa}. In such scenarios, 
it may be beneficial to partition a graph by accounting for the aforementioned asymmetries 
of edge-cut costs. This can be accomodated by appropriately defining the inter-partition 
cost function in our framework.

{\it Distributed implementation}. Our streaming algorithm requires computing marginal value 
indices that can be computed in a distributed fashion by maintaining local views on a global state. 
For concretness, let us consider the traditional objective where the inter-partition cost is a linear 
function of the total number of cut edges and the intra-partition cost is a sum of convex functions 
of vertex cardinalities of individual partitions. In this case, computing the marginal value indices 
requires to compute per each vertex arrival: (1) the number of neighbors of given vertex that were 
already assigned to given cluster of vertices, and (2) the number of vertices that were already 
assigned per cluster. The former corresponds to a set-intersection query and can be efficiently 
implemented by standard methods, e.g. using data structures such as minwise hashing~\cite{satuluri}. 
The latter is a simple count tracking problem. Further optimizations could be made by trading accuracy 
for reduction of communication overhead by updating of the local views at a smaller rate than the rate 
of vertex arrival. 

\clearpage
\chapter{PEGASUS: A System for Large-Scale Graph Processing}
\label{pegasuschapter}
\lhead{\emph{PEGASUS: A System for Large-Scale Graph Processing}} 
\section{Introduction} 
\label{sec:pegasusintro} 

In this Chapter we describe \pegasus, an open source Peta Graph Mining library which 
performs typical graph mining tasks such as computing 
the diameter of a graph, computing the radius of each node, finding the connected components, 
(see also Chapter~\ref{hadichapter}), and computing the importance score of nodes.
The main idea behind \pegasus is to capitalize on matrix-vector multiplication 
as a main primitive for the software engineer. 
Inspired by the work of \cite{Tsourakakis:2008:FCT:1510528.1511415}
which showed that triangles can be estimated by few matrix-vector multiplications, 
\pegasus introduces a set of different operators which solve a variety of graph mining tasks
together with an optimized implementation of matrix-vector multiplications in \mapreduce. 
\pegasus is a solid engineering effort which allows us to manipulate large-scale graphs. 
Since the introduction of \pegasus, other large-scale graph processing systems
have been introduced, among them Google's Pregel \cite{malewicz}, Linkedin's Giraph \cite{giraph}
and GraphLab \cite{DBLP:conf/uai/LowGKBGH10}. 
It is worth mentioning that Giraph uses several algorithms and ideas from \pegasus, including
the connected components algorithm. Also, \pegasus has been included in \hadoop for Windows Azure \cite{azur}.

\spara{Outline:} This Chapter is organized as follows: 
 Section~\ref{sec:pegasusidea} presents the proposed method.
Section~\ref{sec:pegasushadoop} presents \hadoop implementations and Section~\ref{sec:pegasusscalability} 
timings. Section~\ref{sec:pegasusatwork} shows findings of \pegasus in several real-world networks.

\section{Proposed Method} 
\label{sec:pegasusidea} 

Consider the following assignment $v' \leftarrow M \times v$ 
where $M \in \field{R}^{m \times n}, v \in \field{R}^n$. 
The $i$-th coordinate of $v'$ is $v'_i = \sum_{j=1}^n m_{i,j}v_j$,
$i=1,\ldots,m$. Typically in our applications, $M$ is the adjacency
matrix represetation of a graph and therefore we are going to assume
in the following that $m=n$, unless otherwise noticed.

There are three types of operations in the previous formula:
\begin{enumerate}
  \item \join: multiply $m_{i,j}$ and $v_j$.
  \item \aggrnp: sum $n$ multiplication results for node $i$.
  \item \assign: overwrite the previous value of $v_i$ with the new result to make $v'_i$.
\end{enumerate}

We introduce an abstraction of the basic matrix-vector multiplication, 
called  \gmvm. The corresponding programming primitive is the \IGMV primitive
on which \pegasus is based. The `Iterative' in  \IGMV denotes
that we apply the $\gimv_{G}$ operation until
a  convergence criterion is met.
Specifically, let us define the operator $\gimv_{G}$ as follows:

\begin{quote}
$v' = M \gimv_{G} v $ \\
where $v'_i$ = \assign$(v_i, $\aggr$(\{ x_j \mid j=1..n$, and $x_j =  $\join$(m_{i,j},v_j) \}))$.
\end{quote}

The functions \join(), \aggrnp(), and \assign() have the following interpretation,
generalizing the product, sum and assignment of the traditional matrix-vector multiplication: 

\begin{enumerate}
  \item \join$(m_{i,j}, v_j)$ : combine $m_{i,j}$ and $v_j$.
  \item \aggr$(x_1, ..., x_n)$ : combine all the results from \join() for node $i$.
  \item \assign$(v_i, v_{new})$ : decide how to update $v_i$ with $v_{new}$.
\end{enumerate}

In the following sections we show how different choices of \join(), \aggr() and \assign()
allow us to solve several important graph mining tasks. 
Before that, we want to highlight the strong connection of \IGMV\ with SQL.
When \aggr() and \assign() can be implemented by user defined functions,
the operator $\gimv_{G}$ can be expressed concisely in terms of SQL.
This viewpoint is important when we implement \IGMV
in large-scale parallel processing platforms, including \hadoop,
if they can be customized to support several SQL primitives including JOIN and GROUP BY.
Suppose we have an {\tt edge} table {\tt E(sid, did, val)}
and a {\tt vector} table {\tt V(id, val)},
corresponding to a matrix and a vector, respectively.
Then,
$\gimv_{G}$ corresponds to the 
SQL statement in Table~\ref{tab:gimvsql}.
We assume
that we have (built-in or user-defined) functions, \aggr() and \join(),
and we also assume
that the resulting table/vector will be fed into the
\assign() function (omitted, for clarity).

\begin{table}
        \centering
        \begin{tabular}{|l|}\hline
SELECT E.sid, \aggrsid(\join(E.val,V.val))\\
~ FROM E, V  \\
~ WHERE E.did=V.id \\
~ GROUP BY E.sid \\
\hline
        \end{tabular}
        \caption{GIM-V in terms of SQL.}
        \label{tab:gimvsql}
\end{table}

In the following sections we show how we can customize
\IGMV, to handle important graph mining operations including
PageRank, Random Walk with Restart, diameter estimation, and connected components.

\subsection{\IGMV\ and PageRank}
\label{sec:gimv_pagerank}

Our first warm-up application of \IGMV is PageRank, 
a famous algorithm that was used by Google to calculate relative importance of web pages \cite{Brin98anatomy}.
The PageRank vector $p$ of $n$ web pages satisfies the following
eigenvector equation:

\begin{quote}
$p = (cE^T + (1-c)U)p$
\end{quote}

where $c$ is a damping factor (usually set to 0.85),
$E$ is the row-normalized adjacency matrix (source, destination),
and $U$ is a matrix with all elements set to $1/n$.

To calculate the eigenvector $p$ we can use the power method, which multiplies an initial vector with the matrix,
several times.
We initialize the current PageRank vector $p^{cur}$ and set
all its elements to $1/n$.
Then the next PageRank $p^{next}$ is calculated by $p^{next} = (cE^T + (1-c)U)p^{cur}$.
We continue to perform the multiplication until $p$ converges.

PageRank is a direct application of \IGMV, i.e., $p^{next} = M \gimv_G p^{cur}$.
Matrix $M$ is $E^T$, i.e., the column-normalized version of the adjacency matrix. 
The three operations are defined as follows:

\begin{enumerate}
  \item \join$(m_{i,j}, v_j)$ = $c \times m_{i,j} \times v_j $
  \item \aggr$(x_1, ..., x_n)$ = $\frac{(1-c)}{n} + \sum_{j=1}^n x_{j}$
  \item \assign$(v_i, v_{new})$ = $v_{new}$
\end{enumerate}

\subsection{\IGMV\ and Random Walk with Restart}

Random Walk with Restart (RWR) is closely related to Personalized pagerank, 
a popular algorithm to measure the relative proximity of vertices with respect to a given vertex 
\cite{Pan04Automatic}.
In RWR, the proximity vector $r_k$ of vertex $k$ satisfies the equation:

\begin{quote}
$r_k = cMr_k + (1-c)e_k$
\end{quote}

\noindent where $e_k$ is the $k$-th unit vector in $\field{R}^n$, 
$c$ is a restart probability parameter which is typically set to 0.85~\cite{Pan04Automatic}
and $M$ is  as in Section \ref{sec:gimv_pagerank}.
In \IGMV, RWR is formulated by $r_k^{next} = M \gimv_G r_k^{cur}$ 
where the three operations are defined as follows:
 
\begin{enumerate}
  \item \join$(m_{i,j}, v_j)$ = $c \times m_{i,j} \times v_j $
  \item \aggr$(x_1, ..., x_n)$ = $(1-c)\delta_{ik} + \sum_{j=1}^n x_{j}$, where 
   $\delta_{ik}$ is the \emph{Kronecker delta}, equal to 1 if $i=k$ and 0 otherwise
  \item \assign$(v_i, v_{new})$ = $v_{new}$
\end{enumerate}

\subsection{\IGMV\ and Diameter Estimation}

In Chapter~\ref{hadichapter} we discussed \hadi, an algorithm that estimates the diameter and radius distribution
of a large-scale graph.  \hadi can be presented within the framework of \pegasus,
since the number of neighbors reachable from vertex $i$ within $h$ hops 
is encoded in a probabilistic bitstring $b^{h}_i$ which is updated as follows \cite{flajolet85probabilistic}:

\begin{quote}
$b^{h+1}_i$ = $b^{h}_i$ BITWISE-OR $\{b^{h}_k \mid (i,k) \in E\}$
\end{quote}

\noindent In \IGMV, the bitstring update of \hadi is represented by

\begin{quote}
$b^{h+1} = M \gimv_G b^{h}$
\end{quote}

\noindent where $M$ is the adjacency matrix, $b^{h+1}$ is a vector of length $n$ which is updated by\\
$b^{h+1}_i = $\assign$(b^{h}_i, $\aggr$(\{ x_j \mid j=1..n$, and $x_j =  $\join$(m_{i,j},b^{h}_j) \}))$, \\
and the three \pegasus operations are defined as follows:

\begin{enumerate}
  \item \join$(m_{i,j}, v_j)$ = $m_{i,j} \times v_j$.
  \item \aggr$(x_1, ..., x_n)$ = BITWISE-OR$\{ x_j \mid j=1..n \}$
  \item \assign$(v_i, v_{new})$ = BITWISE-OR$(v_i,v_{new})$.
\end{enumerate}

The $\gimv_G$ operation is run iteratively until the bitstring of each vertex remains the same.

\subsection{\IGMV\ and  Connected Components}

We propose \hcc, a new algorithm for finding connected components in large graphs.
The main idea is as follows: for each vertex $i$ in the graph, 
we maintain a component identification number (id) $c^h_i$ which is the 
minimum vertex id within $h$ hops from $i$.

Initially, $c^h_i$ of vertex $i$ is set to $i$, i.e., $c^0_i = i$.
In each iteration, each vertex sends its current $c^h_i$ to its neighbors.
Then $c^{h+1}_i$ is set to the minimum value among its current component 
id and the received component ids from its neighbors.
The crucial observation is that this communication between neighbors 
can be formulated in \IGMV as follows:

\begin{quote}
$c^{h+1} = M \gimv_G c^{h}$
\end{quote}

where $M$ is the adjacency matrix, $c^{h+1}$ is a vector of length $n$ which is updated by
$c^{h+1}_i = $\assign$(c^{h}_i, $\aggr$(\{ x_j \mid j=1..n$, and $x_j =  $\join$(m_{i,j},c^{h}_j) \}))$, 
and the three \pegasus operations are defined as follows:

\begin{enumerate}
  \item \join$(m_{i,j}, v_j)$ = $m_{i,j} \times v_j$.
  \item \aggr$(x_1, ..., x_n)$ = $\min \{ x_j \mid j=1..n \}$.
  \item \assign$(v_i, v_{new})$ = $\min (v_i,v_{new})$.
\end{enumerate}

By repeating this process, component ids of nodes in a component are set to the minimum node id of the component.
We iteratively do the multiplication until component ids converge.
The upper bound of the number of iterations in \hcc is $d$, where $d$ is the diameter of the graph. 
We notice that because of the small-world phenomenon, see Section~\ref{sec:hadiintro}, 
the diameter of real graphs is small, and therefore in practice \hcc completes after a 
small number of iterations.
For a recent work with better practical performance, see \cite{seidl2012cc}.

\section{\hadoop Implementation} 
\label{sec:pegasushadoop} 

Given the main goal of the \pegasus project is to provide an efficient system to the user/programmer, 
we discuss different \hadoop implementation approaches, starting out with a naive implementation
and progressing to faster methods for \IGMV.
The proposed versions are evaluated in Section~\ref{sec:pegasusscalability}. 

\subsection{\IGMV BASE: Naive Multiplication}

\IGMV BASE is a two-stage algorithm whose pseudo code is in Algorithm~\ref{alg:igmvs1} and \ref{alg:igmvs2}.
The inputs are an edge file and a vector file. 
Each line of the edge file has the form $(id_{src}, id_{dst}, mval)$ 
which corresponds to a non-zero entry in the djacency matrix. 
Similarly, each line of the vector file has the form  $(id, vval)$ 
which corresponds to an element in vector $v$.
\PassA performs the \join operation by combining columns of 
matrix ($id_{dst}$ of $M$) with rows of the vector ($id$ of $V$).
The output of \PassA are (key, value) pairs where the key is the 
source vertex id of the matrix ($id_{src}$ of $M$) and 
the value is the partially combined result (\join($mval, vval$)).
This output of \PassA becomes the input of \PassB.
\PassB combines all partial results from \PassA and updates the vector. 
The \aggr() and \assign() operations are done in line 15 of \PassB, where
the ``self'' and ``others'' tags in line 15 and line 21 of \PassA are
needed by \PassB to distinguish cases appropriately. 
We note that in Algorithm~\ref{alg:igmvs1} and \ref{alg:igmvs2}, 
Output($k$, $v$) means to output data with the key $k$ and the value $v$.

\begin{algorithm}[!t]
\begin{algorithmic}[1]
\REQUIRE Matrix $M=\{ (id_{src}, (id_{dst}, mval)) \}$,
    Vector $V=\{(id, vval)\}$
\ENSURE Partial vector $V'=\{(id_{src}, \join(mval, vval)\}$
\STATE Stage1-Map(Key $k$, Value $v$):
\IF{$(k,v)$ is of type V}
    \STATE Output($k, v$); \hfill // ($k$: $id$, $v$: $vval$)
\ELSIF{($k,v$) is of type M}
    \STATE $(id_{dst}, mval) \leftarrow v$;
    \STATE Output($id_{dst}, (k,mval)$); \hfill // ($k$: $id_{src}$)
\ENDIF
\STATE
\STATE Stage1-Reduce(Key $k$, Value $v[1..m$]):
\STATE $saved\_kv \leftarrow $[ ];
\STATE $saved\_v \leftarrow $[ ];
\FOR{$v \in v[1..m]$}
    \IF{$(k,v)$ is of type V}
        \STATE $saved\_v \leftarrow v$;
        \STATE Output($k, (``self",saved\_v)$);
    \ELSIF{($k,v$) is of type M}
        \STATE Add $v$ to $saved\_kv$; \hfill // (v: $(id_{src}, mval)$)
    \ENDIF
\ENDFOR
\FOR{$(id_{src}', mval') \in saved\_kv$}
    \STATE Output($id_{src}',$ $(``others", $\join$(mval',saved\_v))$);
\ENDFOR
\end{algorithmic}
\caption{\IGMV BASE Stage 1.}
\label{alg:igmvs1}
\end{algorithm}

\begin{algorithm}[!t]
\begin{algorithmic}[1]
\REQUIRE Partial vector $V'=\{(id_{src}, vval')\}$
\ENSURE Result Vector $V=\{(id_{src}, vval)\}$
\STATE Stage2-Map(Key k, Value v):
\STATE Output($k, v$);
\STATE
\STATE Stage2-Reduce(Key k, Value v[1..m]):
\STATE $others\_v \leftarrow  $[ ];
\STATE $self\_v \leftarrow  $[ ];
\FOR{$v \in v[1..m]$}
    \STATE $(tag, v') \leftarrow v$;
    \IF{$tag$ = ``same"}
        \STATE $self\_v \leftarrow v'$;
    \ELSIF{$tag$ = ``others"}
        \STATE Add $v'$ to $others\_v$;
    \ENDIF
\ENDFOR
\STATE Output($k, $\assign$(self\_v, $\texttt{combineAll$_k$}$(others\_v))$);
\end{algorithmic}
\caption{\IGMV BASE Stage 2.}
\label{alg:igmvs2}
\end{algorithm}

\subsection{\IGMV BL: Block Multiplication}
\label{sec:method_bl}

\IGMV BL is a fast algorithm for \IGMV which is based on block multiplication.
The main idea is to group elements of the input matrix into blocks/submatrices of size $b$ by $b$. 
Also, we group elements of input vectors into blocks of length $b$.
In practice, grouping means we place all elements of a group into one line of input file.
Each block contains only non-zero elements of the matrix/vector.
The format of a matrix block with $k$ nonzero elements is 
($row_{block}, col_{block}, row_{elem_1},$ $col_{elem_1}, mval_{elem_1}, ...,$ $row_{elem_k}, col_{elem_k}, mval_{elem_k}$).
Similarly, the format of a vector block with $k$ nonzero elements is 
($id_{block}, id_{elem_1}, vval_{elem_1}, ..., id_{elem_k},$ $vval_{elem_k}$).
Only blocks with at least one nonzero elements are saved to disk.
This block encoding forces nearby edges in the adjacency matrix to be closely 
located; it is different from \hadoop's default behavior which does not guarantee co-locating them.
After grouping, \IGMV is performed on blocks, not on individual elements.
\IGMV BL is illustrated in Figure~\ref{fig:igmv_bl}.

\begin{figure}
\begin{center}
  \includegraphics[width=0.6\textwidth]{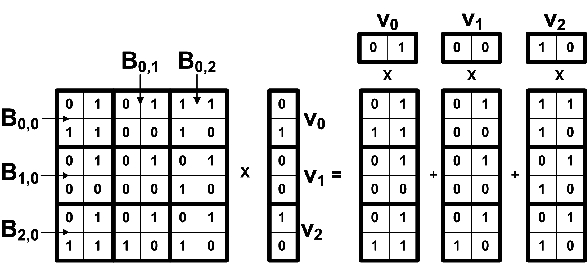}
  \caption{ \label{fig:igmv_bl}
  \IGMV BL using 2 x 2 blocks.
  B$_{i,j}$ represents a matrix block, and v$_i$ represents a vector block.
  The matrix and vector are joined block-wise, not element-wise.
  }
\end{center}
\end{figure}


In Section~\ref{sec:pegasusscalability}, we observe that \IGMV BL is 
at least 5 times faster than \IGMV BASE. 
There are two main reasons for this speed-up.

\begin{itemize}
  \item \textbf{Sorting Time} Block encoding decreases the number of items to be sorted in the shuffling stage of \hadoop.
We observe that one of the main efficiency bottlenecks  in \hadoop is 
its shuffling stage where network transfer, sorting, and disk I/O take place. 
  \item \textbf{Compression} The size of the data decreases significantly by converting edges and vectors to block format.
The reason is that in \IGMV BASE we need $2 \times 4=8$ bytes to save each (srcid, dstid) pair. 
However in \IGMV BL we can specify each \textit{block} using a block row id and a block column id with 
two 4-byte Integers, and refer to elements inside the block using $2 \times \log{b}$ bits.
This is possible because we can use $\log{b}$ bits to refer to a row or column inside a block.
By this block method we decrease the edge file size. For instance, using 
block encoding we are able to decrease the size of the YahooWeb graph more than 50\%. 
\end{itemize}


\subsection{\IGMV CL: Clustered Edges}

We use co-clustering heuristics, see \cite{PapadimitriouICDM08} 
as a preprocessing step to obtain a better clustering of the edge set. 
Figure~\ref{fig:igmv_cl} illustrates the concept. 
The preprocessing step needs to be performed only once. 
If the number of iterations required for the execution of an algorithm 
is large, then it is beneficial to perform this preprocessing step. 
Notice that we have two variants of \IGMV: 
\IGMV CL and \IGMV BL-CL, which are \IGMV BASE and \IGMV BL with clustered edges 
respectively. 

\begin{figure}
\begin{center}
  \includegraphics[width=0.6\textwidth]{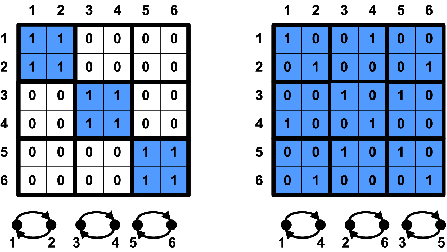}
  \caption{ \label{fig:igmv_cl}
  Clustered vs. non-clustered adjacency matrices for two isomorphic graphs. The edges are grouped into 2 by 2 blocks.
  The left graph uses only 3 blocks while the right graph uses 9 blocks.
  }
\end{center}
\end{figure}

\subsection{\IGMV DI: Diagonal Block Iteration}

Reducing the number of iterations required for executing an algorithm 
in \mapreduce mitigates the computational cost a lot, 
since the main bottleneck of \IGMV is its shuffling and disk I/O steps. 
In \hcc, it is possible to decrease the number of iterations
when the graph has long chains. The main idea is to multiply diagonal matrix blocks 
and corresponding vector blocks as much as possible in one iteration.
This is illustrated in Figure~\ref{fig:cc_propagation}.

\begin{figure}
\begin{center}
  \includegraphics[width=0.6\textwidth]{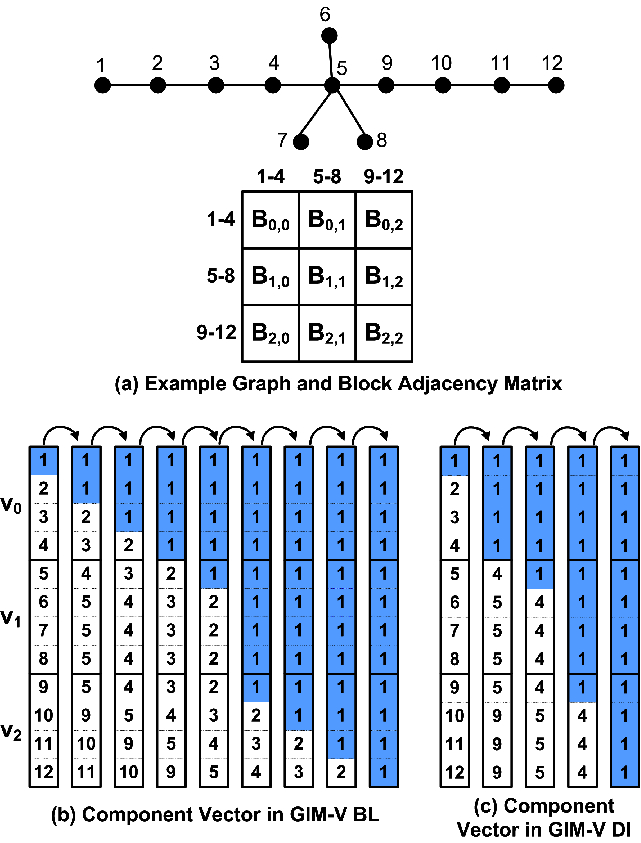}
  \caption{ \label{fig:cc_propagation}
  Propagation of component id(=1) when block width is 4.
  Each element in the adjacency matrix of (a) represents a 4 by 4 block;
  each column in (b) and (c) represents the vector after each iteration.
  \IGMV DL finishes in 4 iterations while \IGMV BL requires 8 iterations.
  }
\end{center}
\end{figure}

\subsection{\IGMV NR: Node Renumbering}
In HCC, the minimum vertex id is propagated to the other parts of 
the graph within at most $d$ steps, where $d$ is the diameter of the graph.
If the vertex with the minimum id (which we call `minimum node') 
is located at the center of the graph, then the number of iterations 
is small, close to $d$/2. However, if it is located at the boundary 
of the network, then the number of iteration can be close to $d$.
Therefore, if we preprocess the edges so that the minimum vertex 
id is swapped to the center vertex id, the number of iterations and the total running time of HCC would decrease.

Finding the center vertex with the minimum radius could be done with the \hadi algorithm. 
However, the algorithm is expensive for the pre-processing step of HCC. Therefore, we instead propose 
the following heuristic for finding the center node: we choose the center vertex by 
sampling from the high-degree vertices.
This heuristic is based on the fact that vertices with large degree have 
small radii~\cite{HadiSDM2010}. 

After finding a center node, we need to renumber the edge file to swap the 
current minimum vertex id with the center vertex id. 
The \mapreduce algorithm for this renumbering is shown 
in Algorithm~\ref{alg:renumbering}. Since the renumbering requires only 
filtering, it can be done with a Map-only job.

\begin{algorithm}[!t]
\begin{algorithmic}[1]
\REQUIRE Edge $E=\{(id_{src}, id_{dst})\}$, \\
    current minimum vertex id $minid_{cur}$, \\
    new minimum vertex id $minid_{new}$
\ENSURE Renumbered Edge $V=\{(id_{src}', id_{dst}')\}$
\STATE Renumber--Map(key $k$, value $v$):
\STATE $src \leftarrow k$;
\STATE $dst \leftarrow v$;
\IF{$src$ = $minid_{cur}$}
    \STATE $src \leftarrow minid_{new}$;
\ELSIF{$src$ = $minid_{new}$}
    \STATE $src \leftarrow minid_{cur}$;
\ENDIF
\IF{$dst$ = $minid_{cur}$}
    \STATE $dst \leftarrow minid_{new}$;
\ELSIF{$dst$ = $minid_{new}$}
    \STATE $dst \leftarrow minid_{cur}$;
\ENDIF
\STATE Output($src, dst$); 
\end{algorithmic}
\caption{\label{alg:renumbering}Renumbering the minimum node}
\end{algorithm}

\subsection{Analysis}

Finally, we analyze the time and space complexity of \IGMV.
It is not hard to observe that 
one iteration of \IGMV takes $O(\frac{n+m}{M} \log \frac{n+m}{M})$ time,
where $M$ stands for the number of machines. 
Assuming uniformity, mappers and reducers of \PassA and \PassB receive $O(\frac{n+m}{M})$ records per machine.
The running time is dominated by the sorting time for $\frac{n+m}{M}$ records.
\IGMV requires $O(V+E)$ space.

\section{Scalability} 
\label{sec:pegasusscalability} 

We perform experiments to answer the following questions:

\begin{itemize}
  \item  How does \IGMV scale up?
  \item Which of the proposed optimizations (block multiplication, clustered edges, 
  and diagonal block iteration, vertex renumbering) gives the highest performance gains?
\end{itemize}

\begin{table}
\begin{center}
\begin{tabular}{|c|c|c|c|}
    \hline
   \textbf{Name} & \textbf{Vertices} & \textbf{Edges} & \textbf{Description}  \\ \hline \hline
   YahooWeb & 1,413 M & 6,636 M & WWW pages in 2002\\ \hline
   LinkedIn & 7.5 M & 58 M & person-person in 2006\\
            & 4.4 M & 27 M & person-person in 2005\\
            & 1.6 M & 6.8 M & person-person in 2004\\
            & 85 K & 230 K & person-person in 2003\\ \hline
   Wikipedia & 3.5 M & 42 M & doc-doc in 2007/02 \\
             & 3 M  & 35 M  & doc-doc in 2006/09 \\
             & 1.6 M  & 18.5 M  & doc-doc in 2005/11 \\  \hline
   Kronecker & 177 K & 1,977 M & synthetic \\
            & 120 K & 1,145 M &  synthetic \\
            & 59 K & 282 M &  synthetic \\
            & 19 K & 40 M &  synthetic\\ \hline
   WWW-Barabasi & 325 K & 1,497 K & WWW pages in nd.edu\\ \hline
   DBLP & 471 K &  112 K & document-document \\ \hline
   flickr & 404 K  & 2.1 M   & person-person \\ \hline
   Epinions &  75 K &  508 K & who trusts whom \\ \hline
 
\end{tabular}
\end{center}
\caption{Order and size of networks.
}
\label{tab:datasets}
\end{table}

\begin{figure*}
\begin{center}
\setlength{\tabcolsep}{0cm}
    \begin{tabular}{c c}
    \includegraphics[width=0.49\textwidth]{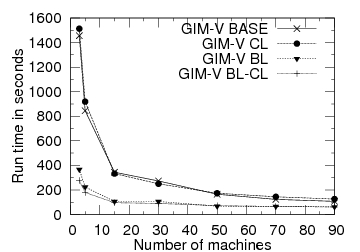}&
    \includegraphics[width=0.49\textwidth]{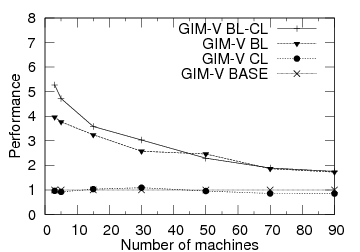} \\
      (a) Running time vs. Machines & (b) Performance vs. Machines \\
    \includegraphics[width=0.49\textwidth]{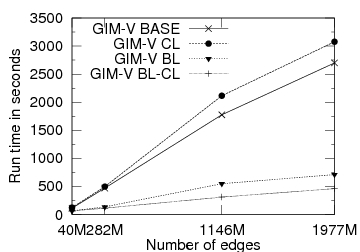} \\
     (c) Running time vs. Edges 
    \end{tabular}
    \caption{
    Scalability and Performance of GIM-V. 
    (a) Running time decreases quickly as more machines are added.
    (b) The performance(=$1/running$ $time$) of 'BL-CL'
    wins more than 5x (for n=3 machines) over the 'BASE'.
    (c) 
    Every version of \IGMV shows linear scalability.
    }
    \label{fig:runtime_performance_machines}
    \end{center}
\end{figure*}

The graphs we use in our experiments   are shown in Table~\ref{tab:datasets}. 
We run \pegasus in M45 \hadoop cluster by Yahoo! and our own cluster composed of 9 machines.
M45 is one of the top 50 supercomputers in the world with the total 1.5 Pb storage and 3.5 Tb memory.
For the performance and scalability experiments, we used synthetic Kronecker graphs~\cite{Leskovec05Realistic} 
since we can generate them with any size, and they are one of the most realistic graphs among synthetic graphs.

\subsection{Results}

We first show how the performance of our method changes as we add more machines. 
Figure~\ref{fig:runtime_performance_machines} shows the running time and performance 
of \IGMV for PageRank with Kronecker graph of 282 million edges, and size 32 blocks if necessary.

In Figure~\ref{fig:runtime_performance_machines} (a), for all of the methods the running time decreases as we add more machines.
Note that clustered edges(\IGMV CL) didn't help performance unless it is combined with block encoding.
When it is combined, however, it showed the best performance (\IGMV BL-CL).

In Figure~\ref{fig:runtime_performance_machines} (b), we see that the relative performance of each method compared to \IGMV BASE method decreases as number of machines increases.
With 3 machines (minimum number of machines which \hadoop `distributed mode' supports), the fastest method(\IGMV BL-CL) ran 5.27 times faster than \IGMV BASE. With 90 machines, \IGMV BL-CL ran 2.93 times faster than \IGMV BASE. This is expected since there are fixed component(JVM load time, disk I/O, network communication) which can not be optimized even if we add more machines.

Next we show how the performance of our methods changes as the input size grows. Figure~\ref{fig:runtime_performance_machines} (c) shows the running time of \IGMV with different number of edges under 10 machines. As we can see, all of the methods scales linearly with the number of edges.

Next, we compare the performance of \IGMV DI and \IGMV BL-CL for \hcc in graphs with long chains.
For this experiment we made a new graph whose diameter is 17, by adding a length 15 chain to the 282 million Kronecker graph which has diameter 2.
As we see in Figure~\ref{fig:gimv_di_bl}, \IGMV DI finished in 6 iteration while \IGMV BL-CL finished in 18 iteration.
The running time of both methods for the first 6 iterations are nearly same. Therefore, the diagonal block iteration method decreases the number of iterations while not affecting the running time of each iteration much.

\begin{figure}
\begin{center}
  \includegraphics[width=0.6\textwidth]{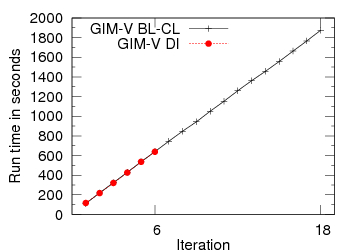}
  \caption{ \label{fig:gimv_di_bl}
Comparison of \IGMV DI and \IGMV BL-CL for \hcc. \IGMV DI finishes in 6 iterations while \IGMV BL-CL finishes in 18 iterations due to long chains.
  }
\end{center}
\end{figure}

Finally, we compare the number of iterations with/without renumbering.
Figure~\ref{fig:dd_li} shows the degree distribution of LinkedIn.
Without renumbering, the minimum vertex has degree 1, which is not surprising since about 46 \% of the vertices have degree 1 due to the power-law behavior of the degree distribution.
We show the number of iterations after changing the minimum vertex to each of the top 5 highest-degree vertices in Figure~\ref{fig:renum_li}.
We see that the renumbering decreased the number of iterations to 81 \% of the original.
Similar results are observed for the Wikipedia graph in Figure~\ref{fig:dd_wiki} and \ref{fig:renum_wikipedia}.
The original minimum vertex has degree 1, and the number of iterations decreased to 83 \% of the original after renumbering.

\begin{figure}
\begin{center}
  \includegraphics[width=0.6\textwidth]{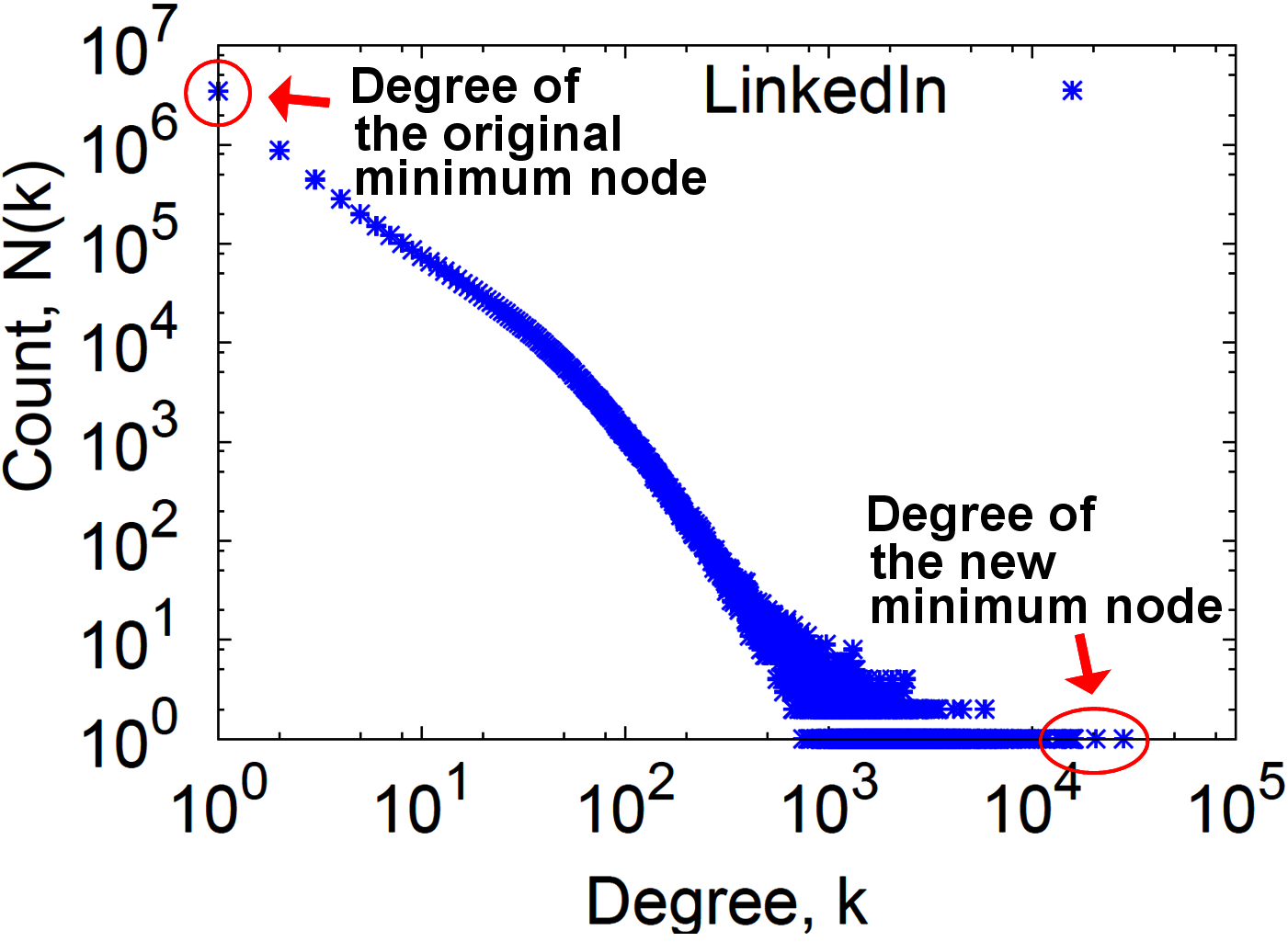}
  \caption{ \label{fig:dd_li}
Degree distribution of LinkedIn. Notice that the original minimum vertex has degree 1, which is highly probable given the power-law behavior of the degree distribution. After the renumbering, the minimum vertex is replaced with a highest-degree node.
  }
\end{center}
\end{figure}

\begin{figure}
\begin{center}
  \includegraphics[width=0.6\textwidth]{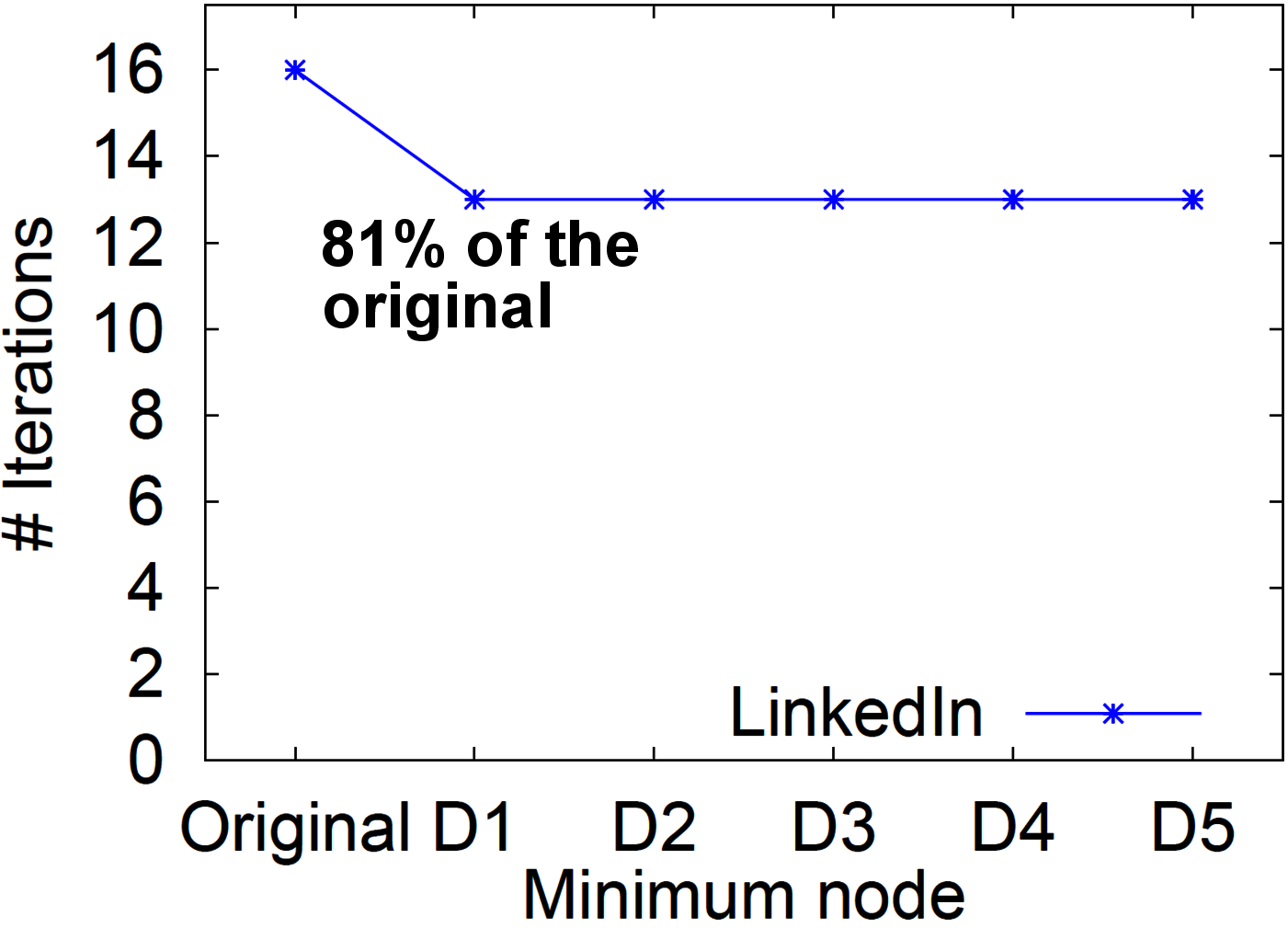}
  \caption{ \label{fig:renum_li}
Number of iterations vs. the minimum vertex of LinkedIn, for connected components.
D$i$ represents the vertex with $i$-th largest degree.
Notice that the number of iterations decreased by 19 \% after renumbering.
}
\end{center}
\end{figure}

\begin{figure}
\begin{center}
  \includegraphics[width=0.6\textwidth]{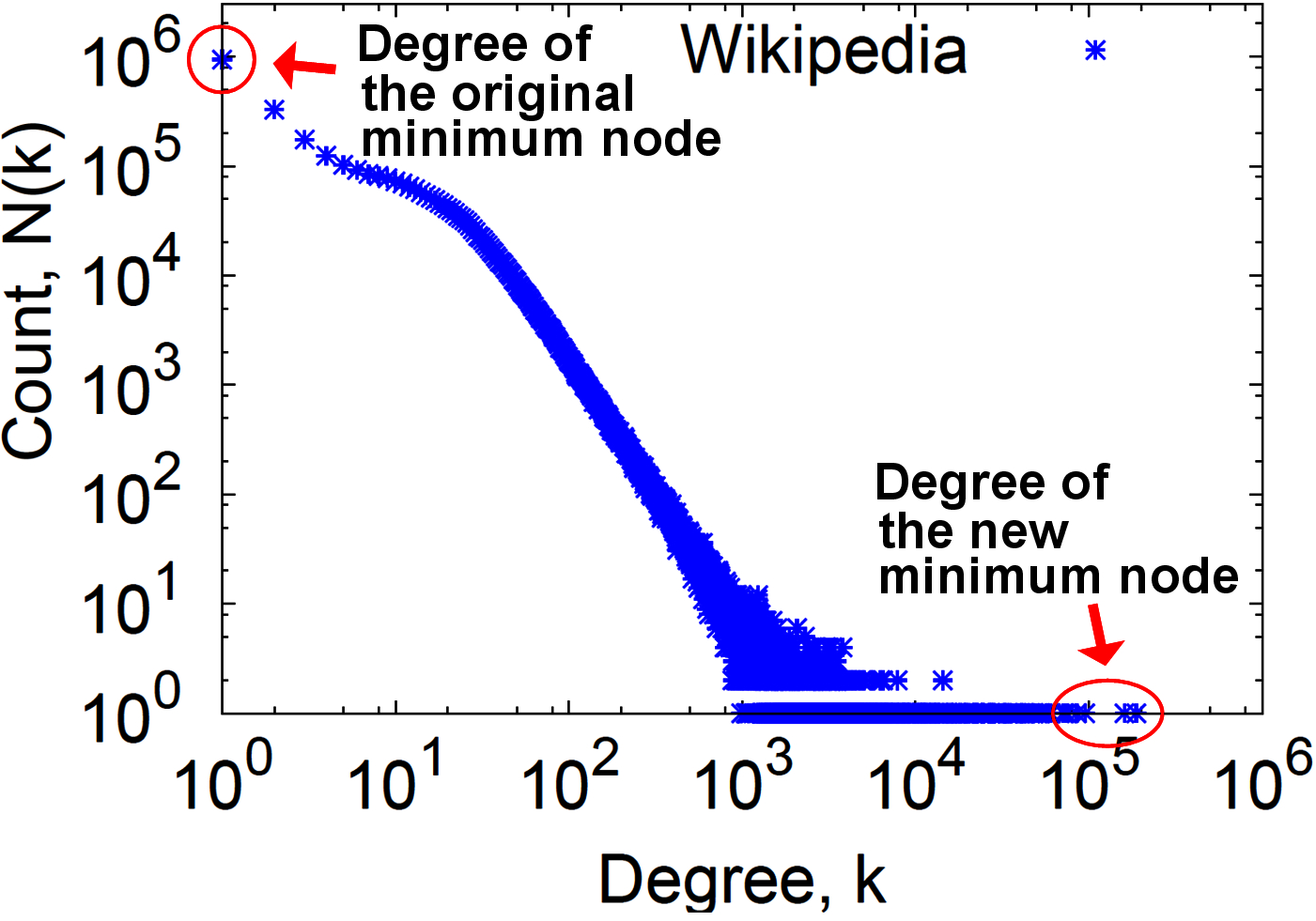}
  \caption{ \label{fig:dd_wiki}
Degree distribution of Wikipedia. Notice that the original minimum vertex has degree 1, as in LinkedIn. After the renumbering, the minimum vertex is replaced with a highest-degree node.
  }
\end{center}
\end{figure}

\begin{figure}
\begin{center}
  \includegraphics[width=0.6\textwidth]{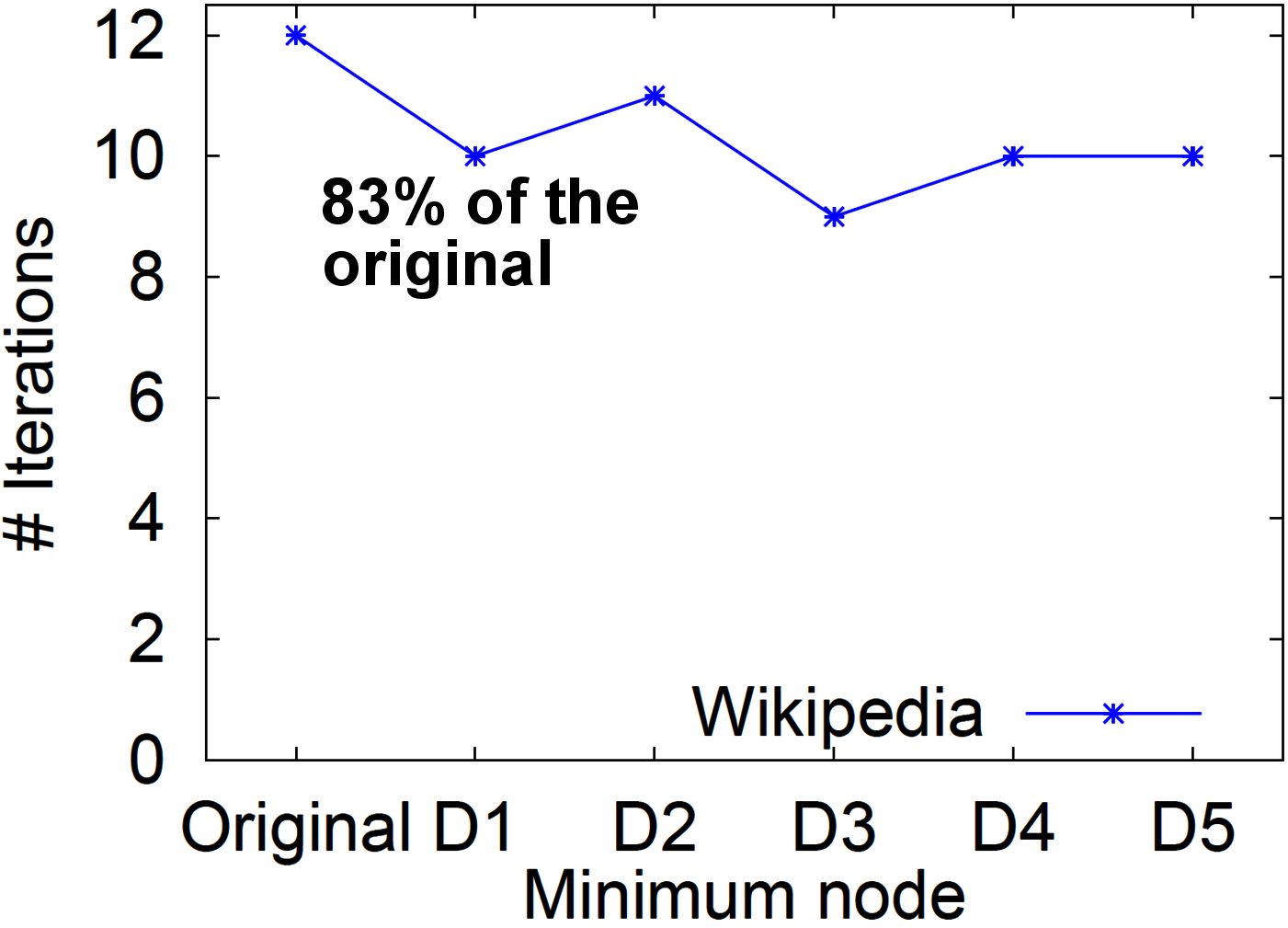}
  \caption{ \label{fig:renum_wikipedia}
Number of iterations vs. the minimum vertex of Wikipedia, for connected components.
D$i$ represents the vertex with $i$-th largest degree.
Notice that the number of iterations decreased by 17 \% after renumbering.
  }
\end{center}
\end{figure} 
 
\section{Pegasus at Work} 
\label{sec:pegasusatwork}

In this section we evaluate \pegasus on real-world networks. 

\subsection{Connected Components of Real Networks}
 
Figure~\ref{fig:evolution_cc} shows the evolution of connected components of LinkedIn and Wikipedia graphs. 
Figure~\ref{fig:yahoo_cc} shows the distribution of connected components in the YahooWeb graph. 
We make the following set of observations.

\begin{figure*}
\begin{center}
    \begin{tabular}{c}
    \includegraphics[width=0.7\textwidth]{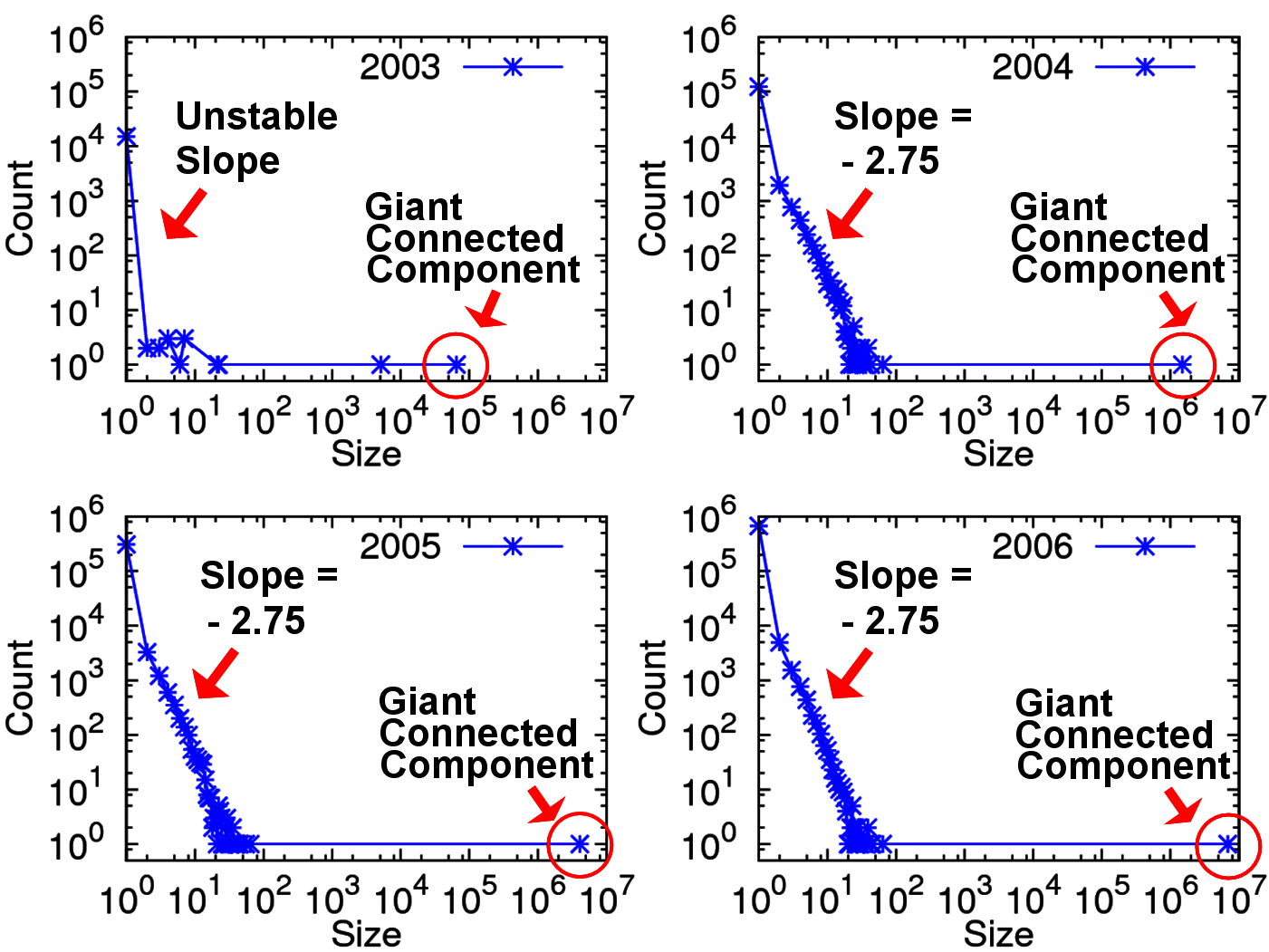}\\
     (a) Connected Components of LinkedIn \\
    \includegraphics[width=0.7\textwidth]{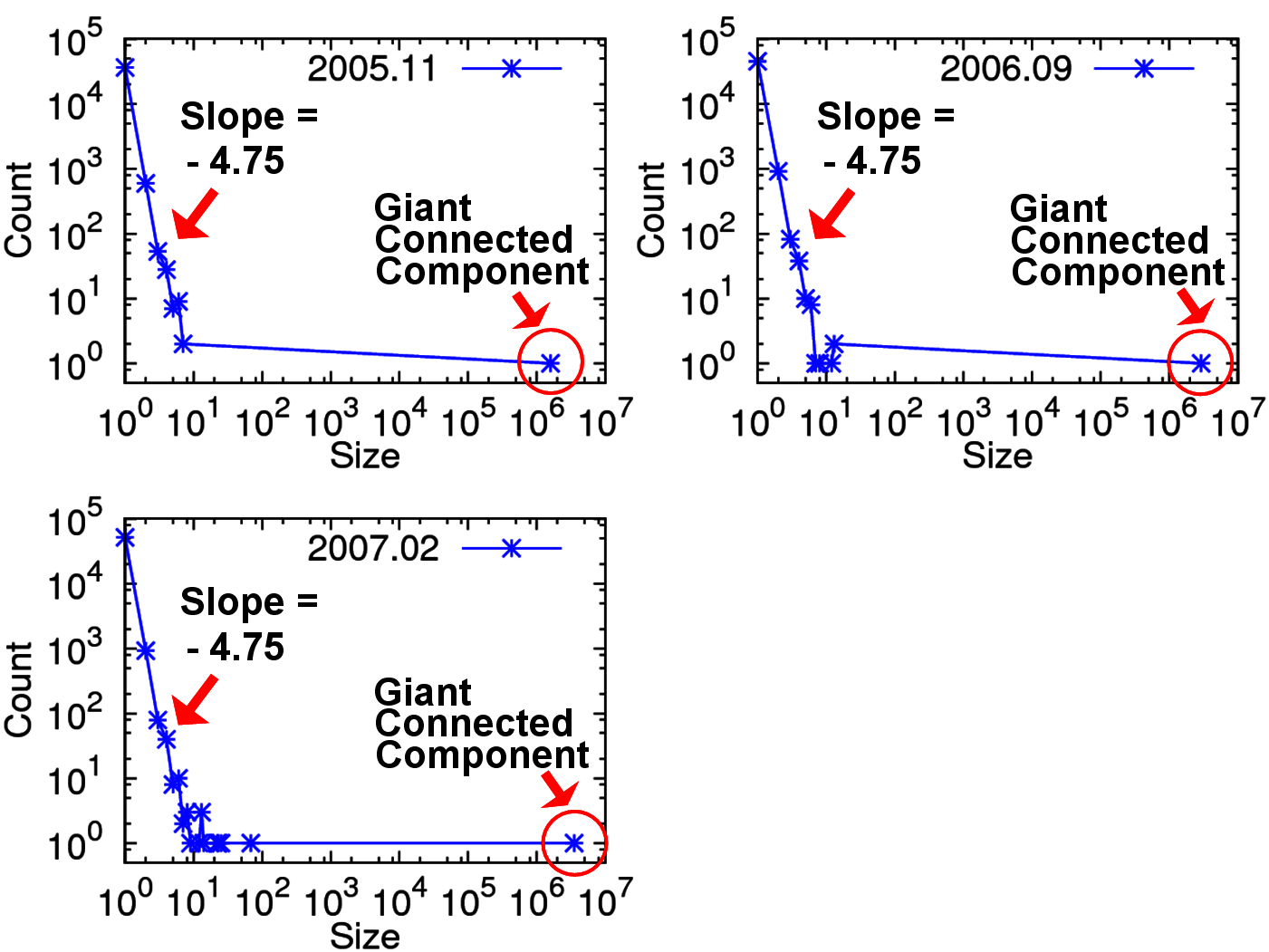} \\
     (b) Connected Components of Wikipedia 
    \end{tabular}
    \caption{
    The evolution of connected components.
    (a)  
    The giant connected component grows for each year. However, the second 
    largest connected component do not grow above Dunbar's number($\approx$ 150) and the 
   slope of the size distribution remains constant after the gelling point at year 2003.
    As in LinkedIn, notice the growth of giant connected component and the constant slope of the size distribution.
    }
    \label{fig:evolution_cc}
    \end{center}
\end{figure*}

\begin{figure}
\begin{center}
  \includegraphics[width=0.6\textwidth]{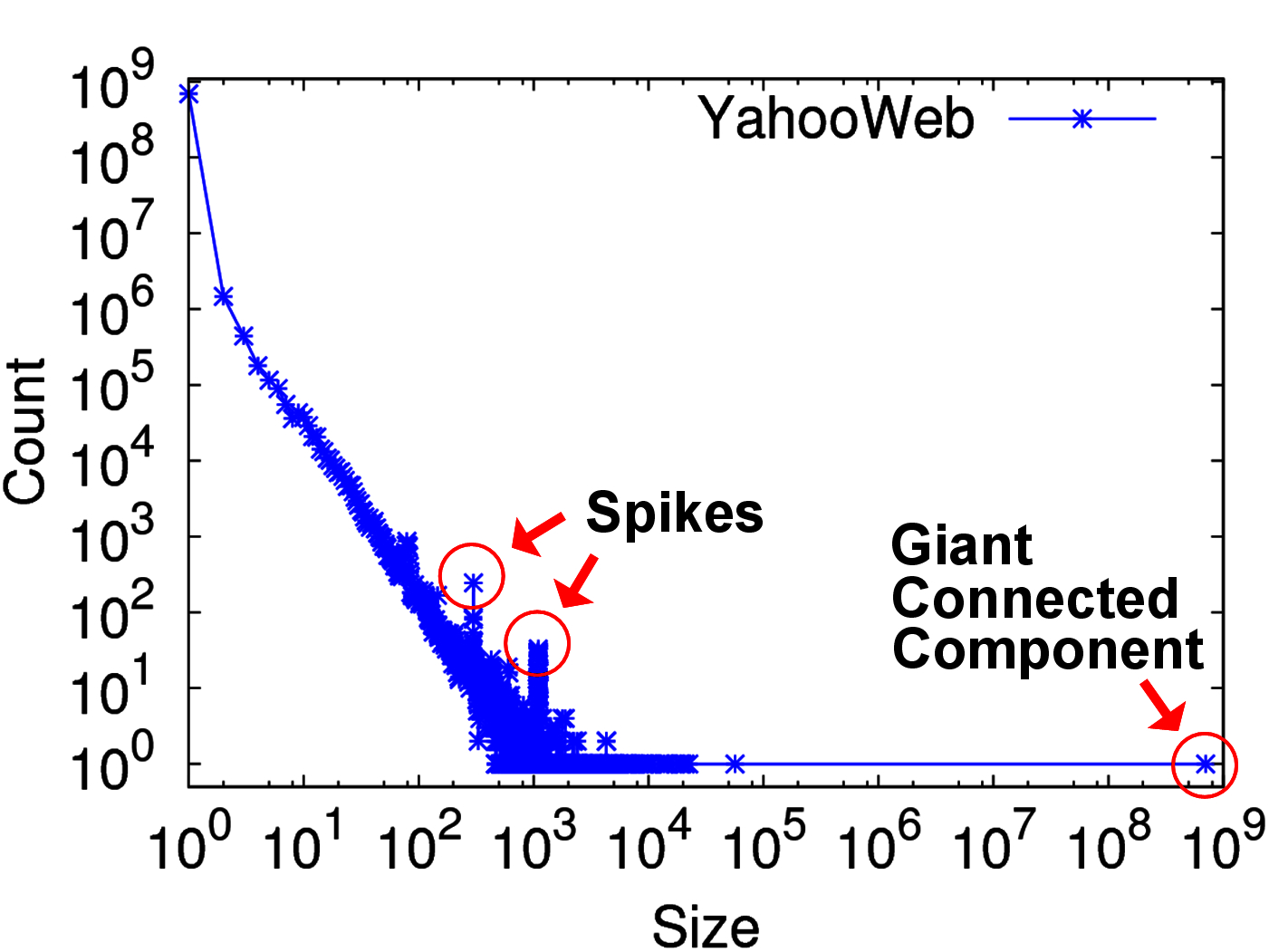}
  \caption{ \label{fig:yahoo_cc}
Connected Components of YahooWeb. Notice the two anomalous spikes which are far from the constant-slope line. Most of them are domain selling or porn sites which are replicated from templates.
  }
\end{center}
\end{figure}

\textbf{Power Laws in Connected Components Distributions}
We observe a power law relation between the count and size of small 
connected components in Figure~\ref{fig:evolution_cc}(a),(b) and Figure~\ref{fig:yahoo_cc}.
This reflects that the connected components in real networks are formed 
by preferential attachment processes.

\textbf{Absorbed Connected Components and Dunbar's number}
The size of the giant component keeps growing while the second and third largest connected components
do not grow beyond size 100, until they are absorbed from the giant component. 
This does not surprise us, since had we had two giant components it is not unlikely
that some new vertex becomes connected to both. 

\textbf{``Anomalous'' Connected Components} 
In Figure~\ref{fig:yahoo_cc}, we see two spikes.
In the first spike at size 300, more than half of the components have exactly the same structure 
and were made from a domain selling company where each component represents a domain to be sold.
The spike happened because the company replicated sites using the same template, 
and injected the disconnected components into WWW network.
In the second spike at size 1101, more than 80 \% of the components are porn sites 
disconnected from the giant connected component.
In general, by By looking carefully the distribution plot of connected components, 
we were able to detect interesting communities 
with special purposes which are disconnected from the rest of the Internet.

\subsection{PageRank scores of Real Networks}

We analyze the PageRank scores of real graphs, using \pegasus.
Figure~\ref{fig:yahoo_pr} and \ref{fig:wwwbb_pr} show the distribution of 
the PageRank scores for the Web graphs, and Figure~\ref{fig:evolution_pr}
shows the evolution of PageRank scores for the LinkedIn and Wikipedia graphs.
We observe power-law relations between the PageRank score
and the number of vertices with such PageRank. 
The top 3 highest PageRank sites for the year 2002 are
{\tt www.careerbank.com}, {\tt access.adobe.com}, and
{\tt top100.rambler.ru}.
As expected, they have huge in-degrees (from $\approx$70K to $\approx$70M).

\begin{figure}
\begin{center}
  \includegraphics[width=0.6\textwidth]{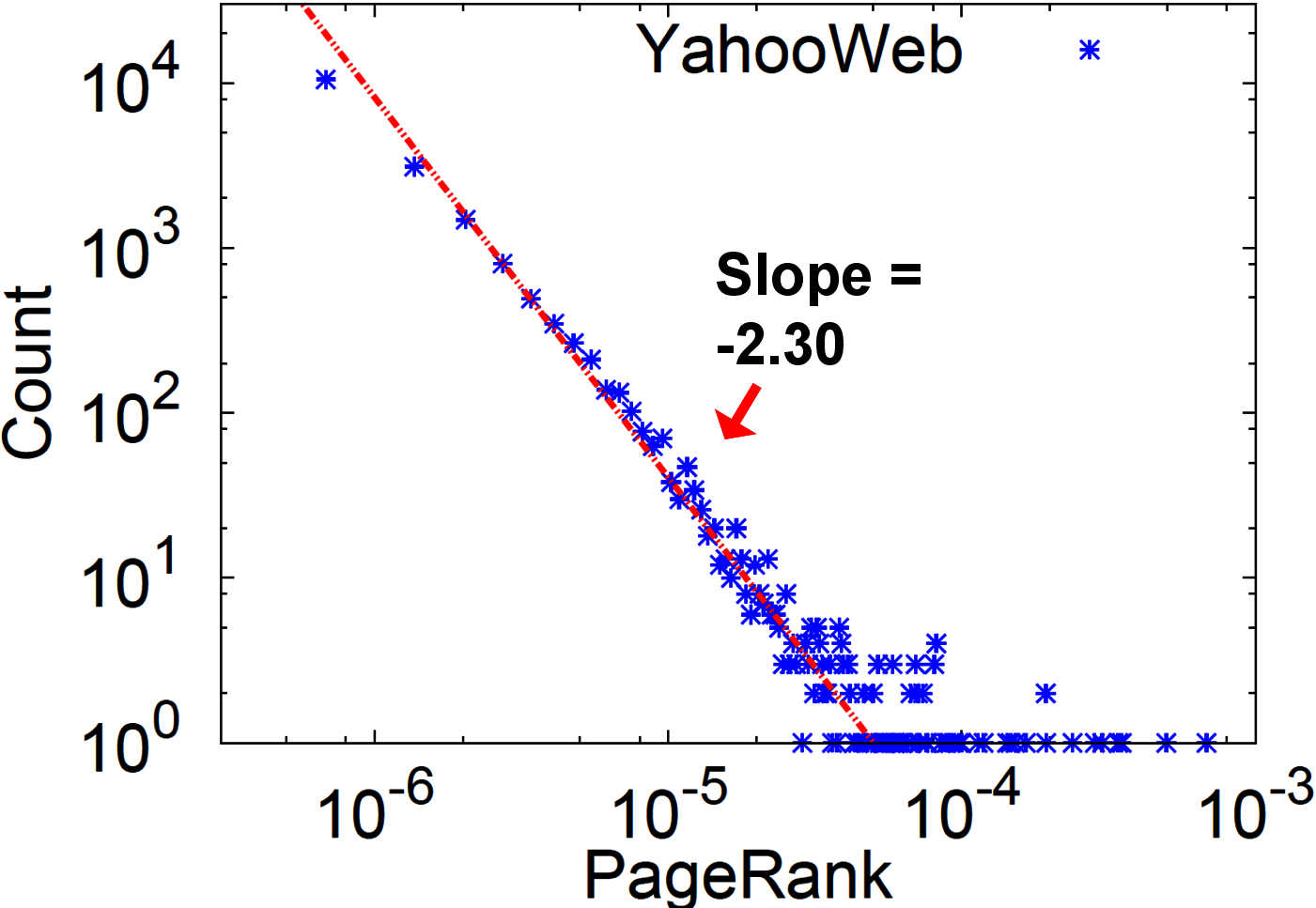}
  \caption{ \label{fig:yahoo_pr}
PageRank distribution of YahooWeb.
The distribution follows a power law with an exponent 2.30.  }
\end{center}
\end{figure}

\begin{figure}
\begin{center}
  \includegraphics[width=0.6\textwidth]{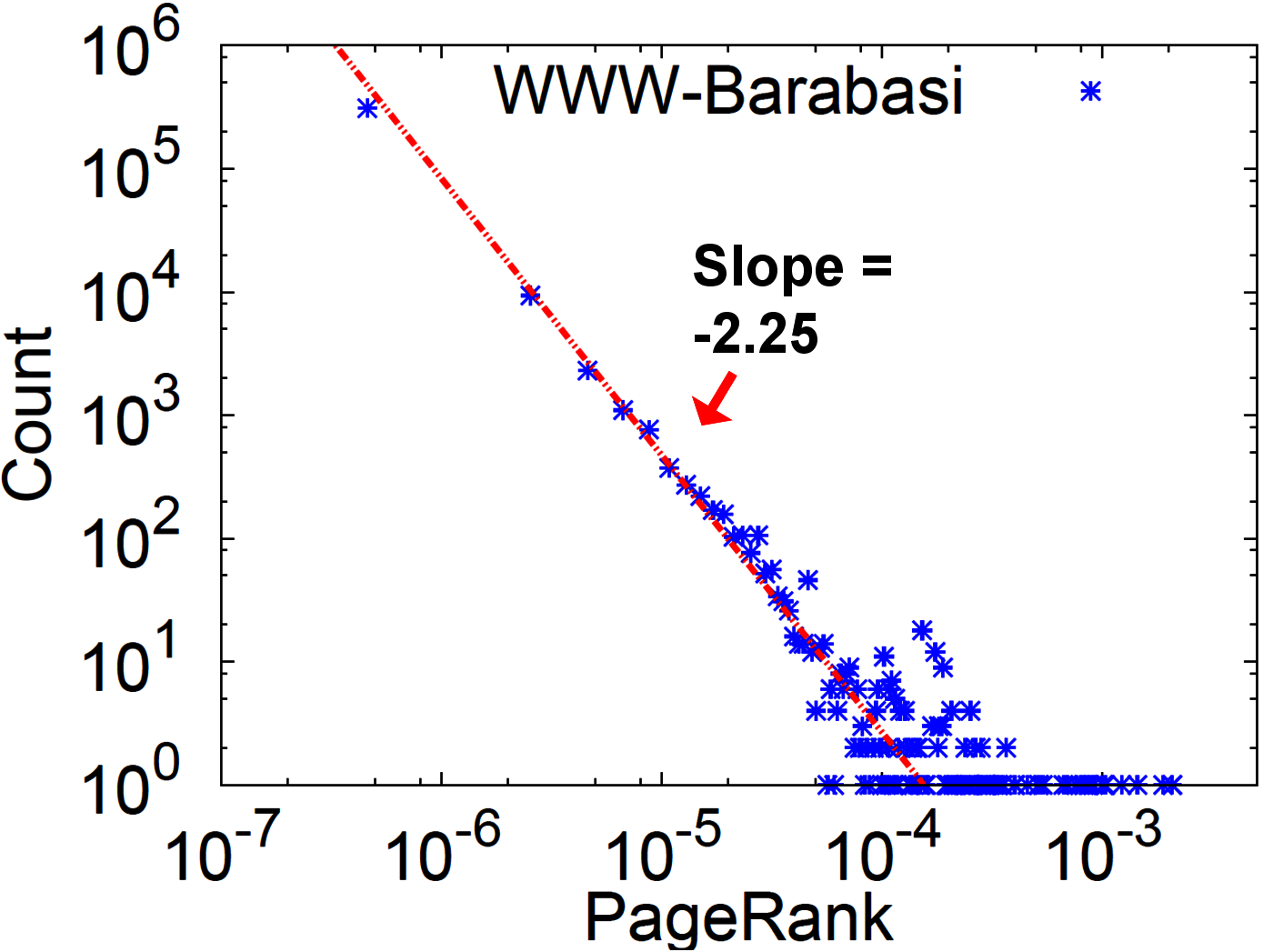}
  \caption{ \label{fig:wwwbb_pr}
PageRank distribution of WWW-Barabasi.
The distribution follows a power law with an exponent 2.25.
  }
\end{center}
\end{figure}

\begin{figure*}
\begin{center}
    \begin{tabular}{c}
    \includegraphics[width=0.7\textwidth]{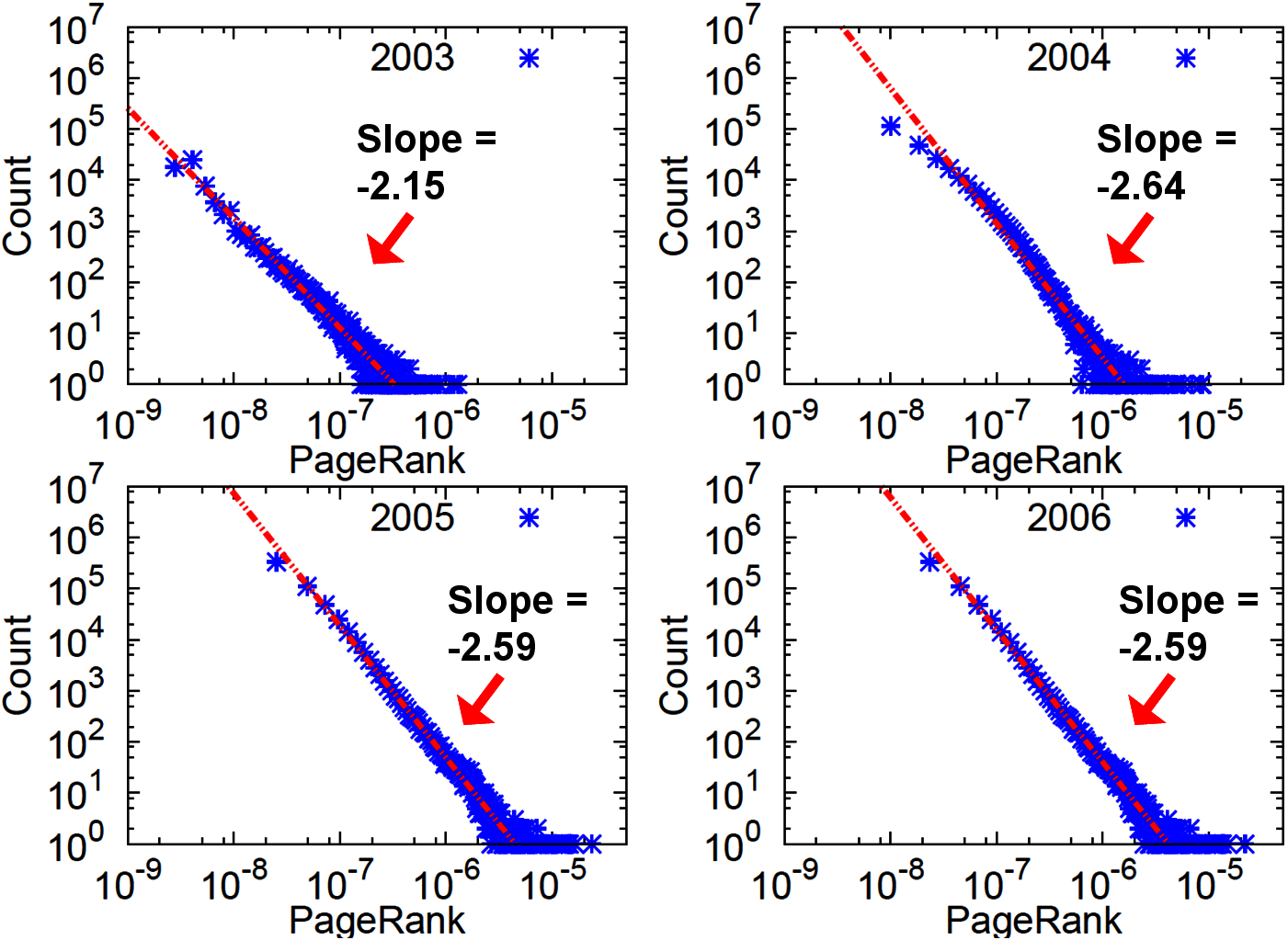} \\
     (a) PageRanks of LinkedIn \\
    \includegraphics[width=0.7\textwidth]{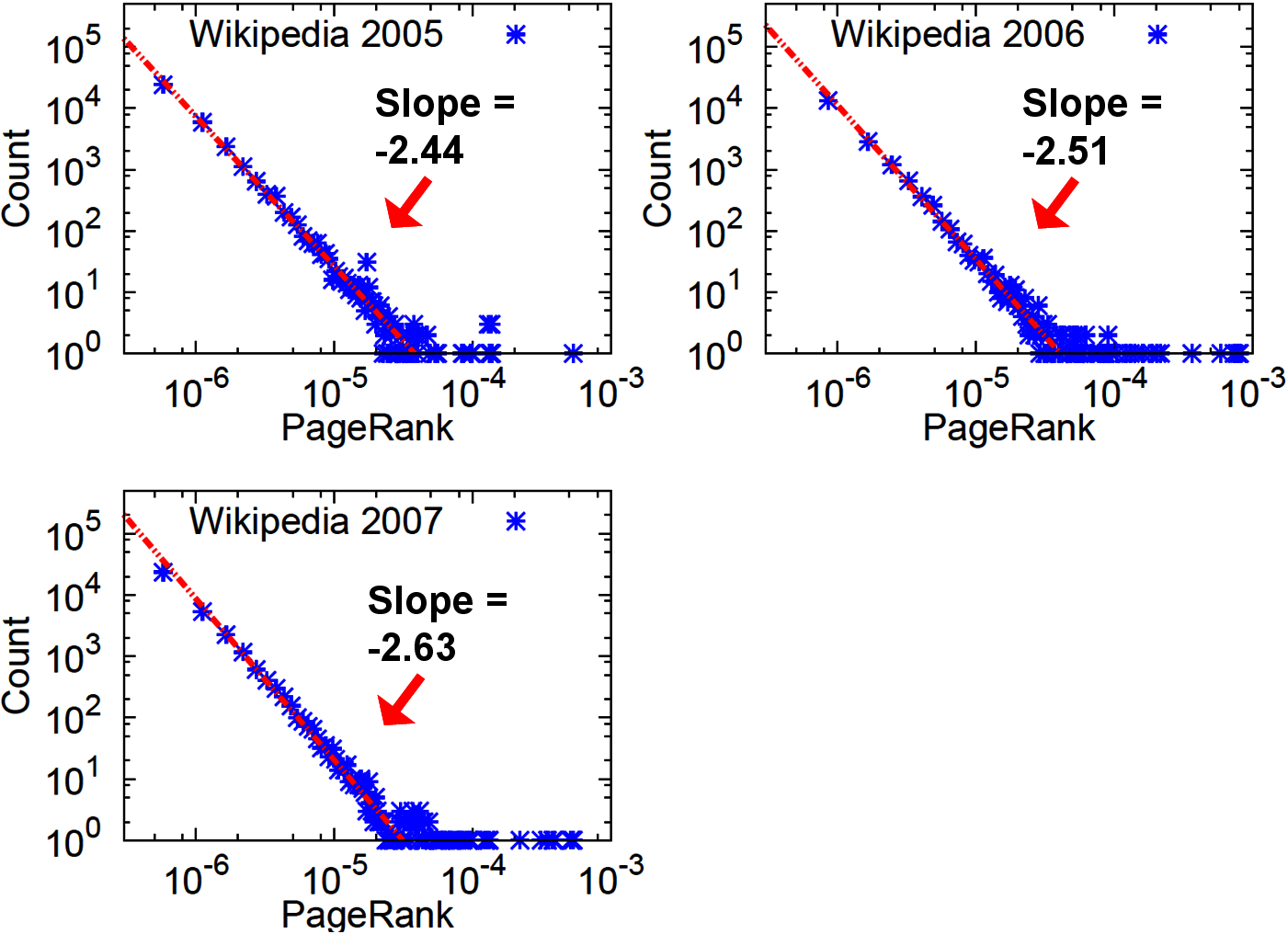} \\
     (b) PageRanks of Wikipedia 
    \end{tabular}
    \caption{
    The evolution of PageRanks.
    (a) The distributions of PageRanks follows a power-law.
    However, the exponent at year 2003, which is around the gelling point, is much different from year 2004, which are after the gelling point.
    The exponent increases after the gelling point and becomes stable. Also notice the maximum PageRank after the gelling point is about 10 times larger than that before the gelling point due to the emergence of the giant connected component. 
    (b) Again, the distributions of PageRanks follows a power-law.
    Since the gelling point is before year 2005, the three plots show similar characteristics: the maximum PageRanks and the slopes are similar.
    }
    \label{fig:evolution_pr}
    \end{center}
\end{figure*}

\subsection{Diameter of Real Networks}

We analyze the diameter and radius of real networks with \pegasus. 
Figure~\ref{fig:dir_radius} shows the radius plot of real networks. We have following observations:

\spara{Small Diameter:} For all the graphs in Figure~\ref{fig:dir_radius}, 
the average diameter is less than 6.09, verifying the six degrees of separation theory.

\textbf{Small Changes in the Diameter over Time: } For the LinkedIn graph, the average diameter remains in the range of 5.28 and 6.09
for all snapshots. 
For the Wikipedia graph, the average diameter remains in the range of 4.76 and 4.99
for all snapshots.
Also, we do not observe a monotone pattern in the growth. 

\textbf{Bimodal Structure of Radius Plot}
For every plot, we observe a bimodal shape which reflects the structure of these real graphs. 
The graphs have one giant connected component where majority of vertices belongs to, and many 
smaller connected components whose size follows a power law.
Therefore, the first mode is at radius zero which comes from single-vertex components; 
the second mode (e.g., at radius 6 in Epinion) comes from the giant connected component.

\begin{figure*}
\begin{center}
  \includegraphics[width=1.0\textwidth]{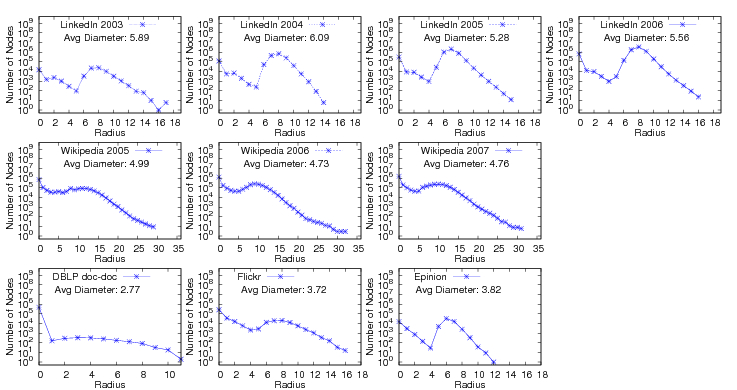}
  \caption{ \label{fig:dir_radius}
  Radius of real graphs. 
  X axis: radius in linear scale.
  Y axis: number of vertices in log scale.
  (Row 1) LinkedIn from 2003 to 2006. 
  (Row 2) Wikipedia from 2005 to 2007. 
  (Row 3) DBLP, flickr, Epinion. 
  Notice that all the radius plots have the bimodal structure due to many smaller connected components(first mode) and the giant connected component(second mode).
  }
\end{center}
\end{figure*}
 

\clearpage
\newpage
\newpage
\vspace*{\fill}
\begingroup
\centering
{\Huge  II {\em Computational Cancer Biology } } 
\endgroup
\vspace*{\fill}
\newpage 


\chapter{Approximation Algorithms for Speeding up Dynamic Programming and Denoising aCGH data}
\label{acghchapter}
\lhead{\emph{Dynamic Programming and Denoising aCGH data}} 
\section{Problem \& Formulation}

Our approach to the aCGH denoising problem, see Section~\ref{subsec:preprocessacgh}, is based 
on the well-established observation that near-by probes tend to have the same DNA copy
number. We formulate the problem of denoising aCGH data as the problem of approximating 
a signal $P$ with another signal $F$ consisting of a few piecewise constant segments.
Specifically, let  $P=(P_1, P_2, \ldots, P_n) \in \field{R}^n$ be the input
signal -in our setting the sequence of the noisy aCGH measurements- 
and let  $C$ be a constant. Our goal is to find a function 
$F:[n]\rightarrow \field{R}$ which optimizes the following objective function:

\begin{equation}
\min_{F} \sum_{i=1}^n (P_i-F_i)^2 + C\times (|\{i:F_i \neq F_{i+1} \} |+1).
\label{eq:optimizationobjective}
\end{equation}

The best known exact algorithm for solving the optimization problem defined by Equation~\ref{eq:optimizationobjective} runs in  $O(n^2)$ time but as our results suggest 
this running time is likely not to be tight. It is worth noting that existing techniques for speeding up dynamic programming \cite{fyao,Eppstein88speedingup,146650} do not apply to our problem. 

The main algorithmic contributions of this Chapter are 
two approximation algorithms for solving this recurrence.
The first algorithm approximates the objective within an additive error $\epsilon$ 
which we can make as small as we wish and its key idea is the reduction of the 
problem to halfspace range queries, a well studied computational geometric problem \cite{Agarwal99geometricrange}. 
The second algorithm carefully decomposes the problem into a ``small'' (logarithmic) 
number of subproblems which satisfy the quadrangle inequality (Monge property).

The remainder of this Chapter is organized as follows: Section~\ref{sec:Trimmervanilla} presents
the vanilla dynamic programming algorithm which runs in $O(n^2)$.
Section~\ref{sec:Trimmertransitionfunction} analyzes properties of the recurrence which 
will be exploited in Sections~\ref{sec:Trimmerhalfspace} and~\ref{sec:Trimmermultiscale} where we describe the
additive and multiplicative approximation algorithms respectively. 
In Section~\ref{sec:Trimmervalidation} we validate our model by performing an extensive
biological analysis of the findings of our segmentation. 

\section{$O(n^2)$ Dynamic Programming Algorithm}
\label{sec:Trimmervanilla}

In order to solve the optimization problem defined by Equation~\ref{eq:optimizationobjective}, we define the key quantity $OPT_i$ given by the following recurrence: 

\begin{align*}
&OPT_0 = 0 \\
&OPT_i = \min_{0 \leq j \leq i-1}  \left[   OPT_{j}+ w(i,j) \right] + C\text{, for~} i>0 \\
&\text{where~} w(i,j) = \sum_{k=j+1}^i \left(P_k - \frac{\sum_{m=j+1}^i P_m}{i-j} \right)^2.
\end{align*}

The above recurrence has a straightforward interpretation:
$OPT_i$ is equal to the minimum cost of fitting a set of piecewise constant segments 
from point $P_1$ to $P_i$ given that  index $j$ is a breakpoint.
The cost of fitting the segment from $j+1$ to $i$ is $C$. 
The weight function $w()$ is the minimum squared error for fitting a constant segment 
on points $\{P_{j+1},\ldots,P_i\}$, which is obtained for the constant segment with value  $\frac{\sum_{m=j+1}^i P_m}{i-j}$,
i.e., the average of the points in the segment.
This recursion directly implies a simple dynamic programming algorithm.
We call this algorithm \trimmer\ and the pseudocode is shown in Algorithm~\ref{alg:vanillatrimmer}.
The main computational bottleneck of \trimmer\ is the computation of the auxiliary matrix $M$, an upper diagonal matrix for which $m_{ij}$ is the minimum squared
error of fitting a segment from points $\{P_i,\ldots,P_j\}$.
To avoid a naive algorithm that would simply find the average of those points
and then compute the squared error, resulting in $O(n^3)$ time,
we use Claim~\ref{thrm:trivial}.

\begin{algorithm}
\caption{\label{alg:vanillatrimmer}CGHTRIMMER algorithm}
\begin{algorithmic}
\STATE Input: Signal $P= (P_1,\ldots,P_n )$, Regularization parameter $C$
\STATE Output:  Optimal Segmentation with respect to our objective (see Equation~\ref{eq:optimizationobjective}) \\ 
\COMMENT{ *Compute an $n \times n$ matrix $M$, where $M_{ji} =\sum_{k=j}^i \left(P_k - \frac{\sum_{m=j}^i P_m}{i-j+1} \right)^2$.\\
$A$ is an auxiliary matrix of averages, i.e., $A_{ji} = \tfrac{ \sum_{k=j}^i P_k }{i-j+1}$.*}
\STATE Initialize matrix $A \in \field{R}^{n \times n}$, $A_{ij}=0, i \neq j$ and $A_{ii}= P_i$. 
\FOR{$i=1$ to $n$} 
\FOR{$j=i+1$ to $n$} 
\STATE $A_{i,j} \leftarrow \frac{j-i}{j-i+1} A_{i,j-1}+ \frac{1}{j-i+1} P_j$ 
\ENDFOR
\ENDFOR
\FOR{$i=1$ to $n$} 
\FOR{$j=i+1$ to $n$} 
\STATE $M_{i,j} \leftarrow M_{i,j-1} + \frac{j-i}{j-i+1}(P_j-A_{i,j-1})^2$
\ENDFOR
\ENDFOR  \\
\COMMENT{ * Solve the Recurrence.*}
\FOR{$i=1$ to $n$} 
\STATE     $\text{OPT}_i  \leftarrow \min_{0 \leq j \leq i-1}{ OPT_{j}+M_{j+1,i}+ C }$
\STATE $\text{BREAK}_i  \leftarrow \text{arg}\min_{1 \leq j \leq i}{ OPT_{j-1}+M_{j,i} + C}$
\ENDFOR
\end{algorithmic}
\end{algorithm}

\begin{claim}
\label{thrm:trivial}
Let $\alpha_{(j)}$ and $m_{(j)}$ be the average and the minimum squared error of fitting a constant segment 
to points $\{ P_1,\ldots, P_j \}$ respectively. Then, 

\begin{equation}
\alpha_{(j)}= \frac{j-1}{j} \alpha_{(j-1)} + \frac{1}{j} P_j,
\label{eq:onlinemean}
\end{equation}

\begin{equation}
m_{(j)}= m_{(j-1)} + \frac{j-1}{j} (P_j-\alpha_{(j-1)})^2.
\label{eq:onlinevar}
\end{equation}

\end{claim} 

The interested reader can find a proof of Claim~\ref{thrm:trivial} in \cite{knuth}. 
Equations~\ref{eq:onlinemean} and~\ref{eq:onlinevar} provide us a way to compute means 
and least squared errors online.
Algorithm~\ref{alg:vanillatrimmer} first computes matrices $A$ and $M$ using Equations~\ref{eq:onlinemean},~\ref{eq:onlinevar}
and then iterates (last {\it for} loop) 
to solve the recurrence by finding the optimal breakpoint for each index $i$. 
The total running time  is $O(n^2)$ (matrices $A$ and $M$ matrices have $O(n^2)$ 
entries and each requires $O(1)$ time to compute). 
Obviously, Algorithm~\ref{alg:vanillatrimmer} uses $O(n^2)$ units of space.

\section{Analysis of The Transition Function} 
\label{sec:Trimmertransitionfunction}

In the following, let $S_i = \sum_{j=1}^i P_j$. The transition function for the dynamic programming for $i>0$ can be rewritten as:

\begin{equation} 
OPT_i = \min_{j<i} OPT_j + \sum_{m=j+1}^i P_m^2 - \frac{(S_i- S_j)^2}{i-j} + C.
\label{eq:Trimmereqopt}
\end{equation}

The transition can be viewed as a weight function $w(j,i)$ that takes the two indices $j$ and $i$ as parameters such that:

\begin{equation}
w(j,i) =  \sum_{m=j+1}^i P_m^2 - \frac{(S_i- S_j)^2}{i-j} + C 
\end{equation}

Note that the weight function does not have the Monge property, as demonstrated by the vector
$P=(P_1,\ldots,P_{2k+1}) = (1, 2, 0, 2, 0, 2, 0, \dots 2, 0, 1)$.
When $C = 1$, the optimal choices of $j$ for $i  = 1, \dots, 2k$ are $j = i-1$, i.e., we fit one segment per point. 
However, once we add in $P_{2k+1} = 1$ the optimal solution changes to fitting all points on a single segment. 
Therefore, preferring a transition to $j_1$ over one to $j_2$ at some index $i$ does not allow us to discard $j_2$ from future considerations.
This is one of the main difficulties in applying techniques based on the increasing order of optimal choices of $j$,
such as the method of Eppstein, Galil and Giancarlo \cite{Eppstein88speedingup} or the method 
of Larmore and Schieber \cite{larmore}, to reduce the complexity of the $O(n^2)$ algorithm we described in Section~\ref{sec:Trimmervanilla}. 

We start by defining $DP_i$ for $i=0,1,..,n$, the solution to a simpler optimization problem. 

\begin{definition} 
Let $DP_i$, $i=0,1,..,n$, satisfy the following recurrence
\begin{equation}
   			DP_i = \left\{ 
							\begin{array}{l r}
  								\min_{j<i} DP_j - \frac{(S_i- S_j)^2}{i-j} + C & \text{if}~ i > 0  \\
  													0                                    & \text{if}~i=0     \\ 
		\end{array} \right. 
\label{eq:Trimmereqdp}
\end{equation}

\end{definition} 

\noindent The following observation stated as Lemma~\ref{lem:Trimmerlemma41} plays a key role in Section~\ref{sec:Trimmerhalfspace}.

\begin{lemma} 
For all $i$, $OPT_i$ can be written in terms of $DP_i$ as 
      $$OPT_i=DP_i+\sum_{m=1}^i P_m^2.$$
\label{lem:Trimmerlemma41}
\end{lemma}

\begin{proof} 
We use strong induction on $i$. For $i=0$ the result trivially holds. Let the result hold for all $j<i$. Then,

\begin{align*}
DP_i &= \min_{j<i} DP_j - \frac{(S_i- S_j)^2}{i-j} + C \\ 
     &= \min_{j<i} OPT_j  - \frac{(S_i- S_j)^2}{i-j} + \sum_{m=j+1}^i P_m^2 - \sum_{m=1}^i P_m^2 + C\\
     &= OPT_i - \sum_{m=1}^i P_m^2\\
\end{align*}

Hence, $OPT_i=DP_i+\sum_{m=1}^i P_m^2$ for all $i$.
\end{proof} 

Observe that the second order moments involved in the expression of $OPT_i$ are absent from $DP_i$. 
Let $\tilde{w}(j,i)$ be the shifted weight function, i.e., $\tilde{w}(j,i) = - \frac{(S_i- S_j)^2}{i-j} + C$.
Clearly, $w(j,i) = \tilde{w}(j,i) +  \sum_{m=i}^j P_m^2 $.

\section{Additive Approximation using Halfspace Queries}
\label{sec:Trimmerhalfspace}

In this Section, we present a novel algorithm which runs in $\tilde{O}(n^{\tfrac{4}{3}+\delta} \log{ (\frac{U}{ \epsilon }  ) } )$ 
time and approximates the optimal objective value within additive $\epsilon$ error.
We derive the algorithm gradually in the following and upon presenting the necessary theory we provide the pseudocode (see Algorithm 2)
at the end of this Section. 
Our proposed method uses the results of \cite{10.1109/SFCS.1992.267816} as stated in Corollary \ref{cor:matousek}
to obtain a fast algorithm for the additive approximation variant of the problem.
Specifically, the algorithm initializes a 4-dimensional halfspace query data structure. 
The algorithm then uses binary searches to compute an accurate estimate of the value $DP_i$ for $i=1,\ldots,n$.
As errors are introduced at each term, we use $\tilde{DP_i}$ to denote the approximate value of $DP_i$ calculated by the binary search,
and $\bar{DP_i}$ to be the optimum value of the transition function computed
by examining the approximate values $\tilde{DP_j}$ for all $j<i$. Formally,

$$\bar{DP}_i = \min_{j<i} \left[ \tilde{DP}_j - \underbrace{\frac{(S_i- S_j)^2}{i-j}}_{\tilde{w}(j,i)} \right] + C.$$

Since the binary search incurs a small additive error at each step, it remains to show that these errors accumulate in a controlled way.
Theorem~\ref{thrm:Trimmermainthrm} states that a small error at each step suffices to give an overall good approximation.
We show inductively that if $\tilde{DP}_i$ approximates $\bar{DP}_i$ within $\epsilon/n$, then $\tilde{DP}_i$ is within $i\epsilon /n$ additive error from the optimal value $DP_i$  for all $i$.

\begin{theorem}
\label{thrm:Trimmermainthrm}
Let $\tilde{DP}_i$ be the approximation of our algorithm to $DP_i$. Then, the following inequality holds: 
\begin{equation} 
|DP_i - \tilde{DP}_i | \leq \frac{\epsilon i}{n}
\end{equation} 
\end{theorem}

\begin{figure} 
\centering
\includegraphics[width=.5\textwidth]{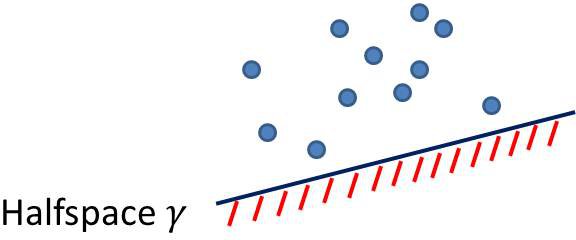}
\caption{Answering whether or not $\bar{DP}_i \leq x+C$ reduces to answering 
whether the point set $\{(j, \tilde{DP_j}, 2S_j, S_j^2+\tilde{DP_j}j) \in \field{R}^4, j<i\}$
has a non-empty intersection with the halfspace 
$\gamma=\{y \in \field{R}^4: a_i y \leq  b_i \}$ where $a_i$ and $b_i$ 
are a 4-dimensional constant vector and a constant which depend on $i$ respectively.
This type of queries can be solved efficiently, see \cite{10.1109/SFCS.1992.267816}. }
\label{fig:acghfig5}
\end{figure}

\begin{proof}

We use induction on the number of points. 
Using the same notation as above, let $\bar{DP}_i = \min_{j<i} \tilde{DP}_j - w(j,i) + C$. 
By construction the following inequality holds: 

\begin{equation}
|\bar{DP}_i - \tilde{DP}_i | \leq \frac{\epsilon}{n} ~\forall i=1,\dots,n
\label{eq:Trimmerbartilde}
\end{equation}

When $i=1$ it is clear that $|DP_1 - \tilde{DP}_1| \leq  \frac{\epsilon}{n}$.
Our inductive hypothesis is the following: 

\begin{equation}
|DP_j - \tilde{DP}_j | \leq \frac{j\epsilon}{n} ~\forall j<i
\label{eq:Trimmerinductivehypothesis}
\end{equation} 

\noindent  It suffices to show that the following inequality holds: 

\begin{equation} 
|DP_i - \bar{DP}_i | \leq \frac{(i-1)\epsilon}{n} 
\label{eq:Trimmersuffices}
\end{equation}
\noindent  since then by the triangular inequality we obtain: 
$$\frac{ i\epsilon }{n} \geq |DP_i - \bar{DP}_i |+ |\bar{DP}_i - \tilde{DP}_i | \geq |DP_i - \tilde{DP}_i |.$$ 

\noindent  Let $j^*,\bar{j}$ be the optimum breakpoints for $DP_i$ and $\bar{DP}_i$ respectively,  $j^*,\bar{j} \leq i-1$.

\begin{align*}
DP_i &= DP_{j^*}+ \tilde{w}(j^*,i)+C \\
& \leq DP_{\bar{j}}+ \tilde{w}(\bar{j},i)+C  \\ 
& \leq \tilde{DP}_{\bar{j}} + \tilde{w}(\bar{j},i)+C+ \frac{\bar{j}\epsilon}{n} \text{(by \ref{eq:Trimmerinductivehypothesis})} \\
& =\bar{DP}_i +  \frac{\bar{j}\epsilon}{n}  \\ 
& \leq \bar{DP}_i +  \frac{(i-1)\epsilon}{n}
\end{align*}

\noindent  Similarly we obtain:

\begin{align*}
\bar{DP}_i &= \tilde{DP}_{\bar{j}}+ \tilde{w}(\bar{j},i)+C \\
& \leq \tilde{DP}_{j^*}+ \tilde{w}(j^*,i)+C  \\ 
& \leq  DP_{j^*}+ \tilde{w}(j^*,i)+C+\frac{j^*\epsilon}{n} \text{(by \ref{eq:Trimmerinductivehypothesis})} \\
& = DP_i +  \frac{j^*\epsilon}{n}  \\ 
& \leq DP_i +  \frac{(i-1)\epsilon}{n}
\end{align*}

\noindent  Combining the above two inequalities, we obtain \ref{eq:Trimmersuffices}. 
\end{proof}

\noindent By substituting $i=n$ in Theorem~\ref{thrm:Trimmermainthrm}  we obtain the following corollary
which proves that $\tilde{DP}_n$ is within $\epsilon$ of $DP_n$.

\begin{corollary}
Let $\tilde{DP}_n$ be the approximation of our algorithm to $DP_n$. Then, 
\begin{equation} 
|DP_n - \tilde{DP}_n | \leq \epsilon. 
\end{equation} 
\end{corollary} 

To use the analysis above in order to come up with an efficient algorithm we need to answer 
two questions: (a) How many binary search queries do we need in order to obtain the desired
approximation? (b) How can we answer each such query efficiently? 
The answer to (a) is simple: as it can easily be seen, the value of the objective function is upper bounded by $U^2n$, 
where $U = \max{ \{ \sqrt{C}, |P_1|, \dots, |P_n|\} }$. 
Therefore, $O(\log(\frac{U^2n}{\epsilon/n})) = \tilde{O}( \log{ (\frac{U}{ \epsilon } ) } )$
iterations of binary search at each index $i$ are sufficient to obtain the desired approximation. 
We reduce the answer to (b) to a well studied computational geometric problem.
Specifically, fix an index $i$, where $1 \leq i \leq n$, and consider the general form of the binary search query 
$\bar{DP}_i \leq x+C$, where $x+C$ is the value on which we query. Note that we use
the expression $x+C$ for convenience, i.e., so that the constant $C$ will be simplified from both sides
of the query. This query translates itself to the following decision problem, see also Figure~\ref{fig:acghfig5}. 
Does there exist an index $j$, such that $j<i$ and the following inequality holds:

\begin{align*}
x &\geq \tilde{DP_j} - \frac{(S_i- S_j)^2}{i-j} \Rightarrow xi + S_i^2 \geq (x, i, S_i, -1)(j, \tilde{DP_j}, 2S_j, S_j^2 + j\tilde{DP_j})^T?
\end{align*}

Hence, the binary search query has been reduced to answering a {\em halfspace query}. 
Specifically, the decision problem for any index $i$ becomes whether the intersection of the point set 
$\text{POINTS}_i = \{(j, \tilde{DP_j}, 2S_j, S_j^2+\tilde{DP_j}j) \in \field{R}^4, j<i\}$ with a hyperplane is empty.
By Corollary \ref{cor:matousek}  \cite{10.1109/SFCS.1992.267816}, for a point set of size $n$, this can be done in $\tilde{O}(n^{\tfrac{1}{3}+\delta})$ per query and $O(n^{\tfrac{1}{3}}\log{n})$ amortized time per insertion of a point.
Hence, the optimal value of $DP_i$ can be found within an additive
constant of $\epsilon/n$ using the binary search in $\tilde{O}(n^{\tfrac{1}{3}}\log{ (\frac{U}{ \epsilon } ) } )$ time.

Therefore, we can proceed from index 1 to $n$, find the approximately optimal value of $OPT_i$ and insert a point corresponding to it into the query structure. We obtain an algorithm which runs in $\tilde{O}(n^{\tfrac{4}{3} + \delta} \log{ (\frac{U}{ \epsilon } ) } )$ time, where
$\delta$ is an arbitrarily small positive constant.
The pseudocode is given in Algorithm~\ref{alg:halfspacealgo}.

\begin{algorithm}
\caption{\label{alg:halfspacealgo}Approximation within additive $\epsilon$ using 4D halfspace queries}

\begin{algorithmic}
\STATE Initialize 4D halfspace query structure $Q$ 
\FOR{$i=1$ to $n$}
\STATE $low \gets 0$ 
\STATE    $high \gets nU^2$ 
\WHILE{$high - low > \epsilon/n$} 
\STATE $m \gets (low+high)/2$  \\
\COMMENT{ *This halfspace emptiness query is efficiently supported by Q.* } \\ 
\STATE $flag \leftarrow (\exists j \text{~such that~} xi + S_i^2 \geq (x, i, S_i, -1)(j, \tilde{DP_j}, 2S_j, S_j^2 + j\tilde{DP_j})^T )$
\IF{$flag$} 
  \STATE $high \gets m$
\ELSE 
\STATE $low \gets m$
\ENDIF

\ENDWHILE
\STATE $\tilde{DP}_i \gets (low+high)/2$ \\ 
\COMMENT{*Point insertions are efficiently supported by Q.*}
\STATE Insert point  $(i, \tilde{DP}_i, 2S_i, S_i^2+\tilde{DP}_ii)$ \text{~in Q}
\ENDFOR
\end{algorithmic}
\end{algorithm}

\section{Multiscale Monge Decomposition}
\label{sec:Trimmermultiscale}

In this Section we present an algorithm which runs in $O(n \log{n} / \epsilon )$
time to approximate the optimal shifted objective value within a multiplicative
factor of $(1+\epsilon)$. Our algorithm is based on a new technique, which we consider
of indepenent interest. 
Let $w'(j,i) = \sum_{m = j+1}^i (i-j)P_m^2 - (S_i-S_j)^2$.
We can rewrite the weight function $w$ as a function of $w'$, namely as 
$w(j,i) = w'(j,i)/(i-j) + C$.
The interested reader can easily check that we can  express $w'(j,i)$ as a sum of non-negative terms, i.e.,

$$w'(j,i) = \sum_{j+1 \leq m_1 < m_2 \leq i} (P_{m_1} - P_{m_2})^2.$$

Recall from Section~\ref{sec:Trimmervanilla} that the weight function $w$ is not Monge. The next lemma shows that the weight function 
$w'$ is a Monge function.

\begin{lemma}
\label{lem:Trimmerlem3}
The weight function $w'(j,i)$ is Monge (concave), i.e.,
for any $i_1<i_2<i_3<i_4$, the following holds:

$$w'(i_1, i_4)+w'(i_2,i_3) \geq w'(i_1,i_3) + w'(i_2, i_4).$$
\end{lemma}

\begin{proof}
Since each term in the summation is non-negative, it suffices to show that
any pair of indices, $(m_1,m_2)$ is summed as many
times on the left hand side as on the right hand side.
If $i_2+1 \leq m_1 < m_2 \leq i_3$, each term is counted
twice on each side. Otherwise, each term is counted once on the left hand
side since $i_1+1 \leq m_1 < m_2 \leq i_4$
and at most once on the right hand side since
$[i_1+1, i_3] \cap [i_2+1, i_4] = [i_2+1, i_3]$.
\end{proof}

\begin{table}
\centering 
\caption{\label{tab:proofcounts} Summary of proof of Lemma~\ref{lem:Trimmerlem3}.}
\begin{tabular}{|c|c|c|} \hline
        ~                 &      $w'(i_1, i_4)+w'(i_2,i_3)$            & $w'(i_1,i_3) + w'(i_2, i_4)$       \\ \hline
$(S_1,S_1)$               &                 1                           &              1                    \\ \hline
$(S_2,S_2)$               &                 2                           &              2                    \\ \hline
$(S_3,S_3)$               &                 1                           &              1                    \\ \hline
$(S_1,S_2)$               &                 1                           &              1                    \\ \hline
$(S_1,S_3)$               &                 1                           &              0                    \\ \hline
$(S_2,S_3)$               &                 1                           &              1                    \\ \hline
\end{tabular}
\end{table}

The proof of Lemma~\ref{lem:Trimmerlem3} is summarized in Table~\ref{tab:proofcounts}. 
Specifically, let $S_j$ be the set of indices $\{i_j+1,\ldots,i_{j+1}\}$, $j=1,2,3$. 
Also, let $(S_j,S_k)$ denote the set of indices $(m_1,m_2)$ which appear in the summation
such that $m_1 \in S_j, m_2 \in S_k$. The two last columns of Table~\ref{tab:proofcounts}
correspond to the left- and right-hand side of the Monge inequality (as in Lemma~\ref{lem:Trimmerlem3})
and contain the counts of appearances of each term.

Our approach is based on the following observations:
\begin{enumerate} 
 \item Consider the weighted directed acyclic graph (DAG) on the vertex set $V=\{0,\ldots,n\}$
 with edge set $E=\{ (j,i): j<i \}$ and weight function $w: E \rightarrow \field{R}$, i.e., 
 edge $(j,i)$ has weight $w(j,i)$. 
 Solving the aCGH denoising problem reduces to finding a shortest path from vertex $0$ to vertex $n$. 
 If we perturb the edge weights within a factor of $(1+\epsilon)$, as long as the weight of each edge is positive,
 then the optimal shortest path distance  is also perturbed within a factor of at most $(1+\epsilon)$.
 \item By Lemma~\ref{lem:Trimmerlem3} we obtain that  the weight function is not Monge essentially because
 of the $i-j$ term in the denominator.
 \item Our goal is to approximate $w$ by a Monge function $w'$ such that $c_1 w' \leq w \leq c_2 w'$
 where $c_1,c_2$ should be known constants. 
\end{enumerate}

In the following we elaborate on the latter goal. 
Fix an index $i$ and note that the optimal breakpoint for that index is some index
$j \in \{ 1, \ldots, i-1\}$. We will ``bucketize'' the range of index $j$ into 
$m=O(\log_{1+\epsilon}{(i)})=O(\log{n} /\epsilon)$ buckets such that the $k$-th 
bucket, $k=1,\ldots,m$, is defined by the set of indices $j$ which satisfy 

$$ l_k=i-(1+\epsilon)^{k} \leq j \leq i-(1+\epsilon)^{k-1}=r_k.$$ 

This choice of bucketization is based on the first two observations 
which guide our approach. Specifically, it is easy to check that 
$(1+\epsilon)^{k-1} \leq i-j \leq (1+\epsilon)^{k}$.
This results, for any given $i$, to approximating $i-j$ by a constant for each possible bucket,
leading to $O(\log{n} /\epsilon)$ different Monge functions (one per bucket)
while incurring a multiplicative error of at most $(1+\epsilon)$.
However, there exists a subtle point, as also  Figure~\ref{fig:acghfig6} indicates.
We need to make sure that each of the Monge functions
is appropriatelly defined so that when we consider the $k$-th Monge subproblem, 
the optimal breakpoint $j_k$ should satisfy $j_k \in [l_k,r_k]$. Having achieved that (see Lemma 6.2)
we can solve efficiently the recurrence. Specifically, $OPT_i$ is computed as follows: 

\begin{align*}
OPT_i
&= \min_{j < i}  \left [OPT_j + \frac{w'(i,j)}{i-j} \right] + C\\
&= \min_k  \left [\min_{j \in [l_k, r_k]}
OPT_j + \frac{w'(i,j)}{i-j} \right] + C\\
&\approx \min_k  \left [\min_{j \in [l_k, r_k]}
OPT_j + \frac{w'(i,j)}{(1+\epsilon)^{k-1}} \right] + C\\
&= \min_k  \left [\min_{j \in [l_k, r_k]}
OPT_j + \frac{w'(i,j)}{c_k} \right] + C.\\
\end{align*}

The following is one of the possible ways to define the $m$ Monge weight functions. In what follows, $\lambda$ is a sufficiently large positive constant.

\begin{equation}
w_k(j,i) = \left\{ 
\begin{array}{l r}
  2^{n-i+j}\lambda & i - j < c_k = (1+\epsilon)^{k-1} \\
  2^{i-j}\lambda & i - j > (1+\epsilon)c_k = (1+\epsilon)^{k} \\
  w'(j,i)/c_k           & \text{otherwise}     \\ 
\end{array} \right. 
\label{eq:Trimmereqwk}
\end{equation}

\begin{lemma}
Given any vector $P$, it is possible to pick $\lambda$ such that $w_k$ is Monge for all $k \geq 1$.
That is, for any 4-tuple $i_1 < i_2 < i_3 < i_4$, $w_k(i_1, i_4)+w_k(i_2,i_3) \geq w_k(i_1,i_3) + w_k(i_2, i_4)$.
\end{lemma}

\begin{proof}
Since $w'(j,i) = \sum_{j+1 \leq m_1 < m_2 \leq i} (P_{m_1} - P_{m_2})^2
\leq (2K)^2 n^2$ where $K = \max_{1 \leq i \leq n} |P_i|$,
we can pick $\lambda$ such that $w_k(j,i) \geq w'(j,i)$.
The rest of the proof is casework based on the lengths of the intervals
$i_3-i_1$, $i_4-i_2$, and how they compare with $c_k$ and $(1+\epsilon) c_k$.
There are 12 such cases in total. We may assume $i_3-i_1 \leq i_4-i_2$ without
loss of generality, leaving thus 6 cases to be considered.

\noindent If $c_k \leq i_3-i_1, i_4-i_2 \leq (1+\epsilon)c_k$, then:
\begin{align*}
w_k(i_1,i_3) + w_k(i_2, i_4) &= (1/c_k)(w'(i_1,i_3) + w'(i_2, i_4))\\
&\leq (1/c_k)(w'(i_1,i_4) + w'(i_2, i_3)) \text{ (by Lemma \ref{lem:Trimmerlem3})}\\
&\leq w_k(i_1,i_4) + w_k(i_2, i_3).
\end{align*}

\noindent  Consider the case $i_3-i_1 < c_k$ and $i_4-i_2 \leq (1+\epsilon) c_k$.
Then as $i_2 > i_1$, $i_3-i_2 \leq i_3-i_1 - 1$ and we have:
\begin{align*}
w_k(i_2, i_3)&= 2^{n-i_3+i_2}\lambda\\
             & \geq 2 \cdot 2^{n-i_3-i_1}\lambda\\   
             & \geq w_k(i_1, i_3) + w_k(i_2, i_4)
\end{align*}

The cases of $c_k \leq i_3-i_1$ and $c_k(1+\epsilon) < i_4-i_2$ can be done similarly.
Note that the cases of $i_3-i_1, i_4 - i_2 < c_k$ and $(1+\epsilon) c_k < i_3-i_1, i_4 - i_2$
are also covered by these.

The only case that remain is $i_3-i_1 < c_k$ and $(1+\epsilon)c_k < i_4 - i_2$.
Since $i_3 - i_2 < i_3 - i_1 < c_k$, we have:
\begin{align*}
w_k(i_2, i_3) &= 2^{n-i_3+i_2}\lambda\\
              &> 2^{n-i_4+i_2}\lambda\\
              &= w_k(i_2, i_4)
\end{align*}

Similarly $w_k(i_1, i_4) = 2^{i_4-i_1}\lambda > 2^{i_3 - i_1}\lambda = w_k(i_1, i_3)$.
Adding them gives the desired result.
\end{proof}

\begin{figure} 
\centering
\includegraphics[width=.75\textwidth]{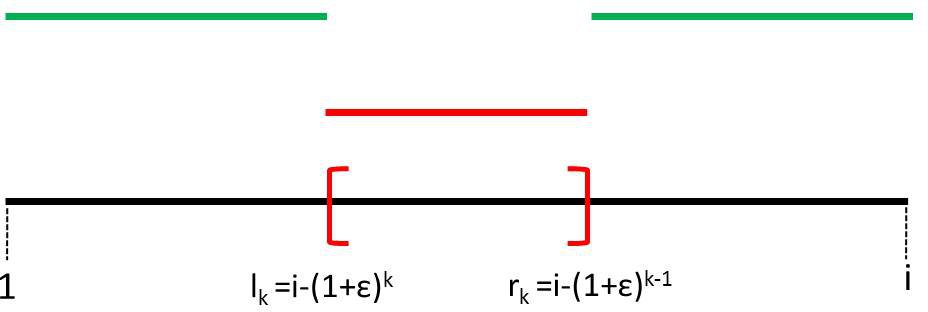}
\caption{Solving the $k$-th Monge problem, $k=1,\ldots,m=O(\log_{1+\epsilon}{(i)})$, for a fixed index $i$. The $k$-th interval 
is the set of indices $\{j: l_k \leq j \leq r_k, l_k=i-(1+\epsilon)^{k}, r_k=i-(1+\epsilon)^{k-1} \}$. Ideally, we wish to define a Monge 
function $w_k'$ whose maximum inside the $k$-th interval (red color) is smaller than the minimum outside that interval (green color).
This ensures that the optimal breakpoint for the $k$-th Monge subproblem lies inside the red interval. }
\label{fig:acghfig6}
\end{figure}

The pseudocode is shown in Algorithm~\ref{alg:mongealgo}. The algorithm computes the 
$OPT$ values online based on the above analysis. 
In order to solve each Monge sub-problem our method calls the routine of Theorem~\ref{thm:larmore}. 
For each index $i$, we compute $OPT_i$ by taking the best value over
queries to all $k$ of the Monge query structures, then we update all the
structures with this value.
Note that storing values of the form $2^k \lambda$ using only their exponent $k$ suffices for comparison,
so introducing $w_k(j,i)$ doesn't result in any change in runtime.
By Theorem~\ref{thm:larmore}, for each $Q_k$, finding $\min_{j<i} Q_k.a_j + w_k(j,i)$
over all $i$ takes $O(n)$ time. Hence, the total runtime is $O(n\log{n}/\epsilon)$.

\begin{algorithm}
\caption{\label{alg:mongealgo}Approximation within a factor of $\epsilon$ using Monge function search}
\begin{algorithmic}
\STATE Maintain $m = \log{n}/\log{(1+\epsilon)}$ Monge function search data structures $Q_1,\dots,Q_m$ where $Q_k$ corresponds to the Monge function $w_k(j,i)$. \\
\COMMENT{*The weight function $w_k(j,i)$ is Monge. Specifically,
$w_k(j,i) = C + \Big( \sum_{m=j+1}^i (i-j)P_m^2 - \frac{(S_i-S_j)^2}{(1+\epsilon)^k} \Big)$ for all $j$ 
which satisfy $(1+\epsilon)^{k-1} \leq i-j \leq (1+\epsilon)^k$. *}
\STATE $OPT_0 \leftarrow 0$ 
\COMMENT{* Recursion basis. *}
\FOR{ $k=1$ to $m$ }
\STATE $Q_k.a_0 \gets 0$ \\
\ENDFOR  \\ 
\COMMENT{*Let $a_j^{(k)}$ denote $Q_k.a_j$, i.e.,  the value $a_j$ of the $k$-th data structure  $Q_k$.*} 
\FOR{ $i=1$ to $n$ }
    \STATE $OPT_i \leftarrow \infty$

    \FOR{ $k=1$ to $m$ } 
      \STATE   $\text{localmin}_k  \leftarrow  \min_{j<i} a_j^{(k)} + w_k(j,i)$
      \STATE   $OPT_i  \leftarrow  \min\{OPT_i, \text{localmin}_k + C \}$
    \ENDFOR 
\ENDFOR 
\FOR{ $k=1$ to $m$ }
\STATE $Q_k.a_i  \leftarrow  OPT_i$
\ENDFOR 
\end{algorithmic}
\end{algorithm}

\section{Validation of Our Model} 
\label{sec:Trimmervalidation}

In this Section we validate our model using the exact algorithm, see Section~\ref{sec:Trimmervanilla}.
In Section~\ref{subsec:Trimmersetup} we describe the datasets and the experimental setup.
In Section~\ref{subsec:Trimmerresults} we show the findings of our method together
with a detailed biological analysis.

\subsection{Experimental Setup and Datasets}
\label{subsec:Trimmersetup} 

\begin{table}
\caption{\label{tab:descriptiondataset}Datasets, papers and the URLs where the datasets can be downloaded. 
\textcolor{red}{$\odot$} and \textcolor{cyan}{$\blacksquare$}   denote
which datasets are synthetic and real respectively. }
\begin{tabular}{|c|c|c|} \hline
        ~                 &      Dataset           &  Availability             \\\hline
\textcolor{red}{$\odot$} &  Lai et al.        & \cite{1181383}                                                           \\ 
      ~                   &                    & \url{http://compbio.med.harvard.edu/}   \\ \hline
\textcolor{red}{$\odot$} & Willenbrock et al. & \cite{citeulike:387317}                                                    \\
    ~                   &                    &  \url{http://www.cbs.dtu.dk/~hanni/aCGH/}                               \\ \hline
\textcolor{cyan}{$\blacksquare$}&  Coriell Cell lines     & \cite{coriell} \\
    ~       &           &  \url{http://www.nature.com/ng/journal/v29/n3/} \\ \hline
\textcolor{cyan}{$\blacksquare$} & Berkeley Breast Cancer & \cite{berkeley} \\             
                                 &           & \url{http://icbp.lbl.gov/breastcancer/}  \\ \hline
\end{tabular}
\end{table}

Our code is implemented in MATLAB\footnote{Code available at URL ~\url{http://www.math.cmu.edu/~ctsourak/CGHTRIMMER.zip}. Faster C code is also available, but since
the competitors were implemented in MATLAB, all the results in this Section refer to our MATLAB implementation.}.  The experiments run in a 4GB RAM,
2.4GHz Intel(R) Core(TM)2 Duo CPU, Windows Vista machine.  Our methods
were compared to existing MATLAB implementations of the CBS algorithm,
available via the Bioinformatics toolbox, and the \cghseg\ algorithm \cite{picard},
courteously provided to us by Franc Picard.  \cghseg\ was run
using heteroscedastic model under the Lavielle criterion
\cite{DBLP:journals/sigpro/Lavielle05}.  Additional tests using the
homoscedastic model showed substantially worse performance and are
omitted here.  All methods were compared using previously developed
benchmark datasets, shown in Table~\ref{tab:descriptiondataset}.
Follow-up analysis of detected regions was conducted by manually
searching for significant genes in the Genes-to-Systems Breast Cancer
Database \url{http://www.itb.cnr.it/breastcancer}~\cite{G2SBCD} and
validating their positions with the UCSC Genome Browser
\url{http://genome.ucsc.edu/}. The Atlas of Genetics and
Cytogenetics in Oncology and Haematology
\url{http://atlasgeneticsoncology.org/} was also used to validate the
significance of reported cancer-associated genes.
It is worth pointing out that since aCGH data are typically given in the log scale, we first exponentiate the points, 
then fit the constant segment by taking the average of the exponentiated values from the hypothesized segment,
and then return to the log domain by taking the logarithm of that constant value. 
Observe that one can fit a constant segment by averaging the log values using Jensen's inequality, but we favor an approach 
more consistent with the prior work, which typically models the data assuming i.i.d.~Gaussian noise
in the linear domain. 

\paragraph{How to pick $C$?} 
The performance of our algorithm depends on the value of the parameter $C$, which determines how much each segment ``costs.'' Clearly, there is a tradeoff between
larger and smaller values: excessively large $C$ will lead the algorithm to output a single segment while excessively small $C$ will result in each point being fit as its own segment.  
We pick our parameter $C$ using data published in \cite{citeulike:387317}. The data was generated by modeling real aCGH data, thus capturing their 
nature better than other simplified synthetic data and also making them a good training dataset for our model.  We used this dataset to generate a Receiver Operating Characteristic (ROC) curve 
using values for $C$ ranging from 0 to 4 with increment 0.01 using one of the four datasets in \cite{citeulike:387317} (``above 20''). 
The resulting curve is shown in Figure~\ref{fig:acghfig2}. Then, we selected $C=0.2$, which achieves high precision/specificity (0.98) and high recall/sensitivity (0.91).
All subsequent results reported were obtained by setting $C$ equal to 0.2.

\begin{figure}
\centering
\includegraphics[width=0.7\textwidth]{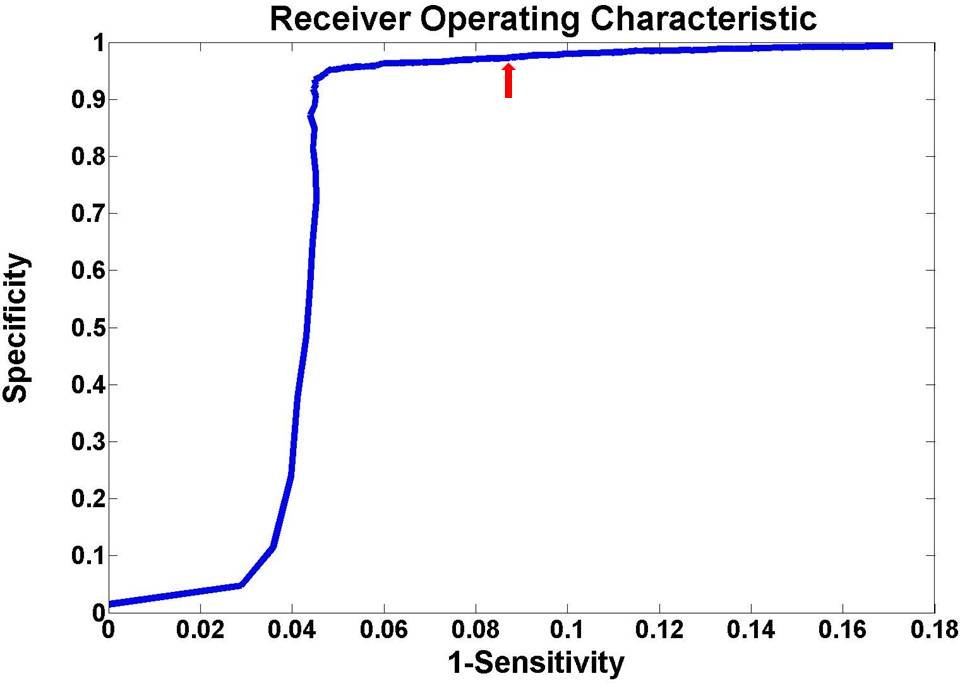}
\caption{ROC curve of \trimmer\ as a function of $C$ on data from \cite{citeulike:387317}. The red arrow indicates
the point (0.91 and 0.98 recall and precision respectively) 
corresponding to $C$=0.2, the value used in all subsequent results.}
\label{fig:acghfig2}
\end{figure}

\subsection{Experimental Results and Biological Analysis}
\label{subsec:Trimmerresults}

We show the results on synthetic data in Section~\ref{subsec:Trimmersynthetic},
on real data where the ground truth is available to us in Section~\ref{subsec:Trimmercoriel}
and on breast cancer cell lines with no ground truth in Section~\ref{subsec:Trimmerbreast}.

\subsubsection{Synthetic Data} 
\label{subsec:Trimmersynthetic}

\begin{figure}
		\begin{tabular}{ccc} 
\includegraphics[width=0.31\textwidth]{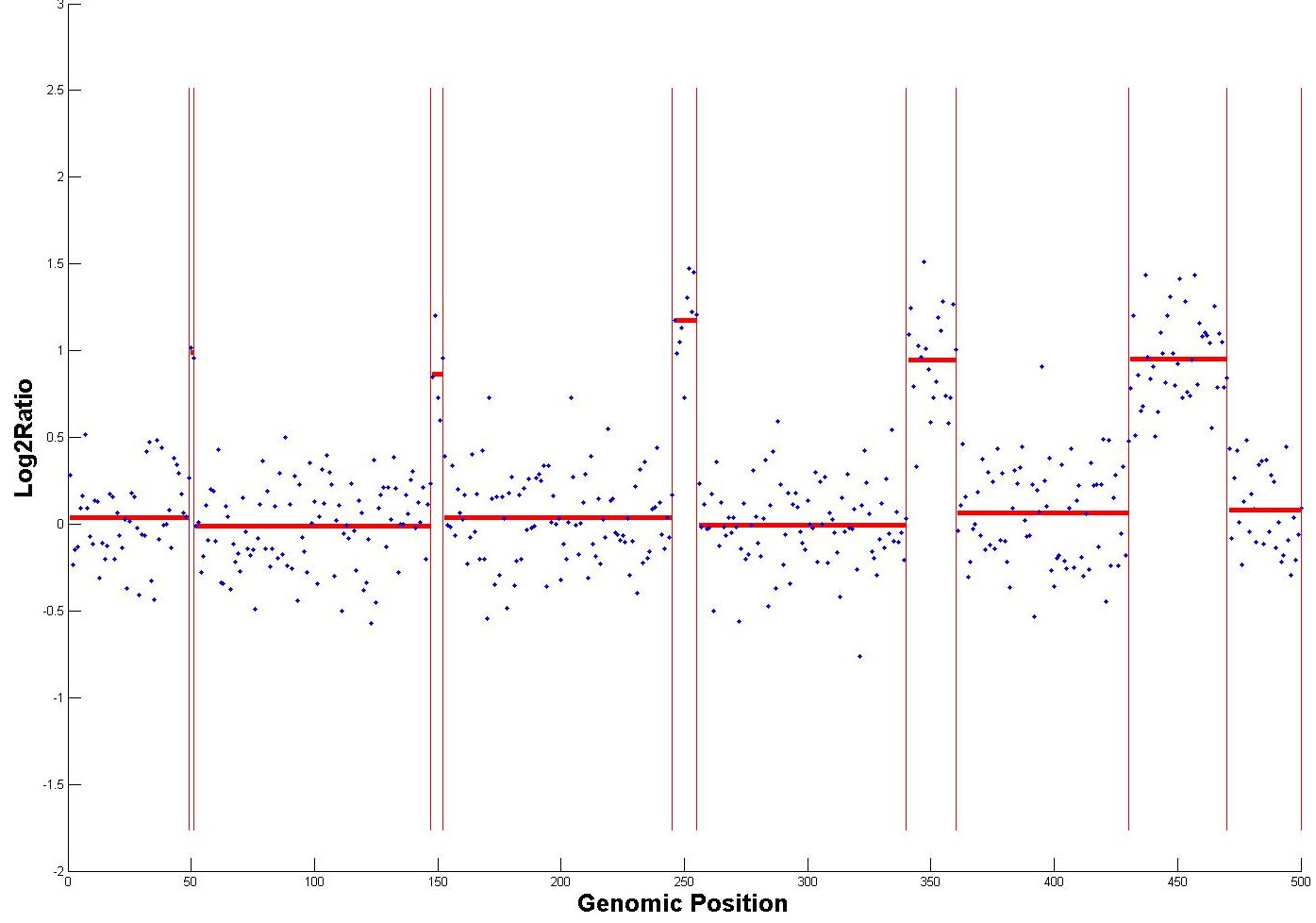} & \includegraphics[width=0.3\textwidth]{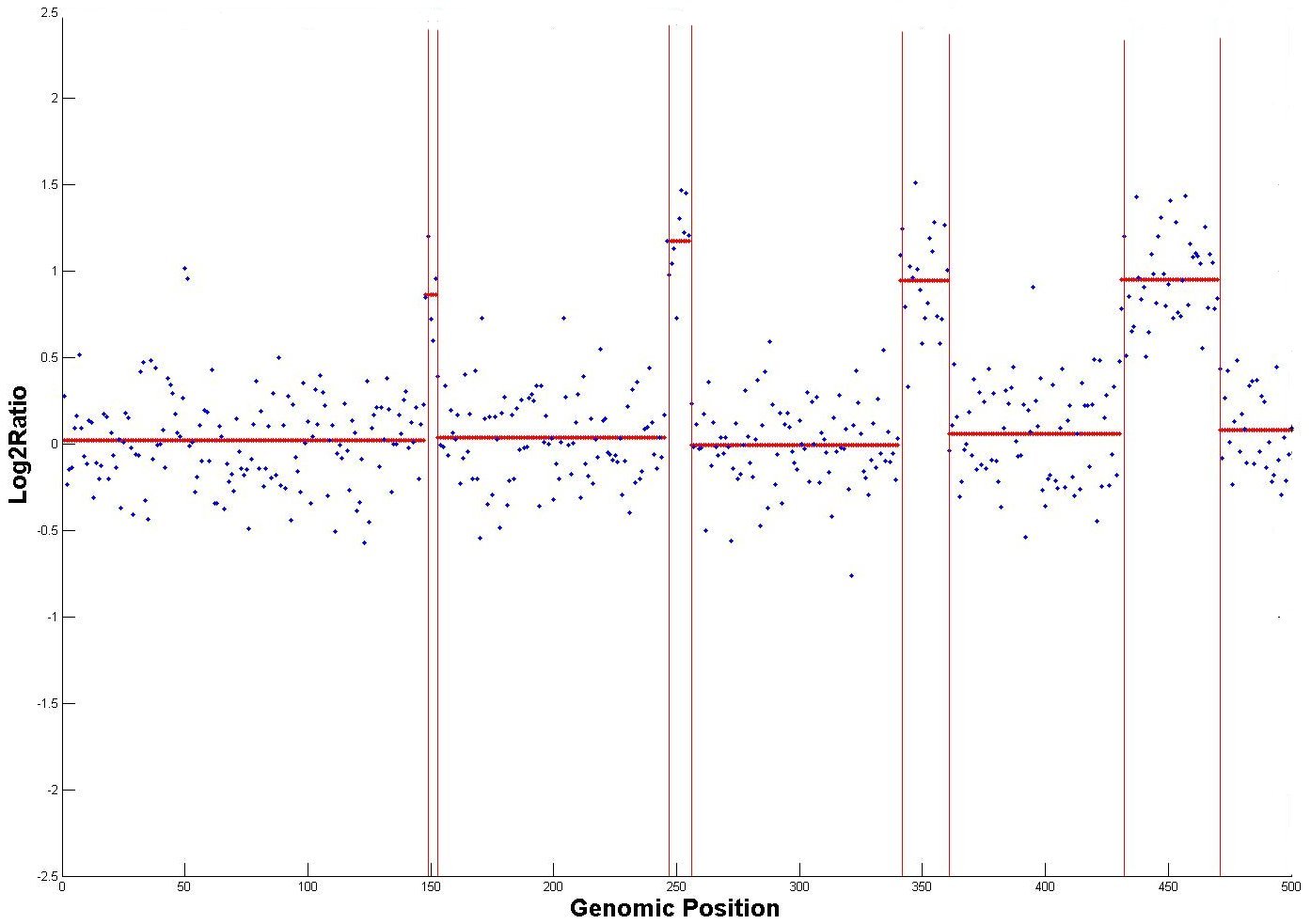} & 		\includegraphics[width=0.3\textwidth]{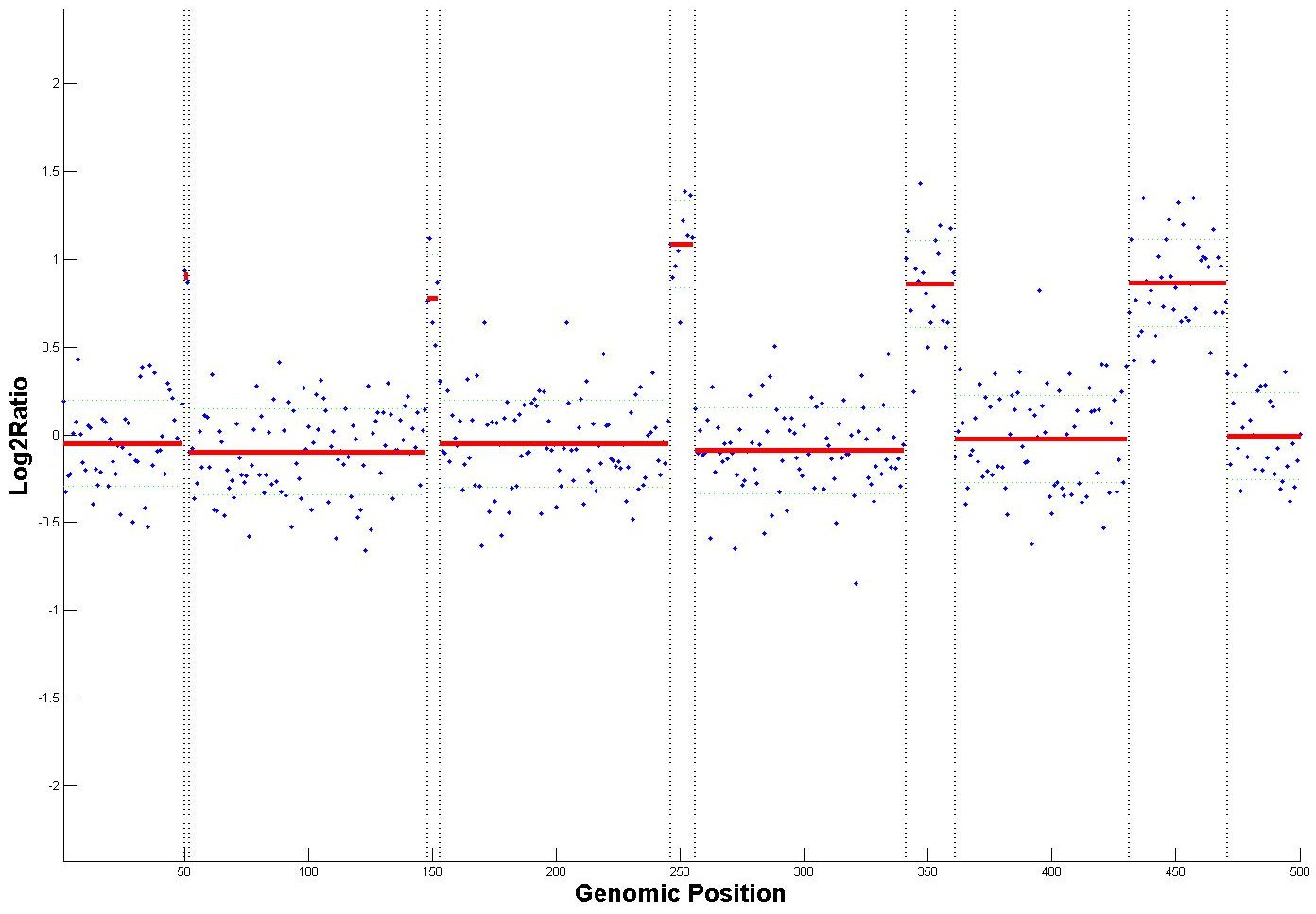}\\
		(a) &(b) &(c) \\
		\trimmer\ &  \dnacopy\  & \cghseg\  \\  
		\end{tabular}
  \caption{Performance of \trimmer, \dnacopy, and \cghseg\ on  denoising synthetic aCGH data from \cite{1181383}. \trimmer\ and \cghseg\ 
exhibit excellent precision and recall whereas \dnacopy\ misses two consecutive genomic positions with DNA copy number equal to 3.}
\label{fig:acghfigsyntheticcorriel}
\end{figure}

We use the synthetic data published in \cite{1181383}. 
The data consist of five aberrations of increasing widths
of 2, 5, 10, 20 and 40 probes, respectively, with Gaussian noise N(0,0.25$^2$).  
Figure~\ref{fig:acghfigsyntheticcorriel} shows the performance of \trimmer, \dnacopy, and \cghseg. 
Both \trimmer\ and \cghseg\ correctly detect all aberrations, while \dnacopy\ misses the first, smallest region.
The running time for \trimmer\ is 0.007 sec, compared to 1.23 sec for \cghseg\ and 60 sec for \dnacopy.

\subsubsection{Coriell Cell Lines}
\label{subsec:Trimmercoriel}

\begin{table}%
\caption{\label{tab:coriellevaluation} Results from applying \trimmer, \dnacopy, and \cghseg\ to 15 cell lines. Rows with listed chromosome numbers (e.g., GM03563/3) corresponded to known gains or losses and are annotated with a check mark if the expected gain or loss was detected or a ``No'' if it was not.  Additional rows list chromosomes on which segments not annotated in the benchmark were detected; we presume these to be false positives. }
\centering
\begin{tabular} {|c|c|c|c|} \hline

Cell Line/Chromosome & \trimmer\ & \dnacopy\ & \cghseg\ \\ \hline
          
GM03563/3     &     \checkmark         & \checkmark    &  \checkmark \\
GM03563/9     &     \checkmark         &    No         &  \checkmark  \\
GM03563/False &        -               &    -          &  - \\ \hline

GM00143/18     &     \checkmark        &   \checkmark   &\checkmark \\
GM00143/False  &      -                &       -      & - \\ \hline
          
GM05296/10 &    \checkmark  & \checkmark  &  \checkmark \\
GM05296/11 &     \checkmark    &  \checkmark  &  \checkmark  \\
GM05296/False &  -   & - & 4,8 \\ \hline
          
GM07408/20 &   \checkmark   &  \checkmark  &  \checkmark \\
GM07408/False &    -   &  - & - \\ \hline 

GM01750/9 &   \checkmark    & \checkmark  &  \checkmark \\
GM01750/14 &   \checkmark    & \checkmark  &  \checkmark  \\
GM01750/False &   -   &  - & - \\ \hline 

GM03134/8 &    \checkmark    & \checkmark  &  \checkmark  \\
GM03134/False &  -  &  - & 1 \\ \hline 
  
GM13330/1 &   \checkmark    &  \checkmark  &   \checkmark  \\
GM13330/4 &   \checkmark    &  \checkmark  &   \checkmark   \\
GM13330/False &   -         &     -        &  -    \\\hline
          
GM03576/2 &   \checkmark    & \checkmark   &  \checkmark \\
GM03576/21 &   \checkmark   & \checkmark   &  \checkmark \\
GM03576/False &   -         &      -       &  - \\\hline
          
GM01535/5 &    \checkmark   &   \checkmark  & \checkmark \\
GM01535/12 &    \checkmark  &    No         & \checkmark \\
GM01535/False & -           &   -           & 8 \\\hline
          
GM07081/7 &    \checkmark   &   \checkmark & \checkmark \\
GM07081/15 &    No   &   No & No \\
GM07081/False &  -   & - & 11 \\\hline
          
GM02948/13 &      \checkmark   &  \checkmark  &  \checkmark  \\
GM02948/False &  7    & 1 &  2 \\ \hline

GM04435/16 &    \checkmark   &  \checkmark  &  \checkmark  \\
GM04435/21 &     \checkmark  &  \checkmark  & \checkmark \\
GM04435/False &  -  & - & 8,17 \\\hline
          
GM10315/22 &     \checkmark   &  \checkmark  & \checkmark\\
GM10315/False &  -    & - & - \\ \hline 

GM13031/17 &       \checkmark  &  \checkmark  & \checkmark  \\
GM13031/False &    -   &  -  &   -  \\ \hline
         
GM01524/6 &        \checkmark   &  \checkmark  & \checkmark  \\
GM01524/False &     -  &  -  &- \\  \hline
\end{tabular}
\end{table}

The first real dataset we use to evaluate our method is the Coriell cell line BAC
array CGH data \cite{coriell}, which is widely considered a ``gold standard'' 
dataset. The dataset is derived from 15 fibroblast cell lines 
using the normalized average of 
$\log_2$ fluorescence relative to a diploid reference.
To call gains or losses of inferred segments, we assign to each segment the mean intensity of its probes and then apply a simple threshold test to determine if the mean is abnormal.  We follow \cite{thresholds} in favoring $\pm$0.3 out of the wide variety of thresholds that have been used \cite{thresholds2}.

Table~\ref{tab:coriellevaluation} summarizes  the performance of \trimmer, \dnacopy\ and \cghseg\ relative to previously annotated gains and losses in the Corielle dataset. The table shows notably better performance for \trimmer\ compared to either alternative method.  \trimmer\ finds 22 of 23 expected segments with one false positive.  \dnacopy\ finds 20 of 23 expected segments with one false positive.  \cghseg\ finds 22 of 23 expected segments with seven false positives.  \trimmer\ thus achieves the same recall as \cghseg\ while outperforming it in precision and the same precision as \dnacopy\ while outperforming it in recall.
In cell line GM03563, \dnacopy\ fails to detect a region of two points
which have undergone a loss along chromosome 9, in accordance with the results 
obtained using the Lai et al. \cite{1181383} synthetic data. 
In cell line GM03134, \cghseg\ makes a false positive along chromosome 1
which both \trimmer\ and \dnacopy\ avoid. 
In cell line GM01535, \cghseg\ makes a false positive along chromosome 8
and \dnacopy\ misses the aberration along chromosome 12.
\trimmer, however, performs ideally on this cell line. 
In cell line GM02948, \trimmer\ makes a false positive along chromosome 7, finding
a one-point segment in 7q21.3d at genomic position 97000 
whose value is equal to 0.732726. All other methods also make false positive errors
on this cell line.  In GM7081, all three methods fail to find an annotated aberration on chromosome 15.  In addition, \cghseg\ finds a false positive on chrosome 11.

\trimmer\ also substantially outperforms the comparative methods in
run time, requiring 5.78 sec for the full data set versus 8.15 min for
\cghseg\ (an 84.6-fold speedup) and 47.7 min for \dnacopy\ (a 495-fold
speedup).

\subsubsection{Breast Cancer Cell Lines}
\label{subsec:Trimmerbreast}

To illustrate further the performance of \trimmer\ and compare it to \dnacopy\ and \cghseg, we applied it to the Berkeley Breast Cancer cell line database \cite{berkeley}.
The dataset consists of 53 breast cancer cell lines that capture most of the recurrent genomic and 
transcriptional characteristics of 145 primary breast cancer cases. We do not have an accepted ``answer key'' for this data set, but it provides a more extensive basis for detailed comparison of differences in performance of the methods on common data sets, as well as an opportunity for novel discovery.  While we have applied the methods to all chromosomes in all cell lines, space limitations prevent us from presenting the full results here. The interested reader can reproduce all the results including the ones 
not presented here\footnote{\url{http://www.math.cmu.edu/~ctsourak/DPaCGHsupp.html}}. 
We therefore arbitrarily selected three of the 53 cell lines and selected three chromosomes per cell line that we believed would best illustrate the comparative performance of the methods.  The Genes-to-Systems Breast Cancer
Database \footnote{\url{http://www.itb.cnr.it/breastcancer}}~\cite{G2SBCD} was used to identify known breast cancer markers in regions predicted to be gained or lost by at least one of the methods. We used the UCSC Genome Browser
\footnote{\url{http://genome.ucsc.edu/}} to verify the placement of genes.
 
We note that \trimmer\ again had a substantial advantage in run time.  For the full data set, \trimmer\ required 22.76 sec, compared to 23.3 min for \cghseg\ (a 61.5-fold increase), and 4.95 hrs for \dnacopy\ (a 783-fold increase).

\paragraph{\textbf{Cell Line BT474:}}
Figure~\ref{fig:acghbt474} shows the performance of each method on the BT474 cell line.  The three methods report different results for chromosome 1, as shown in Figures~\ref{fig:acghbt474}(a,b,c), with all three detecting amplification in the q-arm but differing in the detail of resolution.  
\trimmer\ is the only method that detects region 1q31.2-1q31.3 as aberrant. 
This region hosts gene NEK7, a candidate oncogene \cite{nek7} and gene KIF14, 
a predictor of grade and outcome in breast cancer \cite{KIF14}.
\trimmer\ and \dnacopy\ annotate the region 1q23.3-1q24.3 as
amplified. This region hosts several genes previously implicated in breast
cancer~\cite{G2SBCD}, such 
as CREG1 (1q24), POU2F1 (1q22-23), RCSD1 (1q22-q24), and BLZF1 (1q24).
Finally, \trimmer\ alone reports independent amplification of the gene CHRM3, a marker of metastasis in breast cancer patients~\cite{G2SBCD}.

\begin{figure*}
		\begin{tabular}{ccc} 
		\includegraphics[width=0.33\textwidth]{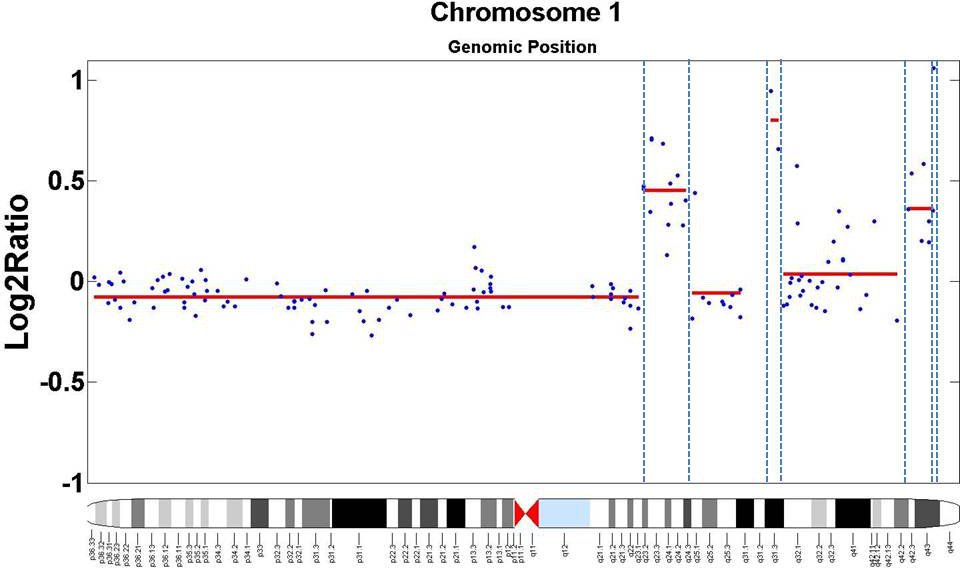} & \includegraphics[width=0.33\textwidth]{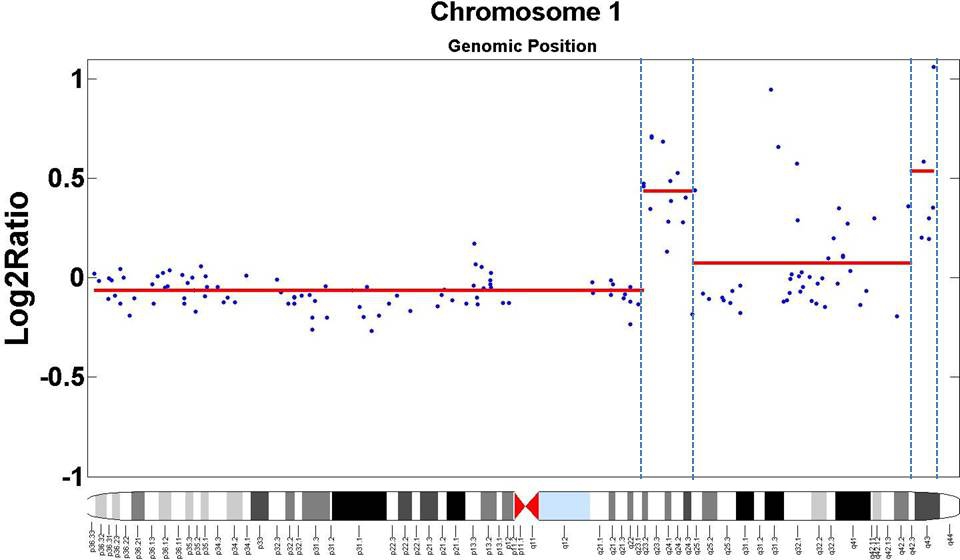} &    \includegraphics[width=0.33\textwidth]{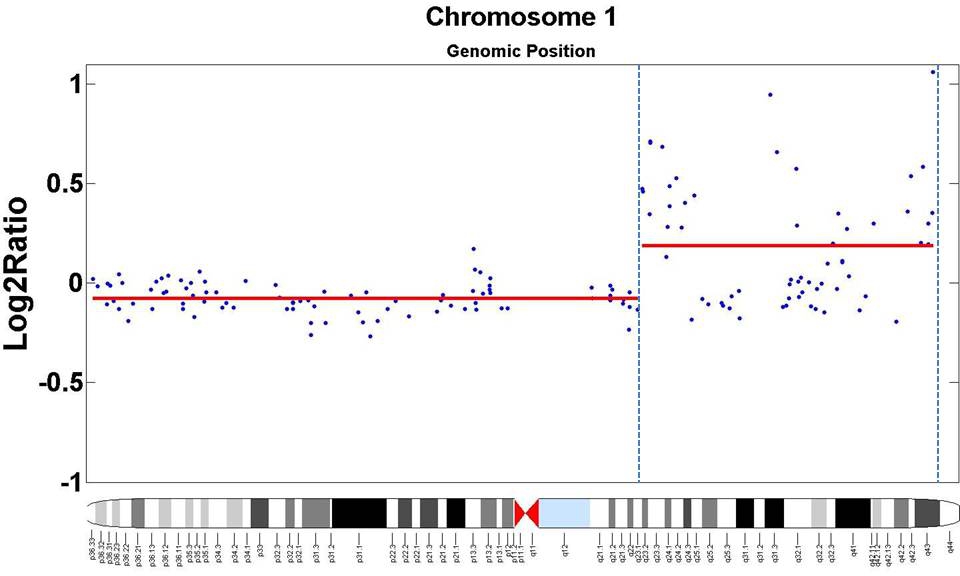}   \\
		(a) &(b) &(c) \\
		\includegraphics[width=0.33\textwidth]{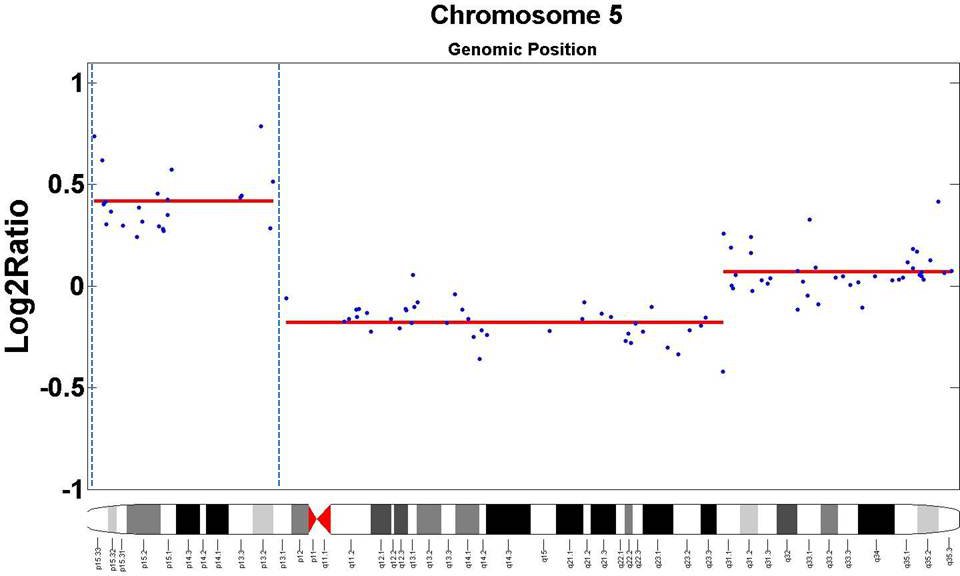} &		\includegraphics[width=0.33\textwidth]{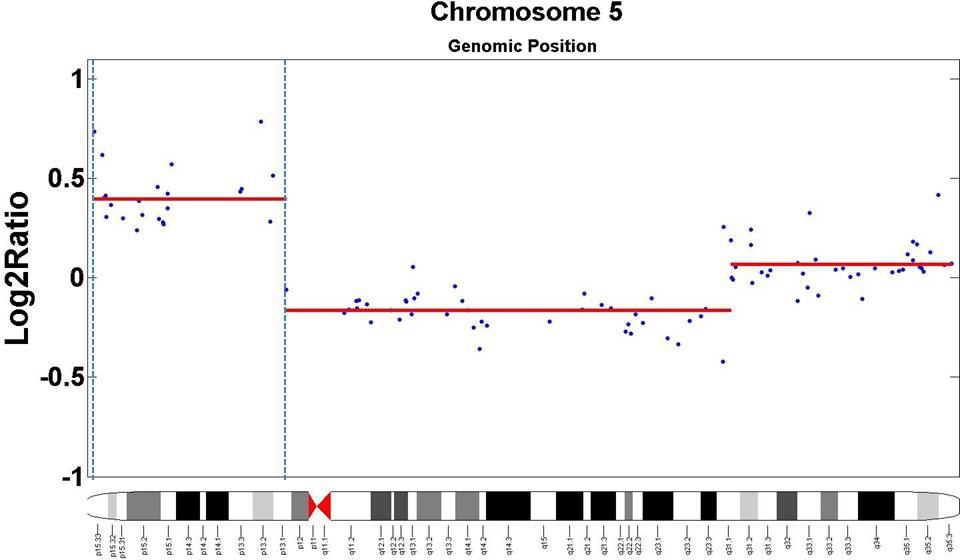} &    \includegraphics[width=0.33\textwidth]{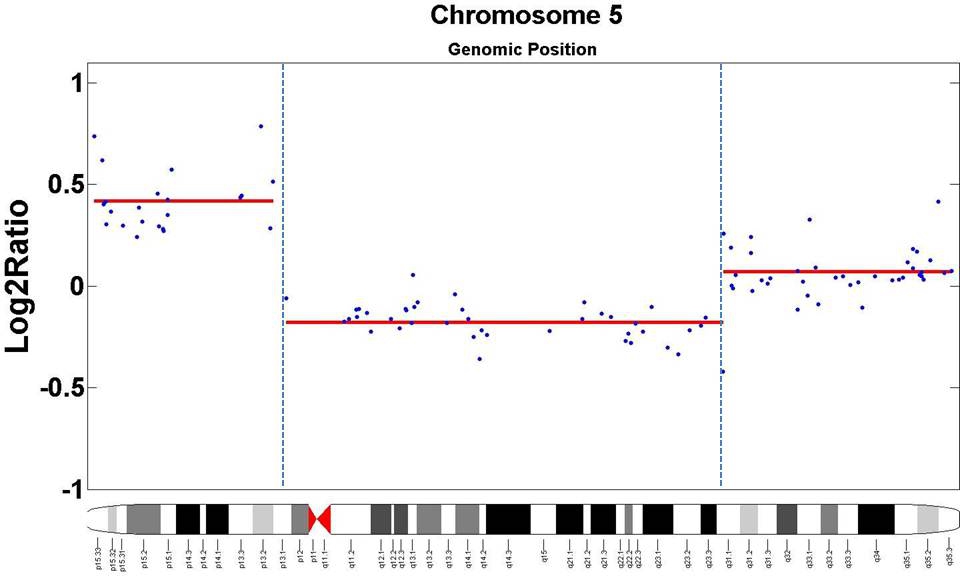}   \\
		(d) &(e) &(f) \\
		\includegraphics[width=0.33\textwidth]{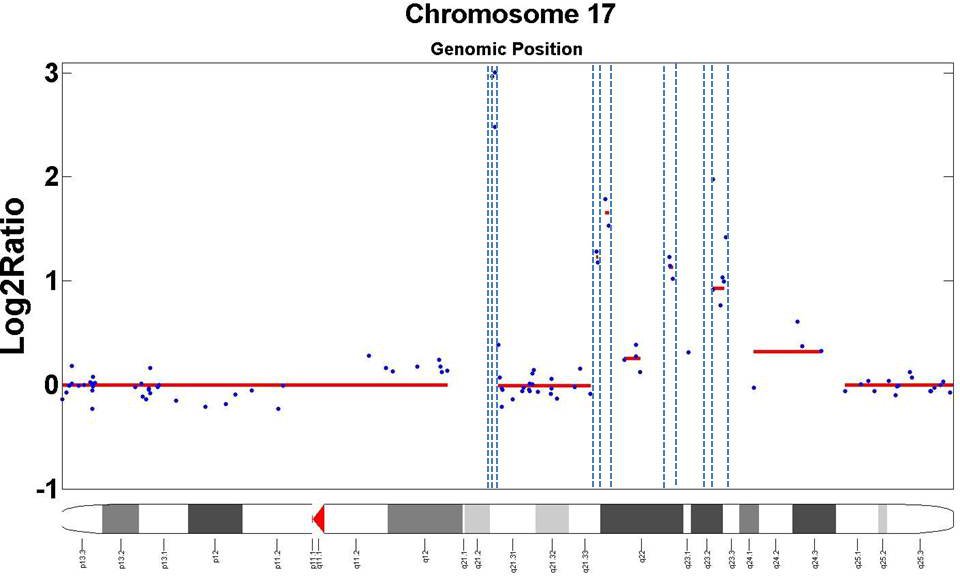} & \includegraphics[width=.33\textwidth]{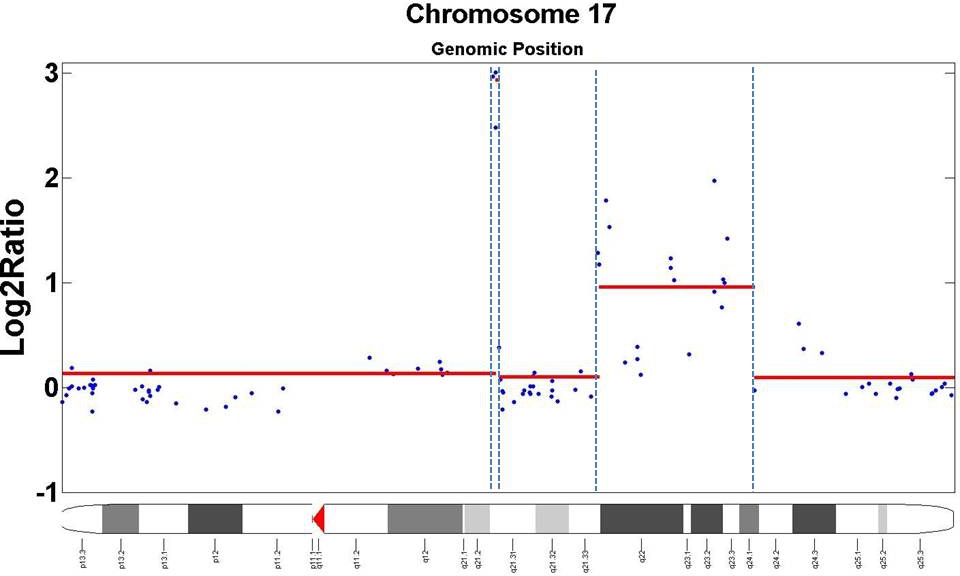} &    \includegraphics[width=0.33\textwidth]{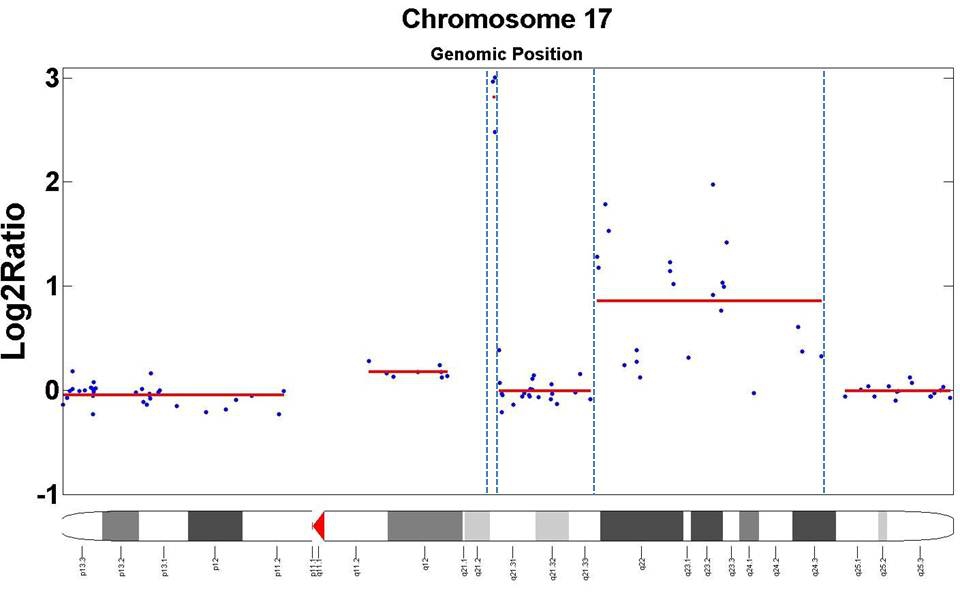}  \\
		(g) &(h) &(i) \\
		   &     & \\ 
		\trimmer\ &  \dnacopy\  & \cghseg\  \\  	
		\end{tabular}
\caption{Visualization of the segmentation output of \trimmer, \dnacopy, and \cghseg\ for the cell line BT474 on chromosomes 1 (a,b,c), 5 (d,e,f), and 17 (g,h,i).  (a,d,g) \trimmer\ output.  (b,e,h) \dnacopy\ output. (c,f,i) \cghseg\ output. Segments exceeding the $\pm$ 0.3 threshold \cite{thresholds} are highlighted.}
\label{fig:acghbt474}
\end{figure*}

For chromosome 5 (Figures~\ref{fig:acghbt474}(d,e,f)), the behavior of the
three methods is almost identical. All methods report
amplification of a region known to contain many breast cancer markers,
including MRPL36 (5p33), ADAMTS16 (5p15.32), POLS (5p15.31), ADCY2
(5p15.31), CCT5 (5p15.2), TAS2R1 (5p15.31), ROPN1L (5p15.2), DAP
(5p15.2), ANKH (5p15.2), FBXL7 (5p15.1), BASP1 (5p15.1), CDH18
(5p14.3), CDH12 (5p14.3), CDH10 (5p14.2 - 5p14.1), CDH9 (5p14.1) PDZD2
(5p13.3), GOLPH3 (5p13.3), MTMR12 (5p13.3), ADAMTS12 (5p13.3 -
5p13.2), SLC45A2 (5p13.2), TARS (5p13.3), RAD1 (5p13.2), AGXT2
(5p13.2), SKP2 (5p13.2), NIPBL (5p13.2), NUP155 (5p13.2), KRT18P31
(5p13.2), LIFR (5p13.1) and GDNF (5p13.2)~\cite{G2SBCD}.  The only
difference in the assignments is that \dnacopy\ fits one more probe to
this amplified segment.

Finally, for chromosome 17 (Figures~\ref{fig:acghbt474}(g,h,i)), like chromosome 1, all methods detect amplification but \trimmer\ predicts a finer breakdown of the amplified region into independently amplified segments.  All three methods detect amplification of a region which includes the major breast cancer biomarkers HER2 (17q21.1) and BRCA1 (17q21) as also the additional markers MSI2 (17q23.2) and  TRIM37 (17q23.2)~\cite{G2SBCD}.  While the more discontiguous picture produced by \trimmer\ may appear to be a less parsimonious explanation of the data, a complex combination of fine-scale gains and losses in 17q is in fact well supported by the literature \cite{Orsetti}.

\paragraph{ \textbf{Cell Line HS578T:}}

\begin{figure*}
		\begin{tabular}{ccc} 
		\includegraphics[width=0.33\textwidth]{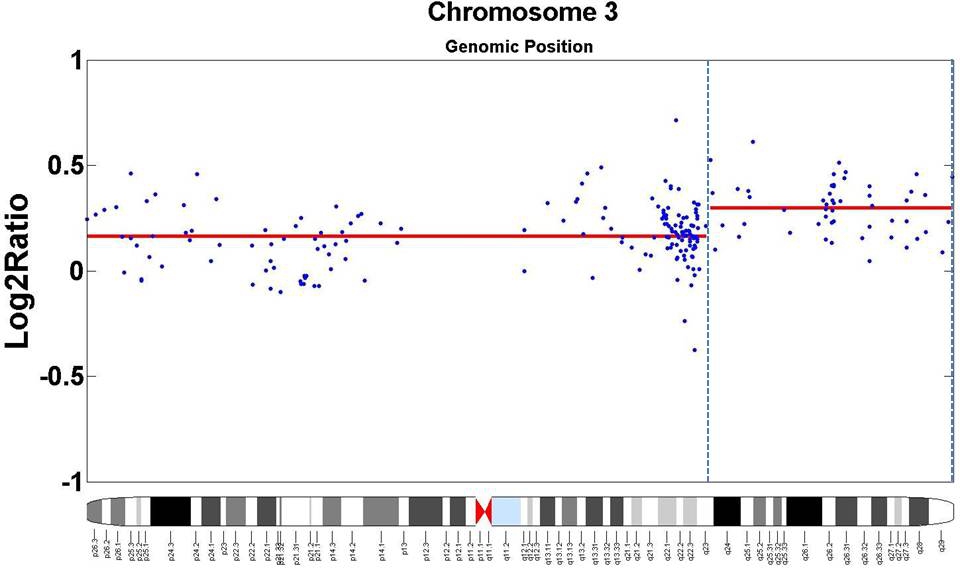} & \includegraphics[width=0.33\textwidth]{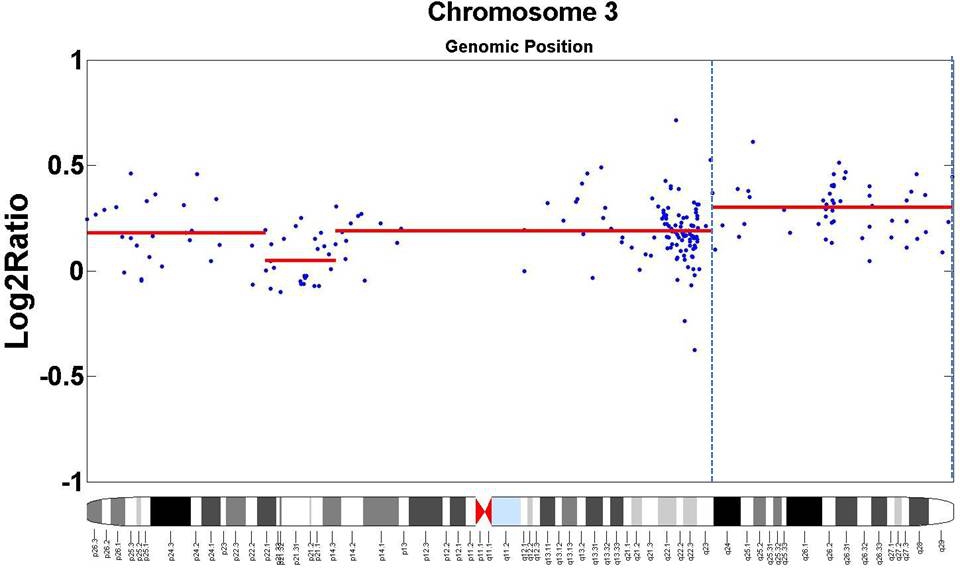} &    \includegraphics[width=0.33\textwidth]{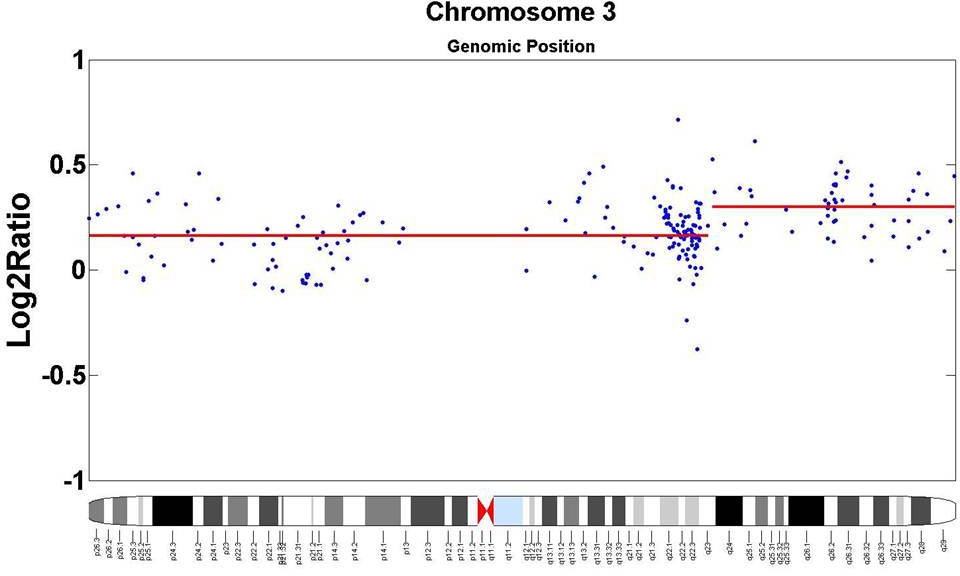} \\
		(a) &(b) &(c) \\
		 \includegraphics[width=0.31\textwidth]{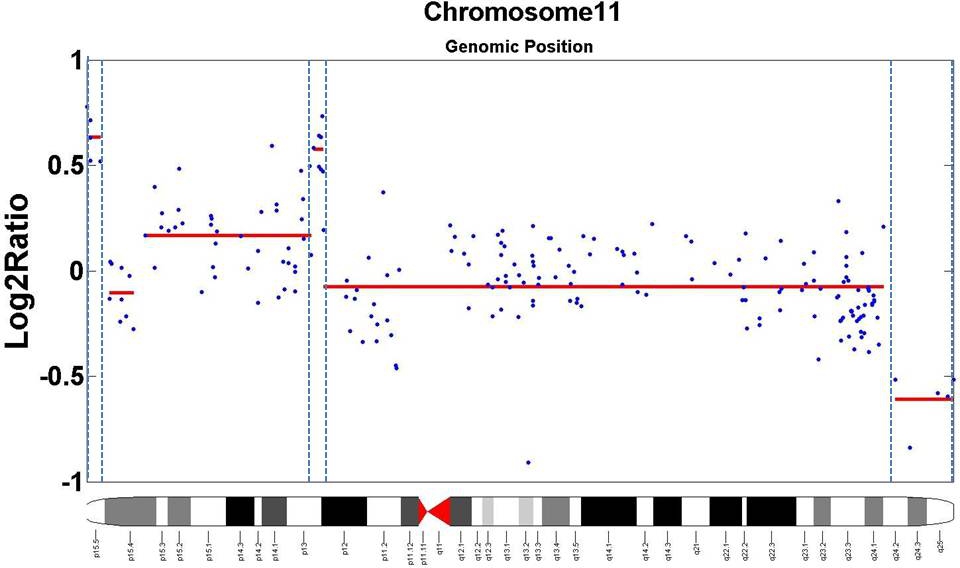} & \includegraphics[width=0.33\textwidth]{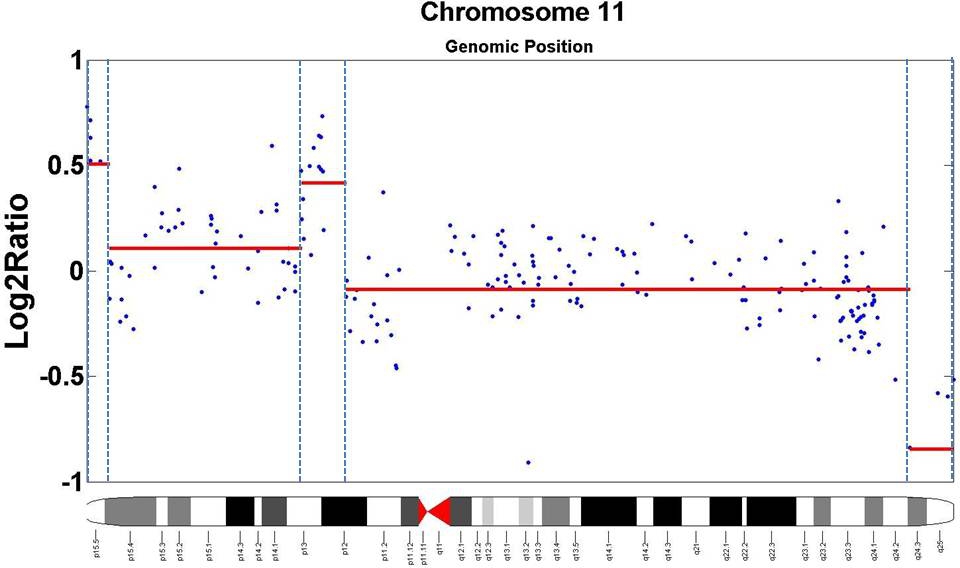} &    \includegraphics[width=0.33\textwidth]{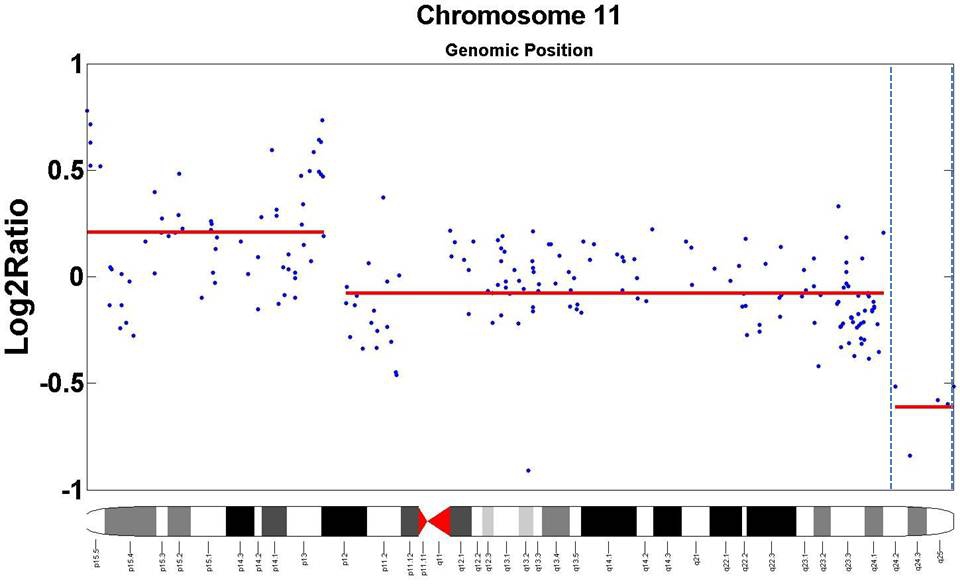} \\
		(d) &(e) &(f) \\
		\includegraphics[width=0.33\textwidth]{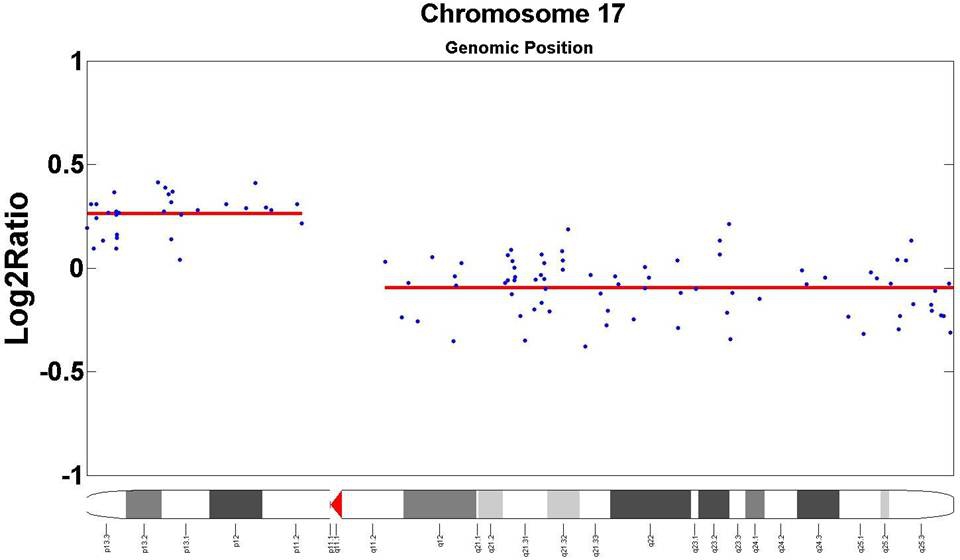} & \includegraphics[width=0.33\textwidth]{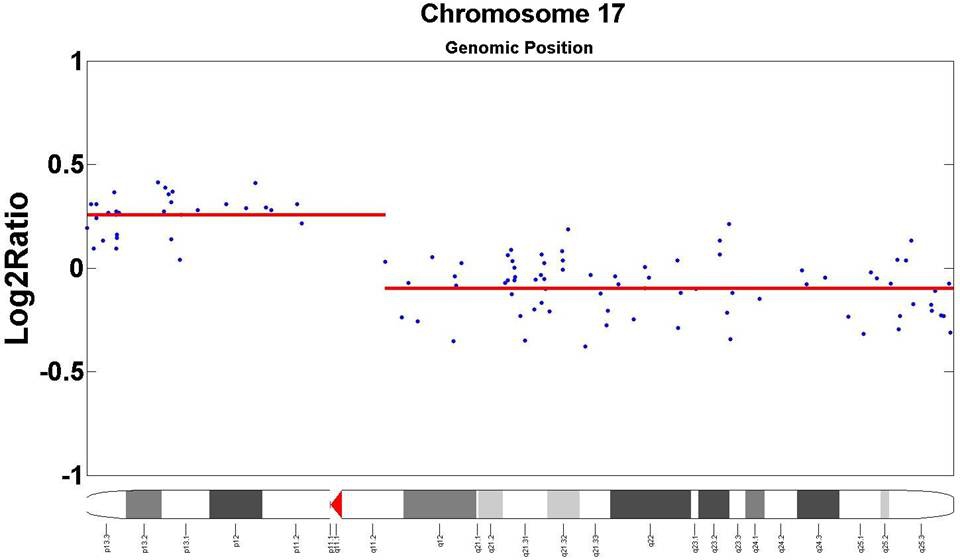} &    \includegraphics[width=0.32\textwidth]{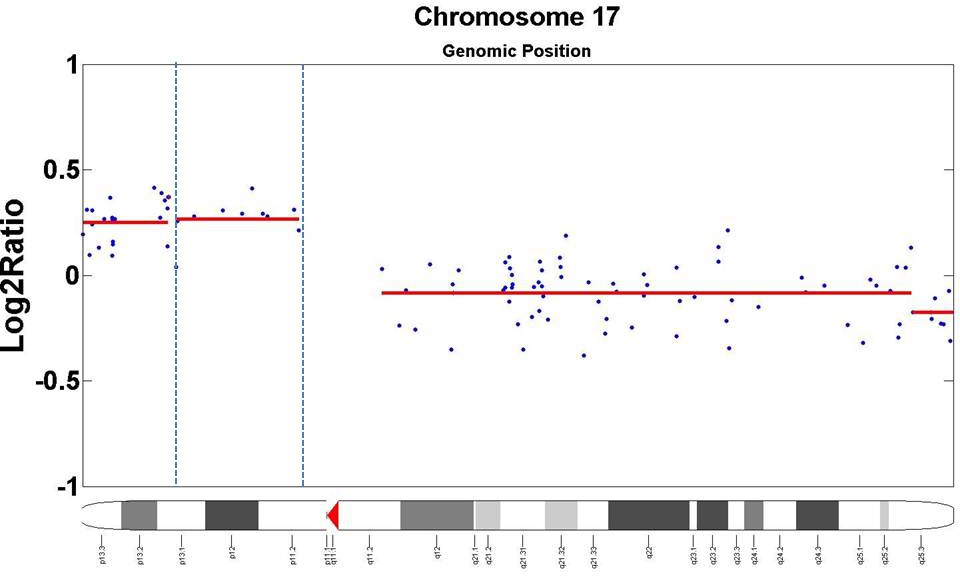} \\
		(g) &(h) &(i) \\
		    &     & \\ 
		\trimmer\ &  \dnacopy\  & \cghseg\  \\  
		\end{tabular}
\caption{Visualization of the segmentation output of \trimmer, \dnacopy, and \cghseg\ for the cell line HS578T on chromosomes 3 (a,b,c), 11 (d,e,f), and 17 (g,h,i).  (a,d,g) \trimmer\ output.  (b,e,h) \dnacopy\ output. (c,f,i) \cghseg\ output. Segments exceeding the $\pm$ 0.3 threshold \cite{thresholds} are highlighted.}
\label{fig:acghHS578T}
\end{figure*}

Figure~\ref{fig:acghHS578T} compares the methods on cell line HS578T for
chromosomes 3, 11 and 17.  Chromosome 3 (Figures~\ref{fig:acghHS578T}(a,b,c))
shows identical prediction of an amplification of 3q24-3qter for all
three methods. This region includes the key breast cancer markers PIK3CA
(3q26.32)~\cite{pik3ca}, and additional breast-cancer-associated genes
TIG1 (3q25.32), MME (3q25.2), TNFSF10 (3q26), MUC4 (3q29), TFRC
(3q29), DLG1 (3q29)~\cite{G2SBCD}. \trimmer\ and \cghseg\ also make
identical predictions of normal copy number in the p-arm, while
\dnacopy\ reports an additional loss between 3p21 and 3p14.3.  We are
unaware of any known gain or loss in this region associated with
breast cancer.

For chromosome 11 (Figures~\ref{fig:acghHS578T}(d,e,f)), the methods again
present an identical picture of loss at the q-terminus
(11q24.2-11qter) but detect amplifications of the p-arm at different
levels of resolution.  \trimmer\ and \dnacopy\ detect gain in the
region 11p15.5, which is the site of the HRAS breast cancer metastasis
marker~\cite{G2SBCD}.  In contrast to \dnacopy, \trimmer\ detects an
adjacent loss region.  While we have no direct evidence this loss is a
true finding, the region of predicted loss does contain EIF3F
(11p15.4), identified as a possible tumor suppressor whose expression
is decreased in most pancreatic cancers and melanomas~\cite{G2SBCD}.
Thus, we conjecture that EIF3F is a tumor suppressor in
breast cancer.

\begin{figure*}
		\begin{tabular}{ccc} 
		\includegraphics[width=0.33\textwidth]{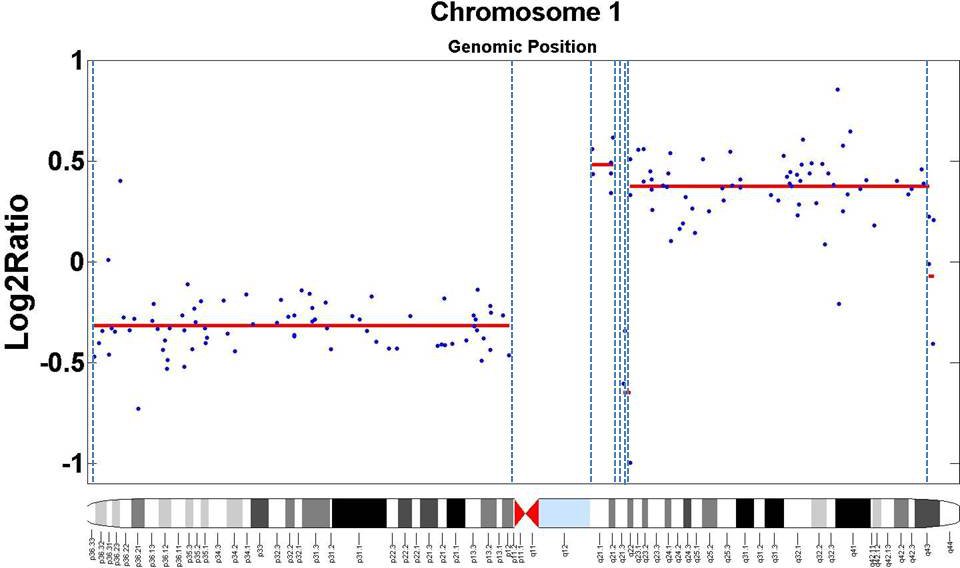} & \includegraphics[width=0.33\textwidth]{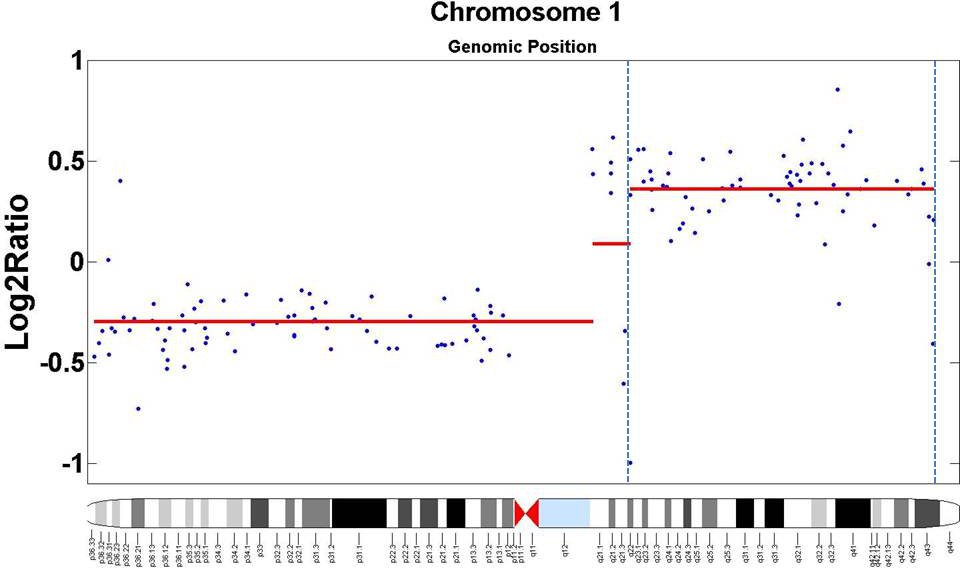} &    \includegraphics[width=0.33\textwidth]{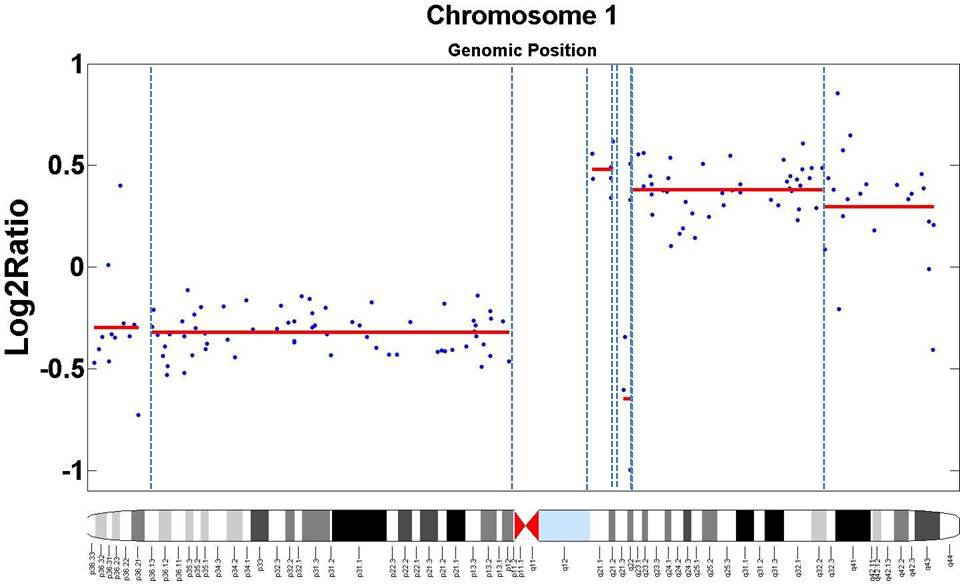}  \\
		(a) &(b) &(c) \\
		\includegraphics[width=0.33\textwidth]{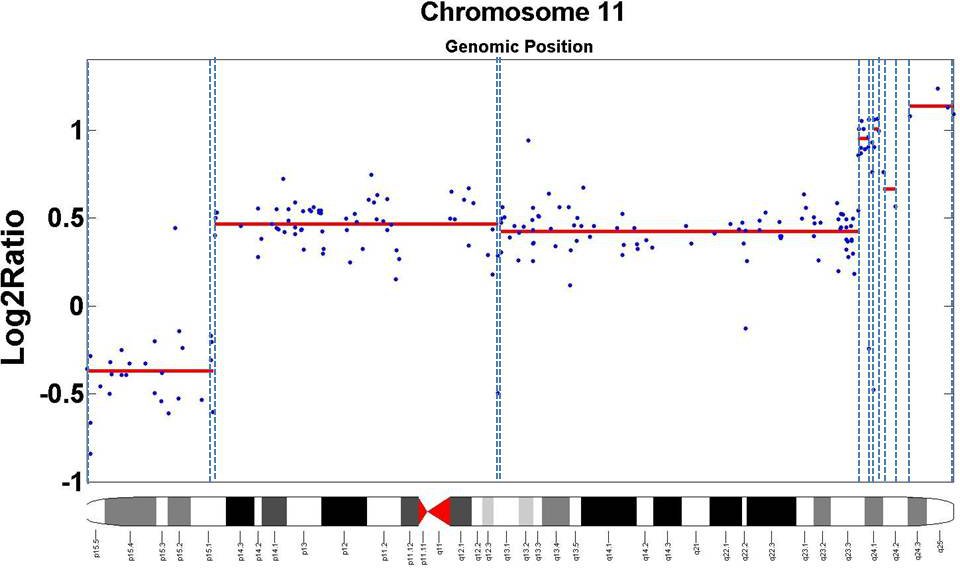} & \includegraphics[width=0.33\textwidth]{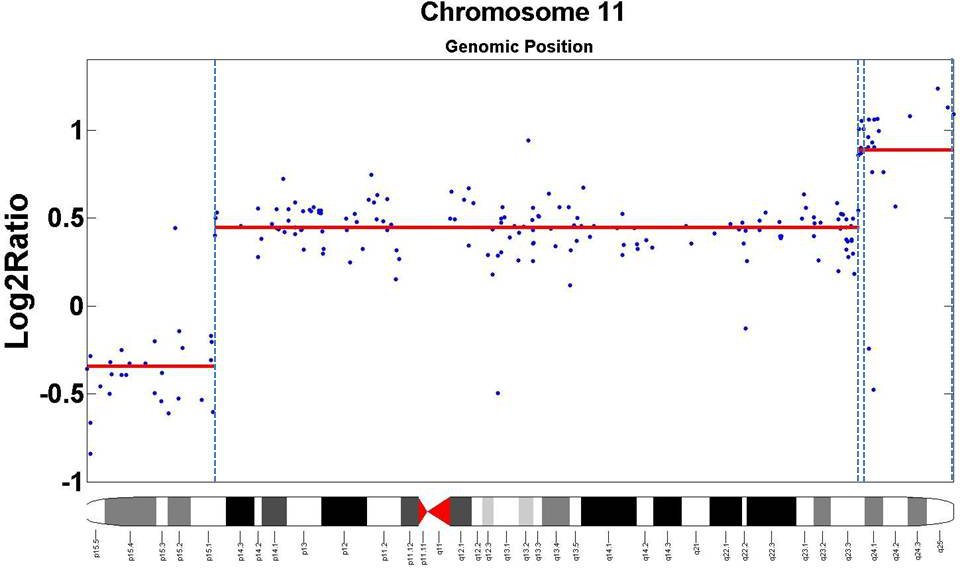} &    \includegraphics[width=0.33\textwidth]{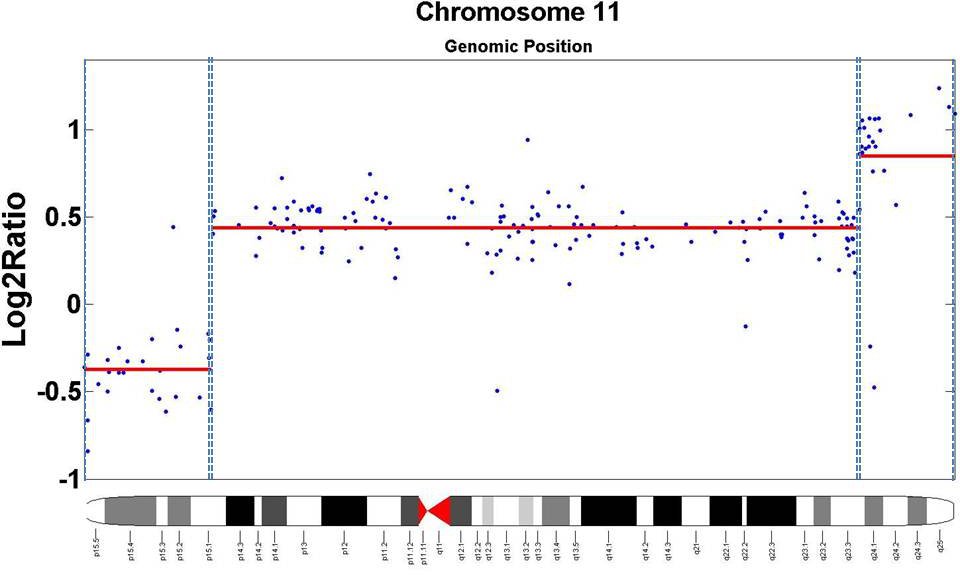} \\
		(d) &(e) &(f) \\
		\includegraphics[width=0.33\textwidth]{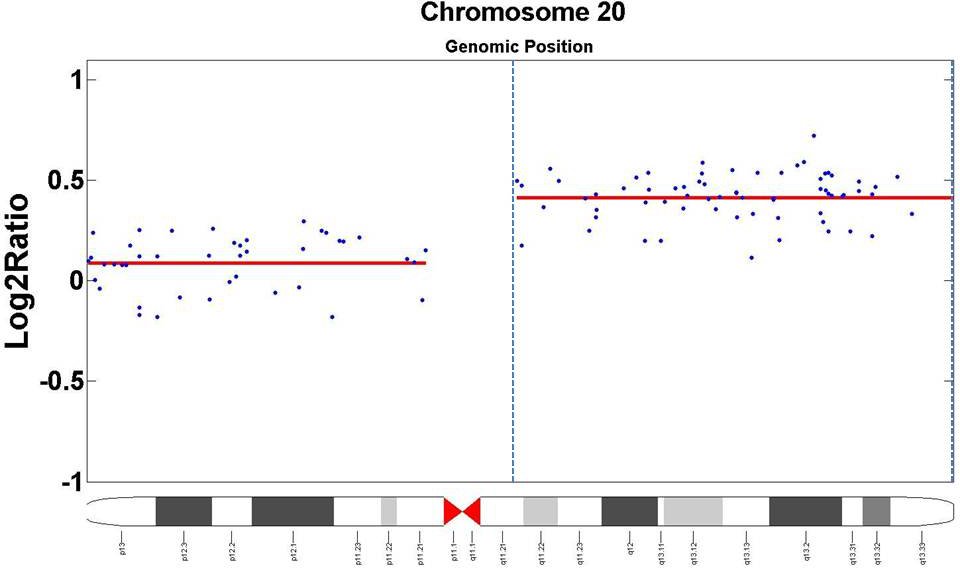} & \includegraphics[width=0.33\textwidth]{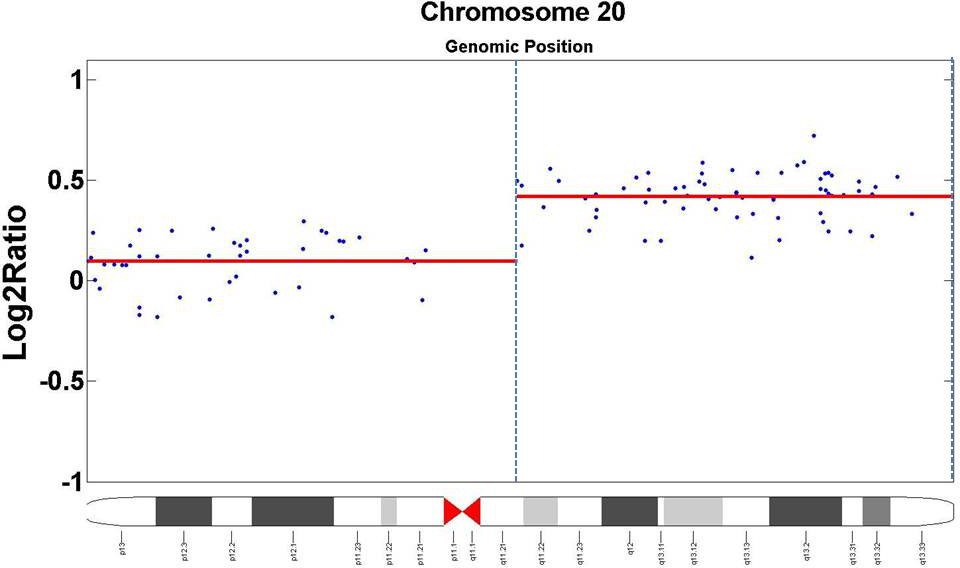} &    \includegraphics[width=0.32\textwidth]{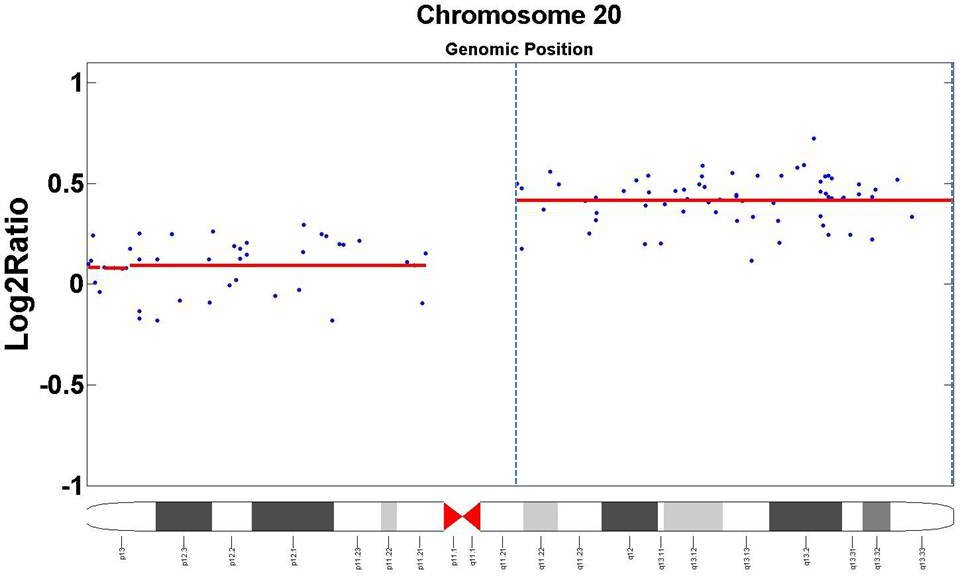} \\
		(g) &(h) &(i) \\
		&     & \\ 
		\trimmer\ &  \dnacopy\  & \cghseg\  \\  
		\end{tabular}
\caption{Visualization of the segmentation output of \trimmer, \dnacopy, and \cghseg\ for the cell line T47D on chromosomes 1 (a,b,c), 11 (d,e,f), and 20 (g,h,i).  (a,d,g) \trimmer\ output.  (b,e,h) \dnacopy\ output. (c,f,i) \cghseg\ output. Segments exceeding the $\pm$ 0.3 threshold \cite{thresholds} are highlighted.}
\label{fig:acghT47D}
\end{figure*}

On chromosome 17  (Figures~\ref{fig:acghHS578T}(g,h,i)), the three methods behave similarly, with all three predicting amplification of the p-arm.  \dnacopy\ places one more marker in the amplified region causing it to cross the centromere while \cghseg\ breaks the amplified region into three segments by predicting additional amplification at a single marker.

\paragraph{\textbf{Cell Line T47D:}}

\begin{figure}
		\begin{tabular}{ccc} 
\includegraphics[width=0.31\textwidth]{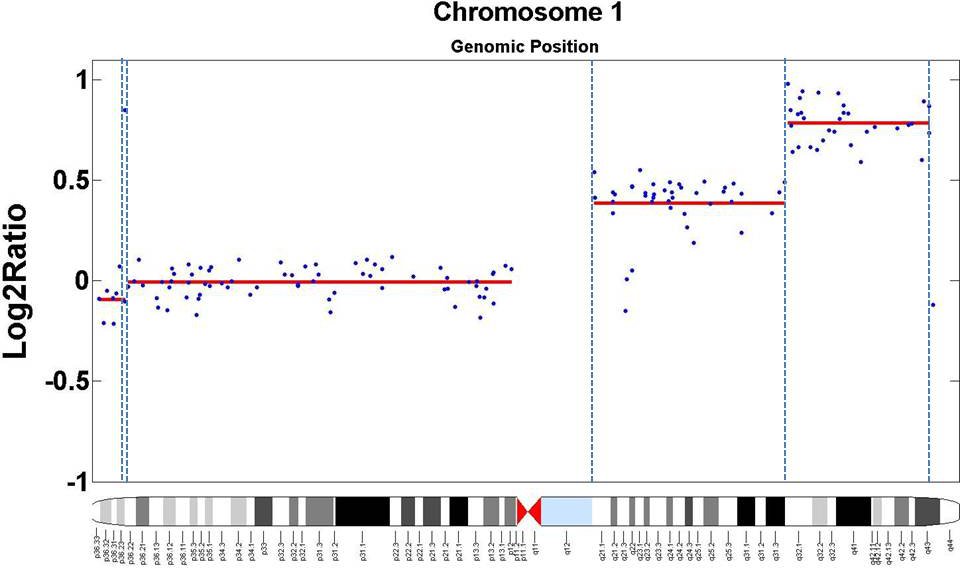} &		\includegraphics[width=0.32\textwidth]{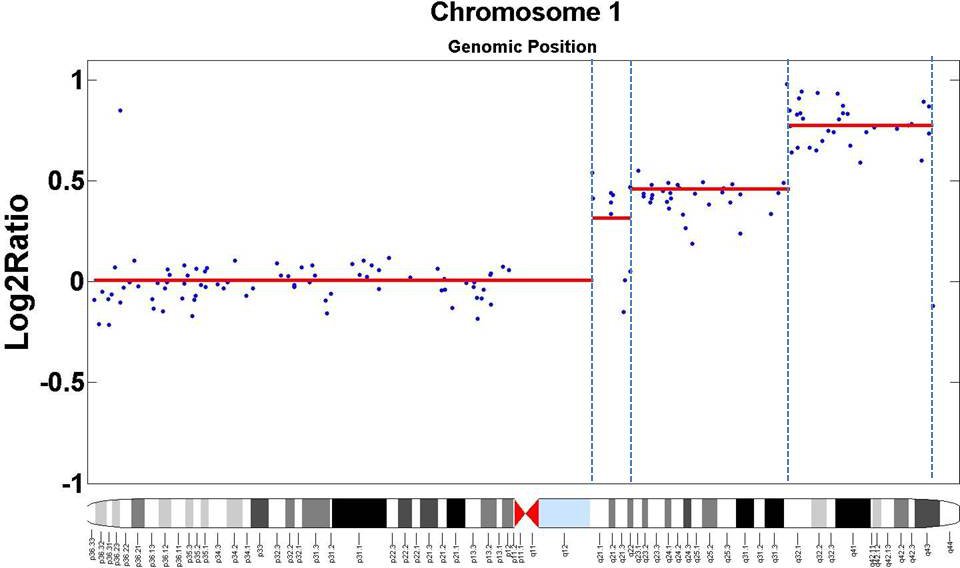} &    \includegraphics[width=0.33\textwidth]{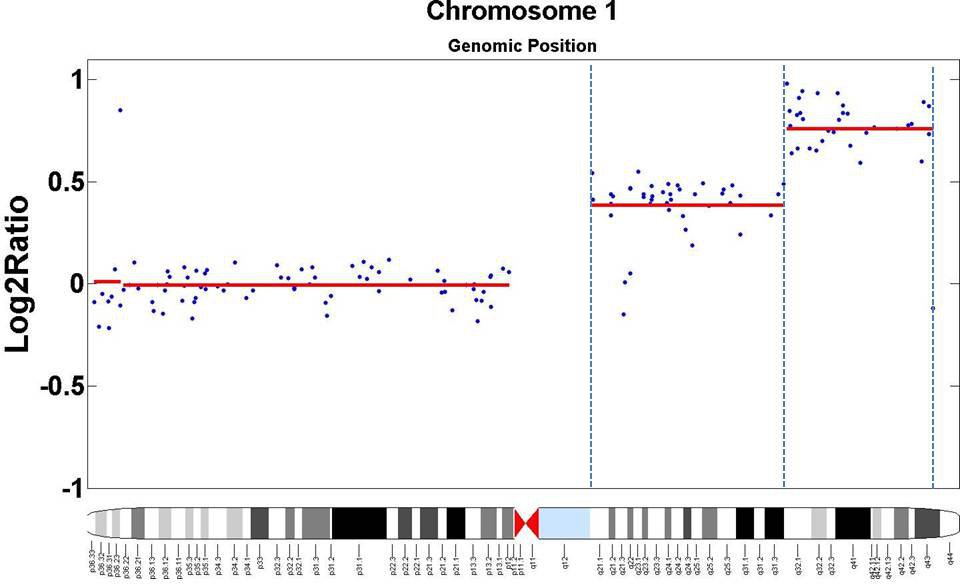} \\
		(a) &(b) &(c) \\
 \includegraphics[width=0.3\textwidth]{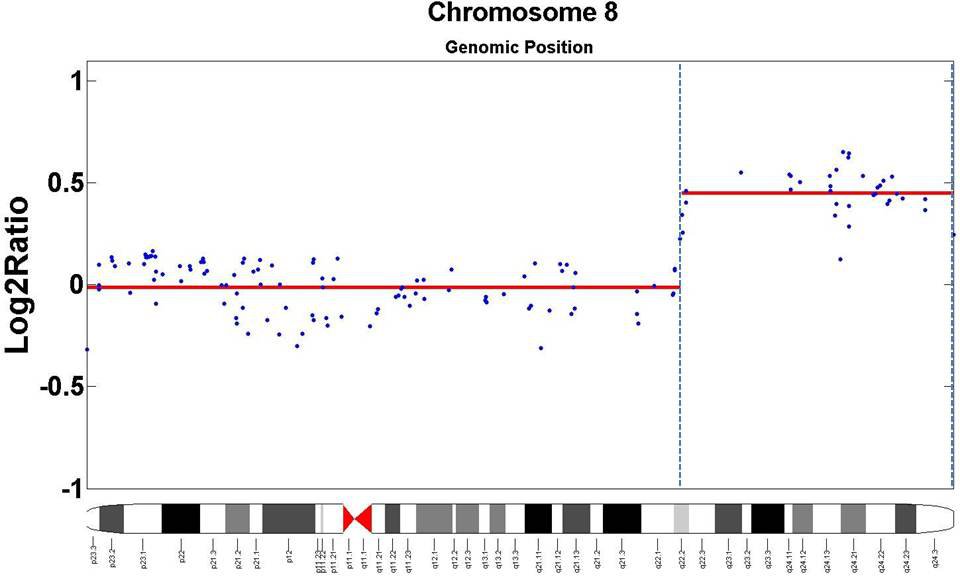} &		\includegraphics[width=0.3\textwidth]{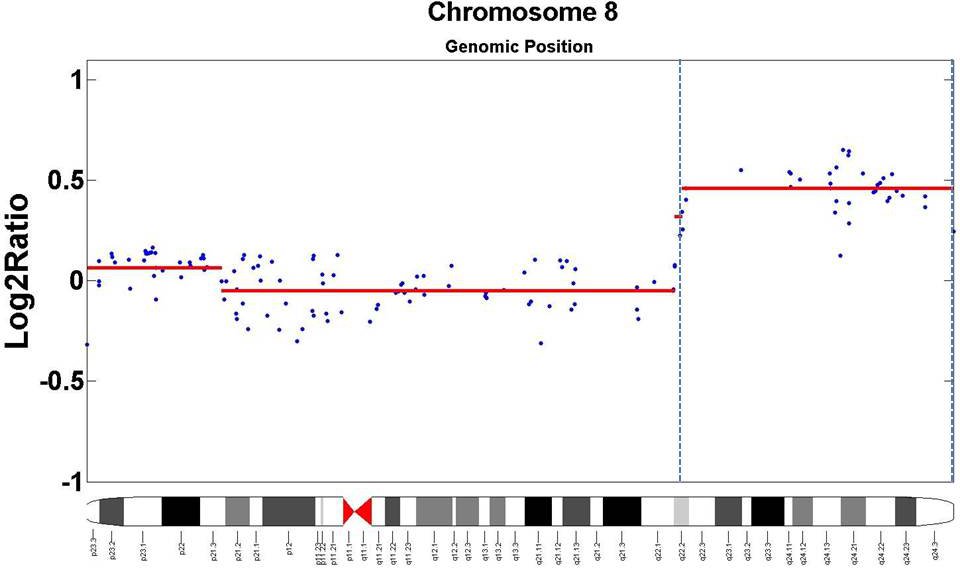} &    \includegraphics[width=0.3\textwidth]{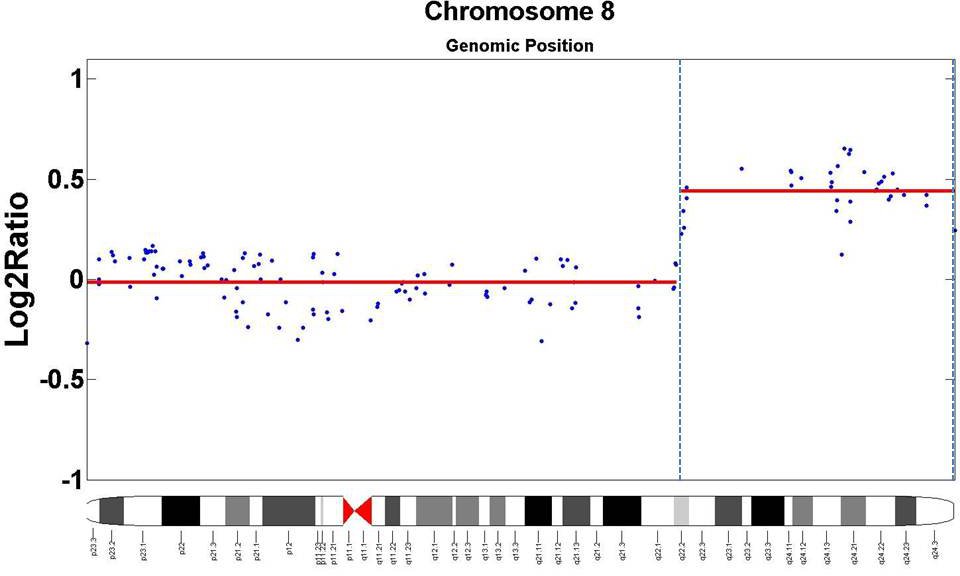} \\
		(d) &(e) &(f) \\
		\includegraphics[width=0.3\textwidth]{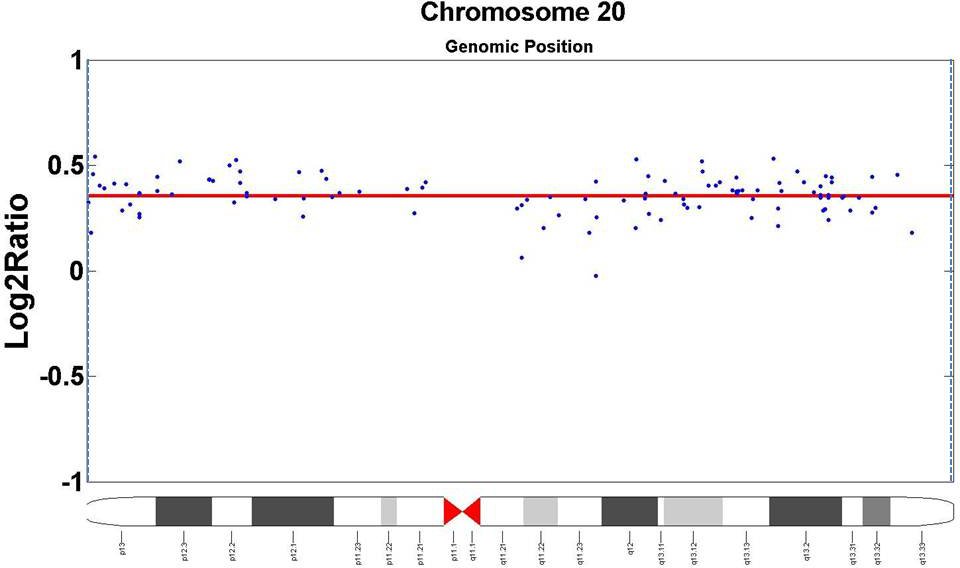} & \includegraphics[width=0.3\textwidth]{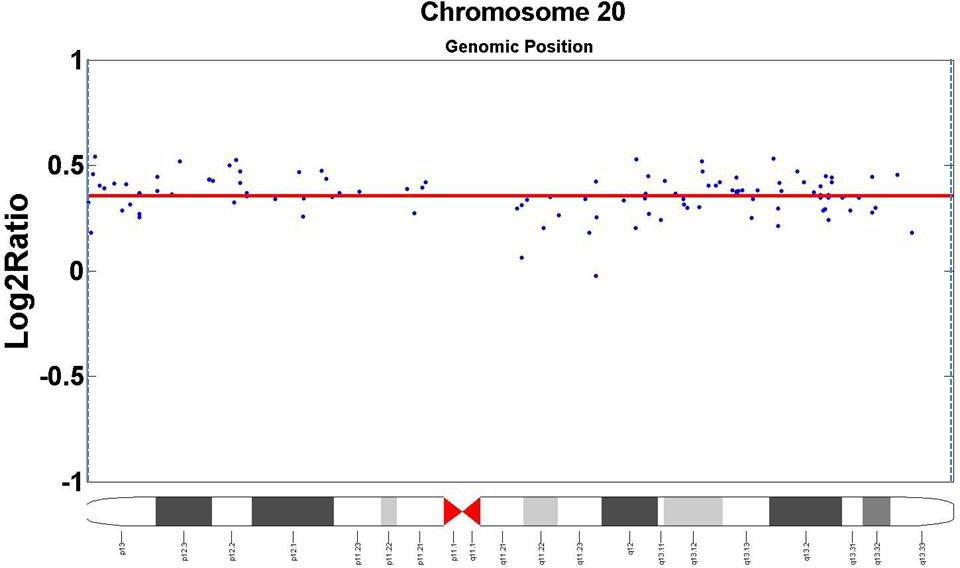} &    \includegraphics[width=0.33\textwidth]{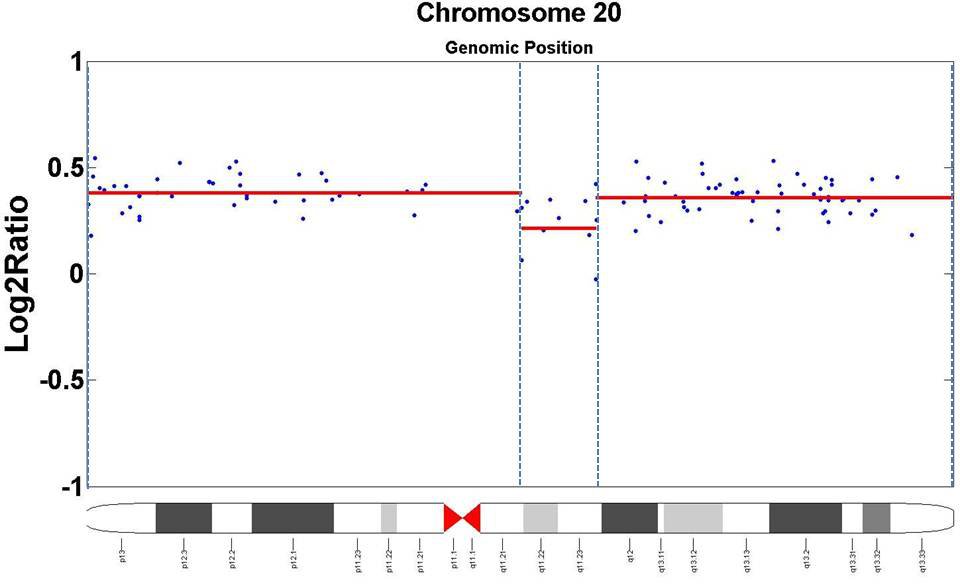} \\
		(g) &(h) &(i) \\
		&     & \\ 
		\trimmer\ &  \dnacopy\  & \cghseg\  \\  
		\end{tabular}
\caption{Visualization of the segmentation output of \trimmer, \dnacopy, and \cghseg for the cell line MCF10A on chromosomes 1 (a,b,c), 8 (d,e,f), and 20 (g,h,i).  (a,d,g) \trimmer\ output.  (b,e,h) \dnacopy\ output. (c,f,i) \cghseg\ output. Segments exceeding the $\pm$ 0.3 threshold \cite{thresholds} are highlighted.}
\label{fig:acghMCF10A}
\end{figure}

Figure~\ref{fig:acghT47D} compares the methods on chromosomes 1, 8, and 20
of  the cell line T47D.  On chromosome 1 (Figure~\ref{fig:acghT47D}(a,b,c)),
all three methods detect loss of the p-arm and a predominant
amplification of the q-arm.  \dnacopy\ infers a presumably spurious extension of the p-arm loss across
the centromere into the q-arm, while the other methods do not.  The main differences
between the three methods appear on the q-arm of chromosome 1.
\trimmer\ and \cghseg\ detect a small region of gain proximal to
the centromere at 1q21.1-1q21.2, followed by a short region of loss
spanning 1q21.3-1q22. \dnacopy\ merges these into a single longer
region of normal copy number.  The existence of a small region of loss
at this location in breast cancers is supported by prior
literature~\cite{chunder}.

The three methods provide comparable segmentations of chromosome 11
(Figure~\ref{fig:acghT47D}(d,e,f)).  All predict loss near the p-terminus,
a long segment of amplification stretching across much of the p- and
q-arms, and additional amplification near the q-terminus.  \trimmer,
however, breaks this q-terminal amplification into several
sub-segments at different levels of amplification while \dnacopy\ and
\cghseg\ both fit a single segment to that region.  We have no
empirical basis to determine which segmentation is correct here.
\trimmer\ does appear to provide a spurious break in the long amplified
segment that is not predicted by the others.

Finally, along chromosome 20 (Figure~\ref{fig:acghT47D}(g,h,i)), the
output of the methods is similar, with all three methods suggesting that
the q-arm has an aberrant copy number, an observation consistent with
prior studies \cite{breast20}.  The only exception is again that
\dnacopy\ fits one point more than the other two methods along the
first segment, causing a likely spurious extension of the p-arm's
normal copy number into the q-arm.

\paragraph{{\bf Cell Line MCF10A}}

Figure~\ref{fig:acghMCF10A} shows the output of each of the three methods on
chromosomes 1, 8, and 20 of the cell line MCF10A.  On this cell line, the
methods all yield similar predictions although from slightly different
segmentations.  All three show nearly identical behavior on chromosome
1 (Figure~\ref{fig:acghMCF10A}(a,b,c)), with normal copy number on the
p-arm and at least two regions of independent amplification of the
q-arm.  Specifically, the regions noted as gain regions host significant genes such as
PDE4DIP a gene associated with breast metastatic to bone (1q22),
ECM1 (1q21.2), ARNT (1q21), MLLT11 (1q21), S100A10 (1q21.3),
S100A13 (1q21.3), TPM3 (1q25) which also plays a role in breast cancer metastasis, 
SHC1 (1q21.3) and CKS1B (1q21.3).
\dnacopy\ provides a slightly different segmentation of the
q-arm near the centromere, suggesting that the non-amplified region
spans the centromere and that a region of lower amplification exists
near the centromere.  
On chromosome 8 (Figure~\ref{fig:acghMCF10A}(d,e,f))
the three algorithms lead to identical copy number predictions after
thresholding, although \dnacopy\ inserts an additional breakpoint at
8q21.3 and a short additional segment at 8q22.2 that do not correspond
to copy number changes.  
All three show significant amplification across chromosome 20 (Figure~\ref{fig:acghMCF10A}(g,h,i)), although in
this case \cghseg\ distinguishes an additional segment from 20q11.22-20q11.23 that is near the amplification threshold.
It is worth mentioning that chromosome 20 
hosts significant breast cancer related genes such as CYP24 and ZNF217.

\chapter{Robust Unmixing of Tumor States in Array Comparative Genomic Hybridization Data }
\label{unmixingchapter}
\lhead{\emph{Robust Unmixing of aCGH Data}} 
\section{Introduction}

In Section~\ref{subsec:unmixinghere} we discussed the importance 
of discovering tumor subtypes and we mentioned that
the recent work by \cite{SchwartzS10} showed promising results. 
They were, however, hampered  by limitations of the hard geometric approach, particularly the sensitivity to 
experimental error and outlier data points caused by the simplex fitting approach.  
An example of simplex fitting in the plane is shown in Figure \ref{fig:unmixingconrob},  
illustrating why the strict containment model used in \cite{EhrlichF87,ChanCHM09,SchwartzS10} 
is extremely sensitive to noise in the data.
In the present Chapter we introduce a soft geometric unmixing model 
for tumor mixture separation, which relaxes the requirement for strict containment using a 
fitting criterion that is robust to noisy measurements. We develop a formalization of the 
problem and derive an efficient gradient-based optimization method.  It is worth mentioning that this 
computational method has also been applied in the context of vertex similarity in social networks 
\cite{tsourakakis2011social}.

\begin{figure*}
\centering	
	\includegraphics[height=1.4in]{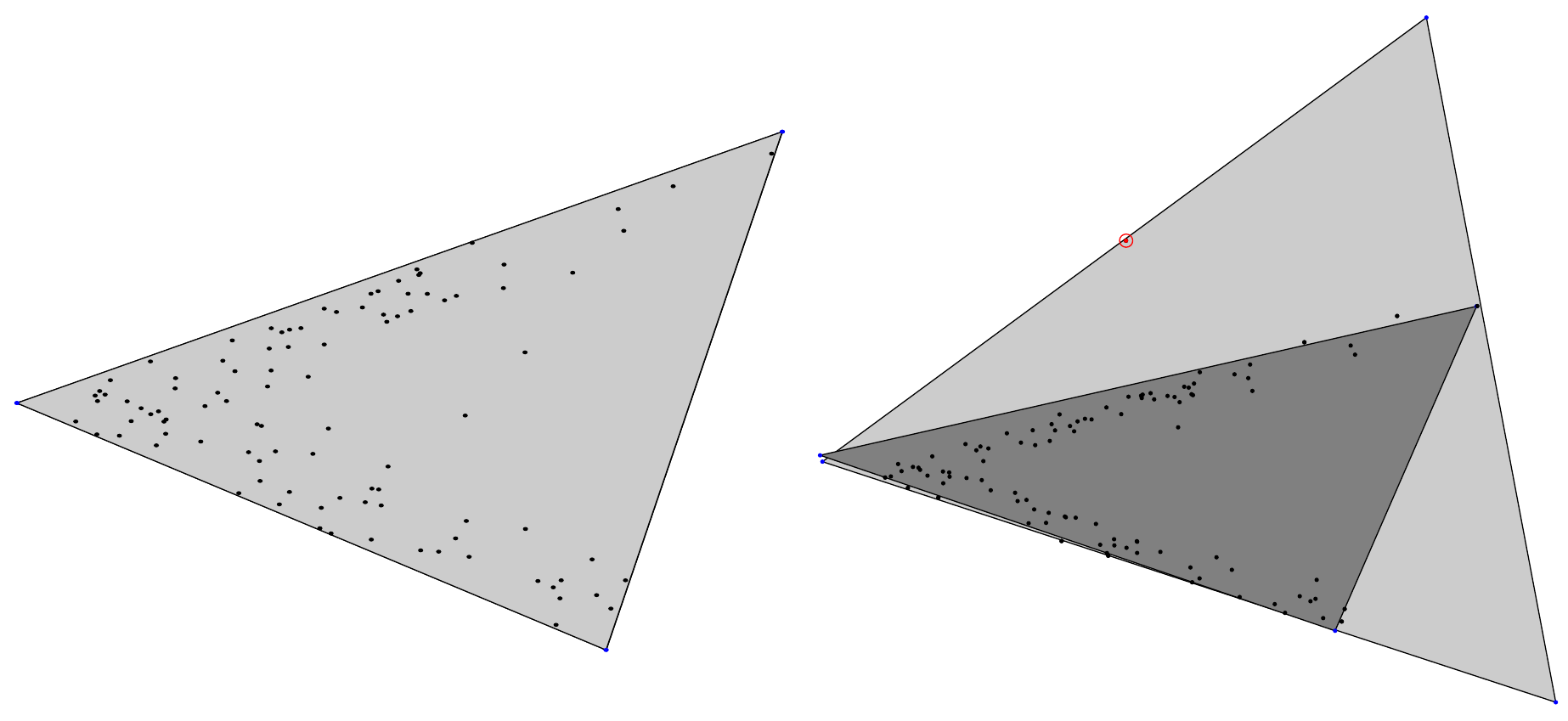}
	\caption{\label{fig:unmixingconrob} {\it Left:} The minimum area fit of a simplex containing 
         the sample points in the plane (shown in black) 
using the program in $\S$ \ref{sec:ISMBgeoun}. On noiseless data, hard geometric unmixing recovers the locations of the 
fundamental components at the vertices.  
{\it Right:} However, the  containment simplex is highly sensitive to noise and outliers in the data. 
A single outlier, circled above, radically changes the shape of the containment 
simplex fit (light gray above). In turn, this changes the estimates of basis distributions 
used to unmix the data. We mitigate this short coming by developing a soft geometric unmixing 
model (see $\S$ \ref{sec:ISMBrobun}) that is comparatively robust to noise. 
The soft fit (shown dark gray) is geometrically very close to the generating sources as seen on the left.}
\end{figure*}

The remainder of this Chapter is organized as follows: Section~\ref{sec:ISMBapproach} presents
our proposed method and Section~\ref{sec:ISMBmethods} we validate our method. 
In Section~\ref{sec:ISMBresults} we present a biological analysis of our findings on 
an aCGH data set taken from \cite{Navin09} and show that the method identifies state sets 
corresponding to known sub-types consistent with much of the analysis performed by the authors.

\section{Approach}
\label{sec:ISMBapproach}

The data are assumed to be given as $g$ genes sampled in $s$ tumors or tumor sections. 
The samples are collected in a matrix, $M \in \Re^{g \times s}$, 
in which each row corresponds to an estimate of gene copy number across the sample 
population obtained with aCGH.  The data in $M$ are processed as raw or baseline 
normalized raw input, rather than as $\log$ ratios. The ``unmixing'' model, described below, 
asserts that each sample $m_i$, a column of $M$, be well approximated by a convex combination 
of a fixed set of $C=[c_0|...|c_k]$ of $k+1$ unobserved basis distributions over the gene 
measurements. Further, the observed measurements are assumed to be perturbed by additive 
noise in the $\log$ domain, {\it i.e.}:  

\begin{equation} 
	\nonumber
	m_i = b^{\log_b \left(CF_i\right) + \eta}
\end{equation} 

where $F_i$ is the vector of coefficients for the convex combination of the $(k+1)$ basis 
distributions and $\eta$ is additive zero mode {\it i.i.d.} noise.

\subsection{Algorithms and Assumptions}
Given the data model above, the inference procedure seeks to recover 
the $k+1$ distributions over gene-copy number or expression that ``unmix'' the data. 
The procedure contains three primary stages: 
\begin{enumerate}
	\item Compute a reduced representation $x_i$ for each sample $m_i$,
	\item Estimate the basis distributions $K_{min}$ in the reduced coordinates and the mixture fractions $F$, 
	\item Map the reduced coordinates $K_{min}$ back into the ``gene space'' recovering $C$.
\end{enumerate} 
The second step in the method is performed by optimizing the objective in $\S$ \ref{sec:ISMBgeoun} or the robust problem formulation in $\S$ 
\ref{sec:ISMBrobun}.
\\
\\
\noindent 
{\it Obtaining the reduced representation} \\ 
We begin our calculations by projecting the data into a $k$  dimension vector space (i.e., the intrinsic dimensionality of a $(k+1)-$vertex simplex). We accomplish this using principal components analysis (PCA) \cite[]{Pearson01}, which decomposes the input matrix $M$ into a set of orthogonal basis vectors of maximum variance and retain only the $k$ components of highest variance.  PCA transforms the $g \times s$ measurement matrix $M$ into a linear combination $XV + A$, where $V$ is a matrix of the principal components of $M$, $X$ provides a representation of each input sample as a linear combination of the components of $V$, and $A$ is a $k \times s$ matrix in which each row contains $s$ copies of the mean value of the corresponding row of $M$. The matrix $X$ thus provides a reduced-dimension representation of $M$, and becomes the input to the sample mixture identification method in Stage 2.  $V$ and $A$ are retained to allow us to later contruct estimated aCGH vectors corresponding to the inferred mixture components in the original dimension $g$. 

Assuming the generative model of the data above, PCA typically recovers a sensible reduced representation, as low magnitude log additive noise induces ``shot-noise'' behavior in the subspace containing the simplex with small perturbations in the orthogonal complement subspace. An illustration of this stage of our algorithm can be found in Figure \ref{fig:pca_mot}.
\\
\\
\noindent 
{\it Sample mixture identification} \\
Stage 2 invokes either a hard geometric unmixing method that seeks the minimum volume simplex enclosing the input point set $X$ (Program \ref{prog:geomix1}) or a soft geometric unmixing method that fits a simplex to the points balancing the desire for a compact simplex with that for containment of the input point set (Program \ref{prog:robustmix1}). For this purpose, we place a prior over simplices, preferring those with small volume that fit or enclose the point set of $X$. This prior captures the intuition that the most plausible set of components explaining a given data set are those that can explain as much as possible of the observed data while leaving in the simplex as little empty volume, corresponding to mixtures that could be but are not observed, as possible.

Upon completion, Stage 2 obtains estimates of the vertex locations $K_{min}$, representing the inferred cell types from the aCGH data in reduced coordinates, and a set of mixture fractions describing the amount of each observed tumor sample attributed to each mixture component. The mixture fractions are encoded in a $(k+1) \times s$ matrix $F$, in which each column corresponds to the inferred mixture fractions of one observed tumor sample and each row corresponds to the amount of a single component attributed to all tumor samples. We define $F_{ij}$ to be the fraction of component $i$ assigned to tumor sample $j$ and $F_j$ to be vector of all mixture fractions assigned to a given tumor sample $j$. To ensure that that the observations are modeled as convex combinations of the basis vertices, we require that  $F {\bf 1} = 1$. 
\\
\\
\noindent 
{\it Cell type identification} \\
\noindent
The reduced coordinate components from Stage 2, $K_{min}$, are projected up to  a $g \times (k+1)$ matrix $C$ in which each column corresponds to one of the $k+1$ inferred components and each row corresponds to the approximate copy number of a single gene in a component. We perform this transformation using the matrices $V$ and $A$ produced by PCA in Stage 1 with the formula $C = V^TK_{min} + A$, augmenting the average to $k+1$ columns.


\begin{figure}
\centering
		\includegraphics[width=3.0in]{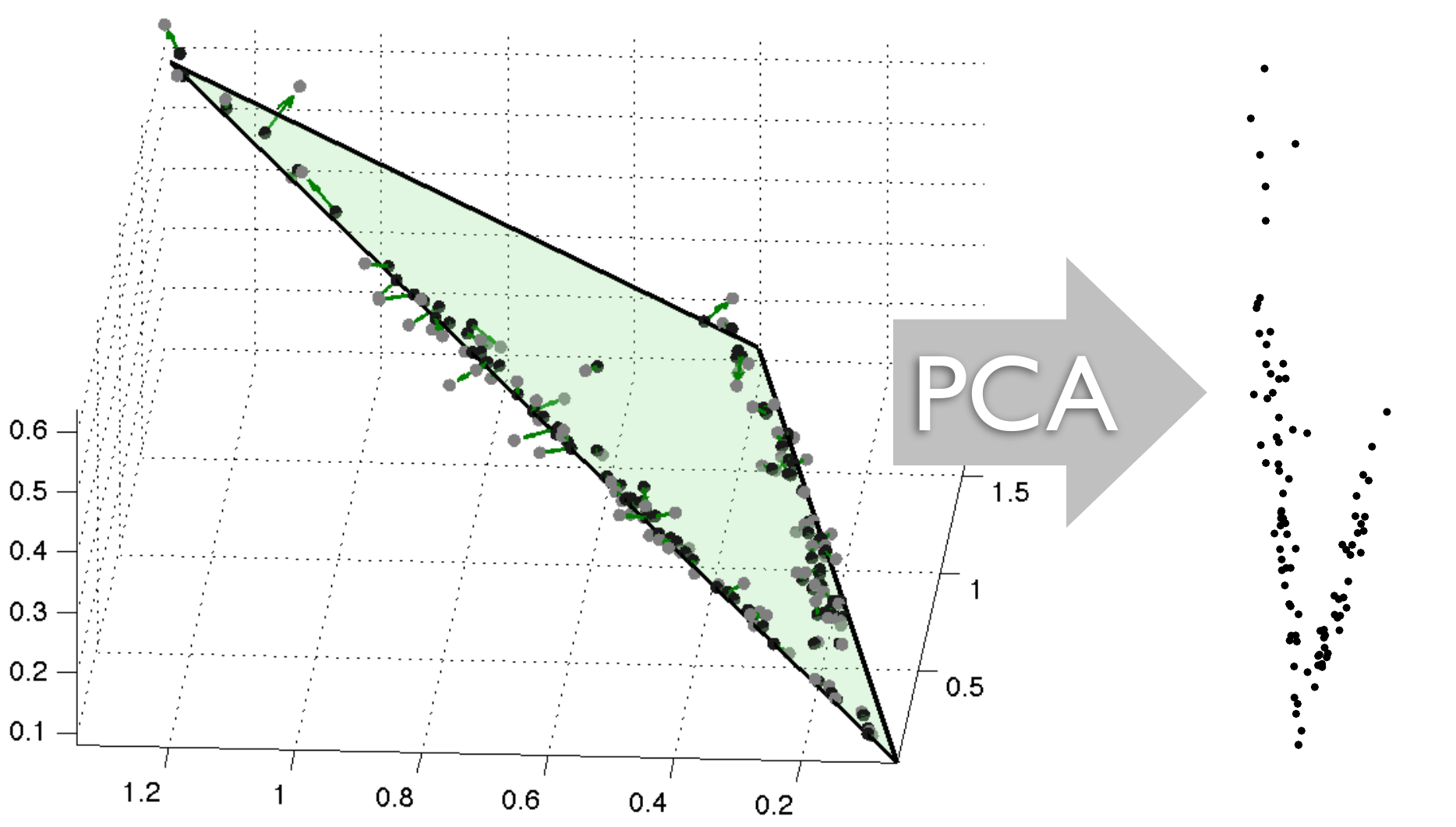}
	\caption{\label{fig:pca_mot} An illustration of the reduced coordinates under the unmixing hypothesis: 
points (show in gray) sampled from the $3-$simplex embedded are $\Re^3$ and then perturbed by log-normal noise, 
producing points shown in black with sample correspondence given the green arrows. 
Note that the dominant subspace remains in the planar variation induced by the simplex, 
and a $2$D reduced representation for simplex fitting is thus sufficient. }
\end{figure}

Finally the complete inference procedure is summarized in the following pseudocode:
\\

\noindent
{\it Given tumor sample matrix $M$, the desired number of mixture components $k$, and the strength of the volume prior $\gamma$:} 
\begin{enumerate} 
	\item Factor the sample matrix $M$ such that $M^T = X V + A$ 
	\item Produce the reduced $k-$dimensional representation by retaining the top $k$ components in $X$ 
	\item Minimize Program \ref{prog:geomix1}, obtaining an estimate of the simplex $K_{min}^{0}$
	\item Minimize Program \ref{prog:robustmix1} starting at $K_{min}^{0}$, obtaining $K_{min}$ and $F$
	\item Obtain the centers $C$ in gene space as $C = A+V^TK_{min}$ 
\end{enumerate}

\subsubsection{Hard Geometric Unmixing}
\label{sec:ISMBgeoun}
Hard geometric unmixing is equivalent to finding a minimum volume $(k+1)-$simplex containing a set of $s$ points $\{X \}$ in $\Re^k$. A non-linear program for hard geometric unmixing can be written as follows: 
\begin{eqnarray} 
	\label{prog:geomix1}
	\underset{K}{min}&:& \log {\tt vol}(K) \\ \nonumber
	\forall i &:& x_i = K F_i \\ \nonumber 
		\forall F_i &:& F_i^T{\bf 1} = 1,~F_i \succeq 0
\end{eqnarray}
where $\log {\tt vol}$ measures the volume of simplex defined by the vertices $K\doteq \left [v_0 | ... | v_k \right ]$ and $F \succeq 0$ requires that $\forall_{ij}. ~F_{ij}\ge 0$. Collectively, the constraints ensure that each point be expressed exactly as a unique convex combination of the vertices.  Exact nonnegative matrix factorization (NNMF), see \cite[]{LeeS99}, can be seen as a relaxation of hard geometric unmixing. Exact NNMF retains the top two constraints while omitting the constraint that the columns of $F$ sum to unity -- thus admitting all positive combinations rather than the restriction to convex combinations as is the case for geometric unmixing. 

Approximate and exponential-time exact minimizers are available for Program \ref{prog:geomix1}, in our experiments we use the approach of \cite{ChanCHM09}, which sacrifices some measure of accuracy for efficiency. 

\subsubsection{Soft Geometric Unmixing} 
\label{sec:ISMBrobun}
Estimates of the target distributions, derived from the fundamental components (simplex vertices), produced by hard geometric unmixing are sensitive to the wide-spectrum noise 
and outliers characteristic of log-additive noise ({\it i.e.}, multiplicative noise in the linear domain). 
The robust formulation below tolerates noise in the sample measurements $m_i$ and subsequently in the reduced 
representations $x_i$, improving the stability of these estimates. The sensitivity of hard geometric unmixing 
is illustrated in Figure 
\ref{fig:unmixingconrob}.  The motivation for soft geometric unmixing is to provide some tolerance to 
experimental error and outliers by relaxing the constraints in Program \ref{prog:geomix1} allowing points to 
lie outside the boundary of the simplex fit to the data.  We extend Program \ref{prog:geomix1} to provide a 
robust formulation as follows:  
\begin{eqnarray} 
	\label{prog:robustmix1} 
	\underset{K}{min}&:& \sum_{i=1}^s \left |x_i - KF_i \right |_p + \gamma \log {\tt vol}(K)\\ \nonumber 
	\forall F_i &:& F_i^T{\bf 1} = 1,~F_i \succeq 0
\end{eqnarray}	
where the term $\left |x_i - KF_i \right |_p$ penalizes the imprecise fit of the simplex to the data and  $\gamma$ 
establishes the strength of the minimum-volume prior. Optimization of Program \ref{prog:robustmix1} is seeded 
with an estimate produced from Program \ref{prog:geomix1} and refined using MATLAB's {\it fminsearch} with analytical 
derivatives for the $\log {\tt vol}$ term and an $LP$-step that determines mixtures components $F_i$  and the distance 
to the boundary for each point outside the simplex.

We observe that when taken as whole, Program~\ref{prog:robustmix1} 
can be interpreted as the negative log likelihood of a Bayesian model of signal formation. In the case of array CGH data, we choose $p=1$ ({\em i.e.}, optimizing relative to an $\ell_1$ norm), as we observe that the errors may be induced by outliers and the $\ell_1$ norm would provide a relatively modest penalty for a few points far from the simplex. From the Bayesian perspective, this is equivalent to relaxing the noise model to assume {\it i.i.d.} heavy-tailed additive noise. To mitigate some of the more pernicious effects of log-normal noise, we also apply a total variation-like smoother to aCGH data in our experiments. Additionally, the method can be readily extended to weighted norms if an explicit outlier model is available. 

\subsubsection{Analysis \& Efficiency}
\label{sec:ISMBeffcomp}
The hard geometric unmixing problem in $\S$ \ref{sec:ISMBgeoun} is a non-convex objective in the present parameterization, and was shown by \cite[]{Packer02} to be NP-hard when $k+1 \ge \log(s)$. For the special case of minimum volume tetrahedra ($k=3$), \cite[]{ZhouS00} demonstrated an exact algorithm with time complexity $\Theta(s^4)$ and a $(1+\epsilon)$ approximate method with complexity $O(s + 1 / \epsilon^6)$. 
Below, we examine the present definition and show that Programs \ref{prog:geomix1} and \ref{prog:robustmix1} have structural properties that may exploited to construct efficient gradient based methods that seek local minima. Such gradient methods can be applied in lieu of or after heuristic or approximate 
combinatorial methods for minimizing Program \ref{prog:geomix1}, such as \cite[]{EhrlichF87,ChanCHM09} or the $(1+\epsilon)$ method of \cite[]{ZhouS00} for simplexes in $\Re^3$. 

We begin by studying the volume penalization term as it appears in both procedures. The volume of a convex body is well known (see \cite[]{Boyd04}) to be a log concave function. In the case of a simplex, analytic partial derivatives with respect to vertex position can used to speed the estimation of the minimum volume configuration $K_{min}$. The volume of a simplex, represented by the vertex matrix $K=\left [v_0|...|v_k \right ]$, can be calculated as: 
\begin{equation} 
	{\tt vol} (K) = c_k \cdot {\tt det}  \left( \Gamma^T K K^T \Gamma  \right)^{1/2} = c_k \cdot {\tt det}~Q
\end{equation} 
where $c_k$ is the volume of the unit simplex defined on $k+1$ points and $\Gamma$ is a fixed vertex-edge incidence matrix such that $\Gamma^T K = \left [v_1-v_0|...|v_k-v_0 \right ]$. The matrix $Q$ is an inner product matrix over the vectors from the special vertex $v_0$ to each of the remaining $k$ vertices. In the case where the simplex $K$ is non-degenerate, these vectors form a linearly independent set and $Q$ is positive definite (PD). While the determinant is log concave over PD matrices, our parameterization is linear over the matrices $K$, not $Q$.  Thus it is possible to generate a degenerate simplex when interpolating between two non-degenerate simplexes $K$ and $K'$. 
For example, let $K$ define a triangle with two vertices on the $y-$axis and produce a new simplex $K'$ by reflecting the triangle $K$ across the $y-$axis. The curve $K(\alpha) = \alpha K+ (1-\alpha)K'$ linearly interpolates between the two.  Clearly, when $\alpha=1/2$, all three vertices of $K(\alpha)$ are co-linear and thus the matrix $Q$ is not full rank and the determinant vanishes. However, in the case of small perturbations, we can expect the simplexes to remain non-degenerate. 

To derive the partial derivative, we begin by substituting the determinant formulation into our volume penalization and arrive at the following calculation: 
\begin{eqnarray}
\nonumber
\log {\tt vol}(K) &=&\log c_k + \frac{1}{2} \log {\tt det} Q \\ \nonumber 
& \propto & \log \prod_{d=1}^k \lambda_d(Q) = \sum_{d=1}^k \lambda_d(Q)
\end{eqnarray}
therefore the gradient of $\log {\tt vol}(K)$ is given by 
\begin{eqnarray} 
	\nonumber
\frac{\partial \log {\tt vol(K)}}{\partial K_{ij}} &=& \sum_{d=1}^k \frac{\partial }{\partial K_{ij}} \lambda_d = \sum_{d=1}^k z_d^T(\Gamma^T E_{ij} E_{ij}^T \Gamma) z_d 
\end{eqnarray}
where the eigenvector $z_d$ satisfies the equality $Q z_d = \lambda_d z_d$ and $E_{ij}$ is the indicator matrix for the entry $ij$. To minimize the volume, we move the vertices along the paths specified by the negative $\log$ gradient of the current simplex volume. The Hessian is derived by an analogous computation, making Newton's method for Program \ref{prog:geomix1}, with log barriers over the equality and inequality constraints, a possible optimization strategy. 

Soft geometric unmixing (Program~\ref{prog:robustmix1} ) trades the equality constraints in Program~\ref{prog:geomix1} for a convex, but  non-differentiable term, in the objective function $\sum_{i=1}^s\left | x_i - K F_i \right|_p$ for $p|1 \le p \le 2$. Intuitively -- points inside the simplex have no impact on the cost of the fit. However, over the course of the optimization, as the shape of the simplex changes points move from the interior to the exterior, at which time they incur a cost. To determine this cost, we solve the nonnegative least squares problem for each mixture fraction $F_i$, $min_F~:~(KF_i - x_i)^T(KF_i-x_i)$. This step simultaneously solves for the mixture fraction, and for exterior points, the distance to the simplex is determined. The simplex is then shifted under a standard shrinkage method based on these distances.

\section{Experimental Methods}
\label{sec:ISMBmethods}

We evaluated our methods using synthetic experiments, allowing us to assess two properties of robust unmixing 
1) the fidelity with which endmembers (sub-types) are identified and 2) the relative effect of noise on 
hard versus robust unmixing.  We then evaluate the robust method on a real world aCGH data set published 
by \cite{Navin09} in which ground truth is not available, but for which we uncover much the structure reported by the authors. 

\subsection{Methods: Synthetic Experiments}
To test the algorithms given in $\S$\ref{sec:ISMBapproach} we simulated data using a 
biologically plausible model of ad-mixtures. Simulated data provides a quantitative 
means of evaluation as ground truth is available for both the components $C$ and the 
mixture fractions $F_i$ associated with each measurement in the synthetic design matrix 
$M$. The tests evaluate and compare hard geometric unmixing $\S$\ref{sec:ISMBgeoun} and soft 
geometric unmixing $\S$\ref{sec:ISMBrobun} in the presence of varying levels of log-additive 
Gaussian noise and varying $k$. By applying additive Gaussian noise in the log domain 
we simulate the heteroscedasticity characteristic of CGH measurements ({\it i.e.} 
higher variance with larger magnitude measurements). By varying $k$, the dimensionality 
of the simplex used to fit the data, we assess the algorithmic sensitivity to this 
parameter as well as that to $\gamma$ governing the strength of the volume prior in 
Program 2. The sample generation process consists of three major steps: 1) mixture 
fraction generation (determining the ratio of sub-types present in a sample), 2) 
end-member ({\it i.e.} sub-type) generation and 3) the sample perturbation by additive noise in the log-ratio domain. 

\subsubsection{Mixture Sampler}
\label{sec:ISMBmix_sampler}
Samples over mixture fractions were generated in a manner analogous to Polya's Urn Process, 
in which previously sampled simplicial components ({\it e.g.}, line segments, triangles, 
tetrahedra) are more likely to be sampled again. This sampling mechanism produces data 
distributions that are similar to those we see in low dimensional projections of aCGH data 
when compared against purely uniform samples over mixtures. An example of a low dimensional 
sample set and the simplex that was used to generate the points is shown in Figure \ref{fig:sample_ex}.

To generate the mixture fractions $F_i$ for the $i^{th}$ sample, the individual 
components in $C^{true}$ are sampled without replacement from a dynamic tree model. 
Each node in the tree contains a dynamic distribution over the remaining components, 
each of which is initialized to the uniform distribution. We then sample $s$ mixtures 
by choosing an initial component according to the root's component distribution and proceed 
down the tree. As a tree-node is reached, its component distribution is updated to reflect 
the frequency with which its children are drawn. To generate the $i^{th}$ sample, 
the fractional values $F_i$ are initialized to zero. As sample generation proceeds, 
the currently selected component $C_j$ updates the mixture as $F{ij} \sim {\tt uniform}[(1/2) f^j_p,1]$ 
where $f_p^j$ is the frequency of $j$'s parent node. For the $i^{th}$ mixture, this process terminates 
when the condition $1 \le \sum_{j=1}^{k+1}F_{ij}$ holds. Therefore, samples generated by long paths 
in the tree will tend to be homogenous combinations of the components $C^{true}$, where as 
short paths will produce lower dimensional substructures. At the end of the process, 
the matrix of fractions $F$ is re-normalized so that the mixtures associated with each 
sample sum to unity.  This defines a mixture $F_{i}^{true}$ for each sample -- {\it i.e.} 
the convex combination over fundamental components generating the sample point. 

\subsubsection{Geometric Sampling of End-members \& Noise}
\label{sec:ISMBgsamp}

\begin{figure}
\centering
	\includegraphics[width=2in]{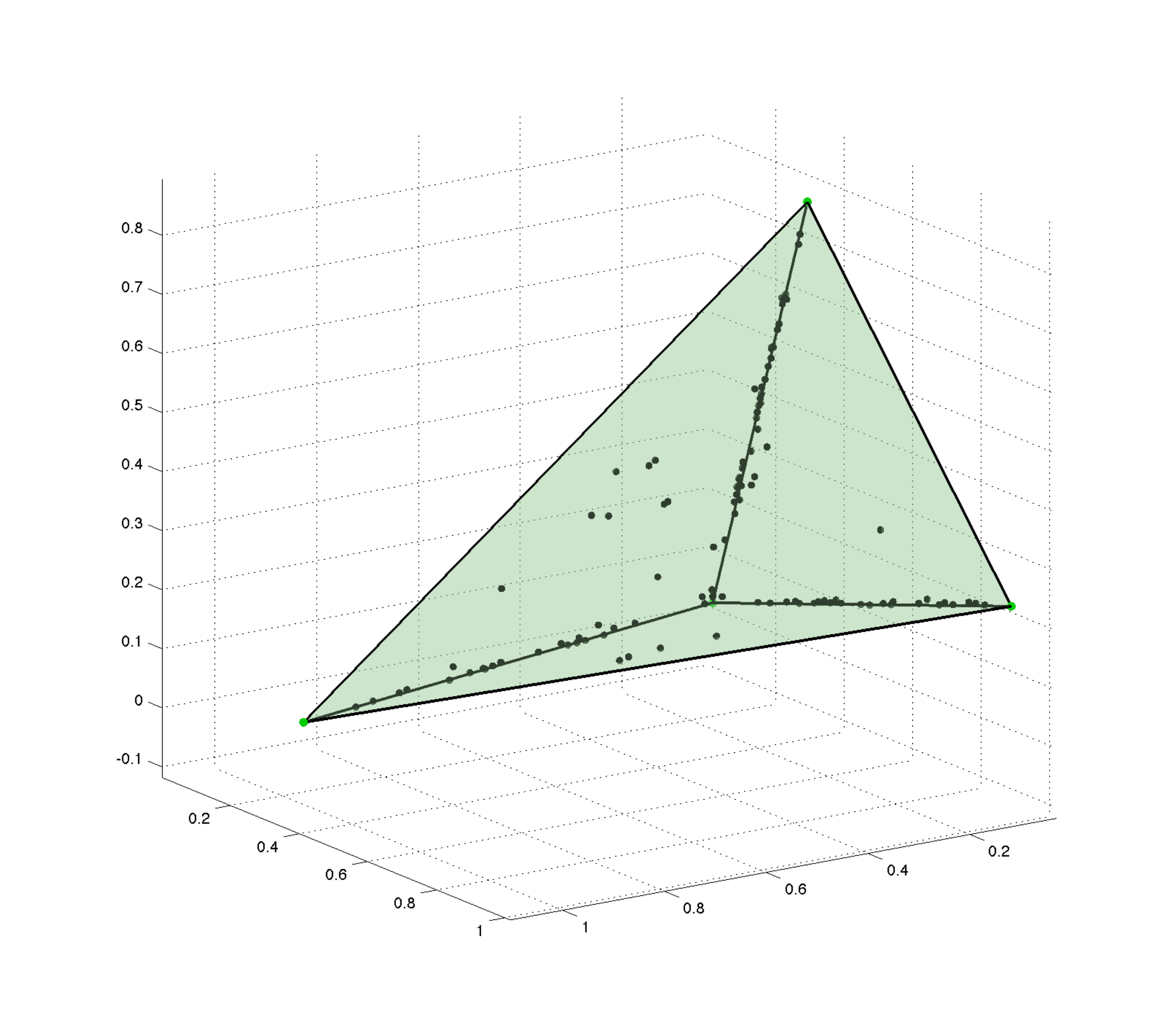}
	\caption{\label{fig:sample_ex} An example sample set generated for $\S$\ref{sec:ISMBgsamp} shown in 
         the ``intrinsic dimensions'' of the model. Note that sample points cleave to the lower 
         dimensional substructure (edges) of the simplex.}
\end{figure}

To determine the location of the end-members we specify an extrinsic dimension 
(number of genes) $g$, and an intrinsic dimension $k$ (requiring $k+1$ components). 
We then simulate $k+1$ components by constructing a $g \times (k+1)$ matrix $C^{true}$ 
of fundamental components in which each column is an end-member ({\it i.e.} sub-type) 
and each row is the copy number of one hypothetical gene, sampled from the unit Gaussian 
distribution and rounded to the nearest integer. Samples $m_i$, corresponding to the columns 
of the data matrix $M$, are then given by: 

\begin{equation}
	m_{i} = 2^{log_2\left(C^{true}F_i^{true} \right ) + \frac{1}{2}\sigma \eta}
\end{equation}
where $\eta \sim {\tt normal}(0,1)$ and the mixture fractions $F_i^{true}$ were obtained as in $\S$\ref{sec:ISMBmix_sampler}.

\begin{figure*}
\centering
	\includegraphics[height=2in, width=3in]{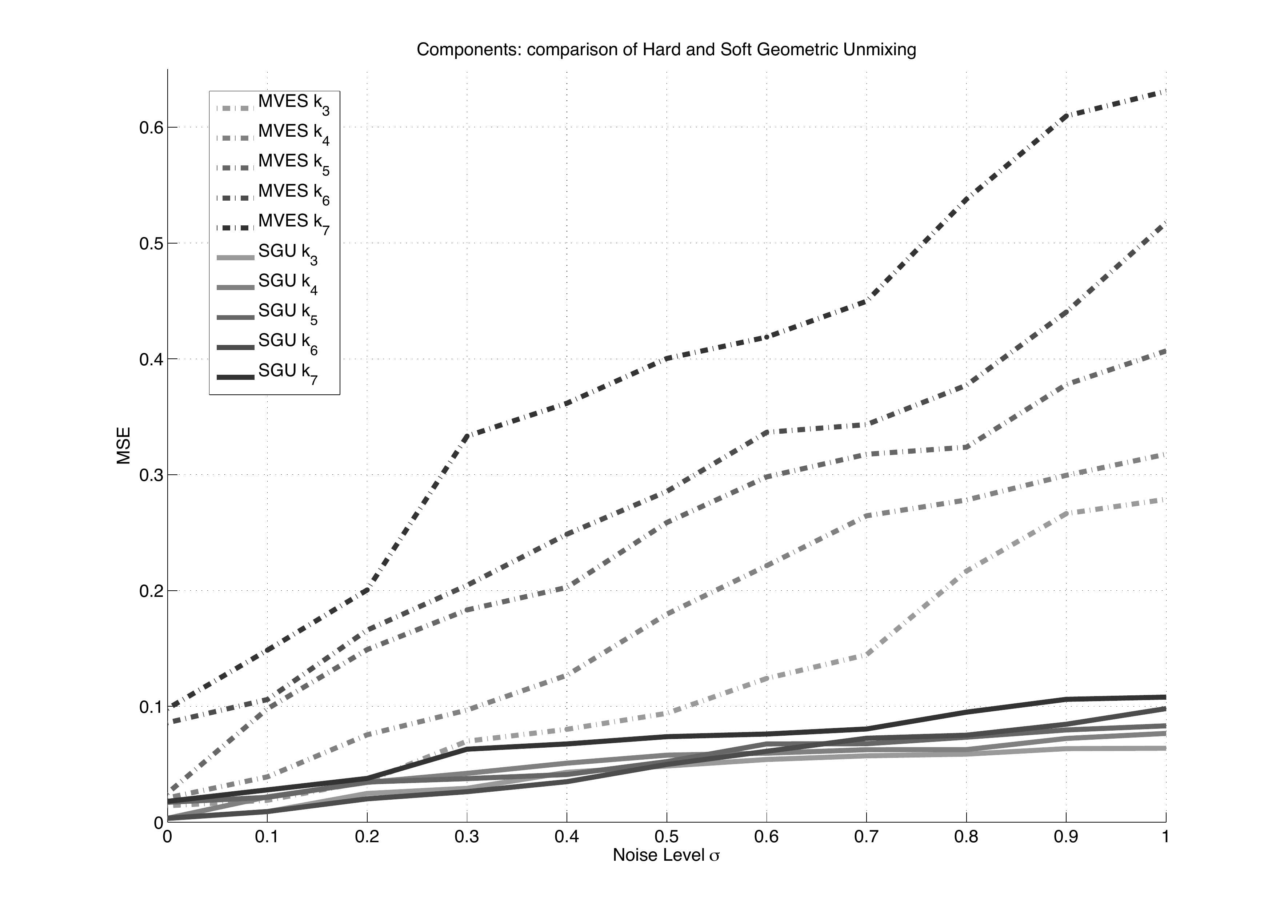}
	\includegraphics[height=2.05in,width=3in]{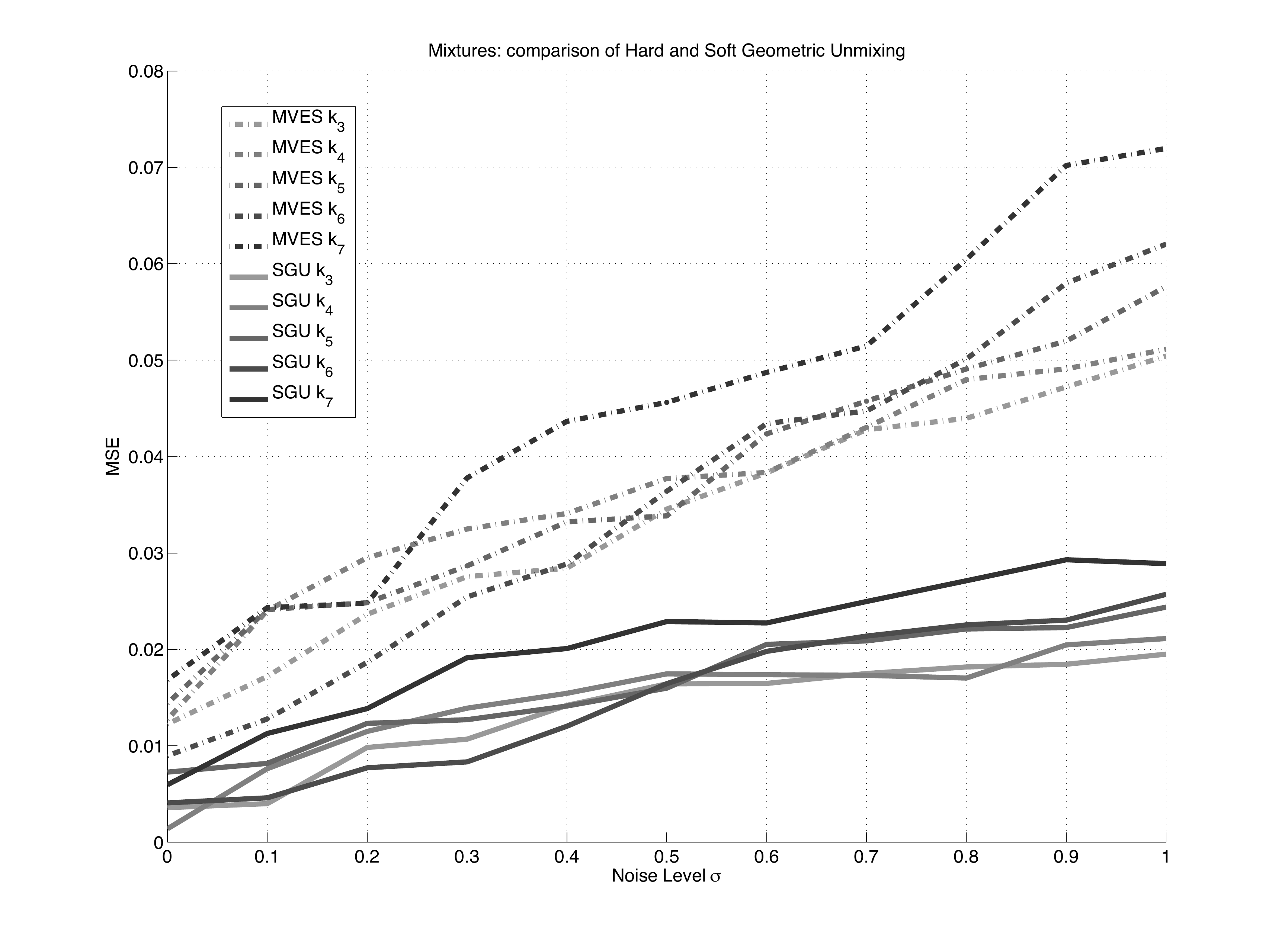}
\caption{\label{fig:comp_mse} Left: mean squared error for the component reconstruction comparing 
Hard Geometric Unmixing (MVES: \cite{ChanCHM09}) and Soft Geometric Unmixing (SGU) introduced 
in $\S$\ref{sec:ISMBrobun} for the experiment described in $\S$\ref{sec:ISMBgsamp} 
with variable $\gamma$. The plot demonstrates that robust unmixing more accurately 
reconstructs the ground truth centers relative to hard unmixing in the presence of 
noise. Right: mean squared error for mixture reconstruction comparing MVES and SGU.}
\end{figure*}


\subsubsection{Evaluation} 
We follow Schwartz and Shackney~\cite{SchwartzS10} in assessing the quality of 
the unmixing methods by independently measuring the accuracy of inferring the components and the mixture fractions.  
We first match inferred mixture components to true mixture components by performing a maximum weighted bipartite 
matching of columns between $C^{true}$ and the inferred components $C^e$, weighted by negative Euclidean distance. 
We will now assume that the estimates have been permuted according to this matching and continue. We then assess 
the quality of the mixture component identification by the root mean square distance over all entries of all 
components between the matched columns of the two C matrices:
\begin{equation}
	{\tt error} = \frac{1}{g(k+1)}\left | | C^{true}-C^{e} | \right |_F^2
\end{equation}
where $|| A ||_F=\sqrt{\sum_{ij}a^2_{ij}}$ denotes the Frobenius norm of the matrix $A$. 

We similarly assess the quality of the mixture fractions by the root mean square distance 
between $F^{true}$ and the inferred fractions $F^e$ over all genes and samples:
\begin{equation}
	{\tt error} = \frac{1}{g(k+1)} \left | | F^{true}-F^e |\right|^2_F. 
\end{equation}
This process was performed for $s = 100$ and $d = 10000$ to approximate a realistic tumor 
expression data set and evaluated for $k = 3$ to $k = 7$ and for $\sigma = \{0, 0.1, 0.2, . . . , 1.0\}$, with 
ten repetitions per parameter.

\section{Results}
\label{sec:ISMBresults}
\subsection{Results: Synthetic Data}
\label{sec:ISMBsynd}

The results for the synthetic experiment are summarized in Figure \ref{fig:comp_mse}. 
The figure shows the trends in MSE for hard geometric unmixing $\S$\ref{sec:ISMBgeoun} 
and soft geometric unmixing $\S$\ref{sec:ISMBrobun} on the synthetic data described above.  
As hard geometric unmixing requires that each sample lie inside the fit simplex, 
as noise levels increase (larger $\sigma$), the fit becomes increasingly inaccurate. 
Further, the method MVES deteriorates to some degree as order $k$ of the simplex increases. 
However, soft geometric unmixing degrades more gracefully in the presence of noise 
if an estimate of the noise level is available with $\pm 0.1$ in our current model. 
The trend of soft unmixing exhibiting lower error and better scaling in $k$ than hard 
unmixing holds for both components and mixture fractions, although components exhibit 
a higher average degree of variability due to the scale of the synthetic measurements 
when compared to the mixture fractions.

\subsection{Array Comparative Gene Hybridization (aCGH) Data }
\label{sec:ISMBreald}
We further illustrate the performance of our methods on a publicly available primary 
Ductal Breast Cancer aCGH Dataset furnished with ~\cite{Navin09}. This dataset 
is of interest in that each tumor sample has been sectored multiple times during 
biopsy which is ideal for understanding the substructure of the tumor population.  
The data consists of 87 aCGH profiles from 14 tumors run on a high-density 
ROMA platform with $83055$ probes.  Profiles are derived from 4-6 sectors per 
tumor, with samples for tumors 5-14 sub-partitioned by cell sorting according 
to total DNA content, and with healthy control samples for tumors 6, 9, 12, and 13.  
For full details, the reader is referred to Navin et al.~\cite{Navin09}.  The processed 
data consists of $\log_{10}$ ratios and which were exponentiated prior to the PCA step 
(Stage 1) of the method.

\begin{figure}[t]
\centering 
		\includegraphics[width=3in]{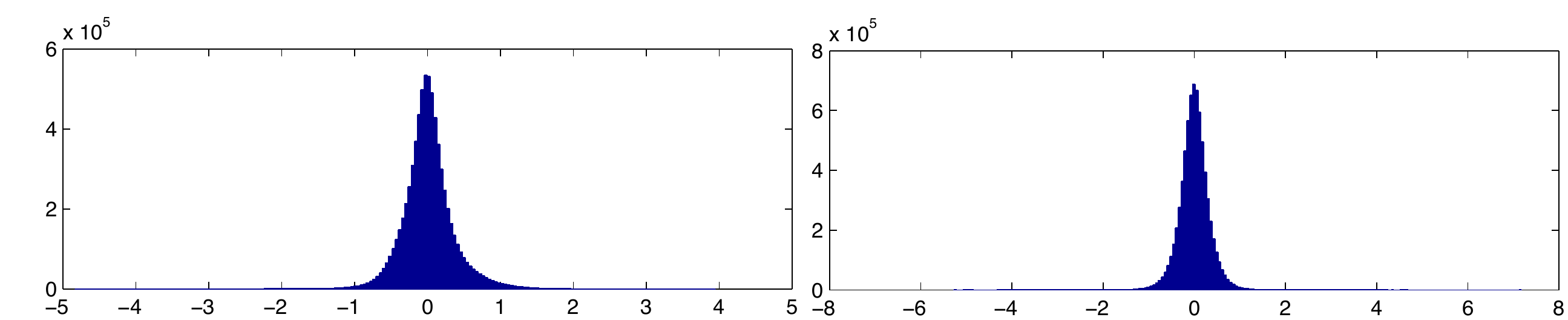}
	\caption{\label{fig:hist_cgh} Empirical motivation for the $\ell_1-\ell_1-$total variation 
       functional for smoothing CGH data. The left plot shows the histogram of values 
         found in the CGH data obtained from the \cite{Navin09} data set. The distribution is well fit 
       by the high kurtosis Laplacian distribution in lieu of a Gaussian. The right plot shows the distribution of 
       differences along the probe array values. As with the values distribution, these frequencies exhibit high kurtosis.}
\end{figure}

\begin{figure}[t]
	\begin{center}
		\includegraphics[width=1.5in]{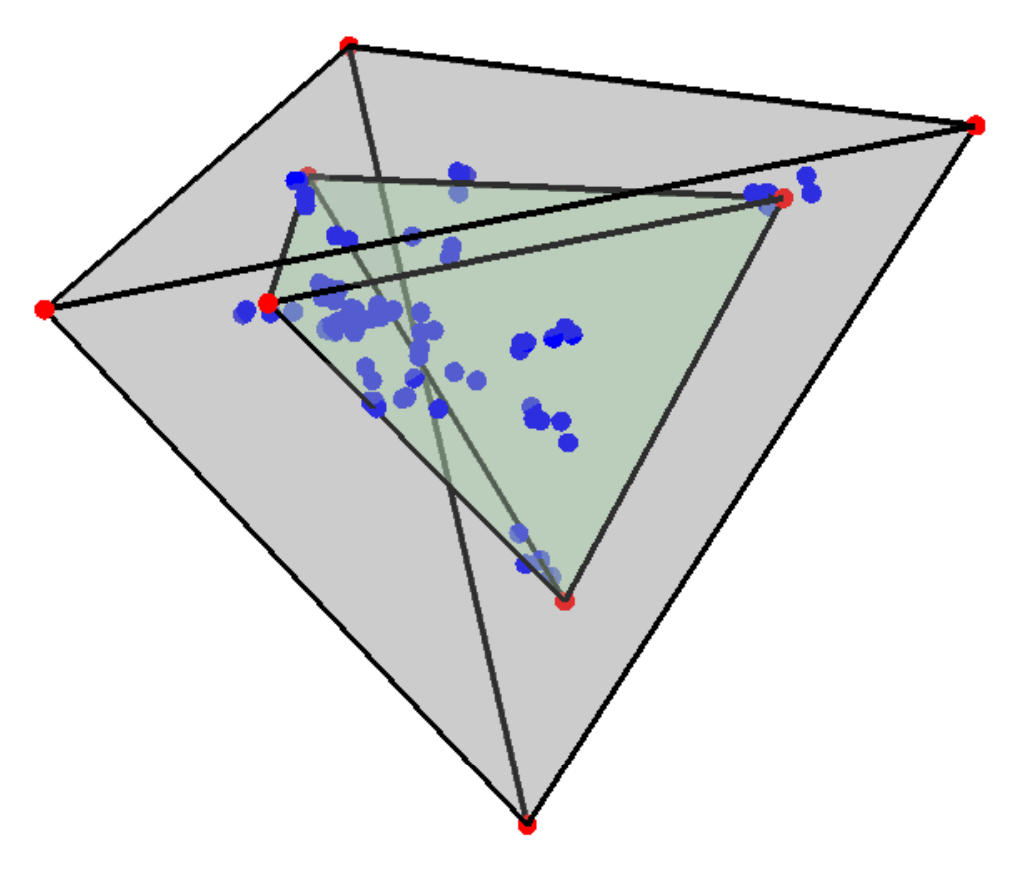}
	\end{center}
	\caption{\label{fig:S4cvr} The simplex fit to the CGH data samples from \cite{Navin09} ductal data 
        set in $\Re^3$. The gray tetrahedron was return by the optimization of Program \ref{prog:geomix1} and 
        the green tetrahedron was returned by the robust unmixing routine.}
\end{figure}

\subsubsection{Preprocessing}
To mitigate the effects of sensor noise on the geometric inference problem we apply a total variation (TV) functional to the raw log-domain data. The $\ell_1-\ell_1-$TV  minimization  is equivalent to a penalized projection onto the over-complete Harr basis preserving a larger degree of the signal variation when compared to discretization methods ({\it e.g} \cite{OlshenVLW06,GuhaLN06}) that employ aggressive priors over the data distribution. The procedure seeks a smooth instance $x$ of the observed signal $s$ by optimizing the following functional: 
\begin{equation} 
	\label{fun:tv}
	\underset{x}{min}~:~ \sum_{i=1}^g \left | x_i - s_i \right |_1 + \lambda \sum_{i=1}^{g-1} \left | x_i - x_{i+1} \right |_1
\end{equation}
The functional \ref{fun:tv} is convex and can be solved readily using Newton's method with log-barrier functions (\cite{Boyd04}). 
The solution $x$ can be taken as the maximum likelihood estimate of a Bayesian model of CGH data formation. 
That is, the above is the negative log-likelihood of a simple Bayesian model of signal formation. 
The measurements $\hat{x}_i$ are assumed to be perturbed by the {\it i.i.d.} Laplacian noise and the 
changes along the probe array are assumed to be sparse. Recall that the Laplacian distribution is 
defined by the probability density function  $\Prob{x} = \frac{1}{z} \exp \Big( \frac{-\left | x \right | }{a} \Big)$. In all experiments the 
strength of the prior $\lambda$ was set to $\lambda=10$. The data fit this model well as illustrated in Figure \ref{fig:hist_cgh}. The dimension of the reduced representation $k$, fixing the number of fundamental components, was determined using the eigengap heuristic during the PCA computation (Stage 1). This rule ceases computing additional principal components when the difference in variances jumps above threshold.

\subsubsection{Unmixing Analysis and Validation}
The raw data was preprocessed as described above and a simplex was fit to the reduced coordinate representation using the soft geometric unmixing method (see  
$\S$\ref{sec:ISMBrobun}). A three dimensional visualization of the resulting fit is shown for the \cite{Navin09} data set in Figure \ref{fig:S4cvr}. To assess the performance with increasing dimensionality, we ran experiments for polytope dimensionality $k$ ranging from $3$ to $9$.Following the eigen-gap heuristic we chose to analyze the results for $k = 6$. The $\gamma$ value was picked according to the estimated noise level in the aCGH dataset and scaled relative to the unit simplex volume (here, $\gamma = 100$). The estimated $6$ components/simplex vertices/pure cancer types are labeled $C_1,C_2,..., C_6$.

Figure \ref{fig:frac} shows mixture fraction assignments for the aCGH data for $k=6$.  While there is typically a non-zero amount of each component in each sample due to imprecision in assignments, the results nonetheless show distinct subsets of tumors favoring different mixture compositions and with tumor cells clearly differentiated from healthy control samples.  The relative consistency within versus between tumors provides a secondary validation that soft unmixing is effective at robustly assigning mixture fractions to tumor samples despite noise inherent to the assay and that produced by subsampling cell populations. It is also consistent with observations of Navin {\em et al.} 

It is not possible to know with certainty the true cell components or mixture fractions of the real data, but we can validate the biological plausibility of our results by examining known sites of amplification in the inferred components.
We selected fourteen benchmark loci frequently amplified in breast cancers through manual literature search.  Table~\ref{Tab:01} lists the
chosen benchmarks and the components exhibiting at least 2-fold amplification of
each.  Figure~\ref{fig:infcomps} visualizes the results, plotting relative amplification of each component as a function of genomic coordinate and highlighting the locations of the benchmark markers.  Thirteen of the fourteen benchmark loci exhibit amplification for a subset of the components, although often at minimal levels.
The components also show amplification of many other sites not in our benchmark set, but we cannot definitively determine which are true sites of amplification and which are false positives.  We further tested for amplification of seven loci reported as amplified by Navin et al.~\cite[]{Navin09} specifically in the tumors examined here and found that six of the seven are specifically amplified in one of our inferred components: PPP1R12A ($C_2$), KRAS ($C_2$), CDC6 ($C_2$), RARA ($C_2$), EFNA5 ($C_2$), PTPN1 ($C_3$), and LPXN (not detected).  Our method did not infer a component corresponding to normal diploid cells as one might expect due to stromal contamination.  This failure may reflect a bias introduced by the dataset, in which many samples were cell sorted to specifically select aneuploid cell fractions, or could reflect an inherent bias of the method towards more distinct components, which would tend to favor components with large amplifications.

We repeated these analyses for the hard unmixing with a higher amplification threshold due to the noise levels in the centers.  It detected amplification at 11 of the 14 loci, with spurious inferences of deletion at four of the 11.  For the seven sites reported in Navin {\em et al.}, hard unmixing identified five (failing to identify EFNA5 or LPXN) and again made spurious inferences of deletions for three of these sites, an artifact the soft unmixing eliminates. 
The results suggest that hard unmixing produces less precise fits of simplexes to the true data.

We can also provide a secondary analysis based on Navin {\em et al.}'s central result that the tumors can be partitioned into monogenomic (those appearing to show essentially a single genotype) and polygenomic (those that appear to contain multiple tumor subpopulations). We test for monogeniety in mixture fractions by finding the minimum correlation coefficient between mixture fractions of consecutive tumor sectors (ignoring normal controls) maximized over all permutations of the sectors.  Those tumors with correlations above the mean over all tumors (0.69) were considered monogenomic and the remainder polygenomic. Navin {\em et al.} assign $\{1,2,6,7,9,11\}$ as monogenomic and $\{3,4,5,8,10,12,13,14\}$ and polygenomic. Our tests classify $\{1,2,5,6,7,8,11\}$ as monogenomic and $\{3,4,10,12,13,14\}$ as polygenomic, disagreeing only in tumors 5 and 8.  Our methods are thus effective at identifying true intratumor heterogeneity in almost all cases without introducing spurious heterogeneity.  By contrast, hard unmixing identifies only tumors 5 and 8 as polygenomic, generally 
obscuring true heterogeneity in the tumors.

\begin{figure}
\begin{center}
\includegraphics[width=0.6\textwidth]{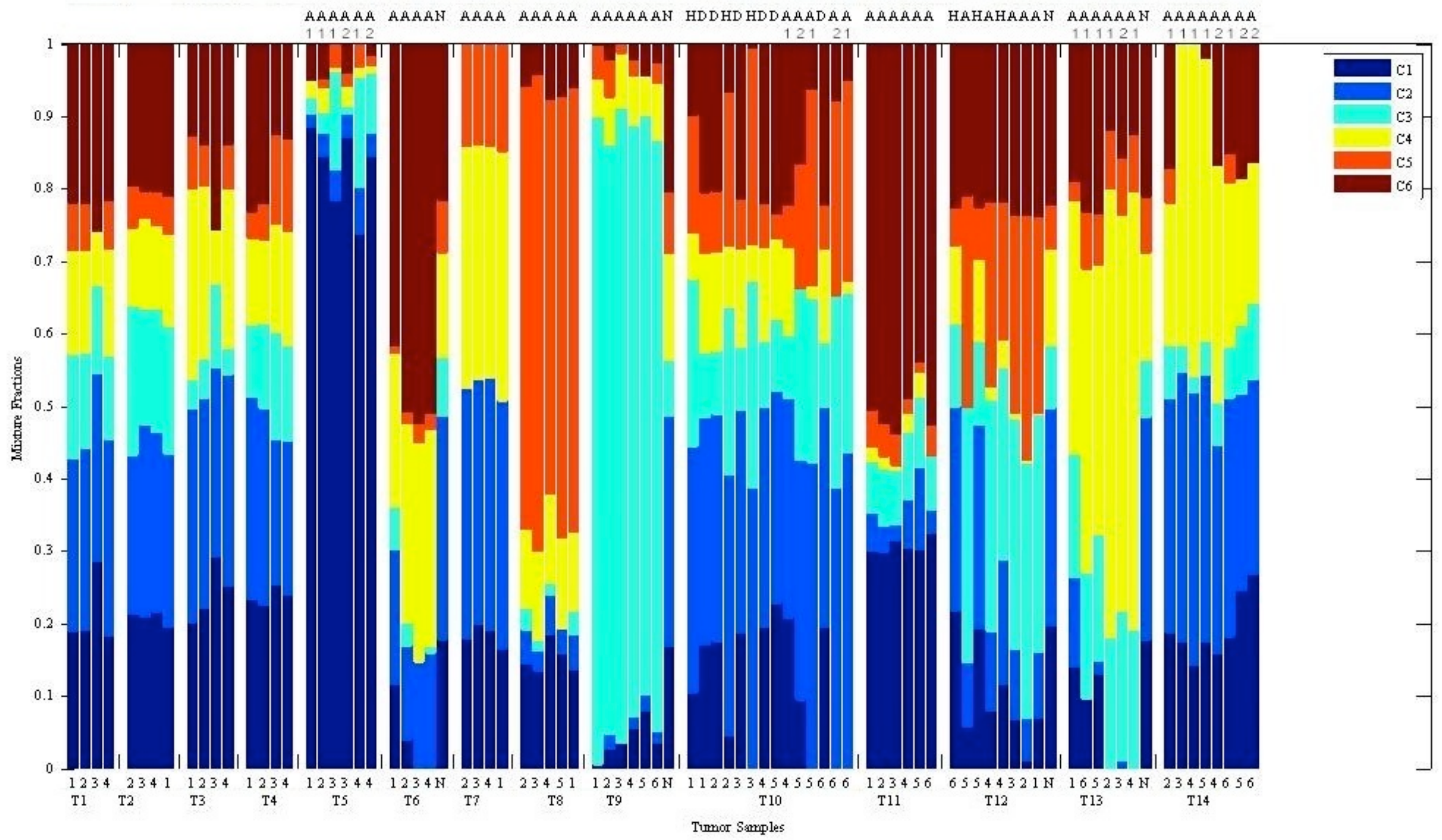}
\end{center}
	\caption{\label{fig:frac} Inferred mixture fractions for six-component soft geometric unmixing 
applied to breast cancer aCGH data.  Data is grouped by tumor, with multiple sectors per tumor placed side-by-side. 
Columns are annotated below by sector or N for normal control and above by cell sorting fraction (D for diploid, 
H for hypodiploid, A for aneuploid, and A1/A2 for subsets of aneupoloid) where cell sorting was used.}
\end{figure}

\begin{table}
\centering 
\begin{tabular}{|c|c|c|} \hline 
Marker         &  Locus         & Component        \\ \hline 
MUC1	&  1q21         & C1,C4 \\	
PIK3CA  &  3q26.3	& C3,C6 \\	
ESR1    &  6q25.1	& C4      \\
EGFR	& 7p12	& C5    \\
c-MYC&	8q24	& C1,C3,C5  \\
PTEN  &	10p23	& none   \\
PGR & 14q23.2   & C6    \\	
CCND1 &	11q13	& C4 	\\
BRCA2 &	13q12.3	& C5 \\ 
ESR2  & 14q23   & C1  \\
BRCA1, &	17q21 & C5,C6 \\ 
ERBB2  & & \\ 
 STAT5A, & 17q11.2 & C5\\
STAT5B & & \\ 
GRB7 & 17q12 & C6 \\ 
CEA & 19q13.2 & C6 \\  \hline
\end{tabular}
\caption{\label{Tab:01} Benchmark set of breast cancer markers selected for validation of real data, 
annotated by gene name, genomic locus, and the set of components exhibiting 
amplification at the given marker.}
\end{table}

\begin{figure}
\begin{center}
\includegraphics[width=3.4in, height=3.5in]{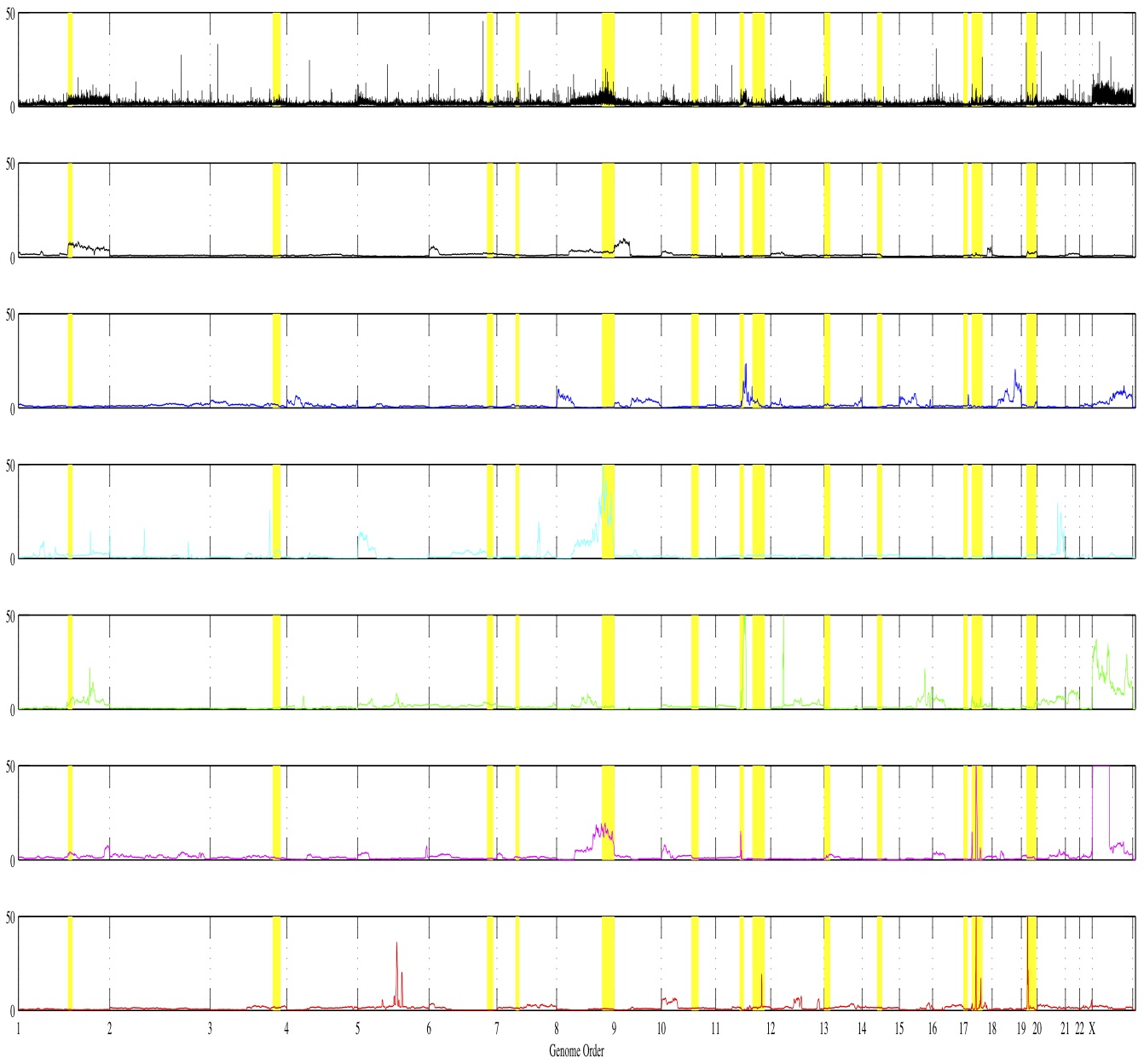}
\end{center}
	\caption{\label{fig:infcomps} Copy numbers of inferred components versus genomic position. 
The average of all input arrays (top) is shown for comparison, with the six components below. 
Benchmarks loci are indicated by yellow vertical bars.}
\end{figure}

Our long-term goal in this work is not just to identify sub-types, but
to describe the evolutionary relationships among them.  We have no
empirical basis for validating any such predictions at the moment but
nonetheless consider the problem informally here for illustrative
purposes.  To explore the question of possible ancestral relationships
among components, we manually examined the most pronounced regions of
shared gain across component.  Figure~\ref{fig:comp_phylo} shows a
condensed view of the six components highlighting several regions of
shared amplification between components.  The left half of the image
shows components 3, 5, and 1, revealing a region of shared gain across
all three components at 9p21 (labeled B).  Components 5 and 1 share an
additional amplification at 1q21 (labeled A).  Components 1 and 5 have
distinct but nearby amplifications on chromosome 17, with component 1
exhibiting amplification at 17q12 (labeled D) and component 5 at 17q21
(labeled C).  We can interpret these images to suggest a possible
evolutionary scenario: component 3 initially acquires an amplification
at 9p21 (the locus of the gene CDKN2B/p15INK4b), an unobserved
descendent of component 3 acquires secondary amplification at 1q21
(the locus of MUC1), and this descendent then diverges into components
1 and 5 through acquisition of independent abnormalities at 17q12
(site of PGAP3) or 17q21 (site of HER2).  The right side of the figure
similarly shows some sharing of sites of amplification between
components 2, 4, and 6, although the amplified regions do not lead to
so simple an evolutionary interpretation.  The figure is consistent
with the notion that component 2 is ancestral to 4, with component 2
acquiring a mutation at 5q21 (site of APC/MCC) and component 4
inheriting that mutation but adding an additional one at 17q21.  We
would then infer that the amplification at the HER2 locus arose
independently in component 6, as well as in component 5.  The figure
thus suggests the possibility that the HER2-amplifying breast cancer
sub-type may arise from multiple distinct ancestral backgrounds in
different tumors.  While we cannot evaluate the accuracy of these
evolutionary scenarios, they provide an illustration of
how the output of this method is intended to be used to make
inferences of evolutionary pathways of tumor states.

\begin{figure*}
	\begin{center}
	\includegraphics[width=6in,height=2.1in]{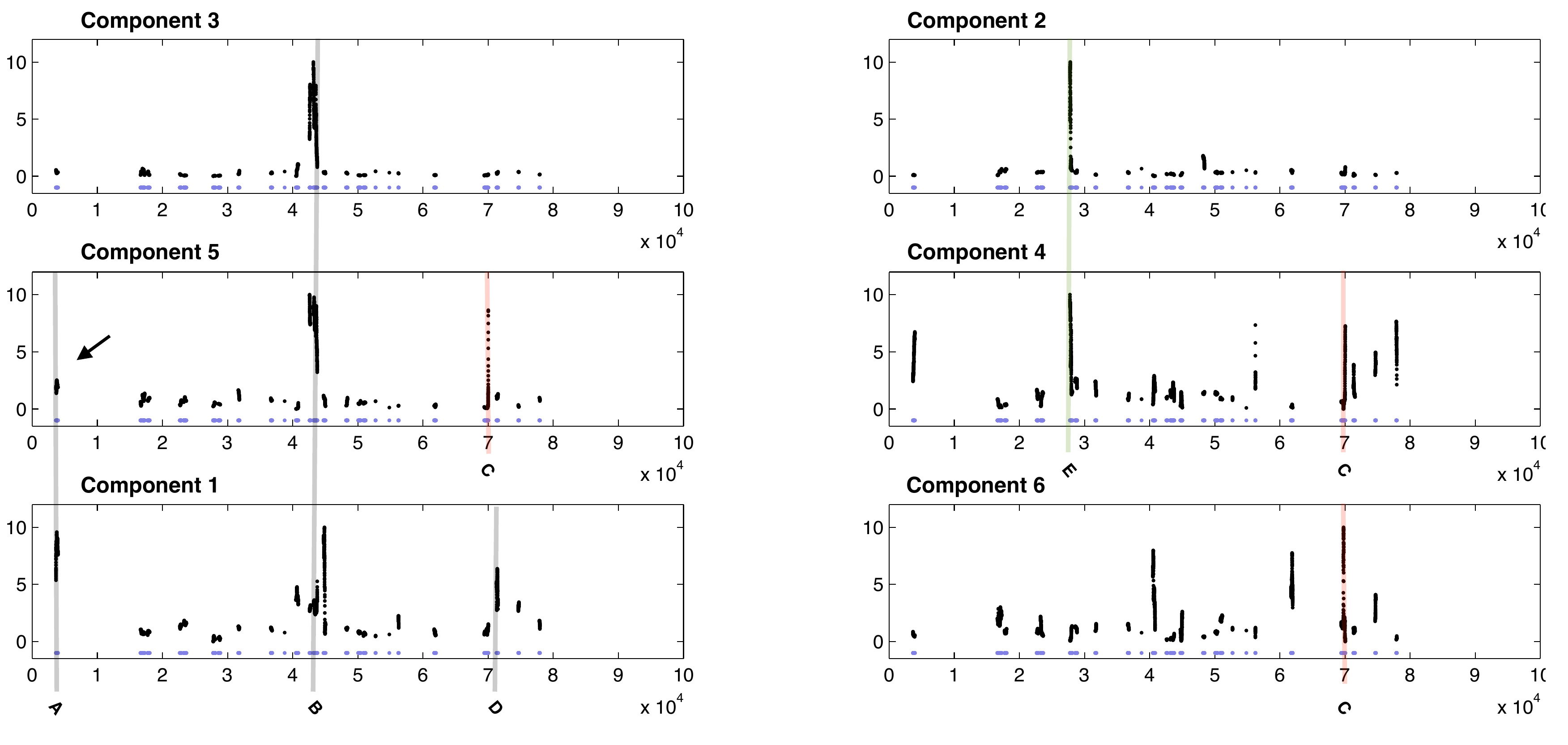}
\end{center}
	\caption{\label{fig:comp_phylo} Plot of amplification per probe highlighting regions of shared amplification across components.  
The lower (blue) dots mark the location of the collected cancer benchmarks set. Bars highlight specific markers of high shared amplification 
for discussion in the text.  {\it Above:} {\bf A}: 1q21 (site of MUC1), {\bf B}: 9p21 (site of CDKN2B), 
{\bf C}: 7q21 (site of HER2), {\bf D}: 17q12 (site of PGAP3), {\bf E}: 5q21 (site of APC/MCC).}
\end{figure*}

\chapter{Perfect Reconstruction of Oncogenetic Trees }
\label{oncotreeschapter}
\lhead{\emph{Perfect Reconstruction of Oncogenetic Trees}} 
\section{Introduction} 
\label{sec:oncotreesintro}

In this Chapter we answer a fundamental question regarding oncogenetic trees, see Section~\ref{subsec:oncotreeshere}. 
Before we state the question, we introduce some notation first. 
Let $T=(V,E,r)$ be a rooted tree on $V$, i.e., every vertex has in-degree at most one 
and there are no cycles, and let $r \in V$ be the root of $T$.
Given a finite family $\mathcal{F}=\{A_1,...A_q\}$ of sets of vertices, i.e., $A_i \subseteq V(T)$ for $i=1,\ldots,q$,
where each $A_i$ is the vertex set of a $r$-rooted sub-tree of $T$, what are the necessary and sufficient conditions, if any,  
which allow us to uniquely reconstruct  $T$? 
In this work we treat this natural combinatorial question, namely:

\begin{center}
{\em ``When can we reconstruct an oncogenetic tree $T$ from a set family $\mathcal{F}$?''}
\end{center}

Despite the fact that in practice aCGH data tend to be noisy and consistent with more than one oncogenetic trees, 
the question is nonetheless interesting and to the best of our knowledge 
remains open so far \cite{papadimitriou,schaeffer}.
Theorem~\ref{thrm:oncotreesthrm} provides the necessary and sufficient conditions to uniquely reconstruct 
an oncogenetic tree. 
We write $u \prec  v$ ($u \nprec v$) to denote that $u$ is (not) a descendant of $v$ in $T$.

\begin{theorem}
\label{thrm:oncotreesthrm}
Let $T$ be an oncogenetic tree and $\mathcal{F}=\{A_1,...A_q\}$ be a finite family of sets of vertices, 
i.e., $A_i \subseteq V(T)$ for $i=1,\ldots,q$, where each $A_i$ is the vertex set of a $r$-rooted  sub-tree of $T$ 
The necessary and sufficient conditions to uniquely reconstruct the tree $T$ from the family $\mathcal{F}$ are
the following:
\begin{enumerate}
		\item For any two distinct vertices $x,y \in V(T)$ such that $(x,y) \in E(T)$, 
                there exists a set $A_i \in \mathcal{F}$ such that $x \in A_i$ and $y \notin A_i$.
	        \item For any two distinct vertices $x,y \in V(T)$ such that $y \nprec x$ and $x \nprec y$
	        there exist sets $A_i,A_j \in \mathcal{F}$ such that $x \in A_i$, $y \notin A_i$ and $x \notin A_j$ and $y \in A_j$.
\end{enumerate} 
\end{theorem}

\noindent In Section~\ref{sec:oncoproof} we prove Theorem~\ref{thrm:oncotreesthrm}. It is worth noticing that our proof 
provides a simple procedure for the reconstruction as well. 

\section{Proof of Theorem~\ref{thrm:oncotreesthrm}} 
\label{sec:oncoproof}

\noindent 
In the following we call a tree $T$ {\it consistent} with the family set $\mathcal{F}$ if  
all sets $A_i \in \mathcal{F}$ are vertices of rooted sub-trees of $T$. 
Notice that when two or more trees are consistent with the input dataset $\mathcal{F}$, 
then we cannot uniquely reconstruct $T$. 

\begin{proof} 
First we prove the necessity of conditions 1,2 and then their sufficiency to reconstruct $T$.

\begin{figure}
		\begin{tabular}{cc} 
                    \includegraphics[width=0.4\textwidth]{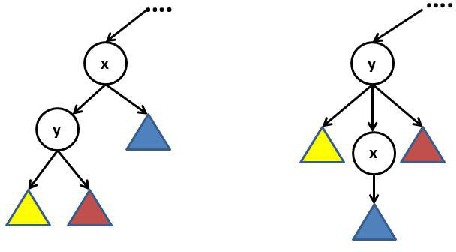} & \includegraphics[width=0.4\textwidth]{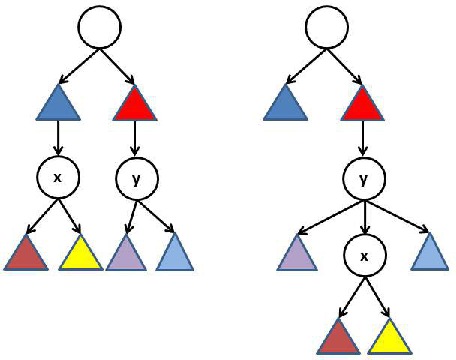}  \\ 
   (a) Necessity of Condition 1 & (b) Necessity of Condition 2  \\ 
		\end{tabular}
  \caption{\label{tab:oncotab1} Illustration of necessity conditions of Theorem~\ref{thrm:oncotreesthrm}.}
\end{figure}

\underline{Necessity:} For the sake of contradiction, assume that Condition 1 does not hold. Therefore, 
there exists two vertices $x,y \in V(T)$ such that there exists no set $A \in \mathcal{F}$ 
that contains one of them.  Then, the two trees 
shown in Figure~\ref{tab:oncotab1}(a) are both consistent with $\mathcal{F}$. Therefore we cannot reconstruct $T$. 
Similarly, assume that Condition 2 does not hold. Specifically assume that 
for all $j$ such that $x \in A_j$, then $y \in A_j$ too (for the symmetric case
the same argument holds). Then, both trees in Figure~\ref{tab:oncotab1}(b) are consistent with $\mathcal{F}$ 
and therefore $T$ is not reconstructable. The symmetric case follows by the same argument. 

\underline{Sufficiency:} Let $x \in V(T)$  and $P_x$ be the vertex set of the unique path from the root $r$ to $x$, i.e., $P_x=\{r,\ldots,x\}$. 
Also, define $F_x$ to be the intersection of all sets in the family $\mathcal{F}$ that contain vertex $x$,
i.e., $F_x =\underset{\mbox{ $A_i \ni x$ }}{ \text{~~}\bigcap A_i }$. 
We prove that $F_x = P_x$. Since by the definition of an oncogenetic tree $P_x \subseteq F_x$ it suffices
to show that $F_x \subseteq P_x$. Assume for the sake of contradiction that $F_x \nsubseteq P_x$.
Then, there exists a vertex $v \in V(T)$ such that $v \in F_x, v \notin P_x$. We consider the following three cases. 

\noindent 
\underline{$\bullet$ {\sc Case 1} ($x \prec v$):} Since by definition each set $A_i \in \mathcal{F}$ is the vertex 
set of a rooted sub-tree of $T$,  $v \in P_x$ by the definition of an oncogenetic tree.

\noindent 
\underline{$\bullet$ {\sc Case 2} ($v \prec x$):} Inductively by condition 1,
there exists $A_i \in \mathcal{F}$ such that $x \in A_i, v \notin A_i$. Therefore, $v \notin F_x$. 

\noindent 
\underline{$\bullet$ {\sc Case 3} ($x \nprec v, v \nprec x$):}
By condition 2, there exists $A_i \in \mathcal{F}$ such that  $x \in A_i$ and $v \notin A_i$.
Hence, $v \notin F_x$. 

In all three cases above, we obtain a contradiction and therefore $v \in F_x \Rightarrow v \in P_x$.
Therefore, $F_x \subseteq P_x$ and subsequently $F_x = P_x$. 
Given this fact, it is easy to reconstruct the tree $T$. We sketch the algorithm: compute for each $x$ the set 
$F_x$ which is the unordered set of vertices of the unique path from $r$ to $x$.
The ordering of the vertices which results in finding the path $P_x$, i.e.,
$(v_0=r \rightarrow v_1 \rightarrow .. v_{k-1} \rightarrow v_k=x)$ is computed using sets in $\mathcal{F}$ 
which contain $v_i$ but not $v_{i+1}$, $i=0,..,k-1$. The existence of such sets is guaranteed by condition 1.  
\end{proof}

\newpage
\vspace*{\fill}
\begingroup
\centering
{\Huge III {\em Conclusion and Future Directions} }
\endgroup
\vspace*{\fill}

\newpage

\clearpage
\chapter{Conclusion}
\label{conclchapter}
\lhead{\emph{Conclusion}}
This Chapter concludes the dissertation by posing both specific open problems 
and broader directions related to our work.

\section{Open Problems} 
\label{subsec:solidopenproblems} 

\subsection{Rainbow Connectivity (Chapter~\ref{rainbowchapter}) } 

\noindent We propose two open problems for future work.

\noindent{\bf Erd\"{o}s-R\'enyi graphs:} In Chapter~\ref{rainbowchapter} 
we give an aymptotically tight result on the rainbow 
connectivity of $G=G(n,p)$ at the connectivity
threshold. It is reasonable to conjecture that this could be tightened.
We conjecture that \whp, $rc(G) = \max\set{Z_1,diameter(G(n,p))}$.

\noindent{\bf Random regular graphs:} Our result on random regular graphs is not so tight. 
It is still reasonable to believe that the above conjecture also holds in
this case. (Of course $Z_1=0$ here).

\subsection{Random Apollonian Graphs (Chapter~\ref{ranchapter})}

\noindent We propose two open problems for future work.

\noindent{\bf Diameter:} We conjecture that $\lim_{t \rightarrow +\infty} \frac{d(G_t)}{\log{t}} \rightarrow c$, 
where $c$ is a constant. Finding $c$ is an interesting research problem. 

\noindent{\bf Conductance:} We conjecture that the conductance of a RAN is $\Theta\big(\frac{1}{\sqrt{t}}\big)$ \whp. 
Figure~\ref{fig:phi} shows that $\phi(G_t) \leq \frac{1}{\sqrt{t}}$.

\begin{figure}[h]
\begin{center} 
\includegraphics[width=0.3\textwidth]{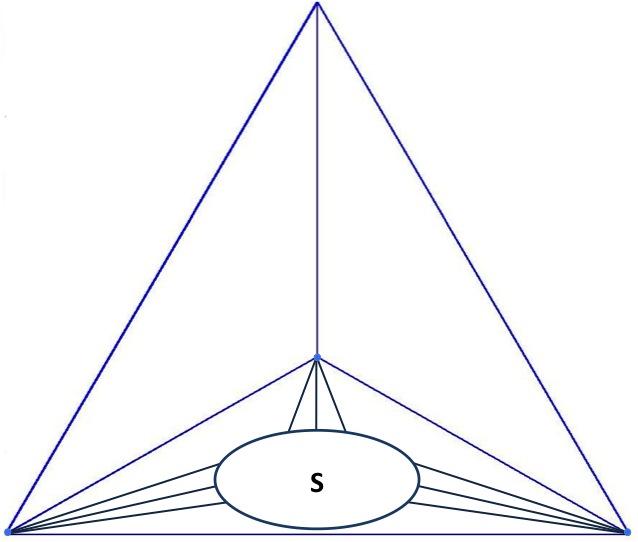}
\end{center}
\caption{By the pigeonhole principle, one of the three initial faces receives $\Theta(t)$ vertices. 
Using Theorem~\ref{thrm:RANthrm1} it is not hard to see that the encircled set of vertices $S$
has conductance $\phi(S) \approx \frac{\sqrt{t}}{t} = \frac{1}{\sqrt{t}}$ \whp.}
\label{fig:phi}
\end{figure}

\subsection{Triangle Counting (Chapter~\ref{trianglecountingchapter}) } 

\begin{figure}[htbp]
\begin{center}
\includegraphics[width=5cm]{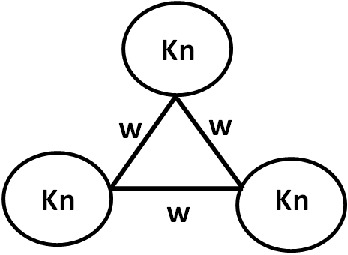}
\end{center}
\caption{\label{fig:sparsifierfig2} Weighted Graph}
\end{figure}

Certain open problems on triangle counting follow: 

\noindent{\bf  Weighted graphs:}  The sampling scheme presented in 
Section~\ref{subsec:Sparsifieralgorithm} can be adapted to weighted graphs: 
multiply the weight of sampled edge by $\frac{1}{p}$ and 
count weighted triangles. However this can be problematic as the graph shown in Figure~\ref{fig:sparsifierfig2} indicates.
Specifically, for a sufficiently large $w$, throwing out any weighted edge results in an arbitrarily 
bad estimate of the count of triangles. Finding a better sampling scheme for weighted graphs is left as a problem 
for future work.

\noindent{\bf Better sparsification:}  
Finding an easy-to-compute quantity which gives us an optimal way to sparsify the graph with respect to triangles
is an interesting research problem.  

\noindent{\bf Random projections and triangles:} 
Can we provide some reasonable condition on $G$ that would guarantee \eqref{condition}? 

\noindent{\bf \mapreduce :} 
Our proposed methods are easily parallelizable and developing a \mapreduce implementation  is a natural practical direction.

\noindent{\bf Streaming model :}
Besides \mapreduce, the streaming model is particularly well tailored for large data. 
Can we design better algorithms for triangle counting for the streaming model \cite{buriol}?
Can we estimate the average number of triangles over all vertices with a given degree $d$?

\subsection{Densest Subgraphs (Chapter~\ref{densestchapter}) }

Our work in Chapter~\ref{densestchapter} leaves several open problems. 

\noindent{\bf Algorithm Design:}  
Can we design approximation algorithms with better additive guarantees
or a multiplicative approximation algorithm without shifting the objective? 
Can we characterize the worst-case behavior of \LSalgo? 
Also designing efficient randomized algorithms, e.g., \cite{evocut} is an interesting
research direction. 

\noindent{\bf Bipartite graphs:}   A natural extension of our objective to a bipartite
graph $G(L \cup R, E)$  is the following. For a set $A \subseteq L, B \subseteq R$
and a real parameter $\alpha$, we define $f_{\alpha}(A,B) = e[A \cup B] - \alpha |A| |B|$.
Maximizing $f_{\alpha}$ is in general a hard problem. This can be easily shown using a 
reduction similar to the one we used in Chapter~\ref{densestchapter} where we used a planted
clique in an Erd\"{o}s-R\'enyi $G(n,1/2)$ graph 
and invoking the results of \cite{feldman2013statistical}
on detecting planted bipartite clique distributions. 
Understanding this objective and evaluating its performance 
is an interesting direction.


\subsection{Structure of the Web Graph  (Chapter~\ref{hadichapter}) } 

\noindent Studying the Web graph is an important problem in graph mining research, given 
its prominence. Two problems related to our work follow. 

\noindent{\bf Improving \hadi:} 
How well do recent advances in moment estimation  \cite{Kane:2010:OAD:1807085.1807094} perform in practice? 
Can we use them to improve the performance of \hadi? 

\noindent{\bf Analyzing the Web graph:}   Using our algorithmic tool \hadi we provided
insights into certain fundamental properties of the Web graph. 
Can we provide further insights by studying other statistics, e.g., \cite{ugandersubgraph}? 
Can we use these structural properties to come up with good models of the Web graph?

\subsection{FENNEL: Streaming Graph Partitioning for Massive Scale Graphs (Chapter~\ref{fennelchapter}) }


\noindent{\bf How to interpolate?}   In Chapter~\ref{fennelchapter} we showed that by selecting 
an exponent $\gamma$ between 1 and 2 allows us to interpolate between the two heuristics 
most frequently used for balanced graph partitioning. 
Can we find an easy-to-compute graph statistic, e.g., power law exponent, 
that allows us to make a good selection for $\gamma$ for a given graph? 

\noindent{\bf Random graphs with hidden partition:}    Assume that we have a dense graph 
with $k$ clusters where the probability of having an edge insider and between two distinct
clusters is $p$ and $q$ respectively, $p>q=\Theta(1)$. 
Can we predict how different heuristics derived from our framework will perform? 

\noindent{\bf More cost functions:}   
It is interesting to test other choices of cost functions and tackle
the assymetric edge cost issue that occurs in standard data center architectures.

\noindent{\bf Communication cost:}   A future goal is to minimize further the 
communication cost  using minwise hashing~\cite{satuluri}.

\subsection{PEGASUS: A System for Large-Scale Graph Processing (Chapter~\ref{pegasuschapter})}

Given that \pegasus is much more an engineering  rather than a theoretical contribution,
we provide a broad research direction, related to \mapreduce, the framework on 
which \pegasus relies. 

\noindent{\bf Understanding the limits of \mapreduce:} 
Afrati, Das Sarma, Salihoglu and Ullman \cite{afrati2012vision}
provide an interesting framework for studying \mapreduce\ algorithms. 
Let the replication rate of any \mapreduce algorithm\ be the average
number of reducers each input is sent to.
Understanding trade-offs between the replication rate and 
the efficiency of \mapreduce\ algorithms is an interesting research direction.

\subsection{Approximation Algorithms for Speeding up Dynamic Programming and Denoising aCGH data (Chapter~\ref{acghchapter})}

In Chapter~\ref{acghchapter}, we present a new formulation for the problem of denoising aCGH data. 
Our formulation has already proved to be valuable in numerous settings \cite{ding2010robust}. 

\noindent{\bf Improving algorithmic performance:} 
Our results strongly indicate that the $O(n^2)$ algorithm, which is
--- to the best of our knowledge --- the fastest exact algorithm, is not tight. 
There is inherent structure in the optimization problem. Lemma~\ref{thrm:Trimmerlemmalast} 
is such an example. 

\begin{figure}
		\begin{tabular}{cc} 
                 \includegraphics[width=0.43\textwidth]{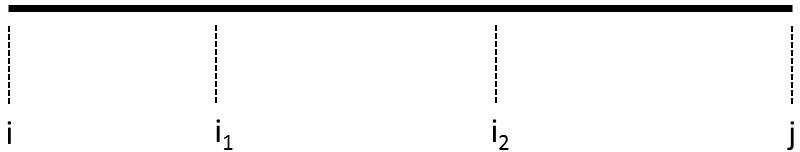} & \includegraphics[width=0.45\textwidth]{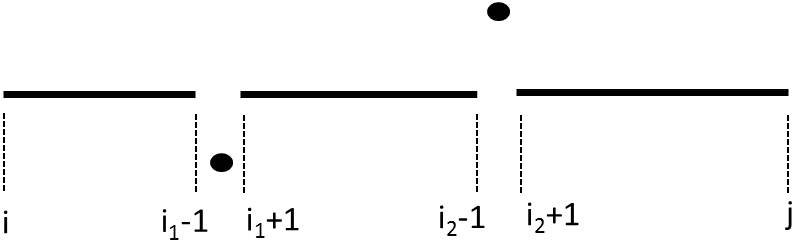} \\ 
		(a) One fitted segment & (b) Five fitted segments \\
		\end{tabular}
  \caption{\label{fig:acghfig8} Lemma~\ref{thrm:Trimmerlemmalast}: if $|P_{i_1}-P_{i_2}|>2\sqrt{2C}$ then Segmentation (b) which 
  consists of five segments (two of which are single points) is superior to Segmentation (a) which 
  consists of a single segment. }
\end{figure}

\begin{lemma}
\label{thrm:Trimmerlemmalast}
If $|P_{i_1}-P_{i_2}|>2\sqrt{2C}$, then in the optimal solution of the dynamic programming using $L_2$ norm, $i_1$ and $i_2$ are
in different segments.
\end{lemma}

\begin{proof}
The proof is by contradiction, see also Figure~\ref{fig:acghfig8}.
Suppose the optimal solution has a segment $[i,j]$ where $i \leq i_1 < i_2 \leq j$, and its optimal
$x$ value is $x^{*}$. Then consider splitting it into 5 intervals $[i, i_1-1]$, $[i_1, i_1]$, $[i_1+1, i_2-1]$, $[i_2, i_2]$, $[i_2+1, r]$.
We let $x=x^{*}$ be the fitted value to the three intervals not containing $i_1$ and $i_2$. 
Also, as  $|P_{i_1}-P_{i_2}|>2\sqrt{2C}$,
$ (P_{i_1}-x)^2 + (P_{i_2}-x)^2>2\sqrt{2C}^2 = 4C$. So by letting $x=P_{i_1}$ in $[i_1, i_1]$ and $x=P_{i_2}$ in $[i_2, i_2]$,
the total decreases by more than $4C$. This is more than the added penalty of having 4 more segments, a contradiction with the
optimality of the segmentation.
\end{proof}

Uncovering and taking advantage of the inherent structure in a principled way should result in a faster 
exact algorithm. This is an interesting research direction which we leave as future work.

\noindent{\bf More applications:}  Another research direction is to find more applications (e.g., histogram construction \cite{1132873}) 
to which our methods are applicable.

\noindent{\bf aCGH Denoising:} Our model does not perform well in very noisy data. Developing novel machine learning techniques 
for detecting signal breakpoints is a practical research direction.

\subsection{Robust Unmixing of Tumor States in Array Comparative Genomic Hybridization Data (Chapter~\ref{unmixingchapter})}

In Chapter~\ref{unmixingchapter} we introduce a geometric framework which provides a novel perspective
on detecting cancer subtypes. From a computational perspective, we aim to fit to a cloud of points a minimum volume
enclosing simplex. 

\noindent{\bf Complex geometric shapes:} Can we fit instead of a simplex, a simplicial complex? 

\subsection{Perfect Reconstruction of Oncogenetic Trees (Chapter~\ref{oncotreeschapter})}

We propose the following open problems on oncogenetic trees. 

\noindent{\bf Enumerative properties of oncogenetic trees:} 
How does one calculate the size of the smallest family $\mathcal{F}$ 
that is uniquely consistent with a given tree $T$? Which trees are extremal in the aforementioned sense? 
A broad research direction is to obtain a deeper understanding of  tumorigenesis 	
through the understanding of combinatorial properties of oncogenetic trees.



\addtocontents{toc}{\vspace{2em}} 

\backmatter


\label{Bibliography}

\lhead{\emph{Bibliography}} 

\bibliographystyle{plainnat}

\bibliography{ref} 

\end{document}